\documentclass[11pt,a4paper]{article}
\usepackage[a4paper,margin=0.75in]{geometry}
\usepackage[T1]{fontenc}
\usepackage{lmodern}
\usepackage{graphicx,xcolor,comment,verbatim}
\usepackage{amsmath,amssymb,amsthm,mathrsfs,bm}
\usepackage{booktabs,array,multirow,rotating}
\usepackage{caption,subcaption}
\usepackage{algorithm,algpseudocode,float}
\usepackage[authoryear,round]{natbib}
\usepackage{bibunits}
\usepackage[hidelinks]{hyperref}
\usepackage{cleveref}
\usepackage{setspace,placeins,indentfirst}
\usepackage[normalem]{ulem}
\usepackage{tikz}
\RenewDocumentCommand{\enlargethispage}{s m}{}
\hypersetup{pdftitle={Multivariate conditional randomisation tests with unknown covariate parameters},pdfauthor={Xiaotong Lin, Jie Xie, Fangqiao Tian, and Dongming Huang}}

\theoremstyle{plain}
\usepackage{bbm}

\def\independenT#1#2{\mathrel{\rlap{$#1#2$}\mkern2mu{#1#2}}}
\newcommand\independent{\protect\mathpalette{\protect\independenT}{\perp}}

\newcommand{\mc}{\mathcal}
\newcommand{\mb}{\mathbf}
\newcommand{\bs}{\boldsymbol}

\newcommand{\tp}{^{\top}}

\newcommand{\bbR}{\mathbb{R}}
\newcommand{\bbE}{\mathbb{E}}
\newcommand{\bbP}{\mathbb{P}}
\newcommand{\al}{\alpha}
\newcommand{\rmd}{\,{\rm d}\,}

\newcommand\restr[2]{{
  \left.\kern-\nulldelimiterspace
  #1
  \littletaller
  \right|_{#2}
  }}
\newcommand{\littletaller}{\mathchoice{\vphantom{\big|}}{}{}{}}

\newcommand{\indp}{\independent}

\usepackage{xspace}
\newcommand{\nameVV}{$M^1 P_1$\xspace}
\newcommand{\nameDP}{Bonf\xspace}

\newtheorem{lemma}{Lemma}

\newtheorem{condition}{Condition}
\newtheorem{corollary}{Corollary}

\newtheorem{theorem}{Theorem}
\newtheorem{example}{Example}
\newtheorem{proposition}{Proposition}

\newtheorem{remark}{Remark}
\numberwithin{equation}{section}
\providecommand{\theHalgorithm}{\arabic{algorithm}}

\makeatletter
\newcommand{\AppendixTableFigureLayout}{
  \setlength{\@fptop}{0pt}
  \setlength{\@fpsep}{12pt plus 2pt minus 2pt}
  \setlength{\@fpbot}{0pt plus 1fil}
  \setlength{\textfloatsep}{12pt plus 2pt minus 2pt}}
\makeatother

\title{Multivariate conditional randomisation tests\\with unknown covariate parameters}
\author{Xiaotong Lin\thanks{Department of Statistics and Data Science, National University of Singapore.}\kern0.4em\thanks{Xiaotong Lin and Jie Xie contributed equally to this work.}\quad
Jie Xie\footnotemark[1]\kern0.4em\footnotemark[2]\quad
Fangqiao Tian\thanks{Department of Mathematics, National University of Singapore.}\quad
Dongming Huang\footnotemark[1]\kern0.4em\thanks{Corresponding author: \href{mailto:huang.dongming@nus.edu.sg}{huang.dongming@nus.edu.sg}.}}
\date{26 November 2026}
\begin{document}
\maketitle
\begin{abstract}
We study conditional randomisation tests for independence between a response and a group of target covariates given background covariates. The conditional law of the targets is modelled by a parametric family with unknown parameters. A randomisation law based on parameter estimates from the observed covariate sample can fail to preserve the exchangeability required for validity. We instead condition on sufficient statistics to remove these parameters from the randomisation law and generate exchangeable copies of the target covariates. The resulting tests control Type I error in finite samples when the covariate model holds, without requiring a response model. We develop implementations for multivariate normal covariates and Gaussian graphical models. In Gaussian linear regression, the multivariate normal test attains minimax separation rates, and we give sufficient conditions for consistency of the graphical model test even when background covariates outnumber observations. We also bound Type I error inflation under departures from the covariate model, including errors in graphs estimated from independent covariate data. In simulations, our tests avoid the Type I error inflation seen in some plug-in tests and achieve power close to that of the corresponding oracle tests. Two exploratory analyses in finance and genomics illustrate the methods.
\end{abstract}
\section{Introduction}\label{sec: intro}
\label{sec:Introduction}

We study whether a group of covariates provides information about a response beyond the other observed covariates. 
We formalise this question as \textit{multivariate conditional independence} (MCI). 
Let $Y$ denote the response variable and $\bs{X}$ be a $p$-dimensional covariate vector. Suppose $\bs{X}$ is divided into two subvectors, $\bs{X}_{\mathcal{T}}$ and $\bs{X}_{\mathcal{S}}$, where $\mathcal{T}$ indexes the covariates of interest and $\mathcal{S}=\mathcal{T}^c$ indexes the background covariates.
The MCI null hypothesis is
\begin{equation}\label{eq: CI-T}
H_0: {Y}\indp {\bs{X}}_{\mc{T}}\mid {\bs{X}}_{\mc{S}},
\end{equation}
meaning that $\bs{X}_{\mc{T}}$ provides no additional information about $Y$ given $\bs{X}_{\mc{S}}$.
We focus on target groups with $|\mc{T}|>1$.
Testing MCI is useful for model simplification: it assesses whether
$\bs{X}_{\mathcal{T}}$ adds information about $Y$ beyond
$\bs{X}_{\mathcal{S}}$, and hence whether the target group can be removed
from subsequent modelling. It also provides joint inference
for predefined groups, such as gene sets \citep{goemanGlobalTestGroups2004}
or market sectors \citep{fanProjectionbasedConditionalDependence}.

Testing MCI in such settings faces three recurring difficulties: small samples can make classical asymptotic approximations unreliable; the dependence of $Y$ on $(\bs{X}_{\mc{T}},\bs{X}_{\mc{S}})$ can be too complex to model credibly; and when $\mc{T}$ is large, coordinatewise tests can miss a substantial joint effect. 

The conditional randomisation test (CRT) of \citet{candes_panning_2017} is a natural starting point: if $\mathcal L(\bs X_{\mathcal T}\mid\bs X_{\mathcal S})$ were known, $M$ copies could be sampled directly and their exchangeability with the observed target would give a valid test with any test statistic. 
In practice, $\mathcal L(\bs X_{\mathcal T}\mid\bs X_{\mathcal S})$ is usually modelled with unknown nuisance parameters, so this direct sampling is not possible.
Estimating these parameters from the covariates and plugging the estimates into the CRT resampling law can fail to preserve the exchangeability required for validity. 
In one simulation setting with a correctly specified Gaussian covariate model, a plug-in CRT has an estimated size of $17.8\%$ at the nominal $0.05$ level (Section~\ref{sec:simulation-unknown-parameters}).

We address this problem and make three contributions.

\noindent\textit{(i) Exact tests with unknown covariate parameters.}
We introduce the multivariate sufficient statistic CRT (MS-CRT) for testing MCI when $\mathcal L(\bs X_{\mathcal T}\mid\bs X_{\mathcal S})$ contains unknown nuisance parameters. 
Following \citet{besag_generalized_1989}, we condition on a sufficient statistic $\Phi$ to remove these parameters from the randomisation law given the background covariate matrix $\mathbf X_{\mathcal S}$.
We then generate copies of the target covariate matrix $\mathbf X_{\mathcal T}$ that are exchangeable with it under this conditioning. When the covariate model holds, this gives finite-sample Type I error control without modelling $Y\mid\bs X$.

We give two implementations. The first, the multivariate normal CRT (MVN-CRT), provides direct conditional sampling for multivariate normal covariates with unknown mean and covariance. 
The second, $G$-CRT, uses Gaussian graphical models for high-dimensional covariates and applies the subset sampler of \citet{lin_goodness--fit_2025} to test MCI.
This sampler combines local residual rotations with the hub-and-spoke construction of \citet{besag_generalized_1989}.
The supplied graph $G$ specifies which entries of the precision matrix may be nonzero, while the mean and the values of those entries remain unknown.
When the neighbourhoods of the target covariates are sufficiently small, these residual rotations permit nontrivial randomisation even when $p\geq n$, where $n$ is the sample size.
The approach also extends to discrete graphical models. 

\noindent\textit{(ii) Power after conditioning.}
In Gaussian linear regression, MVN-CRT with suitable test statistics attains the minimax separation rates in the $\ell_2$ and $\ell_\infty$ norms of the conditional covariance between $\bs X_{\mathcal T}$ and $Y$ given $\bs X_{\mathcal S}$.

For $G$-CRT, randomised copies can preserve part of the target
group's association with the response, so a strong observed
signal alone does not guarantee power. We couple the nonlinear
residual rotation construction to an oracle Gaussian construction and establish sufficient separation between the observed and randomised
signal contributions. This yields uniform consistency, even with more background covariates than observations.

\noindent\textit{(iii) Robustness to misspecification.}
We bound Type I error inflation by the expected total variation distance from the true conditional law of $\mathbf X_{\mathcal T}$ given $\mathbf X_{\mathcal S}$ to the laws that agree with the model conditionally on $(\mathbf X_{\mathcal S},\Phi)$.
This bound is no larger than that based on the original parametric family. It also shows that exact validity can extend beyond that family.
For Gaussian graphical models whose graph is estimated from independent covariate data, we quantify how errors in that graph, including omitted edges, affect Type I error (Section~\ref{sec: ROBT}).

Beyond the theoretical results, we evaluate MVN-CRT and $G$-CRT numerically. Our simulations compare MVN-CRT with an oracle CRT and three plug-in CRTs using the same test statistic. The plug-in CRTs estimate the conditional covariate law from the covariate sample used for testing and can substantially inflate Type I error even under a correctly specified Gaussian model. In the settings studied, MVN-CRT has estimated size near the nominal level and rejection rates close to those of the oracle across the alternatives considered.
Across the linear and nonlinear response models examined, $G$-CRT likewise has estimated size near the nominal level. In comparisons with oracle joint CRTs that use the same test statistics, $G$-CRT achieves similar power without knowing the covariate mean or precision matrix (Appendix~\ref{app:matched-ridge}, Table~\ref{tab:matched-direct-comparison}).
Two exploratory analyses in finance and genomics illustrate the methods. 

The remainder of the paper is organised as follows. Section~\ref{sec: CRT} reviews the oracle CRT. Sections~\ref{sec: MSCRT} and~\ref{sec: construction} present MS-CRT and its constructions; Sections~\ref{sec:power} and~\ref{sec: ROBT} study power and robustness. Sections~\ref{sec: Simulation} and~\ref{sec: Application} present simulations and applications, followed by the discussion in Section~\ref{sec: Discussion}. Further literature, proofs, algorithms, numerical results, and application details are provided in the supplementary material.

\subsection{Related literature}\label{sec: literature}
\label{sec: literature CRT}

As shown by \citet{shah2020}, valid conditional independence testing with nontrivial power necessarily requires assumptions on the class of null distributions.
Our work is most directly related to model-X conditional independence testing, especially the CRT of \citet{candes_panning_2017}, which places assumptions on the covariate distribution $
\mathcal{L}(\bs{X})$ while leaving the conditional law of the response $\mathcal{L}(Y\mid \bs{X})$ unrestricted.

Existing CRT variants have addressed computational cost, the construction of test statistics, and the effect of using estimated or approximate conditional covariate laws \citep{bates2020causal, shaer_model-x_2023, tansey_holdout_2022, liu_fast_2022, niu2024spa, bellot_conditional_2019, berrett_conditional_2020, barber_testing_2021, niu_reconciling_2024, li_maxway_2023}. MS-CRT addresses the complementary problem of recovering exact exchangeable randomisation for a multivariate target group when the structured conditional law $\mathcal L(\bs{X}_{\mc T}\mid \bs{X}_{\mc S})$ contains unknown nuisance parameters.

The conditional permutation test remains valid under covariate model misspecification when the model still gives the correct conditional distribution of permutations \citep{berrett_conditional_2020}. MS-CRT likewise remains valid under misspecification that preserves the specified conditional distribution of $\mathbf X_{\mathcal T}$ given $(\mathbf X_{\mathcal S},\Phi)$.

Conditioning plays different roles in related methods. In MS-CRT, conditioning on the sufficient statistic $\Phi$ removes the unknown covariate parameters from the resampling law. For a scalar target $X$ and background covariates $Z$, the MAXWAY CRT of \citet{li_maxway_2023} uses a low-dimensional summary $h(Z)$ that ideally satisfies $X\indp Z\mid h(Z)$, together with an analogous summary for the response. Its aim is to improve robustness to errors in an estimated covariate law. Appendix~\ref{app:model-x-inference} compares the exchangeability conditions. In multiple testing, \citet{fithian_conditional_2022} condition on sufficient statistics to remove nuisance parameters from the null distributions used to calibrate rejection thresholds for false discovery rate control. \citet{wang2022high} analyse the power of a univariate CRT that conditions on sufficient statistics. Section~\ref{sec: MVN-CRT power} shows that MVN-CRT attains minimax separation rates under a Gaussian linear model, with the target dimension allowed to grow.
\par

Model-X knockoffs also rely on exchangeability \citep{candes_panning_2017}. \citet{huang_relaxing_2020} construct conditional knockoffs by conditioning on sufficient statistics for unknown covariate parameters. Knockoff sampling imposes columnwise swap invariance, whereas MS-CRT uses joint exchangeability of the observed target matrix and its randomised copies conditional on $(\mathbf X_{\mc S},\Phi)$.
 Appendix~\ref{rem:knockoff} compares knockoff sampling with the construction of exchangeable copies for MS-CRT.

Our power analysis is related to work on high-dimensional regression such as \citet{shah2018goodness}, \citet{WANG2013134}, \citet{verzelen_goodness--fit_2010}, and \citet{ning_general_2017}. These papers study coefficient testing or regression goodness-of-fit. 
Throughout, MS-CRT remains a model-X procedure: its finite-sample validity follows from randomisation and conditional exchangeability under assumptions on the covariate law. 
In Section~\ref{sec:power}, the Gaussian linear model serves as a tractable setting for our power analysis, not as the modelling basis for validity. Appendix~\ref{app:regression-ci-comparison} provides a more detailed review of related literature beyond the model-X setting.

\subsection{Notation}\label{sec: Notation}
We write vectors in bold italic and matrices in bold upright; random vectors and matrices are uppercase and their realisations lowercase. 
Column $j$, row $i$, and entry $(i,j)$ of a matrix $\mathbf X$ are $\mathbf X_j$, $\boldsymbol X_{i\cdot}$, and $X_{ij}$.
Let $[p]:=\{1,\dots,p\}$ for a positive integer $p$. For $\mathcal T\subseteq[p]$, write $\boldsymbol X_{\mathcal T}:=(X_j)_{j\in\mathcal T}$ for the subvector of a vector $\boldsymbol X$ and $\mathbf X_{\mathcal T}$ for the submatrix of a matrix $\mathbf X$ formed by the columns in $\mathcal T$; the subscript $-\mathcal T$ refers to $\mathcal T^c$. Write $\|\cdot\|_q$ for the $\ell_q$ norm ($1\leq q\leq\infty$), $\mathrm{diag}(\mathbf A)$ for the diagonal matrix with the diagonal entries of a square matrix $\mathbf A$, and $\mathbf A\succ\mathbf 0$ if $\mathbf A$ is positive definite. 
For matrices, $\|\cdot\|_{\rm op}$, $\|\cdot\|_1$, and $\|\cdot\|_\infty$ denote the operator norm induced by the Euclidean norm, maximum absolute column sum, and maximum absolute row sum, respectively.
The $p$-variate normal distribution with mean $\boldsymbol\mu$ and covariance matrix $\mathbf\Sigma\succ\mathbf 0$ is $\mathcal N_p(\boldsymbol\mu,\mathbf\Sigma)$, and its precision matrix is $\mathbf\Omega:=\mathbf\Sigma^{-1}=(\omega_{ij})$. The distribution of $\boldsymbol X$ and its conditional distribution given $\boldsymbol Y$ are $\mathcal L(\boldsymbol X)$ and $\mathcal L(\boldsymbol X\mid\boldsymbol Y)$. Write $\boldsymbol X\indp\boldsymbol Y$ for independence, $\boldsymbol X\indp\boldsymbol Y\mid\boldsymbol Z$ for conditional independence given $\boldsymbol Z$, and $\boldsymbol X\overset{d.}{=}\boldsymbol Y$ for equality in distribution. For probability distributions $P_1$ and $P_2$ on the same measurable space, $d_{\mathrm{TV}}(P_1,P_2):=\sup_A|P_1(A)-P_2(A)|$ is their total variation distance, where $A$ ranges over all measurable sets.

\section{Oracle CRT and conditional exchangeability}\label{sec: CRT}

We first describe the CRT of \citet{candes_panning_2017} for testing~\eqref{eq: CI-T} when $\mathcal L(\bs X_{\mc T}\mid\bs X_{\mc S})$ is known, and call it the \emph{oracle CRT}.
Suppose that $\{(Y_i,\bs X_i):i=1,\ldots,n\}$ are i.i.d. copies of $(Y,\bs X)$.
Let $\bs Y=(Y_1,\ldots,Y_n)$, and let $\mathbf X$ be the matrix with rows $\bs X_1,\ldots,\bs X_n$.
Fix a positive integer $M$. Conditional on $\mb X_{\mc S}$, draw
$\widetilde{\mb X}_{\mc T}^{(1)},\ldots,\widetilde{\mb X}_{\mc T}^{(M)}$
independently from $\mathcal{L}(\mb X_{\mc T}\mid\mb X_{\mc S})$ and independently
of $(\bs Y,\mb X_{\mc T})$. For any statistic
\(T(\bs y,\mb x_{\mc T},\mb x_{\mc S})\) whose larger values indicate
stronger evidence against \(H_0\), define
$T_0:=T(\bs Y,\mb X_{\mc T},\mb X_{\mc S})$ and
$T_m:=T\left(\bs Y,\widetilde{\mb X}_{\mc T}^{(m)},\mb X_{\mc S}\right)$,
$m=1,\ldots,M$.
Define the p-value by
\begin{equation}\label{eqn:pval-1}
\mathrm{pVal}_{T}=\frac{1}{M+1}\left(1+\sum_{m=1}^{M}\mathbf{1}\{T_m\geq T_0\}\right).
\end{equation}

Under $H_0$, the observed target matrix $\mb X_{\mc T}$ and the sampled copies
$\widetilde{\mb X}_{\mc T}^{(1)},\ldots,
\widetilde{\mb X}_{\mc T}^{(M)}$
are conditionally exchangeable given $(\bs Y,\mb X_{\mc S})$. Consequently,
$T_0,T_1,\ldots,T_M$ are conditionally exchangeable, and
$\mathbb{P}\left(\mathrm{pVal}_{T}\leq\alpha\right)\leq\alpha$
for every $\alpha\in(0,1)$.
The choice of \(T\) determines how evidence across the target covariates is
combined and affects power, but not validity. 

In practice, $\mathcal{L}(\mb X_{\mc T}\mid\mb X_{\mc S})$ is rarely fully known, so direct sampling is generally unavailable.
Consider a parametric family $P_\theta(\cdot\mid\mb X_{\mc S})$ for this conditional law.
A plug-in CRT computes an estimate $\widehat\theta=\widehat\theta(\mb X_{\mc T},\mb X_{\mc S})$ from the observed covariate sample.
Conditional on $(\mb X_{\mc S},\widehat\theta)$, it draws $\widetilde{\mb X}_{\mc T}^{(m)}$ i.i.d. from the fitted law $P_{\widehat\theta}(\cdot\mid\mb X_{\mc S})$.
However, the observed target matrix $\mb X_{\mc T}$ is additionally constrained to produce the given value of $\widehat\theta$, so it need not follow the fitted law $P_{\widehat\theta}(\cdot\mid\mb X_{\mc S})$.
Consequently, the copies need not be exchangeable with $\mb X_{\mc T}$, even when the parametric family is correctly specified, and the resulting CRT is not guaranteed to be valid in finite samples.
Section~\ref{sec:simulation-unknown-parameters} reports the Type I error
inflation observed for the plug-in procedures considered there.

\section{General MS-CRT and exact validity}\label{sec: MSCRT}\label{sec: MS-CRT}\label{sec: CRT statistic}

Suppose now that $\mathcal L(\mathbf X_{\mc T}\mid\mathbf X_{\mc S})=P_\theta(\cdot\mid\mathbf X_{\mc S})$ for an unknown nuisance parameter $\theta$.
Instead of estimating $\theta$, we condition on a statistic $\Phi=\phi(\mathbf X_{\mc T};\mathbf X_{\mc S})$ that is sufficient for $\theta$, so that the distribution of $\mathbf X_{\mc T}$ given $(\mathbf X_{\mc S},\Phi)$ is free of $\theta$.

The \textit{multivariate sufficient statistic CRT} (MS-CRT) draws $M$ copies $\widetilde{\mathbf X}_{\mc T}^{(m)}$ of $\mathbf X_{\mc T}$ without using $\bs Y$, such that
$\bigl(\mathbf X_{\mathcal T},\widetilde{\mathbf X}_{\mathcal T}^{(1)},\ldots,\widetilde{\mathbf X}_{\mathcal T}^{(M)}\bigr)$
is exchangeable conditional on $(\mathbf X_{\mathcal S},\Phi)$.
We call these $M+1$ matrices \emph{candidates}. 

As in Section~\ref{sec: CRT}, we score each candidate using the same statistic $T$.
We also allow $T$ to depend on a common statistic $H$ of $\bs Y$, $\mathbf X_{\mc S}$, and all candidates, provided that $H$ is unchanged by permutations of the candidates. 
Define the observed score 
$T_0=T(\bs Y,\mathbf X_{\mc T},\mathbf X_{\mc S};H)$ and the copy scores
$T_m=T(\bs Y,\widetilde{\mathbf X}_{\mc T}^{(m)},\mathbf X_{\mc S};H)$,
$m=1,\ldots,M$.
Larger scores indicate stronger evidence against $H_0$, and the p-value $\mathrm{pVal}_T$ is computed as in~\eqref{eqn:pval-1}.
The following proposition establishes finite-sample validity.

\begin{proposition}\label{prop: CRT valid}
Suppose the copies are generated without using $\bs Y$ and
$(\mathbf X_{\mc T},\widetilde{\mathbf X}_{\mc T}^{(1)},\ldots,
\widetilde{\mathbf X}_{\mc T}^{(M)})$
is exchangeable conditional on $(\mathbf X_{\mc S},\Phi)$.
Then, under $\bs Y\indp\mathbf X_{\mc T}\mid\mathbf X_{\mc S}$,
we have $\mathbb P(\mathrm{pVal}_T\leq\alpha)\leq\alpha$
for every $\alpha\in(0,1)$.
\end{proposition}

Appendix~\ref{app:scoring-conditions} extends this result to more general scoring rules, including those using randomised algorithms such as random forests.
Appendix~\ref{app:pvalue-conventions} defines randomised variants of the p-value.

\begin{remark}[Choice of scores]\label{rem:score}
Scores can summarise overall model fit or combine measures of importance for individual target covariates by taking their sum or maximum.
Such measures include the absolute values of estimated linear regression coefficients and variable importance scores from a random forest fitted to predict the response.
Table~\ref{tab:mscrt-statistic-catalogue} in Appendix~\ref{app:experimental-pvalue-conventions} lists several test statistics used in our numerical studies. 
\end{remark}

\section{MS-CRT constructions}\label{sec: construction}
We give two constructions of exchangeable covariate copies.
The multivariate normal method uses direct conditional sampling, whereas the graphical method uses the hub-and-spoke construction to combine local updates that preserve sufficient statistics. 

\subsection{Multivariate normal CRT}\label{sec:mvn-CRT}

Suppose $\bs{X} \sim \mathcal{N}_p(\bs{0}, \mathbf{\Sigma})$ with unknown covariance matrix $\mathbf{\Sigma}\succ\mathbf0$ (see Remark~\ref{rem: general mean for MVN} for the extension to unknown means). Let $t = |\mathcal{T}|$ and $s = |\mathcal{S}| = p - t$.
Suppose \(s<n\). 
The quantities $\mathbf X_{\mc T}^{\top}\mathbf X_{\mc S}$ and $\mathbf X_{\mc T}^{\top}\mathbf X_{\mc T}$ form a sufficient statistic 
\(\Phi_{\rm MVN}\)
for the unknown parameters in $\mathcal L(\mathbf X_{\mc T}\mid\mathbf X_{\mc S})$. 

To construct exchangeable copies, decompose $\mathbf X_{\mc T}$ into its projection onto $\operatorname{col}(\mathbf X_{\mc S})$ and a residual matrix whose columns lie in $\operatorname{col}(\mathbf X_{\mc S})^\perp$.
Conditional on $(\mathbf X_{\mc S},\Phi_{\rm MVN})$, the projection and the residual Gram matrix are fixed, while the residual matrix is uniformly distributed over all matrices with columns in this orthogonal complement and the same Gram matrix.
Therefore, we draw residual matrices from this conditional law and add back the fixed projection to obtain exchangeable copies.
This geometry has also been exploited in the multivariate rotation tests of \citet[Section 5.4]{langsrud_rotation_2005}. 

To introduce the sampling algorithm, let \(\mathbf{P}_{\mc{R}}\) be an $(n-s)\times n$ matrix whose rows form an orthonormal
basis of \(\operatorname{col}(\mathbf X_{\mc S})^\perp\);
its rows can be obtained by applying 
the Gram--Schmidt procedure to the rows of $\mathbf\Pi=\mathbf{I}_{n} - \mathbf{X}_{\mc{S}} \left(\mathbf{X}_{\mc{S}}^{\top} \mathbf{X}_{\mc{S}}\right)^{-1} \mathbf{X}_{\mc{S}}^{\top}$ with zero residuals discarded.
Here $\mathbf{X}_{\mc{S}}^{\top} \mathbf{X}_{\mc{S}}$ is invertible a.s. since $s<n$. 
Let \(r = \min\{n - s, t\}\), and define the matrix \(\mathbf{Q} \in \mathbb{R}^{r \times t}\) as
\begin{equation}\label{construction of Q}
\mathbf{Q} = \left\{ \begin{array}{cc}
  \left(\mathbf{X}_{\mc{T}}^{\top} \mathbf{P}_{\mc{R}}^{\top} \mathbf{P}_{\mc{R}} \mathbf{X}_{\mc{T}}\right)^{1/2}&n-s>t,\\
  \mathbf{P}_{\mc{R}} \mathbf{X}_{\mc{T}}&n-s\leq t,
\end{array}\right.
\end{equation}
where $\mathbf{A}^{1/2}$ denotes a positive definite square root.
In both cases, $\mathbf Q^\top\mathbf Q$ is the Gram matrix of the residual columns.
Algorithm~\ref{alg:multi gaussian CRT} outlines the \textit{multivariate normal CRT} (MVN-CRT), and Proposition~\ref{prop: multi gaussian CRT valid} establishes its validity. 

\begin{algorithm}[htbp]
  \caption{Multivariate normal CRT (MVN-CRT)}\label{alg:multi gaussian CRT}
  \begin{itemize}
\item \textbf{Input:} Response vector $\bs{Y}\in\bbR^n$, $n \times p$ covariate matrix $\mathbf{X}$, partition $\mc{T}\cup \mc{S}=[p]$, statistic $T$ as in Section~\ref{sec: MSCRT}, and number of randomised copies $M$. Suppose $n>s$, where $s=|\mc{S}|$.

\item \textbf{Step 1:}
Independently for each $j\in [M]$, sample $\mathbf{U}^{(j)}\in \bbR^{(n-s)\times r}$ as follows.
Sample $r$ independent vectors $\{\bs{V}^{(j)}_{1}, \ldots, \bs{V}^{(j)}_{r}\}$ from the standard $(n-s)$-dimensional normal distribution and apply the Gram--Schmidt process to obtain the columns of $\mathbf{U}^{(j)}$.

\item \textbf{Step 2:}
For each $j\in [M]$, define $\mathbf{\widetilde{X}}_{\mc{T}}^{(j)}=\mathbf{X}_{\mc{S}}\left(\mathbf{X}_{\mc{S}}^{\top}\mathbf{X}_{\mc{S}}\right)^{-1}\mathbf{X}_{\mc{S}}^{\top}\mathbf{X}_{\mc{T}}+\mathbf{P}_{\mc{R}}^{\top}\mathbf{U}^{(j)}\mathbf{Q}$. 
\item \textbf{Step 3:} Compute the scores $T_0,\ldots,T_M$ and output the p-value ${\rm pVal}_T$ in \eqref{eqn:pval-1}.
  \end{itemize}
\end{algorithm}

\begin{proposition}\label{prop: multi gaussian CRT valid}
Suppose the rows of $\mathbf{X}$ are i.i.d. observations sampled from $\mathcal{N}_p(\bs{0}, \mathbf{\Sigma})$ and $n>s$.
For the copies $\widetilde{\mathbf{X}}_{\mc{T}}^{(j)}$ ($j=1,\ldots,M$) generated by Algorithm~\ref{alg:multi gaussian CRT}, the following properties hold:
\\*
\textup{1.}\enspace ${\mathbf{X}}_{\mc{T}}^{\top}{\mathbf{X}_{\mc{S}}}=\left(\widetilde{\mathbf{X}}_{\mc{T}}^{(j)}\right)^{\top}{\mathbf{X}_{\mc{S}}}$ and ${\mathbf{X}_{\mc{T}}^{\top}}{\mathbf{X}_{\mc{T}}}=\left(\widetilde{\mathbf{X}}_{\mc{T}}^{(j)}\right)^{\top}\widetilde{\mathbf{X}}_{\mc{T}}^{(j)}$ for $j=1,\ldots,M$.\\*
\textup{2.}\enspace Conditional on \((\mathbf X_{\mc S},\Phi_{\rm MVN})\), the sequence \((\mathbf X_{\mc T},\allowbreak \widetilde{\mathbf X}_{\mc T}^{(1)},\allowbreak \ldots,\allowbreak \widetilde{\mathbf X}_{\mc T}^{(M)})\) is exchangeable.\\*
\textup{3.}\enspace 
Under $H_0:\bs Y\indp\mathbf X_{\mc T}\mid\mathbf X_{\mc S}$, the p-value in \eqref{eqn:pval-1} is valid by Proposition~\ref{prop: CRT valid}.
\end{proposition}
\begin{remark}\label{rem: general mean for MVN}
For unknown means, add a dummy index $0$ to $\mc S$ and the corresponding column $\bs1_n$ to $\mathbf X_{\mc S}$. The same construction then applies for $n>s+1$. 
\end{remark}

\subsection{Gaussian graphical models}
\label{sec:ggm-gcrt}\label{sec:ggm}
When the number of background covariates $s$ approaches or exceeds $n$, MVN-CRT has little or no room for randomisation because it uses the orthogonal complement of all background columns. 
In a Gaussian graphical model (GGM), we can randomise each target column one at a time by conditioning on its neighbourhood. 
If these neighbourhoods are small relative to $n$, we can apply these local updates to randomise $\mathbf X_{\mc T}$ even when $s>n$. 

For a fixed undirected graph $G=([p],\mc E)$, suppose the rows of $\mathbf X$ are i.i.d. with distribution in
\begin{equation}\label{eq: model MG}
\mathcal M_G=\bigl\{\mathcal{N}_p(\bs\mu,\mathbf\Omega^{-1}):\;
\bs\mu\in\mathbb R^p,\ \mathbf\Omega\succ\mathbf0,\
\omega_{ij}=0\ \text{if }i\ne j\text{ and }(i,j)\notin\mc E\bigr\}.
\end{equation}
The conditional independence graph of a Gaussian distribution has an edge $(i,j)$ exactly when $i\ne j$ and $\omega_{ij}\ne0$. 
The entries of $\mathbf\Omega$ corresponding to edges of $G$ are unrestricted, so $G$ may contain extra edges.
A sufficient statistic is
$\psi_G(\mathbf X)=\Bigl(\sum_{r=1}^n\bs X_{r\cdot},\,
\bigl((\mathbf X^\top\mathbf X)_{ij}\bigr)_{i=j\text{ or }(i,j)\in\mc E}\Bigr)$.
The law of $\mathbf X_{\mc T}$ conditional on $(\mathbf X_{\mc S}=\mathbf x_{\mc S},\psi_G(\mathbf X)=\psi)$ is parameter-free and supported on
\begin{equation}\label{eq:GGM_conditional_space}
\mathcal X_{\mc T,\mathbf x_{\mc S},\psi}
=\{\mathbf x_{\mc T}\in\mathbb R^{n\times|\mc T|}:
\psi_G([\mathbf x_{\mc T},\mathbf x_{\mc S}])=\psi\}.
\end{equation}
Here each column in the concatenation has its original index. 
The coupled linear and quadratic constraints in \eqref{eq:GGM_conditional_space} generally make direct sampling from this conditional law difficult.
We instead generate exchangeable copies by combining local residual rotations through the hub-and-spoke construction of \citet[Section~2]{besag_generalized_1989}.
This sampler was used by \citet{lin_goodness--fit_2025} to test the goodness of fit of a given GGM; here we assume a GGM for the covariates and use the sampler to test MCI.

\subsubsection{Residual rotation}

For a target node $i$, let $N_i$ be its neighbours in $G$. The local Markov property gives
$\mathbf X_i\mid\mathbf X_{-i}\sim
\mathcal{N}_n(\alpha_i\bs1_n+\mathbf X_{N_i}\bs\beta_i,\sigma_i^2\mathbf I_n)$.
Here $(\alpha_i,\bs\beta_i,\sigma_i^2)$ are nuisance parameters with $\sigma_i^2>0$. 
Let $\mathbf P_i$ be the orthogonal projection onto $\operatorname{col}([\bs1_n,\mathbf X_{N_i}])$, and define $\bs F_i=\mathbf P_i\mathbf X_i$ and $\bs R_i=(\mathbf I_n-\mathbf P_i)\mathbf X_i$. 
A sufficient statistic for $(\alpha_i,\bs\beta_i,\sigma_i^2)$ is $\Psi_i(\mathbf X_i;\mathbf X_{N_i})=(\bs F_i,\|\bs R_i\|_2)$.
The \emph{residual rotation} of
\citet[Algorithm~1]{lin_goodness--fit_2025} updates column
$\mathbf X_i$ as follows.
Draw a fresh $\bs W\sim\mathcal{N}_n(\bs0,\mathbf I_n)$ and set $\widetilde{\bs R}_i=(\mathbf I_n-\mathbf P_i)\bs W$. 
Replace column $i$ of the current matrix by
\begin{equation}\label{eq:ggm-local-update}
\mathbf X_i^{\mathrm{new}}=
\begin{cases}
\bs F_i+\widetilde{\bs R}_i\dfrac{\|\bs R_i\|_2}{\|\widetilde{\bs R}_i\|_2},
& \|\widetilde{\bs R}_i\|_2>0,\\
\mathbf X_i, & \|\widetilde{\bs R}_i\|_2=0.
\end{cases}
  \end{equation}
The updated column $\mathbf X_i^{\mathrm{new}}$ has the same fitted value $\bs F_i$ and residual norm $\|\bs R_i\|_2$ as $\mathbf X_i$, so the local sufficient statistic $\Psi_i$ is unchanged.
With the other columns fixed, $\Psi_i$ determines every entry of $\psi_G$ involving column $i$, so the global sufficient statistic is also unchanged.
Conditional on $(\mathbf X_{-i},\Psi_i)$,
the update is reversible.
Lemma~\ref{lem:ggm-local-update} in Appendix~\ref{app:residual-rotation}
establishes these properties.

The randomisation takes place in the residual space $\operatorname{col}(\mathbf I_n-\mathbf P_i)$,
whose dimension is $d_i=n-\operatorname{rank}([\bs1_n,\mathbf X_{N_i}])$.
Under $\mathcal M_G$, $d_i=\max\{n-|N_i|-1,0\}$ a.s.
When $d_i=0$, both $\bs R_i$ and $\widetilde{\bs R}_i$ are zero; when $d_i>0$, both are nonzero a.s.
For a nonzero residual, $d_i=1$ allows only a change of sign,
whereas $d_i\geq2$ allows continuous randomisation of its direction.
In particular, if $|N_i|\leq n-3$ for all $i\in\mc T$,
all target columns admit continuous randomisation even when $p\gg n$.

\subsubsection{Exchangeable copies and the $G$-CRT}

We combine the residual rotations to generate copies of $\mathbf X_{\mc T}$ that are jointly exchangeable with the observed target matrix conditional on $(\mathbf X_{\mc S},\psi_G(\mathbf X))$.
Starting from $\mathbf X$, we apply a prescribed sequence of local updates to obtain an intermediate matrix, called the \emph{hub}.
Each copy is then generated from this same hub by applying the local updates in reverse order with fresh independent randomness.
The procedure is summarised in Algorithm~\ref{alg: exchangeable}. 

\begin{algorithm}[htbp]
\caption[Exchangeable copies for $G$-CRT via the hub-and-spoke construction]{Exchangeable copies for $G$-CRT via the hub-and-spoke construction}\label{alg: exchangeable}
\begin{itemize}
\item \textbf{Input:} Covariate matrix $\mathbf X$, graph $G$, partition $\mc T\cup\mc S=[p]$, number of copies $M\geq1$, number of sweeps $L\geq1$, and an update order $\mc I$ with entries in $\mc T$.
\item \textbf{Require:} Use the residual rotation update in \eqref{eq:ggm-local-update}, or other local updates satisfying Condition~\ref{cond:local-update} in Appendix~\ref{app:general-graphical-updates}. 
\item \textbf{Step 1:} Starting from $\mathbf X$, apply the local updates in order $\mc I$ for $L$ sweeps to obtain $\mathbf X^{(\mathrm{hub})}$.
\item \textbf{Step 2:} For each $m\in[M]$, start a separate branch at $\mathbf X^{(\mathrm{hub})}$ and independently apply the same updates in the reverse order for $L$ sweeps to obtain $\widetilde{\mathbf X}^{(m)}$.
\item \textbf{Output:} $\widetilde{\mathbf X}^{(1)},\ldots,\widetilde{\mathbf X}^{(M)}$.
\end{itemize}
\end{algorithm}

In Algorithm~\ref{alg: exchangeable}, the graph $G$, partition $(\mc T,\mc S)$, update order $\mc I$, number $L$ of sweeps (complete passes through $\mc I$), and number $M$ of copies are fixed independently of the inference data $(\bs Y,\mathbf X)$.
At each local update, $\mathbf P_i$, $\bs F_i$, and $\bs R_i$ are recomputed from the current matrix.

Each local update is reversible with respect to $\mc{L}(\mathbf X\mid\mathbf X_{\mc S},\psi_G(\mathbf X))$. 
Therefore, conditional on the hub, the initial matrix $\mathbf{X}$ has the same distribution as a matrix obtained by applying the local updates of Step~1 in the reverse order from the hub. 
Each copy $\widetilde{\mathbf X}^{(m)}$ in Step~2 is constructed in this way, so $\mathbf{X}$ and all copies are i.i.d. conditional on the hub. 
Averaging over the hub gives the exchangeability in Proposition~\ref{prop:ggm-gcrt-valid}. 

\begin{proposition}\label{prop:ggm-gcrt-valid}
Suppose the rows of $\mathbf X$ are i.i.d. from $\mathcal M_G$. Apply Algorithm~\ref{alg: exchangeable} with the residual rotations in \eqref{eq:ggm-local-update}. Then:
\\*
\textup{1.}\enspace For each $m\in[M]$, the copy $\widetilde{\mathbf X}^{(m)}$ satisfies $\widetilde{\mathbf X}^{(m)}_{\mc S}=\mathbf X_{\mc S}$ and $\psi_G(\widetilde{\mathbf X}^{(m)})=\psi_G(\mathbf X)$.\\*
\textup{2.}\enspace Conditional on $(\mathbf X_{\mc S},\psi_G(\mathbf X))$, the complete tuple $(\mathbf X_{\mc T},\widetilde{\mathbf X}^{(1)}_{\mc T},\ldots,\widetilde{\mathbf X}^{(M)}_{\mc T})$ is exchangeable.\\*
\textup{3.}\enspace Under $H_0:\bs Y\indp\mathbf X_{\mc T}\mid\mathbf X_{\mc S}$, the p-value in \eqref{eqn:pval-1} is valid by Proposition~\ref{prop: CRT valid}.
\end{proposition}

The exchangeability holds for any $L\geq1$ without any mixing assumption. 
We call the resulting test the \emph{$G$-CRT}. 

\subsubsection[Graph requirement for $G$-CRT]{Graph requirement for $G$-CRT}\label{sec:graph-choice}\label{sec: G-CRT assumption}
Since only $\mathbf X_{\mc T}$ is updated, the assumption
$\mathcal L(\bs X)\in\mathcal M_G$ can be weakened.
For Gaussian covariates, it suffices to require the local Markov property
only at the target nodes:
\begin{equation}\label{eq:local Markov property}
\forall i\in\mc T, ~~ X_i\indp\bs X_{-(N_i\cup\{i\})}\mid\bs X_{N_i}.
\end{equation}
Therefore, omitted true edges within $\mc{S}$ do not affect the local
updates or the validity of $G$-CRT (see
Corollary~\ref{cor:ggm-weaker-graph-valid} in
Appendix~\ref{app:general-graphical-updates}).

Furthermore, \eqref{eq:local Markov property} holds if $G$ is a supergraph of the conditional independence graph. 
Extra edges incident to a target node $i$ preserve validity but increase
$|N_i|$ and can reduce the residual dimension $d_i$.
Appendix~\ref{app: CRT super} examines the effect of extra edges on power.

Graph information may come from domain knowledge, independent auxiliary covariate data, or sample splitting. Conditional on the auxiliary data, the selected graph and test specification are fixed, and validity requires the selected neighbourhoods to satisfy \eqref{eq:local Markov property}. Section~\ref{sec:ggm-robustness-rewrite} quantifies how errors in the
estimated graph affect Type I error.

\subsection{General local updates for graphical models}\label{sec:general-graphical-models}

To apply Algorithm~\ref{alg: exchangeable} to general graphical models where \eqref{eq:local Markov property} holds, we require for each target node $i$ a sufficient statistic $\Psi_i(\mathbf X_i;\mathbf X_{N_i})$ for the nuisance parameters of $\mc{L}(\mathbf X_i \mid \mathbf X_{N_i})$ and a local update that does not depend on these parameters. 
The update must preserve $\Psi_i$ and make $(\mathbf X_i,\mathbf X_i^{\mathrm{new}})$ exchangeable conditional on $(\mathbf X_{-i},\Psi_i)$, as the residual rotation does in Section~\ref{sec:ggm}.
Appendix~\ref{app:general-graphical-updates} states the precise conditions and conclusions. 

For a target node $i$ taking values in a finite set, a local update is given as follows. Group the sample rows by identical values of its neighbours. Within each group, conditional on the unchanged columns and the counts of the node's values, all arrangements with those counts are equally likely, regardless of the unknown conditional probabilities. 
Therefore, we can obtain an updated column exchangeable with the original column by drawing independent uniform permutations within the groups. Example~\ref{ex:gcrt-finite-state} in Appendix~\ref{app:general-graphical-updates} gives the full construction.

\section[Power analysis]{Power analysis}\label{sec:power}
We study the power of MVN-CRT and $G$-CRT.
For a tractable analysis, we assume the response follows a Gaussian linear model 
\begin{equation}\label{linear model}
    {Y} \mid {\bs{X}} \sim \mathcal{N}\left({\bs{X}}_{\mc{T}}^{\top}\bs{\beta}_{\mc{T}} + {\bs{X}}_{\mc{S}}^{\top}\bs{\beta}_{\mc{S}}, \sigma^2\right),
\end{equation}
where $\bs{\beta}_{\mc{T}}\in \bbR^t$, $\bs{\beta}_{\mc{S}}\in \bbR^s$, and $\sigma>0$.
This additional assumption is not needed for the validity of our tests, but it enables the theoretical analysis of the power under alternatives.  
For the model in \eqref{linear model}, if
$\mathcal L(\bs X_{\mc T}\mid\bs X_{\mc S})$ is nondegenerate,
testing \eqref{eq: CI-T} is equivalent to testing $H_0: \bs{\beta}_{\mc{T}} = \bs{0}$.

Section~\ref{sec: MVN-CRT power} shows that MVN-CRT attains minimax separation rates. Section~\ref{sec:gcrt-ridge-power} gives sufficient conditions for $G$-CRT to have power tending to one. 
These results demonstrate that 
our tests remain powerful after conditioning on sufficient statistics. 

\subsection[Minimax rates for MVN-CRT]{Minimax rates for MVN-CRT}\label{sec: MVN-CRT power}

For the covariate model $\bs{X}\sim\mathcal{N}_p(\bs{0},\mathbf{\Sigma})$, partition $\mathbf{\Sigma}$ according to $(\mc{T},\mc{S})$. Then
\begin{equation}\label{null of X and Z}
  \bs{X}_{\mc{S}} \sim \mathcal{N}_s(\bs{0}, \mathbf{\Sigma}_{\mc{S}\mc{S}}) \quad \text{and} \quad \bs{X}_{\mc{T}} \mid \bs{X}_{\mc{S}} \sim \mathcal{N}_t(\mathbf{C} \bs{X}_{\mc{S}}, \mathbf{\Sigma}_{\mc{T}\mid \mc{S}}),
\end{equation}
where $\mathbf{C} = \mathbf{\Sigma}_{\mc{T}\mc{S}} \mathbf{\Sigma}_{\mc{S}\mc{S}}^{-1} \in \mathbb{R}^{t \times s}$ and $\mathbf{\Sigma}_{\mc{T}\mid \mc{S}} = \mathbf{\Sigma}_{\mc{T}\mc{T}} - \mathbf{\Sigma}_{\mc{T}\mc{S}} \mathbf{\Sigma}_{\mc{S}\mc{S}}^{-1} \mathbf{\Sigma}_{\mc{S}\mc{T}}$.
Under the covariate model in \eqref{null of X and Z} and the response model in \eqref{linear model}, 
the parameter is $\theta=(\bs\beta_{\mc T},\bs\beta_{\mc S},\mathbf\Sigma,\sigma)$ and the parameter space is 
$$
\begin{aligned}
\Theta=\bigl\{\theta:\;&\bs\beta_{\mc T}\in\mathbb R^t,\ \bs\beta_{\mc S}\in\mathbb R^s,\ M_0^{-1}\leq\lambda_{\min}(\mathbf\Sigma)\leq\lambda_{\max}(\mathbf\Sigma)\leq M_0,\\
&0<\sigma\leq M_1,\ \|\bs\beta_{\mc T}\|_2^2+\|\bs\beta_{\mc S}\|_2^2\leq M_2\bigr\},
\end{aligned}
$$
where $M_0>1$ and $M_1,M_2>0$ are universal constants. The null parameter subspace is $\Theta_0=\{\theta\in\Theta:\bs\beta_{\mc T}=\bs0\}$.
The bounds on the eigenvalues of $\mathbf\Sigma$ and on $\sigma$ are regularity assumptions on the design and noise level used in related work \citep{10.1214/16-AOS1461, 10.1214/17-AOS1604}.
Let $ \bs{\beta} = (\bs{\beta}_{\mathcal{T}}^{\top}, \bs{\beta}_{\mathcal{S}}^{\top})^{\top} \in \mathbb{R}^{p \times 1} $.
Since $\operatorname{Var}(Y_i) = \sigma^2 + \bs{\beta}^{\top} \mathbf{\Sigma} \bs{\beta}$, the bound on the norm of $\bs{\beta}$ can be interpreted as a constraint on the variance of \(Y_i\), which is a common assumption.

The distribution of $\bs{X}$ in \eqref{null of X and Z} and the conditional distribution of $Y$ in \eqref{linear model} imply that
$$\bbE_\theta\left((\bs{X}_{\mc{T}}-\bbE_\theta(\bs{X}_{\mc{T}}\mid \bs{X}_{\mc{S}}))(Y-\bbE_\theta(Y\mid \bs{X}_{\mc{S}}))\mid \bs{X}_{\mc{S}}\right)=\mathbf{\Sigma}_{\mc{T}\mid \mc{S}}\bs{\beta}_{\mc{T}}.$$
We measure separation from the null by the $\ell_2$ and
$\ell_\infty$ norms of $\mathbf{\Sigma}_{\mc{T}\mid \mc{S}}\bs{\beta}_{\mc{T}}$.

\textbf{$\ell_2$ separation.}
Define the alternative space $\Theta_2(h)=\{\theta\in\Theta:\|\mathbf\Sigma_{\mc T\mid\mc S}\bs\beta_{\mc T}\|_2\geq h\}$.
For separation in the $\ell_2$ norm, suppose $n>s+t$ and define $$
T_{\ell_2}(\bs Y,\mathbf X_{\mc T},\mathbf X_{\mc S})=
\|\mathbf P_{\mc R}\bs Y\|_2^{-2}
\left\|\bs Y^\top\mathbf P_{\mc R}^\top\mathbf P_{\mc R}\mathbf X_{\mc T}
\left(\mathbf X_{\mc T}^\top\mathbf P_{\mc R}^\top\mathbf P_{\mc R}\mathbf X_{\mc T}\right)^{-1/2}\right\|_2^2.
$$
The statistic $T_{\ell_2}$ is the $R^2$ from regressing the response residual on the target residuals after removing the background covariates. 
If the response residual is zero, define the observed value and all copy values of $T_{\ell_2}$ to be zero.
Under the Gaussian linear model~\eqref{linear model}, MVN-CRT with $T_{\ell_2}$ is a Monte Carlo version of the partial $F$-test.

\textbf{$\ell_\infty$ separation.}
Define the alternative space $\Theta_\infty(h)=\{\theta\in\Theta:\|\mathbf\Sigma_{\mc T\mid\mc S}\bs\beta_{\mc T}\|_\infty\geq h\}$.
For separation in the $\ell_\infty$ norm, we consider the \textbf{MaxCor} statistic
\begin{equation}\label{eq:mvn-maxcor-statistic}
T_{\ell_\infty}(\bs{Y},\mathbf{X}_{\mc{T}},\mathbf{X}_{\mc{S}})=\max_{j\in \mc{T}} \left|\frac{\bs{Y}^{\top}\mathbf{P}_{\mc{R}}^{\top} \mathbf{P}_{\mc{R}} \mathbf{X}_j}{\left\|\mathbf{P}_{\mc{R}}\bs{Y}\right\|_2\|\mathbf{P}_{\mc{R}}\mathbf{X}_j\|_2}\right|,
\end{equation}
which is the maximum absolute normalised inner product between the residual of $\bs Y$ and the target residuals after regression on $\mathbf X_{\mc S}$, with no additional subtraction of their sample means.

We consider an asymptotic regime in which $n\to\infty$, the dimensions $s$ and $t$ may vary with $n$, 
and the separation $h=h_n$ may also depend on $n$.
Theorem~\ref{power analysis} characterises the asymptotic power of MVN-CRT over these two alternative spaces.
For a test statistic $T$ and level $\alpha$, define the MVN-CRT rejection indicator $\phi_{T,\alpha}=\mathbf{1}\{\mathrm{pVal}_T\leq\alpha\}$.

\begin{theorem}\label{power analysis}
Fix $0<\alpha<\beta<1$ and suppose $s\leq c_1n$ for a constant $c_1\in(0,1)$. Let $M$ be an integer with $M\geq2/\alpha$.
  \\*
\textup{1.}\enspace \textbf{($\ell_2$ separation)} Suppose $t\leq c_2n$ for a constant $c_2\in(0,1-c_1)$.
    With the test statistic $T_{\ell_2}$, we have
    $$\liminf_{n\to\infty}\inf_{\theta\in {\Theta}_2(h_{1,n})}\bbE_\theta \phi_{T_{\ell_2},\al}\geq \beta,$$
   where \(h_{1,n}=C_1t^{1/4}/\sqrt n\) for a constant \(C_1>0\) depending only on $(\alpha,\beta,M_0,M_1,M_2,c_1,c_2)$.\\*
\textup{2.}\enspace \textbf{($\ell_\infty$ separation)} Suppose $\log t = o(n)$. With the test statistic $T_{\ell_\infty}$, we have
    $$\liminf_{n\to\infty}\inf_{\theta\in {\Theta}_\infty(h_{2,n})}\bbE_\theta \phi_{T_{\ell_\infty}, \alpha}\geq \beta,$$
    where \(h_{2,n}=C_2\sqrt{\log(et)/n}\) for a constant \(C_2>0\) depending only on $(\alpha,\beta,M_0,M_1,M_2,c_1)$.
\end{theorem}

Under the respective parameter and dimension conditions of Theorem~\ref{power analysis}, the power of the MVN-CRT converges to one for alternatives with $\|\mathbf{\Sigma}_{\mc T\mid\mc S}\bs{\beta}_{\mc T}\|_2\gg t^{1/4}/\sqrt n$, and for alternatives with $\|\mathbf{\Sigma}_{\mc T\mid\mc S}\bs{\beta}_{\mc T}\|_\infty\gg\sqrt{\log(et)/n}$.
The next theorem gives matching lower bounds for the two separation rates in Theorem~\ref{power analysis}.
For $0<\alpha<1$, define the class of all level-$\alpha$ tests for $\Theta_0$ by
$$
\Phi_\alpha(\Theta_0) = \left\{\psi: (\bs{Y}, \mathbf{X}) \to [0,1] \mid \sup_{\theta \in \Theta_0} \mathbb{E}_\theta [\psi] \leq \alpha \right\}.
$$

\begin{theorem}\label{lower bound}
Fix $0<\alpha<\beta<1$.
\\*
\textup{1.}\enspace \textbf{($\ell_2$ separation)} Suppose $t\le C_1n^2$ for a constant $C_1>0$. There exists $c_3>0$, depending only on $(\alpha,\beta,M_0,M_1,M_2,C_1)$, such that, with $h_{3,n}=c_3t^{1/4}/\sqrt n$,
\begin{equation}\label{eq:lower-dense-global}
 \limsup_{n\to\infty}\ \sup_{\psi\in\Phi_\alpha(\Theta_0)}
 \ \inf_{\theta\in\Theta_2(h_{3,n})}
 \mathbb E_\theta\psi\le\beta.
\end{equation}
\noindent\textup{2.}\enspace \textbf{($\ell_\infty$ separation)} Suppose $\log t\le C_2n$ for a constant $C_2>0$. There exists $c_4>0$, depending only on $(\alpha,\beta,M_0,M_1,M_2,C_2)$, such that, with $h_{4,n}=c_4\sqrt{\log(et)/n}$,
\begin{equation}\label{eq:lower-sparse-global}
 \limsup_{n\to\infty}\ \sup_{\psi\in\Phi_\alpha(\Theta_0)}
 \ \inf_{\theta\in\Theta_\infty(h_{4,n})}
 \mathbb E_\theta\psi\le\beta.
\end{equation}
\end{theorem}

Theorems~\ref{power analysis} and~\ref{lower bound} establish that
MVN-CRT attains the minimax separation rates $t^{1/4}/\sqrt n$ and
$\sqrt{\log(et)/n}$ for $\ell_2$ and $\ell_\infty$ separation,
respectively, under the corresponding conditions of
Theorem~\ref{power analysis}.
Related lower bounds for Gaussian regression are given by \citet{verzelen_goodness--fit_2010}.
We discuss their connection with our results in Appendix~\ref{app:regression-ci-comparison}.

\subsection[Power of G-CRT]{Power of $G$-CRT}\label{sec:gcrt-ridge-power}
We establish sufficient conditions for uniform consistency
of $G$-CRT under Gaussian linear alternatives, allowing the
number of background covariates to exceed the sample size.

Assume that the $n$ observations are independent and follow the Gaussian linear response model in \eqref{linear model}, with an unrestricted intercept and covariate distribution in $\mathcal M_G$ from \eqref{eq: model MG}.
Let $t=|\mathcal T|\geq2$ and $s=p-t\geq1$. 
Ridge regression allows adjustment for the background covariates even when $s\geq n$.
We consider symmetrically centred ridge residual scores, constructed as follows.

Fix a background penalty
constant $\ell>0$ and a positive deterministic target penalty $\lambda_n$.
Regress the response and every candidate target column on all background
covariates, with an unpenalised intercept and penalty $\ell n$.
Compute the response residual vector $\bs Y_R$ once and share it across
candidates. Let $\mathbf X_R^{(v)}\in\mathbb R^{n\times t}$ collect the
target residuals for candidate $v=0,\ldots,M$, where $v=0$ corresponds to the observed target data. Define
\begin{equation}\label{gcrtcent:eq:residual-scores}
 \bs c_v=\mathbf X_R^{(v)\top}\bs Y_R/n,\qquad
 \mathbf D_v=\mathbf X_R^{(v)\top}\mathbf X_R^{(v)}/n
                  +(\lambda_n/n)\mathbf I_t,\qquad
 \mathbf D_{\rm pool}=\frac1{M+1}\sum_{v=0}^M\mathbf D_v.
\end{equation}
Here $\bs c_v$ is the residual score vector, $\mathbf D_v$ is the regularised residual Gram matrix, and $\mathbf D_{\rm pool}$ is the average of these matrices over all candidates.
The leading expansion of $\bs c_v$ contains a background contribution
common to all candidates. To remove this contribution, define the centre and scores
\begin{equation}\label{eq:symmetric-centered-statistic}
 \bar{\bs c}=\frac1{M+1}\sum_{v=0}^M\bs c_v,\qquad
 T_v^{\rm cen}=(\bs c_v-\bar{\bs c})^\top\mathbf W(\bs c_v-\bar{\bs c}),
 \qquad v=0,\ldots,M,
\end{equation}
where $\mathbf W\in\mathbb R^{t\times t}$ is a common measurable positive definite matrix that is unchanged by permuting the candidates.
We consider the choices $\mathbf W=\mathbf I_t$ and $\mathbf W=\mathbf D_{\rm pool}^{-1}$. 
The latter inverse pooled Gram weight accounts for the scale and correlation of the target residuals.
Write $p_{\rm cen}$ for the p-value in \eqref{eqn:pval-1} with the scores in \eqref{eq:symmetric-centered-statistic}.
Under the sampling and null conditions of Proposition~\ref{prop: CRT valid},
the scores remain exchangeable and $p_{\rm cen}$ is valid in finite samples. 

The common centre and weight treat all candidates symmetrically.
The effect of centring on CRT power is discussed and illustrated in simulations by 
\citet{liu_fast_2022}. 

The power conditions involve population quantities that need not be known to compute the test. We use $\mathbf C$ and $\mathbf\Sigma_{\mathcal T\mid\mathcal S}$ as defined in Section~\ref{sec: MVN-CRT power}. Let $\bs\gamma=\bs\beta_{\mathcal S}+\mathbf C^\top\bs\beta_{\mathcal T}$ be the coefficient vector of $\bs X_{\mathcal S}$ in $\mathbb E(Y\mid\bs X_{\mathcal S})$. Let $\mathbf\Theta=\mathbf\Sigma_{\mathcal T\mid\mathcal S}^{-1}$ and define $\mathbf J=(J_{jk})_{j,k\in\mathcal T}$ by $J_{jj}=0$ and $J_{jk}=-\Theta_{jk}/\Theta_{jj}$ for $j\ne k$. Here $J_{jk}$ is the coefficient of $X_k$ in the regression of $X_j$ on the other target and background covariates.

The classes of alternatives are defined as follows. Fix positive constants 
$v_\pm,s_\pm, \sigma_\pm, \bar{c}, c_+,\allowbreak G_+,R$
and a constant $a\in[0,1)$. 
For a positive deterministic sequence $h_n$, let $\mathcal C_n(h_n)$ be the class of the Gaussian models (indexed by $\vartheta$) satisfying
\begin{equation}\label{gcrtcent:eq:population-conditions}
\begin{aligned}
&  1\leq s\leq c_+n,  \qquad v_-\mathbf I_t\preceq\mathbf\Sigma_{\mathcal T\mid\mathcal S}\preceq v_+\mathbf I_t,
 \qquad
 s_-\mathbf I_s\preceq\mathbf\Sigma_{\mathcal S\mathcal S}
                    \preceq s_+\mathbf I_s,\\
 &\|\mathbf C\|_{\rm op}\leq\bar c,
 \qquad \max\{\|\mathbf J\|_1,\|\mathbf J\|_\infty\}\leq a,
 \qquad \sigma_-\leq\sigma\leq\sigma_+,\\
 &\|\bs\gamma\|_2/\sigma\leq G_+,
 \qquad h_n^2\leq
 \frac{\bs\beta_{\mathcal T}^\top\mathbf\Sigma_{\mathcal T\mid\mathcal S}\bs\beta_{\mathcal T}}{\sigma^2}
 \leq R^2.
\end{aligned}
\end{equation}
In \eqref{gcrtcent:eq:population-conditions}, $h_n$ is a lower bound on the standardised target signal strength. 

\begin{theorem}\label{thm:gcrt-centered-power}
Fix $\ell,K>0$, integers $L\geq1$ and $M\geq2$, and $\alpha\in(0,1)$ with $1/(M+1)\leq\alpha$. Let $\lambda_n$ be a deterministic sequence satisfying $0<\lambda_n\leq Kn$. Generate copies using Algorithm~\ref{alg: exchangeable} with the residual rotations in \eqref{eq:ggm-local-update}, updating each target once per sweep. If
\begin{equation}\label{gcrtcent:eq:rates}
 \sup_{\vartheta\in\mathcal C_n(h_n)}
 \frac{\max_{j\in\mathcal T}|N_j\cap\mathcal S|}{n}
 \longrightarrow0,
 \qquad
 \sup_{\vartheta\in\mathcal C_n(h_n)}
 \frac{\sum_{j\in\mathcal T}|N_j\cap\mathcal T|+t}{nh_n^2}
 \longrightarrow0,
\end{equation}
then, for either $\mathbf W=\mathbf D_{\rm pool}^{-1}$ or $\mathbf W=\mathbf I_t$, $G$-CRT using the statistic in \eqref{eq:symmetric-centered-statistic} satisfies
\begin{equation}\label{gcrtcent:eq:power}
 \inf_{\vartheta\in\mathcal C_n(h_n)}
 \mathbb P_\vartheta(p_{\rm cen}\leq\alpha)\longrightarrow1.
\end{equation}
\end{theorem}

The conditions in \eqref{gcrtcent:eq:rates} ensure that the approximation errors in the residual score vectors are negligible relative to the target signal.
The first limits the number of background covariates used in each local update relative to $n$, which allows the background dimension $s$ to exceed $n$.
The second requires stronger signals for larger or more densely connected target groups.
To illustrate the scope, Appendix~\ref{app:gcrt-ridge-example} gives an explicit family satisfying the conditions of Theorem~\ref{thm:gcrt-centered-power} with five connected targets, $s=2n$, and a dense background coefficient vector $\bs\gamma$.

Symmetric centring allows the uniform consistency result in Theorem~\ref{thm:gcrt-centered-power} to hold for any direction of $\bs\gamma$ satisfying $\|\bs\gamma\|_2\leq \sigma G_+$.
With the uncentred ridge statistic, the background and target contributions can cancel asymptotically, and the power can tend to zero even for a fixed nonzero target signal (Appendix~\ref{sec:gcrt-ridge-comparison}).

The difficulty in proving Theorem~\ref{thm:gcrt-centered-power}
is to establish $T_0^{\rm cen}>T_m^{\rm cen}$ ($m=1,\ldots,M$) for the dependent
candidates generated by Algorithm~\ref{alg: exchangeable}.
The projections and rescaling in \eqref{eq:ggm-local-update}
depend on the current data, and the generated copies are not jointly Gaussian given $\mathbf X_{\mathcal S}$.
To overcome this difficulty, we couple the residual rotation
construction to an oracle Gaussian construction that uses
population parameters and follows the same scan order and hub
structure. The resulting comparison copies are jointly Gaussian
given $\mathbf X_{\mathcal S}$; their cross moments with the
original residuals identify the target signal contribution
in each copy. A contraction bound yields a positive lower bound for the
leading term in $T_0^{\rm cen}-T_m^{\rm cen}$.

To transfer this separation to the actual test, we evaluate
the residual rotation and the update of the oracle construction
at the same actual state and bound their difference. We then
control how these local errors propagate and accumulate through
the sweeps and branches. 
We combine these error bounds with concentration bounds for the comparison scores to show that the coupling error is negligible relative to the positive leading term.
Consequently,
$T_0^{\rm cen}$ exceeds $\max_{1\leq m\leq M}T_m^{\rm cen}$ with probability tending to one.

\section{Robustness to covariate model misspecification}\label{sec: ROBT}

We now consider departures from the working covariate model and establish a general robustness bound for MS-CRT.
 We then apply the MS-CRT bound to Gaussian graphical models to quantify the effect of errors in the selected graph.
We use $\mathbb P_*$, $\mathbb E_*$, and $\mathcal L_*$ for probabilities, expectations, and distributions under the true data law, including the test's auxiliary randomness.

\subsection{Robustness of MS-CRT}\label{sec:mscrt-robustness}

MS-CRT uses the conditional covariate law given the background covariates and a sufficient statistic. Covariate laws that agree under this conditioning permit the same exchangeable copy construction, so they provide an enlarged comparison class for measuring the effect of misspecification.

At a fixed $\mathbf x_{\mc S}$, let ${\cal P}_\Theta(\mathbf x_{\mc S})=\{P_{\theta,\mathbf x_{\mc S}}:\theta\in\Theta\}$ be the working family for $\mathbf X_{\mc T}\mid\mathbf X_{\mc S}=\mathbf x_{\mc S}$, with unknown parameter $\theta$.
Let $\Phi=\phi(\mathbf X_{\mc T};\mathbf x_{\mc S})$ be the sufficient statistic used by MS-CRT for this working family. Let $\mc V(\mathbf x_{\mc S})$ be a measurable set satisfying $\Phi\in\mc V(\mathbf x_{\mc S})$ almost surely under every working law. Assume that there is a common conditional kernel $L_{\mathbf x_{\mc S},\varphi}$ for $\mathbf X_{\mc T}$ given $\Phi=\varphi$ under these laws, defined for $\varphi\in\mc V(\mathbf x_{\mc S})$ and not depending on $\theta$.

For any law \(R\) of \(\mathbf X_{\mc T}\) conditional on 
\(\mathbf X_{\mc S}=\mathbf x_{\mc S}\), let \(\mathbf Z\sim R\) and set
\(U=\phi(\mathbf Z;\mathbf x_{\mc S})\). Define
\begin{equation}
\label{eq:Phi-comparison-class-at-fixed-XS}
\widetilde{\mc P}_{\Phi}(\mathbf x_{\mc S})
=
\left\{
R:
U\in\mc V(\mathbf x_{\mc S})\text{ almost surely and }
\mathbf Z\mid U
\sim
L_{\mathbf x_{\mc S},U}
\right\}.
\end{equation}
Thus, the laws in \eqref{eq:Phi-comparison-class-at-fixed-XS} have the same
conditional distribution given the sufficient statistic, but may have
different marginal distributions of that statistic. 
Equation~\eqref{eq:Phi-comparison-class-at-fixed-XS} immediately implies the inclusion
${\cal P}_\Theta(\mathbf x_{\mc S})\subseteq
\widetilde{\mc P}_{\Phi}(\mathbf x_{\mc S})$. 

Fix a copy construction that does not use $\bs Y$, and use the statistic $T$ as in Section~\ref{sec: MSCRT}. Let $W$ collect the auxiliary randomness used by this procedure, independent of $(\bs Y,\mathbf X)$.
Assume the following for every $\varphi\in\mc V(\mathbf x_{\mc S})$. If $\mathbf X_{\mc T}\sim L_{\mathbf x_{\mc S},\varphi}$, then $\mathbf X_{\mc T}$ and the $M$ copies generated by this construction are jointly exchangeable. Let $p_T(\bs y,\mathbf x_{\mc S},\mathbf x_{\mc T},w)$ be the p-value in \eqref{eqn:pval-1} returned by this procedure for fixed arguments. The full technical setting is given in Appendix~\ref{app:additional-robustness}.

\begin{theorem}[Robustness of MS-CRT]
\label{thm: MS-CRT robustness}
Assume $H_0:\bs Y\indp\mathbf X_{\mc T}\mid\mathbf X_{\mc S}$. 
Define the pointwise approximation radius $\Delta_{\Phi}(\mathbf x_{\mc S}) = \inf_{R\in\widetilde{\mc P}_{\Phi}(\mathbf x_{\mc S})} d_{\mathrm{TV}}\!\left( \mathcal L_*(\mathbf X_{\mc T}\mid \mathbf X_{\mc S}=\mathbf x_{\mc S}), R \right)$.
For almost every $(\bs{y},\mathbf x_{\mc S})$ under the true joint law of
$(\bs{Y},\mathbf X_{\mc S})$ and every \(\alpha\in(0,1)\),
\begin{equation}
\label{eq:MSCRT-conditional-robustness}
\mathbb P_*\!\left(
p_T(\bs{y},\mathbf x_{\mc S},\mathbf X_{\mc T},W)\leq\alpha
\mid
\bs{Y}=\bs{y},\mathbf X_{\mc S}=\mathbf x_{\mc S}
\right)
\leq
\alpha+\Delta_{\Phi}(\mathbf x_{\mc S}). 
\end{equation}
\end{theorem}

Averaging \eqref{eq:MSCRT-conditional-robustness} over the joint law of
$(\bs{Y},\mathbf X_{\mc S})$ (assuming $\Delta_\Phi$ is measurable) gives
\begin{equation}
\label{eq:MSCRT-unconditional-robustness}
\mathbb P_*\!\left(
p_T(\bs{Y},\mathbf X_{\mc S},\mathbf X_{\mc T},W)\leq\alpha
\right)
\leq
\alpha+
\mathbb E_*\!\left[\Delta_{\Phi}(\mathbf X_{\mc S})\right].
\end{equation}

If $\mathcal L_*(\mathbf X_{\mc T}\mid\mathbf X_{\mc S}=\mathbf x_{\mc S})\in\widetilde{\mc P}_{\Phi}(\mathbf x_{\mc S})$, then $\Delta_{\Phi}(\mathbf x_{\mc S})=0$, so \eqref{eq:MSCRT-conditional-robustness} gives nominal conditional Type I error control at $\mathbf x_{\mc S}$. If this membership holds for almost every value of $\mathbf X_{\mc S}$, \eqref{eq:MSCRT-unconditional-robustness} gives nominal unconditional Type I error control, even when the true conditional law is outside ${\mc P}_{\Theta}(\mathbf x_{\mc S})$.

The enlarged comparison class $\widetilde{\mc P}_{\Phi}(\mathbf x_{\mc S})$ contains all mixtures of laws in ${\mc P}_{\Theta}(\mathbf x_{\mc S})$ and can therefore be strictly larger than ${\mc P}_{\Theta}(\mathbf x_{\mc S})$; see Appendix~\ref{app:comparison-class-properties} for formal statements.

For comparison, consider an ordinary CRT using $P_{\widehat\theta,\mathbf x_{\mc S}}$, where $\widehat\theta\in\Theta$ is estimated independently of the labelled data.
\citet{berrett_conditional_2020} bound its Type I error inflation by the total variation distance between the true and fitted conditional covariate laws.
Since ${\mc P}_{\Theta}(\mathbf x_{\mc S})\subseteq\widetilde{\mc P}_{\Phi}(\mathbf x_{\mc S})$,
\begin{equation}
\label{eq:MSCRT-working-family-bound-comparison}
d_{\mathrm{TV}}\!\left(
\mathcal L_*(\mathbf X_{\mc T}\mid\mathbf X_{\mc S}=\mathbf x_{\mc S}),
P_{\widehat\theta,\mathbf x_{\mc S}}
\right)
\geq
\inf_{P\in{\mc P}_{\Theta}(\mathbf x_{\mc S})}
d_{\mathrm{TV}}\!\left(
\mathcal L_*(\mathbf X_{\mc T}\mid\mathbf X_{\mc S}=\mathbf x_{\mc S}),P
\right)
\geq\Delta_\Phi(\mathbf x_{\mc S}).
\end{equation}
By Theorem~\ref{thm: MS-CRT robustness} and \eqref{eq:MSCRT-working-family-bound-comparison}, the Type I error inflation of MS-CRT is no larger than the best approximation error within the original parametric family.
For the ordinary CRT, the total variation distance between the true and fitted conditional covariate laws bounds conditional Type I error inflation uniformly over test statistics and significance levels. This bound is asymptotically sharp as $M\to\infty$ \citep{berrett_conditional_2020}.
Particular CRT statistics can have smaller inflation under suitable conditions \citep{niu_reconciling_2024}.
Under correct specification, $\Delta_\Phi(\mathbf x_{\mc S})=0$, so MS-CRT has no Type I error inflation, whereas the ordinary CRT may still have inflated Type I error for some test statistics and significance levels.
Appendix~\ref{app:ordinary-crt-comparison} gives the full comparison.

\subsection{Robustness for Gaussian graphical models}\label{sec:ggm-robustness-rewrite}
We consider $G$-CRT with a graph selected from a covariate sample independent of the inference sample.
Conditioning on the sufficient statistic avoids parameter estimation when generating copies, but the selected graph may still omit true edges.
We use Theorem~\ref{thm: MS-CRT robustness} to quantify the resulting Type I error inflation.

For a fixed target group, the $G$-CRT construction in
Section~\ref{sec: construction} is conditionally valid if the
selected graph contains every true edge incident to the group.
Extra edges and errors entirely within the background group
preserve this validity.
The following result bounds the Type I error inflation caused by
omitted edges between the target and background groups, provided
the graph contains all true edges within the target group.

Suppose the $n$ inference observations are i.i.d., with
$\bs X\sim\mathcal N_p(\bs\mu,\mathbf\Sigma)$,
$\mathbf\Sigma\succ0$, and $H_0$ holding for fixed nonempty
target and background sets $\mc T$ and $\mc S=\mc T^c$.
Define $\mathbf\Omega=\mathbf\Sigma^{-1}$.
For $i\ne j$, an edge $(i,j)$ is a true edge if $\Omega_{ij}\ne0$;
let $\mc E_{\mc T}^*$ be the set of true edges with both
endpoints in $\mc T$.
Let $\mathcal U$ contain $n_u$ i.i.d.\ covariate observations
from the same law, independent of the inference sample,
together with any independent randomisation used for selection.

We use only $\mathcal U$ to choose the graph
$\widehat G=([p],\widehat{\mc E})$, $(\mc I,L,M)$ in
Algorithm~\ref{alg: exchangeable}, and the statistic $T$ as in
Section~\ref{sec: MSCRT}; these choices are independent of
$(\bs Y,\mathbf X)$.
Define the matrix $\mathbf\Omega_{\mc T\mc S}^{\mathrm{miss}}$ of omitted entries of the precision matrix between
the target and background groups by
$\Omega_{ij}^{\mathrm{miss}}=\Omega_{ij}\mathbf{1}\{(i,j)\notin\widehat{\mc E}\}$,
$i\in\mc T$, $j\in\mc S$.
Thus $\mathbf\Omega_{\mc T\mc S}^{\mathrm{miss}}$ is determined by $\mathcal U$.

\begin{proposition}[Independent graph selection]\label{prop:learned-graph-size-rewrite}
\label{prop:learned-graph-size}
Generate the copies by Algorithm~\ref{alg: exchangeable}
with the residual rotations of Section~\ref{sec:ggm-gcrt},
and use fresh auxiliary randomness $W$, independent of
$(\mathcal U,\bs Y,\mathbf X)$.
Assume that, for some deterministic $\delta\in[0,1]$ and
$f\geq0$, there is an event $\mathcal F$ determined by
$\mathcal U$, with $\mathbb P_*(\mathcal F)\geq1-\delta$,
on which $\mc E_{\mc T}^*\subseteq\widehat{\mc E}$ and
$\operatorname{tr}(\mathbf\Omega_{\mc T\mc T}^{-1}\mathbf\Omega_{\mc T\mc S}^{\mathrm{miss}}
\mathbf\Sigma_{\mc S\mc S}(\mathbf\Omega_{\mc T\mc S}^{\mathrm{miss}})^\top)\leq f$.

Then, for every $\alpha\in(0,1)$, the p-value in \eqref{eqn:pval-1} satisfies
\begin{equation}\label{eq:learned-graph-size-rewrite}
\mathbb P_*(\mathrm{pVal}_T\leq\alpha)
\leq\min\left\{1,\alpha+(1-\alpha)\delta
+\frac{1-\delta}{2}\sqrt{nf}\right\}.
\end{equation}
\end{proposition}

On $\mathcal F$, coverage of all true edges within $\mc T$ allows us to choose a reference Gaussian conditional distribution compatible with $\widehat G$ that preserves the true conditional covariance.
The trace measures the average squared difference between the true and reference conditional means for one observation, weighted by their common inverse conditional covariance.
Appendix~\ref{app:proof-learned-graph-size} gives the reference construction and the full proof.

If $\widehat G$ contains every true edge incident to $\mc T$
on $\mathcal F$, then $\mathbf\Omega_{\mc T\mc S}^{\mathrm{miss}}=\mathbf0$ and $f$ can be taken
to be zero. Equation~\eqref{eq:learned-graph-size-rewrite} then
gives $\mathbb P_*(\mathrm{pVal}_T\leq\alpha)
\leq\alpha+(1-\alpha)\delta$.
If some true edges between $\mc T$ and $\mc S$ are omitted on
$\mathcal F$, their contribution to this bound is controlled
by $\sqrt{nf}$.

The graph construction based on thresholded estimates from constrained $\ell_1$ minimisation for inverse matrix estimation (CLIME) \citep{cai_constrained_2011} includes all true edges incident to $\mc T$ with high probability under the population and edge signal conditions given in Appendix~\ref{app:clime-graph-selection}.
Under the additional growth conditions stated there, this
guarantee can hold even when $n_u=o(n)$.

If the signal requirement is imposed only on true edges within $\mc T$, these edges are still included with high probability. By contrast, edges between $\mc T$ and $\mc S$ may be omitted when the corresponding entries of the precision matrix are small in absolute value relative to the CLIME estimation error.
Corollary~\ref{cor:learned-graph-threshold} in Appendix~\ref{app:clime-graph-selection} gives an explicit choice of $f$ under its assumptions.

\section{Simulation studies}\label{sec: Simulation}
Section~\ref{sec:simulation-unknown-parameters} compares MVN-CRT with an oracle CRT and three plug-in CRTs, all using the same test statistic.
Section~\ref{sec:simulation-g-crt} compares $G$-CRT with coordinatewise procedures and classical tests under four response models in low- and high-dimensional settings.

\subsection[Unknown parameters in the randomisation law]{Unknown parameters in the randomisation law}
\label{sec:simulation-unknown-parameters}

We compare the oracle CRT, the plug-in MLE CRT, the plug-in CRT with a degrees of freedom adjustment (d.f.), the plug-in CRT based on Ledoit--Wolf covariance shrinkage (LW), and MVN-CRT. The oracle CRT knows the true conditional covariate law, which is unavailable to the other procedures. The plug-in procedures estimate this law from the same covariate matrix used for testing, whereas MVN-CRT conditions on Gaussian sufficient statistics. Within each experiment, all five procedures use the same test statistic, number of randomised copies, and p-value rule.

Here we consider $(n,s,t)=(60,45,10)$, $(80,50,10)$, and $(80,60,40)$. Let $d=n-s-1$ be the residual degrees of freedom. The background covariates are independent standard normal variables. Given these covariates, the target covariates follow a Gaussian linear model with an intercept and identity error covariance. For signal strength $\theta\geq0$, the response follows a linear model in all covariates with target coefficient vector $\theta t^{-1/2}\bs{1}_t$ and independent $\mathcal{N}(0,1)$ errors; $\theta=0$ gives the null. All three settings satisfy $n>s+1$. The third setting has $p>n$ but $s<n$ and illustrates MVN-CRT with a singular residual Gram matrix, since $t=40>d=19$.

Define $\mathbf D=(\bs{1}_n,\mb X_{\mc S})$, and let $\mathbf P_{\mathbf D}$ be the orthogonal projection onto its column space. The response residual is $\bs R_Y=(\mathbf I_n-\mathbf P_{\mathbf D})\bs Y$. For an observed or randomised target matrix $\mathbf x_{\mc T}\in\mathbb R^{n\times t}$, define its residual matrix $\mathbf r=(\mathbf I_n-\mathbf P_{\mathbf D})\mathbf x_{\mc T}$. All five procedures use the ridge gain statistic
\begin{equation}\label{eq:simulation-ridge-gain}
T_{\mathrm{ridge}}(\mathbf r)
=\bs R_Y^\top\mathbf r(\mathbf r^\top\mathbf r+\mathbf I_t)^{-1}\mathbf r^\top\bs R_Y.
\end{equation}
We fix the ridge penalty at one, so the statistic is well defined even when $t>d$. The MLE and d.f. plug-in procedures fit the conditional means by least squares and estimate the covariance from the residuals, with denominator $n$ (MLE) or $d$ (d.f.). The LW procedure uses the covariance estimator and diagonal variance adjustment in Appendix D.10 of \citet{liu_fast_2022}, adapted to the multivariate target by conditioning the fitted Gaussian distribution on the background covariates. MVN-CRT uses the construction in Section~\ref{sec:mvn-CRT}, with the intercept adjustment in Remark~\ref{rem: general mean for MVN}. For all procedures, we compute p-values by \eqref{eqn:pval-1} and reject at level $0.05$. Appendix~\ref{app:joint-crt-benchmark} gives the data-generating model and sampling details, including the interpretation of the plug-in covariance estimates when $t>d$.

\begin{table}[!htbp]
    \centering
\caption{Estimated size (\%) at nominal level $0.05$ for the five CRT procedures, all using the ridge gain statistic. Here $d=n-s-1$. Each size estimate is based on $5000$ replications with $M=999$ randomised copies per replication. Monte Carlo standard errors are in parentheses.}
    \label{tab:joint-crt-benchmark-size}
    \begin{tabular*}{\linewidth}{@{\extracolsep{\fill}}rrrr@{\hspace{12pt}}rrrrr@{}}
\hline
$n$ & $s$ & $t$ & $d$ & \shortstack{Oracle\\CRT} & \shortstack{Plug-in\\MLE} & \shortstack{Plug-in\\(d.f.)} & MVN-CRT & LW \\
\hline
$60$ & $45$ & $10$ & $14$ & $4.7\;(0.3)$ & $52.5\;(0.7)$ & $13.4\;(0.5)$ & $4.8\;(0.3)$ & $4.9\;(0.3)$ \\
$80$ & $50$ & $10$ & $29$ & $5.1\;(0.3)$ & $11.8\;(0.5)$ & $6.3\;(0.3)$ & $5.2\;(0.3)$ & $5.2\;(0.3)$ \\
$80$ & $60$ & $40$ & $19$ & $4.5\;(0.3)$ & $100.0\;(0.0)$ & $97.7\;(0.2)$ & $4.5\;(0.3)$ & $17.8\;(0.5)$ \\
\hline
\end{tabular*}

\end{table}

Table~\ref{tab:joint-crt-benchmark-size} shows that the estimated sizes of the oracle CRT and MVN-CRT remain near the nominal level in all three settings. The MLE and d.f. plug-in procedures are oversized, even in the first two settings, where their covariance estimates are nonsingular. The degrees of freedom adjustment reduces but does not remove the Type I error inflation. The LW procedure has estimated sizes $4.9\%$, $5.2\%$, and $17.8\%$. Thus, regularisation substantially reduces the Type I error inflation, but does not restore nominal size in the third setting, despite a positive definite fitted conditional covariance. These results show that plug-in CRTs fitted on the observed covariate sample can have substantial Type I error inflation even under a correctly specified Gaussian model.

\begin{figure}[!htbp]
    \centering
    \includegraphics[width=0.95\textwidth]{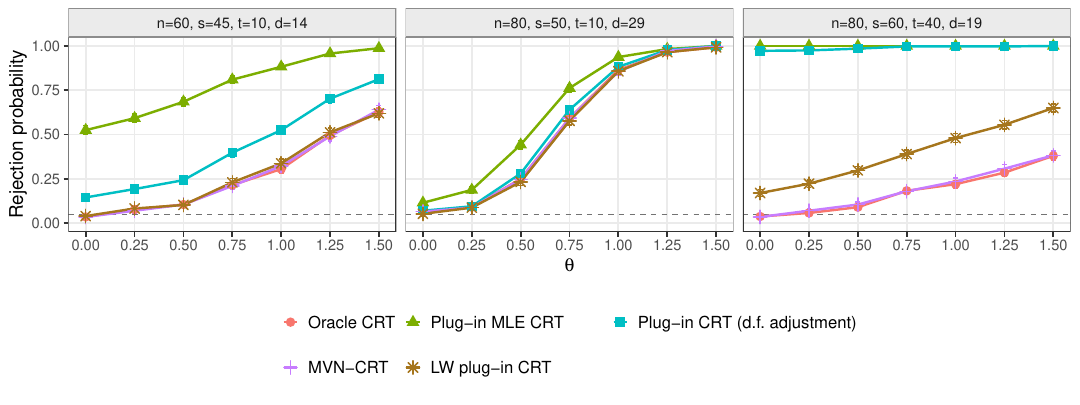}
    \caption{Rejection rates at nominal level $0.05$ for the five procedures and settings in Table~\ref{tab:joint-crt-benchmark-size}. At each signal strength, including $\theta=0$, we use $400$ replications with $M=199$ randomised copies per replication. Error bars show one Monte Carlo standard error. The points at $\theta=0$ estimate size in this experiment. Curves for oversized plug-in procedures are shown for reference and do not provide valid power comparisons.}
    \label{fig:oracle-plugin-power-curve}
    \caption*{Alt text: Three panels plot rejection probability against signal strength for five CRT procedures. MVN-CRT and oracle CRT curves nearly overlap. Some plug-in curves are already above the dashed nominal level at zero signal, particularly in the third panel.}
\end{figure}

In Figure~\ref{fig:oracle-plugin-power-curve}, the oracle CRT and MVN-CRT have similar rejection rates over the displayed alternatives; the largest observed absolute difference is $2.3$ percentage points after rounding. These rejection rates and their Monte Carlo standard errors suggest little power loss from conditioning on sufficient statistics in these Gaussian settings. The higher rejection rates of the oversized plug-in procedures do not establish a power advantage.

Appendix~\ref{app:mvn-nongaussian} provides additional simulations to examine the size of MVN-CRT under non-Gaussian covariate models. With standardised multivariate $t_5$ and $t_3$ covariate errors, its estimated sizes are close to the nominal level of $0.05$, although some settings with $t_3$ covariate errors show appreciable relative inflation at smaller nominal levels.

\subsection[Group testing with graphical covariate models]{Group testing with graphical covariate models}
\label{sec:simulation-g-crt}

We assess the power of $G$-CRT under Gaussian graphical covariate models when the graph is known but the mean and precision matrix are unknown.
In each of the four response models, we test a group of eight target covariates, with $(n,p)=(50,20)$ in the low-dimensional setting and $(n,p)=(80,120)$ in the high-dimensional setting.
In both dimension settings, the graph has maximum degree $12$.
We vary a signal parameter $\theta$, with $\theta=0$ giving the null, and test at nominal level $0.05$.
Appendix~\ref{app:simulation-details} gives the experimental designs, statistic definitions, and full results, including all figures cited in this subsection.

Under a Gaussian linear response model, $G$-CRT with regression statistics has power close to that of the $F$-test in low dimensions (Figure~\ref{fig:llw}). In high dimensions, the LM-L1-R-SSR statistic uses the negative residual sum of squares from a linear regression of the response residuals on the target covariates, after an $\ell_1$-penalised linear fit on the background covariates. At moderate signal strengths, $G$-CRT with this statistic has substantially higher rejection rates than the coordinatewise distilled CRT (dCRT) \citep{liu_fast_2022} and the de-sparsified Lasso \citep{van_de_geer_asymptotically_2014}, both with Bonferroni correction (Figure~\ref{fig:hlw}). Under a logistic regression model, in both low and high dimensions, $G$-CRT can have higher power than the coordinatewise procedures considered (Figures~\ref{fig:lglw} and~\ref{fig:hglw}).

For nonlinear continuous and binary responses, we compute the RF-RR statistic in two steps: we first obtain response residuals from a random forest regression on the background covariates, and then sum the positive variable importance scores from a second random forest regression of these residuals on the target covariates. In high dimensions, $G$-CRT with RF-RR has substantially higher power than the coordinatewise procedures considered, including dCRT with random forests (Figures~\ref{fig:hlm} and~\ref{fig:hglm}).

Appendix~\ref{app:matched-ridge} (Table~\ref{tab:matched-direct-comparison}) compares $G$-CRT with a coordinatewise test using the same graph and with oracle joint CRTs using the same test statistics. In these experiments, $G$-CRT has power close to that of the corresponding oracle joint CRTs without knowing the covariate mean or precision matrix. Both $G$-CRT and the oracle joint CRTs can have higher or lower power than the coordinatewise test, depending on the setting and the test statistic. 

Additional simulations examine the effect of using supergraphs on power
(Appendix~\ref{app: CRT super}), the power of $G$-CRT with a single
target covariate (Appendix~\ref{app:crt-univariate}), and its estimated
size and power under a heavy-tailed, non-Gaussian covariate design
(Appendix~\ref{app: CRT vio}). In the heavy-tailed design studied,
$G$-CRT has estimated sizes close to the nominal level. With statistics based on residual sums of squares, $G$-CRT has power comparable to that of the $F$-test in the low-dimensional setting and higher than that of the coordinatewise procedures considered in the high-dimensional setting (Figure~\ref{fig: linear regression-t3}).

\section{Real-data illustrations}
\label{sec: Application}
This section presents two exploratory data examples illustrating the usefulness of the $G$-CRT. 

\subsection{Dependence of fund returns}\label{sec: fund}

Some fund holdings are disclosed only with a delay \citep{agarwal2013uncovering}. A fund's returns may therefore contain information about its investment strategy before these holdings become public. We test whether the returns of a specified stock group are conditionally associated with the fund return $Y$ given the remaining stock returns, using a fund with disclosed holdings as a reference.

In our example, $Y$ is the average return of the SPDR Dow Jones Industrial Average ETF Trust (DIA) over a block of five trading days, and $\bs X$ contains the corresponding transformed average returns of 103 large stocks.
We consider the period from 1 September 2020 to 30 June 2022, which yields data with $n=92$ and $p=103$.
As a working approximation, we treat the return observations as i.i.d. across blocks.
Further details are given in Appendix~\ref{app:stock-dataset}.
We consider two target sets $\mc{T}$:
\\*
\textup{1.}\enspace $\mc{T}_1=\{\text{CAT, DOW, HON, MMM, PG}\}$:
    five stocks from traditional sectors such as materials and manufacturing;\\*
\textup{2.}\enspace $\mc{T}_2=\{\text{AMGN, CRM, MSFT, UNH}\}$:
    four stocks in the biotechnology, technology, and healthcare sectors.

DIA's disclosed holdings show that both $\mc T_1$ and $\mc T_2$ contain stocks whose returns contribute directly to $Y$. We use these two target sets to compare the evidence against $Y\indp \bs{X}_{\mc T}\mid \bs{X}_{-\mc T}$ reported by the different tests.

We compare $G$-CRT using LM-L1-R-SSR and RF-RR with the de-sparsified Lasso and two coordinatewise dCRT baselines using Gaussian Lasso models (L1) or random forests (RF).
In simulations under fitted null models (Appendix~\ref{app:stock-preliminary}), the four CRTs reject in $0.8\%$--$7.3\%$ of replications at the $0.05$ threshold, compared with $26.0\%$ and $15.0\%$ for the de-sparsified Lasso with $\mc T_1$ and $\mc T_2$, respectively.

We therefore focus on the four CRTs. On the observed data,
the $G$-CRT with LM-L1-R-SSR and dCRT (L1) reject for both
$\mc T_1$ and $\mc T_2$ at the $0.05$ threshold.
The $G$-CRT with RF-RR rejects only for $\mc T_2$, whereas
dCRT (RF) rejects for neither group
(Table~\ref{tab: fund p-value} in
Appendix~\ref{app:stock-fund-dependence}).

To compare detection rates at lower signal-to-noise ratios,
we add independent Gaussian noise to the DIA response for $\mc T_2$.
The covariate matrix and graph remain fixed, and all four procedures
use the same perturbed response in each replication.
Further details are given in Appendix~\ref{app:stock-fund-dependence}.

Figure~\ref{fig: stock noise} shows that the $G$-CRT with LM-L1-R-SSR
has the highest empirical detection rate at all five noise levels,
tied with dCRT (L1) at the lowest level.
At a noise standard deviation of $0.10$ percentage points, its detection
rate is $94.3\%$, compared with $72.8\%$ for dCRT (L1).

\begin{figure}[!htbp]
    \centering
    \includegraphics[width=0.80\linewidth]{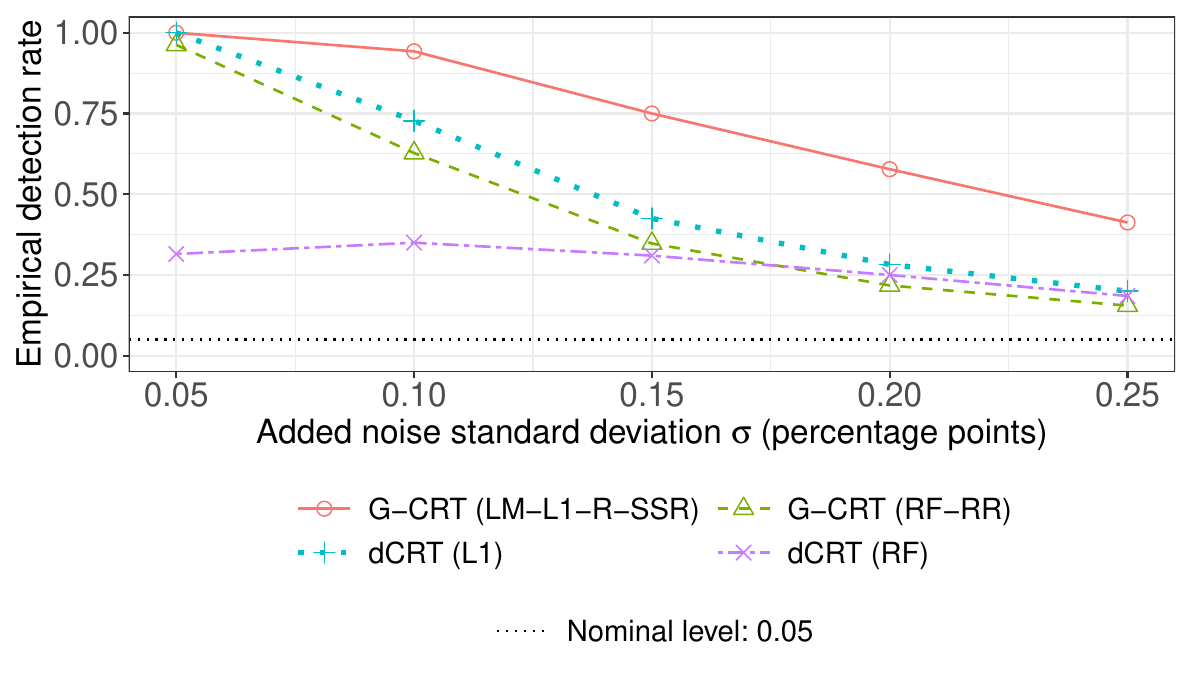}
    \caption{Empirical detection rates at level $0.05$ for $\mc T_2$ after adding independent $\mathcal{N}(0,\sigma^2)$ noise to each entry of the observed DIA response, based on 400 replications per noise level. 
All four methods returned finite p-values in every replication.}\label{fig: stock noise}
    \caption*{Alt text: Four curves plot empirical detection rate against added noise standard deviation. G-CRT (LM-L1-R-SSR) has the highest rate throughout, tied with dCRT (L1) at the lowest noise level. Its rate falls from about 1 to 0.4; the other rates end below 0.2.}
\end{figure}

\subsection{Breast cancer relapse}\label{sec: breast}
We next consider a breast cancer relapse example in which the usual chi-squared approximation for the likelihood ratio test is unreliable at the available sample size. The breast cancer relapse data were collected by \citet{wang_gene-expression_2005} to identify genes associated with metastasis. They are available from the NCBI GEO database \citep{barrett2012ncbi} under accession GSE2034. We analyse 209 oestrogen-receptor-positive patients using the 60 probes identified for this subgroup in the original study.

The response $Y$ indicates whether relapse occurred within four years and, if so, whether brain relapse was recorded. A preliminary multinomial Lasso analysis of the transformed expression levels of the probes selects $\mc S$ by the one-standard-error rule and $\mc S\cup\mc T$ by minimum cross-validation error, yielding $|\mc S|=3$ and $|\mc T|=18$.
To assess whether the additional probes provide information about relapse beyond the smaller selection, we test $Y\indp \bs{X}_{\mc T}\mid \bs{X}_{\mc S}$. 
We compare the likelihood ratio test for multinomial logistic regression, the $G$-CRT with a random forest classification statistic, and Bonferroni-adjusted coordinatewise dCRT baselines using Gaussian Lasso models or random forests. Because $\mc S$ and $\mc T$ are selected using the response, this comparison is exploratory. Further details of the data and analysis are given in Appendix~\ref{app:bcr-details}.

The likelihood ratio test gives a p-value of approximately $9.58\times10^{-6}$ under its chi-squared reference distribution. To examine this approximation, we use 400 replications under conditional independence, keeping $\mc S$ and $\mc T$ fixed. Responses are generated from a multinomial model fitted using only $\bs{X}_{\mc S}$. At the nominal $0.05$ level, the likelihood ratio test rejects $40.0\%$ of the time, compared with $7.3\%$ for $G$-CRT and $3.3\%$ for each dCRT baseline.

On the observed data, the $G$-CRT with the RF statistic rejects at
the $0.05$ threshold (p-value $0.018$), whereas the coordinatewise
dCRT baselines using Gaussian Lasso models and random forests
do not (p-values $0.697$ and $0.407$, respectively).

\section{Discussion}
\label{sec: Discussion}
By conditioning on sufficient statistics for the parametric family $\mathcal P_\Theta(\mathbf x_{\mc S})$ of laws of $\mathbf X_{\mc T}$ given $\mathbf X_{\mc S}=\mathbf x_{\mc S}$, MS-CRT gives finite-sample Type I error control without estimating the covariate parameters for resampling or specifying a response model.
Our robustness theory bounds Type I error inflation by the best total variation approximation to the true conditional covariate law within the enlarged comparison class. This bound is no larger than the corresponding bound on worst-case Type I error inflation for an ordinary CRT using any law in $\mathcal P_\Theta(\mathbf x_{\mc S})$. For covariates from Gaussian graphical models, we also quantify how errors in a graph estimated from independent covariate data affect Type I error.

Our power analysis shows that conditioning allows effective detection in the Gaussian linear models studied here.
MVN-CRT attains minimax separation rates.
For $G$-CRT, we establish sufficient conditions for power consistency with fixed numbers of sweeps and randomised copies.

Two questions remain. First, tractable samplers that produce useful exchangeable copies are needed for covariate models beyond those implemented here. Such samplers could extend MS-CRT to a wider range of Markov random fields, such as exponential family graphical models \citep{wainwright2008graphical}. Second, the power analysis for $G$-CRT remains open for the statistics based on $\ell_1$-penalised fits and random forests used in our simulations.

\section*{Acknowledgements}
D. Huang was partially supported by Singapore MOE AcRF Tier-1 Grant A-8004149-00-00. 
GPT-6 Astra assisted with mathematical derivations, code development, and manuscript revision. The authors directed this work, checked the material used for accuracy, and take responsibility for the manuscript.

\begingroup
\hypersetup{pdfborder={0 0 0}}
\bibliography{reference}
\bibliographystyle{crt-compact}
\endgroup

\clearpage
\begin{bibunit}[crt-compact]
\makeatletter
\def\@extra@b@citeb{.supp}
\def\@extra@binfo{.supp}
\makeatother
\appendix

\makeatletter
\let\c@proposition\c@theorem
\let\c@lemma\c@theorem
\let\c@corollary\c@theorem
\let\c@condition\c@definition
\let\c@assumption\c@definition
\makeatother
\setcounter{theorem}{0}
\setcounter{example}{0}
\setcounter{definition}{0}
\setcounter{remark}{0}
\setcounter{setting}{0}
\setcounter{problem}{0}
\setcounter{algorithm}{0}
\setcounter{figure}{0}
\setcounter{table}{0}
\renewcommand{\thetheorem}{S\arabic{theorem}}
\renewcommand{\theexample}{S\arabic{example}}
\renewcommand{\theproposition}{\thetheorem}
\renewcommand{\thedefinition}{S\arabic{definition}}
\renewcommand{\theremark}{S\arabic{remark}}
\renewcommand{\thelemma}{\thetheorem}
\renewcommand{\theassumption}{\thedefinition}
\renewcommand{\thecondition}{\thedefinition}
\renewcommand{\thecorollary}{\thetheorem}
\renewcommand{\thesetting}{S\arabic{setting}}
\renewcommand{\theproblem}{S\arabic{problem}}
\renewcommand{\thealgorithm}{S\arabic{algorithm}}
\renewcommand{\thefigure}{S\arabic{figure}}
\renewcommand{\thetable}{S\arabic{table}}
\renewcommand{\theHfigure}{\thefigure}
\renewcommand{\theHtable}{\thetable}
\renewcommand{\theHtheorem}{S\arabic{theorem}}
\renewcommand{\theHlemma}{\theHtheorem}
\renewcommand{\theHalgorithm}{S\arabic{algorithm}}
\renewcommand{\theHproposition}{\theHtheorem}
\renewcommand{\theHcorollary}{\theHtheorem}
\renewcommand{\theHdefinition}{S\arabic{definition}}
\renewcommand{\theHcondition}{\theHdefinition}
\renewcommand{\theHassumption}{\theHdefinition}
\renewcommand{\theHexample}{S\arabic{example}}
\renewcommand{\theHremark}{S\arabic{remark}}

\section*{Appendix}
Appendix~\ref{app: literature} provides a more detailed and complete version of the literature review discussed in the main text. Appendix~\ref{app:joint-crt-benchmark} gives the data-generating model and sampling details for the Gaussian CRT comparison in Section~\ref{sec:simulation-unknown-parameters}.
Appendix~\ref{app:supplementary-results} contains supplementary methods and results on robustness and power. 
Appendix~\ref{app:proofs} contains the omitted proofs. Appendix~\ref{app:simulation-details} gives computational settings and additional results for the simulation studies in Section~\ref{sec: Simulation}. 
Appendix~\ref{app:power-numerical-comparisons} gives numerical comparisons for the $G$-CRT power analysis in Section~\ref{sec:gcrt-ridge-power}.
Appendix~\ref{app:application-details} shows detailed information about the applications of our testing
methods to real-world data in Section~\ref{sec: Application} in the main text.
Code for reproducing our experiments is maintained on GitHub (\url{https://github.com/stathuang/multivariate-crt/}).

\section{Literature review}\label{app: literature}
This appendix expands the literature review in the main text.

Testing conditional independence $\bs{X}_{\mc T}\indp Y\mid \bs{X}_{\mc S}$ is central in variable selection, causal inference, and graphical modelling. 
At any fixed sample size, \citet{shah2020} show that, over the class of absolutely continuous joint distributions they consider, any test with Type I error at most \(\alpha\) uniformly across the distributions satisfying conditional independence has power at most \(\alpha\) against every alternative in the class.
Therefore, obtaining nontrivial power requires restrictions on the class of null distributions.

\subsection{Conditional randomisation tests}\label{app:model-x-inference}

The ``model-X'' framework assumes that $F_{\bs{X}}$, the distribution of $\bs{X}$, is known or satisfies some specific assumptions, while making no modelling assumption on $F_{Y\mid \bs{X}}$, the conditional distribution of $Y$ given $\bs{X}$, for validity. Response models may still be used to construct powerful statistics. \citet{candes_panning_2017} introduced the conditional randomisation test (CRT) and model-X knockoffs under this framework.
The CRT has been extended and developed in various directions in recent years.
\citet{bates2020causal} applied the idea of the CRT to develop a method for causal inferences using genetic data that include parents and offspring.
\citet{shaer_model-x_2023} developed a sequential extension of the CRT for i.i.d. data streams, with Type I error control under optional stopping.
To accelerate the computation of test statistics,
\citet{tansey_holdout_2022} proposed the holdout randomisation test, while \citet{liu_fast_2022} introduced the distilled CRT (dCRT). \cite{niu2024spa} use a saddlepoint approximation to dCRT to achieve similar statistical performance with significantly reduced computational demands.

Most existing CRT methods have been developed for a univariate target variable. 
One notable multivariate example is \citet{ham_using_2022}, who studied the joint significance of a group of factors in conjoint analysis.

\paragraph{Unknown conditional covariate laws}
The original CRT requires the ability to sample from the true conditional distribution $\bs X_{\mc T}\mid\bs X_{\mc S}$, which is often unavailable in practice. Several extensions have been developed to accommodate an estimated or approximate conditional law.
\citet{bellot_conditional_2019} proposed to use generative adversarial networks to approximate the conditional distribution, while
\citet{berrett_conditional_2020} introduced non-uniform permutations and bounded the Type I error inflation caused by errors in the approximate conditional covariate law. 
Other approaches study robustness to estimation error. \citet{niu_reconciling_2024} studied the double robustness of the dCRT, relating its asymptotic validity to the joint accuracy of estimates of the target and response means conditional on the background covariates.
The approximate co-sufficient sampling (aCSS) method of \citet{barber_testing_2021} can also be used to implement a CRT. It conditions on an approximately sufficient estimator and uses estimated parameters in the conditional sampling law.
When resampling relies on estimated distributions or parameters, these methods provide asymptotic validity or robustness guarantees under conditions on the estimation error. 
In contrast, when the parametric covariate model is correctly specified, MS-CRT achieves exact finite-sample validity by conditioning on sufficient statistics, without estimating the nuisance parameters for resampling. 

\paragraph{Comparison with the MAXWAY CRT}
The MAXWAY CRT of \citet{li_maxway_2023} uses summaries of the background covariates to improve robustness to errors in the estimated conditional covariate law. In their notation, let $X$ be the scalar target, $Y$ the response and $Z$ the background covariates. The ideal target summary $h(Z)$ and response summary $g(Z)$ satisfy $X\indp Z\mid h(Z)$ and $Y\indp Z\mid g(Z)$, respectively. Fitted summaries need not satisfy these conditions exactly.

Let $\mathbf x$, $\mathbf y$ and $\mathbf Z$ be the corresponding target and response data vectors and matrix of background covariates, with $h$ and $g$ applied rowwise to $\mathbf Z$. With $h$, $g$ and the fitted law treated as given, Algorithm~2 of \citet{li_maxway_2023} draws copies from an estimate of $\mathcal L(\mathbf x\mid(g(\mathbf Z),h(\mathbf Z)))$, independently of the observed $(\mathbf x,\mathbf y)$ conditional on these summaries. Its test statistic uses the background covariates only through these summaries. For this algorithm, their Theorem~1 establishes validity under $\mathbf x\indp\mathbf y\mid\mathbf Z$ if either (i) the observed target data and the copies are exchangeable given $\mathbf Z$; or~(ii)~they are exchangeable given $(g(\mathbf Z),h(\mathbf Z))$, and $\mathbf y\indp\mathbf Z\mid g(\mathbf Z)$. The requirement on the response summary applies specifically to condition (ii).

The conditions on the ideal summaries $h$ and $g$ do not by themselves make the conditional law of the target known, and their Algorithm~2 uses an estimate of this law. In MS-CRT, $\Phi=\phi(\mathbf X_{\mc T};\mathbf X_{\mc S})$ is sufficient for the unknown parameters of the specified conditional covariate model, so the law of $\mathbf X_{\mc T}$ given $(\mathbf X_{\mc S},\Phi)$ is free of these parameters. This conditional law can be used for resampling without estimating those parameters. Our Proposition~\ref{prop: CRT valid} applies the same conditional exchangeability argument under its stated sampling and scoring conditions, with exchangeability conditional on $(\mathbf X_{\mc S},\Phi)$.

\paragraph{Power analysis}
Regarding the power of tests based on the CRT, \citet{wang2022high} derived explicit expressions for the asymptotic power of CRTs based on several commonly used statistics in high-dimensional linear regression with a univariate target.
\citet{katsevich_power_2022} identified a likelihood statistic that yields the most powerful conditionally valid CRT against point alternatives and derived the limiting power of the dCRT against local semiparametric alternatives for a target of fixed dimension. 
Our analysis instead allows the target dimension \(t\) to increase with \(n\) and establishes that the MVN-CRT attains the minimax separation rates in the $\ell_2$ and $\ell_\infty$ norms of the conditional covariance vector in the Gaussian linear regression setting of Section~\ref{sec: MVN-CRT power}. 
The model for $Y\mid \bs{X}$ is imposed only for this power analysis; finite-sample validity follows separately from conditional exchangeability under the covariate model.

\subsection{Other methods for testing conditional independence}\label{app:regression-ci-comparison}

Beyond model-X inference, other methods for testing MCI can be broadly divided into regression methods based on models for $Y\mid \bs{X}$ and nonparametric methods.
We first discuss regression methods and power comparisons with MVN-CRT, and then review nonparametric methods.

\paragraph{Regression methods and model assumptions}
Regression methods assume that the conditional distribution $F_{Y\mid\bs{X}}$ belongs to a specific parametric or semiparametric regression model.
Within this approach, one line of high-dimensional inference focuses on constructing confidence intervals and testing hypotheses for low-dimensional parameters within regression models.
Along this line, \citet{zhang2014confidence} developed a bias correction method for constructing confidence intervals for low-dimensional linear combinations of the regression coefficients in high-dimensional linear models. \citet{van_de_geer_asymptotically_2014} and \citet{javanmard_confidence_2013} further developed various debiasing methods based on Lasso to achieve optimality.
Related developments concern minimax lengths and adaptivity of confidence intervals \citep{10.1214/16-AOS1461} and debiasing with an estimated covariate model \citep{10.1093/jrsssb/qkae039}.
For binary responses, \citet{cai2023statistical} extended the debiasing method to high-dimensional generalised linear models (GLMs).

For Gaussian linear regression, the asymptotic normality result in Theorem~2.2 of \citet{van_de_geer_asymptotically_2014} requires the number of nonzero regression coefficients to be $o(\sqrt n/\log p)$, together with conditions on the design. Such sparsity requirements can be restrictive in applications.
Other approaches to inference in linear models include permutation tests \citep{diciccio2017robust,lei2021assumption,wen_residual_2024} and goodness-of-fit or generalised $F$-type statistics \citep{verzelen_goodness--fit_2010,WANG2013134}.
For high-dimensional linear models, \citet{zhu2018linear} and \citet{bradic2022testability} also develop inference that allows nonsparse regression coefficients.

Joint tests of regression coefficients provide a way to assess the significance of a group of covariates within a response model. \citet{ning_general_2017} develop asymptotically valid decorrelated score tests for sparse high-dimensional models, including hypotheses whose dimension grows (Section~H.2 of their supplementary material).

\citet{shah2018goodness} construct Residual Prediction tests for group significance under a model for $Y\mid\bs{X}$, whereas MS-CRT uses a covariate model to obtain finite-sample validity without specifying this response law. Their framework also covers regression goodness of fit, and its theoretical guarantees in high dimensions assume a sparse Gaussian linear model. \citet{jankova_goodness--fit_2019} develop related residual prediction tests for goodness of fit and group significance in high-dimensional generalised linear models.

\paragraph{Power and minimax rates}
We compare the power of MVN-CRT under $\ell_2$ separation with results for regression $F$-tests, using the notation of Section~\ref{sec: MVN-CRT power}.
We then relate the lower bounds in Theorem~\ref{lower bound} to existing minimax results for Gaussian regression.

\citet{verzelen_goodness--fit_2010} derive nonasymptotic power bounds for combinations of $F$-tests over subsets of the target variables in Gaussian linear models with Gaussian covariates.
For comparison with MVN-CRT, consider their single $F$-test of all $t$ targets.
Under the dimension conditions of our Theorem~\ref{power analysis}(1), and for all sufficiently large $n$, their Theorem~3.3 gives a sufficient signal-to-noise ratio $\bs{\beta}_{\mc{T}}^{\top}\mathbf{\Sigma}_{\mc{T}\mid\mc{S}}\bs{\beta}_{\mc{T}}/\sigma^2$ of order $\sqrt t/(n-s)$, for fixed significance and power levels.
The covariance eigenvalue bounds defining $\Theta$ imply $\bs{\beta}_{\mc{T}}^{\top}\mathbf{\Sigma}_{\mc{T}\mid\mc{S}}\bs{\beta}_{\mc{T}}\asymp\|\mathbf{\Sigma}_{\mc{T}\mid\mc{S}}\bs{\beta}_{\mc{T}}\|_2^2$.
With fixed positive $\sigma^2$, the corresponding separation scale for $\|\mathbf{\Sigma}_{\mc{T}\mid\mc{S}}\bs{\beta}_{\mc{T}}\|_2$ is $t^{1/4}/\sqrt{n-s}$.
Since $n-s$ is of order $n$ under these dimension conditions, this agrees with the $\ell_2$ separation rate of MVN-CRT.

Beyond the Gaussian setting, \citet[Section~3]{WANG2013134} analyse the generalised $F$-test.
Their covariate and moment conditions allow non-Gaussian covariates and errors, while the response model assumes a linear conditional mean and constant conditional variance.
For $t/n\to\rho\in(0,1)$ and $s=o(n)$, they derive an explicit local asymptotic power formula that quantifies how the signal-to-noise ratio and $\rho$ affect rejection probability.
In the Gaussian setting, our power bound in Theorem~\ref{power analysis}(1) allows $s$ proportional to $n$, under its stated dimension conditions.

The minimax lower bounds in Theorem~\ref{lower bound} can also be obtained by applying results of \citet{verzelen_goodness--fit_2010} to suitable submodels of $\Theta$.
The $\ell_2$ bound follows from their Proposition~4.2, and the $\ell_\infty$ bound from their Theorem~4.3 applied to alternatives supported on one coordinate.
These lower bounds hold even when the covariate covariance and response variance are known, so they also apply to our problem with unknown parameters.

\paragraph{Nonparametric methods}
Nonparametric methods include kernel methods \citep{fukumizu2007kernel,zhang2011kernel,doran2014permutation,strobl2019approximate,scetbon2022asymptotic}, 
measures of conditional dependence \citep{wang2015conditional,kubkowski2021gain},
and nearest-neighbour methods \citep{runge_conditional_2018,sen_model-powered_2017}.
The local permutation scheme of \citet{runge_conditional_2018} and the nearest-neighbour sampling procedure of \citet{sen_model-powered_2017} both construct reference samples using observations with similar background covariates.
\citet{kim2022local} study the related approach of permuting target values within bins of the background covariates and establish Type I error bounds and power guarantees under smoothness assumptions.
\citet{chen2026local} extend this analysis to bins chosen adaptively from the observed background covariates.
These nonparametric methods impose fewer restrictions on distributional form. Existing studies often focus on a single target covariate in settings with $p<n$, whereas we consider joint tests of multiple covariates.

\subsection{Comparison between knockoffs and exchangeable copies} \label{rem:knockoff}

\paragraph{Exchangeability requirements}
Under the sampling assumptions of Proposition~\ref{prop: CRT valid}, MS-CRT uses joint exchangeability of the target matrices $\{\mathbf{X}_{\mc{T}}, \widetilde{\mathbf{X}}_{\mc{T}}^{(1)}, \ldots, \widetilde{\mathbf{X}}_{\mc{T}}^{(M)}\}$ conditional on $(\mathbf{X}_{\mc{S}},\Phi)$ to calibrate the p-value in \eqref{eqn:pval-1}. The model-X knockoff method of \citet{candes_panning_2017} instead uses columnwise swap invariance for FDR control. Its knockoff matrix $\overline{\mathbf X}$ is constructed so that $[\mathbf X,\overline{\mathbf X}]_{\mathrm{swap}(j)}\stackrel{d}{=}[\mathbf X,\overline{\mathbf X}]$ for each $j\in\{1,\ldots,p\}$, where $\mathrm{swap}(j)$ exchanges the $j$-th columns of the two matrices.

For $M=1$, columnwise swap invariance for the law of $\mathbf X_{\mc T}$ conditional on $(\mathbf X_{\mc S},\Phi)$ is stronger than the exchangeability required by MS-CRT: swapping all target columns exchanges the two target matrices. For $M>1$, however, generating $M$ knockoffs independently conditional on the observed matrix need not give the joint exchangeability of all $M+1$ matrices required by MS-CRT (see also Appendix~A.7 in \citet{lin_goodness--fit_2025}).

\paragraph{Sampling constructions}
To obtain this joint exchangeability, the hub-and-spoke construction of \citet{besag_generalized_1989} uses a transition that preserves the conditional law together with its time reversal; a reversible kernel is a special case. This construction does not require columnwise swap invariance. The Metropolized knockoff sampler (Metro) of \citet{bates_metropolized_2021}, which constructs knockoffs through sequential Metropolis--Hastings updates, can also provide such a transition if it is valid for this conditional law. Its columnwise swap invariance then gives an exchangeable pair of input and output matrices, so its sampling kernel is reversible with respect to this law \citep[Proposition~1]{bates_metropolized_2021}.

The Gaussian $G$-CRT in Section~\ref{sec:ggm} generates its copies through the hub-and-spoke construction, using the residual rotations and subset sampler of \citet{lin_goodness--fit_2025}. Here the residual rotations provide the local transitions, and the hub-and-spoke construction makes the observed target matrix and its copies exchangeable using finitely many local updates.

\section[Details of the Gaussian CRT comparison]{Details of the Gaussian CRT comparison}\label{app:joint-crt-benchmark}
\begingroup
\makeatletter
\def\toclevel@paragraph{2}
\makeatother
This appendix describes the data-generating model and sampling procedures for the comparison in Section~\ref{sec:simulation-unknown-parameters}. Table~\ref{tab:joint-crt-benchmark-size} and Figure~\ref{fig:oracle-plugin-power-curve} in the main text report estimated size and rejection rates for the oracle CRT, the three plug-in CRTs, and MVN-CRT.

Let $s=|\mc S|$ and $t=|\mc T|=p-s$.
Let $\mathbf{B}\in\mathbb R^{(s+1)\times t}$ and $\bs{\gamma}\in\mathbb R^{s+1}$. 
Independently across \(i=1,\ldots,n\), generate
$$
\bs{X}_{\mc S,i}\sim \mathcal{N}_s(\bs{0},\mathbf{I}_s),\qquad
\bs{\varepsilon}_i\sim \mathcal{N}_t(\bs{0},\mathbf{I}_t),\qquad
\varepsilon_i\sim \mathcal{N}(0,1),
$$
where $\bs{X}_{\mc S,i}$, $\bs{\varepsilon}_i$, and $\varepsilon_i$ are mutually independent, and set
$$
\begin{gathered}
\bs{X}_{\mc T,i}=\mathbf{B}^\top(1,\bs{X}_{\mc S,i}^\top)^\top+\bs{\varepsilon}_i,\\
Y_i=\bs{\gamma}^\top(1,\bs{X}_{\mc S,i}^\top)^\top+\theta \bs{\varepsilon}_i^\top \bs{v}+\varepsilon_i,\\
\bs{v}=t^{-1/2}\bs{1}_t.
\end{gathered}
$$
The case \(\theta=0\) gives $Y\indp \bs{X}_{\mc T}\mid \bs{X}_{\mc S}$ and is used to assess estimated size, whereas \(\theta>0\) gives a dense target effect.

For each combination of $(n,s,t)$, the entries of $\mathbf{B}$ and $\bs{\gamma}$ are drawn independently from \(\mathcal{N}(0,0.7^2/(s+1))\) and \(\mathcal{N}(0,1/(s+1))\), respectively, and the resulting coefficients are held fixed across all simulation replications and values of \(\theta\).

Let \(\mb X_{\mc S}\in\mathbb R^{n\times s}\) and \(\mb X_{\mc T}\in\mathbb R^{n\times t}\) denote the matrices whose \(i\)th rows are $\bs{X}_{\mc S,i}^\top$ and $\bs{X}_{\mc T,i}^\top$, respectively, and let \(\bs Y=(Y_1,\ldots,Y_n)^\top\). Set $\mathbf{D}=(\bs{1}_n,\mb X_{\mc S})\in\mathbb R^{n\times(s+1)}$. 
All experimental settings in Section~\ref{sec:simulation-unknown-parameters} have \(n>s+1\), so $\mathbf{D}$ has full column rank almost surely. 

All randomised covariate copies are generated without using \(\bs Y\); the response enters only when the CRT statistic is evaluated.

\paragraph{Oracle CRT.}
For the oracle CRT, the construction in Section~\ref{sec: CRT} reduces to
$$
\widetilde{\mb X}_{\mc T}^{(m)}
=\mathbf{D}\mathbf{B}+\bs{\varepsilon}^{(m)},\qquad m=1,\ldots,M,
$$
where the entries of each $\bs{\varepsilon}^{(m)}\in\mathbb R^{n\times t}$ are independent standard normal variables. The matrices $\bs{\varepsilon}^{(1)},\ldots,\bs{\varepsilon}^{(M)}$ are independent of one another and of \((\mb X_{\mc T},\bs Y)\) conditional on \(\mb X_{\mc S}\).

\paragraph{Plug-in CRTs based on residual covariance estimates.}
Define the projection matrix
$$
\mathbf{P}_{\mathbf{D}}=\mathbf{D}(\mathbf{D}^\top \mathbf{D})^{-1}\mathbf{D}^\top.
$$
Define the least squares estimator and residual matrix by
$$
\widehat{\mathbf{B}}=(\mathbf{D}^\top \mathbf{D})^{-1}\mathbf{D}^\top\mb X_{\mc T},
\qquad
\mathbf{R}_{\mc T}
=\mb X_{\mc T}-\mathbf{D}\widehat{\mathbf{B}}
=(\mathbf{I}_n-\mathbf{P}_{\mathbf{D}})\mb X_{\mc T}.
$$
Set \(d=n-s-1\), the residual degrees of freedom.
We use the residual covariance estimate
$$
\widehat{\mathbf{\Sigma}}_{\mathrm{MLE}}
=\frac{1}{n}\mathbf{R}_{\mc T}^\top \mathbf{R}_{\mc T},
$$
and its degrees of freedom adjustment 
$$
\widehat{\mathbf{\Sigma}}_{\mathrm{df}}
=\frac{1}{d}\mathbf{R}_{\mc T}^\top \mathbf{R}_{\mc T}.
$$
When $t\leq d$, $\widehat{\mathbf{\Sigma}}_{\mathrm{MLE}}$ is the Gaussian maximum likelihood estimator almost surely. When $t>d$, this estimate is singular, and the Gaussian likelihood has no maximiser over positive-definite covariance matrices. In this case, we continue to call the procedure based on the residual Gram matrix divided by $n$ ``plug-in MLE''; both plug-in procedures sample from the corresponding degenerate Gaussian laws.

The plug-in MLE CRT uses $\widehat{\mathbf{\Sigma}}_{\mathrm{MLE}}$, whereas the plug-in CRT with the degrees of freedom adjustment uses $\widehat{\mathbf{\Sigma}}_{\mathrm{df}}$. Apart from this choice, the two procedures use the same conditional Gaussian law for the randomised copies.
For \(\star\in\{\mathrm{MLE},\mathrm{df}\}\), conditional on $(\mb X_{\mc S},\widehat{\mathbf{B}},\widehat{\mathbf{\Sigma}}_\star)$, generate the rows of \(\widetilde{\mb X}_{\mc T,\star}^{(m)}\), \(m=1,\ldots,M\), independently across \(i\) and \(m\) according to
$$
\widetilde{\bs{X}}_{\mc T,i,\star}^{(m)}
\sim
\mathcal{N}_t\left(
\widehat{\mathbf{B}}^\top(1,\bs{X}_{\mc S,i}^\top)^\top,
\widehat{\mathbf{\Sigma}}_{\star}
 \right),
 \qquad i=1,\ldots,n.
$$
Conditional on $(\mb X_{\mc S},\widehat{\mathbf{B}},\widehat{\mathbf{\Sigma}}_\star)$, the randomised copies are independent of one another and of \((\mb X_{\mc T},\bs Y)\). The observed target matrix, however, need not follow this fitted law conditional on the same quantities, so the observed and randomised matrices need not be exchangeable.
\par

\paragraph{Plug-in CRT based on Ledoit--Wolf covariance shrinkage (LW).}
We use the covariance estimator and variance adjustment from Appendix D.10 of \citet{liu_fast_2022}.
Covariance shrinkage provides a positive definite estimate even when the sample covariance is singular. We estimate the unknown mean by the sample mean and obtain the joint conditional law of the target group by conditioning the fitted Gaussian distribution on the background covariates. The original dCRT comparison tests one target covariate at a time.

Let $\overline{\bs X}\in\mathbb R^p$ be the vector of column means of $\mathbf X$, and define $\mathbf Z=\mathbf X-\bs 1_n\overline{\bs X}^{\top}$, with $i$th row $\bs Z_i^\top$. Define $\widehat{\mathbf\Sigma}_{\mathrm{emp}}=\mathbf Z^\top\mathbf Z/(n-1)$, $\eta=\operatorname{tr}(\widehat{\mathbf\Sigma}_{\mathrm{emp}})/p$, and $\delta^2=p^{-1}\|\widehat{\mathbf\Sigma}_{\mathrm{emp}}-\eta\mathbf I_p\|_F^2$, where $\|\cdot\|_F$ is the Frobenius norm. The shrinkage coefficient and covariance estimate are
$$
\bar b^2=\frac{\sum_{i=1}^n\|\bs Z_i\bs Z_i^\top-\widehat{\mathbf\Sigma}_{\mathrm{emp}}\|_F^2}{p(n-1)^2},\qquad
\rho=\frac{\min(\delta^2,\bar b^2)}{\delta^2},\qquad
\widehat{\mathbf\Sigma}_{\mathrm{LW},0}=(1-\rho)\widehat{\mathbf\Sigma}_{\mathrm{emp}}+\rho\eta\mathbf I_p,
$$
with $\rho=1$ when $\delta^2=0$. For each $j=1,\ldots,p$, let $\widehat{\bs\beta}_j=(\widehat{\mathbf\Sigma}_{\mathrm{LW},0})_{-j,-j}^{-1}(\widehat{\mathbf\Sigma}_{\mathrm{LW},0})_{-j,j}$ and $\widehat v_j=n^{-1}\|\mathbf Z_j-\mathbf Z_{-j}\widehat{\bs\beta}_j\|^2$. The variance adjustment is
$$
a_j=\sqrt{\widehat v_j(\widehat{\mathbf\Sigma}_{\mathrm{LW},0}^{-1})_{jj}},\qquad
\mathbf D_a=\operatorname{diag}(a_1,\ldots,a_p),\qquad
\widehat{\mathbf\Sigma}_{\mathrm{LW}}=\mathbf D_a\widehat{\mathbf\Sigma}_{\mathrm{LW},0}\mathbf D_a.
$$
It gives $(\widehat{\mathbf\Sigma}_{\mathrm{LW}}^{-1})_{jj}^{-1}=\widehat v_j$, so it matches the fitted conditional variance of each coordinate given all the others to the corresponding residual mean square. 

Define $\widehat{\mathbf V}=(\widehat{\mathbf\Sigma}_{\mathrm{LW}})_{\mc T\mc T}-(\widehat{\mathbf\Sigma}_{\mathrm{LW}})_{\mc T\mc S}(\widehat{\mathbf\Sigma}_{\mathrm{LW}})_{\mc S\mc S}^{-1}(\widehat{\mathbf\Sigma}_{\mathrm{LW}})_{\mc S\mc T}$ and $\widehat{\bs\mu}_i=\overline{\bs X}_{\mc T}+(\widehat{\mathbf\Sigma}_{\mathrm{LW}})_{\mc T\mc S}(\widehat{\mathbf\Sigma}_{\mathrm{LW}})_{\mc S\mc S}^{-1}(\bs X_{\mc S,i}-\overline{\bs X}_{\mc S})$. We draw the rows of each randomised copy of the target matrix from $\mathcal{N}_t(\widehat{\bs\mu}_i,\widehat{\mathbf V})$, independently across observations and copies conditional on $\mathbf X$. The fit uses only $\mathbf X$ and is held fixed across copies.

\paragraph{MVN-CRT.}
We use Algorithm~\ref{alg:multi gaussian CRT} with $\mathbf D$ in place of $\mb X_{\mc S}$ to include an intercept, as described in Remark~\ref{rem: general mean for MVN}.
The plug-in MLE CRT and the plug-in CRT with a degrees of freedom adjustment use $\mathbf R_{\mc T}^{\top}\mathbf R_{\mc T}$ to estimate the conditional covariance.
MVN-CRT instead conditions on this Gram matrix when drawing residual matrices and adds back the fixed projection $\mathbf D\widehat{\mathbf B}$.

\paragraph{Test statistic and Monte Carlo p-value.}
We use $T_{\mathrm{ridge}}$ as defined in \eqref{eq:simulation-ridge-gain}, with $\bs R_Y=(\mathbf I_n-\mathbf P_{\mathbf D})\bs Y$. The observed statistic is $T_0=T_{\mathrm{ridge}}(\mathbf R_{\mc T})$. For each of the five procedures, the randomised statistics are
$$
\widetilde{\mathbf R}_{\mc T}^{(m)}
=(\mathbf I_n-\mathbf P_{\mathbf D})\widetilde{\mb X}_{\mc T}^{(m)},
\qquad
T_m=T_{\mathrm{ridge}}\!\left(\widetilde{\mathbf R}_{\mc T}^{(m)}\right).
$$
For every procedure and setting, the ridge penalty is one. We compute the p-value by \eqref{eqn:pval-1} and reject when it is at most $0.05$.

The estimated sizes in Table~\ref{tab:joint-crt-benchmark-size} use $5000$ null replications with $M=999$ randomised copies per replication. For Figure~\ref{fig:oracle-plugin-power-curve}, we use $400$ replications at each $\theta\in\{0,0.25,0.5,0.75,1,1.25,1.5\}$, with $M=199$ copies per replication, including at $\theta=0$.
\endgroup

\section{Additional methods and theoretical results}\label{app:supplementary-results}
Appendix~\ref{app:additional-methods} gives the scoring conditions, randomised p-values and local graphical update conditions for the methods in Sections~\ref{sec: MSCRT} and~\ref{sec: construction}. Appendix~\ref{app:additional-robustness} supplies the technical setting, properties of the enlarged comparison class and the ordinary CRT comparison for the robustness analysis in Section~\ref{sec:mscrt-robustness}. Appendix~\ref{app:gaussian-robustness-bounds} gives Type I error bounds for $G$-CRT with a misspecified graph under the Gaussian covariate model of Section~\ref{sec:ggm-robustness-rewrite}.  Appendix~\ref{app:additional-power} provides additional results for the $G$-CRT power analysis in Section~\ref{sec:gcrt-ridge-power}.

\subsection{Additional methodological details}\label{app:additional-methods}

\subsubsection[Scoring conditions.]{Scoring conditions.}\label{app:scoring-conditions}

Scoring may involve a randomised fitting or tuning procedure, such as fitting a random forest. The random choices used by the procedure may be shared across candidates or generated separately for each candidate. We give conditions on these choices and the scoring rule under which exchangeable candidates yield exchangeable scores under the null hypothesis, extending the validity result in Proposition~\ref{prop: CRT valid}.

To allow fitting and tuning to use all candidates, we describe their scores jointly. 
For the observed candidate and the randomised copies, let
$\mathbf X_{\mc T}^{(0)}=\mathbf X_{\mc T}$ and
$\mathbf X_{\mc T}^{(m)}=\widetilde{\mathbf X}_{\mc T}^{(m)}$
for $m=1,\ldots,M$. 
To define the scoring rule, 
let $\bs z=(\mathbf z_0,\ldots,\mathbf z_M)$ be a tuple of candidate matrices $\mathbf z_v\in\mathbb R^{n\times|\mc T|}$, $v=0,\ldots,M$, and 
let $W_{\mathrm s}\sim\nu_{\mathrm s}$ collect the random choices used to
compute the scores, including any randomness in fitting or tuning.
Let $\mathscr T(\bs y,\mathbf x_{\mc S},\bs z;w_{\mathrm s})
\in\mathbb R^{M+1}$ be a jointly measurable scoring rule
that returns one score for each candidate, where $\bs y$ and
$\mathbf x_{\mc S}$ are response and background values and
$w_{\mathrm s}$ is a realised value of $W_{\mathrm s}$.

In the following, we describe the conditions for $\mathscr T(\bs y,\mathbf x_{\mc S},\bs z;w_{\mathrm s})$. 

The scores must follow the same permutation as the candidates and their assigned random choices, with shared random choices unchanged.
For a permutation $\pi$ of $\{0,\ldots,M\}$, define
$\pi\bs z=(\mathbf z_{\pi(0)},\ldots,\mathbf z_{\pi(M)})$
and permute score vectors in the same way.
If all auxiliary randomness is shared across the candidates, define
$\pi w_{\mathrm s}=w_{\mathrm s}$ for every $\pi$.
If all auxiliary randomness is assigned separately to the candidates,
let $W_{\mathrm s,v}$ collect the random choices associated with
candidate $v$, $v=0,\ldots,M$, and let
$W_{\mathrm s}=(W_{\mathrm s,0},\ldots,W_{\mathrm s,M})$
have independent identically distributed components and define
$\pi w_{\mathrm s}=(w_{\mathrm s,\pi(0)},\ldots,w_{\mathrm s,\pi(M)})$.
When shared and separate auxiliary randomness are both used, write
$W_{\mathrm s}=(W_{\mathrm s,\mathrm{sh}},
W_{\mathrm s,0},\ldots,W_{\mathrm s,M})$, where
$W_{\mathrm s,\mathrm{sh}}$ collects the shared random choices and
$W_{\mathrm s,v}$ collects the random choices associated with candidate
$v$, $v=0,\ldots,M$.
For every permutation $\pi$, define
$\pi w_{\mathrm s}=(w_{\mathrm s,\mathrm{sh}},
w_{\mathrm s,\pi(0)},\ldots,w_{\mathrm s,\pi(M)})$.

We impose the following conditions on the entire $W_{\mathrm s}$ so that the candidates and their scoring randomness are jointly invariant under candidate permutations.
Let $W_{\mathrm c}$ collect all random choices used to generate
$(\widetilde{\mathbf X}_{\mc T}^{(m)})_{m=1}^{M}$.
For every $\pi$, we require
\begin{equation}\label{eq:mscrt-scoring-randomness-conditions}
\begin{aligned}
W_{\mathrm s}&\indp(\bs Y,\mathbf X,W_{\mathrm c}),\\
\mathcal L(\pi W_{\mathrm s})&=\mathcal L(W_{\mathrm s})=\nu_{\mathrm s}.
\end{aligned}
\end{equation}
In the mixed case, this invariance concerns the joint distribution
of the entire tuple; identical marginal distributions of the separate
components alone do not suffice, and the shared and separate components
may be dependent.
Fitting and tuning may depend on the response, the background covariates and all candidates. We require the scoring rule to treat the observed candidate and the randomised copies symmetrically: for every $\pi$ and all arguments,
\begin{equation}\label{eq:mscrt-randomised-scoring-equivariance}
\mathscr T(\bs y,\mathbf x_{\mc S},\pi\bs z;\pi w_{\mathrm s})
=\pi\mathscr T(\bs y,\mathbf x_{\mc S},\bs z;w_{\mathrm s}).
\end{equation}
Define
$(T_0,\ldots,T_M)=\mathscr T(\bs Y,\mathbf X_{\mc S},
(\mathbf X_{\mc T}^{(v)})_{v=0}^{M};W_{\mathrm s})$.

For deterministic scoring, take $W_{\mathrm s}$ to be constant and omit
the argument $w_{\mathrm s}$; then
\eqref{eq:mscrt-randomised-scoring-equivariance} reduces to
\begin{equation}\label{eq:mscrt-scoring-equivariance}
\mathscr T(\bs y,\mathbf x_{\mc S},\pi\bs z)
=\pi\mathscr T(\bs y,\mathbf x_{\mc S},\bs z).
\end{equation}

The scoring rule defined by $T$ and $H$ in Section~\ref{sec: MSCRT} satisfies \eqref{eq:mscrt-scoring-equivariance}, since $H$ is unchanged by candidate permutations and the same $T$ is applied to every candidate.
More generally, under the assumptions on the copies and the null hypothesis in Proposition~\ref{prop: CRT valid}, conditions \eqref{eq:mscrt-scoring-randomness-conditions} and \eqref{eq:mscrt-randomised-scoring-equivariance} imply $\mathbb P(\mathrm{pVal}_T\leq\alpha)\leq\alpha$ for every $\alpha\in(0,1)$.
Appendix~\ref{app:proof-prop-crt-valid} proves this extension.

\subsubsection{Randomised p-values}\label{app:pvalue-conventions}
The conservative p-value in \eqref{eqn:pval-1} can be randomised to reduce conservatism. For the observed score $T_0$ and copy scores $T_1,\ldots,T_M$, define $A=\sum_{m=1}^M\mathbf1\{T_m>T_0\}$ and $S=1+\sum_{m=1}^M\mathbf1\{T_m=T_0\}$, where the $1$ accounts for the observation itself. Any rounding is applied by the same rule to all candidates before these counts are computed. Let $U\sim\operatorname{Unif}(0,1)$ be independent of the data and all randomness used to generate the copies and scores, and define $J=\lceil US\rceil$. Conditional on the scores, $J$ is uniform on $\{1,\ldots,S\}$. Define the conservative, discrete randomised and continuous randomised p-values as follows:
\begin{align}
 p_{\mathrm{cons}}&=\frac{A+S}{M+1},\label{eq:crt-pvalue-conservative}\\
 p_{\mathrm{disc}}&=\frac{A+J}{M+1},\label{eq:crt-pvalue-discrete}\\
 p_{\mathrm{cont}}&=\frac{A+US}{M+1}.\label{eq:crt-pvalue-continuous}
\end{align}
When the scores are exchangeable under the null hypothesis, $p_{\mathrm{cons}}$ is valid by Proposition~\ref{prop: CRT valid}, $p_{\mathrm{disc}}$ is uniform on $\{1/(M+1),\ldots,1\}$, and $p_{\mathrm{cont}}$ is uniform on $(0,1)$. Unlike the two discrete rules, the continuous rule permits p-values below $1/(M+1)$.

\subsubsection{General graphical updates and conditional exchangeability}\label{app:general-graphical-updates}

This subsection gives the general local conditions and a finite-state example for the extension of Algorithm~\ref{alg: exchangeable} described in Section~\ref{sec:general-graphical-models}. Let the rows of $\mathbf X$ be i.i.d. and assume \eqref{eq:local Markov property}. Suppose the family of conditional distributions for each node admits a sufficient statistic $\Psi_i(\mathbf X_i;\mathbf X_{N_i})$.

\begin{condition}\label{cond:local-update}\label{cond:gcrt-local}\label{condition:local update}
For each $i\in\mc T$, a fixed, parameter-free local kernel replaces only column $i$ of its current input $\mathbf X$ by $\mathbf X_i^{\mathrm{new}}$ and satisfies the following properties under the covariate model:
\begin{enumerate}
\item $\Psi_i(\mathbf X_i^{\mathrm{new}};\mathbf X_{N_i})=\Psi_i(\mathbf X_i;\mathbf X_{N_i})=\Psi_i$ almost surely.
\item Conditional on $(\mathbf X_{-i},\Psi_i)$, the pair $(\mathbf X_i,\mathbf X_i^{\mathrm{new}})$ is exchangeable almost surely.
\end{enumerate}
The neighbourhood values and local statistic are computed from the current state at each call, and each call uses fresh independent randomness as specified in Algorithm~\ref{alg: exchangeable}.
\end{condition}

\begin{proposition}\label{prop: exchangeable}
Suppose \eqref{eq:local Markov property} holds, Condition~\ref{cond:local-update} holds, and the rows of $\mathbf X$ are i.i.d. 
Apply Algorithm~\ref{alg: exchangeable} under its requirements for fixed choices and randomness. Then:
\begin{enumerate}
\item \textit{Conditional exchangeability.} Every copy satisfies $\widetilde{\mathbf X}^{(m)}_{\mc S}=\mathbf X_{\mc S}$, and
$$
\bigl(\mathbf X_{\mc T},\widetilde{\mathbf X}^{(1)}_{\mc T},\ldots,
\widetilde{\mathbf X}^{(M)}_{\mc T}\bigr)
$$
is exchangeable conditional on $\mathbf X_{\mc S}$.

\item \textit{Validity.} Under $H_0:\bs Y\indp\mathbf X_{\mc T}\mid\mathbf X_{\mc S}$, for scoring rules satisfying the conditions in Appendix~\ref{app:scoring-conditions}, the p-value in \eqref{eqn:pval-1} satisfies $\mathbb P(\mathrm{pVal}_T\leq\alpha)\leq\alpha$ for every $\alpha\in(0,1)$.

\end{enumerate}

\end{proposition}

The following corollary weakens the graph assumption in Proposition~\ref{prop:ggm-gcrt-valid}. For Gaussian covariates, the graph need only contain every true edge incident to a target node.

\begin{corollary}[A weaker graph assumption for Gaussian covariates]
\label{cor:ggm-weaker-graph-valid}
Suppose the rows of $\mathbf X$ are i.i.d. from
$\mathcal N_p(\bs\mu,\mathbf\Sigma)$ with $\mathbf\Sigma\succ0$,
and let $\mathbf\Omega=\mathbf\Sigma^{-1}$.
For each $i\in\mc T$, assume that $N_i$ contains every
$j\in[p]\setminus\{i\}$ with $\omega_{ij}\ne0$.
Apply Algorithm~\ref{alg: exchangeable} with the residual rotations
in \eqref{eq:ggm-local-update}, under its requirements on fixed
choices and independent randomisation.
Then the conclusions of Proposition~\ref{prop:ggm-gcrt-valid} hold.
Under $H_0$, the p-value assertion also holds for scoring rules
satisfying the conditions in Appendix~\ref{app:scoring-conditions}.
\end{corollary}

The following example applies Proposition~\ref{prop: exchangeable} to finite-state graphical models, with local updates constructed by permutations.

\begin{example}[Finite-state graphical models]\label{ex:gcrt-finite-state}
Assume \eqref{eq:local Markov property} and let $X_i$ take values in $\{1,\ldots,K_i\}$. For a target node $i$ and a neighbour configuration $\bs v$, define
$$
\mc H_{\bs v}=\{j\in[n]:\bs X_{j,N_i}=\bs v\},\qquad
n_{\bs v,k}=\sum_{j\in\mc H_{\bs v}}\mathbf1\{X_{ji}=k\}.
$$
The table of counts $\Psi_i=\bigl(n_{\bs v,k}\bigr)_{\mc H_{\bs v}\ne\varnothing,\,k\in[K_i]}$ is sufficient for the nodewise conditional law.
Conditional on $(\mathbf X_{-i},\Psi_i)$, all arrangements with these counts are equally likely. Draw an independent uniform permutation $\pi_{\bs v}$ within each such set and set
$$
X_{ji}^{\mathrm{new}}=X_{\pi_{\bs v}(j),i},\qquad j\in\mc H_{\bs v}.
$$
Proposition~\ref{prop: exchangeable} then gives a valid $G$-CRT. A set permits a change only if it contains at least two distinct observed values; singleton or constant sets give identity updates. Many small sets can therefore leave little randomisation.
\end{example}

Appendix~\ref{app:proof-prop-exchangeable} contains the proofs of the results in this subsection and of Proposition~\ref{prop:ggm-gcrt-valid} in the main text.

\subsection{Robustness and the enlarged comparison class}\label{app:additional-robustness}

We give the technical setting for Theorem~\ref{thm: MS-CRT robustness}, then establish convexity and mixture properties of the enlarged comparison class in \eqref{eq:Phi-comparison-class-at-fixed-XS}. 
We then compare the robustness of MS-CRT and the ordinary CRT.

\subsubsection{Technical setting for Theorem~\ref{thm: MS-CRT robustness}.}
The comparison in Theorem~\ref{thm: MS-CRT robustness} changes the target law while keeping the copy construction and scoring rule fixed. Fix $\mathbf x_{\mc S}$ and use the working laws $P_{\theta,\mathbf x_{\mc S}}$, $\theta\in\Theta$, defined in Section~\ref{sec:mscrt-robustness}.
Let $(\mathcal X,\mathscr X)$ be the common measurable space on which the working laws are defined.
Let $\Phi=\phi(\mathbf X_{\mc T};\mathbf x_{\mc S})$ be the sufficient statistic used by MS-CRT.
Let $\mc V(\mathbf x_{\mc S})$ be a subset of the range of $\phi(\cdot;\mathbf x_{\mc S})$.
Assume that $\{\mathbf z\in\mathcal X:\phi(\mathbf z;\mathbf x_{\mc S})\in\mc V(\mathbf x_{\mc S})\}\in\mathscr X$
and $P_{\theta,\mathbf x_{\mc S}}(\Phi\in\mc V(\mathbf x_{\mc S}))=1$ for every $\theta\in\Theta$.
Equip $\mc V(\mathbf x_{\mc S})$ with the $\sigma$-algebra
$$
\mathscr V(\mathbf x_{\mc S})
:=
\left\{
B\subseteq\mc V(\mathbf x_{\mc S}):
\{\mathbf z\in\mathcal X:\phi(\mathbf z;\mathbf x_{\mc S})\in B\}\in\mathscr X
\right\}.
$$

Assume that there is a collection of probability laws $\{L_{\mathbf x_{\mc S},\varphi}:\varphi\in\mc V(\mathbf x_{\mc S})\}$ on $(\mathcal X,\mathscr X)$ such that, for every $A\in\mathscr X$, the map $\varphi\mapsto L_{\mathbf x_{\mc S},\varphi}(A)$ is $\mathscr V(\mathbf x_{\mc S})$-measurable, and

$$
L_{\mathbf x_{\mc S},\varphi}
\!\left(
\{\mathbf z:\phi(\mathbf z;\mathbf x_{\mc S})=\varphi\}
\right)=1,
\qquad
\varphi\in\mc V(\mathbf x_{\mc S}).
$$

For every $\theta\in\Theta$, a version of the conditional law under $P_{\theta,\mathbf x_{\mc S}}$ satisfies

$$
\mathcal L(\mathbf X_{\mc T}\mid\Phi)
=
L_{\mathbf x_{\mc S},\Phi},
\qquad P_{\theta,\mathbf x_{\mc S}}\text{-a.s.}
$$

Let $W_{\mathrm c}\sim\nu_{\mathrm c}$ supply the randomness for generating copies, and let $W_{\mathrm s}\sim\nu_{\mathrm s}$ be the auxiliary scoring randomness specified in Appendix~\ref{app:scoring-conditions}.
Define $W=(W_{\mathrm c},W_{\mathrm s})$ and $\nu=\nu_{\mathrm c}\otimes\nu_{\mathrm s}$, with $W\sim\nu$ and $W\indp(\bs Y,\mathbf X)$.
Represent the copy construction by a jointly measurable map $Q=(Q_1,\ldots,Q_M)$, where

$$
\widetilde{\mathbf X}_{\mc T}^{(m)}
=
Q_m(\mathbf X_{\mc S},\mathbf X_{\mc T},W_{\mathrm c}),
\qquad m\in[M].
$$

Assume that, for every $\varphi\in\mc V(\mathbf x_{\mc S})$, if $\mathbf Z\sim L_{\mathbf x_{\mc S},\varphi}$ and $W_{\mathrm c}\sim\nu_{\mathrm c}$ are independent, then
\begin{equation}
\label{eq:Q-exchangeability-at-fixed-XS}
\left(
\mathbf Z,
Q_1(\mathbf x_{\mc S},\mathbf Z,W_{\mathrm c}),
\ldots,
Q_M(\mathbf x_{\mc S},\mathbf Z,W_{\mathrm c})
\right)
\quad\text{is exchangeable}.
\end{equation}

The scoring rule $\mathscr T$ is measurable.
For deterministic scoring, require \eqref{eq:mscrt-scoring-equivariance} and take $W_{\mathrm s}$ to be constant.
For randomised scoring, require \eqref{eq:mscrt-randomised-scoring-equivariance}, with $W_{\mathrm s}$ satisfying \eqref{eq:mscrt-scoring-randomness-conditions}.
For fixed $(\bs y,\mathbf x_{\mc S},\mathbf x_{\mc T},w)$, with $w=(w_{\mathrm c},w_{\mathrm s})$, define
\begin{equation}\label{eq:robustness-score-array}
\begin{aligned}
\mathbf z_0&=\mathbf x_{\mc T},\qquad
\mathbf z_m=Q_m(\mathbf x_{\mc S},\mathbf x_{\mc T},w_{\mathrm c}),\quad m\in[M],\\
(t_0,\ldots,t_M)&=
\mathscr T(\bs y,\mathbf x_{\mc S},(\mathbf z_0,\ldots,\mathbf z_M);w_{\mathrm s}).
\end{aligned}
\end{equation}
Let $p_T(\bs y,\mathbf x_{\mc S},\mathbf x_{\mc T},w)$ be the p-value in \eqref{eqn:pval-1} computed from the scores in \eqref{eq:robustness-score-array}.
The same $(Q,\nu,\mathscr T)$ is used under the true law and every $R\in\widetilde{\mc P}_{\Phi}(\mathbf x_{\mc S})$.
Under $R$, evaluate $p_T(\bs y,\mathbf x_{\mc S},\mathbf Z,W)$ with $(\mathbf Z,W)\sim R\otimes\nu$.
For every $\alpha\in(0,1)$, the exchangeability and scoring conditions give rejection probability at most $\alpha$ under $R$.
The proof of Theorem~\ref{thm: MS-CRT robustness} in Appendix~\ref{app:proof-mscrt-robustness} then bounds the change in rejection probability by the total variation distance between the true target law and $R$.

\subsubsection{Convexity and mixtures of the enlarged comparison class}\label{app:comparison-class-properties}

We use the working family $\mc P_{\Theta}(\mathbf x_{\mc S})$ and enlarged comparison class $\widetilde{\mc P}_{\Phi}(\mathbf x_{\mc S})$ defined in Section~\ref{sec:mscrt-robustness}, with the common conditional kernel from the technical setting for Theorem~\ref{thm: MS-CRT robustness}.
The first result concerns finite mixtures of laws in the enlarged comparison class; the second concerns arbitrary mixtures of the original working laws. Both preserve the common conditional kernel, even when the resulting law lies outside the working family.

\begin{proposition}\label{prop:convexity-of-enlarged-class}
Fix \(\mathbf x_{\mc S}\). Let \(J\geq 1\) be an integer, and let
\(R_1,\ldots,R_J\in
\widetilde{\mc P}_{\Phi}(\mathbf x_{\mc S})\). For nonnegative weights
\(\pi_1,\ldots,\pi_J\) satisfying \(\sum_{j=1}^J\pi_j=1\), define
\(R^\pi=\sum_{j=1}^J\pi_jR_j\). Then
\(R^\pi\in\widetilde{\mc P}_{\Phi}(\mathbf x_{\mc S})\). Hence
\(\widetilde{\mc P}_{\Phi}(\mathbf x_{\mc S})\) is convex.
In particular, since \({\mc P}_{\Theta}(\mathbf x_{\mc S})\subseteq
\widetilde{\mc P}_{\Phi}(\mathbf x_{\mc S})\),
the convex hull 
$\operatorname{co}\!\left(
{\mc P}_{\Theta}(\mathbf x_{\mc S})
\right)
\subseteq
\widetilde{\mc P}_{\Phi}(\mathbf x_{\mc S})$.
\end{proposition}

The proof is given in Appendix~\ref{app:proof-convexity-enlarged-class}.

\begin{proposition}\label{prop:mixtures-in-enlarged-class}
Fix \(\mathbf x_{\mc S}\). Assume that, for every measurable set \(A\), the
map \(\theta\mapsto P_{\theta,\mathbf x_{\mc S}}(A)\) is measurable. Let
\(\Pi_{\mathbf x_{\mc S}}\) be a probability law on \(\Theta\) and define
$$
P^\Pi_{\mathbf x_{\mc S}}(\cdot)
=
\int_\Theta
P_{\theta,\mathbf x_{\mc S}}(\cdot)\,
\Pi_{\mathbf x_{\mc S}}(d\theta)
$$
as a probability law for \(\mathbf X_{\mc T}\). Then
\(
P^\Pi_{\mathbf x_{\mc S}}
\in
\widetilde{\mc P}_{\Phi}(\mathbf x_{\mc S})
\).
\end{proposition}

The proof is given in Appendix~\ref{app:proof-mixtures-enlarged-class}.

\subsubsection{Comparison to the ordinary CRT}\label{app:ordinary-crt-comparison}

For comparison, we study how approximation of the conditional law of $\mathbf X_{\mc T}$ given $\mathbf X_{\mc S}$ affects the Type I error of an ordinary CRT. At a fixed $\mathbf x_{\mc S}$, write the true conditional law as
$\mathcal L_*(\mathbf X_{\mc T}\mid
\mathbf X_{\mc S}=\mathbf x_{\mc S})$, and let
$P_{\mathbf x_{\mc S}}$ denote the working law used to generate the CRT copies. The working law may differ from the true conditional law and is constructed independently of the labelled data. If it is estimated from an auxiliary sample, the bounds in Theorem~\ref{thm: multivariate CRT robustness} are understood conditionally on that sample.

When the observed \(\mathbf X_{\mc T}\) has conditional law
\(\mathcal L_*(\mathbf X_{\mc T}\mid
\mathbf X_{\mc S}=\mathbf x_{\mc S})\) and the CRT copies are independent
draws from \(P_{\mathbf x_{\mc S}}\), Theorems~4 and~5 of
\citet{berrett_conditional_2020} specialise to the following bounds on
conditional Type I error inflation.

\begin{theorem}\label{thm: multivariate CRT robustness}
Assume
$H_0:\bs{Y}\indp\mathbf X_{\mc T}\mid\mathbf X_{\mc S}$.
For a fixed integer $M\geq1$, suppose that, conditional on
$\mathbf X_{\mc S}=\mathbf x_{\mc S}$, the copies
$\widetilde{\mathbf X}_{\mc T}^{(1)},\ldots,
\widetilde{\mathbf X}_{\mc T}^{(M)}$ are independent draws from
$P_{\mathbf x_{\mc S}}$ and are generated independently of
$(\bs{Y},\mathbf X_{\mc T})$. For any statistic $T(\bs y,\mathbf x_{\mc T},\mathbf x_{\mc S})$ applied separately to each candidate, define
$\mathrm{pVal}_T$ by \eqref{eqn:pval-1}. Then, for almost every
$(\bs{y},\mathbf x_{\mc S})$ under the true joint law of
$(\bs{Y},\mathbf X_{\mc S})$ and every $\alpha\in(0,1)$,
$$
\begin{aligned}
&\mathbb P_*\!\left(
\mathrm{pVal}_T\leq\alpha
\mid
(\bs{Y}=\bs{y},\mathbf X_{\mc S}=\mathbf x_{\mc S})
\right)\\
\leq{}&
\alpha+
d_{\mathrm{TV}}\!\left(
\mathcal L_*(\mathbf X_{\mc T}\mid
\mathbf X_{\mc S}=\mathbf x_{\mc S}),
P_{\mathbf x_{\mc S}}
\right).
\end{aligned}
$$
Moreover, there exists a statistic $\widetilde T$ for which, for almost
every such $(\bs{y},\mathbf x_{\mc S})$, as $M\to\infty$,
\begin{equation}\label{eq:ordinary-CRT-worst-case-inflation-lower-bound}
\begin{aligned}
&\sup_{\alpha\in(0,1)}
\left\{
\mathbb P_*\!\left(
\mathrm{pVal}_{\widetilde T}\leq\alpha
\mid
(\bs{Y}=\bs{y},\mathbf X_{\mc S}=\mathbf x_{\mc S})
\right)-\alpha
\right\}\\
\geq{}&
d_{\mathrm{TV}}\!\left(
\mathcal L_*(\mathbf X_{\mc T}\mid
\mathbf X_{\mc S}=\mathbf x_{\mc S}),
P_{\mathbf x_{\mc S}}
\right)
-\frac{1}{2}(1+o(1))\sqrt{\frac{\log M}{M}}.
\end{aligned}
\end{equation}
\end{theorem}

Let $\Delta_M=\sup_T\sup_{\alpha\in(0,1)}\{\mathbb P_*(\mathrm{pVal}_T\leq\alpha\mid(\bs{Y}=\bs{y},\mathbf X_{\mc S}=\mathbf x_{\mc S}))-\alpha\}$ denote the worst-case conditional excess Type I error over test statistics and significance levels. \Cref{thm: multivariate CRT robustness} shows that, for sufficiently large \(M\), \(\Delta_M\) is close to the total variation distance between the true and working conditional laws. Thus, small \(\Delta_M\) requires this distance to be small.
This condition concerns the conditional law of the entire
$n\times|\mc T|$ target matrix rather than the conditional law for a single observation. Under factorisation across rows, Kullback--Leibler divergences for individual observations add across the $n$ rows, while the number of parameters in each conditional model may also increase with $|\mc T|$. 
\par

To illustrate the implication of this condition, suppose that the true conditional law for each observation belongs to a correctly specified family with parameter dimension $d$, and that $P_{\mathbf x_{\mc S}}$ is estimated from $N$ independent unlabelled covariate observations. Assume that both the true and fitted conditional laws factor across rows and that the average Kullback--Leibler divergence between their conditional laws for individual rows, evaluated at the observed background covariates, is $O_{\mathbb P}(d/N)$. Additivity and Pinsker's inequality then give
$$
d_{\mathrm{TV}}\!\left(
\mathcal L_*(\mathbf X_{\mc T}\mid
\mathbf X_{\mc S}=\mathbf x_{\mc S}),
P_{\mathbf x_{\mc S}}
\right)
=
O_{\mathbb P}\!\left(\sqrt{\frac{nd}{N}}\right).
$$
The upper bound in \Cref{thm: multivariate CRT robustness} controls $\Delta_M$ at the same rate, so $N\gg nd$ is a sufficient scaling for this bound to be small. This recovers $N\gg nk$ for a univariate target with $k$ parameters, as discussed by \citet{berrett_conditional_2020}.
In the multivariate case, if the target coordinates are conditionally independent given $\mathbf X_{\mc S}$ and each coordinate has a separate conditional model with $k$ parameters, then $d=tk$ and this sufficient scaling becomes $N\gg ntk$.

When the working family is correctly specified, MS-CRT removes the covariate parameters by conditioning on the sufficient statistic.
It therefore provides nominal Type I error control without an auxiliary sample for parameter estimation.

Under misspecification, MS-CRT measures approximation error by the distance to the enlarged comparison class, which contains the ordinary CRT's working law whenever $P_{\mathbf x_{\mc S}}\in{\cal P}_\Theta(\mathbf x_{\mc S})$. The comparison in \eqref{eq:MSCRT-working-family-bound-comparison}, together with \Cref{thm: multivariate CRT robustness}, therefore gives $\Delta_\Phi(\mathbf x_{\mc S})\leq\lim_{M\to\infty}\Delta_M$. Thus \Cref{thm: MS-CRT robustness} provides, for every $M$, a bound on the conditional Type I error inflation of MS-CRT that is no larger than the ordinary CRT's limiting worst-case conditional inflation.

\subsection[Robustness to graph misspecification]{Robustness to graph misspecification}\label{app:gaussian-robustness-bounds}

We study the effect of graph misspecification on the Type I error of $G$-CRT with Gaussian covariates. We first give a bound for Gaussian conditional laws, then apply it to graph selection.

\makeatletter
\@ifundefined{r@eq:gaussian-misspecification-size-supp}{}{
  \expandafter\let\expandafter\ggmrobustnesshistoricaleq
    \csname r@eq:gaussian-misspecification-size-supp\endcsname
  \expandafter\let\csname r@eq:gaussian-misspecification-size\endcsname
    \ggmrobustnesshistoricaleq
  \protected@write\@auxout{}{
    \string\newlabel{eq:gaussian-misspecification-size}{\ggmrobustnesshistoricaleq}
  }
}
\makeatother

\subsubsection[A Type I error bound for Gaussian conditional laws]{A Type I error bound for Gaussian conditional laws.}
\label{app:gaussian-conditional-laws}

Suppose the $n$ observations are independent and identically distributed, with the true conditional law for each observation given by $\bs X_{\mc T}\mid\bs X_{\mc S}=\bs x_{\mc S}\sim \mathcal{N}_t(\bs\mu(\bs x_{\mc S}),\mathbf V)$, where $t=|\mc T|$ and $\mathbf V\succ0$. Choose a reference law $\mathcal{N}_t(\bs\mu_0(\bs x_{\mc S}),\mathbf V_0)$ with $\mathbf V_0\succ0$, whose product over the $n$ rows belongs to $\widetilde{\mc P}_{\Phi}(\mathbf x_{\mc S})$ for almost every $\mathbf x_{\mc S}$ under the true law. Define $\bs b(\bs x_{\mc S})=\bs\mu(\bs x_{\mc S})-\bs\mu_0(\bs x_{\mc S})$ and assume that
\begin{equation}\label{eq:gaussian-covariate-discrepancy-supp}
B=\operatorname{tr}(\mathbf V_0^{-1}\mathbf V)-t
+\log\frac{\det\mathbf V_0}{\det\mathbf V}
+\mathbb E_*\!\left[\bs b(\bs X_{\mc S})^\top\mathbf V_0^{-1}\bs b(\bs X_{\mc S})\right]<\infty.
\end{equation}
The quantity $B$ in \eqref{eq:gaussian-covariate-discrepancy-supp} is twice the expected Kullback--Leibler divergence for one observation. Under $H_0$, additivity across the $n$ observations, Pinsker's inequality and Jensen's inequality give
\begin{equation}\label{eq:gaussian-misspecification-size-supp}
\mathbb P_*\!\left(p_T(\bs Y,\mathbf X_{\mc S},\mathbf X_{\mc T},W)\leq\alpha\right)
\leq\min\left\{1,\alpha+\frac{\sqrt{nB}}{2}\right\}.
\end{equation}

The proof is given in Appendix~\ref{app:proof-gaussian-discrepancy}.

For the graphical results in Appendix~\ref{app:thresholded-graph-selection}, assume the joint Gaussian covariate model of Section~\ref{sec:ggm-robustness-rewrite}. Let
$\mathbf\Omega=\mathbf\Sigma^{-1}$ be the true precision matrix,
and define $\mathbf B=\mathbf\Omega_{\mc T\mc S}$.

\subsubsection[Type I error after thresholded graph selection]{Type I error after thresholded graph selection.}\label{app:thresholded-graph-selection}\label{app:clime-graph-selection}
\Cref{cor:learned-graph-threshold} converts an entrywise error bound
for the estimated precision matrix into a Type I error bound.
Proposition~\ref{prop:clime-estimation-guarantee} supplies the estimation
bound for CLIME; combining the two gives sufficient auxiliary sample sizes.

For a matrix $\mathbf F$, define
$\|\mathbf F\|_{\max}=\max_{i,j}|F_{ij}|$.

\begin{corollary}\label{cor:learned-graph-threshold}
In the setting of \Cref{prop:learned-graph-size-rewrite}, suppose a symmetric estimator from $\mathcal U$ and a deterministic $r_u\geq0$ satisfy
\begin{equation}\label{eq:learned-graph-estimation-event}
\mathbb P_*(\|\widehat{\mathbf\Omega}-\mathbf\Omega\|_{\max}\leq r_u)\geq1-\delta.
\end{equation}
Include an edge $(i,j)$ in $\widehat G$ exactly when $i\ne j$ and $|\widehat\Omega_{ij}|>r_u$. Suppose every nonzero off-diagonal entry of $\mathbf\Omega_{\mc T\mc T}$ exceeds $2r_u$ in absolute value. Let $k_u$ count the true edges between $\mc T$ and $\mc S$ with $0<|\Omega_{ij}|\leq2r_u$. If $\lambda_{\min}(\mathbf\Omega_{\mc T\mc T})\geq\kappa_T>0$ and $\lambda_{\max}(\mathbf\Sigma_{\mc S\mc S})\leq \Lambda_S$, then
\begin{equation}\label{eq:learned-graph-threshold-size}
\mathbb P_*(\mathrm{pVal}_T\leq\alpha)
\leq\min\left\{1,\alpha+(1-\alpha)\delta
+(1-\delta)r_u\sqrt{\frac{n \Lambda_S k_u}{\kappa_T}}\right\}.
\end{equation}
If every true edge incident to $\mc T$ corresponds to an entry of the precision matrix with absolute value greater than $2r_u$, the bound is $\alpha+(1-\alpha)\delta$, without the spectral assumptions.
\end{corollary}

On the event $\|\widehat{\mathbf\Omega}-\mathbf\Omega\|_{\max}\leq r_u$,
every true edge with $|\Omega_{ij}|>2r_u$ is selected. Thus $k_u$
counts the weak true edges between $\mc T$ and $\mc S$ that may be omitted.
Unlike a node's degree, which counts all incident edges regardless of
strength, $k_u$ counts only weak true edges between the two groups.

The proof is given in Appendix~\ref{app:proof-learned-graph-threshold}.

For sequences satisfying \Cref{cor:learned-graph-threshold},
\eqref{eq:learned-graph-threshold-size} gives size at most
$\alpha+o(1)$ if $\delta\to0$ and
$n \Lambda_S k_u r_u^2/\kappa_T\to0$.

To apply \Cref{cor:learned-graph-threshold}, we use CLIME
\citep{cai_constrained_2011} to obtain the estimation guarantee
\eqref{eq:learned-graph-estimation-event} with an explicit Gaussian error radius.

Let $\widehat{\mathbf\Sigma}$ be the unbiased sample covariance of
$n_u\geq2$ auxiliary observations. CLIME estimates the columns of
$\mathbf\Omega$ through the identity $\mathbf\Sigma\mathbf\Omega=\mathbf I_p$.

\medskip
For known bounds $\max_i\Sigma_{ii}\leq v$ and
$\|\mathbf\Omega\|_1:=\max_j\sum_i|\Omega_{ij}|\leq K$,
and a failure probability $\delta\in(0,1)$, define
\begin{equation}\label{eq:learned-graph-clime-radius}
\ell_u=\log(4p^2/\delta),\qquad
a_u=2v\left\{\sqrt{\frac{\ell_u}{n_u-1}}+\frac{\ell_u}{n_u-1}\right\},
\qquad \lambda_u=K a_u,\qquad r_u=4K^2a_u.
\end{equation}

For each $j\in[p]$, choose $\bs b\in\mathbb R^p$ to minimise $\|\bs b\|_1$ subject to $\|\widehat{\mathbf\Sigma}\bs b-\bs e_j\|_\infty\leq\lambda_u$, where $\bs e_j$ is the $j$th coordinate vector in $\mathbb R^p$. If there are multiple minimisers, choose the one with minimum Euclidean norm.
Collect the selected vectors as the columns of a matrix $\widetilde{\mathbf\Omega}$.

We obtain a symmetric estimate $\widehat{\mathbf\Omega}$ from these separate column estimates. Its diagonal entries are those of $\widetilde{\mathbf\Omega}$. For each $i<j$, choose the entry of smaller absolute value from $\widetilde{\Omega}_{ij}$ and $\widetilde{\Omega}_{ji}$, and assign that value to both $\widehat\Omega_{ij}$ and $\widehat\Omega_{ji}$. If the absolute values agree, use a fixed deterministic rule to choose between the two entries.

If any optimisation is infeasible, set the diagonal of $\widehat{\mathbf\Omega}$ to zero and all its off-diagonal entries to $2r_u$, so the threshold rule returns the complete graph.

\begin{proposition}[CLIME estimation guarantee]
\label{prop:clime-estimation-guarantee}
Suppose the auxiliary sample contains $n_u\geq2$ i.i.d. observations
from $\mathcal N_p(\bs\mu,\mathbf\Sigma)$, where
$\mathbf\Sigma\succ0$, $\max_i\Sigma_{ii}\leq v$, and
$\|\mathbf\Omega\|_1\leq K$ with $\mathbf\Omega=\mathbf\Sigma^{-1}$.
Let $\widehat{\mathbf\Omega}$ be the estimator constructed as above.
For any $\delta\in(0,1)$, choose $a_u$, $\lambda_u$, and $r_u$ as in
\eqref{eq:learned-graph-clime-radius}.
Then \eqref{eq:learned-graph-estimation-event} holds.
\end{proposition}
Appendix~\ref{app:proof-learned-graph-clime} gives the proof.

We can combine Proposition~\ref{prop:clime-estimation-guarantee} with
\Cref{cor:learned-graph-threshold} to express the Type I error bound
in terms of the auxiliary sample size $n_u$.
To obtain $\mathbb P_*(\mathrm{pVal}_T\leq\alpha)\leq\alpha+\varepsilon$,
fix $\alpha\in(0,1)$ and $0<\varepsilon<1-\alpha$.
In the setting of Proposition~\ref{prop:learned-graph-size-rewrite},
use the CLIME construction with $\delta=\varepsilon/2$ in
\eqref{eq:learned-graph-clime-radius}, so
$\ell_u=\log(8p^2/\varepsilon)$, and threshold at $r_u$.
The contribution $(1-\alpha)\delta$ in
\eqref{eq:learned-graph-threshold-size} is then at most
$\varepsilon/2$, leaving $\varepsilon/2$ for the contribution
from any omitted edges.

To illustrate, we next give sufficient sample sizes for two settings, distinguished
by which true edges the selected graph must cover. To express the required bounds on $r_u$ as conditions on $n_u$,
we first simplify \eqref{eq:learned-graph-clime-radius}.
Whenever $n_u-1>\ell_u$, we have
$\ell_u/(n_u-1)\leq\sqrt{\ell_u/(n_u-1)}$, so
\eqref{eq:learned-graph-clime-radius} gives
\begin{equation}\label{eq:clime-radius-from-sample-size}
r_u\leq16vK^2\sqrt{\frac{\ell_u}{n_u-1}}.
\end{equation}

\begin{example}[Full coverage]\label{ex:learned-graph-full-coverage}
We give a sufficient auxiliary sample size for coverage of every
true edge incident to $\mc T$ on the estimation event.
If such edges exist, let $\beta_{\min}$ be the smallest absolute value
of their corresponding entries of the precision matrix:
\begin{equation*}
\beta_{\min}:=
\min\left\{
|\Omega_{ij}| \; : \; i\in\mc T, ~~ j\in[p]\setminus\{i\},  ~~ \Omega_{ij}\neq0 \right\}.
\end{equation*}
The test has size at most $\alpha+\varepsilon$ if the auxiliary sample size satisfies
\begin{equation}\label{eq:learned-graph-full-coverage-sample-size}
n_u-1>\ell_u\max\left\{1,\frac{1024v^2K^4}{\beta_{\min}^2}\right\}.
\end{equation}

Indeed, under \eqref{eq:learned-graph-full-coverage-sample-size},
\eqref{eq:clime-radius-from-sample-size} gives
$2r_u\leq32vK^2\sqrt{\ell_u/(n_u-1)}<\beta_{\min}$.
Consequently, $|\Omega_{ij}|>2r_u$ for every true edge incident to
$\mc T$, so coverage holds on the estimation event.
In particular, $k_u=0$, so the contribution from omitted edges vanishes.
The last assertion of \Cref{cor:learned-graph-threshold} therefore gives
size at most $\alpha+(1-\alpha)\delta\leq\alpha+\varepsilon$,
since $\delta=\varepsilon/2$. 
If no true edge is incident to $\mc T$, coverage holds automatically.

The same coverage argument also yields asymptotic Type I error control.
When there are true edges incident to $\mc T$, take $\delta=\delta_n\to0$ in
\eqref{eq:learned-graph-clime-radius}.
Suppose the bounds $v$ and $K$ do not depend on $n$.
Then size is at most $\alpha+o(1)$ if
$\log(p/\delta_n)=o(n_u)$ and
$\log(p/\delta_n)/\beta_{\min}^2=o(n_u)$.
Under these conditions, \eqref{eq:learned-graph-clime-radius}
gives $r_u=O(\sqrt{\log(p/\delta_n)/n_u})=o(\beta_{\min})$,
so the same coverage argument applies eventually. 
Furthermore, if $\log(p/\delta_n)/\beta_{\min}^2=o(n)$, the sufficient lower
bound on $n_u$ in
\eqref{eq:learned-graph-full-coverage-sample-size} is itself $o(n)$.
\end{example}

\begin{example}[Allowing weak edge omissions]\label{ex:learned-graph-weak-omissions}
Returning to $\delta=\varepsilon/2$, we allow omission of weak true
edges between $\mc T$ and $\mc S$, while requiring coverage of all true
edges within $\mc T$ on the estimation event.
Let $\beta_T$ be the smallest nonzero $|\Omega_{ij}|$ with distinct
$i,j\in\mc T$, if any.
The count $k_u$ depends on $n_u$ through $r_u$.
To obtain an upper bound on $k_u$ that does not depend on $n_u$, first define
\begin{equation}\label{eq:learned-graph-baseline-radius}
M_0=\max\left\{1,\frac{1024v^2K^4}{\beta_T^2}\right\},
\qquad
\rho_0=8vK^2\left(M_0^{-1/2}+M_0^{-1}\right).
\end{equation}
If there are no true edges within $\mc T$, omit the term involving
$\beta_T$ and take $M_0=1$.
Define the fixed count
\begin{equation}\label{eq:learned-graph-weak-edge-count}
\bar k=\#\left\{(i,j)\in\mc T\times\mc S:
                  0<|\Omega_{ij}|<2\rho_0\right\}.
\end{equation}
Whenever $n_u-1>\ell_u M_0$, \eqref{eq:learned-graph-clime-radius}
gives $r_u<\rho_0$, so $k_u\leq\bar k$.
Under the spectral assumptions of \Cref{cor:learned-graph-threshold},
$\mathbb P_*(\mathrm{pVal}_T\leq\alpha)\leq\alpha+\varepsilon$ whenever
\begin{equation}\label{eq:learned-graph-auxiliary-sample-size}
n_u-1>\ell_u\max\left\{
M_0,
\frac{1024v^2K^4n\Lambda_S\bar k}{\kappa_T\varepsilon^2}
\right\}.
\end{equation}
Indeed, the $M_0$ term in \eqref{eq:learned-graph-auxiliary-sample-size}
ensures $k_u\leq\bar k$.
It also guarantees coverage of all true edges within $\mc T$
on the estimation event.
If there are true edges within $\mc T$,
\eqref{eq:clime-radius-from-sample-size} and the definition of $M_0$ in
\eqref{eq:learned-graph-baseline-radius} give
$2r_u<32vK^2M_0^{-1/2}\leq\beta_T$,
ensuring that these edges are selected on the estimation event.
If there are no such edges, coverage within $\mc T$ is automatic.
For the contribution from omitted edges,
\eqref{eq:clime-radius-from-sample-size} and $k_u\leq\bar k$ give
$$
r_u\sqrt{\frac{n\Lambda_S k_u}{\kappa_T}}
\leq16vK^2\sqrt{\frac{n\Lambda_S\bar k\ell_u}{\kappa_T(n_u-1)}}
\leq\frac{\varepsilon}{2},
$$
where the last inequality follows from the second term in
\eqref{eq:learned-graph-auxiliary-sample-size} when $\bar k>0$;
if $\bar k=0$, then $k_u=0$ and the bound is immediate.
With
$\delta=\varepsilon/2$, substitution into \eqref{eq:learned-graph-threshold-size} gives the claimed bound $\alpha+\varepsilon$.
\end{example}

\subsection[Additional power results for G-CRT]{Additional power results for $G$-CRT}\label{app:additional-power}

This appendix supplements the $G$-CRT power analysis in
Section~\ref{sec:gcrt-ridge-power}.
Appendix~\ref{app:gcrt-ridge-example} gives an explicit family
satisfying Theorem~\ref{thm:gcrt-centered-power}.
Appendix~\ref{sec:gcrt-ridge-comparison} states two power comparisons
in Proposition~\ref{prop:gcrt-centering-comparisons}.
The model and sampling details are in
Appendix~\ref{app:gcrt-common-conditions}, and the proofs are in
Appendices~\ref{app:gcrt-ridge-power} and~\ref{app:gcrt-centering-comparisons}.

We use the Gaussian linear model and notation of
Section~\ref{sec:gcrt-ridge-power}.
Let $\mathbf C=\mathbf\Sigma_{\mathcal T\mathcal S}
\mathbf\Sigma_{\mathcal S\mathcal S}^{-1}$ be the coefficient matrix
in the conditional mean of the targets given the background,
and let $\mathbf V=\mathbf\Sigma_{\mathcal T\mid\mathcal S}$ be their
conditional covariance matrix.
The vector $\bs\gamma=\bs\beta_{\mathcal S}
+\mathbf C^\top\bs\beta_{\mathcal T}$ is the coefficient of $\bs X_{\mathcal S}$
in $\mathbb E(Y\mid\bs X_{\mathcal S})$.
Define the standardised target signal strength $r\geq0$ by
\begin{equation}\label{gcrtcent:eq:signal-strength}
r^2=\bs\beta_{\mathcal T}^\top\mathbf V\bs\beta_{\mathcal T}/\sigma^2,
\end{equation}
where $\sigma$ is the error standard deviation.

\subsubsection{Five connected targets with a dense background}\label{app:gcrt-ridge-example}
This family shows that the conditions of
Theorem~\ref{thm:gcrt-centered-power} can hold simultaneously with
dependent targets, more background variables than observations, and a
background contribution to the response that does not vanish.
The background graph is complete, while each target has $o(n)$
neighbours in the background group.
Let $t=5$, $s=2n$, and $k_n=2\lfloor\sqrt n/2\rfloor$.
We construct unit vectors $\bs u,\bs w\in\mathbb R^s$.
The support of $\bs u$ will determine the background neighbours of
each target, while $\bs w$ will supply a dense component of $\bs\gamma$.
Among the first $k_n$ background indices, define
$I_n^+=\{1,\ldots,k_n/2\}$ and
$I_n^-=\{k_n/2+1,\ldots,k_n\}$.
Let $I_n=I_n^+\cup I_n^-=\{1,\ldots,k_n\}$.
To make the vectors orthogonal, we take $\bs u$ constant on $I_n$
and zero elsewhere, and give $\bs w$ values of equal magnitude
and opposite sign on the two halves.
Specifically, define
\begin{equation}\label{gcrtint:eq:example-vectors}
 u_j=\frac{\mathbf1\{j\in I_n\}}{\sqrt{k_n}},\qquad
 w_j=\begin{cases}
 (2k_n)^{-1/2},&j\in I_n^+,\\
 -(2k_n)^{-1/2},&j\in I_n^-,\\
 \{2(s-k_n)\}^{-1/2},&j\notin I_n.
 \end{cases}
\end{equation}
Let $\mathbf W_{\rm star}$ be the $5\times5$ matrix with
$(\mathbf W_{\rm star})_{1j}=(\mathbf W_{\rm star})_{j1}=1$
for $j=2,\ldots,5$ and all other entries zero.
It is the adjacency matrix of the star centred at target node 1.
For $\sigma\in[1/2,2]$, $0<r\leq2$, and $0\leq\delta\leq0.13r$, define
\begin{equation}\label{gcrtint:eq:example-model}
\begin{aligned}
 \mathbf\Theta&=\mathbf I_5-0.02\mathbf W_{\rm star},
 &\bs a&=\bigl(0.1/\sqrt{5}\bigr)\times\bs1_5,\qquad \mathbf A=\bs a\bs u^\top,\\
 \mathbf\Sigma_{\mathcal S\mathcal S}^{-1}
 &=\mathbf I_s+0.2\bs u\bs u^\top+0.1\bs w\bs w^\top,
 &\mathbf V&=\mathbf\Theta^{-1},\qquad \mathbf C=\mathbf V\mathbf A,\\
 \bs\gamma&=\sigma\{\sqrt{1-\delta^2}\bs w+\delta\bs u\},
 &\bs\beta_{\mathcal T}&=
 \frac{\sigma r\bs1_5}{\sqrt{\bs1_5^\top\mathbf V\bs1_5}}.
\end{aligned}
\end{equation}
The five target coefficients in \eqref{gcrtint:eq:example-model} are equal,
and their normalisation gives standardised target signal strength $r$.
To describe power uniformly over a range of signals, let $h_n$ be a
positive deterministic sequence with $h_n\leq2$.
We consider all these models with $h_n\leq r\leq2$, so $h_n$ is the
weakest signal strength included at sample size $n$.

The joint precision matrix has background block
$\mathbf\Sigma_{\mathcal S\mathcal S}^{-1}+\mathbf A^\top\mathbf V\mathbf A$
and is positive definite by its Schur complement.
The eigenvalues of $\mathbf W_{\rm star}$ are $2,-2,0,0,0$, so the spectral bounds hold with
$v_-=0.9$, $v_+=1.1$, $s_-=0.75$, $s_+=1$, $\bar c=0.11$, $c_+=2.1$,
and $G_+=R=2$.
For the matrix $\mathbf J$ defined in Section~\ref{sec:gcrt-ridge-power},
we have
$\mathbf J=0.02\mathbf W_{\rm star}$, so
$\max\{\|\mathbf J\|_1,\|\mathbf J\|_\infty\}=0.08<1$.
Thus the population bounds in
\eqref{gcrtcent:eq:population-conditions} hold for this family.

Every component of $\bs\gamma$ is nonzero, since $\delta\leq0.26<1/\sqrt3$,
and the variance of $\bs\gamma^\top\bs X_{\mathcal S}$ is at least $5\sigma^2/6$.
The background graph is complete: inside $I_n$ the off-diagonal entries of the precision
matrix are at least $(0.2+\bs a^\top\mathbf V\bs a-0.05)/k_n>0$;
all other entries equal $0.1w_iw_j\ne0$.
Each target has $k_n=o(n)$ neighbours in the background group, and
$\sum_{j\in\mathcal T}|N_j\cap\mathcal T|+t=13$.

For $G$-CRT, take $\ell=0.2$ and $L=\lambda_n=1$.
Generate the randomised copies using Algorithm~\ref{alg: exchangeable}
with the conditional independence graph of this family and the residual
rotations in \eqref{eq:ggm-local-update}, updating each target once per sweep.
For fixed $M\geq2$ and $\alpha\in(0,1)$ with
$1/(M+1)\leq\alpha$, use the symmetrically centred ridge residual scores in
\eqref{eq:symmetric-centered-statistic}, with either
$\mathbf W=\mathbf D_{\rm pool}^{-1}$ or $\mathbf W=\mathbf I_t$.
If $nh_n^2\to\infty$, the graph counts satisfy \eqref{gcrtcent:eq:rates}.
Theorem~\ref{thm:gcrt-centered-power} then implies that the rejection
probability tends to one uniformly over all models in this family with
$h_n\leq r\leq2$.
This allows the weakest signal $h_n$ to decrease to zero more slowly
than $n^{-1/2}$.

\subsubsection{Power comparisons for centred and uncentred scores}
\label{sec:gcrt-ridge-comparison}

Theorem~\ref{thm:gcrt-centered-power} establishes uniform consistency
for the $G$-CRT using symmetrically centred ridge residual scores
(vector centring).
We compare the power of this test with two other tests: the $G$-CRT
using the uncentred ridge statistic, and the $G$-CRT using scalar gains
that are centred and squared.
Proposition~\ref{prop:gcrt-centering-comparisons} demonstrates two
differences in power in a common Gaussian model.
Part (a) considers two fixed alternatives with the same nonzero
signal strength and opposite target directions.
The uncentred test has power tending to one under one alternative
and to zero under the other.
With either form of centring, power tends to one under both alternatives.
Part (b) gives a weak orthogonal alternative under which only the test
with vector centring is consistent.

\paragraph{Statistics.}
We compare three $G$-CRT tests using the same candidates and the p-value
formula in \eqref{eqn:pval-1}.
Use the residual score vectors $\bs c_v$ and regularised residual Gram
matrices $\mathbf D_v$ in \eqref{gcrtcent:eq:residual-scores},
where candidate $v=0$ is the observation and $v=1,\ldots,M$ are
the randomised copies.

\begin{enumerate}
\item[(1)] \textbf{Uncentred ridge statistic.}
For candidate $v$, define the uncentred gain
$\mathcal A_v=\bs c_v^\top\mathbf D_v^{-1}\bs c_v$.
To interpret this gain, compare the penalised fit at zero coefficients
with the fit obtained by ridge regression of the response residuals
on the target residuals.
For a candidate matrix $\mathbf x$ and its ridge residual matrix
$\mathbf x_R$, let $\widehat{\bs b}$ be the fitted ridge coefficient vector:
\begin{equation}\label{gcrtint:eq:target-fit}
 \widehat{\bs b}=\arg\min_{\bs b\in\mathbb R^t}
 \{\|\bs Y_R-\mathbf x_R\bs b\|^2+\lambda_n\|\bs b\|^2\}.
\end{equation}
The background fit, target residualisation, and penalties are those of
Section~\ref{sec:gcrt-ridge-power}.
For candidate $v$, the decrease in the penalised objective in
\eqref{gcrtint:eq:target-fit} from the zero coefficient vector to the fitted coefficient is
$n\mathcal A_v$.

For $\|\bs Y_R\|>0$, normalising this improvement by
$\|\bs Y_R\|^2$ gives the uncentred ridge statistic:
\begin{equation}\label{gcrtint:eq:statistic}
 T_{\mathrm{G,ridge}}(\mathbf x)=
 \frac{\|\bs Y_R\|^2-
       \{\|\bs Y_R-\mathbf x_R\widehat{\bs b}\|^2
                           +\lambda_n\|\widehat{\bs b}\|^2\}}
 {\|\bs Y_R\|^2}.
\end{equation}
Define $T_{\mathrm{G,ridge}}(\mathbf x)=0$ if $\|\bs Y_R\|=0$,
and write $T_v$ for its value on candidate $v$.

Thus, for $\|\bs Y_R\|>0$,
\begin{equation}\label{gcrtb:eq:rank}
 T_v=\frac{\mathcal A_v}{\|\bs Y_R\|^2/n}.
\end{equation}
Let $\mathrm{pVal}_{\rm unc}$ and $\mathrm{pVal}_T$ be the p-values in
\eqref{eqn:pval-1} computed from $\mathcal A_v$ and $T_v$, respectively.
The positive denominator in \eqref{gcrtb:eq:rank} is common to all
candidates, so their ranks coincide and
$\mathrm{pVal}_T=\mathrm{pVal}_{\rm unc}$.
If $\|\bs Y_R\|=0$, all scores and gains are zero and both p-values are one.

\item[(2)] \textbf{Vector centring.}
Let $\bar{\bs c}=(M+1)^{-1}\sum_{v=0}^M\bs c_v$.
With $\mathbf W=\mathbf D_{\rm pool}^{-1}$, the statistic in
\eqref{eq:symmetric-centered-statistic} is
$$
T_v^{\rm cen}=(\bs c_v-\bar{\bs c})^\top
\mathbf D_{\rm pool}^{-1}(\bs c_v-\bar{\bs c}),
\qquad v=0,\ldots,M.
$$
Let $p_{\rm cen}$ be the p-value in \eqref{eqn:pval-1}
computed from $T_v^{\rm cen}$.
The comparison of $\mathcal A_v$ with $T_v^{\rm cen}$ also changes
the weight, from $\mathbf D_v^{-1}$ to $\mathbf D_{\rm pool}^{-1}$.

\item[(3)] \textbf{Scalar centring and squaring.}
Subtracting a common mean from all scalar gains leaves their ranks
unchanged. The squared deviations make the rank comparison sensitive
to departures on either side of that mean.
For scalar centring and squaring,
let $\bar{\mathcal A}=(M+1)^{-1}\sum_{v=0}^M\mathcal A_v$ and
define
\begin{equation}\label{gcrtb:eq:scalar-centered-statistic}
\mathcal B_v=(\mathcal A_v-\bar{\mathcal A})^2,
\qquad v=0,\ldots,M.
\end{equation}
Let $\mathrm{pVal}_{\mathcal B}$ be the p-value in
\eqref{eqn:pval-1} computed from $\mathcal B_v$.
This symmetric transformation preserves exchangeability of the scores under the null hypothesis, so
$\mathrm{pVal}_{\mathcal B}$ is valid under the null;
see Appendix~\ref{app:scoring-conditions}.
\end{enumerate}

\paragraph{Model with two target covariates.}
Let $Z,E_1,E_2,\varepsilon$ be independent standard Gaussian variables.
Define $E_+=(E_1+E_2)/\sqrt2$, $E_-=(E_1-E_2)/\sqrt2$, and, for
$r\geq0$ and $\theta\in[0,2\pi)$, consider
\begin{equation}\label{gcrtb:eq:direction-example}
 X_{\mathcal S}=Z,\qquad X_1=Z+E_1,\qquad X_2=Z+E_2,\qquad
 Y=Z+r(\cos\theta\,E_++\sin\theta\,E_-)+\varepsilon.
\end{equation}
Let $\bs v_+=(1,1)^\top/\sqrt2$ and $\bs v_-=(1,-1)^\top/\sqrt2$.
The target coefficient is
$r(\cos\theta\,\bs v_++\sin\theta\,\bs v_-)$.
Here $\mathbf V=\mathbf I_2$, $\mathbf C=(1,1)^\top$,
$\gamma=\sigma=1$, so $r$ is the standardised target signal strength.

The covariate distribution and $\mathbb E(Y\mid Z)=Z$ remain fixed
as $r$ and $\theta$ vary.
The common background contribution to the residual score vectors
$\bs c_v$ lies along $\bs v_+$; see \eqref{gcrtb:eq:direction-limits}.
At fixed $r>0$, changing $\theta$ from $0$ to $\pi$ reverses the
target signal without changing its strength.
The choice $\theta=\pi/2$ makes the target signal orthogonal to
this background contribution.

\begin{proposition}[Power comparisons]\label{prop:gcrt-centering-comparisons}
Consider $n$ independent observations from
\eqref{gcrtb:eq:direction-example}. Generate the randomised copies
using Algorithm~\ref{alg: exchangeable} with the conditional independence graph and
the residual rotations in \eqref{eq:ggm-local-update}, updating each
target once per sweep. Let $\ell=\lambda_n=L=1$, and fix an integer
$M\geq2$ and $\alpha$ with $1/(M+1)\leq\alpha<1$.
For $G$-CRT using $\mathcal A_v$, $T_v^{\rm cen}$, or $\mathcal B_v$,
the following conclusions hold as $n\to\infty$.
\begin{enumerate}
\item[(a)] \textbf{Fixed alternatives.}
Let $r=1/(2\sqrt2)$. Then
$$
\mathbb P(\mathrm{pVal}_{\rm unc}\leq\alpha)\longrightarrow
\begin{cases}
1,&\theta=0,\\
0,&\theta=\pi.
\end{cases}
$$
For both $\theta=0$ and $\theta=\pi$,
$\mathbb P(p_{\rm cen}\leq\alpha)\to1$ and
$\mathbb P(\mathrm{pVal}_{\mathcal B}\leq\alpha)\to1$.

\item[(b)] \textbf{Weak orthogonal alternative.}
Let $\theta=\pi/2$ and $r=r_n=n^{-1/3}$. Then
\begin{equation}\label{gcrtcent:eq:strict-comparison}
 \mathbb P(p_{\rm cen}\leq\alpha)\longrightarrow1,\qquad
 \mathbb P(\mathrm{pVal}_{\rm unc}\leq\alpha)
 \longrightarrow\frac{\lfloor\alpha(M+1)\rfloor}{M+1}.
\end{equation}
Also, $\mathbb P(\mathrm{pVal}_{\mathcal B}\leq\alpha)
\to\lfloor\alpha(M+1)\rfloor/(M+1)$.
\end{enumerate}
\end{proposition}

Appendix~\ref{app:gcrt-centering-comparisons} proves
Proposition~\ref{prop:gcrt-centering-comparisons}, with separate
arguments for parts (a) and (b).
We next explain how the background contribution and target signal
produce the different power limits in parts (a) and (b).

\paragraph{Interpretation of the fixed alternatives.}
In part (a), the target signal reinforces the common background
contribution to $\bs c_0$ when $\theta=0$.
Then $\mathcal A_0$ converges in probability to a larger limit than
$\mathcal A_m$ for each $m=1,\ldots,M$.
When $\theta=\pi$, the two contributions cancel asymptotically in $\bs c_0$.
Each $\mathcal A_m$ converges in probability to a positive limit, while
$\mathcal A_0\xrightarrow{\mathbb P}0$, so
$\mathrm{pVal}_{\rm unc}\xrightarrow{\mathbb P}1$.
Thus changing only the target direction, at the same signal strength,
changes the limit of $\mathbb P(\mathrm{pVal}_{\rm unc}\leq\alpha)$
from one to zero.
For $\theta=\pi$ and any fixed integer $L\geq1$, we have
$\mathbb P(\mathrm{pVal}_{\rm unc}\leq\alpha)\to0$ as $n\to\infty$;
see Appendix~\ref{app:gcrt-centering-comparisons}.
For both directions, $\mathbb P(p_{\rm cen}\leq\alpha)\to1$ and
$\mathbb P(\mathrm{pVal}_{\mathcal B}\leq\alpha)\to1$.
These fixed alternatives therefore do not distinguish the detection
capabilities of $T_v^{\rm cen}$ and $\mathcal B_v$.

\paragraph{Interpretation of the weak alternative.}
In part (b), the model is $Y=Z+r_nE_-+\varepsilon$.
The target coefficient $\bs\beta_{\mathcal T}=r_n\bs v_-$ is
orthogonal to the common background contribution along $\bs v_+$;
see \eqref{gcrtb:eq:direction-limits}.
The subtraction $\bs c_v-\bar{\bs c}$ in $T_v^{\rm cen}$ removes
this common contribution before the quadratic form is evaluated.
Here $r_n\to0$ and $nr_n^2=n^{1/3}\to\infty$.
The graph and dimensions are fixed, and the remaining conditions of
Theorem~\ref{thm:gcrt-centered-power} hold in this model, so the theorem
gives $\mathbb P(p_{\rm cen}\leq\alpha)\to1$.

Write $\mathcal A_v(r)$ for $\mathcal A_v$ evaluated at signal strength $r$.
Orthogonality removes the leading term linear in $r_n$ from
$\mathcal A_0(r_n)-\mathcal A_0(0)$.
Equations~\eqref{gcrtb:eq:observed-signal-perturbation} and
\eqref{gcrtb:eq:signal-perturbation} show that its leading
term is $r_n^2=n^{-2/3}$, with smaller changes
$\mathcal A_m(r_n)-\mathcal A_m(0)$ for $m=1,\ldots,M$.
These changes are negligible relative to the $n^{-1/2}$ fluctuations
in \eqref{gcrtb:eq:null-gain-limit}.
After a common shift and rescaling, both $(\mathcal A_v(r_n))_{v=0}^M$
and $(\mathcal A_v(0))_{v=0}^M$ converge to the same joint limit;
see \eqref{gcrtb:eq:joint-rank-limit}.
For $M\geq2$, this limit and its centred squared deviations are
exchangeable and have no ties, so the ranks used by
$\mathrm{pVal}_{\rm unc}$ and $\mathrm{pVal}_{\mathcal B}$ are
asymptotically uniform.
Their common limiting rejection probability,
$\lfloor\alpha(M+1)\rfloor/(M+1)$, is the nominal level rounded down
to the discrete p-value grid for fixed $M$.

\paragraph{Numerical comparisons.}
For $r=1/(2\sqrt2)$ and $\theta=0$ or $\pi$, the parametrisation used
in the numerical comparisons is
$Y=(1-2b)X_{\mathcal S}+b(X_1+X_2)+\varepsilon$ with
$b=1/4$ or $-1/4$, respectively.
Appendix~\ref{app:direction-mechanism}
(Table~\ref{tab:direction-mechanism}) shows a pronounced finite-sample
power difference between the two fixed alternatives in part~(a).
The uncentred test's rejection rate rises towards one under the positive
alternative and remains near zero under the negative alternative.

Under the positive fixed alternative, Appendix~\ref{app:centering-simulation}
(Figure~\ref{fig:centering-comparison}) reports higher observed rejection
rates for the uncentred test than for the test with vector centring at
$n=40$ and $120$, although both are asymptotically consistent.

Appendix~\ref{app:connected-centering}
(Figure~\ref{fig:connected-centering}, Table~\ref{tab:connected-centering-paired})
reports further comparisons with dependent targets, including comparisons
at common weights and under opposite signs of the background contribution.

\section{Proofs}
\label{app:proofs}
This appendix contains proofs of the main results and the supplementary results in Appendix~\ref{app:supplementary-results}, together with the required technical lemmas.\par

Appendices~\ref{app:proof-prop-crt-valid}--\ref{app:proof-prop-exchangeable} establish finite-sample validity and give the required technical preparations. Appendices~\ref{app:proof-thm-power} and~\ref{app:proof-thm-lower-bound} prove the MVN-CRT power results and detection lower bounds; Appendices~\ref{app:gcrt-ridge-power} and~\ref{app:gcrt-centering-comparisons} prove the $G$-CRT power results. Appendix~\ref{app:proof-comparison-class-properties} gives the proofs for MS-CRT robustness and the properties of the enlarged comparison class; Appendix~\ref{app: pf robust} gives the proofs for robustness to graph misspecification.

\subsection{Proof of Proposition~\ref{prop: CRT valid}}\label{app:proof-prop-crt-valid}

\begin{proof}[Proof of Proposition~\ref{prop: CRT valid}]

We prove the validity result for the general scoring rules in Appendix~\ref{app:scoring-conditions}. Proposition~\ref{prop: CRT valid} is a special case.

To apply \eqref{eq:mscrt-randomised-scoring-equivariance} with $\bs Y$
held fixed, we first show that
$(\mathbf X_{\mc T}^{(v)})_{v=0}^{M}$ is exchangeable conditional on
$(\bs Y,\mathbf X_{\mc S},\Phi)$.
Under $H_0$, $\bs Y\indp\mathbf X_{\mc T}\mid\mathbf X_{\mc S}$.
Because $\Phi$ is a function of $(\mathbf X_{\mc T},\mathbf X_{\mc S})$,
we also have $\bs Y\indp(\mathbf X_{\mc T},\Phi)\mid\mathbf X_{\mc S}$.
The weak-union property of conditional independence gives
\begin{equation}\label{eq:key-prop-step1}
\bs{Y} \indp \mathbf{X}_{\mc{T}}\mid (\mathbf{X}_{\mc{S}},\Phi).
\end{equation}

The matrices $(\widetilde{\mathbf X}_{\mc T}^{(m)})_{m=1}^{M}$ are
generated without using $\bs Y$, so
\begin{equation}\label{eq:key-prop-step2}
\bs{Y} \indp \left(\widetilde{\mathbf{X}}_{\mathcal{T}}^{(1)}, \ldots, \widetilde{\mathbf{X}}_{\mathcal{T}}^{(M)}\right) \mid \left(\mathbf{X}_{\mathcal{T}}, \mathbf{X}_{\mathcal{S}}, \Phi \right).
\end{equation}

We next prove
$\bs Y\indp(\mathbf X_{\mc T}^{(v)})_{v=0}^{M}
\mid(\mathbf X_{\mc S},\Phi)$.
We recall the contraction property of conditional independence:
for random objects $U$, $V$, and $Z$,
if $\bs Y\indp U\mid Z$ and $\bs Y\indp V\mid(U,Z)$, then
$\bs Y\indp(U,V)\mid Z$; see
\citet[Lemma 4.3]{dawid1979conditional}.
To apply the contraction property of conditional independence, let
$$
Z=(\mathbf X_{\mc S},\Phi),\qquad
U=\mathbf X_{\mc T},\qquad
V=(\widetilde{\mathbf X}_{\mc T}^{(1)},\ldots,
\widetilde{\mathbf X}_{\mc T}^{(M)}).
$$
Equation~\eqref{eq:key-prop-step1} gives $\bs Y\indp U\mid Z$,
and \eqref{eq:key-prop-step2} gives $\bs Y\indp V\mid(U,Z)$.
Since $U=\mathbf X_{\mc T}$ and
$V=(\widetilde{\mathbf X}_{\mc T}^{(1)},\ldots,
\widetilde{\mathbf X}_{\mc T}^{(M)})$, contraction gives
\begin{equation}\label{eq:mscrt-all-candidates-independent-response}
\bs Y\indp
\bigl(\mathbf X_{\mc T},\widetilde{\mathbf X}_{\mc T}^{(1)},
\ldots,\widetilde{\mathbf X}_{\mc T}^{(M)}\bigr)
\mid(\mathbf X_{\mc S},\Phi).
\end{equation}

The conditional independence in
\eqref{eq:mscrt-all-candidates-independent-response} implies that,
almost surely,
\begin{equation}\label{eq:mscrt-candidate-law-with-response}
\mathcal L\bigl((\mathbf X_{\mc T}^{(v)})_{v=0}^{M}
\mid(\bs Y,\mathbf X_{\mc S},\Phi)\bigr)
=
\mathcal L\bigl((\mathbf X_{\mc T}^{(v)})_{v=0}^{M}
\mid(\mathbf X_{\mc S},\Phi)\bigr).
\end{equation}
Let $\pi$ be an arbitrary permutation of $\{0,\ldots,M\}$.
We apply \eqref{eq:mscrt-candidate-law-with-response} to both
$(\mathbf X_{\mc T}^{(v)})_{v=0}^{M}$ and
$\pi(\mathbf X_{\mc T}^{(v)})_{v=0}^{M}$.
The sampling assumption of Proposition~\ref{prop: CRT valid} makes
$(\mathbf X_{\mc T}^{(v)})_{v=0}^{M}$ exchangeable conditional on
$(\mathbf X_{\mc S},\Phi)$.
Thus,
\begin{equation}\label{eq:mscrt-candidate-permutation-with-response}
\begin{aligned}
&\mathcal L\bigl(\pi(\mathbf X_{\mc T}^{(v)})_{v=0}^{M}
\mid(\bs Y,\mathbf X_{\mc S},\Phi)\bigr)\\
={}&\mathcal L\bigl(\pi(\mathbf X_{\mc T}^{(v)})_{v=0}^{M}
\mid(\mathbf X_{\mc S},\Phi)\bigr)\\
={}&\mathcal L\bigl((\mathbf X_{\mc T}^{(v)})_{v=0}^{M}
\mid(\mathbf X_{\mc S},\Phi)\bigr)\\
={}&\mathcal L\bigl((\mathbf X_{\mc T}^{(v)})_{v=0}^{M}
\mid(\bs Y,\mathbf X_{\mc S},\Phi)\bigr).
\end{aligned}
\end{equation}

We next prove that $(T_0,\ldots,T_M)$ is exchangeable conditional on
$(\bs Y,\mathbf X_{\mc S},\Phi)$.
The matrices $(\mathbf X_{\mc T}^{(v)})_{v=0}^{M}$ are measurable
functions of $(\mathbf X,W_{\mathrm c})$.
The independence assumption in
\eqref{eq:mscrt-scoring-randomness-conditions} therefore implies
$W_{\mathrm s}\indp
(\bs Y,\mathbf X,(\mathbf X_{\mc T}^{(v)})_{v=0}^{M})$.
Hence, almost surely,
\begin{equation}\label{eq:mscrt-candidates-scoring-product}
\begin{aligned}
&\mathcal L\bigl(((\mathbf X_{\mc T}^{(v)})_{v=0}^{M},W_{\mathrm s})
\mid(\bs Y,\mathbf X_{\mc S},\Phi)\bigr)\\
={}&\mathcal L\bigl((\mathbf X_{\mc T}^{(v)})_{v=0}^{M}
\mid(\bs Y,\mathbf X_{\mc S},\Phi)\bigr)\otimes\nu_{\mathrm s},
\end{aligned}
\end{equation}
where $\otimes$ denotes the product of probability measures.
For every fixed permutation $\pi$, the variable $\pi W_{\mathrm s}$ is
also independent of $(\bs Y,\mathbf X,W_{\mathrm c})$, and
\eqref{eq:mscrt-scoring-randomness-conditions} gives
$\mathcal L(\pi W_{\mathrm s}\mid(\bs Y,\mathbf X_{\mc S},\Phi))
=\nu_{\mathrm s}$.
Equations~\eqref{eq:mscrt-candidate-permutation-with-response} and
\eqref{eq:mscrt-candidates-scoring-product} now give
\begin{equation}\label{eq:mscrt-joint-permutation-invariance}
\begin{aligned}
&\mathcal L\bigl((\pi(\mathbf X_{\mc T}^{(v)})_{v=0}^{M},
\pi W_{\mathrm s})\mid(\bs Y,\mathbf X_{\mc S},\Phi)\bigr)\\
={}&\mathcal L\bigl(\pi(\mathbf X_{\mc T}^{(v)})_{v=0}^{M}
\mid(\bs Y,\mathbf X_{\mc S},\Phi)\bigr)\otimes\nu_{\mathrm s}\\
={}&\mathcal L\bigl((\mathbf X_{\mc T}^{(v)})_{v=0}^{M}
\mid(\bs Y,\mathbf X_{\mc S},\Phi)\bigr)\otimes\nu_{\mathrm s}\\
={}&\mathcal L\bigl(((\mathbf X_{\mc T}^{(v)})_{v=0}^{M},
W_{\mathrm s})\mid(\bs Y,\mathbf X_{\mc S},\Phi)\bigr).
\end{aligned}
\end{equation}

For any permutation $\pi$, the definition of $(T_0,\ldots,T_M)$ and
\eqref{eq:mscrt-randomised-scoring-equivariance} give, conditional on
$(\bs Y,\mathbf X_{\mc S},\Phi)$,
\begin{equation}\label{eq:mscrt-score-exchangeability}
\begin{aligned}
\pi(T_0,\ldots,T_M)
&=\pi\mathscr T\bigl(\bs Y,\mathbf X_{\mc S},
(\mathbf X_{\mc T}^{(v)})_{v=0}^{M};W_{\mathrm s}\bigr)\\
&=\mathscr T\bigl(\bs Y,\mathbf X_{\mc S},
\pi(\mathbf X_{\mc T}^{(v)})_{v=0}^{M};\pi W_{\mathrm s}\bigr)\\
&\overset d=\mathscr T\bigl(\bs Y,\mathbf X_{\mc S},
(\mathbf X_{\mc T}^{(v)})_{v=0}^{M};W_{\mathrm s}\bigr)\\
&=(T_0,\ldots,T_M),
\end{aligned}
\end{equation}
where the first and last equalities use the definition of
$(T_0,\ldots,T_M)$, the second equality uses
\eqref{eq:mscrt-randomised-scoring-equivariance}, and the equality in
distribution follows from
\eqref{eq:mscrt-joint-permutation-invariance}.
Since $\pi$ is arbitrary, \eqref{eq:mscrt-score-exchangeability} proves
that $(T_0,\ldots,T_M)$ is exchangeable conditional on
$(\bs Y,\mathbf X_{\mc S},\Phi)$.

By \eqref{eq:mscrt-score-exchangeability}, the standard rank argument
for exchangeable scores, with ties counted conservatively as in
\eqref{eqn:pval-1}, gives, for every $\alpha\in(0,1)$,
\begin{equation}\label{eq:mscrt-conditional-rank-bound}
\mathbb P\bigl(\mathrm{pVal}_T\leq\alpha
\mid(\bs Y,\mathbf X_{\mc S},\Phi)\bigr)\leq\alpha.
\end{equation}
Taking expectations in \eqref{eq:mscrt-conditional-rank-bound}
completes the proof.
\end{proof}

\subsection{Preliminaries on Haar measure and the Stiefel manifold}
\label{app:haar-preliminaries-draft}

This subsection is needed for the theoretical analysis of the MVN-CRT.
It records the basic facts about the Stiefel manifold and Haar
(invariant) measure that are used in the proofs of
Proposition~\ref{prop: multi gaussian CRT valid} and
Theorem~\ref{power analysis}.
In Eaton's notation, the Stiefel manifold is the space $\mathcal{F}_{d,n}$ of
\(n\times d\) matrices with orthonormal columns. The existence and uniqueness
of the invariant probability measure on this space follow from Theorem~6.3
together with Example~6.16 of \citet{eaton1983}.

For integers \(n \ge d\), define the Stiefel manifold
$$
W_{n,d}
=
\left\{
\mathbf{u} \in \mathbb{R}^{n\times d} : \mathbf{u}^\top \mathbf{u} = \mathbf{I}_d
\right\}.
$$
When \(d=n\), this reduces to the orthogonal group
$$
O(n)
=
\left\{
\mathbf{b} \in \mathbb{R}^{n\times n} : \mathbf{b}^\top \mathbf{b} = \mathbf{b}\mathbf{b}^\top = \mathbf{I}_n
\right\}.
$$
The group \(O(n)\) acts on \(W_{n,d}\) by left multiplication:
$$
(\mathbf{b},\mathbf{u})\mapsto \mathbf{b}\mathbf{u},
\qquad
\mathbf{b}\in O(n),\ \mathbf{u}\in W_{n,d}.
$$

We denote by \(\mu_{n,d}\) the unique probability measure on \(W_{n,d}\) that is
invariant under this action, namely,
$$
\mu_{n,d}(A)=\mu_{n,d}(\mathbf{b}A)
$$
for every Borel set \(A\subseteq W_{n,d}\) and every $\mathbf{b}\in O(n)$; see
Theorem~6.3 and Example~6.16 of \citet{eaton1983}. In the special case
\(d=n\), \(\mu_{n,n}\) is the usual Haar probability measure on the orthogonal
group \(O(n)\).

\begin{lemma}[Left orthogonal invariance and uniqueness]
\label{lem:haar-left-invariance-draft}
Let \(n\ge d\).
\begin{enumerate}
    \item If $\mathbf{U}\sim \mu_{n,d}$, then for every fixed $\mathbf{b}\in O(n)$,
    $$
    \mathbf{b}\mathbf{U} \stackrel{d}{=} \mathbf{U}.
    $$
    \item Conversely, if $\mathbf{U}_0$ is a \(W_{n,d}\)-valued random matrix such that
    $$
    \mathbf{b}\mathbf{U}_0 \stackrel{d}{=} \mathbf{U}_0
    \qquad\text{for every fixed } \mathbf{b}\in O(n),
    $$
    then $\mathbf{U}_0\sim \mu_{n,d}$.
\end{enumerate}
\end{lemma}

\begin{proof}
Part (1) is exactly the left-invariance of \(\mu_{n,d}\). For part (2), let
\(\nu\) denote the law of $\mathbf{U}_0$ on \(W_{n,d}\). The assumed distributional
invariance implies that $\nu(A)=\nu(\mathbf{b}A)$ for every Borel set
\(A\subseteq W_{n,d}\) and every $\mathbf{b}\in O(n)$. Hence \(\nu\) is a left
\(O(n)\)-invariant probability measure on \(W_{n,d}\). By the uniqueness stated
in Example~6.16 of \citet{eaton1983}, we conclude that \(\nu=\mu_{n,d}\), or
equivalently $\mathbf{U}_0\sim \mu_{n,d}$.
\end{proof}
\par

For later reference, \citet[Proposition~7.2]{eaton1983} gives a probabilistic construction of \(\mu_{n,d}\): if $\mathbf{X}$ is an \(n\times d\) standard Gaussian matrix, then the orthonormal factor obtained by applying Gram--Schmidt orthogonalisation to the columns of $\mathbf{X}$ has distribution \(\mu_{n,d}\). This complements the result on existence and uniqueness from Theorem~6.3 and Example~6.16.

\begin{lemma}[Conditional law on matrices with a fixed Gram matrix]
\label{lem:conditional-law-on-gram-fiber}
Let \(d,t\geq1\) be integers, let \(\mathcal G\) be a
\(\sigma\)-field, and let $\mathbf{Z}$ be an
\(\mathbb R^{d\times t}\)-valued random matrix. Suppose that, for every
$\mathbf{b}\in O(d)$,
\begin{equation}
\label{eq:conditional-left-orthogonal-invariance}
\mathcal L(\mathbf{b}\mathbf{Z}\mid\mathcal G)
=
\mathcal L(\mathbf{Z}\mid\mathcal G)
\quad\text{almost surely}.
\end{equation}
Set
$$
\mathbf{H}=\mathbf{Z}^\top \mathbf{Z},
\qquad
\mathcal H=\mathcal G\vee\sigma(\mathbf{H}),
$$
where $\sigma(\mathbf H)$ is the $\sigma$-field generated by
$\mathbf H$ and $\vee$ denotes the smallest $\sigma$-field containing
both arguments. 
For $\mathbf{h}\in\mathbb R^{t\times t}$, define
$$
\mathcal F_{\mathbf{h}}
=
\left\{
\mathbf{z}\in\mathbb R^{d\times t}:\mathbf{z}^\top \mathbf{z}=\mathbf{h}
\right\}.
$$
For $\mathbf{b}\in O(d)$ and
\(A\in\mathcal B(\mathbb R^{d\times t})\), write
$$
\mathbf{b}^{-1}A
=
\left\{
\mathbf{z}\in\mathbb R^{d\times t}:\mathbf{b}\mathbf{z}\in A
\right\}.
$$
Then there exist a version \(K\) of
$\mathcal L(\mathbf{Z}\mid\mathcal H)$ and an event \(\Omega_0\) with
\(\mathbb P(\Omega_0)=1\) such that, for every
\(\omega\in\Omega_0\),
$$
K\!\left(\omega,\mathcal F_{\mathbf{H}(\omega)}\right)=1
$$
and
\begin{equation}
\label{eq:conditional-gram-kernel-invariance}
K(\omega,\mathbf{b}^{-1}A)=K(\omega,A)
\end{equation}
for every $\mathbf{b}\in O(d)$ and every
\(A\in\mathcal B(\mathbb R^{d\times t})\). Moreover, for every
\(\omega\in\Omega_0\), \(K(\omega,\cdot)\) is the unique probability
measure supported on $\mathcal F_{\mathbf{H}(\omega)}$ that satisfies
\eqref{eq:conditional-gram-kernel-invariance}.
\end{lemma}

\begin{proof}
By \citet[Proposition~1.31]{eaton1983}, for every nonempty
$\mathcal F_{\mathbf h}$ and every $\mathbf x\in\mathcal F_{\mathbf h}$,
\begin{equation}\label{eq:gram-fiber-is-orthogonal-orbit}
\mathcal F_{\mathbf h}=\{\mathbf b\mathbf x:\mathbf b\in O(d)\}.
\end{equation}
Let $\lambda_d$ be the Haar probability measure
on $O(d)$. For $\mathbf x\in\mathbb R^{d\times t}$ and Borel
$A\subseteq\mathbb R^{d\times t}$, define
\begin{equation}\label{eq:orthogonal-orbit-measure}
\nu_{\mathbf x}(A)=\int_{O(d)}\mathbf1_A(\mathbf b\mathbf x)\,\lambda_d(d\mathbf b).
\end{equation}
Orthogonal transformations preserve the Gram matrix, and left
invariance of $\lambda_d$ gives
\begin{equation}\label{eq:orthogonal-orbit-measure-properties}
\nu_{\mathbf x}(\mathcal F_{\mathbf x^\top\mathbf x})=1,
\qquad
\nu_{\mathbf x}(\mathbf b^{-1}A)=\nu_{\mathbf x}(A),
\quad \mathbf b\in O(d).
\end{equation}
The orbit $\mathcal F_{\mathbf x^\top\mathbf x}$ is compact as the image of the
compact group $O(d)$ under the continuous map $\mathbf b\mapsto\mathbf b\mathbf x$.
It is homeomorphic to the quotient of $O(d)$ by the closed stabiliser
$\{\mathbf b\in O(d):\mathbf b\mathbf x=\mathbf x\}$ of $\mathbf x$, and hence is a
homogeneous space for $O(d)$
\citep[Section~6.3, pp.~207--208]{eaton1983}.
By \citet[Theorem~6.3(iii)]{eaton1983}, invariant
integrals on this space are unique up to scale, so $\nu_{\mathbf x}$ is its unique invariant Borel
probability measure.

For each fixed $\mathbf b\in O(d)$, apply
\eqref{eq:conditional-left-orthogonal-invariance} to the continuous map
$\mathbf z\mapsto(\mathbf z,\mathbf z^\top\mathbf z)$.
Since $(\mathbf b\mathbf Z)^\top(\mathbf b\mathbf Z)=\mathbf H$, this gives
\begin{equation}\label{eq:conditional-rotation-joint-gram-law}
\mathcal L(\mathbf b\mathbf Z,\mathbf H\mid\mathcal G)
=\mathcal L(\mathbf Z,\mathbf H\mid\mathcal G)
\quad\text{almost surely}.
\end{equation}
By a monotone class argument,
\eqref{eq:conditional-rotation-joint-gram-law} implies that
$\mathbb P(\Gamma\cap\{\mathbf b\mathbf Z\in A\})
=\mathbb P(\Gamma\cap\{\mathbf Z\in A\})$ for every
$\Gamma\in\mathcal H$ and Borel $A$.
Thus, for each fixed $\mathbf b\in O(d)$ and Borel $A$,
$$
\mathbb P(\mathbf b\mathbf Z\in A\mid\mathcal H)
=\mathbb P(\mathbf Z\in A\mid\mathcal H) \quad\text{almost surely},
$$
with an exceptional null set that may depend on $\mathbf b$ and $A$.

To choose a version invariant under all rotations at once, let $L$
be a regular conditional distribution of $\mathbf Z$ given $\mathcal H$.
Such a kernel exists for this arbitrary $\sigma$-field because
$\mathbb R^{d\times t}$ is a standard Borel space
\citep[Theorem~8.5(i) and the discussion on p.~168]{kallenberg2021foundations}.
Define
\begin{equation}\label{eq:haar-averaged-conditional-kernel}
K(\omega,A)
=\int_{O(d)}L(\omega,\mathbf b^{-1}A)\,\lambda_d(d\mathbf b).
\end{equation}
The continuous action $(\mathbf b,\mathbf z)\mapsto\mathbf b\mathbf z$
and \citet[Lemma~3.2(i)]{kallenberg2021foundations} ensure joint
measurability of the integrand, so $K$ is a probability kernel.
For every $\Gamma\in\mathcal H$ and Borel $A$, Fubini's theorem
\citep[Theorem~1.29]{kallenberg2021foundations} gives
\begin{equation}\label{eq:haar-averaged-kernel-is-version}
\begin{aligned}
\mathbb E\bigl(\mathbf1_\Gamma K(\cdot,A)\bigr)
&=\int_{O(d)}\mathbb P(\Gamma\cap\{\mathbf b\mathbf Z\in A\})\,\lambda_d(d\mathbf b)\\
&=\int_{O(d)}\mathbb P(\Gamma\cap\{\mathbf Z\in A\})\,\lambda_d(d\mathbf b)\\
&=\mathbb P(\Gamma\cap\{\mathbf Z\in A\}),
\end{aligned}
\end{equation}
which verifies that $K$ is a version of
$\mathcal L(\mathbf Z\mid\mathcal H)$.

$K$ is also invariant under all orthogonal
transformations: For every $\omega$, $\mathbf b_0\in O(d)$, and Borel $A$,
left invariance of $\lambda_d$ gives
\begin{equation}\label{eq:haar-averaged-kernel-invariance}
K(\omega,\mathbf b_0^{-1}A)
=\int_{O(d)}L(\omega,(\mathbf b_0\mathbf b)^{-1}A)\,\lambda_d(d\mathbf b)
=K(\omega,A).
\end{equation}

\smallskip

Since $\mathbf H$ is $\mathcal H$-measurable,
\citet[Theorem~8.5(ii)]{kallenberg2021foundations} gives
$$
\begin{aligned}
L(\omega,\mathcal F_{\mathbf H(\omega)})
&=\int_{\mathbb R^{d\times t}}
\mathbf1_{\{\mathbf z^\top\mathbf z=\mathbf H(\omega)\}}
\,L(\omega,d\mathbf z)\\
&=\mathbb E\!\left[
\mathbf1_{\{\mathbf Z^\top\mathbf Z=\mathbf H\}}
\mid\mathcal H\right](\omega)
=1
\quad\text{almost surely},
\end{aligned}
$$
where $\mathbf H=\mathbf Z^\top\mathbf Z$.
Thus $L(\omega,\mathcal F_{\mathbf H(\omega)})=1$ on an
$\mathcal H$-measurable event $\Omega_0=\left\{\omega: L\left(\omega, \mathcal{F}_{\mathbf{H}(\omega)}\right)=1\right\}$ of probability one.
Since orthogonal transformations preserve each $\mathcal F_{\mathbf h}$,
\eqref{eq:haar-averaged-conditional-kernel} gives the same support
property for $K$ on $\Omega_0$.
For every $\omega\in\Omega_0$, $K(\omega,\cdot)$ is supported on
$\mathcal F_{\mathbf H(\omega)}$ and is invariant under $O(d)$ by
\eqref{eq:haar-averaged-kernel-invariance}.
The uniqueness of the invariant probability measure on this orbit
then proves the final assertion.
\end{proof}

\begin{lemma}[Exchangeability on an orthogonal orbit]
\label{lem:mvn-orbit-exchangeability}
Let \(d,t,M\geq1\) be integers, let \(\mathcal H\) be a
\(\sigma\)-field, and let $\mathbf{R}\in\mathbb R^{d\times t}$ be a random
matrix such that $\mathbf{R}^\top \mathbf{R}$ is \(\mathcal H\)-measurable. Suppose that there exist a version $K$ of
$\mathcal L(\mathbf R\mid\mathcal H)$ and an event $\Omega_0$
with $\mathbb P(\Omega_0)=1$ such that, for every $\omega\in\Omega_0$,
$K(\omega,\cdot)$ is the unique probability measure supported on
$\mathcal F_{\mathbf R(\omega)^\top\mathbf R(\omega)}$
that is invariant under left multiplication by $O(d)$. Let
$\mathbf{B}^{(1)},\ldots,\mathbf{B}^{(M)}$ be i.i.d. Haar matrices on \(O(d)\), and
suppose that
$$
\sigma\!\left(\mathbf{B}^{(1)},\ldots,\mathbf{B}^{(M)}\right)
\indp
\mathcal H\vee\sigma(\mathbf{R}).
$$
Then
$$
\mathbf{R},\ \mathbf{B}^{(1)}\mathbf{R},\ldots,\mathbf{B}^{(M)}\mathbf{R}
$$
are conditionally independent with common conditional law \(K\) given
\(\mathcal H\). In particular, they are exchangeable conditional on
\(\mathcal H\).
\end{lemma}

\begin{proof}
Use the orbit probability measure $\nu_{\mathbf x}$ defined by
\eqref{eq:orthogonal-orbit-measure} in the proof of
Lemma~\ref{lem:conditional-law-on-gram-fiber}.
For every Borel $A$, the nonnegative function
$(\mathbf b,\mathbf x)\mapsto\mathbf1_A(\mathbf b\mathbf x)$
is jointly measurable, since matrix multiplication is continuous.
Thus \citet[Lemma~1.28(ii)]{kallenberg2021foundations},
applied with the probability measure $\lambda_d$, implies that
$\mathbf x\mapsto\nu_{\mathbf x}(A)$ is measurable.

Equation~\eqref{eq:orthogonal-orbit-measure-properties} and the uniqueness
assumption on $K$ give $\nu_{\mathbf R(\omega)}=K(\omega,\cdot)$ for every
$\omega\in\Omega_0$.
Conditional on $\mathcal H\vee\sigma(\mathbf R)$, the independence
assumption makes the matrices $\mathbf B^{(m)}\mathbf R$, $m=1,\ldots,M$,
independent with common law $\nu_{\mathbf R}=K$ almost surely.

Let $\mathbf B^{(0)}=\mathbf I_d$, take bounded Borel functions
$f_0,\ldots,f_M$ on $\mathbb R^{d\times t}$, and define
$Kf_m=\int f_m(\mathbf z)K(\cdot,d\mathbf z)$.
Let $\mathcal J=\mathcal H\vee\sigma(\mathbf R)$.
The factor $f_0(\mathbf R)$ is $\mathcal J$-measurable.
For $m=1,\ldots,M$, the conditional law of $\mathbf B^{(m)}\mathbf R$
given $\mathcal J$ is $K$, so
$\mathbb E(f_m(\mathbf B^{(m)}\mathbf R)\mid\mathcal J)=Kf_m$.
Their conditional independence given $\mathcal J$ therefore gives, almost surely,
$$
\begin{aligned}
\mathbb E\!\left(
\prod_{m=0}^M f_m(\mathbf B^{(m)}\mathbf R)\mid\mathcal J
\right)
&=f_0(\mathbf R)\,
\mathbb E\!\left(
\prod_{m=1}^M f_m(\mathbf B^{(m)}\mathbf R)\mid\mathcal J
\right)\\
&=f_0(\mathbf R)\prod_{m=1}^M
\mathbb E\!\left(
 f_m(\mathbf B^{(m)}\mathbf R)\mid\mathcal J
\right)\\
&=f_0(\mathbf R)\prod_{m=1}^M Kf_m.
\end{aligned}
$$
Since $\mathcal H\subseteq\mathcal J$, the tower property gives
\begin{align}
\label{eq:mvn-orbit-product-factorization}
\mathbb E\!\left(
\prod_{m=0}^M f_m(\mathbf B^{(m)}\mathbf R)\mid\mathcal H
\right)
&=\mathbb E\!\left[
\mathbb E\!\left(
\prod_{m=0}^M f_m(\mathbf B^{(m)}\mathbf R)\mid\mathcal J
\right)\mid\mathcal H
\right]\notag\\
&=\mathbb E\!\left(
f_0(\mathbf R)\prod_{m=1}^M Kf_m\mid\mathcal H
\right)\notag\\
&=\prod_{m=0}^M Kf_m,
\end{align}
where the last equality uses the $\mathcal H$-measurability of each
$Kf_m$ and the identity
$\mathbb E(f_0(\mathbf R)\mid\mathcal H)=Kf_0$.
Equation~\eqref{eq:mvn-orbit-product-factorization} proves the stated
conditional independence and common conditional law, and therefore
conditional exchangeability.
\end{proof}

\subsection{Proof of Proposition~\ref{prop: multi gaussian CRT valid}}\label{app:proof-prop-multi-gaussian-crt-valid}

Recall that \(\widetilde{\mathbf X}_{\mc T}^{(1)},\ldots,\widetilde{\mathbf X}_{\mc T}^{(M)}\) are the copies generated from Algorithm~\ref{alg:multi gaussian CRT}, and that the sufficient statistic is denoted by 
\[
\Phi_{\rm MVN}
=
\left(
\mathbf X_{\mc T}^{\top}\mathbf X_{\mc S},
\mathbf X_{\mc T}^{\top}\mathbf X_{\mc T}
\right).
\]
We first prove that, for each \(m\in[M]\),
\begin{equation}\label{equal sufficient}
\mathbf{X}_{\mc{T}}^T{\mathbf{X}_{\mc{S}}}
=
\left(\widetilde{\mathbf{X}}_{\mc{T}}^{(m)}\right)^T{\mathbf{X}_{\mc{S}}},
\qquad
\mathbf{X}_{\mc{T}}^T\mathbf{X}_{\mc{T}}
=
\left(\widetilde{\mathbf{X}}_{\mc{T}}^{(m)}\right)^T
\widetilde{\mathbf{X}}_{\mc{T}}^{(m)} .
\end{equation}
We then prove that
\begin{equation}\label{exchangeable distribution}
\mathbf X_{\mc T},
\widetilde{\mathbf X}_{\mc T}^{(1)},
\ldots,
\widetilde{\mathbf X}_{\mc T}^{(M)}
\quad
\text{are exchangeable conditional on }(\mathbf X_{\mc S},\Phi_{\rm MVN}).
\end{equation}
Together with Proposition~\ref{prop: CRT valid}, these two facts imply the stated finite-sample validity.

\begingroup
\makeatletter
\def\toclevel@paragraph{3}
\makeatother
\paragraph{Proof of \eqref{equal sufficient}}
Note that $$\widetilde{\mathbf{X}}_{\mc{T}}=\mathbf{X}_{\mc{S}}\left(\mathbf{X}_{\mc{S}}^T\mathbf{X}_{\mc{S}}\right)^{-1}\mathbf{X}_{\mc{S}}^T\mathbf{X}_{\mc{T}}+\mathbf{P}_{\mc{R}}^T\mathbf{U}\mathbf{Q},$$
where $\mathbf{U}$ has distribution $\mu_{n-s,r}$, the invariant probability measure on $W_{n-s,r}$ (see Proposition 7.2 in \citet{eaton1983}).
By the definition of $\mathbf P_{\mc R}$ in Section~\ref{sec:mvn-CRT} and \eqref{construction of Q}, we have
$$\mathbf{P}_{\mc{R}}{\mathbf{X}_{\mc{S}}}=\mathbf{0},
$$ and 
\begin{equation}\label{eq:mvn-residual-gram}
\mathbf Q^\top\mathbf Q
=
\mathbf X_{\mc T}^\top\mathbf X_{\mc T}
-
\mathbf X_{\mc T}^\top\mathbf X_{\mc S}
(\mathbf X_{\mc S}^\top\mathbf X_{\mc S})^{-1}
\mathbf X_{\mc S}^\top\mathbf X_{\mc T}.
\end{equation}

Therefore,
\begin{align*}
  \widetilde{\mathbf{X}}_{\mc{T}}^T{\mathbf{X}_{\mc{S}}}&=\mathbf{X}_{\mc{T}}^T\mathbf{X}_{\mc{S}}\left(\mathbf{X}_{\mc{S}}^T\mathbf{X}_{\mc{S}}\right)^{-1}\mathbf{X}_{\mc{S}}^T{\mathbf{X}_{\mc{S}}}=\mathbf{X}_{\mc{T}}^T\mathbf{X}_{\mc{S}}\\
  \widetilde{\mathbf{X}}_{\mc{T}}^T\widetilde{\mathbf{X}}_{\mc{T}}&=\mathbf{X}_{\mc{T}}^T\mathbf{X}_{\mc{S}}\left(\mathbf{X}_{\mc{S}}^T\mathbf{X}_{\mc{S}}\right)^{-1}\mathbf{X}_{\mc{S}}^T\mathbf{X}_{\mc{S}}\left(\mathbf{X}_{\mc{S}}^T\mathbf{X}_{\mc{S}}\right)^{-1}\mathbf{X}_{\mc{S}}^T\mathbf{X}_{\mc{T}} +\mathbf{Q}^T\mathbf{U}^T\mathbf{P}_{\mc{R}}\mathbf{P}_{\mc{R}}^T \mathbf{U}\mathbf{Q}  \\
  &=\mathbf{X}_{\mc{T}}^T\mathbf{X}_{\mc{S}}\left(\mathbf{X}_{\mc{S}}^T\mathbf{X}_{\mc{S}}\right)^{-1}\mathbf{X}_{\mc{S}}^T\mathbf{X}_{\mc{T}}
  + \mathbf{Q}^T \mathbf{Q}\\
  &=\mathbf{X}_{\mc{T}}^T {\mathbf{X}_{\mc{T}}}.
\end{align*}

\paragraph{Proof of \eqref{exchangeable distribution}}

Let
$$
\mathbf Z=\mathbf P_{\mc R}\mathbf X_{\mc T},
\qquad
\mathbf G_{\mc T\mc S}=\mathbf X_{\mc T}^{\top}\mathbf X_{\mc S},
\qquad
\mathbf{F}=
\mathbf X_{\mc S}
(\mathbf X_{\mc S}^{\top}\mathbf X_{\mc S})^{-1}\mathbf G_{\mc T\mc S}^\top,
$$
and define
$$
\mathbf{\Sigma}_{\mc T\mid\mc S}
=
\mathbf{\Sigma}_{\mc T\mc T}
-
\mathbf{\Sigma}_{\mc T\mc S}
\mathbf{\Sigma}_{\mc S\mc S}^{-1}
\mathbf{\Sigma}_{\mc S\mc T}.
$$

To compute the conditional law of $\mathbf Z$, define the
regression residual matrix
\begin{equation}\label{eq:mvn-projected-residual-def}
  \bs\varepsilon
=\mathbf X_{\mc T}
-\mathbf X_{\mc S}\mathbf\Sigma_{\mc S\mc S}^{-1}
 \mathbf\Sigma_{\mc S\mc T}.
\end{equation}
By the Gaussian conditioning formula, the independent rows of
$\bs\varepsilon$ have Gaussian distributions with conditional mean zero and conditional covariance
$\mathbf\Sigma_{\mc T\mid\mc S}$ given $\mathbf X_{\mc S}$.
Since $\mathbf P_{\mc R}\mathbf X_{\mc S}=\mathbf0$,
we have $\mathbf Z=\mathbf P_{\mc R}\bs\varepsilon$ and
$\mathbb E(\mathbf Z\mid\mathbf X_{\mc S})=\mathbf0$.
For $a,b\in\{1,\ldots,n-s\}$ and $j,k\in\{1,\ldots,t\}$,
\begin{equation}\label{eq:mvn-projected-residual-covariance}
\begin{aligned}
&\operatorname{Cov}(Z_{aj},Z_{bk}\mid\mathbf X_{\mc S})\\
={}&\sum_{i=1}^n\sum_{\ell=1}^n
(P_{\mc R})_{ai}(P_{\mc R})_{b\ell}
\operatorname{Cov}(\varepsilon_{ij},\varepsilon_{\ell k}
\mid\mathbf X_{\mc S})\\
={}&(\Sigma_{\mc T\mid\mc S})_{jk}
\sum_{i=1}^n(P_{\mc R})_{ai}(P_{\mc R})_{bi}
=\delta_{ab}(\Sigma_{\mc T\mid\mc S})_{jk},
\end{aligned}
\end{equation}
where $\delta_{ab}$ is the Kronecker delta, the second equality in
\eqref{eq:mvn-projected-residual-covariance} uses conditional
independence of the residual rows, and the last equality uses
$\mathbf P_{\mc R}\mathbf P_{\mc R}^{\top}=\mathbf I_{n-s}$.
Given $\mathbf X_{\mc S}$, the matrix $\mathbf P_{\mc R}$ is fixed,
so $\mathbf Z$ is jointly Gaussian and its rows are independent.
In other words, we have
\begin{equation}
\label{eq:mvn-residual-conditional-law}
\mathbf Z
\mid\mathbf X_{\mc S}
\sim
\mathcal{N}_{n-s,t}
\left(
\mathbf 0_{(n-s)\times t},
\mathbf I_{n-s}\otimes\mathbf{\Sigma}_{\mc T\mid\mc S}
\right).
\end{equation}
Since $\mathbf{\Sigma}_{\mc T\mid\mc S}\succ\mathbf0$, the conditional law in
\eqref{eq:mvn-residual-conditional-law} has density
\begin{equation}
\label{eq:mvn-residual-conditional-density}
f(\mathbf z\mid\mathbf X_{\mc S})
\propto
\exp\!\left\{
-\frac12
\operatorname{tr}\!\left(
\mathbf{\Sigma}_{\mc T\mid\mc S}^{-1}\mathbf z^\top \mathbf z
\right)
\right\},
\qquad
\mathbf z\in\mathbb R^{(n-s)\times t}.
\end{equation}
The density in \eqref{eq:mvn-residual-conditional-density} depends on
$\mathbf z$ only through $\mathbf z^\top \mathbf z$. Since
$(\mathbf{b}\mathbf z)^\top \mathbf{b}\mathbf z=\mathbf z^\top \mathbf z$ for every $\mathbf{b}\in O(n-s)$, and the transformation $\mathbf z\mapsto\mathbf b\mathbf z$ has absolute Jacobian $|\det\mathbf b|^t=1$, it follows that
\begin{equation}
\label{eq:mvn-residual-left-invariance}
\mathcal L(\mathbf{b}\mathbf Z\mid\mathbf X_{\mc S})
=
\mathcal L(\mathbf Z\mid\mathbf X_{\mc S})
\qquad
\text{for every }\mathbf{b}\in O(n-s).
\end{equation}
Since $\mathbf{\Sigma}_{\mc T\mid\mc S}\succ\mathbf 0$, \(\operatorname{rank}(\mathbf Z)=\min\{n-s,t\}\) almost surely. 

To compute the cross-covariance between $\mathbf Z$ and $\mathbf G_{\mc T\mc S}=\mathbf X_{\mc T}^{\top}\mathbf X_{\mc S}$, use \eqref{eq:mvn-projected-residual-def} to write
$$
\mathbf G_{\mc T\mc S}
=\mathbf\Sigma_{\mc T\mc S}\mathbf\Sigma_{\mc S\mc S}^{-1}
\mathbf X_{\mc S}^{\top}\mathbf X_{\mc S}
+\bs\varepsilon^{\top}\mathbf X_{\mc S}.
$$
Given $\mathbf X_{\mc S}$, only
$\bs\varepsilon^{\top}\mathbf X_{\mc S}$ is random.
For $a\in\{1,\ldots,n-s\}$, $j,k\in\{1,\ldots,t\}$ and
$u\in\{1,\ldots,s\}$,
\begin{equation}\label{eq:mvn-residual-cross-covariance}
\begin{aligned}
&\operatorname{Cov}(Z_{aj},(G_{\mc T\mc S})_{ku}\mid\mathbf X_{\mc S})\\
={}&\sum_{i=1}^n\sum_{\ell=1}^n
(P_{\mc R})_{ai}(X_{\mc S})_{\ell u}
\operatorname{Cov}(\varepsilon_{ij},\varepsilon_{\ell k}
\mid\mathbf X_{\mc S})\\
={}&(\Sigma_{\mc T\mid\mc S})_{jk}
\sum_{i=1}^n(P_{\mc R})_{ai}(X_{\mc S})_{iu}\\
={}&(\Sigma_{\mc T\mid\mc S})_{jk}
(\mathbf P_{\mc R}\mathbf X_{\mc S})_{au}=0,
\end{aligned}
\end{equation}
where conditional independence of the residual rows gives the second
equality in \eqref{eq:mvn-residual-cross-covariance}, and
$\mathbf P_{\mc R}\mathbf X_{\mc S}=\mathbf0$ gives the last equality.
Both $\mathbf Z$ and $\mathbf G_{\mc T\mc S}$ are affine functions of
$\bs\varepsilon$ given $\mathbf X_{\mc S}$.
Moreover,
conditional on \(\mathbf X_{\mc S}\), the pair $(\mathbf Z,\mathbf G_{\mc T\mc S})$ is jointly
Gaussian, and its cross-covariance vanishes because
\(\mathbf P_{\mc R}\mathbf X_{\mc S}=\mathbf 0\). Hence
\begin{equation}\label{eq:mvn-residual-cross-independence}
\mathbf Z\indp \mathbf G_{\mc T\mc S}
\quad\text{conditional on }\mathbf X_{\mc S}.
\end{equation}

\textbf{Part 1.}
Suppose \(n-s>t\).

It suffices to show that, conditional on \((\mathbf X_{\mc S},\Phi_{\rm MVN})\), the orthonormal factors of the observed residual matrix and its randomised copies are i.i.d. with distribution \(\mu_{n-s,t}\).

By \eqref{construction of Q}, we can write $\mathbf Q=(\mathbf Z^{\top}\mathbf Z)^{1/2}$, which is the unique positive-definite square root of $\mathbf{X}_{\mc{T}}^{\top} \mathbf{P}_{\mc{R}}^{\top} \mathbf{P}_{\mc{R}} \mathbf{X}_{\mc{T}}$ a.s. 
Let $\mathbf{U}_0=\mathbf Z\mathbf Q^{-1}$. Then $\mathbf Z=\mathbf{U}_0\mathbf Q$ and $\mathbf{U}_0\in W_{n-s,t}$. 
Applying \citet[Proposition~7.4]{eaton1983} conditionally on
\(\mathbf X_{\mc S}\), we have
\begin{equation}\label{eq:mvn-haar-U0}
  \mathbf{U}_0\sim\mu_{n-s,t},
\qquad
\mathbf{U}_0\indp\mathbf Q
\quad\text{conditional on }\mathbf X_{\mc S}.
\end{equation}

Because $(\mathbf{U}_0,\mathbf Q)$ is a measurable function of \(\mathbf Z\),
\eqref{eq:mvn-residual-cross-independence} gives
$$
(\mathbf{U}_0,\mathbf Q)\indp \mathbf G_{\mc T\mc S}
\quad\text{conditional on }\mathbf X_{\mc S}.
$$

Since $\mathbf{U}_0\indp\mathbf Q\mid\mathbf X_{\mc S}$ and
$(\mathbf{U}_0,\mathbf Q)\indp \mathbf G_{\mc T\mc S} \mid\mathbf X_{\mc S}$, we obtain
\begin{equation}\label{eq:mvn-stiefel-independence}
\mathbf{U}_0
\indp
\left(
\mathbf Q, \mathbf X_{\mc T}^{\top}\mathbf X_{\mc S}
\right)
\quad\text{conditional on }\mathbf X_{\mc S}.
\end{equation}

Given \(\mathbf X_{\mc S}\), \eqref{eq:mvn-residual-gram} and the uniqueness of the positive-definite square root show that \(\Phi_{\rm MVN}\) determines \((\mathbf Q,\mathbf X_{\mc T}^{\top}\mathbf X_{\mc S})\); conversely, rearranging \eqref{eq:mvn-residual-gram} recovers \(\mathbf X_{\mc T}^{\top}\mathbf X_{\mc T}\) from \((\mathbf Q,\mathbf X_{\mc T}^{\top}\mathbf X_{\mc S})\). 
Thus, given \(\mathbf X_{\mc S}\), 
\(\Phi_{\rm MVN}\) and
$\left(
\mathbf Q,\mathbf X_{\mc T}^{\top}\mathbf X_{\mc S}
\right)$
determine each other. 
Therefore, \eqref{eq:mvn-haar-U0} and \eqref{eq:mvn-stiefel-independence} yield the conditional law
$$
\mathbf{U}_0\mid(\mathbf X_{\mc S},\Phi_{\rm MVN})
\sim\mu_{n-s,t}.
$$
Moreover, conditional on $(\mathbf X_{\mc S},\Phi_{\rm MVN})$,
both $\mathbf Q$ and $\mathbf F$ are fixed, and
$\mathbf X_{\mc T}=\mathbf F+\mathbf P_{\mc R}^{\top}\mathbf U_0\mathbf Q$.
Algorithm~\ref{alg:multi gaussian CRT} generates $\mathbf{U}^{(1)},\ldots,\mathbf{U}^{(M)}$ independently of the data and independently of one another. 
By \citet[Proposition~7.2]{eaton1983}, each $\mathbf{U}^{(m)}$ has distribution \(\mu_{n-s,t}\).
Together with
$\mathbf{U}_0\mid(\mathbf X_{\mc S},\Phi_{\rm MVN})\sim\mu_{n-s,t}$,
these independence and distributional properties show that, conditional on
\((\mathbf X_{\mc S},\Phi_{\rm MVN})\),
$\mathbf{U}_0,\mathbf{U}^{(1)},\ldots,\mathbf{U}^{(M)}$ are i.i.d. with distribution
\(\mu_{n-s,t}\).
\par

Since $\mathbf F$ and $\mathbf Q$ are fixed under the conditioning, applying the same affine transformation $\mathbf{u}\mapsto \mathbf{F}+\mathbf P_{\mc R}^{\top}\mathbf{u}\mathbf Q$ to these i.i.d.\ matrices proves the required exchangeability conditional on \((\mathbf X_{\mc S},\Phi_{\rm MVN})\).

\textbf{Part 2.} Suppose $n-s\leq t$.
Let $\mathbf H=\mathbf Z^\top\mathbf Z$.
Since $\mathbf H$ is a measurable function of $\mathbf Z$,
\eqref{eq:mvn-residual-cross-independence} implies
$(\mathbf Z,\mathbf H)\indp\mathbf G_{\mc T\mc S}\mid\mathbf X_{\mc S}$.
The weak union property of conditional independence therefore gives
\begin{equation}\label{eq:mvn-gram-conditioning-law}
\mathcal L(\mathbf Z\mid\mathbf X_{\mc S},\mathbf G_{\mc T\mc S},\mathbf H)
=\mathcal L(\mathbf Z\mid\mathbf X_{\mc S},\mathbf H)
\quad\text{almost surely}.
\end{equation}

By \eqref{eq:mvn-residual-left-invariance},
Lemma~\ref{lem:conditional-law-on-gram-fiber} applies to
$\mathbf Z$ with $\mathcal G=\sigma(\mathbf X_{\mc S})$.
It supplies a version $K$ of
$\mathcal L(\mathbf Z\mid\mathbf X_{\mc S},\mathbf H)$ that,
on one event of probability one, is the unique probability measure
supported on
$$
\{\mathbf z\in\mathbb R^{(n-s)\times t}:\mathbf z^\top\mathbf z=\mathbf H\}
$$
and invariant under every left multiplication by $O(n-s)$.
By \eqref{eq:mvn-gram-conditioning-law}, the same $K$ is a version
of $\mathcal L(\mathbf Z\mid\mathbf X_{\mc S},\mathbf G_{\mc T\mc S},\mathbf H)$,
with these properties on the same event.
Moreover, \eqref{eq:mvn-residual-gram} gives
$\sigma(\mathbf X_{\mc S},\mathbf G_{\mc T\mc S},\mathbf H)
=\sigma(\mathbf X_{\mc S},\Phi_{\rm MVN})$.

In this case $r=n-s$ and $\mathbf Q=\mathbf Z$. The matrices
$\mathbf U^{(1)},\ldots,\mathbf U^{(M)}$ generated by
Algorithm~\ref{alg:multi gaussian CRT} are therefore i.i.d. Haar matrices in
$O(n-s)$, jointly independent of the observed data.
Apply Lemma~\ref{lem:mvn-orbit-exchangeability} with
$\mathbf R=\mathbf Z$,
$\mathcal H=\sigma(\mathbf X_{\mc S},\Phi_{\rm MVN})$,
and $\mathbf B^{(m)}=\mathbf U^{(m)}$ to obtain that
$\mathbf Z,\mathbf U^{(1)}\mathbf Z,\ldots,\mathbf U^{(M)}\mathbf Z$
are conditionally i.i.d. given
$(\mathbf X_{\mc S},\Phi_{\rm MVN})$.

Under the same conditioning, $\mathbf P_{\mc R}$ and $\mathbf F$ are fixed, and
$$
\mathbf X_{\mc T}=\mathbf F+\mathbf P_{\mc R}^{\top}\mathbf Z,
\qquad
\widetilde{\mathbf X}_{\mc T}^{(m)}
=\mathbf F+\mathbf P_{\mc R}^{\top}\mathbf U^{(m)}\mathbf Z.
$$
Hence
$\mathbf X_{\mc T},\widetilde{\mathbf X}_{\mc T}^{(1)},\ldots,
\widetilde{\mathbf X}_{\mc T}^{(M)}$
are conditionally i.i.d., and therefore exchangeable, given
$(\mathbf X_{\mc S},\Phi_{\rm MVN})$.
This proves \eqref{exchangeable distribution} for $n-s\leq t$.

\endgroup

\subsection{Residual rotation}\label{app:residual-rotation}

The following statement records the local properties of the update of \citet[Section~2.1, Algorithm~1 and Proposition~1]{lin_goodness--fit_2025} used in Section~\ref{sec:ggm}.

\begin{lemma}[Local update by residual rotation]\label{lem:ggm-local-update}
Suppose the rows of $\mathbf X$ are i.i.d. from $\mc M_G$,
and fix $i\in\mc T$.
The update in \eqref{eq:ggm-local-update} preserves both
$\Psi_i(\mathbf X_i;\mathbf X_{N_i})=(\bs F_i,\|\bs R_i\|_2)$
and $\psi_G(\mathbf X)$.
The vectors $\mathbf X_i$ and $\mathbf X_i^{\mathrm{new}}$ are
conditionally independent and identically distributed given either
$(\mathbf X_{N_i},\Psi_i)$ or $(\mathbf X_{-i},\Psi_i)$.
In particular, the pair $(\mathbf X_i,\mathbf X_i^{\mathrm{new}})$
is exchangeable conditional on $(\mathbf X_{-i},\Psi_i)$.
\end{lemma}

\begin{proof}
The local Markov property and independent rows give
$$
\mathcal L(\mathbf X_i\mid\mathbf X_{-i})
=\mathcal L(\mathbf X_i\mid\mathbf X_{N_i})
=\mathcal{N}_n(\alpha_i\bs1_n+\mathbf X_{N_i}\bs\beta_i,\sigma_i^2\mathbf I_n).
$$

Let $\mathbf A_i=[\bs1_n,\mathbf X_{N_i}]$.
Conditional on either $\mathbf X_{-i}$ or $\mathbf X_{N_i}$,
the projection $\mathbf P_i$ is fixed.
The projections $\bs F_i$ and $\bs R_i$ are jointly Gaussian
with cross-covariance
$\sigma_i^2\mathbf P_i(\mathbf I_n-\mathbf P_i)=0$,
and hence are independent.
The residual has mean zero and covariance
$\sigma_i^2(\mathbf I_n-\mathbf P_i)$, so it is isotropic in
$\operatorname{col}(\mathbf A_i)^\perp$. 
If $n>|N_i|+1$, this subspace has positive dimension.

Let $q_i=n-\operatorname{rank}(\mathbf A_i)>0$, and let
$\mathbf B_i\in\mathbb R^{n\times q_i}$ have orthonormal columns
spanning $\operatorname{col}(\mathbf A_i)^\perp$.
Conditional on either $\mathbf X_{-i}$ or $\mathbf X_{N_i}$, the vector
$\sigma_i^{-1}\mathbf B_i^\top\bs R_i$ is standard Gaussian.
Its direction is uniform on the unit sphere and independent of its norm.
This includes $q_i=1$, where the direction is a uniform sign.
Since $\bs F_i$ and $\bs R_i$ are conditionally independent,
the direction $\bs R_i/\|\bs R_i\|_2$ remains uniform on the unit sphere of
$\operatorname{col}(\mathbf A_i)^\perp$ after we also condition on
$\Psi_i=(\bs F_i,\|\bs R_i\|_2)$.

Conditional on either $\mathbf X_{-i}$ or $\mathbf X_{N_i}$, $\mathbf B_i^\top\bs W$
is standard Gaussian and independent of $\mathbf X_i$.
For fixed $\xi=(\mathbf X_{-i},\Psi_i)$, let $Q_{i,\xi}$ be the law of
$\bs F_i+\|\bs R_i\|_2\mathbf B_i\bs u$, where $\bs u$ is uniform
on the unit sphere in $\mathbb R^{q_i}$.
The original and fresh residual directions are conditionally independent
given $\xi$, so \eqref{eq:ggm-local-update} gives
$$
\mathcal L(\mathbf X_i,\mathbf X_i^{\mathrm{new}}\mid\xi)
=Q_{i,\xi}\otimes Q_{i,\xi}.
$$
The same argument with $\mathbf X_{N_i}$ in place of $\mathbf X_{-i}$
shows that $\mathbf X_i$ and $\mathbf X_i^{\mathrm{new}}$ are conditionally
i.i.d. given $(\mathbf X_{N_i},\Psi_i)$. 
The update preserves $\bs F_i$ and $\|\bs R_i\|_2$, and hence $\Psi_i$.

Since $\mathbf X_{N_i}$ has a joint density,
$\operatorname{rank}(\mathbf A_i)=\min\{n,|N_i|+1\}$ almost surely.
If $n>|N_i|+1$, both residual norms are nonzero almost surely.
The stated convention defines the update when either norm is zero.
If $n\leq|N_i|+1$, then $\mathbf P_i=\mathbf I_n$ almost surely,
so the update leaves $\mathbf X_i$ unchanged, which gives preservation
and conditional exchangeability directly.

Since $\mathbf P_i$ is an orthogonal projection onto
$\operatorname{col}(\mathbf A_i)$, preservation of $\bs F_i=\mathbf P_i\mathbf X_i^{\mathrm{new}}=\mathbf P_i \mathbf X_i$ and
$\|\bs R_i\|_2=\|(\mathbf I_n-\mathbf P_i)\mathbf X_i^{\mathrm{new}}\|_2=\|(\mathbf I_n-\mathbf P_i)\mathbf X_i\|_2$ gives
$$
\begin{aligned}
\bs1_n^\top\mathbf X_i^{\mathrm{new}}
&=\bs1_n^\top\bs F_i=\bs1_n^\top\mathbf X_i,\\
\mathbf X_{N_i}^\top\mathbf X_i^{\mathrm{new}}
&=\mathbf X_{N_i}^\top\bs F_i=\mathbf X_{N_i}^\top\mathbf X_i,\\
\|\mathbf X_i^{\mathrm{new}}\|_2^2
&=\|\bs F_i\|_2^2+\|\bs R_i\|_2^2=\|\mathbf X_i\|_2^2.
\end{aligned}
$$
All other entries of $\psi_G$ involve unchanged columns,
so the update preserves $\psi_G(\mathbf X)$.
\end{proof}

\subsection[Proofs for graphical models]{Proofs for graphical models}\label{app:proof-prop-exchangeable}

\subsubsection[Proof of Proposition~\ref{prop: exchangeable}]{Proof of Proposition~\ref{prop: exchangeable}}\label{app:proof-general-graphical-validity}

\begin{proof}
Condition on any independent auxiliary data used for the fixed choices. For almost every $\mathbf s$, let $\pi_{\mathbf s}$ be the conditional law of the full matrix given $\mathbf X_{\mc S}=\mathbf s$. Each local kernel $K_i$ changes only column $i\in\mc T$. We first show that $K_i$ is reversible with respect to $\pi_{\mathbf s}$.

Let $\xi_i=(\mathbf X_{-i},\Psi_i)$, and let $\mathbf X'$ be the full matrix after one update of column $i$. Since the other columns agree, Condition~\ref{cond:local-update} gives conditional symmetry of $(\mathbf X,\mathbf X')$ given $\xi_i$. Moreover, $\xi_i$ includes $\mathbf X_{\mc S}$ because $i\in\mc T$. For bounded measurable $f,g$,
$$
\begin{aligned}
\int f(\mathbf x)(K_i g)(\mathbf x)\,\pi_{\mathbf s}(d\mathbf x)
&=\mathbb E\{f(\mathbf X)g(\mathbf X')\mid\mathbf X_{\mc S}=\mathbf s\}\\
&=\mathbb E\!\left[\mathbb E\{f(\mathbf X)g(\mathbf X')\mid\xi_i\}
                  \mid\mathbf X_{\mc S}=\mathbf s\right]\\
&=\mathbb E\!\left[\mathbb E\{f(\mathbf X')g(\mathbf X)\mid\xi_i\}
                  \mid\mathbf X_{\mc S}=\mathbf s\right]\\
&=\int g(\mathbf x)(K_i f)(\mathbf x)\,\pi_{\mathbf s}(d\mathbf x),
\end{aligned}
$$
where $(K_i g)(\mathbf x)=\int g(\mathbf x')K_i(\mathbf x,d\mathbf x')$.
For indicator functions, this gives the following detailed balance: 
\begin{equation}\label{eq:graphical-local-detailed-balance}
\pi_{\mathbf s}(d\mathbf x)K_i(\mathbf x,d\mathbf x')
=\pi_{\mathbf s}(d\mathbf x')K_i(\mathbf x',d\mathbf x).
\end{equation}
Thus every $K_i$ is reversible with respect to $\pi_{\mathbf s}$. Taking $f=1$ also gives $\pi_{\mathbf s}K_i=\pi_{\mathbf s}$, so subsequent updates encounter the conditional laws on the same sets of full probability.

Let $i_1,\ldots,i_R$ be the complete list of updates in Step~1, with $R=L|\mc I|$. Write products of kernels in chronological order and define
$$
K_+=K_{i_1}\cdots K_{i_R},\qquad
K_-=K_{i_R}\cdots K_{i_1}.
$$
Both kernels are the identity if $R=0$. Define $\langle f,g\rangle_{\pi_{\mathbf s}}=\int fg\,d\pi_{\mathbf s}$. By \eqref{eq:graphical-local-detailed-balance},
$\langle f,K_i g\rangle_{\pi_{\mathbf s}}=\langle K_i f,g\rangle_{\pi_{\mathbf s}}$.
For each $j=1,\ldots,R$, we apply the reversibility of
$K_{i_1},\ldots,K_{i_j}$ in turn to obtain
$$
\langle f,K_+g\rangle_{\pi_{\mathbf s}}
=\left\langle K_{i_j}\cdots K_{i_1}f,
 K_{i_{j+1}}\cdots K_{i_R}g\right\rangle_{\pi_{\mathbf s}}.
$$
At $j=R$ the right side is $\langle K_-f,g\rangle_{\pi_{\mathbf s}}$; the empty product is the identity. For indicator functions, this gives
\begin{equation}\label{eq:graphical-composite-time-reversal}
\pi_{\mathbf s}(d\mathbf x)K_+(\mathbf x,d\mathbf h)
=\pi_{\mathbf s}(d\mathbf h)K_-(\mathbf h,d\mathbf x).
\end{equation}

Equation~\eqref{eq:graphical-composite-time-reversal} is the time-reversal identity in \citet[equation~(1)]{besag_generalized_1989}, with their backward transition equal to $K_+$ and their forward transition equal to $K_-$. It does not require $K_+=K_-$.

Write $\mathbf X^{(0)}=\mathbf X$ and $\mathbf X^{(m)}=\widetilde{\mathbf X}^{(m)}$ for $m=1,\ldots,M$. The hub is $\mathbf X^{(\mathrm{hub})}$, and the branches use independent randomness. 
The joint law is
$$
\begin{aligned}
&\mathcal L(\mathbf X^{(0)},\mathbf X^{(\mathrm{hub})},\mathbf X^{(1)},\ldots,
 \mathbf X^{(M)}\mid\mathbf X_{\mc S}=\mathbf s)\\
={}&\pi_{\mathbf s}(d\mathbf x_0)K_+(\mathbf x_0,d\mathbf h)
             \prod_{m=1}^M K_-(\mathbf h,d\mathbf x_m)\\
={}&\pi_{\mathbf s}(d\mathbf h) K_-(\mathbf h,d\mathbf x_0)
             \prod_{m=1}^M K_-(\mathbf h,d\mathbf x_m),
\end{aligned}
$$
where the first equality follows from the sampling construction and the
conditional independence of the branches given
$(\mathbf X^{(0)},\mathbf X^{(\mathrm{hub})},\mathbf X_{\mc S})$, 
and the second follows from the time-reversal identity
\eqref{eq:graphical-composite-time-reversal}.
Thus the hub has marginal law $\pi_{\mathbf s}$, and conditional on $(\mathbf X_{\mc S}=\mathbf s,\mathbf X^{(\mathrm{hub})}=\mathbf h)$ all $M+1$ matrices are independent with common law $K_-(\mathbf h,\cdot)$. 
The conditional i.i.d. law is invariant under permutations.
After integration over the hub, the target matrices are therefore
jointly exchangeable conditional on $\mathbf X_{\mc S}$.

Under $H_0$, the updates use no response information and all sampling randomness is independent of the inference data. The conditional independence argument in Proposition~\ref{prop: CRT valid} therefore applies, with a constant $\Phi$.

The resulting p-value is valid conditional on the independent
auxiliary data. We obtain unconditional validity by taking
expectations over these data.

\end{proof}

\subsubsection[Proof of Proposition~\ref{prop:ggm-gcrt-valid}]{Proof of Proposition~\ref{prop:ggm-gcrt-valid}}\label{app:proof-ggm-gcrt-valid}

\begin{proof}

Lemma~\ref{lem:ggm-local-update} verifies Condition~\ref{cond:local-update} for the residual rotations and shows that they preserve $\psi_G(\mathbf X)$. Algorithm~\ref{alg: exchangeable} leaves the background columns unchanged, so the first assertion holds.

Let
$\mathcal Z=(\mathbf X,\widetilde{\mathbf X}^{(1)},\ldots,
\widetilde{\mathbf X}^{(M)})$ and $\Psi=\psi_G(\mathbf X)$.
Since all candidates have the same background columns,
Proposition~\ref{prop: exchangeable} implies that $\mathcal Z$
is exchangeable conditional on $\mathbf X_{\mc S}$.
The common value $\Psi$ of $\psi_G$ is unchanged by any
permutation $\pi$ of the candidates. Hence
$$
\mathcal L(\pi\mathcal Z,\Psi\mid\mathbf X_{\mc S})
=
\mathcal L(\mathcal Z,\Psi\mid\mathbf X_{\mc S})
\quad\text{almost surely}.
$$
Therefore, $\mathcal Z$ is exchangeable conditional on
$(\mathbf X_{\mc S},\Psi)$, proving the second assertion.

Under $H_0$, Proposition~\ref{prop: CRT valid} with $\Phi=\psi_G(\mathbf X)$ gives the p-value assertion for the scores defined in Section~\ref{sec: MSCRT}.

\end{proof}

\subsubsection[Proof of Corollary~\ref{cor:ggm-weaker-graph-valid}]{Proof of Corollary~\ref{cor:ggm-weaker-graph-valid}}\label{app:proof-weaker-gaussian-graph}

\begin{proof}
Condition on any independent auxiliary data used for selection. For each $i\in\mc T$, the neighbourhood assumption in Corollary~\ref{cor:ggm-weaker-graph-valid} and Gaussian conditioning give
$$
\mathcal L(\mathbf X_i\mid\mathbf X_{-i})
=\mathcal N_n\!\left(
\mu_i\bs1_n-\frac{1}{\omega_{ii}}
\sum_{j\in N_i}\omega_{ij}(\mathbf X_j-\mu_j\bs1_n),
\frac{1}{\omega_{ii}}\mathbf I_n
\right).
$$
This law depends on $\mathbf X_{-i}$ only through $\mathbf X_{N_i}$, giving \eqref{eq:local Markov property}. Its mean lies in $\operatorname{col}([\bs1_n,\mathbf X_{N_i}])$, so the conditional Gaussian calculation in the proof of Lemma~\ref{lem:ggm-local-update} establishes Condition~\ref{cond:local-update}. Each update of column $i$ preserves $\bs1_n^\top\mathbf X_i$,
$\mathbf X_{N_i}^\top\mathbf X_i$, and $\mathbf X_i^\top\mathbf X_i$.
Since all other columns are unchanged, $\psi_G(\mathbf X)$ is preserved
without any zero restrictions on $\mathbf\Omega_{\mc S\mc S}$. Proposition~\ref{prop: exchangeable} then gives exchangeability conditional on $\mathbf X_{\mc S}$. Since $\psi_G(\widetilde{\mathbf X}^{(m)})=\psi_G(\mathbf X)$ for every $m$, the argument in Appendix~\ref{app:proof-ggm-gcrt-valid} gives exchangeability conditional on $(\mathbf X_{\mc S},\psi_G(\mathbf X))$. Under $H_0$ and the stated scoring conditions, the extension of
Proposition~\ref{prop: CRT valid} in Appendix~\ref{app:scoring-conditions}
gives the p-value assertion.
\end{proof}

\subsubsection[Proof of Example~\ref{ex:gcrt-finite-state}]{Proof of Example~\ref{ex:gcrt-finite-state}}\label{app:proof-finite-state-graphical}

\begin{proof}
Given $\mathbf X_{-i}$, an arrangement of column $i$ has probability
$$
\prod_{\bs v}\prod_{k=1}^{K_i}p_i(k\mid\bs v)^{n_{\bs v,k}},
\qquad \text{ where } p_i(k\mid\bs v)=\mathbb P(X_i=k\mid\bs X_{N_i}=\bs v).
$$

Let $n_{\bs v}=|\mc H_{\bs v}|$ be the number of rows in group
$\mc H_{\bs v}$, and define
$c_{\bs v}=n_{\bs v}!/\prod_{k=1}^{K_i}n_{\bs v,k}!$, the number of
distinct arrangements within that group having the prescribed counts.
The conditional probability of an arrangement given $\mathbf X_{-i}$
depends only on $\Psi_i$.
After we also condition on $\Psi_i$, the original column is therefore
uniform over the $\prod_{\bs v}c_{\bs v}$ compatible columns.
Thus, for any compatible column $\mathbf x$,
$\mathbb P(\mathbf X_i=\mathbf x\mid\mathbf X_{-i},\Psi_i)
=\prod_{\bs v}c_{\bs v}^{-1}$ almost surely.

To determine the update law, fix two compatible columns
$\mathbf x,\mathbf y$.
Within group $\mc H_{\bs v}$, exactly
$\prod_{k=1}^{K_i}n_{\bs v,k}!$ of the $n_{\bs v}!$ permutations
transform the arrangement in $\mathbf x$ into that in $\mathbf y$.
A uniform permutation therefore produces the prescribed arrangement
with probability $\prod_{k=1}^{K_i}n_{\bs v,k}!/n_{\bs v}!
=c_{\bs v}^{-1}$.
The permutations are independent across groups, so
$$
\begin{aligned}
\mathbb P(\mathbf X_i^{\mathrm{new}}=\mathbf y\mid
          \mathbf X_i=\mathbf x,\mathbf X_{-i},\Psi_i)
&=\prod_{\bs v}c_{\bs v}^{-1}\\
&=\mathbb P(\mathbf X_i=\mathbf y\mid\mathbf X_{-i},\Psi_i).
\end{aligned}
$$
The update law does not depend on $\mathbf x$ and equals the
conditional law of the original column.
Hence $\mathbf X_i$ and $\mathbf X_i^{\mathrm{new}}$ are conditionally
i.i.d. given $(\mathbf X_{-i},\Psi_i)$.
The permutations also preserve the counts $\Psi_i$, so
Condition~\ref{cond:local-update} holds.

\end{proof}

\subsection{Proof of Theorem \ref{power analysis}}\label{app:proof-thm-power}

We first give a lower bound on the rejection probability for finitely many randomised copies, in terms of tail probabilities for the observed and copy statistics. We then state the concentration inequalities needed for the separate arguments under $\ell_2$ and $\ell_\infty$ separation.
Denote $T_0=T(\bs{Y},\mathbf{X}_{\mc{T}},\mathbf{X}_{\mc{S}})$
and $T_m=T(\bs{Y},\widetilde{\mathbf{X}}_{\mc{T}}^{(m)},\mathbf{X}_{\mc{S}})$.

\par

The following bound applies without independence among the copies.
\begin{lemma}\label{lem: null tail bound}
Let \(M\ge1\) be an integer, let \(T_0,T_1,\ldots,T_M\) be real random variables, let \(0<\alpha<1\), and let \(k=\lfloor\alpha(M+1)\rfloor\). Define \(p_M=(1+\sum_{m=1}^M\mathbf1\{T_m\ge T_0\})/(M+1)\). If \(k=0\), then \(\mathbb P(p_M\le\alpha)=0\). If \(k\ge1\), then for every deterministic \(q\),
\begin{equation}\label{eq:finite-M-power-reduction}
 \mathbb P(p_M\le\alpha)\ge \mathbb P(T_0>q)-\frac1k\sum_{m=1}^M\mathbb P(T_m\ge q).
\end{equation}
\end{lemma}
\begin{proof}
Let
$$
N_0=\sum_{m=1}^M\mathbf1\{T_m\ge T_0\},
\qquad N_q=\sum_{m=1}^M\mathbf1\{T_m\ge q\}.
$$
For \(k=0\),
$$
\alpha(M+1)<1
\quad\Longrightarrow\quad
p_M=\frac{1+N_0}{M+1}\ge\frac1{M+1}>\alpha
\quad\Longrightarrow\quad
\mathbb P(p_M\le\alpha)=0.
$$
For \(k\ge1\),
$$
\begin{aligned}
p_M>\alpha
&\iff N_0>\alpha(M+1)-1\\
&\iff N_0\ge\lfloor\alpha(M+1)\rfloor=k,
\end{aligned}
$$
and
$$
\begin{aligned}
T_0>q
&\quad\Longrightarrow\quad
\mathbf1\{T_m\ge T_0\}\le\mathbf1\{T_m\ge q\}
\quad(m\in[M])\\
&\quad\Longrightarrow\quad N_0\le N_q.
\end{aligned}
$$
Hence
$$
\begin{aligned}
\mathbb P(p_M\le\alpha)
&\ge\mathbb P(T_0>q,\ p_M\le\alpha)\\
&=\mathbb P(T_0>q)-\mathbb P(T_0>q,\ p_M>\alpha)\\
&=\mathbb P(T_0>q)-\mathbb P(T_0>q,\ N_0\ge k)\\
&\ge\mathbb P(T_0>q)-\mathbb P(T_0>q,\ N_q\ge k)\\
&\ge\mathbb P(T_0>q)-\mathbb P(N_q\ge k)\\
&\ge\mathbb P(T_0>q)-\frac{\mathbb E N_q}{k}
\qquad\text{(Markov's inequality)}\\
&=\mathbb P(T_0>q)-\frac1k\sum_{m=1}^M\mathbb P(T_m\ge q).
\end{aligned}
$$
\end{proof}

Let \(\delta=1-\beta\) and \(a=\alpha\delta/4\). Since \(M\ge2/\alpha\), we have $k\geq\alpha(M+1)-1\geq\alpha M-1\geq\alpha M/2\geq1$. By \eqref{eq:finite-M-power-reduction}, it is enough to find a deterministic threshold \(q\) such that
\begin{equation}\label{eq:power-tail-targets}
 \mathbb P_\theta(T_0>q)\ge1-\delta/2,
 \qquad \mathbb P_\theta(T_m\ge q)\le a\quad(m\in[M]).
\end{equation}

Indeed, these bounds give
$\mathbb P_\theta(p_M\le\alpha)
\ge1-\delta/2-Ma/k\ge1-\delta=\beta$,
uniformly in $M\ge2/\alpha$.

For both separation norms, the remaining proof has two steps:
\begin{enumerate}
  \item Bound the upper tail of the copy statistic \(T_m\) to choose a deterministic threshold \(q\) satisfying \(\mathbb P_\theta(T_m\ge q)\le a\).
  \item Show that, under the corresponding alternatives, the observed statistic \(T_0\) exceeds \(q\) with probability at least \(1-\delta/2\).
\end{enumerate}
We prove these bounds uniformly over the respective alternative classes for all sufficiently large \(n\). The threshold $q$ and separation constants may depend on the fixed error probabilities \(a\) and \(\delta\).

\begingroup
\subsubsection{Auxiliary results}
Recall from Appendix~\ref{app:haar-preliminaries-draft} that
$\mu_{n,d}$ is the unique probability measure on the Stiefel
manifold $W_{n,d}$ invariant under left multiplication by $O(n)$.
The next two lemmas give concentration bounds under $\mu_{n,d}$.
In the next two lemmas, zero coefficients give an identically zero function; the displayed bounds are stated for nonzero coefficients.
\begin{lemma}[{\citep[Section 7.1]{10.1214/23-EJP966}}]\label{first order}
  Suppose $n>d$ and $\mathbf{A}\sim \mu_{n,d}$, and $f_1(\mathbf{A})=\bs{m}^T{\rm  vec}(\mathbf{A})$ for some $\bs{m}\in \bbR^{nd\times 1}$. For any \( t \geq 0 \), we have
  $$\mu_{n,d} \left( | f_1 | \geq t \right)\leq 2\exp\left(-\frac{(n-1)t^2}{8\left\|\bs{m}\right\|_2^2}\right).$$
\end{lemma}

Let $\mathbf C\in\mathbb R^{n\times d}$ satisfy
$\operatorname{vec}(\mathbf C)=\bs m$.
Equation~(7.1) of \citet{10.1214/23-EJP966} uses the supremum of the
Frobenius (Hilbert--Schmidt) norm of the orthogonal projection of $\mathbf C$
onto the tangent space at $\mathbf A$.
Every such projection has norm at most
$\|\mathbf C\|_{\mathrm{HS}}=\|\bs m\|_2$.
Replacing that supremum by $\|\bs m\|_2$ increases the upper bound
and gives the stated inequality.

\begin{lemma}[{\citep[Theorem 7.1]{10.1214/23-EJP966}}]\label{second order}
  Suppose $n>d$ and $\mathbf{A}\sim \mu_{n,d}$, and $f_2(\mathbf{A})={\rm vec }(\mathbf{A})^T\mathbf{m}{\rm vec}(\mathbf{A})$ for some matrix $\mathbf{m}\in \bbR^{nd\times nd}$. For any \( t \geq 0 \), we have
  \begin{equation}\label{eq:stiefel-quadratic-concentration}
  \mu_{n,d} \left( | f_2 - \operatorname{tr}(\mathbf{m})/n | \geq t \right)
  \leq 2 \exp \left( -\frac{1}{C} \min \left( \frac{(n-2)^2 t^2}{\|\mathbf{m}\|_{{\rm HS}}^2}, \frac{(n-2) t}{\|\mathbf{m}\|_{{\rm op}}} \right) \right),
  \end{equation}
  where we may take $C=128e^2/\log (2)$.
  Here $\|\cdot\|_{\mathrm{HS}}$ denotes the Frobenius norm, also written $\|\cdot\|_F$.
\end{lemma}

The following moment calculation explains the centring in
Lemmas~\ref{first order} and~\ref{second order}.
For $i,k\in[n]$ and $j,l\in[d]$, invariance under a change of sign
of row $i$ gives $\mathbb E A_{ij}=0$ and
$\mathbb E(A_{ij}A_{kl})=0$ when $i\ne k$.
Row permutation invariance makes $\mathbb E(A_{ij}A_{il})$
independent of $i$.
Since $\mathbf A^\top\mathbf A=\mathbf I_d$, we have
$\sum_{i=1}^n A_{ij}A_{il}=\mathbf1\{j=l\}$.
On taking expectations, we obtain
$n\,\mathbb E(A_{ij}A_{il})=\mathbf1\{j=l\}$.
Together, these identities give
$$
\mathbb E(A_{ij}A_{kl})
=\frac{\mathbf1\{i=k\}\mathbf1\{j=l\}}n.
$$
Equivalently,
$\mathbb E[\operatorname{vec}(\mathbf A)
\operatorname{vec}(\mathbf A)^\top]=\mathbf I_{nd}/n$.
The definitions of $f_1$ and $f_2$ then yield
$\mathbb E f_1=0$ and
$\mathbb E f_2=\operatorname{tr}(\mathbf m)/n$.

The cited theorem applies to symmetric coefficient matrices. Replacing $\mathbf m$ by $(\mathbf m+\mathbf m^\top)/2$ leaves the quadratic form and trace unchanged and does not increase either norm, so \eqref{eq:stiefel-quadratic-concentration} also holds for arbitrary $\mathbf m$. For $n=2$, \eqref{eq:stiefel-quadratic-concentration} is trivial because its right-hand side is two.
\begin{lemma}[Restated {\citep[Lemma 8.1]{birge2001alternative}}]\label{chi square concentration}
  Let \( X \) be a noncentral \(\chi^2\) variable with \( D \) degrees of freedom and noncentrality parameter \( B\). Then for all \( x > 0 \),
  \begin{equation}\label{eq: upper tail bound}
      \mathbb{P} \left( X \geq (D + B) + 2\sqrt{(D + 2B)x} + 2x \right) \leq \exp(-x),
  \end{equation}
  and
  \begin{equation}\label{eq: lower tail bound}
      \mathbb{P} \left( X \leq (D + B) - 2\sqrt{(D + 2B)x} \right) \leq \exp(-x).
  \end{equation}
\end{lemma}
\par\endgroup

\subsubsection{Residual model}
We next derive the residual model used in both the $\ell_2$
and $\ell_\infty$ power arguments.
Let $d=n-s$, $\mathbf Z=\mathbf P_{\mc R}\mathbf X_{\mc T}$,
$\bs R=\mathbf P_{\mc R}\bs Y$, and
$\mathbf V=\mathbf\Sigma_{\mc T\mid\mc S}$.

Define the original error vector
$\bs\varepsilon^{\mathrm{full}}
=\bs Y-\mathbf X_{\mc T}\bs\beta_{\mc T}
       -\mathbf X_{\mc S}\bs\beta_{\mc S}$.
By \eqref{linear model} and independence of the observations,
$\bs\varepsilon^{\mathrm{full}}\sim
\mathcal N_n(\bs0,\sigma^2\mathbf I_n)$ independently of $\mathbf X$.

Conditional on $\mathbf X_{\mc S}$, the matrix $\mathbf P_{\mc R}$
defined in Section~\ref{sec:mvn-CRT} is fixed and satisfies
$\mathbf P_{\mc R}\mathbf P_{\mc R}^{\top}=\mathbf I_d$.
Thus the projected error
$\bs\varepsilon=\mathbf P_{\mc R}\bs\varepsilon^{\mathrm{full}}$
is Gaussian with mean zero and covariance $\sigma^2\mathbf I_d$.
Under the same conditioning, $\bs\varepsilon^{\mathrm{full}}$
and $\mathbf X_{\mc T}$ are independent, so their fixed projections
$\bs\varepsilon$ and $\mathbf Z$ are independent.

Since $\mathbf P_{\mc R}\mathbf X_{\mc S}=\mathbf0$,
the projected response satisfies
$\bs R=\mathbf Z\bs\beta_{\mc T}+\bs\varepsilon$.
Together with the conditional residual distribution in
\eqref{eq:mvn-residual-conditional-law}, this gives, conditional on
$\mathbf X_{\mc S}$,
\begin{equation}\label{eq:power-residual-model}
\begin{gathered}
 \bs{Z}_{l\cdot}\stackrel{\mathrm{iid}}\sim
 \mathcal{N}_t(\bs{0},\mathbf{V})\quad(l\in[d]),\\
 \bs{R}=\mathbf{Z}\bs{\beta}_{\mc T}+\bs{\varepsilon},\\
 \bs{\varepsilon}\sim
 \mathcal{N}_d(\bs{0},\sigma^2\mathbf{I}_d),
 \qquad \mathbf{Z}\indp\bs{\varepsilon}.
\end{gathered}
\end{equation}
Conditional probability bounds whose right-hand sides do not depend
on $\mathbf X_{\mc S}$ also hold unconditionally by the tower property.

\subsubsection{Power analysis for \texorpdfstring{$\ell_2$}{ell2} separation}
\label{app:mvn-dense-power-proof}

We prove the two probability bounds in \eqref{eq:power-tail-targets} using separate lemmas for the copy statistic and the observed statistic. We work under the residual model \eqref{eq:power-residual-model}, with $\mathbf V\succ\mathbf0$ and $\sigma>0$. We first state the lemmas and deduce the conclusion for $\ell_2$ separation in Theorem~\ref{power analysis}; the proofs of the lemmas follow that deduction.

\begin{lemma}[Copy statistic for $\ell_2$ separation]\label{lem:dense-copy-tail}

Suppose $d\ge4$ and $1\le t<d$. For each $m\in[M]$, let $T_m$ be the statistic $T_{\ell_2}$ computed from the $m$-th copy generated by Algorithm~\ref{alg:multi gaussian CRT}. For $0<a<1$, define $C_0=128e^2/\log2$, $y=\log(2/a)$ and $K=2C_0y$, and let
\begin{equation}\label{eq:dense-copy-threshold}
 q_a=\frac td+K\frac{\sqrt t}{d}.
\end{equation}
Under \eqref{eq:power-residual-model}, $\mathbb P_\theta(T_m\ge q_a\mid(\bs Y,\mathbf X))\le a$ almost surely. Fix $0<c<1$ and let $c'=(1+c)/2$. Uniformly over $1\le t\le cd$, we have $q_a\le c'<1$ for all sufficiently large $d$ depending only on $a$ and $c$.

\end{lemma}

\begin{lemma}[Observed statistic for $\ell_2$ separation]\label{lem:dense-observed-tail}
Under \eqref{eq:power-residual-model}, suppose $1\le t<d$, let $r=d-t$, and define $\rho^2=\bs\beta_{\mc T}^{\top}\mathbf V\bs\beta_{\mc T}/\sigma^2>0$. Let $K\ge0$, $0<c'<1$ and $0<\delta<1$. Let $q_a$ be any deterministic threshold satisfying
\begin{equation}\label{eq:dense-observed-threshold-conditions}
 0\le q_a\le c'<1,\qquad q_a\le\frac td+K\frac{\sqrt t}{d}.
\end{equation}
Define $x=\log(6/\delta)$ and $H=K+(4+2K)\sqrt x+6x$. If
\begin{equation}\label{eq:dense-observed-signal-conditions}
 d\ge16x,\qquad d\rho^2\ge\max\left\{\frac{8H}{1-c'},4x\right\}\sqrt t,
\end{equation}
then $\mathbb P_\theta(T_{\ell_2}>q_a)\ge1-\delta/2$.
\end{lemma}

\begin{proof}[Proof of Theorem~\ref{power analysis}]

We prove the conclusion for $\ell_2$ separation.

Recall that $\delta=1-\beta$ and $a=\alpha\delta/4$. Let $c=c_2/(1-c_1)<1$ and $c'=(1+c)/2$. The dimension assumptions give $d=n-s\ge(1-c_1)n$ and $t\le cd$. Define $y=\log(2/a)$, $C_0=128e^2/\log2$, $K=2C_0y$, $x=\log(6/\delta)$ and $H=K+(4+2K)\sqrt x+6x$ as in the two lemmas. For all sufficiently large $n$, Lemma~\ref{lem:dense-copy-tail} gives a deterministic $q_a$ satisfying \eqref{eq:dense-observed-threshold-conditions} and $\mathbb P_\theta(T_m\ge q_a\mid(\bs Y,\mathbf X))\le a$. Taking expectations gives $\mathbb P_\theta(T_m\ge q_a)\le a$ for each $m\in[M]$. To apply Lemma~\ref{lem:dense-observed-tail}, it remains to verify \eqref{eq:dense-observed-signal-conditions} uniformly over the class defined by $\ell_2$ separation.

By the definition of $\Theta$ in Section~\ref{sec: MVN-CRT power}, the eigenvalues of $\mathbf{\Sigma}$ lie in $[M_0^{-1},M_0]$ and $\sigma\le M_1$. Since $\mathbf{V}^{-1}=(\mathbf{\Sigma}^{-1})_{\mc T\mc T}$, the eigenvalue bounds for principal submatrices imply that the eigenvalues of $\mathbf{V}$ also lie in $[M_0^{-1},M_0]$.
Recall that
$\rho^2=\bs\beta_{\mc T}^{\top}\mathbf V\bs\beta_{\mc T}/\sigma^2$.
The spectral theorem and $\lambda_{\max}(\mathbf V)\le M_0$ give
$\|\mathbf V\bs\beta_{\mc T}\|_2^2
\le M_0\bs\beta_{\mc T}^{\top}\mathbf V\bs\beta_{\mc T}$.
Since $\sigma\le M_1$, we obtain
$\rho^2\ge
\|\mathbf V\bs\beta_{\mc T}\|_2^2/(M_0M_1^2)$.
By the definition of $\Theta_2(h)$ in Section~\ref{sec: MVN-CRT power}, membership in $\Theta_2(C_1t^{1/4}/\sqrt n)$ gives $\|\mathbf{V}\bs{\beta}_{\mc T}\|_2\ge C_1t^{1/4}/\sqrt n$. Together with $d=n-s\ge(1-c_1)n$ from Theorem~\ref{power analysis}, this yields
\begin{equation}\label{eq:dense-signal-scale}
 d\rho^2
 \ge\frac{d\|\mathbf{V}\bs{\beta}_{\mc T}\|_2^2}{M_0M_1^2}
 \ge\frac{d C_1^2\sqrt t}{nM_0M_1^2}
 \ge\frac{(1-c_1)C_1^2}{M_0M_1^2}\sqrt t.
\end{equation}

Choose $C_1$ so that $(1-c_1)C_1^2/(M_0M_1^2)\ge\max\{8H/(1-c'),4x\}$. Equation \eqref{eq:dense-signal-scale} then ensures $d\rho^2\ge\max\{8H/(1-c'),4x\}\sqrt t$ uniformly over $\Theta_2(C_1t^{1/4}/\sqrt n)$. Moreover, $d=n-s\ge(1-c_1)n\ge16x$ for all sufficiently large $n$, since $x=\log(6/\delta)$ is fixed. Thus \eqref{eq:dense-observed-signal-conditions} holds uniformly over the alternative class. Lemma~\ref{lem:dense-observed-tail} gives $\mathbb P_\theta(T_0>q_a)\ge1-\delta/2$. Together with $\mathbb P_\theta(T_m\ge q_a)\le a$ from Lemma~\ref{lem:dense-copy-tail}, this verifies both bounds in \eqref{eq:power-tail-targets}. The choices of $C_1$ and the lower bound on $n$ depend only on the fixed parameters allowed in Theorem~\ref{power analysis}. Lemma~\ref{lem: null tail bound} then gives the asserted power.
\end{proof}

\begin{proof}[Proof of Lemma~\ref{lem:dense-copy-tail}]

We first express each copy statistic as a quadratic form. Since $d>t$, $\mathbf{V}\succ\mathbf{0}$ and $\sigma>0$, equation \eqref{eq:power-residual-model} implies that $\mathbf{Z}$ has rank $t$ and $\bs{R}\ne\bs{0}$ almost surely. We condition on observed data $(\bs{Y},\mathbf X)$ for which these properties hold. The matrix $\mathbf Q=(\mathbf{Z}^\top \mathbf{Z})^{1/2}$ in \eqref{construction of Q} is then fixed and invertible.
\par

Fix $m\in[M]$ and write $\mathbf{U}=\mathbf{U}^{(m)}$. The Gaussian draws in Step 1 of Algorithm~\ref{alg:multi gaussian CRT} are independent of the observed data. Applying Gram--Schmidt orthogonalisation to these Gaussian draws, as described in Appendix~\ref{app:haar-preliminaries-draft}, therefore gives $\mathbf{U}\mid(\bs{Y},\mathbf X)\sim\mu_{d,t}$. By Step 2 of the algorithm, the copy residual is $\widetilde{\mathbf{Z}}=\mathbf P_{\mc R}\widetilde{\mathbf X}_{\mc T}^{(m)}=\mathbf{U}\mathbf Q$. Since $\mathbf Q$ is invertible, $\operatorname{col}(\widetilde{\mathbf{Z}})=\operatorname{col}(\mathbf{U})$, whose distribution is invariant under every fixed orthogonal transformation by Lemma~\ref{lem:haar-left-invariance-draft}. Substituting $\widetilde{\mathbf{Z}}$ into the statistic $T_{\ell_2}$ gives
\begin{equation}\label{eq:dense-copy-projection}
\begin{aligned}
 T_m
 &=\frac{\bs{R}^\top\widetilde{\mathbf{Z}}(\widetilde{\mathbf{Z}}^\top\widetilde{\mathbf{Z}})^{-1}\widetilde{\mathbf{Z}}^\top \bs{R}}{\|\bs{R}\|_2^2}\\
 &=\frac{\bs{R}^\top \mathbf{U}\mathbf Q(\mathbf Q^\top\mathbf Q)^{-1}\mathbf Q^\top \mathbf{U}^\top \bs{R}}{\|\bs{R}\|_2^2}
 =\frac{\|\mathbf{U}^\top \bs{R}\|_2^2}{\|\bs{R}\|_2^2},
\end{aligned}
\end{equation}
where the second equality in \eqref{eq:dense-copy-projection} uses $\mathbf{U}^\top \mathbf{U}=\mathbf{I}_t$, and the last uses $\mathbf Q(\mathbf Q^\top\mathbf Q)^{-1}\mathbf Q^\top=\mathbf{I}_t$.

Let $\bs v=\bs R/\|\bs R\|_2$ and $\mathbf m=\mathbf I_t\otimes\bs v\bs v^\top$. Equation \eqref{eq:dense-copy-projection} gives $T_m=\operatorname{vec}(\mathbf U)^\top\mathbf m\operatorname{vec}(\mathbf U)$. The eigenvalues of $\mathbf m$ are $1$ with multiplicity $t$ and $0$ otherwise, so $\operatorname{tr}(\mathbf m)=\|\mathbf m\|_{\mathrm{HS}}^2=t$ and $\|\mathbf m\|_{\mathrm{op}}=1$. Conditional on $(\bs Y,\mathbf X)$, Lemma~\ref{second order} therefore gives, for every $u>0$,
\begin{equation}\label{eq:dense-copy-quadratic-tail}
\begin{aligned}
 &\mathbb P_\theta\!\left(T_m\ge\frac td+u\mid(\bs Y,\mathbf X)\right)\\
 \le{}&2\exp\left(-\frac1{C_0}\min\left\{\frac{(d-2)^2u^2}{t},(d-2)u\right\}\right).
\end{aligned}
\end{equation}
Take $u=K\sqrt t/d=2C_0y\sqrt t/d$. Since $d\ge4$ and $t\ge1$, we have $(d-2)u/C_0\ge y\sqrt t\ge y$ and $(d-2)^2u^2/(C_0t)\ge C_0y^2\ge y$, where the last inequality uses $C_0y\ge128e^2>1$. Equations \eqref{eq:dense-copy-threshold} and \eqref{eq:dense-copy-quadratic-tail} imply $\mathbb P_\theta(T_m\ge q_a\mid(\bs Y,\mathbf X))\le2e^{-y}=a$.

Finally, if $t\le cd$, equation \eqref{eq:dense-copy-threshold} gives $q_a\le c+K\sqrt{c/d}\le c+(1-c)/2=c'<1$ whenever $d\ge4K^2c/(1-c)^2$. This lower bound on $d$, together with $d\ge4$, depends only on $a$ and $c$.

\end{proof}

\begin{proof}[Proof of Lemma~\ref{lem:dense-observed-tail}]
We show that the observed statistic exceeds $q_a$ with probability at least $1-\delta/2$. Recall the residual model \eqref{eq:power-residual-model}, where $\mathbf{Z}=\mathbf P_{\mc R}\mathbf X_{\mc T}$ is the residual covariate matrix, $\bs{R}=\mathbf P_{\mc R}\bs{Y}$ is the residual response vector, and $\mathbf{V}=\mathbf{\Sigma}_{\mc T\mid\mc S}$. Let $\mathbf{P}_{\mathbf{Z}}=\mathbf{Z}(\mathbf{Z}^\top \mathbf{Z})^{-1}\mathbf{Z}^\top$, and define $A=\bs{R}^\top \mathbf{P}_{\mathbf{Z}}\bs{R}/\sigma^2$, $D=\bs{R}^\top(\mathbf{I}_d-\mathbf{P}_{\mathbf{Z}})\bs{R}/\sigma^2$, and $\Lambda=\|\mathbf{Z}\bs{\beta}_{\mc T}\|^2/\sigma^2$. 
The observed statistic is $T_{\ell_2}=A/(A+D)$.
We work conditional on $(\mathbf Z,\mathbf X_{\mc S})$.
By \eqref{eq:power-residual-model},
$\bs\varepsilon/\sigma\sim\mathcal N_d(\bs0,\mathbf I_d)$
under this conditioning.

We first find the conditional distributions of $A$ and $D$. 
Since $\mathbf Z$ has rank $t$ almost surely, choose matrices
$\mathbf O_1\in\mathbb R^{d\times t}$ and
$\mathbf O_2\in\mathbb R^{d\times r}$ whose columns form
orthonormal bases of $\operatorname{col}(\mathbf Z)$
and its orthogonal complement, respectively.
Define $\bs\varepsilon_h=\mathbf O_h^\top\bs\varepsilon/\sigma$
for $h=1,2$.
These vectors are jointly Gaussian with mean zero,
covariance matrices $\mathbf I_t$ and $\mathbf I_r$, and
cross-covariance $\mathbf O_1^\top\mathbf O_2=\mathbf0$.
They are therefore independent.

Since $\mathbf Z\bs\beta_{\mc T}$ lies in
$\operatorname{col}(\mathbf Z)$, the quadratic forms become
$$
\begin{aligned}
A&=\left\|
\frac{\mathbf O_1^\top\mathbf Z\bs\beta_{\mc T}}{\sigma}
+\bs\varepsilon_1\right\|_2^2,\\
D&=\|\bs\varepsilon_2\|_2^2.
\end{aligned}
$$
The squared norm of the mean in the first expression is
$\|\mathbf O_1^\top\mathbf Z\bs\beta_{\mc T}/\sigma\|_2^2
=\Lambda$.
Thus $A$ and $D$ are conditionally independent with
$A\sim\chi_t^2(\Lambda)$ and $D\sim\chi_r^2$.

We now condition only on $\mathbf X_{\mc S}$ and define $\rho^2=\bs{\beta}_{\mc T}^\top \mathbf{V}\bs{\beta}_{\mc T}/\sigma^2>0$, since $\bs{\beta}_{\mc T}\ne\bs{0}$ under the alternatives and $\mathbf{V}\succ\mathbf{0}$. By \eqref{eq:power-residual-model}, the components $(\mathbf{Z}\bs{\beta}_{\mc T})_l$, $l\in[d]$, are then independent $\mathcal{N}(0,\sigma^2\rho^2)$ variables. Dividing them by $\sigma\rho$ gives independent standard normal variables, so
\begin{equation}\label{eq:dense-noncentrality-distribution}
 \frac{\Lambda}{\rho^2}
 =\frac{\|\mathbf{Z}\bs{\beta}_{\mc T}\|_2^2}{\sigma^2\rho^2}
 =\sum_{l=1}^d\left(\frac{(\mathbf{Z}\bs{\beta}_{\mc T})_l}{\sigma\rho}\right)^2
 \sim\chi_d^2
 \quad\text{conditional on }\mathbf X_{\mc S}.
\end{equation}

For $x=\log(6/\delta)$, define $\lambda_0=(d-2\sqrt{dx})\rho^2$ and $g(z)=t+z-2\sqrt{(t+2z)x}$. We use Lemma~\ref{chi square concentration} three times. First, \eqref{eq:dense-noncentrality-distribution} gives the central $\chi_d^2$ law of $\Lambda/\rho^2$ conditional on $\mathbf X_{\mc S}$. Its lower tail bound \eqref{eq: lower tail bound}, with $d$ degrees of freedom and noncentrality zero, gives
\begin{equation}\label{eq:dense-noncentrality-tail}
\begin{aligned}
 \mathbb P_\theta(\Lambda\le\lambda_0\mid\mathbf X_{\mc S})
 &=\mathbb P_\theta(\Lambda/\rho^2\le d-2\sqrt{dx}\mid\mathbf X_{\mc S})\\
 &\le e^{-x}.
\end{aligned}
\end{equation}
Second, conditional on $(\mathbf{Z},\mathbf X_{\mc S})$, the noncentrality parameter $\Lambda$ is fixed and $A\sim\chi_t^2(\Lambda)$. Using \eqref{eq: lower tail bound} with $t$ degrees of freedom and noncentrality $\Lambda$ yields
\begin{equation}\label{eq:dense-observed-numerator-tail}
 \mathbb P_\theta\!\left(A\le t+\Lambda-2\sqrt{(t+2\Lambda)x}
 \mid (\mathbf{Z},\mathbf X_{\mc S})\right)
 =\mathbb P_\theta(A\le g(\Lambda)\mid (\mathbf{Z},\mathbf X_{\mc S}))
 \le e^{-x}.
\end{equation}
Third, conditional on $(\mathbf{Z},\mathbf X_{\mc S})$, $D\sim\chi_r^2$. Its upper tail bound \eqref{eq: upper tail bound}, with $r$ degrees of freedom and noncentrality zero, gives
\begin{equation}\label{eq:dense-observed-remainder-tail}
 \mathbb P_\theta(D\ge r+2\sqrt{rx}+2x\mid (\mathbf{Z},\mathbf X_{\mc S}))\le e^{-x}.
\end{equation}
We now remove the conditioning. Taking expectations in \eqref{eq:dense-noncentrality-tail}--\eqref{eq:dense-observed-remainder-tail} gives the three unconditional bounds
\begin{equation}\label{eq:dense-unconditional-tail-bounds}
\begin{aligned}
 \mathbb P_\theta(\Lambda<\lambda_0)
 &=\mathbb E_\theta\!\left[\mathbb P_\theta(\Lambda<\lambda_0\mid\mathbf X_{\mc S})\right]\le e^{-x},\\
 \mathbb P_\theta(A<g(\Lambda))
 &=\mathbb E_\theta\!\left[\mathbb P_\theta(A<g(\Lambda)\mid (\mathbf{Z},\mathbf X_{\mc S}))\right]\le e^{-x},\\
 \mathbb P_\theta(D>r+2\sqrt{rx}+2x)
 &=\mathbb E_\theta\!\left[\mathbb P_\theta(D>r+2\sqrt{rx}+2x\mid (\mathbf{Z},\mathbf X_{\mc S}))\right]\le e^{-x}.
\end{aligned}
\end{equation}
Finally, the union bound applied to \eqref{eq:dense-unconditional-tail-bounds} gives
\begin{equation}\label{eq: strong signal}
\begin{aligned}
 &\mathbb P_\theta(\Lambda\ge\lambda_0,\ A\ge g(\Lambda),\ D\le r+2\sqrt{rx}+2x)\\
 \ge{}&1-\mathbb P_\theta(\Lambda<\lambda_0)
             -\mathbb P_\theta(A<g(\Lambda))
             -\mathbb P_\theta(D>r+2\sqrt{rx}+2x)\\
 \ge{}&1-3e^{-x}=1-\delta/2.
\end{aligned}
\end{equation}
Let $H=K+(4+2K)\sqrt x+6x$. The conditions on $d$ and $d\rho^2$ in \eqref{eq:dense-observed-signal-conditions} are hypotheses of the lemma, and \eqref{eq:dense-observed-threshold-conditions} gives $0\le q_a\le c'<1$. On the event in \eqref{eq: strong signal}, we will show that
\begin{equation}\label{eq:dense-observed-margin}
\begin{aligned}
 (1-q_a)A-q_aD
 &\stackrel{\mathrm{(i)}}{\ge}(1-q_a)g(\lambda_0)-q_a(r+2\sqrt{rx}+2x)\\
 &\stackrel{\mathrm{(ii)}}{\ge}(1-c')\lambda_0/2-H\sqrt t\\
 &\stackrel{\mathrm{(iii)}}{\ge}H\sqrt t>0.
\end{aligned}
\end{equation}
Equation \eqref{eq:dense-observed-margin} implies $T_{\ell_2}=A/(A+D)>q_a$; its three steps are justified as follows:
\begin{enumerate}
\item For (i), the assumed conditions give $\lambda_0=(d-2\sqrt{dx})\rho^2\ge d\rho^2/2\ge2x\sqrt t\ge2x$. Hence $g'(z)=1-2\sqrt{x}/\sqrt{t+2z}\ge1-2\sqrt{x}/\sqrt{4x}=0$ for $z\ge\lambda_0\ge2x$. On the event in \eqref{eq: strong signal}, we therefore have $A\ge g(\Lambda)\ge g(\lambda_0)$ and $D\le r+2\sqrt{rx}+2x$. Multiplying these bounds by $1-q_a\ge0$ and $-q_a\le0$, respectively, gives (i).

\item For (ii), using $\sqrt{(t+2\lambda_0)x}\le\sqrt{tx}+\sqrt{2\lambda_0x}$ and then $2\sqrt{2\lambda_0x}\le\lambda_0/2+4x$ gives $g(\lambda_0)\ge t+\lambda_0/2-2\sqrt{tx}-4x$. The threshold bound in \eqref{eq:dense-observed-threshold-conditions} gives $t-q_ad\ge-K\sqrt t$ and $q_a\sqrt r\le(t+K\sqrt t)\sqrt r/d\le(1+K)\sqrt t$, since $r\le d$, $t\le d$, and $d\ge1$. Substitution yields
\begin{equation}\label{eq:dense-observed-margin-expansion}
\begin{aligned}
 &(1-q_a)g(\lambda_0)-q_a(r+2\sqrt{rx}+2x)\\
 \ge{}&(1-q_a)(t+\lambda_0/2-2\sqrt{tx}-4x)
              -q_a(r+2\sqrt{rx}+2x)\\
 ={}&(t-q_ad)+(1-q_a)\lambda_0/2
       -(1-q_a)(2\sqrt{tx}+4x)-2q_a\sqrt{rx}-2q_ax\\
 \ge{}&-K\sqrt t+(1-c')\lambda_0/2
       -2\sqrt{tx}-4x-2(1+K)\sqrt{tx}-2x\\
 ={}&(1-c')\lambda_0/2-\{K+(4+2K)\sqrt x\}\sqrt t-6x\\
 \ge{}&(1-c')\lambda_0/2-H\sqrt t,
\end{aligned}
\end{equation}
where the first inequality in \eqref{eq:dense-observed-margin-expansion} substitutes the bound on $g(\lambda_0)$, and the first equality uses $t+r=d$. The second inequality applies $t-q_ad\ge-K\sqrt t$ and $q_a\sqrt r\le(1+K)\sqrt t$, derived from \eqref{eq:dense-observed-threshold-conditions}, together with $0\le q_a\le c'<1$, term by term. The last inequality uses $6x\le6x\sqrt t$ and the definition of $H$.

\item For (iii), $\lambda_0\ge d\rho^2/2$ and the condition $d\rho^2\ge\max\{8H/(1-c'),4x\}\sqrt t$ give $(1-c')\lambda_0/2\ge(1-c')d\rho^2/4\ge2H\sqrt t$. Subtracting $H\sqrt t$ proves (iii).
\end{enumerate}

Equations \eqref{eq: strong signal} and \eqref{eq:dense-observed-margin} therefore give $\mathbb P_\theta(T_{\ell_2}>q_a)\ge1-\delta/2$, as claimed.
\end{proof}
\subsubsection{Power analysis for \texorpdfstring{$\ell_\infty$}{ell-infinity} separation}

We use the residual coordinates $\mathbf Z=\mathbf P_{\mc R}\mathbf X_{\mc T}$ and $\bs R=\mathbf P_{\mc R}\bs Y$, with $d=n-s$ and $\mathbf V=\mathbf\Sigma_{\mc T\mid\mc S}$, from \eqref{eq:power-residual-model}. Throughout this subsection, $\mathbf V\succ\mathbf0$ and $\sigma>0$. We first state two lemmas for the probability bounds in \eqref{eq:power-tail-targets}. The first supplies a deterministic threshold with a small upper-tail probability for each copy statistic. The second gives a sufficient condition on the population correlation of a single coordinate for the observed statistic to exceed that threshold with high probability. We then combine the two lemmas to prove the conclusion for $\ell_\infty$ separation in Theorem~\ref{power analysis}, before giving their detailed proofs.

\begin{lemma}[Copy statistic for $\ell_\infty$ separation]\label{lem:sparse-copy-tail}
Let $t\ge1$ and $0<a<1$, and suppose $d-1>8\log(2t/a)$. For each $m\in[M]$, let $T_m$ be the statistic $T_{\ell_\infty}$ computed from the $m$-th copy generated by Algorithm~\ref{alg:multi gaussian CRT}. Define
\begin{equation}\label{eq:sparse-copy-threshold}
 q_a=\sqrt{\frac{8\log(2t/a)}{d-1}}\in(0,1).
\end{equation}
Under \eqref{eq:power-residual-model}, we have $\mathbb P_\theta(T_m\ge q_a\mid(\bs Y,\mathbf X))\le a$ almost surely, and hence $\mathbb P_\theta(T_m\ge q_a)\le a$.
\end{lemma}

\begin{lemma}[Observed statistic for $\ell_\infty$ separation]\label{lem:sparse-observed-tail}
Under \eqref{eq:power-residual-model}, let $d,t\ge1$, $0<\delta<1$, and let $q_a\in[0,1)$ be a deterministic threshold. Define $v_j=V_{jj}$, $\nu_j=(\mathbf V\bs\beta_{\mc T})_j$ for $j\in\mc T$, and $w=\sigma^2+\bs\beta_{\mc T}^\top\mathbf V\bs\beta_{\mc T}$. Suppose there is a fixed coordinate $j\in\mc T$ with $\nu_j\ne0$, and define $\gamma=|\nu_j|/\sqrt{v_jw}$. If
\begin{equation}\label{eq:sparse-correlation-condition}
 d\ge16\log(8/\delta),\qquad
 \gamma>4q_a,\qquad
 d\gamma^2\ge8(4/\delta-1),
\end{equation}
then $\mathbb P_\theta(T_{\ell_\infty}>q_a\mid\mathbf X_{\mc S})\ge1-\delta/2$ almost surely, and hence $\mathbb P_\theta(T_{\ell_\infty}>q_a)\ge1-\delta/2$.
\end{lemma}

\begin{proof}[Proof of Theorem~\ref{power analysis}]
We prove the conclusion for $\ell_\infty$ separation by combining Lemmas~\ref{lem:sparse-copy-tail} and~\ref{lem:sparse-observed-tail} to verify the two probability bounds in \eqref{eq:power-tail-targets}. Let $T_0=T_{\ell_\infty}$ be the observed statistic and $T_m$ the statistic from copy $m\in[M]$. Recall that $\delta=1-\beta$ and $a=\alpha\delta/4$.

We first choose the threshold using Lemma~\ref{lem:sparse-copy-tail}. Since $s\le c_1n$, we have $d=n-s\ge(1-c_1)n$. The assumption $\log t=o(n)$ and the fixed choice of $a$ give $\log(2t/a)=\log t+\log(2/a)=o(n)$. Hence $d-1>8\log(2t/a)$ for all sufficiently large $n$, so Lemma~\ref{lem:sparse-copy-tail} applies and gives $q_a\in(0,1)$ from \eqref{eq:sparse-copy-threshold} with $\mathbb P_\theta(T_m\ge q_a)\le a$ for every $m\in[M]$.

To compare this threshold with the signal strength, define $K_q=4\sqrt{\{1+\log(2/a)\}/(1-c_1)}$. For $n\ge2/(1-c_1)$,
\begin{equation}\label{eq:sparse-threshold-scale}
\begin{aligned}
 q_a&=\sqrt{\frac{8\log(2t/a)}{d-1}}\\
 &\le4\sqrt{\frac{\log(2t/a)}{(1-c_1)n}}\\
 &\le K_q\sqrt{\frac{\log(et)}{n}},
\end{aligned}
\end{equation}
where the first inequality in \eqref{eq:sparse-threshold-scale} uses $d-1\ge(1-c_1)n/2$, and the second uses $\log(2t/a)\le\{1+\log(2/a)\}\log(et)$ for $t\ge1$.

We next verify the conditions of Lemma~\ref{lem:sparse-observed-tail}. Fix $\theta\in\Theta_\infty(C_2\sqrt{\log(et)/n})$, with $C_2>0$ to be specified below. Choose $j\in\mc T$ so that $|\nu_j|=\|\mathbf V\bs\beta_{\mc T}\|_\infty$. The definition of the alternative class gives $|\nu_j|\ge C_2\sqrt{\log(et)/n}>0$. This coordinate depends only on $\theta$, so it is deterministic for each fixed parameter value. Use $v_j,w,\gamma$ as defined in Lemma~\ref{lem:sparse-observed-tail}.

The parameter space in Section~\ref{sec: MVN-CRT power} gives $\sigma\le M_1$ and $\|\bs\beta_{\mc T}\|_2^2\le M_2$. Since $\mathbf V^{-1}=(\mathbf\Sigma^{-1})_{\mc T\mc T}$ is a principal submatrix, $\lambda_{\min}(\mathbf V^{-1})\ge\lambda_{\min}(\mathbf\Sigma^{-1})\ge M_0^{-1}$, hence $\lambda_{\max}(\mathbf V)\le M_0$. Consequently,
\begin{equation}\label{eq:sparse-variance-bounds}
\begin{aligned}
 v_j&\le\lambda_{\max}(\mathbf V)\le M_0,\\
 w&=\sigma^2+\bs\beta_{\mc T}^\top\mathbf V\bs\beta_{\mc T}\\
 &\le\sigma^2+\lambda_{\max}(\mathbf V)\|\bs\beta_{\mc T}\|_2^2\\
 &\le M_1^2+M_0M_2.
\end{aligned}
\end{equation}
Choose $C_2$ so that $C_2>4K_q\sqrt{M_0(M_1^2+M_0M_2)}$ and $(1-c_1)C_2^2/\{M_0(M_1^2+M_0M_2)\}\ge8(4/\delta-1)$. These choices depend only on $\alpha,\beta,M_0,M_1,M_2,c_1$.

The definition $\gamma=|\nu_j|/\sqrt{v_jw}$ gives
\begin{equation}\label{eq:sparse-signal-scale}
\begin{aligned}
 \gamma
 &\ge\frac{C_2}{\sqrt{M_0(M_1^2+M_0M_2)}}\sqrt{\frac{\log(et)}{n}}\\
 &>4K_q\sqrt{\frac{\log(et)}{n}}\\
 &\ge4q_a,
\end{aligned}
\end{equation}
where the first bound on $\gamma$ in \eqref{eq:sparse-signal-scale} uses the lower bound on $|\nu_j|$ and \eqref{eq:sparse-variance-bounds}, the strict inequality uses the first choice of $C_2$, and the comparison with $q_a$ uses \eqref{eq:sparse-threshold-scale}.

We have
\begin{equation}\label{eq:sparse-correlation-scale}
\begin{aligned}
 d\gamma^2
 &\ge\frac{(1-c_1)C_2^2\log(et)}{M_0(M_1^2+M_0M_2)}\\
 &\ge8(4/\delta-1),
\end{aligned}
\end{equation}
where the first inequality in \eqref{eq:sparse-correlation-scale} uses the first lower bound on $\gamma$ in \eqref{eq:sparse-signal-scale} and $d\ge(1-c_1)n$, and the second uses $\log(et)\ge1$ and the second choice of $C_2$.

Finally, $d\ge16\log(8/\delta)$ for all sufficiently large $n$. Thus all three conditions in \eqref{eq:sparse-correlation-condition} hold, and Lemma~\ref{lem:sparse-observed-tail} gives $\mathbb P_\theta(T_0>q_a)\ge1-\delta/2$.

It remains to translate these bounds into power. Let $p_M$ and $k=\lfloor\alpha(M+1)\rfloor$ be as in Lemma~\ref{lem: null tail bound}. Since $M\ge2/\alpha$, we have $k\ge\alpha M-1\ge\alpha M/2\ge1$. Equation~\eqref{eq:finite-M-power-reduction} yields
\begin{equation}\label{eq:sparse-power-conclusion}
\begin{aligned}
 \mathbb E_\theta\phi_{T_{\ell_\infty},\alpha}
 &=\mathbb P_\theta(p_M\le\alpha)\\
 &\ge\mathbb P_\theta(T_0>q_a)-\frac1k\sum_{m=1}^M\mathbb P_\theta(T_m\ge q_a)\\
 &\ge1-\frac\delta2-\frac Mk a\\
 &\ge1-\frac\delta2-\frac{2a}{\alpha}=1-\delta=\beta,
\end{aligned}
\end{equation}
where the second inequality in \eqref{eq:sparse-power-conclusion} uses the bounds from Lemmas~\ref{lem:sparse-copy-tail} and~\ref{lem:sparse-observed-tail}, and the last uses $M/k\le2/\alpha$ and $a=\alpha\delta/4$. All parameter bounds are uniform over the alternative class. Along every dimension sequence satisfying the theorem's assumptions, the required lower bound on $n$ is independent of $\theta$. Hence the power bound holds uniformly for all sufficiently large $n$, which proves the stated lower limit.
\end{proof}

\begin{proof}[Proof of Lemma~\ref{lem:sparse-copy-tail}]
Equation \eqref{eq:power-residual-model} gives $\mathbf Z_j\ne\bs0$ for every $j\in\mc T$ and $\bs R\ne\bs0$ almost surely, since $V_{jj}>0$ and $\sigma>0$. Condition on observed data $(\bs Y,\mathbf X)$ with these properties, fix $m\in[M]$, and write $\mathbf U=\mathbf U^{(m)}$. Applying Gram--Schmidt orthogonalisation to the independent standard Gaussian draws in Algorithm~\ref{alg:multi gaussian CRT} gives $\mathbf U\mid(\bs Y,\mathbf X)\sim\mu_{d,\min\{d,t\}}$; see Appendix~\ref{app:haar-preliminaries-draft}. Both definitions of $\mathbf Q$ in \eqref{construction of Q} satisfy $\mathbf Q^\top\mathbf Q=\mathbf Z^\top\mathbf Z$. Thus the copy residual $\mathbf P_{\mc R}\widetilde{\mathbf X}_j^{(m)}=\mathbf U\mathbf Q_j$ has fixed length $\|\mathbf Q_j\|_2=\|\mathbf Z_j\|_2>0$.

Define $\bs v=\bs R/\|\bs R\|_2$ and, for each $j\in\mc T$,
$\bs w_j=\mathbf Q_j/\|\mathbf Q_j\|_2$.
These unit vectors are fixed under the conditioning.
Since $\mathbf U$ has orthonormal columns,
$\|\mathbf U\bs w_j\|_2^2
=\bs w_j^\top\mathbf U^\top\mathbf U\bs w_j=1$.
Moreover, for every $\mathbf O\in O(d)$, left invariance of
the conditional law of $\mathbf U$ gives
$\mathbf O(\mathbf U\bs w_j)
=(\mathbf O\mathbf U)\bs w_j
\overset d=\mathbf U\bs w_j$.
Lemma~\ref{lem:haar-left-invariance-draft}, applied to $W_{d,1}$,
therefore shows that $\mathbf U\bs w_j$ has conditional law
$\mu_{d,1}$, the uniform distribution on the unit sphere in
$\bbR^d$.
This conclusion also holds when $t\ge d$, for which
$\mathbf U$ is square.

Since $d>1$, Lemma~\ref{first order} applies to the
$d\times1$ vector $\mathbf U\bs w_j$.
For the linear function $\bs v^\top(\mathbf U\bs w_j)$,
the coefficient vector has norm $\|\bs v\|_2=1$.
The lemma therefore gives, for $q\ge0$,

$$
 \mathbb P_\theta\!\left(
 \frac{|\bs R^\top\mathbf U\mathbf Q_j|}{\|\bs R\|_2\|\mathbf Q_j\|_2}\ge q
 \,\middle|\,(\bs Y,\mathbf X)\right)
 \le2\exp\!\left(-\frac{(d-1)q^2}{8}\right).
$$
The union bound over the $t$ columns, with $q_a$ from \eqref{eq:sparse-copy-threshold}, gives
\begin{equation}\label{eq:sparse-copy-tail}
 \mathbb P_\theta(T_m\ge q_a\mid(\bs Y,\mathbf X))
 \le2t\exp\!\left(-\frac{(d-1)q_a^2}{8}\right)=a.
\end{equation}
We then take expectations over $(\bs Y,\mathbf X)$ in \eqref{eq:sparse-copy-tail} to obtain $\mathbb P_\theta(T_m\ge q_a)\le a$ for every $m\in[M]$.
\end{proof}

\begin{proof}[Proof of Lemma~\ref{lem:sparse-observed-tail}]
Fix the coordinate $j$ in the statement and define $\eta_j=\operatorname{sign}(\nu_j)$ and $\tau_j^2=w-\nu_j^2/v_j$. Let $Z_{lj}$ and $R_l$ be the $l$-th components of $\mathbf Z_j$ and $\bs R$, respectively, for $l\in[d]$. By \eqref{eq:power-residual-model}, conditional on $\mathbf X_{\mc S}$, the pairs $(Z_{lj},R_l)$, $l\in[d]$, are independent centred Gaussian pairs with variances $v_j,w$ and covariance $\nu_j$. Conditional on $\mathbf X_{\mc S}$, the Cauchy--Schwarz inequality applied to $Z_{lj}$ and $(\mathbf Z\bs\beta_{\mc T})_l$ gives
$$
\begin{aligned}
 \nu_j^2
 &=\operatorname{Cov}\!\left(Z_{lj},(\mathbf Z\bs\beta_{\mc T})_l\mid\mathbf X_{\mc S}\right)^2\\
 &\le\operatorname{Var}(Z_{lj}\mid\mathbf X_{\mc S})\,
 \operatorname{Var}\!\left((\mathbf Z\bs\beta_{\mc T})_l\mid\mathbf X_{\mc S}\right)
 =v_j\bs\beta_{\mc T}^\top\mathbf V\bs\beta_{\mc T}.
\end{aligned}
$$
Since $v_j>0$, substitution into the definition of $\tau_j^2$ yields
$$
 \tau_j^2
 =\sigma^2+\bs\beta_{\mc T}^\top\mathbf V\bs\beta_{\mc T}-\frac{\nu_j^2}{v_j}
 \ge\sigma^2>0.
$$ For each $l\in[d]$, the conditional distribution formula for the bivariate Gaussian pair $(Z_{lj},R_l)$ gives
$$
 R_l\mid(Z_{lj},\mathbf X_{\mc S})
 \sim \mathcal{N}\!\left(\frac{\nu_j}{v_j}Z_{lj},\,w-\frac{\nu_j^2}{v_j}\right).
$$
The pairs $(Z_{lj},R_l)$ are independent across $l$ conditional on $\mathbf X_{\mc S}$. Hence, conditional on the entire column $\mathbf Z_j$ and $\mathbf X_{\mc S}$, the components $R_l$ remain independent, and the conditional distribution of each $R_l$ depends on $\mathbf Z_j$ only through $Z_{lj}$. Their conditional means form the vector $(\nu_j/v_j)\mathbf Z_j$, and their conditional covariance matrix is $\tau_j^2\mathbf I_d$. Thus
\begin{equation}\label{eq:sparse-conditional-regression}
 \bs R\mid(\mathbf Z_j,\mathbf X_{\mc S})
 \sim \mathcal{N}_d\!\left(\frac{\nu_j}{v_j}\mathbf Z_j,\tau_j^2\mathbf I_d\right).
\end{equation}

We first control the residual norms. Conditional on $\mathbf X_{\mc S}$, both $\|\mathbf Z_j\|_2^2/v_j$ and $\|\bs R\|_2^2/w$ have a $\chi_d^2$ distribution. Let $x=\log(8/\delta)$. The lower and upper bounds in Lemma~\ref{chi square concentration}, respectively, give
\begin{equation}\label{eq:sparse-residual-norm-tails}
\begin{aligned}
 &\mathbb P_\theta\!\left(\|\mathbf Z_j\|_2^2<v_jd/2\,\middle|\,\mathbf X_{\mc S}\right)\\
 ={}&\mathbb P_\theta\!\left(\frac{\|\mathbf Z_j\|_2^2}{v_j}<\frac d2\,\middle|\,\mathbf X_{\mc S}\right)
 \le\mathbb P_\theta\!\left(\frac{\|\mathbf Z_j\|_2^2}{v_j}<d-2\sqrt{dx}\,\middle|\,\mathbf X_{\mc S}\right)
 \le e^{-x},\\
 &\mathbb P_\theta\!\left(\|\bs R\|_2^2>2wd\,\middle|\,\mathbf X_{\mc S}\right)\\
 ={}&\mathbb P_\theta\!\left(\frac{\|\bs R\|_2^2}{w}>2d\,\middle|\,\mathbf X_{\mc S}\right)
 \le\mathbb P_\theta\!\left(\frac{\|\bs R\|_2^2}{w}>d+2\sqrt{dx}+2x\,\middle|\,\mathbf X_{\mc S}\right)
 \le e^{-x},
\end{aligned}
\end{equation}
where the first inequality in each line of \eqref{eq:sparse-residual-norm-tails} follows from event inclusion, since $d\ge16x$ implies $d/2\le d-2\sqrt{dx}$ for the lower tail and $2d\ge d+2\sqrt{dx}+2x$ for the upper tail.
A union bound then gives
\begin{equation}\label{eq:sparse-norm-bound}
 \mathbb P_\theta\!\left(
 \|\mathbf Z_j\|_2^2\ge v_jd/2,\ \|\bs R\|_2^2\le2wd
 \,\middle|\,\mathbf X_{\mc S}\right)
 \ge1-2e^{-x}=1-\delta/4.
\end{equation}

We next bound the probability that the signed inner product falls below half its conditional mean. Conditional on $(\mathbf Z_j,\mathbf X_{\mc S})$, equation \eqref{eq:sparse-conditional-regression} gives conditional mean $|\nu_j|\|\mathbf Z_j\|_2^2/v_j$ and variance $\tau_j^2\|\mathbf Z_j\|_2^2$ for $\eta_j\bs R^\top\mathbf Z_j$. The Gaussian tail bound and integration over $\mathbf Z_j$ give
\begin{equation}\label{eq:sparse-observed-innerproduct-tail}
\begin{aligned}
 &\mathbb P_\theta\!\left(
 \eta_j\bs R^\top\mathbf Z_j\le\frac{|\nu_j|}{2v_j}\|\mathbf Z_j\|_2^2
 \,\middle|\,\mathbf X_{\mc S}\right)\\
 ={}&\mathbb E_\theta\!\left[\left.
 \mathbb P_\theta\!\left(
 \eta_j\bs R^\top\mathbf Z_j-\frac{|\nu_j|}{v_j}\|\mathbf Z_j\|_2^2
 \le-\frac{|\nu_j|}{2v_j}\|\mathbf Z_j\|_2^2
 \,\middle|\,(\mathbf Z_j,\mathbf X_{\mc S})\right)
 \right|\mathbf X_{\mc S}\right]\\
 \le{}&\mathbb E_\theta\!\left(\left.
 \exp\!\left(-\frac{\nu_j^2\|\mathbf Z_j\|_2^2}{8v_j^2\tau_j^2}\right)
 \right|\mathbf X_{\mc S}\right)\\
 ={}&\left(1+\frac{\nu_j^2}{4v_j\tau_j^2}\right)^{-d/2}\\
 \le{}&(1+\gamma^2/4)^{-d/2}
 \le(1+d\gamma^2/8)^{-1}\le\delta/4,
\end{aligned}
\end{equation}
where the steps in \eqref{eq:sparse-observed-innerproduct-tail} are justified as follows:
\begin{enumerate}
\item The first equality uses iterated expectation given $(\mathbf Z_j,\mathbf X_{\mc S})$ and subtracts the conditional mean $|\nu_j|\|\mathbf Z_j\|_2^2/v_j$ from both sides of the event.
\item For the first inequality, \eqref{eq:sparse-conditional-regression} gives conditional variance $\tau_j^2\|\mathbf Z_j\|_2^2$ given $(\mathbf Z_j,\mathbf X_{\mc S})$. Apply $\mathbb P(G\le-u)\le e^{-u^2/2}$ for $G\sim \mathcal{N}(0,1)$ with $u=|\nu_j|\|\mathbf Z_j\|_2/(2v_j\tau_j)$, then take conditional expectation given $\mathbf X_{\mc S}$.
\item For the second equality, let $H=\|\mathbf Z_j\|_2^2/v_j$, so $H\mid\mathbf X_{\mc S}\sim\chi_d^2$. Apply $\mathbb E(e^{-cH}\mid\mathbf X_{\mc S})=(1+2c)^{-d/2}$ with $c=\nu_j^2/(8v_j\tau_j^2)>0$.
\item The second inequality uses $\tau_j^2\le w$, hence $\nu_j^2/(4v_j\tau_j^2)\ge\nu_j^2/(4v_jw)=\gamma^2/4$, by definition of $\gamma$.
\item The third inequality follows from Bernoulli's inequality, $(1+\gamma^2/4)^{d/2}\ge1+d\gamma^2/8$, since \eqref{eq:sparse-correlation-condition} gives $d\ge16\log(8/\delta)>2$.
\item The final inequality uses $d\gamma^2\ge8(4/\delta-1)$ from \eqref{eq:sparse-correlation-condition}, so $1+d\gamma^2/8\ge4/\delta$.
\end{enumerate}

By \eqref{eq:sparse-norm-bound} and \eqref{eq:sparse-observed-innerproduct-tail}, with conditional probability at least $1-\delta/2$, the two norm bounds hold and $\eta_j\bs R^\top\mathbf Z_j\ge|\nu_j|\|\mathbf Z_j\|_2^2/(2v_j)$. On this event,
\begin{equation}\label{eq:sparse-innerproduct-bound}
\begin{aligned}
 T_{\ell_\infty}
 &=\max_{k\in\mc T}\frac{|\bs R^\top\mathbf Z_k|}{\|\bs R\|_2\|\mathbf Z_k\|_2}\\
 &\ge\frac{|\bs R^\top\mathbf Z_j|}{\|\bs R\|_2\|\mathbf Z_j\|_2}\\
 &\ge\frac{|\nu_j|}{2v_j}\frac{\|\mathbf Z_j\|_2}{\|\bs R\|_2}\\
 &\ge\frac{|\nu_j|}{2v_j}\frac{\sqrt{v_jd/2}}{\sqrt{2wd}}\\
 &=\frac{|\nu_j|}{4\sqrt{v_jw}}
 =\frac\gamma4>q_a,
\end{aligned}
\end{equation}
where the first inequality in \eqref{eq:sparse-innerproduct-bound} bounds the maximum below by its term at $j$, the second uses the bound on the signed inner product, the third substitutes the norm bounds in \eqref{eq:sparse-norm-bound}, and the last uses \eqref{eq:sparse-correlation-condition}. This proves the conditional probability bound; its expectation over $\mathbf X_{\mc S}$ gives the unconditional bound.
\end{proof}

\paragraph{Consistency.}
The consistency statement following Theorem~\ref{power analysis} follows directly from that theorem. Under its respective parameter and dimension conditions, consider any sequence of alternatives for which $\sqrt n\|\mathbf V\bs\beta_{\mc T}\|_2/t^{1/4}\to\infty$ in the case of $\ell_2$ separation, or $\sqrt{n/\log(et)}\|\mathbf V\bs\beta_{\mc T}\|_\infty\to\infty$ in the case of $\ell_\infty$ separation. For each fixed $\beta\in(\alpha,1)$, the corresponding signal eventually exceeds the separation threshold with the finite constant $C_1$ or $C_2$ supplied by Theorem~\ref{power analysis}, so the limit inferior of the power is at least $\beta$. Since $\beta$ can be arbitrarily close to one and the power is at most one, the power converges to one.
\subsection{Proof of Theorem \ref{lower bound}}\label{app:proof-thm-lower-bound}\label{app:stronger-lower-bound}
We prove Theorem~\ref{lower bound} by first establishing the following stronger result.

For a single null point with identity covariance $\theta_*=(\bs0,\bs\beta_{\mc S}^*,\mathbf I_p,\sigma_*)$, the class $\Phi_\alpha(\{\theta_*\})$ imposes the size condition only at that point: $\mathbb E_{\theta_*}\psi\le\alpha$. For $q\in\{2,\infty\}$, define
\begin{equation}\label{eq:lower-alternative-at-null}
 \Theta_q(h;\theta_*)=
 \left\{\theta\in\Theta_q(h):
 \mathbf\Sigma=\mathbf I_p,\quad
 \mathcal L_\theta(Y,\bs X_{\mc S})=\mathcal L_{\theta_*}(Y,\bs X_{\mc S})\right\},
\end{equation}
where $\mathcal L_\theta(Y,\bs X_{\mc S})$ is the joint distribution of the response and background covariates under $\theta$. Thus these alternatives preserve both the covariate distribution and the joint distribution of $(Y,\bs X_{\mc S})$ at the chosen null point.

\begin{proposition}\label{prop:lower-general-null}
Fix $0<\alpha<\beta<1$, $0<\kappa\le M_1$ and $0<\xi<1$, and define
\begin{equation}\label{eq:lower-null-class}
 \Theta_0(\kappa,\xi)=
 \left\{\theta_*=(\bs0,\bs\beta_{\mc S}^*,\mathbf I_p,\sigma_*)\in\Theta_0:
 \sigma_*\ge\kappa,\quad\|\bs\beta_{\mc S}^*\|_2^2\le\xi M_2\right\}.
\end{equation}
This is a subclass of $\Theta_0$ with identity covariance, positive noise bounded below, and slack in the constraint on the coefficient norm.
\begin{enumerate}
\item \textbf{($\ell_2$ separation)} Suppose $t\le C_1n^2$ for a constant $C_1>0$. There exists $c_3>0$, depending only on $(\alpha,\beta,M_1,M_2,C_1,\kappa,\xi)$, such that, with $h_{3,n}=c_3t^{1/4}/\sqrt n$,
\begin{equation}\label{eq:lower-dense-uniform}
 \limsup_{n\to\infty}\ \sup_{\theta_*\in\Theta_0(\kappa,\xi)}
 \ \sup_{\psi\in\Phi_\alpha(\{\theta_*\})}
 \ \inf_{\theta\in\Theta_2(h_{3,n};\theta_*)}
 \mathbb E_\theta\psi\le\beta.
\end{equation}
\item \textbf{($\ell_\infty$ separation)} Suppose $\log t\le C_2n$ for a constant $C_2>0$. There exists $c_4>0$, depending only on $(\alpha,\beta,M_1,M_2,C_2,\kappa,\xi)$, such that, with $h_{4,n}=c_4\sqrt{\log(et)/n}$,
\begin{equation}\label{eq:lower-sparse-uniform}
 \limsup_{n\to\infty}\ \sup_{\theta_*\in\Theta_0(\kappa,\xi)}
 \ \sup_{\psi\in\Phi_\alpha(\{\theta_*\})}
 \ \inf_{\theta\in\Theta_\infty(h_{4,n};\theta_*)}
 \mathbb E_\theta\psi\le\beta.
\end{equation}
\end{enumerate}
\end{proposition}

To prove the lower bounds, we construct alternatives with signals at the claimed separation rates whose mixture remains close to a null distribution. Such closeness limits the difference between a test's average power over these alternatives and its rejection probability under the null, so no level-$\alpha$ test can have high power uniformly over them; Lemma~\ref{chi sqaure inequality} makes this comparison precise. At null points with identity covariance and general coefficients for the background covariates, Lemma~\ref{calculation for alternatives} gives the expected squared deviation of the mixture likelihood ratio from one. We spread the signal across all target coordinates with random signs for the $\ell_2$ lower bound, and place it at a randomly chosen coordinate for the $\ell_\infty$ lower bound.

\subsubsection{Auxiliary results}
We use the Gaussian likelihood ratio identity of \citet[Lemma~2]{bradic2022testability} to evaluate the mixture moment for our alternatives.
\begin{lemma}\label{chi sqaure inequality}
  For a measurable test \(0\le\psi_*\le1\) and parameter points \(\theta_*,\theta_1,\ldots,\theta_N\), suppose \(\mathbb P_{\theta_j}\ll\mathbb P_{\theta_*}\) for every \(j\), and the mixture likelihood ratio has finite second moment. Then
  $$\left|N^{-1}\sum_{j=1}^N \bbE_{\theta_j} \psi_*-\bbE_{\theta_*}\psi_*\right|\leq \left[\bbE_{\theta_*}\left(N^{-1}\sum_{j=1}^N\frac{\rmd \bbP_{\theta_j}}{\rmd \bbP_{\theta_*}}-1\right)^2\right]^{1/2}.$$
\end{lemma}
\begin{proof}

\begin{equation}\label{eq:lower-testing-comparison-derivation}
  \begin{aligned}
    \left| N^{-1} \sum_{j=1}^N \mathbb{E}_{\theta_j} \psi_* - \mathbb{E}_{\theta_*} \psi_* \right|
    &= \left| N^{-1} \sum_{j=1}^N \left( \mathbb{E}_{\theta_*} \psi_* \frac{\rmd\mathbb{P}_{\theta_j}}{\rmd\mathbb{P}_{\theta_*}} - \mathbb{E}_{\theta_*} \psi_* \right) \right|
    \\ &\quad\leq \mathbb{E}_{\theta_*} \psi_* \left| N^{-1} \sum_{j=1}^N \frac{\rmd\mathbb{P}_{\theta_j}}{\rmd\mathbb{P}_{\theta_*}} - 1 \right|\\
    &\stackrel{(i)}{\leq} \mathbb{E}_{\theta_*} \left| N^{-1} \sum_{j=1}^N \frac{\rmd\mathbb{P}_{\theta_j}}{\rmd\mathbb{P}_{\theta_*}} - 1 \right|
    \stackrel{(ii)}{\leq} \left[ \mathbb{E}_{\theta_*} \left( N^{-1} \sum_{j=1}^N \frac{\rmd\mathbb{P}_{\theta_j}}{\rmd \mathbb{P}_{\theta_*}} - 1 \right)^2 \right]^{1/2},
  \end{aligned}
\end{equation}

where the equality in \eqref{eq:lower-testing-comparison-derivation} is the change-of-measure identity, the first inequality is the triangle inequality, (i) uses $0\le\psi_*\le1$, and (ii) is the Cauchy--Schwarz inequality under $\mathbb P_{\theta_*}$.
\end{proof}

\begin{lemma}\label{calculation for alternatives}
  Consider $\theta_*=(\bs{0},\bs{\beta}_{\mc{S}}^*,\mathbf{I}_p,\sigma_*)$ and a fixed sequence $\left\{\bs{\delta}_j\in \bbR^t\right\}_{j=1}^N$ satisfying that $\left\|\bs{\delta}_j\right\|=r$ for some $r>0$. Define alternative points $\theta_j$ as
  $$\theta_j=\left(\rho\bs{\delta}_j,\bs{\beta}_{\mc{S}}^*,\mathbf{I}_p,\sqrt{\sigma_*^2-\rho^2 r^2}\right)$$
  for some $\rho\in (0, \sigma_*/r)$. Then we have
  $$\begin{aligned}&\bbE_{\theta_*}\left(N^{-1}\sum_{j=1}^{N}\frac{\rmd \bbP_{\theta_j}}{\rmd \bbP_{\theta_*}}-1\right)^2\\={}&N^{-2}\sum_{j_1=1}^N\sum_{j_2=1}^N\left[1-\frac{\rho^2}{\sigma_*^2}\bs{\delta}_{j_1}^T\bs{\delta}_{j_2}\right]^{-n}-1.\end{aligned}$$
  
\end{lemma}
\begin{proof}

We first remove the common background term and standardise
the response.
Let $U=(Y-\bs X_{\mc S}^\top\bs\beta_{\mc S}^*)/\sigma_*$.
The common invertible transformation leaves the likelihood
ratios unchanged.
Since the covariate covariance is $\mathbf I_p$,
$(U,\bs X_{\mc T})$ is independent of $\bs X_{\mc S}$
under the null and every alternative.
The distribution of $\bs X_{\mc S}$ is unchanged, so its
density cancels from each likelihood ratio.

For $j=1,\ldots,N$, define
$\bs b_j=\rho\bs\delta_j/\sigma_*$.
Let $f_*$ and $f_j$ be the densities of
$(U,\bs X_{\mc T})$ under $\theta_*$ and $\theta_j$.
These are centred Gaussian densities with covariance matrices
$\mathbf I_{t+1}$ and $\mathbf\Sigma_j=\mathbf I_{t+1}+\mathbf A_j$,
respectively, where
$$
\mathbf A_j=
\begin{pmatrix}
0&\bs b_j^\top\\
\bs b_j&\mathbf0_{t\times t}
\end{pmatrix}.
$$

Since $\|\bs b_j\|_2=\rho r/\sigma_*<1$, each $\mathbf\Sigma_j$ is positive definite and $\lambda_{\min}(\mathbf\Sigma_j^{-1})=(1+\|\bs b_j\|_2)^{-1}>1/2$. Thus, for $j_1,j_2\in\{1,\ldots,N\}$, the matrix $\mathbf\Sigma_{j_1}^{-1}+\mathbf\Sigma_{j_2}^{-1}-\mathbf I_{t+1}$ is positive definite, ensuring that the Gaussian likelihood ratio product has finite expectation. Applying \citet[Lemma~2]{bradic2022testability} gives
\begin{equation}\label{eq:lower-gaussian-product-integral}
 \int_{\mathbb R^{t+1}}\frac{f_{j_1}(\bs w)f_{j_2}(\bs w)}{f_*(\bs w)}\,\rmd\bs w
 =\det(\mathbf I_{t+1}-\mathbf A_{j_1}\mathbf A_{j_2})^{-1/2}.
\end{equation}
The block form of $\mathbf A_j$ and the matrix determinant lemma give
\begin{equation}\label{eq:lower-transformed-determinant}
\begin{aligned}
 \det(\mathbf I_{t+1}-\mathbf A_{j_1}\mathbf A_{j_2})
 &=(1-\bs b_{j_1}^\top\bs b_{j_2})
   \det(\mathbf I_t-\bs b_{j_1}\bs b_{j_2}^\top)\\
 &=(1-\bs b_{j_1}^\top\bs b_{j_2})^2.
\end{aligned}
\end{equation}
Moreover, $|\bs b_{j_1}^\top\bs b_{j_2}|\le\rho^2r^2/\sigma_*^2<1$, so $1-\bs b_{j_1}^\top\bs b_{j_2}>0$. Combining \eqref{eq:lower-gaussian-product-integral} and \eqref{eq:lower-transformed-determinant} over the $n$ independent observations yields
\begin{equation}\label{eq:lem11CaiGuo2017}
 \bbE_{\theta_*}\left(
 \frac{\rmd\bbP_{\theta_{j_1}}}{\rmd\bbP_{\theta_*}}
 \frac{\rmd\bbP_{\theta_{j_2}}}{\rmd\bbP_{\theta_*}}\right)
 =\left(1-\frac{\rho^2}{\sigma_*^2}
 \bs\delta_{j_1}^\top\bs\delta_{j_2}\right)^{-n}.
\end{equation}

  Each likelihood ratio has expectation one under $\theta_*$. Expanding the square therefore gives
  $$\begin{aligned}
    \bbE_{\theta_*}\left(N^{-1}\sum_{j=1}^{N}\frac{\rmd \bbP_{\theta_j}}{\rmd \bbP_{\theta_*}}-1\right)^2
    &=N^{-2}\sum_{j_1=1}^N\sum_{j_2=1}^N\bbE_{\theta_*}\left(\frac{\rmd \bbP_{\theta_{j_1}}}{\rmd \bbP_{\theta_*}}\times \frac{\rmd \bbP_{\theta_{j_2}}}{\rmd \bbP_{\theta_*}}\right)-1\\
    &=N^{-2}\sum_{j_1=1}^N\sum_{j_2=1}^N\left[1-\frac{\rho^2}{\sigma_*^2}\bs{\delta}_{j_1}^T\bs{\delta}_{j_2}\right]^{-n}-1.
  \end{aligned}$$
  This proves the lemma.
\end{proof}

\subsubsection[Proof of Proposition~\ref{prop:lower-general-null} and Theorem~\ref{lower bound}]{Proof of Proposition~\ref{prop:lower-general-null} and deduction of Theorem~\ref{lower bound}}
\label{app:proof-lower-general-null}
\begin{proof}
Fix constants $0<\kappa\le M_1$ and $0<\xi<1$. We prove the bounds at every null point $\theta_*=(\bs0,\bs\beta_{\mc S}^*,\mathbf I_p,\sigma_*)\in\Theta_0$ satisfying
\begin{equation}\label{eq:lower-null-anchor}
 \sigma_*\ge\kappa,\qquad \|\bs\beta_{\mc S}^*\|_2^2\le\xi M_2.
\end{equation}
The null points with identity covariance satisfying \eqref{eq:lower-null-anchor} form the class $\Theta_0(\kappa,\xi)$ defined in \eqref{eq:lower-null-class}.

\paragraph{Alternative parameters and their likelihoods.}
We construct alternatives that preserve the joint distribution of $(Y,\bs X_{\mc S})$ and the covariate distribution. Define $a_0=\min\{1/\sqrt2,\sqrt{(1-\xi)M_2}/M_1\}$. For $\bs b\in\mathbb R^t$ with $0<\|\bs b\|_2=a\le a_0$, define
\begin{equation}\label{eq:lower-prior-parameters}
 \bs\beta_{\mc T}(\bs b)=\sigma_*\bs b,\qquad
 \bs\beta_{\mc S}(\bs b)=\bs\beta_{\mc S}^*,\qquad
 \theta_{\bs b}=\bigl(\sigma_*\bs b,\bs\beta_{\mc S}^*,\mathbf I_p,\sigma_*\sqrt{1-a^2}\bigr).
\end{equation}
The coefficient constraint follows from
\begin{equation}\label{eq:lower-coefficient-support}
 \|\bs\beta_{\mc T}(\bs b)\|_2^2+\|\bs\beta_{\mc S}(\bs b)\|_2^2
 =\sigma_*^2a^2+\|\bs\beta_{\mc S}^*\|_2^2
 \le M_1^2a_0^2+\xi M_2\le M_2,
\end{equation}
where the inequalities in \eqref{eq:lower-coefficient-support} use \eqref{eq:lower-null-anchor}, $\sigma_*\le M_1$, and the definition of $a_0$. The covariance $\mathbf I_p$ is admissible since $M_0>1$, and the noise standard deviation $\sigma_*\sqrt{1-a^2}$ lies in $[\sigma_*/\sqrt2,\sigma_*]\subset(0,M_1]$. Hence every $\theta_{\bs b}$ lies in $\Theta$.

We next verify the distributional restriction in \eqref{eq:lower-alternative-at-null}. Since the covariance is $\mathbf I_p$, the target vector $\bs X_{\mc T}\sim \mathcal{N}_t(\bs0,\mathbf I_t)$ is independent of $\bs X_{\mc S}$. Under \eqref{eq:lower-prior-parameters}, the Gaussian linear model \eqref{linear model} therefore gives $Y\mid\bs X_{\mc S}$ a normal distribution with mean $\bs X_{\mc S}^\top\bs\beta_{\mc S}^*$ and variance $\sigma_*^2a^2+\sigma_*^2(1-a^2)=\sigma_*^2$, as under $\theta_*$. Since the distribution of $\bs X_{\mc S}$ is unchanged, $\mathcal L_{\theta_{\bs b}}(Y,\bs X_{\mc S})=\mathcal L_{\theta_*}(Y,\bs X_{\mc S})$. 

We next consider finite mixtures of these alternatives. Let $N\ge1$ be an integer and let $\bs b_1,\ldots,\bs b_N\in\mathbb R^t$ be fixed vectors satisfying $\|\bs b_j\|_2=a$ for every $j\in\{1,\ldots,N\}$, where $0<a\le a_0$. Each vector $\bs b_j$ defines an alternative parameter $\theta_{\bs b_j}$ through \eqref{eq:lower-prior-parameters}. Define $\Pi$ by choosing each of these vectors with probability $1/N$, and let $\bs B$ and $\bs B'$ be independent random vectors with distribution $\Pi$. Let $L_\Pi=N^{-1}\sum_{j=1}^N\rmd\mathbb P_{\theta_{\bs b_j}}/\rmd\mathbb P_{\theta_*}$ be the resulting mixture likelihood ratio for the full sample. Apply Lemma~\ref{calculation for alternatives} with $\rho=\sigma_*$, $r=a$, and $\bs\delta_j=\bs b_j$. Its condition $0<\rho<\sigma_*/r$ holds since $0<a\le a_0\le1/\sqrt2<1$. This gives
\begin{equation}\label{eq:lower-mixture-moment}
 \begin{aligned}
 \mathbb E_{\theta_*}(L_\Pi-1)^2
 &=\frac1{N^2}\sum_{j=1}^N\sum_{k=1}^N
 (1-\bs b_j^\top\bs b_k)^{-n}-1\\
 &=\mathbb E_{\bs{B},\bs{B}'\sim\Pi}(1-\bs{B}^\top \bs{B}')^{-n}-1,
 \end{aligned}
\end{equation}
where the first equality in \eqref{eq:lower-mixture-moment} follows from Lemma~\ref{calculation for alternatives} with the stated substitutions, and the second follows because $\bs B$ and $\bs B'$ are independent draws from the uniform prior $\Pi$.
Since $|\bs B^\top\bs B'|\le a^2\le1/2$, every denominator in \eqref{eq:lower-mixture-moment} is positive and the mixture likelihood ratio has a finite second moment.

\paragraph{$\ell_2$ separation.}
Choose a constant $c>0$ such that $c\sqrt{C_1}\le a_0^2$ and $\sqrt{e^{c^2/2}-1}<\beta-\alpha$, where $C_1$ is the constant in Theorem~\ref{lower bound}. Such a choice exists because $a_0>0$ and $\beta>\alpha$. Let $a^2=c\sqrt t/n$, and define $c_3=\kappa\sqrt c$ and $h_{3,n}=c_3t^{1/4}/\sqrt n$.

We first construct the prior and verify that its parameter points belong to $\Theta_2(h_{3,n};\theta_*)$. Choose $\bs b_1,\ldots,\bs b_N$ to enumerate the vectors in $\{a\bs s/\sqrt t:\bs s\in\{-1,1\}^t\}$, with $N=2^t$, and let $\Pi$ be uniform on these vectors. Since $t\le C_1n^2$, we have $a^2=c\sqrt t/n\le c\sqrt{C_1}\le a_0^2$. The conditional covariance is $\mathbf I_t$, so \eqref{eq:lower-prior-parameters} gives each prior point the signal norm
\begin{equation}\label{eq:lower-dense-signal}
 \|\bs\beta_{\mc T}(\bs b)\|_2
 =\sigma_*a
 \ge\kappa\sqrt c\,\frac{t^{1/4}}{\sqrt n},
\end{equation}
where the inequality in \eqref{eq:lower-dense-signal} uses \eqref{eq:lower-null-anchor}. Together with the parameter and distributional restrictions verified for \eqref{eq:lower-prior-parameters}, equation \eqref{eq:lower-dense-signal} shows that every $\theta_{\bs b_j}$ belongs to $\Theta_2(h_{3,n};\theta_*)$.

We next bound $\mathbb E_{\theta_*}(L_\Pi-1)^2$ in \eqref{eq:lower-mixture-moment}. A draw $\bs B\sim\Pi$ can be written as $\bs B=a\bs S/\sqrt t$, where $\bs S=(S_1,\ldots,S_t)^\top$ has independent coordinates, each equal to $1$ or $-1$ with probability $1/2$. For an independent copy $\bs S'$ of $\bs S$, let $\bs B'=a\bs S'/\sqrt t$. Then $\bs B^\top\bs B'=a^2W/t$, where $W=\sum_{i=1}^t S_iS_i'$. The products $S_iS_i'$ are independent uniform signs because the pairs $(S_i,S_i')$ are independent and each product equals $1$ or $-1$ with probability $1/2$.

For $|z|\le1/2$, the power series gives $-\log(1-z)\le z+z^2$: indeed, its terms of order at least two sum to at most $z^2/\{2(1-|z|)\}\le z^2$. Since $|a^2W/t|\le a^2\le1/2$, applying this bound with $z=a^2W/t$ gives
$$
 \mathbb E(1-a^2W/t)^{-n}
 \le\mathbb E\exp\left(\frac{na^2}{t}W+\frac{na^4}{t^2}W^2\right)
 =\mathbb E\exp(AW+BW^2),
$$
where $A=c/\sqrt t$ and $B=c^2/(nt)$. The moment generating function of $W$, a sum of $t$ independent uniform signs, is $\mathbb E e^{uW}=\prod_{j=1}^t\{(e^u+e^{-u})/2\}=(\cosh u)^t$ for every $u\in\mathbb R$. To bound $\mathbb E\exp(AW+BW^2)$, introduce $G_0\sim \mathcal{N}(0,1)$ independently of $W$. The normal moment generating function gives $e^{BW^2}=\mathbb E(e^{\sqrt{2B}G_0W}\mid W)$. For $n>2c^2$, we have $2tB<1$, and hence
\begin{equation}\label{eq:lower-dense-gaussian-moment}
\begin{aligned}
 \mathbb E\exp(AW+BW^2)
 &=\mathbb E_{G_0}\mathbb E_W\exp\{(A+\sqrt{2B}G_0)W\}\\
 &=\mathbb E_{G_0}\bigl\{\cosh(A+\sqrt{2B}G_0)\bigr\}^t\\
 &\le\mathbb E_{G_0}\exp\{t(A+\sqrt{2B}G_0)^2/2\}\\
 &=\frac{e^{tA^2/2}}{\sqrt{2\pi}}\int_{\mathbb R}
    \exp\left\{-\frac{1-2tB}{2}u^2+tA\sqrt{2B}\,u\right\}\,\rmd u\\
 &=\frac{1}{\sqrt{1-2tB}}
    \exp\left\{\frac{tA^2}{2(1-2tB)}\right\},
\end{aligned}
\end{equation}
where the first equality in \eqref{eq:lower-dense-gaussian-moment} uses independence and Tonelli's theorem, the second equality uses the moment generating function of $W$, the inequality uses $\cosh v\le e^{v^2/2}$, the third equality writes the expectation using the standard normal density, and the last equality follows by completing the square with centre $tA\sqrt{2B}/(1-2tB)$.

With $tA^2=c^2$ and $tB=c^2/n$, equations \eqref{eq:lower-dense-gaussian-moment} and \eqref{eq:lower-mixture-moment} yield
\begin{equation}\label{bound dense}
 \mathbb E_{\theta_*}(L_\Pi-1)^2
 \le(1-2c^2/n)^{-1/2}\exp\left\{\frac{c^2}{2(1-2c^2/n)}\right\}-1
 \longrightarrow e^{c^2/2}-1.
\end{equation}
The bound in \eqref{bound dense} is uniform over $1\le t\le C_1n^2$.

For any $\psi\in\Phi_\alpha(\{\theta_*\})$, we average over the prior and apply Lemma~\ref{chi sqaure inequality} to obtain
\begin{equation}\label{eq:lower-testing-bound}
 \begin{aligned}
 \inf_{\theta\in\Theta_2(h_{3,n};\theta_*)}\mathbb E_\theta\psi
 &\le\frac1N\sum_{j=1}^N\mathbb E_{\theta_{\bs b_j}}\psi\\
 &\le\alpha+\left\{\mathbb E_{\theta_*}(L_\Pi-1)^2\right\}^{1/2}.
 \end{aligned}
\end{equation}
The upper bound on the second moment in \eqref{bound dense} depends only on $c$ and $n$. Hence \eqref{eq:lower-testing-bound} holds uniformly over $\theta_*\in\Theta_0(\kappa,\xi)$ and $\psi\in\Phi_\alpha(\{\theta_*\})$, and taking the two suprema followed by the asymptotic upper limit gives
\begin{equation}\label{eq:lower-pointwise-testing}
 \limsup_{n\to\infty}\ \sup_{\theta_*\in\Theta_0(\kappa,\xi)}
 \ \sup_{\psi\in\Phi_\alpha(\{\theta_*\})}
 \inf_{\theta\in\Theta_2(h_{3,n};\theta_*)}\mathbb E_\theta\psi
 \le\alpha+\sqrt{e^{c^2/2}-1}<\beta.
\end{equation}
This proves \eqref{eq:lower-dense-uniform}. The definition of $a_0$ and the two conditions on $c$ show that $c_3$ depends only on $(\alpha,\beta,M_1,M_2,C_1,\kappa,\xi)$, as asserted.

\paragraph{$\ell_\infty$ separation.}
Choose a constant $c>0$ such that $c\le1/4$, $c(C_2+1)\le a_0^2$ and $2\sqrt c<\beta-\alpha$, where $C_2$ is the constant in Theorem~\ref{lower bound}. Such a choice exists because $a_0>0$ and $\beta>\alpha$. Let $a^2=c\log(et)/n$, and define $c_4=\kappa\sqrt c$ and $h_{4,n}=c_4\sqrt{\log(et)/n}$.

We first construct the prior and verify that its parameter points belong to $\Theta_\infty(h_{4,n};\theta_*)$. Take $N=t$ and $\bs b_j=a\bs e_j$ for $j\in[t]$, where $\bs e_j$ is the $j$-th coordinate vector, and let $\Pi$ be uniform on these vectors. Since $\log t\le C_2n$ and $n\ge1$, we have $a^2\le c(C_2+1/n)\le a_0^2$. The conditional covariance is $\mathbf I_t$, so \eqref{eq:lower-prior-parameters} gives each prior point the signal norm
\begin{equation}\label{eq:lower-sparse-signal}
 \|\bs\beta_{\mc T}(a\bs e_j)\|_\infty
 =\sigma_*a
 \ge\kappa\sqrt c\sqrt{\frac{\log(et)}n},
\end{equation}
where the inequality in \eqref{eq:lower-sparse-signal} uses \eqref{eq:lower-null-anchor}. Together with the parameter and distributional restrictions verified for \eqref{eq:lower-prior-parameters}, equation \eqref{eq:lower-sparse-signal} shows that every $\theta_{a\bs e_j}$ belongs to $\Theta_\infty(h_{4,n};\theta_*)$.

We next bound $\mathbb E_{\theta_*}(L_\Pi-1)^2$ in \eqref{eq:lower-mixture-moment}. For independent $\bs B,\bs B'\sim\Pi$, we have $\mathbb P(\bs B^\top\bs B'=a^2)=1/t$ and $\mathbb P(\bs B^\top\bs B'=0)=1-1/t$. Thus \eqref{eq:lower-mixture-moment} gives
\begin{equation}\label{eq:sparse-lower-mixture}
 \mathbb E_{\theta_*}(L_\Pi-1)^2=\frac{(1-a^2)^{-n}-1}{t}
 \le\frac{(et)^{2c}-1}{t}
 \le2c\frac{\log(et)(et)^{2c}}{t}\le4c,
\end{equation}
where the first inequality in \eqref{eq:sparse-lower-mixture} uses $a^2\le1/2$ and $-\log(1-a^2)\le2a^2$, the second uses $e^z-1\le ze^z$ with $z=2c\log(et)\ge0$, and the last uses $c\le1/4$ and the fact that $u\mapsto\log(eu)/\sqrt u$ attains its maximum $2/\sqrt e$ over $u\ge1$ at $u=e$:
$$
 2c\frac{\log(et)(et)^{2c}}{t}
 \le2c\sqrt e\,\frac{\log(et)}{\sqrt t}
 \le2c\sqrt e\,\frac{2}{\sqrt e}=4c.
$$

We can now bound the power of any $\psi\in\Phi_\alpha(\{\theta_*\})$:
\begin{equation}\label{eq:lower-sparse-testing-bound}
 \begin{aligned}
 \inf_{\theta\in\Theta_\infty(h_{4,n};\theta_*)}\mathbb E_\theta\psi
 &\le\frac1t\sum_{j=1}^t\mathbb E_{\theta_{a\bs e_j}}\psi\\
 &\le\alpha+\left\{\mathbb E_{\theta_*}(L_\Pi-1)^2\right\}^{1/2}
 \le\alpha+2\sqrt c<\beta,
 \end{aligned}
\end{equation}
where the first inequality in \eqref{eq:lower-sparse-testing-bound} uses the inclusion of the prior points in the alternative class, the second uses Lemma~\ref{chi sqaure inequality} and $\mathbb E_{\theta_*}\psi\le\alpha$, the third uses \eqref{eq:sparse-lower-mixture}, and the last uses the choice of $c$. The bound in \eqref{eq:lower-sparse-testing-bound} is independent of $n$, $t$, $\theta_*$ and $\psi$. Taking the suprema over $\theta_*\in\Theta_0(\kappa,\xi)$ and $\psi\in\Phi_\alpha(\{\theta_*\})$, followed by the asymptotic upper limit, proves \eqref{eq:lower-sparse-uniform}. The definition of $a_0$ and the three conditions on $c$ show that $c_4$ depends only on $(\alpha,\beta,M_1,M_2,C_2,\kappa,\xi)$, as asserted.

\paragraph{Proof of Theorem~\ref{lower bound} from Proposition~\ref{prop:lower-general-null}.}
Let $\kappa=M_1/2$ and $\xi=1/2$, and choose $\theta_*=(\bs0,\bs0,\mathbf I_p,M_1/2)\in\Theta_0(\kappa,\xi)$. For $q\in\{2,\infty\}$, we have $\Phi_\alpha(\Theta_0)\subseteq\Phi_\alpha(\{\theta_*\})$ and $\Theta_q(h;\theta_*)\subseteq\Theta_q(h)$. These inclusions imply that the bounds \eqref{eq:lower-dense-uniform} and \eqref{eq:lower-sparse-uniform} also hold for the test and alternative classes in \eqref{eq:lower-dense-global} and \eqref{eq:lower-sparse-global}.

The chosen values of $\kappa$ and $\xi$ give the stated dependence of the constants in Theorem~\ref{lower bound}. This completes the proof.
\end{proof}

\subsection{Proofs for the $G$-CRT power analysis}\label{app:gcrt-ridge-power}

Appendix~\ref{app:gcrt-common-conditions} gives the model,
sampling specifications, and proof notation, together with a general
tail bound for Gaussian quadratic forms.
Appendix~\ref{app:gcrt-centered-proof} states two lemmas and uses them to
prove Theorem~\ref{thm:gcrt-centered-power}: the approximations to the
residual score vectors and regularised residual Gram matrices, together with a lower bound on
the difference between the target signal contributions, imply that the
observed centred statistic exceeds all randomised statistics with
probability tending uniformly to one.
Appendices~\ref{app:gcrt-residual-approximation-proof} and~\ref{app:gcrt-score-approximation-proof} prove
Lemma~\ref{gcrtb:lem:endpoints} using three auxiliary lemmas, whose
proofs are given in Appendices~\ref{app:gcrt-ridge-concentration-proof}--\ref{app:gcrt-comparison-cross-moment-proof}.

In Appendices~\ref{app:gcrt-ridge-power} and~\ref{app:gcrt-centering-comparisons},
we measure vector remainders in the Euclidean norm and matrix remainders
in the operator norm. For a positive deterministic sequence $a_n$,
$\bs x_n=O_{\mathbb P}(a_n)$ means
$\|\bs x_n\|_2=O_{\mathbb P}(a_n)$, and
$\mathbf A_n=O_{\mathbb P}(a_n)$ means
$\|\mathbf A_n\|_{\mathrm{op}}=O_{\mathbb P}(a_n)$.
The same convention applies to $o_{\mathbb P}(a_n)$.

\subsubsection[Model, sampling and uniform conditions for G-CRT power]{Model, sampling and uniform conditions for $G$-CRT power}\label{app:gcrt-common-conditions}
We use the Gaussian linear model and population notation in
Appendix~\ref{app:additional-power}.
With $\beta_0=\mathbb E(Y_i)$, the residual representation is
\begin{equation}\label{gcrtint:eq:response}
\begin{aligned}
 \bs E_i&=\bs X_{i,\mathcal T}-\bs\mu_{\mathcal T}
             -\mathbf C(\bs X_{i,\mathcal S}-\bs\mu_{\mathcal S}),\\
 Y_i&=\beta_0+(\bs X_{i,\mathcal S}-\bs\mu_{\mathcal S})^\top\bs\gamma
                   +\bs E_i^\top\bs\beta_{\mathcal T}+\sigma\varepsilon_i,
\end{aligned}
\end{equation}
where $\bs E_i\sim \mathcal{N}_t(0,\mathbf V)$ is independent of
$\bs X_{i,\mathcal S}$ and the errors $\varepsilon_i$ are independent
standard Gaussians, independent of all covariates.
The covariate precision matrix $\mathbf\Omega=\mathbf\Sigma^{-1}$
has target block
\begin{equation}\label{gcrtint:eq:target-precision}
\mathbf\Theta=\mathbf V^{-1}
\end{equation}
and block $-\mathbf A=-\mathbf\Theta\mathbf C$ between the target and background covariates.
The supplied graph contains the conditional independence graph of the covariate distribution. The target set,
graph, order, and tuning parameters are fixed independently of the inference
sample. Independently selected specifications are treated conditionally
on the auxiliary data, with all bounds uniform in the selected specification.
Means and the response intercept are unrestricted.

We use the population constants and the matrix
$\mathbf J$ defined in Section~\ref{sec:gcrt-ridge-power}.
For $j\in\mathcal T$, define $b_j=|N_j\cap\mathcal S|$ and
$g_j=|N_j\cap\mathcal T|$. For each population let
\begin{equation}\label{gcrtb:eq:error-rates}
 b_{\max}=\max_j b_j,\qquad d_n=b_{\max}/n,\qquad
 e_n=\left\{\frac{\sum_jg_j+t}{n}\right\}^{1/2},\qquad
 \eta_n=e_n+d_nr.
\end{equation}

We use the class $\mathcal C_n(h_n)$ defined by
\eqref{gcrtcent:eq:population-conditions} in
Section~\ref{sec:gcrt-ridge-power}, assume that it is nonempty,
and impose the graph and predetermined order requirements stated here.
Fix $\ell,K>0$, integers $L,M\geq1$,
and a predetermined deterministic sequence $0<\lambda_n\leq Kn$.
The intercept is unpenalised and the background penalty is $\ell n$.
The actual sampler generates one shared hub with $L$ full forward sweeps
and $M$ branches using the complete reversed update list and fresh
independent randomness. 
By \eqref{gcrtb:eq:error-rates}, we express the condition in
\eqref{gcrtcent:eq:rates} through $d_n$ and $e_n$ as
\begin{equation}\label{gcrts:eq:general-rates}
 \sup_{\vartheta\in \mathcal C_n(h_n)}(d_n+e_n/h_n)\longrightarrow0.
\end{equation}
In particular, \eqref{gcrts:eq:general-rates} implies that $t=o(n)$ and, for each $j\in\mathcal T$,
the residual space $\operatorname{col}([\bs1_n,\mathbf X_{N_j}])^\perp$
has dimension $n-1-b_j-g_j\to\infty$.
All four conclusions of Lemma~\ref{gcrtb:lem:endpoints} hold uniformly over
$\mathcal C_n(h_n)$ for every fixed $M\geq1$.
The power conclusion of Theorem~\ref{thm:gcrt-centered-power}
requires $M\geq2$.
Define $\tau_0=\frac12\{\ell/(\ell+c_+s_+)\}^2$ for this class.

\medskip
\textbf{Handling the intercept.} 
To remove the unpenalised intercept, let $N=n-1$ and express the centred data
columns in a fixed orthonormal basis of
$\{\bs x\in\mathbb R^n:\sum_{i=1}^n x_i=0\}$.
Write $\mathbf S\in\mathbb R^{N\times s}$, $\mathbf X_c\in\mathbb R^{N\times t}$
and $\bs Y_c\in\mathbb R^N$ for the background data, target data and response
in these coordinates. This change preserves inner products of centred vectors.

Explicitly, let $\mathbf U\in\mathbb R^{n\times N}$ contain that
orthonormal basis, so $\mathbf U^\top\mathbf U=\mathbf I_N$,
$\mathbf U^\top\bs1_n=0$, and
$\mathbf U\mathbf U^\top=\mathbf I_n-\bs1_n\bs1_n^\top/n$.
For any two raw columns $\bs x,\bs y$,
$$
(\mathbf U^\top\bs x)^\top(\mathbf U^\top\bs y)
=\bs x^\top(\mathbf I_n-\bs1_n\bs1_n^\top/n)\bs y,
$$
the inner product after centring. If a raw matrix $\mathbf R$ has
independent centred Gaussian rows with covariance $\mathbf\Sigma_R$,
then $\widetilde{\mathbf R}=\mathbf U^\top\mathbf R$ satisfies
\begin{equation}\label{gcrtb:eq:centered-row-covariance}
\operatorname{Cov}(\widetilde{\mathbf R}_{a\cdot}^{\top},
                   \widetilde{\mathbf R}_{b\cdot}^{\top})
=\sum_{i=1}^n U_{ia}U_{ib}\mathbf\Sigma_R
=\mathbf1\{a=b\}\mathbf\Sigma_R,
\qquad 1\leq a,b\leq N.
\end{equation}
Joint Gaussianity thus makes its rows independent with the original
row covariance. Apply \eqref{gcrtb:eq:centered-row-covariance} separately to the background
deviations, the target residuals and the response noise. Their mutual
independence is preserved by the deterministic maps $\mathbf U^\top$,
and constant mean terms vanish because $\mathbf U^\top\bs1_n=0$.

With the centred data, the model in \eqref{gcrtint:eq:response} becomes
\begin{equation}\label{gcrtint:eq:centered-model}
 \mathbf X_c=\mathbf S\mathbf C^\top+\mathbf E,\qquad
 \bs Y_c=\mathbf S\bs\gamma+\mathbf E\bs\beta_{\mathcal T}
                  +\sigma\bs\varepsilon,
\end{equation}
where $\mathbf S$ has independent
$\mathcal{N}_s(0,\mathbf\Sigma_{\mathcal S\mathcal S})$ rows,
$\mathbf E$ has independent $\mathcal{N}_t(0,\mathbf V)$ rows, and
$\bs\varepsilon\sim \mathcal{N}_N(0,\mathbf I_N)$; these three objects are independent.

\medskip
\textbf{Ridge residuals, scores, and Gram matrices.}
We now express the ridge residuals used to construct
\eqref{gcrtcent:eq:residual-scores} in the centred coordinates of
\eqref{gcrtint:eq:centered-model}.
We regress the response $\bs Y_c$ and each observed or randomised
target column separately on the same background matrix $\mathbf S$,
with penalty $\ell n$.

Let $\bs y_c\in\mathbb R^N$ represent any one of these vectors,
and define $\widehat{\bs b}$ as the minimiser, over
$\bs b\in\mathbb R^s$, of
$\|\bs y_c-\mathbf S\bs b\|^2+\ell n\|\bs b\|^2$.
The first-order condition gives
$(\mathbf S^\top\mathbf S+\ell n\mathbf I_s)\widehat{\bs b}
=\mathbf S^\top\bs y_c$.
Since $\ell>0$, the matrices
$\mathbf S^\top\mathbf S+\ell n\mathbf I_s$ and
$\mathbf S\mathbf S^\top+\ell n\mathbf I_N$ are positive definite,
including when $s\geq n$.
The fitted residual is therefore
$\bs y_c-\mathbf S\widehat{\bs b}=\mathbf H\bs y_c$, where
\begin{equation}\label{gcrtint:eq:ridge-residual-operators}
\mathbf H=\mathbf I_N-
\mathbf S(\mathbf S^\top\mathbf S+\ell n\mathbf I_s)^{-1}\mathbf S^\top,
\qquad \mathbf Q=\mathbf H^2.
\end{equation}
The identity
$(\mathbf S\mathbf S^\top+\ell n\mathbf I_N)\mathbf S
=\mathbf S(\mathbf S^\top\mathbf S+\ell n\mathbf I_s)$ gives
$$
\begin{aligned}
(\mathbf S\mathbf S^\top+\ell n\mathbf I_N)\mathbf H
&=\mathbf S\mathbf S^\top+\ell n\mathbf I_N
-\mathbf S(\mathbf S^\top\mathbf S+\ell n\mathbf I_s)
 (\mathbf S^\top\mathbf S+\ell n\mathbf I_s)^{-1}\mathbf S^\top\\
&=\ell n\mathbf I_N,
\end{aligned}
$$
so
$$
\mathbf H=\ell n(\mathbf S\mathbf S^\top+\ell n\mathbf I_N)^{-1}.
$$

Let $\mathbf X_c^{(0)}=\mathbf X_c$, and let $\mathbf X_c^{(v)}$,
$1\leq v\leq M$, be the randomised target matrices in the same coordinates.
With $a_n=\lambda_n/n$, the residual score vectors and regularised residual
Gram matrices are
\begin{equation}\label{gcrtint:eq:centered-residual-scores}
 \bs c_v=\mathbf X_c^{(v)\top}\mathbf Q\bs Y_c/n,\qquad
 \mathbf D_v=\mathbf X_c^{(v)\top}\mathbf Q\mathbf X_c^{(v)}/n+a_n\mathbf I_t.
\end{equation}
The definitions in \eqref{gcrtint:eq:centered-residual-scores}
express the quantities in \eqref{gcrtcent:eq:residual-scores}
after removing the intercept.
The response residual vector $\bs Y_R$ is represented by
$\mathbf H\bs Y_c$ in these orthonormal coordinates. Hence
\begin{equation}\label{gcrtb:eq:response-norm}
 \|\bs Y_R\|^2=\|\mathbf H\bs Y_c\|^2
              =\bs Y_c^\top\mathbf Q\bs Y_c.
\end{equation}

Let $\widehat{\mathbf B}_n=\mathbf S^\top\mathbf Q\mathbf S/n\in\mathbb R^{s\times s}$
be the Gram matrix of the residualised background covariates.

The following bounds hold for every background matrix $\mathbf S$.
For a squared singular value $x$ of $\mathbf S/\sqrt n$, the corresponding
eigenvalue of $\widehat{\mathbf B}_n$ is $\ell^2x/(\ell+x)^2\leq\ell/4$,
since $(\ell+x)^2\geq4\ell x$.
The corresponding singular value of $\mathbf H\mathbf S$ is
$\sqrt n\,\ell\sqrt{x}/(\ell+x)\leq\sqrt{\ell n}/2$.
The eigenvalues of $\mathbf H$ are $\ell/(\ell+x)\in(0,1]$,
including the value one on the null space of $\mathbf S^\top$.
Since $\mathbf Q=\mathbf H^2$, these estimates give
\begin{equation}\label{gcrtint:eq:ridge-bounds}
 0\preceq\widehat{\mathbf B}_n\preceq(\ell/4)\mathbf I_s,\qquad
 \|\mathbf H\mathbf S\|_{\mathrm{op}}\leq\sqrt{\ell n}/2,
 \qquad 0\prec\mathbf Q\preceq\mathbf I_N.
\end{equation}

\medskip
\textbf{Population quantities and bounds.}
For deterministic approximations to $\bs c_v$ and $\mathbf D_v$ in
\eqref{gcrtint:eq:centered-residual-scores},
$v=0,\ldots,M$, we average over $\mathbf S$.
Define
\begin{equation}\label{gcrtb:eq:population}
\begin{aligned}
 \mathbf B_n&=\mathbb E\widehat{\mathbf B}_n,
 & \bs m_n&=\mathbf C\mathbf B_n\bs\gamma,\\
 \tau &=\operatorname{tr}(\mathbf Q)/n, 
&\bar\tau_n&=\mathbb E\tau,
\qquad \bs u_n =\bar\tau_n\mathbf V\bs\beta_{\mathcal T},\\
 \bar{\mathbf D}_n&=\bar\tau_n\mathbf V+
       \mathbf C\mathbf B_n\mathbf C^\top+(\lambda_n/n)\mathbf I_t.
\end{aligned}
\end{equation}
The vectors $\bs m_n$ and $\bs u_n$ are respectively the mean background and
target signal contributions to the observed residual score vector $\bs c_0$. The bound \eqref{gcrtb:eq:gram-approximation} in Lemma~\ref{gcrtb:lem:endpoints} uses $\bar{\mathbf D}_n$ to approximate
the regularised residual Gram matrices $\mathbf D_v$, $v=0,\ldots,M$.
For each $n$, these quantities are deterministic functions
of the population parameters and the predetermined penalties.

Since $\mathbf B_n=\mathbb E\widehat{\mathbf B}_n$, taking expectations in the first bound in
\eqref{gcrtint:eq:ridge-bounds}, we have 
\begin{equation}\label{gcrtb:eq:background-matrix-bound}
 0\preceq\mathbf B_n\preceq(\ell/4)\mathbf I_s.
\end{equation}
Also, the last bound in \eqref{gcrtint:eq:ridge-bounds} gives
$\tau=\operatorname{tr}(\mathbf Q)/n\leq N/n$.
Taking expectations, we obtain
\begin{equation}\label{gcrtb:eq:mean-trace-upper-bound}
 \bar\tau_n\leq N/n\leq1.
\end{equation}

For a lower bound on $\bar\tau_n$, take $n\geq2$ and define
$f(x)=\{\ell/(\ell+x)\}^2$ for $x\geq0$.
Let $x_1,\ldots,x_N$ be the eigenvalues of
$\mathbf S\mathbf S^\top/n$, including zeros.
Since $f''(x)=6\ell^2/(\ell+x)^4>0$, Jensen's inequality gives
$$
\tau=\frac1n\sum_{i=1}^N f(x_i)
\geq\frac Nn f\left(\frac1N\sum_{i=1}^N x_i\right)
=\frac Nn f\left(\frac{\operatorname{tr}(\mathbf S\mathbf S^\top)}{nN}\right).
$$
The row covariance in \eqref{gcrtint:eq:centered-model} gives
$\mathbb E\operatorname{tr}(\mathbf S\mathbf S^\top)
=\sum_{i=1}^N\mathbb E\|\bs S_{i\cdot}\|^2
=N\operatorname{tr}(\mathbf\Sigma_{\mathcal S\mathcal S})$.
Applying Jensen's inequality again over $\mathbf S$, we have
$\bar\tau_n\geq(N/n)f\{\operatorname{tr}(\mathbf\Sigma_{\mathcal S\mathcal S})/n\}
\geq(N/n)f(c_+s_+)$, where the last inequality is because $f$ is decreasing and $\operatorname{tr}(\mathbf\Sigma_{\mathcal S\mathcal S})\leq c_+s_+n$ from 
\eqref{gcrtcent:eq:population-conditions}. 
For $n\geq2$, since $N/n=1-1/n\geq1/2$, we obtain
\begin{equation}\label{gcrtb:eq:mean-trace-lower-bound}
 \bar\tau_n\geq f(c_+s_+)/2=:\tau_0,
\end{equation}
where $\tau_0>0$ depends only on the fixed constants $\ell$, $c_+$,
and $s_+$.

\medskip
\textbf{Gaussian quadratic form tail bound.}

The following bound controls Gaussian quadratic forms in the Gram matrix
and score approximation proofs. We adapt the moment generating function
argument in the proof of \citet[Lemma~1]{10.1214/aos/1015957395}
to allow eigenvalues of either sign.

\begin{lemma}\label{gcrtb:lem:quadratic-tail}
  Let $q\geq1$ be an integer, let $\mathbf B\in\mathbb R^{q\times q}$
be symmetric, and let $\bs Z\sim\mathcal N_q(\bs0,\mathbf I_q)$.
For every $x>0$,
\begin{equation}\label{gcrtb:eq:quadratic-tail}
 \mathbb P\left(|\bs Z^\top\mathbf B\bs Z-\operatorname{tr}(\mathbf B)|
 >2\|\mathbf B\|_F\sqrt x+2\|\mathbf B\|_{\mathrm{op}}x\right)\leq2e^{-x}.
\end{equation}
\end{lemma}
  
\begin{proof}
The claim is immediate if $\mathbf B=\mathbf0$.
Suppose $\mathbf B\ne\mathbf0$ and define
$T=\bs Z^\top\mathbf B\bs Z-\operatorname{tr}(\mathbf B)$.
Let $\lambda_1,\ldots,\lambda_q$ be the eigenvalues of $\mathbf B$.
For a real parameter $\upsilon$, orthogonal diagonalisation gives
\begin{equation}\label{gcrtb:eq:quadratic-log-mgf}
\begin{aligned}
\log\mathbb E e^{\upsilon T}
&=\sum_{j=1}^q
  \left\{-\upsilon\lambda_j-\frac12\log(1-2\upsilon\lambda_j)\right\}\\
&\leq
\frac{\upsilon^2\|\mathbf B\|_F^2}
     {1-2|\upsilon|\|\mathbf B\|_{\mathrm{op}}},
\qquad
|\upsilon|<\frac{1}{2\|\mathbf B\|_{\mathrm{op}}},
\end{aligned}
\end{equation}
where the inequality follows by applying
$-u-\tfrac12\log(1-2u)\leq u^2/(1-2|u|)$, valid for $|u|<1/2$, to
$u=\upsilon\lambda_j$ and using
$|\lambda_j|\leq\|\mathbf B\|_{\mathrm{op}}$ and
$\sum_j\lambda_j^2=\|\mathbf B\|_F^2$.
For $x>0$, choose
$\upsilon_x=\sqrt{x}/(\|\mathbf B\|_F+
2\|\mathbf B\|_{\mathrm{op}}\sqrt{x})$.
Since $\mathbf B\ne\mathbf0$, we have
$0<\upsilon_x<1/(2\|\mathbf B\|_{\mathrm{op}})$, so
\eqref{gcrtb:eq:quadratic-log-mgf} applies at both
$\upsilon=\upsilon_x$ and $\upsilon=-\upsilon_x$.
For either choice of sign, Markov's inequality applied to
$e^{\pm\upsilon_xT}$ gives 
$$
\begin{aligned}
& \mathbb P\left(\pm T>
2\|\mathbf B\|_F\sqrt{x}+2\|\mathbf B\|_{\mathrm{op}}x\right)\\
\leq{}&
e^{-\upsilon_x(2\|\mathbf B\|_F\sqrt{x}+2\|\mathbf B\|_{\mathrm{op}}x)}
\mathbb E e^{\pm\upsilon_xT}\\
\leq{}&
\exp\left\{
-\upsilon_x(2\|\mathbf B\|_F\sqrt{x}+2\|\mathbf B\|_{\mathrm{op}}x)
+\frac{\upsilon_x^2\|\mathbf B\|_F^2}
 {1-2\upsilon_x\|\mathbf B\|_{\mathrm{op}}}
\right\}
=e^{-x},
\end{aligned}
$$
where the second inequality uses
\eqref{gcrtb:eq:quadratic-log-mgf}, and the final equality follows
by substituting the chosen $\upsilon_x$.
The union bound for the two tails yields \eqref{gcrtb:eq:quadratic-tail}.
\end{proof}

\subsubsection{Proof of Theorem~\ref{thm:gcrt-centered-power}}\label{app:gcrt-centered-proof}

To describe the leading target signal contribution to the randomised
score vectors, we introduce matrices determined by the prescribed
update order.
Relabel the target variables as $1,\ldots,t$ in their predetermined update order.
Let $\mathbf J_<,\mathbf J_>$ be the strictly lower triangular and strictly upper triangular parts of
$\mathbf J$. Define
\begin{equation}\label{gcrtint:eq:scan}
 \mathbf F_{\to}=(\mathbf I_t-\mathbf J_<)^{-1}\mathbf J_>,\qquad
 \mathbf F_{\leftarrow}=(\mathbf I_t-\mathbf J_>)^{-1}\mathbf J_<,\qquad
 \mathbf F=\mathbf F_{\leftarrow}^{L}\mathbf F_{\to}^{L}.
\end{equation}
In \eqref{gcrtint:eq:scan}, the rightmost factor
$\mathbf F_{\to}^{L}$ corresponds to the $L$ forward sweeps to the hub,
followed by $\mathbf F_{\leftarrow}^{L}$ for the $L$ reverse sweeps
along each branch.

The bounds \eqref{gcrtb:eq:gram-approximation}--\eqref{gcrtb:eq:randomised-score-approximation} in Lemma~\ref{gcrtb:lem:endpoints} approximate the residual score vectors
$\bs c_v$ and the regularised residual Gram matrices $\mathbf D_v$ defined in
\eqref{gcrtint:eq:centered-residual-scores}.
The leading background contribution $\bs m_n$ is common to all candidates;
the target signal contributions are $\bs u_n$ for the observation
and $\mathbf F\bs u_n$ for the randomised copies.
Lemma~\ref{gcrtcent:lem:scan-displacement} bounds the norm of
$(\mathbf I_t-\mathbf F)\bs u_n$ from below.
After centring removes $\bs m_n$, these bounds imply
$T_0^{\rm cen}>T_m^{\rm cen}$ for all $m=1,\ldots,M$
with probability tending uniformly to one.

\begin{lemma}\label{gcrtb:lem:endpoints}
Under the Gaussian model
\eqref{gcrtint:eq:response}, assume the population conditions
\eqref{gcrtcent:eq:population-conditions} and use the graph and sampling
specifications in Appendix~\ref{app:gcrt-common-conditions}, with fixed
$\ell,K>0$, integers $L,M\geq1$, and a deterministic sequence
$0<\lambda_n\leq Kn$.
Suppose \eqref{gcrts:eq:general-rates} holds.
Then, uniformly over $\mathcal C_n(h_n)$,
\begin{subequations}\label{gcrtb:eq:population-expansion}
\begin{align}
 \max_{0\leq v\leq M}\|\mathbf D_v-\bar{\mathbf D}_n\|_{\mathrm{op}}
   &=O_{\mathbb P}(e_n),
 \label{gcrtb:eq:gram-approximation}\\
 \bs c_0&=\bs m_n+\bs u_n+O_{\mathbb P}(\sigma\eta_n),
 \label{gcrtb:eq:observed-score-approximation}\\
 \max_{1\leq m\leq M}\|\bs c_m-\bs m_n-\mathbf F\bs u_n\|
   &=O_{\mathbb P}(\sigma\eta_n).
 \label{gcrtb:eq:randomised-score-approximation}
\end{align}
Moreover, define $C_{\mathrm{inv}}=2/(\tau_0v_-)$, with $\tau_0$ as in
\eqref{gcrtb:eq:mean-trace-lower-bound}.
This constant is independent of $n$ and the population
$\vartheta\in\mathcal C_n(h_n)$, and
\begin{equation}\label{gcrtb:eq:inverse-gram-bound}
 \inf_{\vartheta\in\mathcal C_n(h_n)}
 \mathbb P_\vartheta\!\left(
 \max_{0\leq v\leq M}\|\mathbf D_v^{-1}\|_{\mathrm{op}}\leq C_{\mathrm{inv}}
 \right)\longrightarrow1.
\end{equation}
\end{subequations}
\end{lemma}

The proof of Lemma~\ref{gcrtb:lem:endpoints} is given in
Appendices~\ref{app:gcrt-residual-approximation-proof} and~\ref{app:gcrt-score-approximation-proof}.

The leading target signal contributions in
\eqref{gcrtb:eq:observed-score-approximation} and \eqref{gcrtb:eq:randomised-score-approximation} are $\bs u_n$ for the observation
and $\mathbf F\bs u_n$ for the randomised copies.
The following lemma bounds the norm of their difference
$(\mathbf I_t-\mathbf F)\bs u_n$ from below, as needed in the proof of
Theorem~\ref{thm:gcrt-centered-power}.
For a positive definite matrix $\mathbf D\in\mathbb R^{t\times t}$ and
$\bs z\in\mathbb R^t$, define
$\|\bs z\|_{\mathbf D}^2=\bs z^\top\mathbf D\bs z$.
Let $\|\mathbf F\|_{\mathbf\Theta}$ be the operator norm induced by
$\|\cdot\|_{\mathbf\Theta}$.

\begin{lemma}\label{gcrtcent:lem:scan-displacement}
Let $\mathbf\Theta$ be
symmetric positive definite, with $J_{jj}=0$ and
$J_{jk}=-\Theta_{jk}/\Theta_{jj}$ for $j\ne k$.
Assume $\|\mathbf J\|_\infty\leq a<1$, and let $L\geq1$ be an integer.
For $\mathbf F$ defined in \eqref{gcrtint:eq:scan},
$\mathbf F$ is self-adjoint and positive semidefinite in the
$\mathbf\Theta$ inner product, and
\begin{equation}\label{gcrtcent:eq:scan-displacement}
 \|\mathbf F\|_{\mathbf\Theta}\leq a^{2L},\qquad
 \|(\mathbf I_t-\mathbf F)\bs x\|_{\mathbf\Theta}
 \geq(1-a^{2L})\|\bs x\|_{\mathbf\Theta}.
\end{equation}
\end{lemma}
The proof of Lemma~\ref{gcrtcent:lem:scan-displacement} is given in Appendix~\ref{app:gcrt-auxiliary-proofs}.

\begin{proof}[Proof of Theorem~\ref{thm:gcrt-centered-power}]

We first establish bounds shared by both choices of weight, then apply
them to the unit weight and the inverse pooled Gram weight.
All stochastic orders and probability limits below are uniform over
$\mathcal C_n(h_n)$.
Recall that $\eta_n=e_n+d_nr$ is defined in
\eqref{gcrtb:eq:error-rates}.
Under the assumptions of Theorem~\ref{thm:gcrt-centered-power},
equations~\eqref{gcrtb:eq:observed-score-approximation} and \eqref{gcrtb:eq:randomised-score-approximation}
in Lemma~\ref{gcrtb:lem:endpoints} give
$\bs c_0=\bs m_n+\bs u_n+\bs\delta_0$ and
$\bs c_m=\bs m_n+\mathbf F\bs u_n+\bs\delta_m$,
where $\max_v\|\bs\delta_v\|=O_{\mathbb P}(\sigma\eta_n)$.
Define
\begin{equation}\label{gcrtcent:eq:signal-displacement}
\bs d_n^*=(\mathbf I_t-\mathbf F)\bs u_n,
\end{equation}
and let
$\bar{\bs\delta}=(M+1)^{-1}\sum_v\bs\delta_v$. Centring gives the exact identities
\begin{equation}\label{gcrtcent:eq:centered-expansion}
 \bs c_0-\bar{\bs c}=\frac{M}{M+1}\bs d_n^*
                      +\bs\delta_0-\bar{\bs\delta},\qquad
 \bs c_m-\bar{\bs c}=-\frac1{M+1}\bs d_n^*
                      +\bs\delta_m-\bar{\bs\delta}.
\end{equation}
Let $\bs e_v=\bs\delta_v-\bar{\bs\delta}$, $v=0,\ldots,M$,
be the centred error vectors in \eqref{gcrtcent:eq:centered-expansion}.
For any common positive definite score weight $\mathbf W_*\in\mathbb R^{t\times t}$,
use $\mathbf W=\mathbf W_*$ in \eqref{eq:symmetric-centered-statistic}.
Write $\langle\bs x,\bs y\rangle_{\mathbf W_*}=\bs x^\top\mathbf W_*\bs y$
and define $\epsilon=\max_{0\leq v\leq M}\|\bs e_v\|_{\mathbf W_*}$,
the largest weighted norm of these centred error vectors.

Substitution of \eqref{gcrtcent:eq:centered-expansion} into
\eqref{eq:symmetric-centered-statistic} and expansion of the
two quadratic forms give
$$
\begin{aligned}
T_0^{\rm cen}-T_m^{\rm cen}
={}&\frac{M-1}{M+1}\|\bs d_n^*\|_{\mathbf W_*}^2
 +\frac{2M}{M+1}\langle\bs d_n^*,\bs e_0\rangle_{\mathbf W_*}
 +\frac{2}{M+1}\langle\bs d_n^*,\bs e_m\rangle_{\mathbf W_*}\\
&+\|\bs e_0\|_{\mathbf W_*}^2-\|\bs e_m\|_{\mathbf W_*}^2.
\end{aligned}
$$
By Cauchy--Schwarz,
$|\langle\bs d_n^*,\bs e_v\rangle_{\mathbf W_*}|
\leq\|\bs d_n^*\|_{\mathbf W_*}\epsilon$ for $v=0,m$.
The coefficients of these two inner products sum to two.
Moreover, $\|\bs e_0\|_{\mathbf W_*}^2$ and
$\|\bs e_m\|_{\mathbf W_*}^2$ both lie in $[0,\epsilon^2]$,
so their difference has absolute value at most $\epsilon^2$.
Therefore, for each $m\in [M]$, we have 
\begin{equation}\label{gcrtcent:eq:quadratic-gap-bound}
 \left|T_0^{\rm cen}-T_m^{\rm cen}
       -\frac{M-1}{M+1}\|\bs d_n^*\|_{\mathbf W_*}^2\right|
 \leq2\|\bs d_n^*\|_{\mathbf W_*}\epsilon+\epsilon^2.
\end{equation}

For the chosen weight $\mathbf W_*$, define the normalised error by
\begin{equation}\label{gcrtcent:eq:unit-first-normalised-error}
Z_n=\frac1{\sigma^2r^2}\max_{1\leq m\leq M}
\left|T_0^{\rm cen}-T_m^{\rm cen}
-\frac{M-1}{M+1}\|\bs d_n^*\|_{\mathbf W_*}^2\right|.
\end{equation}
The quantities $Z_n$ and $\epsilon$ depend on $\mathbf W_*$.
We next bound the signal and the errors in the Euclidean norm.

The definitions of $\bs u_n$ in \eqref{gcrtb:eq:population},
$\mathbf\Theta$ in \eqref{gcrtint:eq:target-precision}, and $r$ in
\eqref{gcrtcent:eq:signal-strength} give
\begin{equation}\label{gcrtcent:eq:signal-norm}
\begin{aligned}
\|\bs u_n\|_{\mathbf\Theta}^2
&=\bar\tau_n^2\bs\beta_{\mathcal T}^{\top}
  \mathbf V\mathbf V^{-1}\mathbf V\bs\beta_{\mathcal T}\\
&=\bar\tau_n^2\bs\beta_{\mathcal T}^{\top}
  \mathbf V\bs\beta_{\mathcal T}
 =\bar\tau_n^2\sigma^2r^2.
\end{aligned}
\end{equation}
The definition of $\bs d_n^*$ in \eqref{gcrtcent:eq:signal-displacement}
and Lemma~\ref{gcrtcent:lem:scan-displacement} with $\bs x=\bs u_n$ give
$\|\bs d_n^*\|_{\mathbf\Theta}\geq(1-a^{2L})\|\bs u_n\|_{\mathbf\Theta}$. 
Note that $\mathbf\Theta\preceq v_-^{-1}\mathbf I_t$ in view of 
\eqref{gcrtcent:eq:population-conditions} and
\eqref{gcrtint:eq:target-precision}. 
We can then obtain a lower bound for $\|\bs d_n^*\|^2$:
\begin{equation}\label{gcrtcent:eq:gap-lower-bound}
\begin{aligned}
\|\bs d_n^*\|^2
&\geq v_-\|\bs d_n^*\|_{\mathbf\Theta}^2
 \geq v_-(1-a^{2L})^2\bar\tau_n^2\sigma^2r^2\\
&\geq v_-\tau_0^2(1-a^{2L})^2\sigma^2r^2,
\end{aligned}
\end{equation}
where the second inequality uses $\|\bs d_n^*\|_{\mathbf\Theta}\geq(1-a^{2L})\|\bs u_n\|_{\mathbf\Theta}$ and \eqref{gcrtcent:eq:signal-norm}, and the third uses $\bar\tau_n\geq\tau_0>0$ from
\eqref{gcrtb:eq:mean-trace-lower-bound} for $n\geq2$.
We also need upper bounds on the signal norm and the centred errors.
The population conditions \eqref{gcrtcent:eq:population-conditions}
and \eqref{gcrtint:eq:target-precision} also imply
$\mathbf\Theta\succeq v_+^{-1}\mathbf I_t$.
We then obtain
\begin{equation}\label{gcrtcent:eq:unit-first-euclidean-controls}
\begin{aligned}
\|\bs d_n^*\|
&\leq\sqrt{v_+}\|\bs d_n^*\|_{\mathbf\Theta}\\
&\leq\sqrt{v_+}\bigl(\|\bs u_n\|_{\mathbf\Theta}
  +\|\mathbf F\bs u_n\|_{\mathbf\Theta}\bigr)\\
&\leq\sqrt{v_+}(1+a^{2L})\bar\tau_n\sigma r
 =O(\sigma r),
\end{aligned}
\end{equation}
and 
\begin{equation}\label{gcrtcent:eq:unit-first-centered-error-control}
\begin{aligned}
\max_{0\leq v\leq M}\|\bs e_v\|
&\leq\max_{0\leq v\leq M}\|\bs\delta_v\|
  +\|\bar{\bs\delta}\|\\
&\leq2\max_{0\leq v\leq M}\|\bs\delta_v\|
 =O_{\mathbb P}(\sigma\eta_n).
\end{aligned}
\end{equation}
In \eqref{gcrtcent:eq:unit-first-euclidean-controls}, the second inequality applies the triangle inequality to
$\bs d_n^*=(\mathbf I_t-\mathbf F)\bs u_n$, the third uses 
\eqref{gcrtcent:eq:signal-norm} and 
$\|\mathbf F\|_{\mathbf\Theta}\leq a^{2L}$ from Lemma~\ref{gcrtcent:lem:scan-displacement}, and the last order $O(\sigma r)$ follows because $a<1$ and 
$\bar\tau_n\leq1$ by \eqref{gcrtb:eq:mean-trace-upper-bound}.
For the centred errors in \eqref{gcrtcent:eq:unit-first-centered-error-control},
the first inequality applies the triangle inequality to
$\bs e_v=\bs\delta_v-\bar{\bs\delta}$, the second uses
$\|\bar{\bs\delta}\|\leq(M+1)^{-1}\sum_{v=0}^M\|\bs\delta_v\|
\leq\max_{0\leq v\leq M}\|\bs\delta_v\|$, and 
the last stochastic order follows from \eqref{gcrtb:eq:observed-score-approximation} and
\eqref{gcrtb:eq:randomised-score-approximation} in Lemma~\ref{gcrtb:lem:endpoints}.

Since $r\geq h_n>0$,
\eqref{gcrtb:eq:error-rates} and \eqref{gcrts:eq:general-rates} together yield
\begin{equation}\label{gcrtcent:eq:relative-error-rate}
\eta_n/r\leq e_n/h_n+d_n\to0.
\end{equation}
Let $\kappa=\{(M-1)/(M+1)\}v_-\tau_0^2(1-a^{2L})^2>0$.

\medskip 
\textbf{Part (1): Unit weight.}

Let $\mathbf W_*=\mathbf I_t$, so $\epsilon=\max_v\|\bs e_v\|$.
By \eqref{gcrtcent:eq:quadratic-gap-bound},
$Z_n\leq\{2\|\bs d_n^*\|\epsilon+\epsilon^2\}/(\sigma^2r^2)$.
The bounds \eqref{gcrtcent:eq:unit-first-euclidean-controls} and
\eqref{gcrtcent:eq:unit-first-centered-error-control} give
$$
\begin{aligned}
\frac{2\|\bs d_n^*\|\epsilon}{\sigma^2r^2}
&=O_{\mathbb P}\!\left(\frac{\sigma^2r\eta_n}{\sigma^2r^2}\right)
 =O_{\mathbb P}(\eta_n/r),\\
\frac{\epsilon^2}{\sigma^2r^2}
&=O_{\mathbb P}\!\left(\frac{\sigma^2\eta_n^2}{\sigma^2r^2}\right)
 =O_{\mathbb P}((\eta_n/r)^2).
\end{aligned}
$$
Consequently,
\begin{equation}\label{gcrtcent:eq:gap-expansion}
Z_n=O_{\mathbb P}(\eta_n/r+(\eta_n/r)^2).
\end{equation}
The limit in \eqref{gcrtcent:eq:relative-error-rate} therefore
gives $Z_n=o_{\mathbb P}(1)$.
On $\{Z_n\leq\kappa/2\}$, the definition of $Z_n$ in
\eqref{gcrtcent:eq:unit-first-normalised-error}, together with
\eqref{gcrtcent:eq:gap-lower-bound} and the definition of $\kappa$, gives
$T_0^{\rm cen}-T_m^{\rm cen}\geq(\kappa-Z_n)\sigma^2r^2
\geq\kappa\sigma^2r^2/2>0$
for every $m=1,\ldots,M$.
Hence $p_{\rm cen}=1/(M+1)\leq\alpha$ on this event, and
$$
\sup_{\vartheta\in\mathcal C_n(h_n)}
\mathbb P_\vartheta(p_{\rm cen}>\alpha)
\leq\sup_{\vartheta\in\mathcal C_n(h_n)}
\mathbb P_\vartheta(Z_n>\kappa/2)\longrightarrow0.
$$

\medskip 
\textbf{Part (2): Inverse pooled Gram weight.}

Let $\mathbf W_*=\mathbf D_{\rm pool}^{-1}$.
To transfer the Euclidean signal and error bounds to this weight,
we show that its eigenvalues lie between fixed positive constants
with probability tending uniformly to one.
For the lower bound on $\bar{\mathbf D}_n$, \eqref{gcrtb:eq:background-matrix-bound} gives
$\mathbf B_n\succeq0$, and
\eqref{gcrtb:eq:mean-trace-lower-bound} gives
$\bar\tau_n\geq\tau_0$ for $n\geq2$.
Since $\lambda_n>0$ and $\mathbf V\succeq v_-\mathbf I_t$ by
\eqref{gcrtcent:eq:population-conditions}, the definition in
\eqref{gcrtb:eq:population} yields
\begin{equation}\label{gcrtcent:eq:population-gram-lower-bound}
\bar{\mathbf D}_n\succeq\bar\tau_n\mathbf V
\succeq\bar\tau_n v_-\mathbf I_t
\succeq\tau_0v_-\mathbf I_t.
\end{equation}
For the upper bound on $\bar{\mathbf D}_n$, \eqref{gcrtb:eq:background-matrix-bound} gives
$\mathbf B_n\preceq(\ell/4)\mathbf I_s$, while
$\|\mathbf C\|_{\mathrm{op}}\leq\bar c$ in
\eqref{gcrtcent:eq:population-conditions} implies
$\mathbf C\mathbf C^\top\preceq\bar c^2\mathbf I_t$. Hence
$$
\mathbf C\mathbf B_n\mathbf C^\top
\preceq(\ell/4)\mathbf C\mathbf C^\top
\preceq(\ell\bar c^2/4)\mathbf I_t.
$$
Together with $\bar\tau_n\leq1$ from
\eqref{gcrtb:eq:mean-trace-upper-bound},
$\mathbf V\preceq v_+\mathbf I_t$ from the population conditions,
and $\lambda_n/n\leq K$ from the penalty assumption, this gives
\begin{equation}\label{gcrtcent:eq:population-gram-upper-bound}
\begin{aligned}
\bar{\mathbf D}_n
&=\bar\tau_n\mathbf V+\mathbf C\mathbf B_n\mathbf C^\top
  +(\lambda_n/n)\mathbf I_t\\
&\preceq\bigl(v_++\ell\bar c^2/4+K\bigr)\mathbf I_t
 =:D_+\mathbf I_t.
\end{aligned}
\end{equation}
Since $\mathbf D_{\rm pool}$ is the average in
\eqref{gcrtcent:eq:residual-scores},
\eqref{gcrtb:eq:gram-approximation} in Lemma~\ref{gcrtb:lem:endpoints} gives
$$
\|\mathbf D_{\rm pool}-\bar{\mathbf D}_n\|_{\mathrm{op}}
\leq\max_{0\leq v\leq M}\|\mathbf D_v-\bar{\mathbf D}_n\|_{\mathrm{op}}
=O_{\mathbb P}(e_n).
$$
As $e_n\to0$ uniformly, the event
$\mathcal E_n=\{\|\mathbf D_{\rm pool}-\bar{\mathbf D}_n\|_{\mathrm{op}}
\leq\tau_0v_-/2\}$ has probability tending uniformly to one.
On $\mathcal E_n$, the eigenvalue perturbation bound applied to
\eqref{gcrtcent:eq:population-gram-lower-bound} and
\eqref{gcrtcent:eq:population-gram-upper-bound} gives
$\frac12\tau_0v_-\mathbf I_t\preceq\mathbf D_{\rm pool}
\preceq2D_+\mathbf I_t$, and hence
\begin{equation}\label{gcrtcent:eq:unit-first-pooled-spectral-bounds}
\frac1{2D_+}\mathbf I_t\preceq\mathbf W_*
\preceq\frac2{\tau_0v_-}\mathbf I_t.
\end{equation}
Thus \eqref{gcrtcent:eq:gap-lower-bound} implies, on
$\mathcal E_n$,
\begin{equation}\label{gcrtcent:eq:unit-first-pooled-signal}
\frac{M-1}{M+1}\|\bs d_n^*\|_{\mathbf W_*}^2
\geq\frac{M-1}{(M+1)2D_+}\|\bs d_n^*\|^2
\geq\frac{\kappa}{2D_+}\sigma^2r^2.
\end{equation}
On $\mathcal E_n$, the upper bound in
\eqref{gcrtcent:eq:unit-first-pooled-spectral-bounds} gives
\begin{equation}\label{gcrtcent:eq:pooled-signal-norm-comparison}
\|\bs d_n^*\|_{\mathbf W_*}\leq\sqrt{2/(\tau_0v_-)}\|\bs d_n^*\|
\end{equation}
and
\begin{equation}\label{gcrtcent:eq:pooled-error-norm-comparison}
\epsilon\leq\sqrt{2/(\tau_0v_-)}\max_v\|\bs e_v\|.
\end{equation}
Since $\mathcal E_n$ has probability tending uniformly to one,
\eqref{gcrtcent:eq:pooled-signal-norm-comparison} and
\eqref{gcrtcent:eq:unit-first-euclidean-controls} imply
$\|\bs d_n^*\|_{\mathbf W_*}=O_{\mathbb P}(\sigma r)$,
while \eqref{gcrtcent:eq:pooled-error-norm-comparison} and
\eqref{gcrtcent:eq:unit-first-centered-error-control} imply
$\epsilon=O_{\mathbb P}(\sigma\eta_n)$.
The inequalities \eqref{gcrtcent:eq:pooled-signal-norm-comparison}
and \eqref{gcrtcent:eq:pooled-error-norm-comparison} hold pointwise
on $\mathcal E_n$ and require no independence between
$\mathbf W_*$ and the errors $(\bs e_v)_{v=0}^M$.

With $\|\bs d_n^*\|_{\mathbf W_*}=O_{\mathbb P}(\sigma r)$
and $\epsilon=O_{\mathbb P}(\sigma\eta_n)$,
\eqref{gcrtcent:eq:quadratic-gap-bound} gives
\eqref{gcrtcent:eq:gap-expansion} for $Z_n$ in
\eqref{gcrtcent:eq:unit-first-normalised-error}, now with
$\mathbf W_*=\mathbf D_{\rm pool}^{-1}$.
By \eqref{gcrtcent:eq:relative-error-rate}, we again obtain
$Z_n=o_{\mathbb P}(1)$.
On $\mathcal E_n\cap\{Z_n\leq\kappa/(4D_+)\}$,
the definition of $Z_n$ in
\eqref{gcrtcent:eq:unit-first-normalised-error}, together with
\eqref{gcrtcent:eq:unit-first-pooled-signal}, gives, for every
$m=1,\ldots,M$,
$$
\begin{aligned}
T_0^{\rm cen}-T_m^{\rm cen}
&\geq\frac{M-1}{M+1}\|\bs d_n^*\|_{\mathbf W_*}^2
      -Z_n\sigma^2r^2\\
&\geq\left(\frac{\kappa}{2D_+}-Z_n\right)\sigma^2r^2\\
&\geq\kappa\sigma^2r^2/(4D_+)>0.
\end{aligned}
$$
Hence
$p_{\rm cen}=1/(M+1)\leq\alpha$.
Consequently,
$$
\begin{aligned}
\sup_{\vartheta\in\mathcal C_n(h_n)}
\mathbb P_\vartheta(p_{\rm cen}>\alpha)
&\leq\sup_{\vartheta\in\mathcal C_n(h_n)}
\mathbb P_\vartheta(\mathcal E_n^c)
+\sup_{\vartheta\in\mathcal C_n(h_n)}
\mathbb P_\vartheta(Z_n>\kappa/(4D_+))\\
&\longrightarrow0.
\end{aligned}
$$

\end{proof}

The next two subsections prove Lemma~\ref{gcrtb:lem:endpoints}.
We first establish the Gram approximation
\eqref{gcrtb:eq:gram-approximation} and derive the inverse Gram bound
\eqref{gcrtb:eq:inverse-gram-bound}.
We then prove the score approximations
\eqref{gcrtb:eq:observed-score-approximation} and
\eqref{gcrtb:eq:randomised-score-approximation} using the same
oracle Gaussian construction.
Throughout these proofs, we use the notation and uniform conditions
in Appendix~\ref{app:gcrt-common-conditions}.
Each auxiliary lemma used in these proofs holds uniformly over
Gaussian models satisfying its own stated conditions, with common
constants.

\subsubsection[\texorpdfstring{Proofs of \eqref{gcrtb:eq:gram-approximation} and \eqref{gcrtb:eq:inverse-gram-bound}}{Proofs of (\ref*{gcrtb:eq:gram-approximation}) and (\ref*{gcrtb:eq:inverse-gram-bound})}]{Proofs of \eqref{gcrtb:eq:gram-approximation} and \eqref{gcrtb:eq:inverse-gram-bound}}\label{app:gcrt-residual-approximation-proof}\label{app:gcrt-inverse-gram-proof}

To prove \eqref{gcrtb:eq:gram-approximation}, we first obtain bounds
conditional on $\mathbf S$ using Gaussian comparison copies, transfer
these bounds to the actual candidates, and then replace the conditional
centre by the population matrix $\bar{\mathbf D}_n$.

The normalisation in the residual rotation \eqref{eq:ggm-local-update}
is nonlinear, so the actual copies and the observation need not be
jointly Gaussian conditional on $\mathbf S$.
We use an oracle Gaussian construction for both the Gram
approximation here and the score approximations
\eqref{gcrtb:eq:observed-score-approximation} and
\eqref{gcrtb:eq:randomised-score-approximation} below.
Its comparison copies are jointly Gaussian conditional on $\mathbf S$.

We now define the oracle Gaussian construction using population
parameters and couple it to the residual rotation construction specified
by Algorithm~\ref{alg: exchangeable} and \eqref{eq:ggm-local-update}.

Let $\mathbf M=\mathbf S\mathbf C^\top$ be the conditional mean of
the original target matrix $\mathbf X_c^{(0)}$ given $\mathbf S$.
Write $\overline{\mathbf E}\in\mathbb R^{N\times t}$ for the current
comparison residual matrix.

For target $j$, let $\mathbf P_j$ be the orthogonal projection onto the
background columns indexed by $N_j\cap\mathcal S$, and define
$\mathbf R_j=\mathbf I_N-\mathbf P_j$ and $\theta_j=\Theta_{jj}$.
These matrices are fixed conditional on $\mathbf S$.
At each update of column $j$ in the oracle construction, use
\begin{equation}\label{gcrtb:eq:gaussian-update}
 \overline{\mathbf E}'_j=\mathbf P_j\overline{\mathbf E}_j+
 \mathbf R_j\sum_{k\ne j}J_{jk}\overline{\mathbf E}_k+
 \theta_j^{-1/2}\mathbf R_j\bs Z_j,
 \qquad \bs Z_j\sim \mathcal{N}_N(0,\mathbf I_N).
\end{equation}
This update leaves $\mathbf P_j\overline{\mathbf E}_j$ unchanged,
and all columns other than $j$ are unchanged.
At each update, use $\bs Z_j=\mathbf U^\top\bs W$
in \eqref{gcrtb:eq:gaussian-update}, where $\bs W\sim\mathcal N_n(0,\mathbf I_n)$
is the fresh vector drawn for the corresponding actual update in
\eqref{eq:ggm-local-update}, and $\mathbf U$ is the fixed matrix of
centred coordinates defined in Appendix~\ref{app:gcrt-common-conditions}.
The two updates therefore use the same vector in centred coordinates.
Since $\mathbf U^\top\mathbf U=\mathbf I_N$, each occurrence of
$\bs Z_j$ is a fresh standard Gaussian vector, independent of the
data and all earlier update randomness.

Let $\overline{\mathbf E}^{(0)}=\mathbf E$.
Starting from this matrix, apply
\eqref{gcrtb:eq:gaussian-update} in forward order for $L$ sweeps
to obtain the comparison hub.
Each of the $M$ branches starts from this hub and applies the updates
in reverse order for $L$ sweeps, with independent update randomness
across branches.
Let $\overline{\mathbf E}^{(m)}$ be the terminal residual matrix of
branch $m$, for $m=1,\ldots,M$.
Define the comparison copies of the oracle construction by
\begin{equation}\label{gcrtb:eq:comparison-copies}
 \overline{\mathbf X}^{(v)}=\mathbf M+\overline{\mathbf E}^{(v)},
 \qquad v=0,\ldots,M.
\end{equation}
In particular, $\overline{\mathbf X}^{(0)}=\mathbf X_c^{(0)}$
by \eqref{gcrtint:eq:centered-model}.

The following lemma gives the approximation bound and the conditional
distributional properties of the comparison residual matrices.

\begin{lemma}[Gaussian comparison]\label{gcrtb:lem:gaussian-comparison}
Assume $\mathbf\Sigma_{\mathcal S\mathcal S}\succ0$,
$0\prec\mathbf V\preceq v_+\mathbf I_t$ and
$\max\{\|\mathbf J\|_1,\|\mathbf J\|_\infty\}\leq a<1$.
Suppose $\max_j(b_j+g_j)<N=n-1$.
For this construction, we have
\begin{equation}\label{gcrtint:eq:coupling}
 \max_{0\leq v\leq M}
 \|\mathbf X_c^{(v)}-\overline{\mathbf X}^{(v)}\|_F/\sqrt n
 =O_{\mathbb P}(e_n).
\end{equation}
Conditional on $\mathbf S$, each residual matrix
$\overline{\mathbf E}^{(v)}$, $v=0,\ldots,M$, has independent
$\mathcal N_t(0,\mathbf V)$ rows, and the collection
$(\overline{\mathbf E}^{(v)})_{v=0}^M$ is jointly Gaussian
and independent of $\bs\varepsilon$.
\end{lemma}

Conditional on $\mathbf S$, the update \eqref{gcrtb:eq:gaussian-update}
is linear in the current comparison residuals and the fresh Gaussian
vector. Since the initial residual matrix $\mathbf E$ is Gaussian
conditional on $\mathbf S$ by \eqref{gcrtint:eq:centered-model},
the resulting comparison residual matrices are jointly Gaussian under
the same conditioning.
The proof of Lemma~\ref{gcrtb:lem:gaussian-comparison} in
Appendix~\ref{app:gcrt-gaussian-comparison-proof} also verifies the
conditional row distributions and controls the discrepancy introduced
at each coupled update and its accumulation over the sweeps and
branches to obtain \eqref{gcrtint:eq:coupling}.

\medskip 

The final step of the Gram approximation replaces the random conditional
centre by $\bar{\mathbf D}_n$.
The following lemma provides the required bounds on
$\mathbf C(\widehat{\mathbf B}_n-\mathbf B_n)\mathbf C^\top$ and
$\tau-\bar\tau_n$. It also provides a vector bound for the
score approximations \eqref{gcrtb:eq:observed-score-approximation} and
\eqref{gcrtb:eq:randomised-score-approximation} in Appendix~\ref{app:gcrt-score-approximation-proof}.

\begin{lemma}[Concentration for ridge regression]\label{gcrtb:lem:ridge-concentration}
Let $\mathbf S$ have independent
$\mathcal N_s(0,\mathbf\Sigma_{\mathcal S\mathcal S})$ rows as in
\eqref{gcrtint:eq:centered-model}, and fix $\ell>0$.
Assume the bounds $1\leq s\leq c_+n$,
$\mathbf\Sigma_{\mathcal S\mathcal S}\preceq s_+\mathbf I_s$,
$\|\mathbf C\|_{\mathrm{op}}\leq\bar c$ and
$\|\bs\gamma\|\leq G_+\sigma$ from
\eqref{gcrtcent:eq:population-conditions}, with $\sigma>0$.
The random quantities
$\mathbf C\widehat{\mathbf B}_n\bs\gamma$, $\mathbf C\widehat{\mathbf B}_n\mathbf C^\top$, and $\tau$
satisfy
\begin{equation}\label{gcrtb:eq:ridge-concentration}
\begin{aligned}
 \|\mathbf C\widehat{\mathbf B}_n\bs\gamma-\bs m_n\|
 &=O_{\mathbb P}\left(\sigma\sqrt{t/n}\right),\\
 \|\mathbf C(\widehat{\mathbf B}_n-\mathbf B_n)\mathbf C^\top\|_{\mathrm{op}}
 &=O_{\mathbb P}\left(\sqrt{t/n}\right),
\end{aligned}
\end{equation}
and
\begin{equation}\label{gcrtb:eq:trace-concentration}
 \tau-\bar\tau_n=O_{\mathbb P}(n^{-1}).
\end{equation}
\end{lemma}

The proof of Lemma~\ref{gcrtb:lem:ridge-concentration} is given in Appendix~\ref{app:gcrt-ridge-concentration-proof}.

\begin{proof}[Proof of \eqref{gcrtb:eq:gram-approximation}]
We first check the residual dimension condition of
Lemma~\ref{gcrtb:lem:gaussian-comparison}.
Since $h_n\leq R$, \eqref{gcrts:eq:general-rates} implies
$e_n\to0$, $d_n\to0$, and $t/n\leq e_n^2\to0$, uniformly.
By \eqref{gcrtb:eq:error-rates},
$\max_j(b_j+g_j)/n\leq d_n+e_n^2\to0$.
Consequently, $\max_j(b_j+g_j)<N=n-1$ throughout the parameter class
for all sufficiently large $n$. The remaining assumptions follow
from the population bounds and sampling specifications.

We proceed in four steps: calculate the conditional means of the
comparison Gram matrices, control their fluctuations, transfer the
bounds to the actual candidates, and replace the random conditional
centre by the population matrix $\bar{\mathbf D}_n$.

\medskip
\noindent\textbf{Step 1: Conditional means of the comparison Gram matrices.}

For the comparison copies in \eqref{gcrtb:eq:comparison-copies},
define the regularised residual Gram matrices, as in
\eqref{gcrtint:eq:centered-residual-scores}, by
\begin{equation}\label{gcrtb:eq:comparison-gram-definition}
\widetilde{\mathbf D}_v=\overline{\mathbf X}^{(v)\top}\mathbf Q
\overline{\mathbf X}^{(v)}/n+a_n\mathbf I_t,
\end{equation}
where $a_n=\lambda_n/n$.

Substitute \eqref{gcrtb:eq:comparison-copies} into
\eqref{gcrtb:eq:comparison-gram-definition} and expand:
\begin{equation}\label{gcrtb:eq:comparison-gram-expansion}
\begin{aligned}
\widetilde{\mathbf D}_v={}&\mathbf C\widehat{\mathbf B}_n\mathbf C^\top+a_n\mathbf I_t
 +\mathbf M^\top\mathbf Q\overline{\mathbf E}^{(v)}/n
 +\overline{\mathbf E}^{(v)\top}\mathbf Q\mathbf M/n\\
 &+\overline{\mathbf E}^{(v)\top}\mathbf Q\overline{\mathbf E}^{(v)}/n.
\end{aligned}
\end{equation}

Conditional on $\mathbf S$, the matrices $\mathbf M$ and
$\mathbf Q$ are fixed, and $\overline{\mathbf E}^{(v)}$ has independent
$\mathcal N_t(0,\mathbf V)$ rows by
Lemma~\ref{gcrtb:lem:gaussian-comparison}.
The terms $\mathbf M^\top\mathbf Q\overline{\mathbf E}^{(v)}/n$
and $\overline{\mathbf E}^{(v)\top}\mathbf Q\mathbf M/n$
therefore have conditional mean zero.
For $1\leq j,k\leq t$, independence of the centred rows gives

\begin{equation}\label{gcrtb:eq:residual-second-moment}
\begin{aligned}
\mathbb E\bigl([\overline{\mathbf E}^{(v)\top}\mathbf Q
 \overline{\mathbf E}^{(v)}]_{jk}\mid\mathbf S\bigr)
 &=\sum_{i,r=1}^N Q_{ir}\mathbb E(
 \overline E^{(v)}_{ij}\overline E^{(v)}_{rk}\mid\mathbf S)\\
 &=\sum_{i=1}^N Q_{ii}V_{jk}
 =\operatorname{tr}(\mathbf Q)V_{jk}.
\end{aligned}
\end{equation}

Since $\tau=\operatorname{tr}(\mathbf Q)/n$,
\eqref{gcrtb:eq:residual-second-moment} yields
$\mathbb E(\widetilde{\mathbf D}_v\mid\mathbf S)=\mathbf D_\lambda$
for $v=0,\ldots,M$, where
\begin{equation}\label{gcrtb:eq:conditional-gram-centre}
\mathbf D_\lambda
=\tau\mathbf V+\mathbf C\widehat{\mathbf B}_n\mathbf C^\top+a_n\mathbf I_t.
\end{equation}

\medskip
\noindent\textbf{Step 2: Conditional fluctuations of the comparison Gram matrices.}

We combine the Gaussian quadratic form bound in
Lemma~\ref{gcrtb:lem:quadratic-tail} with the following deterministic
inequalities for matrix and vector norms.
Let $\mathcal N$ be a deterministic $1/4$-net of the
unit sphere in $\mathbb R^t$ with $|\mathcal N|\leq9^t$.
For every symmetric $\mathbf B\in\mathbb R^{t\times t}$ and every
$\bs b\in\mathbb R^t$, the standard net inequalities give

\begin{equation}\label{gcrtb:eq:net-norm-bounds}
\|\mathbf B\|_{\mathrm{op}}\leq2\max_{\bs z\in\mathcal N}
 |\bs z^\top\mathbf B\bs z|,\qquad
\|\bs b\|\leq\frac43\max_{\bs z\in\mathcal N}|\bs z^\top\bs b|.
\end{equation}

The inequalities in \eqref{gcrtb:eq:net-norm-bounds} hold for every
such matrix $\mathbf B$ and vector $\bs b$; we use them here and
in the later concentration arguments.

\medskip

We now apply Lemma~\ref{gcrtb:lem:quadratic-tail} and
\eqref{gcrtb:eq:net-norm-bounds} to control the conditional
fluctuations of the comparison Gram matrices.
It suffices to control
$\bs z^\top(\widetilde{\mathbf D}_v-\mathbf D_\lambda)\bs z$
over the finite net $\mathcal N$.
We first obtain a tail bound for a fixed unit vector
$\bs z\in\mathbb R^t$ and candidate $v=0,\ldots,M$ by expanding
$\bs z^\top\widetilde{\mathbf D}_v\bs z$.

The expansion \eqref{gcrtb:eq:comparison-gram-expansion} and
the definition in \eqref{gcrtb:eq:conditional-gram-centre} give
\begin{equation}\label{gcrtb:eq:directional-gram-fluctuation}
\begin{aligned}
\bs z^\top(\widetilde{\mathbf D}_v-\mathbf D_\lambda)\bs z
={}&\underbrace{\frac1n\Bigl\{
(\overline{\mathbf E}^{(v)}\bs z)^\top\mathbf Q
(\overline{\mathbf E}^{(v)}\bs z)
-(\bs z^\top\mathbf V\bs z)\operatorname{tr}(\mathbf Q)
\Bigr\}}_{\text{centred quadratic term}}\\
&+\underbrace{\frac2n(\mathbf M\bs z)^\top\mathbf Q
(\overline{\mathbf E}^{(v)}\bs z)}_{\text{linear term}},
\end{aligned}
\end{equation}
where, conditional on $\mathbf S$, the first term is a centred
quadratic form and the second is a linear form in
$\overline{\mathbf E}^{(v)}\bs z$.
The two cross terms combine because $\mathbf Q$ is symmetric.

We standardise the
Gaussian vector in \eqref{gcrtb:eq:directional-gram-fluctuation}: Lemma~\ref{gcrtb:lem:gaussian-comparison} gives
\begin{equation}\label{gcrtb:eq:comparison-directional-gaussian-law}
\overline{\mathbf E}^{(v)}\bs z\mid\mathbf S
\sim\mathcal N_N(\bs0,(\bs z^\top\mathbf V\bs z)\mathbf I_N).
\end{equation}
Since $\mathbf V\succ0$ and $\|\bs z\|=1$, define
$\bs g=\overline{\mathbf E}^{(v)}\bs z/
\sqrt{\bs z^\top\mathbf V\bs z}$;
then $\bs g\mid\mathbf S\sim\mathcal N_N(\bs0,\mathbf I_N)$.

For the centred quadratic term in
\eqref{gcrtb:eq:directional-gram-fluctuation}, 
substitution yields
$$
\begin{aligned}
&\frac{(\overline{\mathbf E}^{(v)}\bs z)^\top\mathbf Q
(\overline{\mathbf E}^{(v)}\bs z)
-(\bs z^\top\mathbf V\bs z)\operatorname{tr}(\mathbf Q)}{n}\\
={}&\bs g^\top\left(\frac{\bs z^\top\mathbf V\bs z}{n}
\mathbf Q\right)\bs g
-\operatorname{tr}\left(\frac{\bs z^\top\mathbf V\bs z}{n}\mathbf Q\right).
\end{aligned}
$$
Conditional on $\mathbf S$, the matrix $\mathbf Q$ is fixed and
symmetric, so Lemma~\ref{gcrtb:lem:quadratic-tail} applies with
$\bs Z=\bs g$ and
$\mathbf B=(\bs z^\top\mathbf V\bs z)\mathbf Q/n$.
By \eqref{gcrtcent:eq:population-conditions},
$\bs z^\top\mathbf V\bs z\leq v_+$, while
\eqref{gcrtint:eq:ridge-bounds} gives
$\|\mathbf Q\|_{\mathrm{op}}\leq1$.
The coefficient matrix therefore has operator norm at most $v_+/n$.
Since $N=n-1$ and
$\|\mathbf Q\|_F\leq\sqrt N\,\|\mathbf Q\|_{\mathrm{op}}\leq\sqrt n$,
its Frobenius norm is at most $v_+/\sqrt n$.

Similarly, the linear term in
\eqref{gcrtb:eq:directional-gram-fluctuation}
has conditional variance
$4(\bs z^\top\mathbf V\bs z)\|\mathbf Q\mathbf M\bs z\|^2/n^2$,
which is at most $v_+\ell\bar c^2/n$ by
\eqref{gcrtint:eq:ridge-bounds}.

For each fixed unit vector $\bs z$, fixed candidate index
$v\in\{0,\ldots,M\}$, and $x>0$, condition on $\mathbf S$.
The Gaussian tail bound and Lemma~\ref{gcrtb:lem:quadratic-tail} bound
the absolute values of the linear and centred quadratic terms in
\eqref{gcrtb:eq:directional-gram-fluctuation} by
$\bar c\sqrt{2v_+\ell}\sqrt{x/n}$ and
$2v_+\{\sqrt{x/n}+x/n\}$, respectively, each with conditional
failure probability at most $2e^{-x}$.
When both bounds hold, the triangle inequality bounds the absolute
value of their sum by the sum of these two thresholds.
Define $C_{\mathrm{Gram}}=2v_++\bar c\sqrt{2v_+\ell}$.
A union bound over these two failure events gives
\begin{equation}\label{gcrtb:eq:directional-gram-tail}
\mathbb P\left(
\left|\bs z^\top(\widetilde{\mathbf D}_v-\mathbf D_\lambda)\bs z\right|
>C_{\mathrm{Gram}}\{\sqrt{x/n}+x/n\}\mid\mathbf S
\right)\leq4e^{-x}.
\end{equation}
The constant $C_{\mathrm{Gram}}$ depends only on the fixed bounds
$v_+$, $\ell$, and $\bar c$;
it is uniform in $\mathbf S$, $\bs z$, $v$, and the population.

We next make \eqref{gcrtb:eq:directional-gram-tail} simultaneous over
$\bs z\in\mathcal N$ and all candidate indices $v=0,\ldots,M$.
A second union bound covers these at most $(M+1)9^t$ pairs.
The matrix inequality in \eqref{gcrtb:eq:net-norm-bounds}, with
$\mathbf B=\widetilde{\mathbf D}_v-\mathbf D_\lambda$, then gives
the corresponding operator norm bound.
Neither union bound requires independence.
Take $x=t\log9+u$, where $u>0$.
The conditional failure probability is at most
$4(M+1)9^t e^{-(t\log9+u)}=4(M+1)e^{-u}$.
Define $C_{\mathrm{Gram,net}}=2(\log9)C_{\mathrm{Gram}}$.
Since $\log9>1$ and $t\log9+u\leq(\log9)(t+u)$,
the factor $2$ in \eqref{gcrtb:eq:net-norm-bounds} gives
$\max_{0\leq v\leq M}\|\widetilde{\mathbf D}_v-\mathbf D_\lambda\|_{\mathrm{op}}
\leq C_{\mathrm{Gram,net}}\{\sqrt{(t+u)/n}+(t+u)/n\}$
outside this failure event.
Both constants depend only on $v_+$, $\ell$, and $\bar c$,
so the bounds are uniform in $\mathbf S$ and the population.
The law of total probability therefore gives the same failure
probability bound $4(M+1)e^{-u}$ without conditioning on $\mathbf S$.
For any prescribed error probability, choose a sufficiently large
fixed $u$, depending only on that probability and fixed $M$,
so $u$ is independent of $n$ and the population.
Since $t+u\leq(1+u)t$ for $t\geq1$ and $t=o(n)$ uniformly,
we have $\sqrt{(t+u)/n}+(t+u)/n=O(\sqrt{t/n})$,
with a constant depending only on this fixed $u$.
This proves

\begin{equation}\label{gcrtb:eq:comparison-gram-concentration}
\max_{0\leq v\leq M}
\|\widetilde{\mathbf D}_v-\mathbf D_\lambda\|_{\mathrm{op}}
=O_{\mathbb P}(\sqrt{t/n}).
\end{equation}

\medskip
\noindent\textbf{Step 3: Transfer to the actual Gram matrices.}

We first bound the operator norms of the comparison Gram matrices.

By \eqref{gcrtb:eq:conditional-gram-centre},
\eqref{gcrtint:eq:ridge-bounds}, and \eqref{gcrtcent:eq:population-conditions},
$\|\mathbf D_\lambda\|_{\mathrm{op}}
\leq v_++\ell\bar c^2/4+K$.
Together with \eqref{gcrtb:eq:comparison-gram-concentration}
and $t=o(n)$, this gives
$\max_{0\leq v\leq M}\|\widetilde{\mathbf D}_v\|_{\mathrm{op}}
=O_{\mathbb P}(1)$.
Using $\mathbf Q=\mathbf H^2$ and
\eqref{gcrtb:eq:comparison-gram-definition}, we have
$\|\mathbf H\overline{\mathbf X}^{(v)}\|_{\mathrm{op}}^2/n
=\|\widetilde{\mathbf D}_v-a_n\mathbf I_t\|_{\mathrm{op}}
\leq\|\widetilde{\mathbf D}_v\|_{\mathrm{op}}$,
where the inequality uses $a_n>0$ and
$\widetilde{\mathbf D}_v-a_n\mathbf I_t\succeq0$.
Thus
\begin{equation}\label{gcrtb:eq:comparison-residual-norm}
\max_{0\leq v\leq M}
\|\mathbf H\overline{\mathbf X}^{(v)}\|_{\mathrm{op}}/\sqrt n
=O_{\mathbb P}(1).
\end{equation}

Let $\mathbf\Delta_v=\mathbf X_c^{(v)}-\overline{\mathbf X}^{(v)}$.
By \eqref{gcrtint:eq:centered-residual-scores} and
\eqref{gcrtb:eq:comparison-gram-definition}, the common penalty cancels, giving
$$
\mathbf D_v-\widetilde{\mathbf D}_v
=\{\mathbf\Delta_v^\top\mathbf Q\overline{\mathbf X}^{(v)}
+\overline{\mathbf X}^{(v)\top}\mathbf Q\mathbf\Delta_v
+\mathbf\Delta_v^\top\mathbf Q\mathbf\Delta_v\}/n.
$$
Since $\|\mathbf H\|_{\mathrm{op}}\leq1$, we have
$\|\mathbf H\mathbf\Delta_v\|_{\mathrm{op}}\leq\|\mathbf\Delta_v\|_F$.
Using $\mathbf Q=\mathbf H^2$, the triangle inequality, and
submultiplicativity, we obtain
$$
\|\mathbf D_v-\widetilde{\mathbf D}_v\|_{\mathrm{op}}
\leq2\frac{\|\mathbf\Delta_v\|_F}{\sqrt n}
       \frac{\|\mathbf H\overline{\mathbf X}^{(v)}\|_{\mathrm{op}}}{\sqrt n}
     +\frac{\|\mathbf\Delta_v\|_F^2}{n}.
$$
Together with \eqref{gcrtb:eq:comparison-residual-norm},
the coupling bound \eqref{gcrtint:eq:coupling} in
Lemma~\ref{gcrtb:lem:gaussian-comparison} yields
$\max_v\|\mathbf D_v-\widetilde{\mathbf D}_v\|_{\mathrm{op}}
=O_{\mathbb P}(e_n+e_n^2)=O_{\mathbb P}(e_n)$.
We decompose
$\mathbf D_v-\mathbf D_\lambda=(\mathbf D_v-\widetilde{\mathbf D}_v)
+(\widetilde{\mathbf D}_v-\mathbf D_\lambda)$.
The bound
$\max_v\|\mathbf D_v-\widetilde{\mathbf D}_v\|_{\mathrm{op}}
=O_{\mathbb P}(e_n)$, together with
\eqref{gcrtb:eq:comparison-gram-concentration} and
$\sqrt{t/n}\leq e_n$, proves the intermediate approximation
\begin{equation}\label{gcrtb:eq:actual-gram-conditional}
\max_{0\leq v\leq M}\|\mathbf D_v-\mathbf D_\lambda\|_{\mathrm{op}}
=O_{\mathbb P}(e_n).
\end{equation}

\medskip
\noindent\textbf{Step 4: Replacement of the conditional centre.}

To obtain the population centre required in
\eqref{gcrtb:eq:gram-approximation}, use the definitions in
\eqref{gcrtb:eq:conditional-gram-centre} and \eqref{gcrtb:eq:population} to write
$$
\mathbf D_\lambda-\bar{\mathbf D}_n
=\mathbf C(\widehat{\mathbf B}_n-\mathbf B_n)\mathbf C^\top
 +(\tau-\bar\tau_n)\mathbf V.
$$
The second bound in \eqref{gcrtb:eq:ridge-concentration} of
Lemma~\ref{gcrtb:lem:ridge-concentration} gives
$\|\mathbf C(\widehat{\mathbf B}_n-\mathbf B_n)\mathbf C^\top\|_{\mathrm{op}}
=O_{\mathbb P}(\sqrt{t/n})$.
Equation~\eqref{gcrtb:eq:trace-concentration} in
Lemma~\ref{gcrtb:lem:ridge-concentration} and
$\|\mathbf V\|_{\mathrm{op}}\leq v_+$ give
$\|(\tau-\bar\tau_n)\mathbf V\|_{\mathrm{op}}
=O_{\mathbb P}(n^{-1})$.
Since $n^{-1}\leq\sqrt{t/n}\leq e_n$, we have
$\|\mathbf D_\lambda-\bar{\mathbf D}_n\|_{\mathrm{op}}
=O_{\mathbb P}(e_n)$.
The bound $\|\mathbf D_\lambda-\bar{\mathbf D}_n\|_{\mathrm{op}}
=O_{\mathbb P}(e_n)$ and \eqref{gcrtb:eq:actual-gram-conditional} give
$\max_v\|\mathbf D_v-\bar{\mathbf D}_n\|_{\mathrm{op}}
=O_{\mathbb P}(e_n)$, proving \eqref{gcrtb:eq:gram-approximation}.

\end{proof}

\begin{proof}[Proof of \eqref{gcrtb:eq:inverse-gram-bound}]

The population spectral lower bound
\eqref{gcrtcent:eq:population-gram-lower-bound} gives
$\bar{\mathbf D}_n\succeq\tau_0v_-\mathbf I_t$ for $n\geq2$.
Since $e_n\to0$ uniformly, \eqref{gcrtb:eq:gram-approximation}
implies that $\max_v\|\mathbf D_v-\bar{\mathbf D}_n\|_{\mathrm{op}}
\leq\tau_0v_-/2$ with probability tending uniformly to one.
On this event, Weyl's inequality gives
$\min_v\lambda_{\min}(\mathbf D_v)\geq\tau_0v_-/2$, and hence
$\max_v\|\mathbf D_v^{-1}\|_{\mathrm{op}}\leq2/(\tau_0v_-)=C_{\mathrm{inv}}$.
This proves \eqref{gcrtb:eq:inverse-gram-bound}.

\end{proof}

\subsubsection[\texorpdfstring{Proof of \eqref{gcrtb:eq:observed-score-approximation} and \eqref{gcrtb:eq:randomised-score-approximation}}{Proof of (\ref*{gcrtb:eq:observed-score-approximation}) and (\ref*{gcrtb:eq:randomised-score-approximation})}]{Proof of \eqref{gcrtb:eq:observed-score-approximation} and \eqref{gcrtb:eq:randomised-score-approximation}}\label{app:gcrt-score-approximation-proof}
We prove these two assertions together using the comparison copies
\eqref{gcrtb:eq:comparison-copies} and the coupling bound
\eqref{gcrtint:eq:coupling} in Lemma~\ref{gcrtb:lem:gaussian-comparison}.
As in Appendix~\ref{app:gcrt-residual-approximation-proof}, we first
identify the conditional means, control the fluctuations, transfer
the bounds to the actual candidates, and replace the conditional
centres by the population centres.

\medskip
\noindent\textbf{Step 1: Conditional means of the comparison score vectors.}

For the same comparison copies in \eqref{gcrtb:eq:comparison-copies},
define the residual score vectors, as in
\eqref{gcrtint:eq:centered-residual-scores}, by
\begin{equation}\label{gcrtb:eq:comparison-score-definition}
\widetilde{\bs c}_v=\overline{\mathbf X}^{(v)\top}\mathbf Q\bs Y_c/n,
\qquad v=0,\ldots,M.
\end{equation}

Substitute \eqref{gcrtb:eq:comparison-copies} and
\eqref{gcrtint:eq:centered-model} into
\eqref{gcrtb:eq:comparison-score-definition} to obtain
one fixed term, two linear terms, and one bilinear term, conditional
on $\mathbf S$:
\begin{equation}\label{gcrtb:eq:comparison-score-decomposition}
\begin{aligned}
 \widetilde{\bs c}_v={}&
 \underbrace{\mathbf C\widehat{\mathbf B}_n\bs\gamma}_{\text{fixed term}}\\
 &+\underbrace{\mathbf M^\top\mathbf Q
       (\mathbf E\bs\beta_{\mathcal T}+\sigma\bs\varepsilon)/n
 +\overline{\mathbf E}^{(v)\top}\mathbf Q\mathbf S\bs\gamma/n}
       _{\text{two linear terms}}\\
 &+\underbrace{\overline{\mathbf E}^{(v)\top}\mathbf Q
       (\mathbf E\bs\beta_{\mathcal T}+\sigma\bs\varepsilon)/n}
       _{\text{bilinear term}}.
\end{aligned}
\end{equation}

Conditional on $\mathbf S$, the matrices $\mathbf M$ and $\mathbf Q$
are fixed, while $\mathbf E$, $\overline{\mathbf E}^{(v)}$, and
$\bs\varepsilon$ are centred by \eqref{gcrtint:eq:centered-model}
and Lemma~\ref{gcrtb:lem:gaussian-comparison}. The two linear terms in
\eqref{gcrtb:eq:comparison-score-decomposition}
therefore have conditional mean zero.
Moreover, Lemma~\ref{gcrtb:lem:gaussian-comparison} gives
independence of $\bs\varepsilon$ and the collection
$(\overline{\mathbf E}^{(v)})_{v=0}^M$ conditional on $\mathbf S$. Hence
$\mathbb E(\overline{\mathbf E}^{(v)\top}\mathbf Q\bs\varepsilon
\mid\mathbf S)=0$ for every $v=0,\ldots,M$.

For $v=0$, we have $\overline{\mathbf E}^{(0)}=\mathbf E$, and
\eqref{gcrtb:eq:residual-second-moment} gives
$\mathbb E(\mathbf E^\top\mathbf Q\mathbf E\mid\mathbf S)/n
=\tau\mathbf V$.
Thus \eqref{gcrtb:eq:comparison-score-decomposition} yields
\begin{equation}\label{gcrtb:eq:comparison-score-means}
\mathbb E(\widetilde{\bs c}_0\mid\mathbf S)
=\mathbf C\widehat{\mathbf B}_n\bs\gamma
  +\tau\mathbf V\bs\beta_{\mathcal T},
\end{equation}
and 
\begin{equation}\label{gcrtb:eq:comparison-randomised-score-mean}
\mathbb E(\widetilde{\bs c}_m\mid\mathbf S)
=\mathbf C\widehat{\mathbf B}_n\bs\gamma
  +\frac1n
  \mathbb E(\overline{\mathbf E}^{(m)\top}\mathbf Q\mathbf E
  \mid\mathbf S)\bs\beta_{\mathcal T},
\end{equation}
for $m=1,\ldots,M$.

The conditional mean in
\eqref{gcrtb:eq:comparison-randomised-score-mean} contains
$n^{-1}\mathbb E(\overline{\mathbf E}^{(m)\top}\mathbf Q\mathbf E
\mid\mathbf S)$. The following lemma bounds its difference from
$\tau\mathbf F\mathbf V$ for the same oracle Gaussian construction defined in
Appendix~\ref{app:gcrt-residual-approximation-proof} and used to couple the comparison copies
to the actual candidates in Lemma~\ref{gcrtb:lem:gaussian-comparison}; it also
provides the operator norm bound on $\mathbf F$ needed when we replace
$\tau$ by $\bar\tau_n$. 

\begin{lemma}[Cross moment of a comparison copy]\label{gcrtb:lem:comparison-cross-moment}
Consider the oracle Gaussian construction defined in
Appendix~\ref{app:gcrt-residual-approximation-proof} and used in
Lemma~\ref{gcrtb:lem:gaussian-comparison}, started at
$\overline{\mathbf E}^{(0)}=\mathbf E$, with independent
$\mathcal N_t(0,\mathbf V)$ rows conditional on $\mathbf S$.
Use $L$ forward sweeps followed by $L$ reverse sweeps in each branch,
where $L\geq1$ is fixed.
Assume $0\prec\mathbf V\preceq v_+\mathbf I_t$,
$\max\{\|\mathbf J\|_1,\|\mathbf J\|_\infty\}\leq a<1$ and
$0\preceq\mathbf Q\preceq\mathbf I_N$, with $\mathbf Q$ determined by
$\mathbf S$. Then
\begin{equation}\label{gcrtb:eq:comparison-cross-moment-bound}
\left\|\frac1n\mathbb E(\overline{\mathbf E}^{(m)\top}\mathbf Q\mathbf E
\mid\mathbf S)-\tau\mathbf F\mathbf V\right\|_{\mathrm{op}}
\leq C_Lv_+d_n,\qquad 1\leq m\leq M,
\end{equation}
where $C_L$ depends only on $L$ and $a$.
Moreover,
\begin{equation}\label{gcrtb:eq:comparison-scan-norm-bound}
\|\mathbf F\|_{\mathrm{op}}\leq\{(1+a)/(1-a)\}^{2L}.
\end{equation}
\end{lemma}

The proof of Lemma~\ref{gcrtb:lem:comparison-cross-moment} is given in
Appendix~\ref{app:gcrt-comparison-cross-moment-proof}.

\medskip

For the oracle Gaussian construction in
Appendix~\ref{app:gcrt-residual-approximation-proof},
\eqref{gcrtcent:eq:population-conditions} and
\eqref{gcrtint:eq:ridge-bounds} verify the assumptions of
Lemma~\ref{gcrtb:lem:comparison-cross-moment}.
Applying the bound \eqref{gcrtb:eq:comparison-cross-moment-bound} to the conditional mean in \eqref{gcrtb:eq:comparison-randomised-score-mean}, we have
\begin{equation}\label{gcrtb:eq:randomised-score-mean-error}
\begin{aligned}
&\|\mathbb E(\widetilde{\bs c}_m\mid\mathbf S)
-\mathbf C\widehat{\mathbf B}_n\bs\gamma-\tau\mathbf F\mathbf V\bs\beta_{\mathcal T}\|\\
\leq{}&\left\|\frac1n
\mathbb E(\overline{\mathbf E}^{(m)\top}\mathbf Q\mathbf E\mid\mathbf S)
-\tau\mathbf F\mathbf V\right\|_{\mathrm{op}}
\|\bs\beta_{\mathcal T}\|\\
\leq{}& C_Lv_+ d_n\|\bs\beta_{\mathcal T}\|\\
\leq{}& \dfrac{C_Lv_+}{\sqrt{v_-}}\sigma d_nr,\qquad m=1,\ldots,M,
\end{aligned}
\end{equation}
where the first inequality applies the operator norm to
$\bs\beta_{\mathcal T}$ in
\eqref{gcrtb:eq:comparison-randomised-score-mean}, the second uses
\eqref{gcrtb:eq:comparison-cross-moment-bound} in
Lemma~\ref{gcrtb:lem:comparison-cross-moment}, with the constant $C_L$
specified there, and the last uses
$\|\bs\beta_{\mathcal T}\|\leq\sigma r/\sqrt{v_-}$, which follows from
\eqref{gcrtcent:eq:population-conditions} and
\eqref{gcrtcent:eq:signal-strength}:
$v_-\|\bs\beta_{\mathcal T}\|^2
\leq\bs\beta_{\mathcal T}^\top\mathbf V\bs\beta_{\mathcal T}
=\sigma^2r^2$.

\medskip
\noindent\textbf{Step 2: Conditional fluctuations of the comparison score vectors.}

The vector inequality in \eqref{gcrtb:eq:net-norm-bounds} reduces
this step to bounds for scalar projections along a finite net.
Fix a unit vector $\bs z\in\mathbb R^t$ and condition on $\mathbf S$.
Using \eqref{gcrtb:eq:comparison-score-decomposition} with
\eqref{gcrtb:eq:comparison-score-means} for $v=0$ and 
\eqref{gcrtb:eq:comparison-randomised-score-mean} for $v\in [M]$, we have 
\begin{equation}\label{gcrtb:eq:directional-score-fluctuation}
\begin{aligned}
&\bs z^\top\{\widetilde{\bs c}_v-
\mathbb E(\widetilde{\bs c}_v\mid\mathbf S)\}\\
={}&\underbrace{\frac1n(\mathbf M\bs z)^\top\mathbf Q
(\mathbf E\bs\beta_{\mathcal T}+\sigma\bs\varepsilon)}_{\text{(I): linear term}}
+\underbrace{\frac1n(\overline{\mathbf E}^{(v)}\bs z)^\top
\mathbf Q\mathbf S\bs\gamma}_{\text{(II): linear term}}\\
&\qquad+\underbrace{\frac1n\bigl\{
(\overline{\mathbf E}^{(v)}\bs z)^\top\mathbf Q
(\mathbf E\bs\beta_{\mathcal T}+\sigma\bs\varepsilon)
-\bs z^\top\mathbb E(\overline{\mathbf E}^{(v)\top}\mathbf Q\mathbf E
\mid\mathbf S)\bs\beta_{\mathcal T}
\bigr\}}_{\text{(III): centred bilinear term}},
\end{aligned}
\end{equation}
where (I) and (II) are centred linear Gaussian forms, and (III) is
the bilinear form minus its conditional mean.

By \eqref{gcrtint:eq:centered-model}, conditional on $\mathbf S$,
$\mathbf E\bs\beta_{\mathcal T}$ and $\sigma\bs\varepsilon$ are
independent centred Gaussian vectors. Since
$\bs\beta_{\mathcal T}^\top\mathbf V\bs\beta_{\mathcal T}=\sigma^2r^2$
by \eqref{gcrtcent:eq:signal-strength}, their sum satisfies
\begin{equation}\label{gcrtb:eq:response-noise-gaussian-law}
\mathbf E\bs\beta_{\mathcal T}+\sigma\bs\varepsilon\mid\mathbf S
\sim\mathcal N_N(\bs0,\sigma^2(1+r^2)\mathbf I_N).
\end{equation}
Conditional on $\mathbf S$, terms (I) and (II) in
\eqref{gcrtb:eq:directional-score-fluctuation} are linear functions
of the Gaussian vectors in \eqref{gcrtb:eq:response-noise-gaussian-law}
and \eqref{gcrtb:eq:comparison-directional-gaussian-law}, with fixed
coefficient vectors $\mathbf Q\mathbf M\bs z/n$ and
$\mathbf Q\mathbf S\bs\gamma/n$, respectively.
Thus both terms are centred Gaussian. Their conditional variances
are the expressions on the left in the following bounds;
we justify the inequalities below:
\begin{equation}\label{gcrtb:eq:linear-score-variance-bounds}
\begin{aligned}
\text{(I):}\quad&\frac{\sigma^2(1+r^2)}{n^2}\|\mathbf Q\mathbf M\bs z\|^2
 \leq\frac{\sigma^2(1+R^2)\ell\bar c^2}{4n},\\
\text{(II):}\quad&\frac{\bs z^\top\mathbf V\bs z}{n^2}
 \|\mathbf Q\mathbf S\bs\gamma\|^2
 \leq\frac{v_+\ell G_+^2\sigma^2}{4n}.
\end{aligned}
\end{equation}

For (II), the factor $\bs z^\top\mathbf V\bs z$ comes from the
conditional covariance matrix
$(\bs z^\top\mathbf V\bs z)\mathbf I_N$ in
\eqref{gcrtb:eq:comparison-directional-gaussian-law}.
To prove the inequalities in
\eqref{gcrtb:eq:linear-score-variance-bounds}, we bound the norms
of the two coefficient vectors.
Recall that $\bar c$ is the fixed upper bound on
$\|\mathbf C\|_{\mathrm{op}}$ in \eqref{gcrtcent:eq:population-conditions}.
Since $\|\bs z\|=1$, we have
$\|\mathbf C^\top\bs z\|\leq\|\mathbf C\|_{\mathrm{op}}\leq\bar c$.
Using $\mathbf M=\mathbf S\mathbf C^\top$, $\mathbf Q=\mathbf H^2$,
and $\|\mathbf H\|_{\mathrm{op}}\leq1$, we obtain
$\|\mathbf Q\mathbf M\bs z\|
\leq\|\mathbf H\mathbf S\|_{\mathrm{op}}\|\mathbf C^\top\bs z\|
\leq\|\mathbf H\mathbf S\|_{\mathrm{op}}\bar c$
and
$\|\mathbf Q\mathbf S\bs\gamma\|
\leq\|\mathbf H\mathbf S\|_{\mathrm{op}}\|\bs\gamma\|$.
The bound $\|\mathbf H\mathbf S\|_{\mathrm{op}}^2\leq\ell n/4$
in \eqref{gcrtint:eq:ridge-bounds} applies to both variance expressions.
For (I), also use $r\leq R$; for (II), use
$\bs z^\top\mathbf V\bs z\leq v_+$ and
$\|\bs\gamma\|\leq G_+\sigma$, all from
\eqref{gcrtcent:eq:population-conditions}.
Let $C_{\mathrm{lin}}=(\sqrt\ell/2)
\max\{\bar c\sqrt{1+R^2},G_+\sqrt{v_+}\}$.
By \eqref{gcrtb:eq:linear-score-variance-bounds}, each linear term
has conditional standard deviation at most
$C_{\mathrm{lin}}\sigma/\sqrt n$.

The two Gaussian vectors in term (III) of
\eqref{gcrtb:eq:directional-score-fluctuation} may be dependent.
We therefore bound their joint covariance before applying
Lemma~\ref{gcrtb:lem:quadratic-tail}.
To remove the scale $\sigma$ from this covariance bound, divide the
second vector by $\sigma$ and define
$$
\bs Z_*=
\begin{pmatrix}
 \overline{\mathbf E}^{(v)}\bs z\\
 (\mathbf E\bs\beta_{\mathcal T}+\sigma\bs\varepsilon)/\sigma
\end{pmatrix}
\in\mathbb R^{2N},
\qquad
\mathbf\Gamma=\operatorname{Cov}(\bs Z_*\mid\mathbf S).
$$
Conditional on $\mathbf S$, the vector $\bs Z_*$ is centred Gaussian
by \eqref{gcrtint:eq:centered-model} and the joint Gaussianity and
noise independence in Lemma~\ref{gcrtb:lem:gaussian-comparison}.
By \eqref{gcrtb:eq:comparison-directional-gaussian-law} and
\eqref{gcrtb:eq:response-noise-gaussian-law}, the two conditional
covariance matrices are $(\bs z^\top\mathbf V\bs z)\mathbf I_N$
and $(1+r^2)\mathbf I_N$, respectively.
Since $\|\bs z\|=1$, the conditions
$\mathbf V\preceq v_+\mathbf I_t$ and $r\leq R$ in
\eqref{gcrtcent:eq:population-conditions} give the respective
upper bounds $v_+\mathbf I_N$ and $(1+R^2)\mathbf I_N$.
For any $\bs u,\bs w\in\mathbb R^N$, define
$A=\bs u^\top\overline{\mathbf E}^{(v)}\bs z$ and
$B=\bs w^\top(\mathbf E\bs\beta_{\mathcal T}
+\sigma\bs\varepsilon)/\sigma$.
Both scalar variables have conditional mean zero given $\mathbf S$.
The conditional Minkowski inequality with $p=2$ gives
\begin{equation}\label{gcrtb:eq:conditional-standard-deviation-bound}
\begin{aligned}
&\sqrt{\operatorname{Var}(A+B\mid\mathbf S)}\\
\leq{}&\sqrt{\operatorname{Var}(A\mid\mathbf S)}
+\sqrt{\operatorname{Var}(B\mid\mathbf S)}\\
\leq{}&\sqrt{v_+}\|\bs u\|
+\sqrt{1+R^2}\|\bs w\|,
\end{aligned}
\end{equation}
where the second inequality uses
$\operatorname{Var}(A\mid\mathbf S)\leq v_+\|\bs u\|^2$ and
$\operatorname{Var}(B\mid\mathbf S)\leq(1+R^2)\|\bs w\|^2$,
obtained by applying the conditional covariance bounds
$v_+\mathbf I_N$ and $(1+R^2)\mathbf I_N$
to $\bs u$ and $\bs w$, respectively.
We square \eqref{gcrtb:eq:conditional-standard-deviation-bound}
and then apply the Cauchy--Schwarz inequality in $\mathbb R^2$
to obtain
$$
\begin{aligned}
&\operatorname{Var}\left(
\bs u^\top\overline{\mathbf E}^{(v)}\bs z
+\frac{\bs w^\top(\mathbf E\bs\beta_{\mathcal T}
+\sigma\bs\varepsilon)}{\sigma}
\mid\mathbf S\right)\\
\leq{}&
\bigl(\sqrt{v_+}\|\bs u\|
+\sqrt{1+R^2}\|\bs w\|\bigr)^2\\
\leq{}&
(v_++1+R^2)(\|\bs u\|^2+\|\bs w\|^2).
\end{aligned}
$$
Hence $\|\mathbf\Gamma\|_{\mathrm{op}}\leq v_++1+R^2$.

Define
$\mathbf B_*=(2n)^{-1}
\begin{pmatrix}0&\mathbf Q\\\mathbf Q&0\end{pmatrix}$.
The scalar
$(\overline{\mathbf E}^{(v)}\bs z)^\top\mathbf Q
(\mathbf E\bs\beta_{\mathcal T}+\sigma\bs\varepsilon)/(n\sigma)$
equals $\bs Z_*^\top\mathbf B_*\bs Z_*$.
Since $\bs Z_*$ is conditionally Gaussian,
$\bs Z_*^\top\mathbf B_*\bs Z_*$ has the same conditional law as
$\bs Z^\top\mathbf\Gamma^{1/2}\mathbf B_*
\mathbf\Gamma^{1/2}\bs Z$,
where $\bs Z\in\mathbb R^{2N}$ is standard Gaussian and independent
of $\mathbf S$.

By \eqref{gcrtint:eq:ridge-bounds},
$\|\mathbf Q\|_{\mathrm{op}}\leq1$ and
$\|\mathbf Q\|_F\leq\sqrt N\leq\sqrt n$.
The symmetric block structure of $\mathbf B_*$ therefore gives
$\|\mathbf B_*\|_{\mathrm{op}}=\|\mathbf Q\|_{\mathrm{op}}/(2n)
\leq1/(2n)$ and
$\|\mathbf B_*\|_F=\|\mathbf Q\|_F/(\sqrt2\,n)
\leq1/\sqrt{2n}$.
Let $C_{\mathrm{quad}}=v_++1+R^2$. Together with $\|\mathbf\Gamma\|_{\mathrm{op}}\leq C_{\mathrm{quad}}$,
these bounds imply
\begin{equation}\label{gcrtb:eq:score-quadratic-coefficient-bounds}
\begin{aligned}
\|\mathbf\Gamma^{1/2}\mathbf B_*\mathbf\Gamma^{1/2}\|_{\mathrm{op}}
&\leq\|\mathbf\Gamma\|_{\mathrm{op}}\|\mathbf B_*\|_{\mathrm{op}}
\leq C_{\mathrm{quad}}/n,\\
\|\mathbf\Gamma^{1/2}\mathbf B_*\mathbf\Gamma^{1/2}\|_F
&\leq\|\mathbf\Gamma\|_{\mathrm{op}}\|\mathbf B_*\|_F
\leq C_{\mathrm{quad}}/\sqrt n.
\end{aligned}
\end{equation}
The conditional mean of this normalised quadratic form is
$\operatorname{tr}(\mathbf\Gamma^{1/2}\mathbf B_*\mathbf\Gamma^{1/2})$.
Thus its centred version corresponds exactly to term (III) in
\eqref{gcrtb:eq:directional-score-fluctuation}, divided by $\sigma$.
Apply Lemma~\ref{gcrtb:lem:quadratic-tail} conditionally on $\mathbf S$
with coefficient matrix $\mathbf\Gamma^{1/2}\mathbf B_*\mathbf\Gamma^{1/2}$.
By \eqref{gcrtb:eq:score-quadratic-coefficient-bounds}, multiplication
by $\sigma$ gives a bound $2C_{\mathrm{quad}}\sigma\{\sqrt{x/n}+x/n\}$ for the
term (III), with conditional failure probability at most
$2e^{-x}$ for every $x>0$.
The Gaussian tail bound and
\eqref{gcrtb:eq:linear-score-variance-bounds} give a bound
$\sqrt2 C_{\mathrm{lin}}\sigma\sqrt{x/n}$ for each of terms (I) and (II), each with
conditional failure probability at most $2e^{-x}$.
Let $C_{\mathrm{score}}=2\sqrt2 C_{\mathrm{lin}}+2C_{\mathrm{quad}}$.
For the fixed $\bs z$ and $v$, the triangle inequality and a union
bound over these three failure events therefore give
\begin{equation}\label{gcrtb:eq:directional-score-tail}
\mathbb P\left(
\left|\bs z^\top\{\widetilde{\bs c}_v-
\mathbb E(\widetilde{\bs c}_v\mid\mathbf S)\}\right|
>C_{\mathrm{score}}\sigma\{\sqrt{x/n}+x/n\}\mid\mathbf S
\right)\leq6e^{-x}.
\end{equation}
The constants $C_{\mathrm{lin}}$, $C_{\mathrm{quad}}$, and
$C_{\mathrm{score}}$ depend only on the fixed bounds
$\ell,\bar c,G_+,v_+,R$, so the bounds are uniform in
$\mathbf S$, $\bs z$, $v$, and the population.

We now use the finite-net argument from the derivation of
\eqref{gcrtb:eq:comparison-gram-concentration}, with the vector
inequality in \eqref{gcrtb:eq:net-norm-bounds} and
$\bs b=\widetilde{\bs c}_v-\mathbb E(\widetilde{\bs c}_v\mid\mathbf S)$.
Apply \eqref{gcrtb:eq:directional-score-tail} with
$x=t\log9+u$, $u>0$, to the at most $(M+1)9^t$ pairs of net
points and candidate indices, and take a second union bound.
The conditional failure probability is at most $6(M+1)e^{-u}$.
Define $C_{\mathrm{net}}=(4\log9/3)C_{\mathrm{score}}$.
Since $\log9>1$, the factor $4/3$ in
\eqref{gcrtb:eq:net-norm-bounds} then gives
$\max_{0\leq v\leq M}\|\widetilde{\bs c}_v-
\mathbb E(\widetilde{\bs c}_v\mid\mathbf S)\|
\leq C_{\mathrm{net}}\sigma\{\sqrt{(t+u)/n}+(t+u)/n\}$
outside this failure event.
Neither union bound requires independence.
The constants are uniform in $\mathbf S$, so integration over
$\mathbf S$ gives the same unconditional failure bound.
As in the derivation of \eqref{gcrtb:eq:comparison-gram-concentration},
we choose a sufficiently large fixed $u$ for each prescribed error
probability. This gives the uniform stochastic order
$O_{\mathbb P}(\sigma\sqrt{t/n})$ for the comparison score fluctuations.
The two concentration bounds are therefore

\begin{equation}\label{gcrtb:eq:comparison-concentration}
\begin{aligned}
 \max_{0\leq v\leq M}
 \|\widetilde{\mathbf D}_v-\mathbf D_\lambda\|_{\mathrm{op}}
   &=O_{\mathbb P}(\sqrt{t/n}),\\
 \max_{0\leq v\leq M}
 \|\widetilde{\bs c}_v-\mathbb E(\widetilde{\bs c}_v\mid\mathbf S)\|
   &=O_{\mathbb P}(\sigma\sqrt{t/n}).
\end{aligned}
\end{equation}

\medskip
\noindent\textbf{Step 3: Transfer to the actual score vectors.}

To transfer the score bound in
\eqref{gcrtb:eq:comparison-concentration} to the actual scores,
we first bound $\bs c_v-\widetilde{\bs c}_v$.
Recall that
$\mathbf\Delta_v=\mathbf X_c^{(v)}-\overline{\mathbf X}^{(v)}$.
The definitions in \eqref{gcrtint:eq:centered-residual-scores} and
\eqref{gcrtb:eq:comparison-score-definition} give
$\bs c_v-\widetilde{\bs c}_v
=\mathbf\Delta_v^\top\mathbf Q\bs Y_c/n$.
Since $\mathbf Q=\mathbf H^2$ and
$\|\mathbf H\|_{\mathrm{op}}\leq1$, we have
\begin{equation}\label{gcrtb:eq:score-transfer-bound}
\begin{aligned}
\|\bs c_v-\widetilde{\bs c}_v\|
&=\frac{\|(\mathbf H\mathbf\Delta_v)^\top
             \mathbf H\bs Y_c\|}{n}\\
&\leq\frac{\|\mathbf\Delta_v\|_F}{\sqrt n}
      \frac{\|\mathbf H\bs Y_c\|}{\sqrt n}.
\end{aligned}
\end{equation}
The coupling bound \eqref{gcrtint:eq:coupling} in
Lemma~\ref{gcrtb:lem:gaussian-comparison} already gives
$\max_v\|\mathbf\Delta_v\|_F/\sqrt n=O_{\mathbb P}(e_n)$.
It therefore remains to show that
$\|\mathbf H\bs Y_c\|/\sqrt n=O_{\mathbb P}(\sigma)$
to obtain a score difference of order $O_{\mathbb P}(\sigma e_n)$
uniformly over $v$.

Conditional on $\mathbf S$, the matrix $\mathbf H$ is fixed.
Thus \eqref{gcrtint:eq:centered-model} and the Gaussian law
\eqref{gcrtb:eq:response-noise-gaussian-law} imply that
$\mathbf H\bs Y_c$ has mean $\mathbf H\mathbf S\bs\gamma$
and covariance
$\sigma^2(1+r^2)\mathbf H\mathbf H^\top
=\sigma^2(1+r^2)\mathbf Q$, since $\mathbf H$ is symmetric
and $\mathbf Q=\mathbf H^2$.
Hence
\begin{equation}\label{gcrtb:eq:response-residual-second-moment}
\begin{aligned}
\mathbb E(\|\mathbf H\bs Y_c\|^2\mid\mathbf S)
&=\|\mathbf H\mathbf S\bs\gamma\|^2+
(\bs\beta_{\mathcal T}^\top\mathbf V\bs\beta_{\mathcal T}+\sigma^2)
\operatorname{tr}(\mathbf Q)\\
&\leq \bigl(1+R^2+\ell G_+^2/4\bigr)\sigma^2n,
\end{aligned}
\end{equation}
where the bound on the first term uses
$\|\mathbf H\mathbf S\|_{\mathrm{op}}^2\leq\ell n/4$ from
\eqref{gcrtint:eq:ridge-bounds} and $\|\bs\gamma\|\leq G_+\sigma$
from \eqref{gcrtcent:eq:population-conditions}.
For the second term, use $\operatorname{tr}(\mathbf Q)\leq n$ from
\eqref{gcrtint:eq:ridge-bounds},
$\bs\beta_{\mathcal T}^\top\mathbf V\bs\beta_{\mathcal T}=\sigma^2r^2$
from \eqref{gcrtcent:eq:signal-strength}, and $r\leq R$ from
\eqref{gcrtcent:eq:population-conditions}.
The bound \eqref{gcrtb:eq:response-residual-second-moment} is uniform
in $\mathbf S$ and the population. Integration over $\mathbf S$
followed by Markov's inequality gives
\begin{equation}\label{gcrtb:eq:response-residual-norm}
\|\mathbf H\bs Y_c\|/\sqrt n=O_{\mathbb P}(\sigma).
\end{equation}

Combining \eqref{gcrtb:eq:score-transfer-bound} with
\eqref{gcrtb:eq:response-residual-norm} and the coupling bound
\eqref{gcrtint:eq:coupling} in
Lemma~\ref{gcrtb:lem:gaussian-comparison} gives
\begin{equation}\label{gcrtb:eq:score-coupling-error}
\max_v\|\bs c_v-\widetilde{\bs c}_v\|=O_{\mathbb P}(\sigma e_n).
\end{equation}
The score bound in \eqref{gcrtb:eq:comparison-concentration} and
\eqref{gcrtb:eq:score-coupling-error} now yield, uniformly over $v$,
\begin{equation}\label{gcrtb:eq:actual-score-conditional}
\begin{aligned}
\bs c_v-\mathbb E(\widetilde{\bs c}_v\mid\mathbf S)
&=(\bs c_v-\widetilde{\bs c}_v)
 +\{\widetilde{\bs c}_v-\mathbb E(\widetilde{\bs c}_v\mid\mathbf S)\}\\
&=O_{\mathbb P}\bigl(\sigma\{e_n+\sqrt{t/n}\}\bigr)\\
&=O_{\mathbb P}(\sigma e_n),
\end{aligned}
\end{equation}
where the last expression is due to $\sqrt{t/n}\leq e_n$ by \eqref{gcrtb:eq:error-rates}. 
We summarise the approximations for the actual Gram matrices and
score vectors as follows:

\begin{equation}\label{gcrtint:eq:actual-expansions}
\begin{aligned}
 \mathbf D_v&=\mathbf D_\lambda+O_{\mathbb P}(e_n),\\
 \bs c_0&=\mathbf C\widehat{\mathbf B}_n\bs\gamma+\tau\mathbf V\bs\beta_{\mathcal T}
                 +O_{\mathbb P}(\sigma e_n),\\
 \bs c_m&=\mathbf C\widehat{\mathbf B}_n\bs\gamma+\tau\mathbf F\mathbf V\bs\beta_{\mathcal T}
                 +O_{\mathbb P}(\sigma\eta_n),
\end{aligned}
\end{equation}
where the first line follows from
\eqref{gcrtb:eq:actual-gram-conditional},
the second combines \eqref{gcrtb:eq:actual-score-conditional}
with the exact conditional mean
\eqref{gcrtb:eq:comparison-score-means},
and the third combines \eqref{gcrtb:eq:actual-score-conditional}
with the conditional mean approximation
\eqref{gcrtb:eq:randomised-score-mean-error},
using $\eta_n=e_n+d_nr$ from \eqref{gcrtb:eq:error-rates}
to absorb the additional error of order $\sigma d_nr$.

\medskip
\noindent\textbf{Step 4: Replacement of the conditional centres.}

The score expansions in \eqref{gcrtint:eq:actual-expansions}
still involve $\mathbf C\widehat{\mathbf B}_n\bs\gamma$ and
$\tau$, both of which depend on the background data $\mathbf S$.
We now replace them by their expectations
$\bs m_n=\mathbf C\mathbf B_n\bs\gamma$ and $\bar\tau_n$,
respectively, defined in \eqref{gcrtb:eq:population}.
We will show that these replacements preserve the error orders
and give \eqref{gcrtb:eq:observed-score-approximation} and
\eqref{gcrtb:eq:randomised-score-approximation}.

We first prove the approximation for the randomised copies.
For $m=1,\ldots,M$, we compare $\bs c_m$ with
$\bs m_n+\mathbf F\bs u_n$.
Using $\bs u_n=\bar\tau_n\mathbf V\bs\beta_{\mathcal T}$ from
\eqref{gcrtb:eq:population}, we have
\begin{equation}\label{gcrtb:eq:randomised-population-decomposition}
\begin{aligned}
&\bs c_m-\bs m_n-\mathbf F\bs u_n\\
={}&\underbrace{\bs c_m-\mathbf C\widehat{\mathbf B}_n\bs\gamma
                -\tau\mathbf F\mathbf V\bs\beta_{\mathcal T}}_{\text{(I)}}\\
&+\underbrace{\mathbf C(\widehat{\mathbf B}_n-\mathbf B_n)\bs\gamma}_{\text{(II)}}
 +\underbrace{(\tau-\bar\tau_n)\mathbf F\mathbf V\bs\beta_{\mathcal T}}_{\text{(III)}}.
\end{aligned}
\end{equation}
Term (I) in \eqref{gcrtb:eq:randomised-population-decomposition}
has norm $O_{\mathbb P}(\sigma\eta_n)$ uniformly over $m$, by the
randomised score expansion in \eqref{gcrtint:eq:actual-expansions}.
For term (II), the first bound in
\eqref{gcrtb:eq:ridge-concentration} of
Lemma~\ref{gcrtb:lem:ridge-concentration} gives
$\|\mathbf C(\widehat{\mathbf B}_n-\mathbf B_n)\bs\gamma\|
=O_{\mathbb P}(\sigma\sqrt{t/n})$.

For term (III) in \eqref{gcrtb:eq:randomised-population-decomposition},
the bounds on $\mathbf V$ in \eqref{gcrtcent:eq:population-conditions}
and the signal definition \eqref{gcrtcent:eq:signal-strength} give
$\|\mathbf V\bs\beta_{\mathcal T}\|
\leq v_+\|\bs\beta_{\mathcal T}\|
\leq v_+\sigma r/\sqrt{v_-}$.
Therefore,
$$
\begin{aligned}
&\|(\tau-\bar\tau_n)\mathbf F\mathbf V\bs\beta_{\mathcal T}\|\\
\leq{}&|\tau-\bar\tau_n|\,\|\mathbf F\|_{\mathrm{op}}
                          \|\mathbf V\bs\beta_{\mathcal T}\|\\
\stackrel{(1)}{\leq}{}&
|\tau-\bar\tau_n| \,
\left(\frac{1+a}{1-a}\right)^{2L}
\frac{v_+}{\sqrt{v_-}}\,\sigma r\\
\stackrel{(2)}{=}{}&O_{\mathbb P}(\sigma r/n),
\end{aligned}
$$
where (1) uses the operator norm bound
\eqref{gcrtb:eq:comparison-scan-norm-bound} in
Lemma~\ref{gcrtb:lem:comparison-cross-moment} and
$\|\mathbf V\bs\beta_{\mathcal T}\|\leq v_+\sigma r/\sqrt{v_-}$,
and (2) uses $\tau-\bar\tau_n=O_{\mathbb P}(n^{-1})$ from
\eqref{gcrtb:eq:trace-concentration} in
Lemma~\ref{gcrtb:lem:ridge-concentration}.
The multiplicative factors depend only on the fixed population bounds
and $L$, so this stochastic order is uniform over
$\mathcal C_n(h_n)$.

The triangle inequality in
\eqref{gcrtb:eq:randomised-population-decomposition} now gives
$$
\begin{aligned}
\max_{1\leq m\leq M}\|\bs c_m-\bs m_n-\mathbf F\bs u_n\|
&=O_{\mathbb P}\bigl(\sigma\{\eta_n+\sqrt{t/n}+r/n\}\bigr)\\
&=O_{\mathbb P}(\sigma\eta_n),
\end{aligned}
$$
where the last order follows from \eqref{gcrtb:eq:error-rates},
$t\geq1$, and $r\leq R$ from \eqref{gcrtcent:eq:population-conditions},
which give
\begin{equation}\label{gcrtb:eq:population-replacement-rates}
\sqrt{t/n}\leq e_n\leq\eta_n,\qquad
\frac rn\leq\frac Rn\leq R\sqrt{t/n}\leq Re_n.
\end{equation}
This proves \eqref{gcrtb:eq:randomised-score-approximation}.

For the observation, we subtract $\bs m_n+\bs u_n$ and obtain
$$
\begin{aligned}
&\bs c_0-\bs m_n-\bs u_n\\
={}&\{\bs c_0-\mathbf C\widehat{\mathbf B}_n\bs\gamma
                -\tau\mathbf V\bs\beta_{\mathcal T}\}\\
&+\mathbf C(\widehat{\mathbf B}_n-\mathbf B_n)\bs\gamma
 +(\tau-\bar\tau_n)\mathbf V\bs\beta_{\mathcal T}\\
={}&O_{\mathbb P}(\sigma e_n)
   +O_{\mathbb P}(\sigma\sqrt{t/n})
   +O_{\mathbb P}(\sigma r/n)\\
={}&O_{\mathbb P}(\sigma e_n),
\end{aligned}
$$
where the three orders are justified term by term as follows:
\begin{enumerate}
\item[1.] The observed score expansion in
\eqref{gcrtint:eq:actual-expansions} gives
$$
\|\bs c_0-\mathbf C\widehat{\mathbf B}_n\bs\gamma
          -\tau\mathbf V\bs\beta_{\mathcal T}\|
=O_{\mathbb P}(\sigma e_n).
$$
\item[2.] The first bound in \eqref{gcrtb:eq:ridge-concentration}
of Lemma~\ref{gcrtb:lem:ridge-concentration} gives
$$
\|\mathbf C(\widehat{\mathbf B}_n-\mathbf B_n)\bs\gamma\|
=O_{\mathbb P}(\sigma\sqrt{t/n}).
$$
\item[3.] The trace bound \eqref{gcrtb:eq:trace-concentration}
in Lemma~\ref{gcrtb:lem:ridge-concentration} gives
$\tau-\bar\tau_n=O_{\mathbb P}(n^{-1})$.
Together with $\|\mathbf V\bs\beta_{\mathcal T}\|
\leq v_+\sigma r/\sqrt{v_-}$, this yields
$$
\begin{aligned}
\|(\tau-\bar\tau_n)\mathbf V\bs\beta_{\mathcal T}\|
&\leq |\tau-\bar\tau_n|\,\frac{v_+}{\sqrt{v_-}}\sigma r\\
&=O_{\mathbb P}(\sigma r/n).
\end{aligned}
$$
\end{enumerate}
By \eqref{gcrtb:eq:population-replacement-rates}, the three bounds give
$\bs c_0-\bs m_n-\bs u_n=O_{\mathbb P}(\sigma e_n)$.
Since $e_n\leq\eta_n$, this proves
\eqref{gcrtb:eq:observed-score-approximation}.
\qed

\subsubsection{Proof of Lemma~\ref{gcrtcent:lem:scan-displacement}}\label{app:gcrt-auxiliary-proofs}

\begin{proof}[Proof of Lemma~\ref{gcrtcent:lem:scan-displacement}]
We first show that $\mathbf F$ is self-adjoint and positive
semidefinite in the $\mathbf\Theta$ inner product.
For $j=1,\ldots,t$, let $\bs e_j$ be the $j$th coordinate vector
in $\mathbb R^t$ and define
$\mathbf K_j=\mathbf I_t-\bs e_j\bs e_j^\top\mathbf\Theta/\Theta_{jj}$.
In this proof, the $j$th coordinate update is the deterministic map
$\bs x\mapsto\mathbf K_j\bs x$ on $\mathbb R^t$.
By the definition of $J_{jk}$, this map replaces $x_j$ by
$\sum_{k\ne j}J_{jk}x_k$ and leaves the other coordinates unchanged.

For any initial vector $\bs x\in\mathbb R^t$, a forward sweep applies
these maps for $j=1,\ldots,t$ in order, each to the current vector,
so its output is
$\bs y=\mathbf K_t\cdots\mathbf K_1\bs x$.
For each $j=1,\ldots,t$,
$y_j=\sum_{k<j}J_{jk}y_k+\sum_{k>j}J_{jk}x_k$:
the coordinates with $k<j$ have already been updated, whereas
those with $k>j$ still have their original values.
Thus $(\mathbf I_t-\mathbf J_<)\bs y=\mathbf J_>\bs x$,
so \eqref{gcrtint:eq:scan} gives $\bs y=\mathbf F_{\to}\bs x$.
Since $\bs x$ was arbitrary, we conclude that
$\mathbf F_{\to}=\mathbf K_t\cdots\mathbf K_1$.
The reverse sweep interchanges $<$ and $>$ and gives
$\mathbf F_{\leftarrow}=\mathbf K_1\cdots\mathbf K_t$.

For $\mathbf A\in\mathbb R^{t\times t}$, let
$\mathbf A^*=\mathbf\Theta^{-1}\mathbf A^\top\mathbf\Theta$
denote its adjoint in the $\mathbf\Theta$ inner product.
Direct multiplication gives
$\mathbf K_j^\top\mathbf\Theta=\mathbf\Theta\mathbf K_j$,
so each $\mathbf K_j$ is self-adjoint: $\mathbf K_j^*=\mathbf K_j$.
The product rule for adjoints gives
$$
\begin{aligned}
\mathbf F_{\to}^*
&=(\mathbf K_t\cdots\mathbf K_1)^*\\
&=\mathbf K_1^*\cdots\mathbf K_t^*\\
&=\mathbf K_1\cdots\mathbf K_t
=\mathbf F_{\leftarrow}.
\end{aligned}
$$
Consequently, \eqref{gcrtint:eq:scan} gives
$\mathbf F=(\mathbf F_{\to}^L)^*\mathbf F_{\to}^L$,
which is self-adjoint and positive semidefinite.

To prove $\|\mathbf F\|_{\mathbf\Theta}\leq a^{2L}$ in
\eqref{gcrtcent:eq:scan-displacement}, we first bound
$\|\mathbf F\|_\infty$, which bounds the spectral radius and
hence the $\mathbf\Theta$ operator norm of the self-adjoint
matrix $\mathbf F$.
Start either sweep from any $\bs x$ satisfying
$\|\bs x\|_\infty\leq1$.
Since $\|\mathbf J\|_\infty\leq a$, each updated coordinate is a
linear combination of the other current coordinates, with
absolute coefficients summing to at most $a$.
Thus the updated coordinate has absolute value at most $a$ whenever
all current coordinates have absolute value at most one.
Since $a<1$, this condition is preserved throughout the sweep.
Every coordinate is updated once, so
$\max(\|\mathbf F_{\to}\|_\infty,\|\mathbf F_{\leftarrow}\|_\infty)\leq a$.

Submultiplicativity and \eqref{gcrtint:eq:scan} give
$$
\|\mathbf F\|_\infty
\leq\|\mathbf F_{\leftarrow}\|_\infty^L
     \|\mathbf F_{\to}\|_\infty^L
\leq a^{2L}.
$$
Since $\mathbf F$ is self-adjoint and positive semidefinite in the
$\mathbf\Theta$ inner product, its induced operator norm equals its spectral radius, namely $\|\mathbf F\|_{\mathbf\Theta}
=\rho(\mathbf F)$.
Since the spectral radius is bounded by every induced matrix norm
\citep[Theorem~5.6.9(a)]{horn_johnson_matrix_2013}, we obtain
$\|\mathbf F\|_{\mathbf\Theta}
=\rho(\mathbf F)\leq\|\mathbf F\|_\infty\leq a^{2L}$.

Finally, for every $\bs x\in\mathbb R^t$, the reverse triangle
inequality gives
$$
\|(\mathbf I_t-\mathbf F)\bs x\|_{\mathbf\Theta}
\geq\|\bs x\|_{\mathbf\Theta}-\|\mathbf F\bs x\|_{\mathbf\Theta}
\geq(1-a^{2L})\|\bs x\|_{\mathbf\Theta}.
$$
This proves both bounds in \eqref{gcrtcent:eq:scan-displacement}.
\end{proof}

\subsubsection{Proof of Lemma~\ref{gcrtb:lem:ridge-concentration}}
\label{app:gcrt-ridge-concentration-proof}

\begin{proof}[Proof of Lemma~\ref{gcrtb:lem:ridge-concentration}]
We use the deterministic bounds \eqref{gcrtint:eq:ridge-bounds}
established in Appendix~\ref{app:gcrt-common-conditions}.
To prove the two concentration bounds in
\eqref{gcrtb:eq:ridge-concentration}, we first establish a gradient
bound for $\bs x^\top\widehat{\mathbf B}_n\bs y$, with deterministic
$\bs x,\bs y\in\mathbb R^s$. These bilinear forms include the
coordinates of $\mathbf C\widehat{\mathbf B}_n\bs\gamma$ and the quadratic
forms of $\mathbf C\widehat{\mathbf B}_n\mathbf C^\top$. Gaussian Poincar\'e
and concentration inequalities will then convert the gradient
bound into variance and tail bounds for these quantities.
We then prove \eqref{gcrtb:eq:trace-concentration} by bounding
the gradient of $\tau$ and applying the Gaussian Poincar\'e inequality again.

\medskip
\noindent\textbf{Step 1: A common gradient bound.}

Recall that $\mathbf S$ is the background data matrix with independent
$\mathcal N_s(0,\mathbf\Sigma_{\mathcal S\mathcal S})$ rows,
as specified after \eqref{gcrtint:eq:centered-model}.
Write
$\mathbf S=\mathbf G\mathbf\Sigma_{\mathcal S\mathcal S}^{1/2}$,
where $\mathbf G\in\mathbb R^{N\times s}$ has independent standard
Gaussian entries.
We will bound the gradient of $\bs x^\top\widehat{\mathbf B}_n\bs y$
with respect to $\mathbf G$, accounting for the dependence of
$\mathbf H$ on $\mathbf S$.

Let $\mathbf T=\mathbf S/\sqrt n$ and
$\mathbf R=(\mathbf T\mathbf T^\top+\ell\mathbf I_N)^{-1}$.
The residual matrix satisfies
$\mathbf H=\ell n(\mathbf S\mathbf S^\top+\ell n\mathbf I_N)^{-1}
=\ell\mathbf R$.
Define $\mathbf L=\mathbf H\mathbf S/\sqrt n=\ell\mathbf R\mathbf T$,
the normalised matrix of background residuals.
Then $\widehat{\mathbf B}_n=\mathbf L^\top\mathbf L$.

For both gradient calculations, we use the following bounds,
valid for every $\mathbf S$:
\begin{equation}\label{gcrtb:eq:ridge-auxiliary-norm-bounds}
\begin{aligned}
 \|\mathbf R\|_{\mathrm{op}}&\leq\ell^{-1},
 &\|\ell\mathbf R\|_{\mathrm{op}}&\leq1,\\
 \|\mathbf R\mathbf T\|_{\mathrm{op}}&\leq\frac{1}{2\sqrt\ell},
 &\|\mathbf L\|_{\mathrm{op}}&\leq\frac{\sqrt\ell}{2},\\
 \|\mathbf T^\top\mathbf R\mathbf T\|_{\mathrm{op}}&\leq1.
\end{aligned}
\end{equation}
The bounds for $\mathbf R$, $\ell\mathbf R$, $\mathbf R\mathbf T$,
and $\mathbf L$ follow from
$\ell\mathbf R=\mathbf H$, $\mathbf Q=\mathbf H^2$, and
$\ell\mathbf R\mathbf T=\mathbf L=\mathbf H\mathbf S/\sqrt n$,
together with \eqref{gcrtint:eq:ridge-bounds}.
For $\mathbf T^\top\mathbf R\mathbf T$, the singular value decomposition
of $\mathbf T$ gives eigenvalues $x/(\ell+x)\leq1$, where $x\geq0$
ranges over the eigenvalues of $\mathbf T^\top\mathbf T$.

We first control $d\mathbf L$ under a perturbation $d\mathbf T$.
The inverse differentiation formula gives
\begin{equation}\label{gcrtb:eq:resolvent-differentials}
\begin{aligned}
 d\mathbf R
 &=-\mathbf R\,d(\mathbf T\mathbf T^\top+\ell\mathbf I_N)\,\mathbf R\\
 &=-\mathbf R(d\mathbf T\mathbf T^\top+
 \mathbf T\,d\mathbf T^\top)\mathbf R.
\end{aligned}
\end{equation}
Since $\mathbf L=\ell\mathbf R\mathbf T$, the product rule and
\eqref{gcrtb:eq:resolvent-differentials} give
\begin{equation}\label{gcrtb:eq:derivative}
\begin{aligned}
 d\mathbf L
 &=\ell\,d\mathbf R\mathbf T+\ell\mathbf R\,d\mathbf T\\
 &=\ell\mathbf R\,d\mathbf T
 -\ell\mathbf R\,d\mathbf T\mathbf T^\top\mathbf R\mathbf T
 -\ell\mathbf R\mathbf T\,d\mathbf T^\top\mathbf R\mathbf T.
\end{aligned}
\end{equation}
The triangle inequality and \eqref{gcrtb:eq:ridge-auxiliary-norm-bounds},
applied to the three terms in the last line of
\eqref{gcrtb:eq:derivative}, give
\begin{equation}\label{gcrtb:eq:ridge-differential-bound}
\begin{aligned}
 \|d\mathbf L\|_F
 &\leq\bigl\{\|\ell\mathbf R\|_{\mathrm{op}}
 +\|\ell\mathbf R\|_{\mathrm{op}}
    \|\mathbf T^\top\mathbf R\mathbf T\|_{\mathrm{op}}
 +\ell\|\mathbf R\mathbf T\|_{\mathrm{op}}^2\bigr\}
 \|d\mathbf T\|_F\\
 &\leq(1+1+1/4)\|d\mathbf T\|_F
 =\frac94\|d\mathbf T\|_F.
\end{aligned}
\end{equation}

For deterministic $\bs x,\bs y\in\mathbb R^s$, the identity
$\widehat{\mathbf B}_n=\mathbf L^\top\mathbf L$ gives
$$
d(\bs x^\top\widehat{\mathbf B}_n\bs y)
=(d\mathbf L\bs x)^\top\mathbf L\bs y
 +(\mathbf L\bs x)^\top d\mathbf L\bs y.
$$
Its absolute value is at most
$2\|\mathbf L\|_{\mathrm{op}}\|d\mathbf L\|_F\|\bs x\|\|\bs y\|$.
The covariance condition
$\mathbf\Sigma_{\mathcal S\mathcal S}\preceq s_+\mathbf I_s$ in
\eqref{gcrtcent:eq:population-conditions} gives
\begin{equation}\label{gcrtb:eq:background-covariance-root-bound}
 \|\mathbf\Sigma_{\mathcal S\mathcal S}^{1/2}\|_{\mathrm{op}}
 \leq\sqrt{s_+}.
\end{equation}
Using the bounds for $\mathbf L$ and $d\mathbf L$ in
\eqref{gcrtb:eq:ridge-auxiliary-norm-bounds} and
\eqref{gcrtb:eq:ridge-differential-bound}, together with
\eqref{gcrtb:eq:background-covariance-root-bound} and
$d\mathbf T=d\mathbf G\mathbf\Sigma_{\mathcal S\mathcal S}^{1/2}/\sqrt n$,
we obtain
\begin{equation}\label{gcrtb:eq:lipschitz}
 \|\nabla_{\mathbf G}(\bs x^\top\widehat{\mathbf B}_n\bs y)\|_F
 \leq\frac94\sqrt{\ell s_+/n}\,\|\bs x\|\|\bs y\|.
\end{equation}

\medskip
\noindent\textbf{Step 2: The vector concentration bound.}

We first prove the vector bound in \eqref{gcrtb:eq:ridge-concentration}.
For $j=1,\ldots,t$, let $\bs e_j$ be the $j$th coordinate vector
in $\mathbb R^t$ and define the scalar function
$f_j(\mathbf G)=\bs e_j^\top\mathbf C\widehat{\mathbf B}_n\bs\gamma$.
Each $f_j$ is smooth, and \eqref{gcrtb:eq:lipschitz} with
$\bs x=\mathbf C^\top\bs e_j$ and $\bs y=\bs\gamma$ shows that it is
globally Lipschitz and square integrable under the Gaussian law of
$\mathbf G$.
The Gaussian Poincar\'e inequality
\citep[Eq.~(2.16)]{ledoux1999logsobolev} states that
\begin{equation}\label{gcrtb:eq:gaussian-poincare}
 \operatorname{Var}(f(\mathbf G))
 \leq\mathbb E\|\nabla_{\mathbf G}f(\mathbf G)\|_F^2
\end{equation}
for every smooth scalar function $f:\mathbb R^{N\times s}\to\mathbb R$
whose value and gradient are square integrable under the standard
Gaussian law of $\mathbf G$.
Applying \eqref{gcrtb:eq:gaussian-poincare} to $f_j$ and using
\eqref{gcrtb:eq:lipschitz} gives
\begin{equation}\label{gcrtb:eq:ridge-coordinate-variance}
\begin{aligned}
 \operatorname{Var}(f_j(\mathbf G))
 &\leq\mathbb E\|\nabla_{\mathbf G}f_j(\mathbf G)\|_F^2\\
 &\leq\frac{81\ell s_+}{16n}
       \|\mathbf C^\top\bs e_j\|^2\|\bs\gamma\|^2.
\end{aligned}
\end{equation}

Since $\bs m_n=\mathbb E(\mathbf C\widehat{\mathbf B}_n\bs\gamma)$
by \eqref{gcrtb:eq:population}, summing
\eqref{gcrtb:eq:ridge-coordinate-variance} over $j$ gives
$$
\begin{aligned}
\mathbb E\|\mathbf C\widehat{\mathbf B}_n\bs\gamma-\bs m_n\|^2
&=\sum_{j=1}^t
  \operatorname{Var}(\bs e_j^\top\mathbf C\widehat{\mathbf B}_n\bs\gamma)\\
&\leq\frac{81\ell s_+}{16n}
  \sum_{j=1}^t\|\mathbf C^\top\bs e_j\|^2\|\bs\gamma\|^2\\
&=\frac{81\ell s_+}{16n}\|\mathbf C\|_F^2\|\bs\gamma\|^2.
\end{aligned}
$$
The population conditions \eqref{gcrtcent:eq:population-conditions} give
$\|\mathbf C\|_{\mathrm{op}}\leq\bar c$ and
$\|\bs\gamma\|\leq G_+\sigma$.
Using $\|\mathbf C\|_F^2\leq t\|\mathbf C\|_{\mathrm{op}}^2\leq t\bar c^2$,
Markov's inequality applied to
$\|\mathbf C\widehat{\mathbf B}_n\bs\gamma-\bs m_n\|^2$ gives, for every $B>0$,
$\mathbb P(\|\mathbf C\widehat{\mathbf B}_n\bs\gamma-\bs m_n\|
>B\sigma\sqrt{t/n})
\leq81\ell s_+\bar c^2G_+^2/(16B^2)$.
Since $\ell$, $s_+$, $\bar c$, and $G_+$ are fixed,
this bound tends to zero as $B\to\infty$, uniformly in $n$
and the population.
This proves the first bound in \eqref{gcrtb:eq:ridge-concentration}.

\medskip
\noindent\textbf{Step 3: The matrix concentration bound.}

For the matrix bound in \eqref{gcrtb:eq:ridge-concentration},
fix a deterministic unit vector $\bs z\in\mathbb R^t$.
In \eqref{gcrtb:eq:lipschitz}, take
$\bs x=\bs y=\mathbf C^\top\bs z$.
The bound $\|\mathbf C\|_{\mathrm{op}}\leq\bar c$ in
\eqref{gcrtcent:eq:population-conditions} gives
$\|\mathbf C^\top\bs z\|\leq\bar c$.
Thus \eqref{gcrtb:eq:lipschitz} shows that
$\bs z^\top\mathbf C\widehat{\mathbf B}_n\mathbf C^\top\bs z$
is globally $L_n$-Lipschitz as a function of $\mathbf G$
with respect to the Frobenius norm, where
$L_n=(9/4)\sqrt{\ell s_+/n}\,\bar c^2=O(n^{-1/2})$.

By \eqref{gcrtb:eq:population}, its mean is
$\bs z^\top\mathbf C\mathbf B_n\mathbf C^\top\bs z$.
Gaussian concentration
\citep[Corollary~2.6]{ledoux1999logsobolev},
applied to this scalar function and its negative, therefore gives,
for every $x>0$,
\begin{equation}\label{gcrtb:eq:ridge-scalar-concentration}
\mathbb P\left(
 \left|\bs z^\top\mathbf C(\widehat{\mathbf B}_n-\mathbf B_n)
                    \mathbf C^\top\bs z\right|
 >L_n\sqrt{2x}
\right)\leq2e^{-x}.
\end{equation}

Let $\mathcal N$ be a deterministic $1/4$-net of the unit sphere
in $\mathbb R^t$ with $|\mathcal N|\leq9^t$.
For the symmetric matrix
$\mathbf B=\mathbf C(\widehat{\mathbf B}_n-\mathbf B_n)\mathbf C^\top$,
combine the matrix inequality in \eqref{gcrtb:eq:net-norm-bounds}
with \eqref{gcrtb:eq:ridge-scalar-concentration} and a union bound over
$\mathcal N$. With $x=t\log9+u$, this gives, for every $u>0$,
\begin{equation}\label{gcrtb:eq:ridge-matrix-concentration}
\mathbb P\left(
 \|\mathbf B\|_{\mathrm{op}}
 >2L_n\sqrt{2(t\log9+u)}
\right)
\leq2\cdot9^t e^{-(t\log9+u)}
=2e^{-u}.
\end{equation}
The factor $2$ in the threshold in
\eqref{gcrtb:eq:ridge-matrix-concentration} comes from the matrix
inequality in \eqref{gcrtb:eq:net-norm-bounds}.
For any fixed error probability, choose $u$ sufficiently large.
Since $t\geq1$, we have $t\log9+u\leq(\log9+u)t$.
Since $L_n=O(n^{-1/2})$, we have 
$2L_n\sqrt{2(t\log9+u)}=O(\sqrt{t/n})$ uniformly over the population class.
This proves the second bound in \eqref{gcrtb:eq:ridge-concentration}.

\medskip
\noindent\textbf{Step 4: The trace concentration bound.}

We next prove \eqref{gcrtb:eq:trace-concentration}.
Since $\mathbf Q=\ell^2\mathbf R^2$, write
$\tau=\ell^2\operatorname{tr}(\mathbf R^2)/n$.
Equation~\eqref{gcrtb:eq:resolvent-differentials}, the symmetry of
$\mathbf R$, and the cyclic invariance of the trace give
\begin{equation}\label{gcrtb:eq:trace-differential}
d\tau
=\frac{2\ell^2}{n}\operatorname{tr}(\mathbf R\,d\mathbf R)
=-\frac{4\ell^2}{n}\operatorname{tr}(\mathbf T^\top\mathbf R^3d\mathbf T).
\end{equation}
Hence $\nabla_{\mathbf T}\tau=-4\ell^2\mathbf R^3\mathbf T/n$.
The bounds for $\mathbf R$ and $\mathbf R\mathbf T$ in
\eqref{gcrtb:eq:ridge-auxiliary-norm-bounds} imply
\begin{equation}\label{gcrtb:eq:trace-product-norm-bound}
\|\mathbf R^3\mathbf T\|_F
\leq\sqrt N\,\|\mathbf R\|_{\mathrm{op}}^2
             \|\mathbf R\mathbf T\|_{\mathrm{op}}
\leq\frac{\sqrt N}{2\ell^{5/2}}.
\end{equation}
For $\mathbf T=\mathbf G\mathbf\Sigma_{\mathcal S\mathcal S}^{1/2}/\sqrt n$,
the chain rule gives
$\nabla_{\mathbf G}\tau
=(\nabla_{\mathbf T}\tau)\mathbf\Sigma_{\mathcal S\mathcal S}^{1/2}/\sqrt n$.
The bound \eqref{gcrtb:eq:background-covariance-root-bound}
and the formula for $\nabla_{\mathbf T}\tau$ therefore yield
$$
\|\nabla_{\mathbf G}\tau\|_F
\leq\sqrt{\frac{s_+}{n}}\frac{4\ell^2}{n}
       \|\mathbf R^3\mathbf T\|_F
\leq\frac2n\sqrt{\frac{s_+}{\ell}},
$$
where the last inequality uses \eqref{gcrtb:eq:trace-product-norm-bound}
and $N\leq n$.
The Gaussian Poincar\'e inequality \eqref{gcrtb:eq:gaussian-poincare}
now gives
$\operatorname{Var}(\tau)\leq\mathbb E\|\nabla_{\mathbf G}\tau\|_F^2
\leq4s_+/(\ell n^2)$.
Since $\bar\tau_n=\mathbb E\tau$ by \eqref{gcrtb:eq:population},
Chebyshev's inequality yields, for every $B>0$,
$\mathbb P(|\tau-\bar\tau_n|>B/n)\leq4s_+/(\ell B^2)$.
This bound tends to zero as $B\to\infty$, uniformly in $n$
and the population, proving \eqref{gcrtb:eq:trace-concentration}.
\end{proof}

\subsubsection{Proof of Lemma~\ref{gcrtb:lem:gaussian-comparison}}
\label{app:gcrt-gaussian-comparison-proof}

\begin{proof}[Proof of Lemma~\ref{gcrtb:lem:gaussian-comparison}]
Step~1 verifies the conditional Gaussian law and independence assertions
for the comparison residual matrices generated by the oracle Gaussian
construction in
Appendix~\ref{app:gcrt-residual-approximation-proof}.
The remaining steps prove the discrepancy bound \eqref{gcrtint:eq:coupling}
for the coupled matrices $\mathbf X_c^{(v)}$ and
$\overline{\mathbf X}^{(v)}$, $v=0,\ldots,M$.

For this bound, Steps~2--3 compare the actual residual rotation
\eqref{eq:ggm-local-update} with the update
\eqref{gcrtb:eq:gaussian-update} of the oracle construction.
Both are evaluated at the same actual state, using the same fresh
Gaussian vector, to obtain the local error and bound its second moment.
We add $\mathbf M_j$ back to the output of
\eqref{gcrtb:eq:gaussian-update} at this actual state to obtain the
intermediate column used below. The oracle construction continues
from its own comparison state.
Step~4 propagates the existing difference between the two states and
incorporates the local errors to bound their discrepancy after one sweep.
Step~5 accumulates these errors over the sweeps and branches.

\medskip
\noindent\textbf{Step 1: Distributional properties of the comparison states.}

By \eqref{gcrtint:eq:centered-model}, $\mathbf E$ has independent
$\mathcal N_t(0,\mathbf V)$ rows and is independent of $\mathbf S$.
For $j=1,\ldots,t$, define
$\bs\Xi_j=\mathbf E_j-\sum_{k\ne j}J_{jk}\mathbf E_k$.
Recall that $\theta_j=\Theta_{jj}$ was defined before
\eqref{gcrtb:eq:gaussian-update}, and that the definition of $\mathbf J$
in Section~\ref{sec:gcrt-ridge-power} gives
$J_{jk}=-\Theta_{jk}/\theta_j$ for $k\ne j$.
Using $\mathbf\Theta=\mathbf V^{-1}$ from
\eqref{gcrtint:eq:target-precision}, the Gaussian conditional distribution
formula, applied row by row, gives
\begin{equation}\label{gcrtb:eq:innovation-conditional-law}
\mathcal L(\bs\Xi_j\mid\mathbf S,\mathbf E_{-j})
=\mathcal N_N(0,\theta_j^{-1}\mathbf I_N).
\end{equation}
The law in \eqref{gcrtb:eq:innovation-conditional-law} does not depend on $(\mathbf S,\mathbf E_{-j})$,
so $\bs\Xi_j$ is independent of $(\mathbf S,\mathbf E_{-j})$.
To verify preservation of the conditional law, suppose that the current
comparison state $\overline{\mathbf E}$ has independent
$\mathcal N_t(0,\mathbf V)$ rows conditional on $\mathbf S$.
Define its innovation
$\overline{\bs\Xi}_j=\overline{\mathbf E}_j-
\sum_{k\ne j}J_{jk}\overline{\mathbf E}_k$.
By \eqref{gcrtb:eq:innovation-conditional-law} applied to the current
comparison state $\overline{\mathbf E}$, the innovation
$\overline{\bs\Xi}_j$ has law
$\mathcal N_N(0,\theta_j^{-1}\mathbf I_N)$ conditional on
$(\mathbf S,\overline{\mathbf E}_{-j})$.
Recall from the definitions preceding
\eqref{gcrtb:eq:gaussian-update} that $\mathbf P_j$ projects onto
the span of the background columns indexed by $N_j\cap\mathcal S$,
and $\mathbf R_j=\mathbf I_N-\mathbf P_j$.
The update \eqref{gcrtb:eq:gaussian-update} therefore gives
the updated innovation as
\begin{equation}\label{gcrtb:eq:comparison-innovation-update}
\overline{\mathbf E}'_j-\sum_{k\ne j}J_{jk}\overline{\mathbf E}_k
=\mathbf P_j\overline{\bs\Xi}_j+
 \theta_j^{-1/2}\mathbf R_j\bs Z_j.
\end{equation}
Conditional on $(\mathbf S,\overline{\mathbf E}_{-j})$,
$\mathbf P_j\overline{\bs\Xi}_j$ and
$\theta_j^{-1/2}\mathbf R_j\bs Z_j$ in
\eqref{gcrtb:eq:comparison-innovation-update}
are independent because $\bs Z_j$ is independent of
$(\mathbf S,\overline{\mathbf E})$.
Since $\mathbf P_j$ and $\mathbf R_j$ are orthogonal projections fixed by
$\mathbf S$, the sum in \eqref{gcrtb:eq:comparison-innovation-update}
is centred Gaussian with covariance
$\theta_j^{-1}(\mathbf P_j+\mathbf R_j)=\theta_j^{-1}\mathbf I_N$.
Thus $\overline{\mathbf E}'_j-\sum_{k\ne j}J_{jk}\overline{\mathbf E}_k$
and $\overline{\bs\Xi}_j$ have the same law conditional on
$(\mathbf S,\overline{\mathbf E}_{-j})$.
All other columns remain fixed, so the update preserves the conditional
law of column $j$ given $(\mathbf S,\overline{\mathbf E}_{-j})$,
and hence the conditional law of the entire matrix given $\mathbf S$.
Starting from $\overline{\mathbf E}^{(0)}=\mathbf E$, induction over the
updates shows that every comparison state has independent
$\mathcal N_t(0,\mathbf V)$ rows conditional on $\mathbf S$.
All comparison states and $\mathbf E$ are jointly Gaussian conditional on
$\mathbf S$, because the updates are linear with coefficients determined
by $\mathbf S$. They are independent of $\bs\varepsilon$, since the
construction uses only the covariates and fresh update randomness.

\medskip
\noindent\textbf{Step 2: Error in one coupled update.}

It remains to prove
\eqref{gcrtint:eq:coupling}. We first bound the local error from the actual
residual rotation \eqref{eq:ggm-local-update}, relative to the update
\eqref{gcrtb:eq:gaussian-update} of the oracle construction evaluated
at the same current actual state. We then control how these errors
propagate through the sweeps and branches.

By Lemma~\ref{lem:ggm-local-update}, each actual update preserves
the Gaussian covariate law, also conditional on $\mathbf S$.
For this calculation, write the current actual target matrix as
$\mathbf M+\mathbf E$.
Here $\mathbf E$ is the current residual matrix, and its columns give the
innovation $\bs\Xi_j=\mathbf E_j-\sum_{k\ne j}J_{jk}\mathbf E_k$.
We first show that the conditional mean of
$\mathbf M_j+\mathbf E_j$ given $(\mathbf S,\mathbf E_{-j})$
lies in the span of the supplied neighbour columns, so that
the actual residual rotation leaves this mean unchanged.
Row $j$ of $\mathbf A=\mathbf\Theta\mathbf C$, together with
$J_{jk}=-\Theta_{jk}/\theta_j$ and $\theta_j=\Theta_{jj}$, gives
$\mathbf C_{j\cdot}-\sum_{k\ne j}J_{jk}\mathbf C_{k\cdot}
=\mathbf A_{j\cdot}/\theta_j$.
Transpose both sides and multiply on the left by $\mathbf S$ to obtain
\begin{equation}\label{gcrtb:eq:mean-column-identity}
\mathbf M_j-\sum_{k\ne j}J_{jk}\mathbf M_k
=\mathbf S\mathbf A_{j\cdot}^\top/\theta_j,
\end{equation}
where $\mathbf M_k=\mathbf S\mathbf C_{k\cdot}^\top$ is the $k$th column
of $\mathbf M=\mathbf S\mathbf C^\top$.
The supplied graph contains the conditional independence graph, so
$A_{jk}=0$ for every $k\in\mathcal S\setminus N_j$.
Consequently, $\mathbf S\mathbf A_{j\cdot}^\top/\theta_j$ is a linear
combination of the supplied background neighbour columns and lies
in the range of $\mathbf P_j$.
Let $\bs\mu_j$ be the conditional mean of $\mathbf M_j+\mathbf E_j$
given $(\mathbf S,\mathbf E_{-j})$.
By \eqref{gcrtb:eq:innovation-conditional-law} applied to the current
actual residual matrix $\mathbf E$, we have
\begin{equation}\label{gcrtb:eq:conditional-column-mean}
\begin{aligned}
\bs\mu_j
&=\mathbf M_j+\sum_{k\ne j}J_{jk}\mathbf E_k\\
&=\frac{\mathbf S\mathbf A_{j\cdot}^\top}{\theta_j}
 +\sum_{k\ne j}J_{jk}(\mathbf M_k+\mathbf E_k),
\end{aligned}
\end{equation}
where the second equality uses \eqref{gcrtb:eq:mean-column-identity}.
Since the supplied graph contains the conditional independence graph,
$\Theta_{jk}=0$ for $k\in\mathcal T\setminus(N_j\cup\{j\})$, and hence
$J_{jk}=-\Theta_{jk}/\theta_j=0$.
The sum in \eqref{gcrtb:eq:conditional-column-mean} can therefore be
restricted to $k\in N_j\cap\mathcal T$.
Together with the background term, this shows that the conditional mean
lies in the span of the current columns indexed by $N_j$.
Let $\mathbf P_D$ be the orthogonal projection onto the span of
the columns of $\mathbf S$ indexed by $N_j\cap\mathcal S$
and the current target columns $\mathbf M_k+\mathbf E_k$,
$k\in N_j\cap\mathcal T$. Its range contains that of
$\mathbf P_j$.

For this update, let $\mathbf X\in\mathbb R^{n\times p}$ be the
current data matrix of the actual sampler, before the coordinate
transformation by $\mathbf U^\top$.
Let $\mathbf P_j^{\rm raw}$ be the orthogonal projection onto
$\operatorname{col}([\bs1_n,\mathbf X_{N_j}])$.
For the fixed basis matrix $\mathbf U$ in
Appendix~\ref{app:gcrt-common-conditions}, the definition of $\mathbf P_D$
gives $\operatorname{col}(\mathbf U^\top\mathbf X_{N_j})
=\operatorname{ran}(\mathbf P_D)$.
Thus the original projection space is the orthogonal sum
$\operatorname{span}(\bs1_n)\oplus\mathbf U\operatorname{ran}(\mathbf P_D)$.
Using $\mathbf U^\top\mathbf U=\mathbf I_N$ and
$\mathbf U\mathbf U^\top=\mathbf I_n-\bs1_n\bs1_n^\top/n$ gives, for
every $\bs x\in\mathbb R^n$,
\begin{equation}\label{gcrtb:eq:rotation-coordinates}
\begin{aligned}
\mathbf P_j^{\rm raw}
 &=\bs1_n\bs1_n^\top/n+\mathbf U\mathbf P_D\mathbf U^\top,\\
\mathbf U^\top\mathbf P_j^{\rm raw}&=\mathbf P_D\mathbf U^\top,\\
(\mathbf I_n-\mathbf P_j^{\rm raw})\bs x
 &=\mathbf U(\mathbf I_N-\mathbf P_D)\mathbf U^\top\bs x,\\
\|(\mathbf I_n-\mathbf P_j^{\rm raw})\bs x\|
 &=\|(\mathbf I_N-\mathbf P_D)\mathbf U^\top\bs x\|.
\end{aligned}
\end{equation}
For the fresh vector $\bs W\sim\mathcal N_n(0,\mathbf I_n)$ in
\eqref{eq:ggm-local-update}, define $\bs Z_j=\mathbf U^\top\bs W$.
Since $\mathbf U$ is deterministic and $\mathbf U^\top\mathbf U=\mathbf I_N$,
we see that $\bs Z_j\sim\mathcal N_N(0,\mathbf I_N)$ and is independent of the data
and all earlier update randomness.
Use this same $\bs Z_j$ in the update
\eqref{gcrtb:eq:gaussian-update} of the oracle construction.
We now rewrite the actual update \eqref{eq:ggm-local-update} in centred
coordinates to compare it with \eqref{gcrtb:eq:gaussian-update}.
Equation~\eqref{gcrtb:eq:rotation-coordinates}, applied to $\mathbf X_j$
with $\mathbf U^\top\mathbf X_j=\mathbf M_j+\mathbf E_j$, transforms
the fitted part into $\mathbf P_D(\mathbf M_j+\mathbf E_j)$ and the
data residual norm into
$\|(\mathbf I_N-\mathbf P_D)(\mathbf M_j+\mathbf E_j)\|$.
Applied to $\bs W$, with $\mathbf U^\top\bs W=\bs Z_j$, it transforms
the Gaussian residual vector into $(\mathbf I_N-\mathbf P_D)\bs Z_j$
and its norm into $\|(\mathbf I_N-\mathbf P_D)\bs Z_j\|$.
Both norm values are unchanged, so the normalising ratio and the
convention of leaving the column unchanged when the denominator is zero
are preserved.

The representation of $\bs\mu_j$ in \eqref{gcrtb:eq:conditional-column-mean} gives
$\bs\mu_j\in\operatorname{ran}(\mathbf P_D)$.
Since $\mathbf M_j+\mathbf E_j=\bs\mu_j+\bs\Xi_j$, we have
$(\mathbf I_N-\mathbf P_D)(\mathbf M_j+\mathbf E_j)
=(\mathbf I_N-\mathbf P_D)\bs\Xi_j$.
To compute the error introduced by one update, evaluate the actual
residual rotation \eqref{eq:ggm-local-update} at the current data matrix
and the update \eqref{gcrtb:eq:gaussian-update} of the oracle construction
at $\overline{\mathbf E}=\mathbf E$, using the same $\bs Z_j$
in centred coordinates.
Because $\operatorname{rank}(\mathbf P_D)\leq b_j+g_j<N$ and
$\bs Z_j$ is fresh,
$\|(\mathbf I_N-\mathbf P_D)\bs Z_j\|>0$ almost surely.
The residual rotation \eqref{eq:ggm-local-update} preserves the fitted
part and residual length and replaces the residual direction.
In centred coordinates, \eqref{gcrtb:eq:rotation-coordinates} therefore
gives the updated column as
$$
\underbrace{\bs\mu_j+\mathbf P_D\bs\Xi_j}_{\text{fitted part}}
+
\underbrace{\|(\mathbf I_N-\mathbf P_D)\bs\Xi_j\|}
           _{\text{residual length}}
\underbrace{
\frac{(\mathbf I_N-\mathbf P_D)\bs Z_j}
     {\|(\mathbf I_N-\mathbf P_D)\bs Z_j\|}
}_{\text{new unit direction}}.
$$
The update \eqref{gcrtb:eq:gaussian-update} of the oracle construction,
with $\mathbf M_j$ added back, gives the intermediate column
$\bs\mu_j+\mathbf P_j\bs\Xi_j+
\theta_j^{-1/2}(\mathbf I_N-\mathbf P_j)\bs Z_j$.
Subtract this intermediate column from the actual updated column and use
$\mathbf I_N-\mathbf P_j=(\mathbf P_D-\mathbf P_j)
 +(\mathbf I_N-\mathbf P_D)$.
The common mean cancels, giving the local error
\begin{equation}\label{gcrtb:eq:local-error}
\begin{aligned}
 \bs\zeta_j={}&\underbrace{(\mathbf P_D-\mathbf P_j)
 (\bs\Xi_j-\theta_j^{-1/2}\bs Z_j)}_{\text{(I)}}\\
 &+\underbrace{\left\{\frac{\|(\mathbf I_N-\mathbf P_D)\bs\Xi_j\|}
 {\|(\mathbf I_N-\mathbf P_D)\bs Z_j\|}-\theta_j^{-1/2}\right\}
 (\mathbf I_N-\mathbf P_D)\bs Z_j}_{\text{(II)}}.
\end{aligned}
\end{equation}
In Step~3, we use the orthogonality of (I) and (II) to compute
$\mathbb E\|\bs\zeta_j\|^2$.

\medskip
\noindent\textbf{Step 3: Second moment of the local error.}

We first determine the ranks of $\mathbf P_j$ and $\mathbf P_D$.
Since Lemma~\ref{lem:ggm-local-update} ensures that each actual update
preserves the Gaussian covariate law,
the $N\times(b_j+g_j)$ matrix of supplied neighbour columns has
independent Gaussian rows with positive definite covariance at each
actual state. Since $b_j+g_j<N$, this matrix has full column rank
almost surely. Consequently,
$\operatorname{rank}(\mathbf P_D)=b_j+g_j$ and
$\operatorname{rank}(\mathbf P_j)=b_j$.
For each fixed $n$, there are only finitely many updates, so these
rank equalities hold simultaneously throughout the construction
almost surely.
Since $\operatorname{ran}(\mathbf P_j)\subseteq
\operatorname{ran}(\mathbf P_D)$, the difference
$\mathbf P_D-\mathbf P_j$ is an orthogonal projection of rank $g_j$,
and $\mathbf I_N-\mathbf P_D$ has rank $d=N-b_j-g_j>0$.
To write (I) and (II) more compactly, let
$\mathbf A_j=\mathbf P_D-\mathbf P_j$ and
$\mathbf B_j=\mathbf I_N-\mathbf P_D$.
Conditional on $(\mathbf S,\mathbf E_{N_j\cap\mathcal T})$,
both projections are fixed, and $\sqrt{\theta_j}\bs\Xi_j$ and
$\bs Z_j$ are independent standard Gaussian vectors.
Term (I) is
$\mathbf A_j(\bs\Xi_j-\theta_j^{-1/2}\bs Z_j)$, a centred Gaussian
vector with conditional covariance $2\theta_j^{-1}\mathbf A_j$.
Thus
\begin{equation}\label{gcrtb:eq:local-first-component-moment}
\mathbb E\left(
 \|\text{(I)}\|^2
 \mid(\mathbf S,\mathbf E_{N_j\cap\mathcal T})\right)
=2\theta_j^{-1}\operatorname{tr}(\mathbf A_j)
=2\theta_j^{-1}g_j.
\end{equation}
Term (II) has norm
$|\|\mathbf B_j\bs\Xi_j\|-\theta_j^{-1/2}\|\mathbf B_j\bs Z_j\||$.
Let $\chi_d,\chi'_d$ be independent lengths of standard Gaussian
vectors in $\mathbb R^d$.
Since $\mathbf B_j$ has rank $d$, the scaled lengths
$\sqrt{\theta_j}\|\mathbf B_j\bs\Xi_j\|$ and
$\|\mathbf B_j\bs Z_j\|$, under the same conditioning, have the
joint law of $(\chi_d,\chi'_d)$.
Consequently,
\begin{equation}\label{gcrtb:eq:local-second-component-moment}
\mathbb E\left(
 \|\text{(II)}\|^2
 \mid(\mathbf S,\mathbf E_{N_j\cap\mathcal T})\right)
=\theta_j^{-1}\mathbb E(\chi_d-\chi'_d)^2.
\end{equation}
Terms (I) and (II) lie in the orthogonal ranges of
$\mathbf A_j$ and $\mathbf B_j$, respectively.
We therefore add \eqref{gcrtb:eq:local-first-component-moment} and
\eqref{gcrtb:eq:local-second-component-moment} and take expectations
to obtain the first equality in \eqref{gcrts:eq:local-moment}:
\begin{equation}\label{gcrts:eq:local-moment}
 \mathbb E\|\bs\zeta_j\|^2
 =2\theta_j^{-1}g_j+\theta_j^{-1}\mathbb E(\chi_d-\chi'_d)^2
 =2\theta_j^{-1}\{g_j+\operatorname{Var}(\chi_d)\}
 \leq2v_+(g_j+1),
\end{equation}
where the second equality uses independence and identical distribution of
$\chi_d,\chi'_d$, and the inequality uses
$\theta_j^{-1}\leq v_+$ (since
$\theta_j=(\mathbf V^{-1})_{jj}\geq v_+^{-1}$
under $0\prec\mathbf V\preceq v_+\mathbf I_t$)
and $\operatorname{Var}(\chi_d)\leq1$ for $d\geq1$
(by Gaussian Poincar\'e \citealp[Eq.~(2.16)]{ledoux1999logsobolev},
applied to the Euclidean norm with Lipschitz constant one).
The bound in \eqref{gcrts:eq:local-moment} holds at every update;
independence between errors at different updates is not required.

\medskip
\noindent\textbf{Step 4: Discrepancy after one sweep.}

Equation~\eqref{gcrts:eq:local-moment} bounds the local error at one update.
We now derive how a full sweep propagates the existing discrepancy
between the actual and comparison states and accumulates these local errors.

\smallskip
\noindent\emph{Propagation through the oracle construction.}

Fix $\mathbf S$. We first bound how much one forward or reverse sweep
of the oracle construction can amplify the difference between two initial
residual states when the same Gaussian vectors are used.
For a forward sweep, the common noise terms cancel in the difference
between the final states. Thus the linear operator for the conditional
mean also propagates the initial difference.

Let $\bs x_j\in\mathbb R^N$ be the initial residual column $j$,
for $j=1,\ldots,t$ in the predetermined order.
Let $\bs y_j$ be the conditional mean of column $j$ after the sweep,
given $(\mathbf S,\bs x_1,\ldots,\bs x_t)$.

When column $j$ is updated, the columns $k<j$ have already been
updated, whereas the columns $k>j$ still have their initial values.
The fresh Gaussian term in \eqref{gcrtb:eq:gaussian-update} has
conditional mean zero, so
\begin{equation}\label{gcrtb:eq:sweep-column-mean}
\bs y_j=\mathbf P_j\bs x_j+
\mathbf R_j\sum_{k<j}J_{jk}\bs y_k+
\mathbf R_j\sum_{k>j}J_{jk}\bs x_k,
\qquad j=1,\ldots,t.
\end{equation}
Let $\bs x,\bs y\in\mathbb R^{Nt}$ stack the columns
$\bs x_j$ and $\bs y_j$, respectively, in the predetermined order
$1,\ldots,t$. Define
$\mathcal P=\operatorname{diag}(\mathbf P_1,\ldots,\mathbf P_t)$ and
$\mathcal R=\operatorname{diag}(\mathbf R_1,\ldots,\mathbf R_t)
=\mathbf I_{Nt}-\mathcal P$.
The $j$th blocks of $(\mathbf J_<\otimes\mathbf I_N)\bs y$
and $(\mathbf J_>\otimes\mathbf I_N)\bs x$ are
$\sum_{k<j}J_{jk}\bs y_k$ and $\sum_{k>j}J_{jk}\bs x_k$,
respectively. Thus \eqref{gcrtb:eq:sweep-column-mean} gives
$$
\bs y=\mathcal P\bs x
+\mathcal R(\mathbf J_<\otimes\mathbf I_N)\bs y
+\mathcal R(\mathbf J_>\otimes\mathbf I_N)\bs x.
$$
We move the term involving $\bs y$ on the right to the left
and collect the terms involving $\bs x$ to obtain
\begin{equation}\label{gcrtb:eq:stacked-sweep-mean}
\{\mathbf I_{Nt}-\mathcal R(\mathbf J_<\otimes\mathbf I_N)\}\bs y
=\{\mathcal P+\mathcal R(\mathbf J_>\otimes\mathbf I_N)\}\bs x.
\end{equation}
The $(j,k)$ block of $\mathcal R(\mathbf J_<\otimes\mathbf I_N)$
is $J_{jk}\mathbf R_j$ for $k<j$ and $\mathbf0$ otherwise.
Consequently, the coefficient matrix
$\mathbf I_{Nt}-\mathcal R(\mathbf J_<\otimes\mathbf I_N)$ in
\eqref{gcrtb:eq:stacked-sweep-mean} is block lower triangular with
identity diagonal blocks, so it is invertible.
Thus $\bs y=\mathcal A_{\to}\bs x$, where
\begin{equation}\label{gcrtb:eq:block-sweep}
 \mathcal A_{\to}=
 \{\mathbf I_{Nt}-\mathcal R(\mathbf J_<\otimes\mathbf I_N)\}^{-1}
 \{\mathcal P+\mathcal R(\mathbf J_>\otimes\mathbf I_N)\}.
\end{equation}

For the reverse sweep, define $\mathcal A_{\leftarrow}$ by interchanging
$\mathbf J_<$ and $\mathbf J_>$ in \eqref{gcrtb:eq:block-sweep}.
The assumption $\max\{\|\mathbf J\|_1,\|\mathbf J\|_\infty\}\leq a<1$
in Lemma~\ref{gcrtb:lem:gaussian-comparison} bounds the maximum absolute
row and column sums of both $\mathbf J_<$ and $\mathbf J_>$ by $a$. Hence $\|\mathbf J_<\otimes\mathbf I_N\|_{\mathrm{op}}=\|\mathbf J_<\|_{\mathrm{op}}\leq a$, and $\|\mathbf J_>\otimes\mathbf I_N\|_{\mathrm{op}}=\|\mathbf J_>\|_{\mathrm{op}}\leq a$. 
Since $\mathcal P$ and $\mathcal R$ are orthogonal projections,
$\|\mathcal P\|_{\mathrm{op}},\|\mathcal R\|_{\mathrm{op}}\leq1$.
For every $\bs v\in\mathbb R^{Nt}$, the reverse triangle inequality gives
$$
\begin{aligned}
\|\{\mathbf I_{Nt}-\mathcal R(\mathbf J_<\otimes\mathbf I_N)\}\bs v\|
&\geq\|\bs v\|-\|\mathcal R(\mathbf J_<\otimes\mathbf I_N)\bs v\|\\
&\geq(1-a)\|\bs v\|.
\end{aligned}
$$
This implies
\begin{equation}\label{gcrtb:eq:sweep-inverse-bound}
\|\{\mathbf I_{Nt}-\mathcal R(\mathbf J_<\otimes\mathbf I_N)\}^{-1}\|_{\mathrm{op}}
\leq(1-a)^{-1}.
\end{equation}
We apply \eqref{gcrtb:eq:sweep-inverse-bound} to the first factor in
\eqref{gcrtb:eq:block-sweep}. Submultiplicativity and the triangle
inequality then give
\begin{equation}\label{gcrtb:eq:sweep-operator-norm-bound}
\begin{aligned}
\|\mathcal A_{\to}\|_{\mathrm{op}}
&\leq\frac{\|\mathcal P\|_{\mathrm{op}}
 +\|\mathcal R\|_{\mathrm{op}}\,
  \|\mathbf J_>\otimes\mathbf I_N\|_{\mathrm{op}}}{1-a}\\
&\leq\frac{1+a}{1-a}=b_a.
\end{aligned}
\end{equation}
The same argument with $\mathbf J_<$ and $\mathbf J_>$ interchanged
gives the reverse version of \eqref{gcrtb:eq:sweep-inverse-bound}
and $\|\mathcal A_{\leftarrow}\|_{\mathrm{op}}\leq b_a$.

We will use \eqref{gcrtb:eq:sweep-operator-norm-bound} to control the propagation
of the discrepancy between the actual and comparison states at the start
of a sweep, while \eqref{gcrtb:eq:sweep-inverse-bound} will control the
accumulated effect of the local errors introduced during that sweep.

\smallskip
\noindent\emph{Addition of the local errors.}

For $j=0,\ldots,t$, let
$\mathbf E^{[j]},\overline{\mathbf E}^{[j]}\in\mathbb R^{N\times t}$
be the actual and comparison residual matrices after the first $j$ updates
of this forward sweep. Define the differences at the start and end by
$$
\bs\Delta_{\rm in}
=\begin{pmatrix}
\mathbf E_1^{[0]}-\overline{\mathbf E}_1^{[0]}\\
\vdots\\
\mathbf E_t^{[0]}-\overline{\mathbf E}_t^{[0]}
\end{pmatrix},
\qquad
\bs\Delta_{\rm out}
=\begin{pmatrix}
\mathbf E_1^{[t]}-\overline{\mathbf E}_1^{[t]}\\
\vdots\\
\mathbf E_t^{[t]}-\overline{\mathbf E}_t^{[t]}
\end{pmatrix}.
$$
Each vector stacks the column differences in the order $1,\ldots,t$
and belongs to $\mathbb R^{Nt}$.

At update $j$, evaluate \eqref{gcrtb:eq:gaussian-update} at
$\mathbf E^{[j-1]}$, using the same $\bs Z_j$ as in the actual update.
Add $\mathbf M_j$ to the output to obtain the intermediate column from Step~2.
The actual updated column $\mathbf M_j+\mathbf E_j^{[j]}$ minus this
intermediate column is the local error $\bs\zeta_j$ in
\eqref{gcrtb:eq:local-error}.
Let $\bs\zeta\in\mathbb R^{Nt}$ stack these local errors in the order
$1,\ldots,t$.

Next evaluate \eqref{gcrtb:eq:gaussian-update} at
$\overline{\mathbf E}^{[j-1]}$.
Again use $\bs Z_j$ and add back $\mathbf M_j$.
This gives the updated comparison column
$\mathbf M_j+\overline{\mathbf E}_j^{[j]}$.
Subtract it from the intermediate column:
the common mean $\mathbf M_j$ and common Gaussian noise
$\theta_j^{-1/2}\mathbf R_j\bs Z_j$ cancel, and the difference is
\begin{equation}\label{gcrtb:eq:comparison-update-difference}
\mathbf P_j(\mathbf E_j^{[j-1]}-\overline{\mathbf E}_j^{[j-1]})
+\mathbf R_j\sum_{k\ne j}J_{jk}
(\mathbf E_k^{[j-1]}-\overline{\mathbf E}_k^{[j-1]}).
\end{equation}

Each column is updated once in this sweep.
Just before update $j$, columns $k<j$ already have their final values
for the sweep. Columns $k\geq j$ still have their initial values.
Thus
$$
\mathbf E_k^{[j-1]}-\overline{\mathbf E}_k^{[j-1]}
=\begin{cases}
\bs\Delta_{{\rm out},k},&k<j,\\
\bs\Delta_{{\rm in},k},&k\geq j.
\end{cases}
$$
After update $j$, its column no longer changes during this sweep, so
$\mathbf E_j^{[j]}-\overline{\mathbf E}_j^{[j]}
=\bs\Delta_{{\rm out},j}$.
We rewrite \eqref{gcrtb:eq:comparison-update-difference} using these
discrepancies and add the local error $\bs\zeta_j$ to obtain
$$
\bs\Delta_{{\rm out},j}
=\mathbf P_j\bs\Delta_{{\rm in},j}
 +\mathbf R_j\sum_{k<j}J_{jk}\bs\Delta_{{\rm out},k}
 +\mathbf R_j\sum_{k>j}J_{jk}\bs\Delta_{{\rm in},k}+\bs\zeta_j.
$$
We stack these equations as in \eqref{gcrtb:eq:stacked-sweep-mean}.
We move the terms involving $\bs\Delta_{{\rm out},k}$ with $k<j$
to the left. The vector $\bs\zeta$ remains on the right.
We then solve for $\bs\Delta_{\rm out}$ using
\eqref{gcrtb:eq:block-sweep}:
\begin{equation}\label{gcrtb:eq:discrepancy-sweep}
 \bs\Delta_{\rm out}
 =\mathcal A_{\to}\bs\Delta_{\rm in}
 +\{\mathbf I_{Nt}-\mathcal R(\mathbf J_<\otimes\mathbf I_N)\}^{-1}
 \bs\zeta.
\end{equation}
We apply the triangle inequality to \eqref{gcrtb:eq:discrepancy-sweep}
and use \eqref{gcrtb:eq:sweep-operator-norm-bound} and
\eqref{gcrtb:eq:sweep-inverse-bound} to obtain
\begin{equation}\label{gcrtb:eq:discrepancy-sweep-norm}
\|\bs\Delta_{\rm out}\|
\leq b_a\|\bs\Delta_{\rm in}\|+\frac{\|\bs\zeta\|}{1-a}.
\end{equation}
For a reverse sweep, the update order is $t,\ldots,1$, while the
stacking order remains $1,\ldots,t$.
Thus \eqref{gcrtb:eq:discrepancy-sweep} holds with $\mathbf J_<$ and
$\mathbf J_>$ interchanged and $\mathcal A_{\to}$ replaced by
$\mathcal A_{\leftarrow}$, using $\bs\Delta_{\rm in}$,
$\bs\Delta_{\rm out}$, and $\bs\zeta$ from this reverse sweep.
The norm bound \eqref{gcrtb:eq:discrepancy-sweep-norm} also holds.
For either direction, the errors $\bs\zeta_j$ occupy separate blocks of
$\bs\zeta$.
Thus \eqref{gcrts:eq:local-moment} gives
$\mathbb E\|\bs\zeta\|^2\leq2v_+(\sum_jg_j+t)$.
Step~5 will iterate \eqref{gcrtb:eq:discrepancy-sweep-norm} over the
$2L$ sweeps from the observation to each branch endpoint.

\medskip
\noindent\textbf{Step 5: Accumulation over sweeps and branches.}

We now iterate the bound \eqref{gcrtb:eq:discrepancy-sweep-norm}
from Step~4 over all $2L$ sweeps from the observation through the hub
to a branch endpoint.
Fix a branch $m\in\{1,\ldots,M\}$.
The first $L$ sweeps are forward and shared by all branches;
the last $L$ are the reverse sweeps of branch $m$.'

We introduce the following notation to write one recursion for these sweeps.
We will always stack column differences and local errors in column order
$1,\ldots,t$, for both forward and reverse sweeps.
\begin{itemize}
\item \emph{Discrepancy after $q$ sweeps.}
For $q=0,\ldots,2L$, subtract the comparison residual matrix from the
actual residual matrix after $q$ sweeps along this path, and let
$\bs\Delta^{(q)}\in\mathbb R^{Nt}$ stack the columns of this difference.

\item \emph{Local errors in sweep $q$.}
For $q=1,\ldots,2L$ and $j=1,\ldots,t$, let
$\bs\zeta_j^{(q)}\in\mathbb R^N$ be the actual updated column minus the
intermediate column from Step~2 when column $j$ is updated in sweep $q$,
as defined in \eqref{gcrtb:eq:local-error}.
Let $\bs\zeta^{(q)}\in\mathbb R^{Nt}$ stack these $t$ local errors.

\item \emph{Matrix for the existing discrepancy.}
Define $\mathcal A_q=\mathcal A_{\to}$ for $1\leq q\leq L$ and
$\mathcal A_q=\mathcal A_{\leftarrow}$ for $L<q\leq2L$.

\item \emph{Matrix for the local errors.}
For $1\leq q\leq L$, define
$\mathcal U_q=\{\mathbf I_{Nt}-\mathcal R(\mathbf J_<\otimes\mathbf I_N)\}^{-1}$.
For $L<q\leq2L$, define
$\mathcal U_q=\{\mathbf I_{Nt}-\mathcal R(\mathbf J_>\otimes\mathbf I_N)\}^{-1}$.
\end{itemize}

Equation~\eqref{gcrtb:eq:discrepancy-sweep} and its reverse version give
$\bs\Delta^{(q)}=\mathcal A_q\bs\Delta^{(q-1)}+\mathcal U_q\bs\zeta^{(q)}$,
with $\bs\Delta^{(0)}=0$.
Equations~\eqref{gcrtb:eq:sweep-operator-norm-bound} and
\eqref{gcrtb:eq:sweep-inverse-bound}, with their reverse versions,
give $\|\mathcal A_q\|_{\mathrm{op}}\leq b_a$ and
$\|\mathcal U_q\|_{\mathrm{op}}\leq(1-a)^{-1}$.
Thus \eqref{gcrtb:eq:discrepancy-sweep-norm} gives, for $q=1,\ldots,2L$,
$$
\|\bs\Delta^{(q)}\|
\leq b_a\|\bs\Delta^{(q-1)}\|
 +\frac{\|\bs\zeta^{(q)}\|}{1-a}.
$$
Starting from $\bs\Delta^{(0)}=0$, the first two applications give $\|\bs\Delta^{(1)}\|\leq\frac{\|\bs\zeta^{(1)}\|}{1-a}$ and 
$$
\begin{aligned}
\|\bs\Delta^{(2)}\| &\leq b_a\|\bs\Delta^{(1)}\|
 +\frac{\|\bs\zeta^{(2)}\|}{1-a}\\
&\leq\frac{b_a\|\bs\zeta^{(1)}\|+\|\bs\zeta^{(2)}\|}{1-a}.
\end{aligned}
$$
We continue this substitution through all $2L$ sweeps to obtain
\begin{equation}\label{gcrtb:eq:iterated-discrepancy-bound}
\|\bs\Delta^{(2L)}\|
\leq\frac{1}{1-a}\sum_{q=1}^{2L}
 b_a^{2L-q}\|\bs\zeta^{(q)}\|.
\end{equation}
Squaring \eqref{gcrtb:eq:iterated-discrepancy-bound} and applying
Cauchy--Schwarz to the $2L$ scalar terms in its sum, we have 
$$
\begin{aligned}
\|\bs\Delta^{(2L)}\|^2
&\leq \frac{1}{(1-a)^2}
 \left(\sum_{q=1}^{2L}b_a^{2L-q}\|\bs\zeta^{(q)}\|\right)^2\\
&\leq \frac{2L}{(1-a)^2}
 \sum_{q=1}^{2L}b_a^{2(2L-q)}\|\bs\zeta^{(q)}\|^2.
\end{aligned}
$$

Taking expectations, we have 
\begin{equation}\label{gcrtb:eq:accumulated-discrepancy-second-moment}
\mathbb E\|\bs\Delta^{(2L)}\|^2
\leq\frac{2L}{(1-a)^2}\sum_{q=1}^{2L}
 b_a^{2(2L-q)}\mathbb E\|\bs\zeta^{(q)}\|^2.
\end{equation}

The local errors $\bs\zeta_j^{(q)}$ occupy separate coordinate blocks of
$\bs\zeta^{(q)}$, so its squared norm is the sum of their squared norms.
Equation~\eqref{gcrts:eq:local-moment} therefore gives, for each $q$,
\begin{equation}\label{gcrtb:eq:sweep-local-error-second-moment}
\begin{aligned}
\mathbb E\|\bs\zeta^{(q)}\|^2
=\sum_{j=1}^t\mathbb E\|\bs\zeta_j^{(q)}\|^2
&\leq2v_+\sum_{j=1}^t(g_j+1)
 =2v_+\left(\sum_{j=1}^t g_j+t\right).
\end{aligned}
\end{equation}
We substitute \eqref{gcrtb:eq:sweep-local-error-second-moment} into
\eqref{gcrtb:eq:accumulated-discrepancy-second-moment} to conclude that
$$
\mathbb E\|\bs\Delta^{(2L)}\|^2
\leq\frac{4Lv_+}{(1-a)^2}
 \left(\sum_{q=1}^{2L}b_a^{2(2L-q)}\right)
 \left(\sum_jg_j+t\right).
$$
At the endpoint,
$\|\bs\Delta^{(2L)}\|^2
=\|\mathbf X_c^{(m)}-(\mathbf M+\overline{\mathbf E}^{(m)})\|_F^2$.
The discrepancy for $v=0$ is zero, and the maximum of the $M$
squared Frobenius discrepancies is bounded by their sum.
The same bound applies to every branch, so
\begin{equation}\label{gcrtb:eq:coupling-maximum-second-moment}
 \mathbb E\max_{0\leq v\leq M}
 \|\mathbf X_c^{(v)}-(\mathbf M+\overline{\mathbf E}^{(v)})\|_F^2
 \leq C_{L,M,a,v_+}\left(\sum_jg_j+t\right).
\end{equation}
The argument uses only pointwise norm bounds and expectations of sums,
so it does not require independence among sweep errors or branches.
By \eqref{gcrtb:eq:error-rates}, $ne_n^2=\sum_jg_j+t$.
Markov's inequality, applied to the squared maximum in
\eqref{gcrtb:eq:coupling-maximum-second-moment}, therefore gives,
for every $B>0$,
$$
\mathbb P\left(
\max_{0\leq v\leq M}
\|\mathbf X_c^{(v)}-(\mathbf M+\overline{\mathbf E}^{(v)})\|_F
>B\sqrt n\,e_n
\right)
\leq\frac{C_{L,M,a,v_+}}{B^2}.
$$
The constant depends only on the fixed $L,M,a,v_+$.
Thus the right-hand side tends to zero as $B\to\infty$, uniformly
in $n$ and over Gaussian models satisfying the lemma's assumptions.
This proves \eqref{gcrtint:eq:coupling} with the stated uniformity.
\end{proof}

\subsubsection{Proof of Lemma~\ref{gcrtb:lem:comparison-cross-moment}}
\label{app:gcrt-comparison-cross-moment-proof}

\begin{proof}[Proof of Lemma~\ref{gcrtb:lem:comparison-cross-moment}]
Use the background projections $\mathbf P_j$ in
\eqref{gcrtb:eq:gaussian-update} and the sweep operators
$\mathcal A_{\to}$ and $\mathcal A_{\leftarrow}$ defined in
\eqref{gcrtb:eq:block-sweep} and its reverse version.
To control
$\mathbb E(\overline{\mathbf E}^{(m)\top}\mathbf Q\mathbf E\mid\mathbf S)/n$,
we compare
$\mathcal A_{\to}$ and $\mathcal A_{\leftarrow}$ with
$\mathbf F_{\to}\otimes\mathbf I_N$ and
$\mathbf F_{\leftarrow}\otimes\mathbf I_N$, respectively.
Recall from \eqref{gcrtb:eq:error-rates} that
$b_{\max}=\max_{j\in\mathcal T}b_j$, where
$b_j=|N_j\cap\mathcal S|$ is the number of background neighbours of target $j$.
We will use the rank bounds on the background projections $\mathbf P_j$
to prove
$$
\left\|\frac1n\mathbb E(\overline{\mathbf E}^{(m)\top}\mathbf Q\mathbf E
\mid\mathbf S)-\tau\mathbf F\mathbf V\right\|_{\mathrm{op}}
\leq C_Lv_+d_n=C_Lv_+\frac{b_{\max}}n,
\qquad 1\leq m\leq M,
$$
where $C_L$ depends only on $L$ and $a$.

\medskip
\noindent\textbf{Step 1: Comparison of individual sweep operators.}

For a block matrix $\mathcal B=(\mathcal B_{jk})_{j,k=1}^t$
with $\mathcal B_{jk}\in\mathbb R^{N\times N}$, define
$$
\begin{aligned}
\alpha_b(\mathcal B)
&=\max\left\{\max_j\sum_{k=1}^t\|\mathcal B_{jk}\|_{\mathrm{op}},
             \max_k\sum_{j=1}^t\|\mathcal B_{jk}\|_{\mathrm{op}}\right\},\\
\nu_b(\mathcal B)
&=\max\left\{\max_j\sum_{k=1}^t\|\mathcal B_{jk}\|_*,
             \max_k\sum_{j=1}^t\|\mathcal B_{jk}\|_*\right\},
\end{aligned}
$$
where $\|\cdot\|_*$ is the nuclear norm (the sum of the singular values).
We will bound the differences
$\mathcal A_{\to}-\mathbf F_{\to}\otimes\mathbf I_N$ and
$\mathcal A_{\leftarrow}-\mathbf F_{\leftarrow}\otimes\mathbf I_N$
in $\nu_b$, and the four individual sweep operators in $\alpha_b$.
These bounds will control the factors in the sweep products in Step~2.
For two such matrices $\mathcal B$ and $\mathcal C$,
$(\mathcal B\mathcal C)_{jk}=\sum_{l=1}^t\mathcal B_{jl}\mathcal C_{lk}$, so
$$
\begin{aligned}
\|(\mathcal B\mathcal C)_{jk}\|_{\mathrm{op}}
&\leq\sum_{l=1}^t\|\mathcal B_{jl}\|_{\mathrm{op}}
                    \|\mathcal C_{lk}\|_{\mathrm{op}},\\
\|(\mathcal B\mathcal C)_{jk}\|_*
&\leq\sum_{l=1}^t\|\mathcal B_{jl}\|_{\mathrm{op}}
                    \|\mathcal C_{lk}\|_*,\\
\|(\mathcal B\mathcal C)_{jk}\|_*
&\leq\sum_{l=1}^t\|\mathcal B_{jl}\|_*
                    \|\mathcal C_{lk}\|_{\mathrm{op}}.
\end{aligned}
$$
Taking the maximum of the block row and column sums gives
\begin{equation}\label{gcrtb:eq:block-norm-products}
\begin{aligned}
\alpha_b(\mathcal B\mathcal C)&\leq\alpha_b(\mathcal B)\alpha_b(\mathcal C),\\
\nu_b(\mathcal B\mathcal C)&\leq
\min\{\alpha_b(\mathcal B)\nu_b(\mathcal C),
\nu_b(\mathcal B)\alpha_b(\mathcal C)\}.
\end{aligned}
\end{equation}
To bound the inverse factor in \eqref{gcrtb:eq:block-sweep}, let
$\mathcal K=\mathcal R(\mathbf J_<\otimes\mathbf I_N)$.
The block description following \eqref{gcrtb:eq:stacked-sweep-mean} gives
$$
\mathcal K_{jk}=
\begin{cases}
J_{jk}\mathbf R_j,&k<j,\\
\mathbf0,&k\geq j.
\end{cases}
$$
Since $\|\mathbf R_j\|_{\mathrm{op}}\leq1$, for each $j=1,\ldots,t$,
$$
\begin{aligned}
\sum_{k=1}^t\|\mathcal K_{jk}\|_{\mathrm{op}}
&=\sum_{k<j}|J_{jk}|\|\mathbf R_j\|_{\mathrm{op}}\\
&\leq\sum_{k<j}|J_{jk}|
\leq\|\mathbf J_<\|_\infty\leq a.
\end{aligned}
$$
Similarly, $\sum_{j=1}^t\|\mathcal K_{jk}\|_{\mathrm{op}}
\leq\|\mathbf J_<\|_1\leq a$ for each $k=1,\ldots,t$.
Thus
\begin{equation}\label{gcrtb:eq:lower-block-norm-bound}
\alpha_b(\mathcal K)\leq a.
\end{equation}

As established after \eqref{gcrtb:eq:stacked-sweep-mean},
$\mathbf I_{Nt}-\mathcal K$ is block lower triangular with identity
diagonal blocks and is therefore invertible. Using the identity
$
(\mathbf I_{Nt}-\mathcal K)^{-1}
=\mathbf I_{Nt}+\mathcal K(\mathbf I_{Nt}-\mathcal K)^{-1}
$, the triangle inequality, \eqref{gcrtb:eq:block-norm-products},
and $\alpha_b(\mathbf I_{Nt})=1$, we have 
$$
\begin{aligned}
\alpha_b((\mathbf I_{Nt}-\mathcal K)^{-1})
&\leq1+\alpha_b(\mathcal K)
       \alpha_b((\mathbf I_{Nt}-\mathcal K)^{-1})\\
&\leq1+a\alpha_b((\mathbf I_{Nt}-\mathcal K)^{-1}).
\end{aligned}
$$
Since $a<1$, rearrangement yields
\begin{equation}\label{gcrtb:eq:sweep-inverse-block-bound}
\alpha_b((\mathbf I_{Nt}-\mathcal K)^{-1})\leq(1-a)^{-1}.
\end{equation}

We next bound $\mathbf F_{\to}\otimes\mathbf I_N$ and
$\mathbf I_{Nt}-\mathbf F_{\to}\otimes\mathbf I_N$ in $\alpha_b$.
If we replace each $\mathbf R_j$ by $\mathbf I_N$, then the row and
column sum calculation leading to \eqref{gcrtb:eq:lower-block-norm-bound}
gives $\alpha_b(\mathbf J_<\otimes\mathbf I_N)\leq a$.
If we also interchange $\mathbf J_<$ and $\mathbf J_>$, then we obtain
$\alpha_b(\mathbf J_>\otimes\mathbf I_N)\leq a$.
If we replace $\mathcal K$ by $\mathbf J_<\otimes\mathbf I_N$, then the
inverse argument leading to \eqref{gcrtb:eq:sweep-inverse-block-bound} gives
$\alpha_b((\mathbf I_{Nt}-\mathbf J_<\otimes\mathbf I_N)^{-1})
\leq(1-a)^{-1}$.
Equation~\eqref{gcrtint:eq:scan} gives
$\mathbf F_{\to}=(\mathbf I_t-\mathbf J_<)^{-1}\mathbf J_>$.
Consequently,
\begin{equation}\label{gcrtb:eq:forward-sweep-factorization}
\mathbf F_{\to}\otimes\mathbf I_N
=(\mathbf I_{Nt}-\mathbf J_<\otimes\mathbf I_N)^{-1}
 (\mathbf J_>\otimes\mathbf I_N).
\end{equation}
We apply the product bound in \eqref{gcrtb:eq:block-norm-products}
to \eqref{gcrtb:eq:forward-sweep-factorization} and obtain
\begin{equation}\label{gcrtb:eq:forward-sweep-block-bounds}
\alpha_b(\mathbf F_{\to}\otimes\mathbf I_N)
\leq\alpha_b((\mathbf I_{Nt}-\mathbf J_<\otimes\mathbf I_N)^{-1})
    \alpha_b(\mathbf J_>\otimes\mathbf I_N)
\leq\frac{a}{1-a}.
\end{equation}
The triangle inequality and \eqref{gcrtb:eq:forward-sweep-block-bounds}
then give
\begin{equation}\label{gcrtb:eq:forward-sweep-complement-block-bound}
\alpha_b(\mathbf I_{Nt}-\mathbf F_{\to}\otimes\mathbf I_N)
\leq1+\alpha_b(\mathbf F_{\to}\otimes\mathbf I_N)
\leq\frac1{1-a}.
\end{equation}

To isolate the background projection in the difference between
the forward sweep operators, we compute
$$
\begin{aligned}
&\{\mathbf I_{Nt}-\mathcal R(\mathbf J_<\otimes\mathbf I_N)\}
 (\mathcal A_{\to}-\mathbf F_{\to}\otimes\mathbf I_N)\\
\stackrel{(1)}{=}{}&
 \mathcal P+\mathcal R\{(\mathbf J_>+\mathbf J_<\mathbf F_{\to})
 \otimes\mathbf I_N\}-\mathbf F_{\to}\otimes\mathbf I_N\\
\stackrel{(2)}{=}{}&
 \mathcal P+\mathcal R(\mathbf F_{\to}\otimes\mathbf I_N)
 -\mathbf F_{\to}\otimes\mathbf I_N\\
\stackrel{(3)}{=}{}&
 \mathcal P(\mathbf I_{Nt}-\mathbf F_{\to}\otimes\mathbf I_N),
\end{aligned}
$$
where step~(1) uses \eqref{gcrtb:eq:block-sweep}, step~(2) uses
$\mathbf J_>+\mathbf J_<\mathbf F_{\to}=\mathbf F_{\to}$ from
\eqref{gcrtint:eq:scan}, and step~(3) uses
$\mathcal R=\mathbf I_{Nt}-\mathcal P$.
Multiplication on the left by
$\{\mathbf I_{Nt}-\mathcal R(\mathbf J_<\otimes\mathbf I_N)\}^{-1}$
then gives
\begin{equation}\label{gcrtb:eq:projection-error}
 \mathcal A_{\to}-\mathbf F_{\to}\otimes\mathbf I_N
 =\{\mathbf I_{Nt}-\mathcal R(\mathbf J_<\otimes\mathbf I_N)\}^{-1}
 \mathcal P(\mathbf I_{Nt}-\mathbf F_{\to}\otimes\mathbf I_N).
\end{equation}
For the middle factor in \eqref{gcrtb:eq:projection-error},
$\mathcal P=\operatorname{diag}(\mathbf P_1,\ldots,\mathbf P_t)$ gives
\begin{equation}\label{gcrtb:eq:background-projection-block-bound}
\nu_b(\mathcal P)=\max_j\|\mathbf P_j\|_*
=\max_j\operatorname{rank}(\mathbf P_j)\leq b_{\max},
\end{equation}
where each $\mathbf P_j$ is an orthogonal projection of rank at most $b_j$.
We therefore have
$$
\begin{aligned}
\nu_b(\mathcal A_{\to}-\mathbf F_{\to}\otimes\mathbf I_N) {}
={}&
\nu_b((\mathbf I_{Nt}-\mathcal K)^{-1}
       \mathcal P(\mathbf I_{Nt}-\mathbf F_{\to}\otimes\mathbf I_N))\\
\stackrel{(1)}{\leq}{}&
\alpha_b((\mathbf I_{Nt}-\mathcal K)^{-1})
\nu_b(\mathcal P(\mathbf I_{Nt}-\mathbf F_{\to}\otimes\mathbf I_N))\\
\stackrel{(2)}{\leq}{}&
\alpha_b((\mathbf I_{Nt}-\mathcal K)^{-1})
\nu_b(\mathcal P)
\alpha_b(\mathbf I_{Nt}-\mathbf F_{\to}\otimes\mathbf I_N)\\
\stackrel{(3)}{\leq}{}&
\frac1{1-a}\,b_{\max}\,\frac1{1-a}
=\frac{b_{\max}}{(1-a)^2},
\end{aligned}
$$
where step~(1) uses
$\nu_b(\mathcal B\mathcal C)\leq\alpha_b(\mathcal B)\nu_b(\mathcal C)$
in \eqref{gcrtb:eq:block-norm-products}.
Step~(2) uses
$\nu_b(\mathcal B\mathcal C)\leq\nu_b(\mathcal B)\alpha_b(\mathcal C)$
in \eqref{gcrtb:eq:block-norm-products}.
Step~(3) uses \eqref{gcrtb:eq:sweep-inverse-block-bound},
\eqref{gcrtb:eq:background-projection-block-bound}, and
\eqref{gcrtb:eq:forward-sweep-complement-block-bound} for the three factors,
respectively.

For the norm of $\mathcal A_{\to}$ itself, note that
$\alpha_b(\mathcal P),\alpha_b(\mathcal R)\leq1$ because their diagonal
blocks are orthogonal projections. By \eqref{gcrtb:eq:block-norm-products}
and $\alpha_b(\mathbf J_>\otimes\mathbf I_N)\leq a$, we have
\begin{equation}\label{gcrtb:eq:upper-block-norm-bound}
\alpha_b(\mathcal R(\mathbf J_>\otimes\mathbf I_N))
\leq\alpha_b(\mathcal R)\alpha_b(\mathbf J_>\otimes\mathbf I_N)
\leq1\cdot a=a.
\end{equation}
By \eqref{gcrtb:eq:block-sweep} and
$\mathcal K=\mathcal R(\mathbf J_<\otimes\mathbf I_N)$, we have
$$
\begin{aligned}
&\alpha_b(\mathcal A_{\to})\\
={}&\alpha_b((\mathbf I_{Nt}-\mathcal K)^{-1}
\{\mathcal P+\mathcal R(\mathbf J_>\otimes\mathbf I_N)\})\\
\stackrel{(1)}{\leq}{}&
\alpha_b((\mathbf I_{Nt}-\mathcal K)^{-1})
\alpha_b(\mathcal P+\mathcal R(\mathbf J_>\otimes\mathbf I_N))\\
\stackrel{(2)}{\leq}{}&
\alpha_b((\mathbf I_{Nt}-\mathcal K)^{-1})
\{\alpha_b(\mathcal P)
 +\alpha_b(\mathcal R(\mathbf J_>\otimes\mathbf I_N))\}\\
\stackrel{(3)}{\leq}{}&
\frac1{1-a}(1+a)=\frac{1+a}{1-a},
\end{aligned}
$$
where step~(1) uses the product bound in
\eqref{gcrtb:eq:block-norm-products}, step~(2) uses the triangle
inequality, and step~(3) uses \eqref{gcrtb:eq:sweep-inverse-block-bound},
$\alpha_b(\mathcal P)\leq1$, and \eqref{gcrtb:eq:upper-block-norm-bound}.

Interchanging $\mathbf J_<$ and $\mathbf J_>$ gives the same bounds
for the reverse sweep:
$\mathbf I_{Nt}-\mathcal R(\mathbf J_>\otimes\mathbf I_N)$ is block
upper triangular with identity diagonal blocks, and the row and
column sum bounds are unchanged. 

In summary, the two bounds needed in Step~2 are
\begin{equation}\label{gcrtb:eq:sweep-difference-block-bound}
\max\{\nu_b(\mathcal A_{\to}-\mathbf F_{\to}\otimes\mathbf I_N),
       \nu_b(\mathcal A_{\leftarrow}-\mathbf F_{\leftarrow}\otimes\mathbf I_N)\}
\leq\frac{b_{\max}}{(1-a)^2}
\end{equation}
and
\begin{equation}\label{gcrtb:eq:sweep-operator-block-bound}
\begin{aligned}
&\max\{\alpha_b(\mathcal A_{\to}),\alpha_b(\mathcal A_{\leftarrow}),
\alpha_b(\mathbf F_{\to}\otimes\mathbf I_N),
\alpha_b(\mathbf F_{\leftarrow}\otimes\mathbf I_N)\}\leq{}& \frac{1+a}{1-a}=: b_a.
\end{aligned}
\end{equation}

\medskip
\noindent\textbf{Step 2: Comparison of the sweep products.}

Define $\mathcal F=\mathcal A_{\leftarrow}^{L}\mathcal A_{\to}^{L}$.
For $q=1,\ldots,L$, let $\mathcal A_q=\mathcal A_{\to}$ and
$\mathcal B_q=\mathbf F_{\to}\otimes\mathbf I_N$;
for $q=L+1,\ldots,2L$, let $\mathcal A_q=\mathcal A_{\leftarrow}$ and
$\mathcal B_q=\mathbf F_{\leftarrow}\otimes\mathbf I_N$.
Thus $\mathcal F=\mathcal A_{2L}\cdots\mathcal A_1$ and, by
\eqref{gcrtint:eq:scan},
$\mathbf F\otimes\mathbf I_N=\mathcal B_{2L}\cdots\mathcal B_1$.
The telescoping identity gives
\begin{equation}\label{gcrtb:eq:sweep-product-telescope}
\mathcal F-\mathbf F\otimes\mathbf I_N
=\sum_{q=1}^{2L}\mathcal A_{2L}\cdots\mathcal A_{q+1}
 (\mathcal A_q-\mathcal B_q)\mathcal B_{q-1}\cdots\mathcal B_1.
\end{equation}
In \eqref{gcrtb:eq:sweep-product-telescope}, any empty product
on either side of $(\mathcal A_q-\mathcal B_q)$
is taken to be $\mathbf I_{Nt}$.
In each summand of \eqref{gcrtb:eq:sweep-product-telescope},
\eqref{gcrtb:eq:sweep-difference-block-bound} gives
$\nu_b(\mathcal A_q-\mathcal B_q)\leq b_{\max}/(1-a)^2$,
and \eqref{gcrtb:eq:sweep-operator-block-bound} bounds the $\alpha_b$
of each of the other $2L-1$ factors by $b_a$.
Apply \eqref{gcrtb:eq:block-norm-products} to each summand and use
the triangle inequality to sum the $2L$ contributions. This gives
\begin{equation}\label{gcrtb:eq:sweep-product-error-bound}
 \nu_b(\mathcal F-\mathbf F\otimes\mathbf I_N)
 \leq2Lb_a^{2L-1}\frac{b_{\max}}{(1-a)^2}
 =C_Lb_{\max},
 \qquad C_L=\frac{2Lb_a^{2L-1}}{(1-a)^2}.
\end{equation}

\medskip
\noindent\textbf{Step 3: The conditional cross moment.}

To prove \eqref{gcrtb:eq:comparison-cross-moment-bound}, we will express
$\mathbb E(\overline{\mathbf E}^{(m)\top}\mathbf Q\mathbf E\mid\mathbf S)/n$
in terms of the sweep product $\mathcal F$ from Step~2.

Recall from \eqref{gcrtint:eq:centered-model} that
$\mathbf E=\mathbf X_c^{(0)}-\mathbf S\mathbf C^\top$
is the residual matrix of the original targets and the initial state
$\overline{\mathbf E}^{(0)}$ of the oracle construction.
We first fix $(\mathbf S,\mathbf E)$ and compute the conditional mean
of the terminal output $\overline{\mathbf E}^{(m)}$.
We then average over $\mathbf E$ conditional on $\mathbf S$ to obtain
the cross moment.

Fix a branch $m\in\{1,\ldots,M\}$.
For $q=0,\ldots,2L$, let $\bs x_q\in\mathbb R^{Nt}$ stack the
residual columns after $q$ sweeps along the path from the original
input through the shared hub to the terminal output of branch $m$.
The first $L$ sweeps are forward and the last $L$ are reverse, so
$$
\begin{aligned}
\bs x_0&=(\mathbf E_1^\top,\ldots,\mathbf E_t^\top)^\top,\\
\bs x_{2L}&=(\overline{\mathbf E}^{(m)\top}_1,\ldots,
                 \overline{\mathbf E}^{(m)\top}_t)^\top.
\end{aligned}
$$
Let $\mathcal H_q$ be the sigma-field generated by $(\mathbf S,\mathbf E)$
and all Gaussian vectors used in the first $q$ sweeps of this path.
In particular, $\mathcal H_0=\sigma(\mathbf S,\mathbf E)$.
Each new sweep uses fresh Gaussian vectors independent of
$\mathcal H_{q-1}$.
The conditional mean relation $\bs y=\mathcal A_{\to}\bs x$ for a
forward sweep follows from \eqref{gcrtb:eq:stacked-sweep-mean} and
\eqref{gcrtb:eq:block-sweep}.
Its reverse version uses $\mathcal A_{\leftarrow}$.
Thus $\mathbb E(\bs x_q\mid\mathcal H_{q-1})=\mathcal A_q\bs x_{q-1}$
for $q=1,\ldots,2L$, where $\mathcal A_q$ is defined in Step~2.
Since $\mathcal A_q$ is fixed conditional on $\mathbf S$, the tower
property yields, for $q=1,\ldots,2L$,
$$
\begin{aligned}
\mathbb E(\bs x_q\mid\mathbf S,\mathbf E)
&=\mathbb E\!\left(\mathbb E(\bs x_q\mid\mathcal H_{q-1})
                   \mid\mathbf S,\mathbf E\right)\\
&=\mathbb E(\mathcal A_q\bs x_{q-1}\mid\mathbf S,\mathbf E)\\
&=\mathcal A_q\mathbb E(\bs x_{q-1}\mid\mathbf S,\mathbf E).
\end{aligned}
$$
We iterate this recursion and use
$\mathbb E(\bs x_0\mid\mathbf S,\mathbf E)=\bs x_0$ to obtain
$$
\begin{aligned}
\mathbb E(\bs x_{2L}\mid\mathbf S,\mathbf E)
&=\mathcal A_{2L}\cdots\mathcal A_1\bs x_0\\
&=\mathcal A_{\leftarrow}^{L}\mathcal A_{\to}^{L}\bs x_0
=\mathcal F\bs x_0.
\end{aligned}
$$
The $j$th block of this identity gives, for $m=1,\ldots,M$ and
$j=1,\ldots,t$,
\begin{equation}\label{gcrtb:eq:branch-conditional-mean}
\mathbb E(\overline{\mathbf E}^{(m)}_j\mid(\mathbf S,\mathbf E))
=\sum_{k=1}^t\mathcal F_{jk}\mathbf E_k,
\end{equation}
where $\mathcal F_{jk}$ is the $(j,k)$ block of $\mathcal F$.
This expectation averages over the randomness in both the shared hub
and branch $m$.

Next, we compute the cross moment from
\eqref{gcrtb:eq:branch-conditional-mean}.
To write it as a matrix product, define
$\mathcal T_Q(\mathcal B)\in\mathbb R^{t\times t}$ by
$\mathcal T_Q(\mathcal B)_{jk}=\operatorname{tr}(\mathbf Q\mathcal B_{jk})/n$
for a block matrix $\mathcal B$.
Conditional on $\mathbf S$, both $\mathcal F$ and $\mathbf Q$ are fixed.
For $m=1,\ldots,M$ and $1\leq j,l\leq t$, we obtain

$$
\begin{aligned}
\left[\mathbb E(\overline{\mathbf E}^{(m)\top}\mathbf Q\mathbf E
\mid\mathbf S)/n\right]_{jl}
&\stackrel{(1)}{=}\frac1n\sum_{k=1}^t
 \mathbb E(\mathbf E_k^\top\mathcal F_{jk}^\top
 \mathbf Q\mathbf E_l\mid\mathbf S)\\
&\stackrel{(2)}{=}\sum_{k=1}^t
 \frac{\operatorname{tr}(\mathcal F_{jk}^\top\mathbf Q)}{n}V_{kl}\\
&\stackrel{(3)}{=}\sum_{k=1}^t\mathcal T_Q(\mathcal F)_{jk}V_{kl},
\end{aligned}
$$
where step~(1) applies the tower property to \eqref{gcrtb:eq:branch-conditional-mean}, step~(2) uses
$\mathbb E(\mathbf E_l\mathbf E_k^\top\mid\mathbf S)=V_{kl}\mathbf I_N$, 
and step~(3) uses the symmetry of $\mathbf Q$ and the definition of
$\mathcal T_Q$.

The entrywise calculation gives the matrix identity
\begin{equation}\label{gcrtb:eq:comparison-cross-moment-identity}
\mathbb E(\overline{\mathbf E}^{(m)\top}\mathbf Q\mathbf E
 \mid\mathbf S)/n
=\mathcal T_Q(\mathcal F)\mathbf V.
\end{equation}

We first establish a norm bound for $\mathcal T_Q(\mathcal B)$
for an arbitrary block matrix $\mathcal B$; we will then apply it with
$\mathcal B=\mathcal F-\mathbf F\otimes\mathbf I_N$.
Since $0\preceq\mathbf Q\preceq\mathbf I_N$, duality between the
nuclear and operator norms gives
$|\operatorname{tr}(\mathbf Q\mathcal B_{jk})|
\leq\|\mathbf Q\|_{\mathrm{op}}\|\mathcal B_{jk}\|_*
\leq\|\mathcal B_{jk}\|_*$.
We apply this trace bound to each term in the absolute row sums:
$$
\begin{aligned}
\|\mathcal T_Q(\mathcal B)\|_\infty
&=\max_j\sum_{k=1}^t|\mathcal T_Q(\mathcal B)_{jk}|\\
&=\frac1n\max_j\sum_{k=1}^t
       |\operatorname{tr}(\mathbf Q\mathcal B_{jk})|\\
&\leq\frac1n\max_j\sum_{k=1}^t\|\mathcal B_{jk}\|_*\\
&\leq\frac{\nu_b(\mathcal B)}n,
\end{aligned}
$$
where the last inequality uses the definition of $\nu_b$ in Step~1.
Similarly, the absolute column sums give
$\|\mathcal T_Q(\mathcal B)\|_1\leq\nu_b(\mathcal B)/n$.
Hence
\begin{equation}\label{gcrtb:eq:trace-block-norm-bound}
\|\mathcal T_Q(\mathcal B)\|_{\mathrm{op}}
\leq\sqrt{\|\mathcal T_Q(\mathcal B)\|_1
          \|\mathcal T_Q(\mathcal B)\|_\infty}
\leq\nu_b(\mathcal B)/n.
\end{equation}
We now bound the error:
$$
\begin{aligned}
\left\|\mathbb E(\overline{\mathbf E}^{(m)\top}\mathbf Q\mathbf E
 \mid\mathbf S)/n-\tau\mathbf F\mathbf V\right\|_{\mathrm{op}}
 &\stackrel{(1)}{=}
 \|\mathcal T_Q(\mathcal F-\mathbf F\otimes\mathbf I_N)\mathbf V\|_{\mathrm{op}}\\
 &\stackrel{(2)}{\leq}\frac{v_+}{n}
 \nu_b(\mathcal F-\mathbf F\otimes\mathbf I_N)\\
 &\stackrel{(3)}{\leq}\frac{C_Lv_+b_{\max}}{n}
 =C_Lv_+d_n,
\end{aligned}
$$
where
\begin{itemize}
\item step~(1) uses \eqref{gcrtb:eq:comparison-cross-moment-identity},
linearity, and
$\mathcal T_Q(\mathbf F\otimes\mathbf I_N)=\tau\mathbf F$,
with $\tau=\operatorname{tr}(\mathbf Q)/n$ from
\eqref{gcrtb:eq:population};
\item step~(2) substitutes $\mathcal B=\mathcal F-\mathbf F\otimes\mathbf I_N$
into \eqref{gcrtb:eq:trace-block-norm-bound} and uses
$\|\mathbf V\|_{\mathrm{op}}\leq v_+$;
\item step~(3) uses \eqref{gcrtb:eq:sweep-product-error-bound};
\item the final equality
uses $d_n=b_{\max}/n$ from \eqref{gcrtb:eq:error-rates}.
\end{itemize}
This proves \eqref{gcrtb:eq:comparison-cross-moment-bound}.
No disjointness of the sets of neighbours in the background group, or commutation of
their projections with $\mathbf Q$, is used in this argument.

Finally, to prove \eqref{gcrtb:eq:comparison-scan-norm-bound},
use the definitions in \eqref{gcrtint:eq:scan} and
$\|\mathbf J_<\|_{\mathrm{op}},\|\mathbf J_>\|_{\mathrm{op}}\leq a$
to obtain
$\|\mathbf F_{\to}\|_{\mathrm{op}},\|\mathbf F_{\leftarrow}\|_{\mathrm{op}}
\leq a/(1-a)\leq b_a$.
Since $\mathbf F=\mathbf F_{\leftarrow}^{L}\mathbf F_{\to}^{L}$,
submultiplicativity gives $\|\mathbf F\|_{\mathrm{op}}\leq b_a^{2L}$.
\end{proof}

\subsection{Proof of Proposition~\ref{prop:gcrt-centering-comparisons}}
\label{app:gcrt-centering-comparisons}

We use the model and sampling specification of
Proposition~\ref{prop:gcrt-centering-comparisons}.
By \eqref{gcrtint:eq:target-precision},
$\mathbf\Theta=\mathbf V^{-1}=\mathbf I_2$.
Thus $\mathbf J=0$ by its definition in Section~\ref{sec:gcrt-ridge-power}.
The definitions in \eqref{gcrtint:eq:scan} then give
$\mathbf F_{\to}=\mathbf F_{\leftarrow}=\mathbf F=0$.
The rate definitions in \eqref{gcrtb:eq:error-rates} give
$d_n=1/n$ and $e_n^2=2/n$.
The scalar scores $\mathcal B_v$ defined in
\eqref{gcrtb:eq:scalar-centered-statistic} satisfy the conditions in
Appendix~\ref{app:scoring-conditions}, so
$\mathrm{pVal}_{\mathcal B}$ is valid under the null hypothesis.
Throughout this proof, $\bs c_v$ and $\mathbf D_v$ are the residual
score vectors and regularised residual Gram matrices defined in
\eqref{gcrtint:eq:centered-residual-scores}.
We first evaluate $\bs m_n$, $\bs u_n$, and $\bar{\mathbf D}_n$
from \eqref{gcrtb:eq:population}. Their limits enter both the statistic comparisons in
part~(a) and the gain expansion in part~(b).

\subsubsection{Common population quantities}
In the centred coordinates of \eqref{gcrtint:eq:centered-model},
the single background column $\mathbf S\in\mathbb R^N$ has
$N=n-1$ independent standard Gaussian entries. Let $U=\|\mathbf S\|_F^2/n$ and
$\omega_n=(1+U)^{-1}$, so $nU\sim\chi^2_{n-1}$.
For $\ell=1$, the definitions in \eqref{gcrtint:eq:ridge-residual-operators} give
$\mathbf H=\mathbf I_N-\mathbf S\mathbf S^\top/(n+\|\mathbf S\|^2)$
and $\mathbf Q=\mathbf H^2$.
Hence $\mathbf H\mathbf S=\omega_n\mathbf S$ and
$\mathbf H\bs x=\bs x$ whenever $\mathbf S^\top\bs x=0$.
Almost surely, $\mathbf Q$ has eigenvalue $\omega_n^2$ on
$\operatorname{col}(\mathbf S)$ and eigenvalue one on its
$(N-1)$-dimensional orthogonal complement.
Starting from the definitions in \eqref{gcrtb:eq:population}, we obtain
\begin{equation}\label{gcrtb:eq:direction-population-identities}
\begin{aligned}
\mathbf B_n
&=\mathbb E\left(\frac{\mathbf S^\top\mathbf Q\mathbf S}{n}\right)
 =\mathbb E\left(\frac{\|\mathbf H\mathbf S\|^2}{n}\right)\\
&=\mathbb E\left(\frac{\omega_n^2\|\mathbf S\|^2}{n}\right)
 =\mathbb E(\omega_n^2 U)\\
&=\mathbb E\left\{\frac{U}{(1+U)^2}\right\},\\
\tau&=\frac{\operatorname{tr}(\mathbf Q)}n
 =\frac{\omega_n^2+N-1}{n}
 =1-\frac{2-\omega_n^2}{n}.
\end{aligned}
\end{equation}

For $u\geq0$, the identity
$u/(1+u)^2=1/4-(u-1)^2/\{4(1+u)^2\}$ yields
$$
0\leq\frac14-\mathbf B_n
=\frac14\mathbb E\left\{\frac{(U-1)^2}{(1+U)^2}\right\}
\leq\frac14\mathbb E(U-1)^2
=\frac{2n-1}{4n^2},
$$
where the last equality uses
$\mathbb E U=(n-1)/n$ and $\operatorname{Var}(U)=2(n-1)/n^2$.
Thus $\mathbf B_n=1/4+O(n^{-1})$.
Recall that $\bar\tau_n=\mathbb E\tau$ in \eqref{gcrtb:eq:population}.
The bound $0<\omega_n\leq1$ therefore gives $\bar\tau_n=1+O(n^{-1})$.

The model \eqref{gcrtint:eq:centered-model} here has
$\mathbf C=\bs1_2\in\mathbb R^{2\times1}$, $\gamma=1$, and $\mathbf V=\mathbf I_2$.
Since $s=1$, $\mathbf B_n$ is a scalar.
With $\lambda_n=1$, \eqref{gcrtb:eq:population} gives
$\bs m_n=\bs1_2\mathbf B_n$,
$\bs u_n=\bar\tau_n\bs\beta_{\mathcal T}$, and
$\bar{\mathbf D}_n=(\bar\tau_n+n^{-1})\mathbf I_2
+\mathbf B_n\bs1_2\bs1_2^\top$.
Consequently,
\begin{equation}\label{gcrtb:eq:direction-limits}
 \bs m_n=(1/4,1/4)^\top+O(n^{-1}),\quad
 \bs u_n=(1+O(n^{-1}))\bs\beta_{\mathcal T},\quad
 \bar{\mathbf D}_n=\mathbf D_\infty+O(n^{-1}),\qquad
 \mathbf D_\infty=\mathbf I_2+\tfrac14\bs1_2\bs1_2^\top.
\end{equation}

\subsubsection{Part (a): fixed alternatives and the corresponding null model}
\label{app:gcrt-direction-score-proof}

\begin{proof}
Part~(a) compares three constructions defined in
Appendix~\ref{sec:gcrt-ridge-comparison}.
The uncentred gains $\mathcal A_v$ evaluate a quadratic form in the
residual score vectors $\bs c_v$.
To form $T_v^{\rm cen}$, we centre these vectors across candidates
before evaluating the quadratic form.
To form $\mathcal B_v$, we centre the scalar gains and square the deviations.

We first obtain the limits of $\bs c_v$ and $\mathbf D_v$ under both
alternatives and their common null model.
We then compare $\mathcal A_v$, $T_v^{\rm cen}$, and $\mathcal B_v$ in this
order, translating each strict limiting inequality into a rejection probability.
By \eqref{gcrtb:eq:rank}, the comparison for $\mathcal A_v$ also applies to $T_v$.

\noindent\textit{(i) Limits of $\bs c_v$ and $\mathbf D_v$.}\par\nopagebreak
Write the response as $Y=Z+b(E_1+E_2)+\varepsilon$.
At $r=1/(2\sqrt2)$, the choices $\theta=0$ and $\theta=\pi$
in \eqref{gcrtb:eq:direction-example}
correspond to $b=1/4$ and $b=-1/4$, respectively.
We also consider their common null model $b=0$.
For these three values of $b$, $\bs\beta_{\mathcal T}=b\bs1_2$.
Equation~\eqref{gcrtb:eq:error-rates} gives $\eta_n=O(n^{-1/2})$.
We have $\mathbf F=0$, and $M$ is fixed as in the proposition.
At $b=\pm1/4$, \eqref{gcrtb:eq:gram-approximation}--\eqref{gcrtb:eq:randomised-score-approximation}
in Lemma~\ref{gcrtb:lem:endpoints} and \eqref{gcrtb:eq:direction-limits} give
\begin{equation}\label{gcrtb:eq:direction-endpoint-limits}
\begin{aligned}
\bs c_0&\xrightarrow{\mathbb P}(1/4+b)\bs1_2,\\
\max_{1\leq m\leq M}\|\bs c_m-\bs1_2/4\|
&\xrightarrow{\mathbb P}0,\\
\max_{0\leq v\leq M}\|\mathbf D_v-\mathbf D_\infty\|_{\mathrm{op}}
&\xrightarrow{\mathbb P}0.
\end{aligned}
\end{equation}

We next show that the limits in \eqref{gcrtb:eq:direction-endpoint-limits}
also hold at $b=0$.
Here $\bs\beta_{\mathcal T}=0$ and $r=0$, so the positive signal lower
bound in Lemma~\ref{gcrtb:lem:endpoints} does not hold.
We will check that the derivation of
\eqref{gcrtint:eq:actual-expansions} still applies at $b=0$.
We will then use concentration bounds to obtain the limits of
$\bs c_v$ and $\mathbf D_v$.
Here $v_-$ is the lower eigenvalue bound, and $h_n$ and $R$ are
the signal bounds in \eqref{gcrtcent:eq:population-conditions}.
Under the remaining population and sampling bounds, the derivation of
\eqref{gcrtint:eq:actual-expansions} uses the signal only through
$r\leq R$ and $\|\bs\beta_{\mathcal T}\|\leq\sigma r/\sqrt{v_-}$;
it does not require $r\geq h_n>0$.
Here $d_n=1/n\to0$ and $e_n=\sqrt{2/n}\to0$.
The definitions in \eqref{gcrtb:eq:error-rates} also give
$\max_j(b_j+g_j)/n\leq d_n+e_n^2=3/n\to0$.
Thus the residual dimension condition
$\max_j(b_j+g_j)<N=n-1$ in
Lemma~\ref{gcrtb:lem:gaussian-comparison} holds for all sufficiently
large $n$, as required in the proof of
\eqref{gcrtb:eq:gram-approximation}.
The conditional moment and coupling bounds therefore yield
\eqref{gcrtint:eq:actual-expansions} also at $b=0$.
The replacements by population expectations justified by
\eqref{gcrtb:eq:ridge-concentration} and \eqref{gcrtb:eq:trace-concentration}
then give $\max_{0\leq v\leq M}\|\bs c_v-\bs1_2/4\|
\xrightarrow{\mathbb P}0$ and
$\max_{0\leq v\leq M}\|\mathbf D_v-\mathbf D_\infty\|_{\mathrm{op}}
\xrightarrow{\mathbb P}0$.

\noindent\textit{(ii) Uncentred gains $\mathcal A_v$.}\par\nopagebreak
The uncentred gain $\mathcal A_v=\bs c_v^\top\mathbf D_v^{-1}\bs c_v$
in Appendix~\ref{sec:gcrt-ridge-comparison} uses $\bs c_v$ without
subtracting the mean across candidates.
We use \eqref{gcrtb:eq:direction-endpoint-limits} to compute its limit.
The matrix $\mathbf D_\infty$ in \eqref{gcrtb:eq:direction-limits}
satisfies $\mathbf D_\infty\succeq\mathbf I_2$ and
$\mathbf D_\infty^{-1}\bs1_2=(2/3)\bs1_2$.
Thus $\bs1_2^\top\mathbf D_\infty^{-1}\bs1_2=4/3$.
The continuous mapping theorem gives
\begin{equation}\label{gcrtb:eq:direction-gain-limits}
 \mathcal A_0\xrightarrow{\mathbb P}\frac43(1/4+b)^2,\qquad
 \mathcal A_m\xrightarrow{\mathbb P}\frac1{12}
 \quad(1\leq m\leq M).
\end{equation}
For $1\leq m\leq M$, \eqref{gcrtb:eq:direction-gain-limits} gives
$$
(\mathcal A_0,\mathcal A_m)\xrightarrow{\mathbb P}
\begin{cases}
(1/3,1/12),&b=1/4,\\
(0,1/12),&b=-1/4.
\end{cases}
$$
A union bound over the fixed $M$ copies gives, with probability tending to one,
$$
\begin{cases}
\mathcal A_0>\max_{1\leq m\leq M}\mathcal A_m,&b=1/4,\\
\mathcal A_0<\min_{1\leq m\leq M}\mathcal A_m,&b=-1/4.
\end{cases}
$$
By \eqref{eqn:pval-1}, with probability tending to one,
$$
\mathrm{pVal}_{\rm unc}=
\begin{cases}
1/(M+1),&b=1/4,\\
1,&b=-1/4.
\end{cases}
$$
Since $1/(M+1)\leq\alpha<1$, it follows that
$$
\mathbb P(\mathrm{pVal}_{\rm unc}\leq\alpha)\longrightarrow
\begin{cases}
1,&b=1/4,\\
0,&b=-1/4.
\end{cases}
$$

\noindent\textit{(iii) Vector centring: $T_v^{\rm cen}$.}\par\nopagebreak
We next compare $T_0^{\rm cen}$ with $T_m^{\rm cen}$ for both
$b=1/4$ and $b=-1/4$.
Following \eqref{eq:symmetric-centered-statistic}, we first form
$\bs c_v-\bar{\bs c}$ and then evaluate its quadratic form with weight
$\mathbf D_{\rm pool}^{-1}$.
The pooled matrix is defined in \eqref{gcrtcent:eq:residual-scores}.
The definition of $\bar{\bs c}$ and
\eqref{gcrtb:eq:direction-endpoint-limits} give
$\bar{\bs c}\xrightarrow{\mathbb P}
\{(1/4+b)+M/4\}\bs1_2/(M+1)
=(1/4+b/(M+1))\bs1_2$.
Consequently,
$$
\bs c_0-\bar{\bs c}\xrightarrow{\mathbb P}\frac{Mb}{M+1}\bs1_2,
\qquad
\bs c_m-\bar{\bs c}\xrightarrow{\mathbb P}-\frac{b}{M+1}\bs1_2
\quad(1\leq m\leq M).
$$
Since $M$ is fixed, \eqref{gcrtb:eq:direction-endpoint-limits} and
the definition of $\mathbf D_{\rm pool}$ in
\eqref{gcrtcent:eq:residual-scores} imply
$\mathbf D_{\rm pool}\xrightarrow{\mathbb P}\mathbf D_\infty$.
Recall from \eqref{gcrtb:eq:direction-limits} that
$\mathbf D_\infty=\mathbf I_2+\tfrac14\bs1_2\bs1_2^\top
\succeq\mathbf I_2$.
Continuity of inversion therefore gives
$\mathbf D_{\rm pool}^{-1}\xrightarrow{\mathbb P}\mathbf D_\infty^{-1}$.
Since $\mathbf D_\infty\bs1_2=(3/2)\bs1_2$, we have
$\mathbf D_\infty^{-1}\bs1_2=(2/3)\bs1_2$.
Thus $b^2\bs1_2^\top\mathbf D_\infty^{-1}\bs1_2
=(1/16)(4/3)=1/12$ for $b=\pm1/4$.
The definition in \eqref{eq:symmetric-centered-statistic} then gives
the following limits in probability:
\begin{equation}\label{gcrtb:eq:fixed-centred-score-limits}
 T_0^{\rm cen}\longrightarrow\frac{M^2}{12(M+1)^2},\qquad
 T_m^{\rm cen}\longrightarrow\frac1{12(M+1)^2}.
\end{equation}
By \eqref{gcrtb:eq:fixed-centred-score-limits},
$T_0^{\rm cen}-T_m^{\rm cen}\xrightarrow{\mathbb P}
(M^2-1)/\{12(M+1)^2\}>0$ for each $m=1,\ldots,M$, since $M\geq2$.
A union bound over the fixed $M$ copies therefore gives
$T_0^{\rm cen}>\max_{1\leq m\leq M}T_m^{\rm cen}$ with probability
tending to one.
By \eqref{eqn:pval-1}, $p_{\rm cen}=1/(M+1)\leq\alpha$ on this event,
so $\mathbb P(p_{\rm cen}\leq\alpha)\to1$ for both signs of $b$.

\noindent\textit{(iv) Scalar centring and squaring: $\mathcal B_v$.}\par\nopagebreak
We first centre the scalar gains and then square the deviations.
This gives $\mathcal B_v=(\mathcal A_v-\bar{\mathcal A})^2$, as defined in
\eqref{gcrtb:eq:scalar-centered-statistic}.
We use the gain limits in \eqref{gcrtb:eq:direction-gain-limits}.
Define $\Delta=\tfrac43\{(1/4+b)^2-(1/4)^2\}=2b/3+4b^2/3$.
Then $\mathcal A_0-\mathcal A_m\xrightarrow{\mathbb P}\Delta$,
where $\Delta=1/4$ for $b=1/4$ and $\Delta=-1/12$ for $b=-1/4$.

Since $\bar{\mathcal A}=(M+1)^{-1}\sum_{v=0}^M\mathcal A_v$,
we have $\bar{\mathcal A}\xrightarrow{\mathbb P}1/12+\Delta/(M+1)$.
Hence, for $m=1,\ldots,M$,
$$
\mathcal A_0-\bar{\mathcal A}\xrightarrow{\mathbb P}\frac{M\Delta}{M+1},
\qquad
\mathcal A_m-\bar{\mathcal A}\xrightarrow{\mathbb P}-\frac{\Delta}{M+1}.
$$
The definition $\mathcal B_v=(\mathcal A_v-\bar{\mathcal A})^2$ then gives
$$
\mathcal B_0\xrightarrow{\mathbb P}\frac{M^2\Delta^2}{(M+1)^2},
\qquad
\mathcal B_m\xrightarrow{\mathbb P}\frac{\Delta^2}{(M+1)^2}.
$$

Because $M\geq2$ and $\Delta\ne0$, we obtain
$\mathcal B_0-\mathcal B_m\xrightarrow{\mathbb P}
(M^2-1)\Delta^2/(M+1)^2>0$.
A union bound over the fixed $M$ copies therefore gives
$\mathbb P(\mathcal B_0>\max_{1\leq m\leq M}\mathcal B_m)\to1$.
On this event, \eqref{eqn:pval-1} gives
$\mathrm{pVal}_{\mathcal B}=1/(M+1)\leq\alpha$.
Thus $\mathbb P(\mathrm{pVal}_{\mathcal B}\leq\alpha)\to1$
for both signs of $b$, completing part~(a).
\end{proof}

\subsubsection{Part (b): weak orthogonal alternative}
\label{app:gcrt-local-rank-proof}

\begin{proof}
Part~(b) compares the rejection probabilities of the three tests at
$\theta=\pi/2$ and $r_n=n^{-1/3}$.
Recall that $\mathrm{pVal}_{\rm unc}$ uses the uncentred gains
$\mathcal A_v$ defined in Appendix~\ref{sec:gcrt-ridge-comparison}.
For vector centring, $p_{\rm cen}$ uses $T_v^{\rm cen}$ in
\eqref{eq:symmetric-centered-statistic} with
$\mathbf W=\mathbf D_{\rm pool}^{-1}$.
For scalar centring and squaring, $\mathrm{pVal}_{\mathcal B}$ uses
$\mathcal B_v$ in \eqref{gcrtb:eq:scalar-centered-statistic}.

Step~(i) establishes $\mathbb P(p_{\rm cen}\leq\alpha)\to1$ by
Theorem~\ref{thm:gcrt-centered-power}.
For the other two tests, we show that the weak signal changes the
gains by less than their random fluctuations, so it does not change
the limiting rank distributions.

Step~(ii) establishes the conditional distribution of the candidates
and separates the common and varying parts of their scores and Gram
matrices. Using this representation, Step~(iii) derives a joint limit
of independent, identically distributed Gaussian variables for the
gain fluctuations at zero signal on the $n^{-1/2}$ scale.
Step~(iv) uses orthogonality to show that changing the signal strength
from zero to $r_n$ changes each gain by $o_{\mathbb P}(n^{-1/2})$.
The same joint limit therefore holds under the weak alternative.
Step~(v) obtains the limiting ranks for $\mathcal A_v$ and
$\mathcal B_v$ from this limit. Both limiting statistic vectors are
exchangeable and have no ties, so both rejection probabilities tend
to $\lfloor\alpha(M+1)\rfloor/(M+1)$.

\medskip
\noindent\textit{(i) Consistency of the test based on $T_v^{\rm cen}$.}\par\nopagebreak
We check the population and rate conditions for this application.
In \eqref{gcrtb:eq:direction-example}, take $\theta=\pi/2$ and
$r_n=n^{-1/3}$. Recall that $\bs v_-=(1,-1)^\top/\sqrt2$,
so $\bs\beta_{\mathcal T}=r_n\bs v_-$.
Here $s=1$, $t=2$, $\mathbf\Sigma_{\mathcal S\mathcal S}=1$,
$\mathbf V=\mathbf I_2$, $\|\mathbf C\|_{\rm op}=\sqrt2$,
$\mathbf J=0$, and $\gamma=\sigma=1$.
The signal definition \eqref{gcrtcent:eq:signal-strength} gives $r=r_n$, so
\eqref{gcrtcent:eq:population-conditions} holds with $h_n=r_n$
and suitable constants independent of $n$.

Each target has one background neighbour and no target neighbours.
The definitions in \eqref{gcrtb:eq:error-rates} give
$$
d_n=\frac1n\longrightarrow0,
\qquad
\frac{e_n^2}{h_n^2}=\frac{2}{nr_n^2}=2n^{-1/3}\longrightarrow0.
$$
The graph and dimensions are fixed, so these rates verify
\eqref{gcrtcent:eq:rates} uniformly over the corresponding population class.
The sampling and tuning parameters are those specified in
Proposition~\ref{prop:gcrt-centering-comparisons}.
Theorem~\ref{thm:gcrt-centered-power} therefore gives
$\mathbb P(p_{\rm cen}\leq\alpha)\to1$.

\medskip
\noindent\textit{(ii) Exact representation of the observed target data and randomised copies.}\par\nopagebreak
We will write each candidate as a shared fitted component plus a
random residual direction. At zero signal, this will separate the common
part of the gains from the fluctuations that determine their ranks.
Each target has only the background column as a neighbour.
In the centred coordinates of \eqref{gcrtint:eq:centered-model}, define
$\mathbf P_z=\mathbf S\mathbf S^\top/\|\mathbf S\|^2$.
The projections defined before \eqref{gcrtb:eq:gaussian-update} satisfy
$\mathbf P_1=\mathbf P_2=\mathbf P_z$.
We use the target matrices $\mathbf X_c^{(v)}$ defined before
\eqref{gcrtint:eq:centered-residual-scores}: $v=0$ is the observation,
and $v=1,\ldots,M$ are the randomised copies.
Let $\mathbf E_j\in\mathbb R^N$ be column $j$ of the random matrix
$\mathbf E\in\mathbb R^{N\times2}$ in \eqref{gcrtint:eq:centered-model}.
The model in \eqref{gcrtb:eq:direction-example} gives
$(\mathbf X_c^{(0)})_j=\mathbf S+\mathbf E_j$ for $j=1,2$.
Here $\mathbf E_1$ and $\mathbf E_2$ are independent
$\mathcal N_N(0,\mathbf I_N)$ vectors, independent of $\mathbf S$.
Define $\widehat a_j$ as the least-squares coefficient from regressing
$(\mathbf X_c^{(0)})_j$ on $\mathbf S$.
The coefficient and the corresponding orthogonal decomposition satisfy
\begin{equation}\label{gcrtb:eq:sphere-decomposition}
\begin{aligned}
 \widehat a_j&=\frac{\mathbf S^\top(\mathbf X_c^{(0)})_j}{\|\mathbf S\|^2}
 =1+\frac{\mathbf S^\top\mathbf E_j}{\|\mathbf S\|^2},\\
 (\mathbf X_c^{(0)})_j&=\widehat a_j\mathbf S+\bs f_j^{(0)},\\
 \bs f_j^{(0)}&=(\mathbf I_N-\mathbf P_z)\mathbf E_j,
 \qquad \rho_j=\|\bs f_j^{(0)}\|.
\end{aligned}
\end{equation}
Let $\mathcal C=\sigma(\mathbf S,\widehat a_1,\widehat a_2,\rho_1,\rho_2)$.
Given $\mathcal C$, the directions $\bs f_j^{(0)}/\rho_j$, $j=1,2$,
are independent and uniform on the unit sphere in
$\operatorname{col}(\mathbf I_N-\mathbf P_z)$, whose dimension is $d=n-2$.
By Lemma~\ref{lem:ggm-local-update}, each actual update leaves
$\widehat a_j$ and $\rho_j$ unchanged and draws a fresh independent
uniform direction. Thus, for all candidates,
\begin{equation}\label{gcrtb:eq:candidate-sphere-representation}
 (\mathbf X_c^{(v)})_j=\widehat a_j\mathbf S+\bs f_j^{(v)},
 \qquad \|\bs f_j^{(v)}\|=\rho_j,
 \qquad j=1,2,\quad v=0,\ldots,M.
\end{equation}
For $j=1,2$ and $v=0,\ldots,M$, the conditional distribution is
$$
 \bs f_j^{(v)}\mid\mathcal C
 \sim\operatorname{Unif}\!\left(
 \left\{\bs u\in\operatorname{col}(\mathbf I_N-\mathbf P_z):
                  \|\bs u\|=\rho_j\right\}\right).
$$
We next show that the pairs $(\bs f_1^{(v)},\bs f_2^{(v)})$,
$v=0,\ldots,M$, are independent given $\mathcal C$,
even though the randomised copies share a hub.
Let $\mu_n$ be the product of these uniform distributions for $j=1,2$.
The original residual pair has conditional law $\mu_n$ by independent
Gaussian projections. Let $\bs f_j^{(\mathrm{hub})}$ be the residual
obtained by subtracting $\widehat a_j\mathbf S$ from centred target
column $j$ at the hub. Define
$$
\mathcal G=\sigma\bigl(\mathcal C,\bs f_1^{(0)},\bs f_2^{(0)},
             \bs f_1^{(\mathrm{hub})},\bs f_2^{(\mathrm{hub})}\bigr).
$$
Given $\mathcal G$, the original residuals and the hub residuals are fixed.
Each branch contains at least one complete reverse sweep. Neither target
is a neighbour of the other, so this sweep refreshes both residual
directions independently in the fixed spheres determined by $\mathcal C$.
Its output has conditional law $\mu_n$, regardless of the starting
directions. The branches use fresh independent randomness, so their
residual pairs are independent given $\mathcal G$.
For measurable sets $A_0,\ldots,A_M$ of residual pairs, it follows that
\begin{equation}\label{gcrtb:eq:branch-sphere-factorisation}
\begin{aligned}
&\mathbb P\left(\bigcap_{m=1}^M
 \{(\bs f_1^{(m)},\bs f_2^{(m)})\in A_m\}\mid\mathcal G\right)\\
={}&\prod_{m=1}^M\mathbb P\left(
 (\bs f_1^{(m)},\bs f_2^{(m)})\in A_m\mid\mathcal G\right)\\
={}&\prod_{m=1}^M\mu_n(A_m).
\end{aligned}
\end{equation}
We now condition only on $\mathcal C$ and include the original pair.
\begin{equation}\label{gcrtb:eq:sphere-factorisation}
\begin{aligned}
&\mathbb P\left(\bigcap_{v=0}^M
 \{(\bs f_1^{(v)},\bs f_2^{(v)})\in A_v\}\mid\mathcal C\right)\\
\stackrel{(1)}{=}{}&\mathbb E\!\left[
 \mathbb P\left(\bigcap_{v=0}^M
 \{(\bs f_1^{(v)},\bs f_2^{(v)})\in A_v\}\mid\mathcal G\right)
 \mid\mathcal C\right]\\
\stackrel{(2)}{=}{}&\mathbb E\!\left[
 \mathbf1\{(\bs f_1^{(0)},\bs f_2^{(0)})\in A_0\}
 \mathbb P\left(\bigcap_{m=1}^M
 \{(\bs f_1^{(m)},\bs f_2^{(m)})\in A_m\}\mid\mathcal G\right)
 \mid\mathcal C\right]\\
\stackrel{(3)}{=}{}&\mathbb E\!\left[
 \mathbf1\{(\bs f_1^{(0)},\bs f_2^{(0)})\in A_0\}
 \prod_{m=1}^M\mu_n(A_m)\mid\mathcal C\right]\\
\stackrel{(4)}{=}{}&\left(\prod_{m=1}^M\mu_n(A_m)\right)
 \mathbb P\left((\bs f_1^{(0)},\bs f_2^{(0)})\in A_0
 \mid\mathcal C\right)\\
\stackrel{(5)}{=}{}&\prod_{v=0}^M\mu_n(A_v),
\end{aligned}
\end{equation}
where the steps are justified as follows:
\begin{itemize}
\item step~(1) uses the tower property and $\mathcal C\subseteq\mathcal G$;
\item step~(2) uses the fact that the original residual pair is
$\mathcal G$-measurable;
\item step~(3) substitutes \eqref{gcrtb:eq:branch-sphere-factorisation};
\item step~(4) uses the fact that $\mu_n$ is determined by $\mathcal C$;
\item step~(5) uses the conditional law $\mu_n$ of the original residual pair.
\end{itemize}
Equation~\eqref{gcrtb:eq:sphere-factorisation} establishes the required
conditional independence. We now use
\eqref{gcrtb:eq:candidate-sphere-representation} to separate the common
and varying parts of the score vectors and Gram matrices.

To compare zero signal with $r=r_n$, keep the covariates, randomised
copies, and independent noise $\bs\varepsilon$ fixed as $r$ varies.
The response is then
$\bs Y_c=\mathbf S+r\mathbf E\bs v_-+\bs\varepsilon$.
Write $\bs c_v(r)$ for the residual score vector in
\eqref{gcrtint:eq:centered-residual-scores} at signal strength $r$.
Write $\mathcal A_v(r)$ for the previously defined uncentred gain
$\mathcal A_v$ evaluated at signal strength $r$, so
\begin{equation}\label{gcrtb:eq:gain-at-signal-strength}
\mathcal A_v(r)=\bs c_v(r)^\top\mathbf D_v^{-1}\bs c_v(r).
\end{equation}
In particular, $\mathcal A_v(0)$ is the gain at zero signal.
The matrix $\mathbf D_v$ depends only on the covariates and is unchanged
as $r$ varies.
At $r=0$, let $\bs Y_0=\mathbf S+\bs\varepsilon$ and
$\bs\varepsilon_\perp=(\mathbf I_N-\mathbf P_z)\bs\varepsilon$.
Since $\mathbf S^\top\bs f_j^{(v)}=0$,
$\bs f_j^{(v)\top}\bs Y_0
=\bs f_j^{(v)\top}\bs\varepsilon_\perp$.
All candidate scores use the same response noise $\bs\varepsilon$.
We therefore also condition on this noise when studying their joint
fluctuations. Its independence from the covariates and the sampling
randomness implies that the factorisation in
\eqref{gcrtb:eq:sphere-factorisation} still holds conditional on
$(\mathcal C,\bs\varepsilon)$.
Recall $\omega_n=n/(n+\|\mathbf S\|^2)$.
The definitions in \eqref{gcrtint:eq:ridge-residual-operators} give
\begin{equation}\label{gcrtb:eq:sphere-ridge-operators}
\mathbf H=\omega_n\mathbf P_z+(\mathbf I_N-\mathbf P_z),
\qquad \mathbf Q=\mathbf H^2.
\end{equation}
Let $\widehat{\bs a}=(\widehat a_1,\widehat a_2)^\top$.
Insert \eqref{gcrtb:eq:candidate-sphere-representation} into
\eqref{gcrtint:eq:centered-residual-scores}.
The matrix $\mathbf Q$ is symmetric and satisfies
$\mathbf Q\mathbf S=\omega_n^2\mathbf S$ and
$\mathbf Q\bs f_j^{(v)}=\bs f_j^{(v)}$.
For $j=1,2$ and $v=0,\ldots,M$, the score component is
$$
\begin{aligned}
(\bs c_v(0))_j
&=\frac1n(\widehat a_j\mathbf S+\bs f_j^{(v)})^\top
       \mathbf Q\bs Y_0\\
&=\frac{\omega_n^2\widehat a_j}{n}\mathbf S^\top\bs Y_0
  +\frac1n\bs f_j^{(v)\top}\bs Y_0\\
&=\frac{\omega_n^2\widehat a_j}{n}\mathbf S^\top\bs Y_0
  +\frac1n\bs f_j^{(v)\top}\bs\varepsilon_\perp.
\end{aligned}
$$
For the Gram matrix, use $\lambda_n=1$.
The mixed terms vanish because $\mathbf S^\top\bs f_j^{(v)}=0$,
so for $j,k=1,2$,
$$
\begin{aligned}
(\mathbf D_v)_{jk}
&=\frac1n(\widehat a_j\mathbf S+\bs f_j^{(v)})^\top
       \mathbf Q(\widehat a_k\mathbf S+\bs f_k^{(v)})
       +\frac{\mathbf1\{j=k\}}n\\
&=\frac{\omega_n^2\|\mathbf S\|^2}{n}\widehat a_j\widehat a_k
  +\frac{\bs f_j^{(v)\top}\bs f_k^{(v)}}n
  +\frac{\mathbf1\{j=k\}}n.
\end{aligned}
$$
Since $\|\bs f_j^{(v)}\|^2=\rho_j^2$, the diagonal residual terms
are common to all candidates. The only remaining dependence of
$\mathbf D_v$ on $v$ is through $\bs f_1^{(v)\top}\bs f_2^{(v)}$.
Collecting the common terms in the score and Gram matrix gives
\begin{equation}\label{gcrtb:eq:null-score}
 \bs c_v(0)=\bs b_n+\frac1n
 \begin{pmatrix}\bs f_1^{(v)\top}\bs\varepsilon_\perp\\
                 \bs f_2^{(v)\top}\bs\varepsilon_\perp\end{pmatrix},
 \qquad
 \mathbf D_v=\mathbf G_n+
 \frac{\bs f_1^{(v)\top}\bs f_2^{(v)}}n
 \begin{pmatrix}0&1\\1&0\end{pmatrix},
\end{equation}
where we define the common terms by
\begin{equation}\label{gcrtb:eq:common-center}
\begin{aligned}
 \bs b_n&=\omega_n^2\widehat{\bs a}\,\mathbf S^\top\bs Y_0/n,\\
 \mathbf G_n&=\omega_n^2\|\mathbf S\|^2
              \widehat{\bs a}\widehat{\bs a}^\top/n
          +\operatorname{diag}(\rho_1^2/n,\rho_2^2/n)
          +n^{-1}\mathbf I_2.
\end{aligned}
\end{equation}

Given $(\mathcal C,\bs\varepsilon)$, the pair
$(\bs b_n,\mathbf G_n)$ is fixed and shared by all candidates.
We will expand $\mathcal A_v(0)$ about this pair to study its fluctuations.
Define the corresponding reference gain by
\begin{equation}\label{gcrtb:eq:common-reference-gain}
\mathcal K_n=\bs b_n^\top\mathbf G_n^{-1}\bs b_n.
\end{equation}
All candidates share this reference gain, so subtracting it from each
$\mathcal A_v(0)$ leaves their ranks unchanged.

We next identify the means of $\bs b_n$ and $\mathbf G_n$ with the
population quantities in \eqref{gcrtb:eq:population}.
First condition on $\mathbf S$.
By \eqref{gcrtb:eq:sphere-decomposition}, the coefficients
$\widehat a_1,\widehat a_2$ are conditionally independent, each with
mean one and variance $\|\mathbf S\|^{-2}$.
They are also conditionally independent of $\bs\varepsilon$.
Since $\mathbb E(\mathbf S^\top\bs Y_0\mid\mathbf S)=\|\mathbf S\|^2$,
\eqref{gcrtb:eq:common-center} gives
$$
\begin{aligned}
\mathbb E(\bs b_n\mid\mathbf S)
&=\frac{\omega_n^2}{n}
  \mathbb E(\widehat{\bs a}\mid\mathbf S)
  \mathbb E(\mathbf S^\top\bs Y_0\mid\mathbf S)\\
&=\frac{\omega_n^2\|\mathbf S\|^2}{n}\bs1_2.
\end{aligned}
$$
For the Gram matrix, we need the second moments.
The conditional means and variances of the coefficients give
$$
\begin{aligned}
\mathbb E(\widehat{\bs a}\widehat{\bs a}^\top\mid\mathbf S)
&=\mathbb E(\widehat{\bs a}\mid\mathbf S)
  \mathbb E(\widehat{\bs a}\mid\mathbf S)^\top
  +\operatorname{Cov}(\widehat{\bs a}\mid\mathbf S)\\
&=\bs1_2\bs1_2^\top+\|\mathbf S\|^{-2}\mathbf I_2.
\end{aligned}
$$
Also, \eqref{gcrtb:eq:sphere-decomposition} gives
$\rho_j^2=\|(\mathbf I_N-\mathbf P_z)\mathbf E_j\|^2$.
Since $\mathbf I_N-\mathbf P_z$ is a projection of rank $n-2$,
$$
\mathbb E(\rho_j^2\mid\mathbf S)
=\operatorname{tr}(\mathbf I_N-\mathbf P_z)=n-2,
\qquad j=1,2.
$$
Equation~\eqref{gcrtb:eq:common-center} therefore yields
$$
\begin{aligned}
\mathbb E(\mathbf G_n\mid\mathbf S)
&=\frac{\omega_n^2\|\mathbf S\|^2}{n}
  (\bs1_2\bs1_2^\top+\|\mathbf S\|^{-2}\mathbf I_2)
  +\frac{n-2}{n}\mathbf I_2+\frac1n\mathbf I_2\\
&=\frac{\omega_n^2\|\mathbf S\|^2}{n}\bs1_2\bs1_2^\top
  +\left(\frac{\omega_n^2+n-2}{n}+\frac1n\right)\mathbf I_2\\
&=\frac{\omega_n^2\|\mathbf S\|^2}{n}\bs1_2\bs1_2^\top
  +(\tau+n^{-1})\mathbf I_2,
\end{aligned}
$$
where the term $(\omega_n^2/n)\mathbf I_2$ comes from the conditional
variances of $\widehat a_1,\widehat a_2$, and the last equality uses
$\operatorname{tr}(\mathbf Q)=\omega_n^2+n-2$ from
\eqref{gcrtb:eq:direction-population-identities}.
Now take expectations over $\mathbf S$ and use
\eqref{gcrtb:eq:population} to obtain
$$
\begin{aligned}
\mathbb E\bs b_n
&=\bs1_2\mathbb E\left(\frac{\omega_n^2\|\mathbf S\|^2}{n}\right)
 =\bs1_2\mathbf B_n=\bs m_n,\\
\mathbb E\mathbf G_n
&=\bs1_2\mathbf B_n\bs1_2^\top
  +(\bar\tau_n+n^{-1})\mathbf I_2
 =\bar{\mathbf D}_n.
\end{aligned}
$$

\medskip
\noindent\textit{(iii) Joint fluctuations of $\mathcal A_v(0)$.}\par\nopagebreak
Our goal is the joint limit of $\sqrt n\{\mathcal A_v(0)-\mathcal K_n\}$
for $v=0,\ldots,M$.
Equation~\eqref{gcrtb:eq:null-score} reduces the score and Gram matrix
fluctuations to three inner products. We first establish a joint
Gaussian limit for the vectors
\begin{equation}\label{gcrtb:eq:sphere-limit}
 \frac1{\sqrt n}
 \begin{pmatrix}
  \bs f_1^{(v)\top}\bs\varepsilon_\perp\\
  \bs f_2^{(v)\top}\bs\varepsilon_\perp\\
  \bs f_1^{(v)\top}\bs f_2^{(v)}
 \end{pmatrix},\qquad v=0,\ldots,M.
\end{equation}
We then apply a Taylor expansion at $(\bs b_n,\mathbf G_n)$.
The normalising factors also enter the Taylor expansion, so we record
their rates here. The variables $\|\mathbf S\|^2$,
$\rho_j^2$, and $\|\bs\varepsilon_\perp\|^2$ have chi-squared
distributions with $n-1$, $n-2$, and $n-2$ degrees of freedom.
Conditional on $\mathbf S$, both $\mathbf S^\top\mathbf E_j/n$ and
$\mathbf S^\top\bs\varepsilon/n$ are centred Gaussian with variance
$\|\mathbf S\|^2/n^2$.
Using \eqref{gcrtb:eq:sphere-decomposition}, we obtain, for $j=1,2$,
\begin{equation}\label{gcrtb:eq:sphere-basic-rates}
\begin{gathered}
\begin{aligned}
 \frac{\|\mathbf S\|^2}{n}&=1+O_{\mathbb P}(n^{-1/2}),
 &\frac{\rho_j^2}{n}&=1+O_{\mathbb P}(n^{-1/2}),\\
 \frac{\|\bs\varepsilon_\perp\|^2}{n}&=1+O_{\mathbb P}(n^{-1/2}),\\
 \frac{\mathbf S^\top\mathbf E_j}{n}&=O_{\mathbb P}(n^{-1/2}),
 &\frac{\mathbf S^\top\bs\varepsilon}{n}&=O_{\mathbb P}(n^{-1/2}),\\
 \widehat a_j&=1+O_{\mathbb P}(n^{-1/2}),
\end{aligned}\\
 \omega_n=\frac1{1+\|\mathbf S\|^2/n}
 =\frac1{2+O_{\mathbb P}(n^{-1/2})}
 =\tfrac12+O_{\mathbb P}(n^{-1/2}).
\end{gathered}
\end{equation}
We now work conditional on $(\mathcal C,\bs\varepsilon)$.
To verify joint convergence of the three variables in \eqref{gcrtb:eq:sphere-limit}, rotate the
residual space so that $\bs\varepsilon_\perp$ lies in its first coordinate.
In these coordinates, represent
$\bs f_j^{(v)}=\rho_j\bs G_j^{(v)}/\|\bs G_j^{(v)}\|$, where the
$\bs G_j^{(v)}\in\mathbb R^d$ are independent standard Gaussian vectors.
Equation~\eqref{gcrtb:eq:sphere-basic-rates} and $d/n\to1$ give
$\rho_j/\sqrt d\xrightarrow{\mathbb P}1$.
Also, $\|\bs G_j^{(v)}\|/\sqrt d\xrightarrow{\mathbb P}1$, so
$\rho_j/\|\bs G_j^{(v)}\|\xrightarrow{\mathbb P}1$.
Write $G_{jk}^{(v)}$ for coordinate $k$ of $\bs G_j^{(v)}$.
For $j=1,2$, the first two coordinates in
\eqref{gcrtb:eq:sphere-limit} satisfy
$$
\frac{\bs f_j^{(v)\top}\bs\varepsilon_\perp}{\sqrt n}
=\frac{\rho_j}{\|\bs G_j^{(v)}\|}
 \frac{\|\bs\varepsilon_\perp\|}{\sqrt n}\,G_{j1}^{(v)}
=G_{j1}^{(v)}+o_{\mathbb P}(1),
$$
where the last equality uses \eqref{gcrtb:eq:sphere-basic-rates}.
For $\bs f_1^{(v)\top}\bs f_2^{(v)}/\sqrt n$,
$$
\frac{\bs f_1^{(v)\top}\bs f_2^{(v)}}{\sqrt n}
=\frac{\rho_1\rho_2}{\|\bs G_1^{(v)}\|\|\bs G_2^{(v)}\|}
 \frac1{\sqrt n}\sum_{k=1}^dG_{1k}^{(v)}G_{2k}^{(v)}.
$$
Here $G_{11}^{(v)}G_{21}^{(v)}/\sqrt n=o_{\mathbb P}(1)$.
Since $d/n\to1$, the central limit theorem gives
$n^{-1/2}\sum_{k=2}^dG_{1k}^{(v)}G_{2k}^{(v)}
\Longrightarrow\mathcal N(0,1)$.
Also, $\rho_1\rho_2/(\|\bs G_1^{(v)}\|\|\bs G_2^{(v)}\|)
\xrightarrow{\mathbb P}1$.
Slutsky's theorem therefore gives
$$
\frac{\bs f_1^{(v)\top}\bs f_2^{(v)}}{\sqrt n}
=\frac1{\sqrt n}\sum_{k=2}^dG_{1k}^{(v)}G_{2k}^{(v)}
 +o_{\mathbb P}(1)
\Longrightarrow\mathcal N(0,1).
$$
The sum over $k=2,\ldots,d$ is independent of
$(G_{11}^{(v)},G_{21}^{(v)})$ for every $n$.
Thus Slutsky's theorem applies jointly to the vector in
\eqref{gcrtb:eq:sphere-limit}, whose limit has three independent
standard Gaussian coordinates.
The conditional independence in \eqref{gcrtb:eq:sphere-factorisation}
and the fixed value of $M$ give the joint independent limit across
all candidates.

For $\bs t\in\mathbb R^{3(M+1)}$, let $\phi_n(\bs t)$ be the joint
characteristic function of the vectors in \eqref{gcrtb:eq:sphere-limit}
conditional on $(\mathcal C,\bs\varepsilon)$.
Conditional on $(\mathcal C,\bs\varepsilon)$, the $\bs G_j^{(v)}$
are independent standard Gaussian vectors.
The scale convergence in \eqref{gcrtb:eq:sphere-basic-rates} and the
joint Slutsky argument therefore give
$$
\phi_n(\bs t)\xrightarrow{\mathbb P}e^{-\|\bs t\|^2/2},
\qquad |\phi_n(\bs t)|\leq1.
$$
The bound implies convergence also in $L^1$.
By the tower property, the unconditional joint characteristic function
is $\mathbb E\phi_n(\bs t)$, and
$$
\left|\mathbb E\phi_n(\bs t)-e^{-\|\bs t\|^2/2}\right|
\leq\mathbb E\left|\phi_n(\bs t)-e^{-\|\bs t\|^2/2}\right|
\longrightarrow0.
$$
The continuity theorem for characteristic functions therefore gives
the unconditional joint convergence
\begin{equation}\label{gcrtb:eq:sphere-joint-limit}
\left(
 \frac1{\sqrt n}
 \begin{pmatrix}
  \bs f_1^{(v)\top}\bs\varepsilon_\perp\\
  \bs f_2^{(v)\top}\bs\varepsilon_\perp\\
  \bs f_1^{(v)\top}\bs f_2^{(v)}
 \end{pmatrix}
\right)_{v=0}^M
\Longrightarrow
\left(
 \begin{pmatrix}Z_{1v}\\Z_{2v}\\Z_{12v}\end{pmatrix}
\right)_{v=0}^M,
\end{equation}
where all displayed $Z$ variables are independent standard Gaussians.

We next locate the random expansion point $(\bs b_n,\mathbf G_n)$.
This will justify the Taylor expansion and the replacement of its
coefficients by deterministic limits.
Since $\bs Y_0=\mathbf S+\bs\varepsilon$,
\eqref{gcrtb:eq:sphere-basic-rates} gives
$$
\frac{\mathbf S^\top\bs Y_0}{n}
=\frac{\|\mathbf S\|^2}{n}
 +\frac{\mathbf S^\top\bs\varepsilon}{n}
=1+O_{\mathbb P}(n^{-1/2}).
$$
We apply the other bounds in \eqref{gcrtb:eq:sphere-basic-rates}
to \eqref{gcrtb:eq:common-center} and obtain
\begin{equation}\label{gcrtb:eq:random-center-bounds}
\begin{aligned}
\bs b_n&=(1/4,1/4)^\top+O_{\mathbb P}(n^{-1/2}),\\
\mathbf G_n&=\mathbf D_\infty+O_{\mathbb P}(n^{-1/2}).
\end{aligned}
\end{equation}
For later use in Step~(iv), \eqref{gcrtb:eq:null-score} and the
joint
limit in \eqref{gcrtb:eq:sphere-joint-limit} also give

\begin{equation}\label{gcrtb:eq:null-center-bounds}
\begin{aligned}
\bs c_0(0)&=(1/4,1/4)^\top+O_{\mathbb P}(n^{-1/2}),\\
\mathbf D_0&=\mathbf D_\infty+O_{\mathbb P}(n^{-1/2}).
\end{aligned}
\end{equation}

We now expand $\mathcal A_v(0)-\mathcal K_n$ and show that the remainder
is negligible on the $n^{-1/2}$ scale.
For $v=0,\ldots,M$, define $\bs h_v=\bs c_v(0)-\bs b_n$
and $\mathbf L_v=\mathbf D_v-\mathbf G_n$.
Since $M$ is fixed, \eqref{gcrtb:eq:null-score} and joint convergence in
\eqref{gcrtb:eq:sphere-joint-limit} give
$\max_v(\|\bs h_v\|+\|\mathbf L_v\|_{\rm op})
=O_{\mathbb P}(n^{-1/2})$.

Recall from \eqref{gcrtb:eq:direction-limits} that
$\mathbf D_\infty=\mathbf I_2+\tfrac14\bs1_2\bs1_2^\top
\succeq\mathbf I_2$. The gain function
$g(\bs c,\mathbf D)=\bs c^\top\mathbf D^{-1}\bs c$
has bounded second derivatives on a fixed convex neighbourhood of
$((1/4,1/4)^\top,\mathbf D_\infty)$.
Equation~\eqref{gcrtb:eq:random-center-bounds} and the
$O_{\mathbb P}(n^{-1/2})$ bound for the perturbations
$(\bs h_v,\mathbf L_v)$ place the expansion point
$(\bs b_n,\mathbf G_n)$ and all perturbed points
$(\bs c_v(0),\mathbf D_v)$ in this neighbourhood with probability
tending to one. To expand along the perturbation $(\bs h_v,\mathbf L_v)$, define
$g_v(s)=g(\bs b_n+s\bs h_v,\mathbf G_n+s\mathbf L_v)$ for $0\leq s\leq1$.
Equation~\eqref{gcrtb:eq:gain-at-signal-strength} at $r=0$ and
\eqref{gcrtb:eq:common-reference-gain} give
$$
\begin{aligned}
\mathcal A_v(0)
&=\bs c_v(0)^\top\mathbf D_v^{-1}\bs c_v(0)
 =g_v(1),\\
\mathcal K_n
&=\bs b_n^\top\mathbf G_n^{-1}\bs b_n
 =g_v(0).
\end{aligned}
$$
If we differentiate
$(\mathbf G_n+s\mathbf L_v)(\mathbf G_n+s\mathbf L_v)^{-1}=\mathbf I_2$,
we obtain
$$
\frac{d}{ds}(\mathbf G_n+s\mathbf L_v)^{-1}
=-(\mathbf G_n+s\mathbf L_v)^{-1}\mathbf L_v
  (\mathbf G_n+s\mathbf L_v)^{-1}.
$$
The product rule then gives
$$
\begin{aligned}
g_v'(0)
&=\bs h_v^\top\mathbf G_n^{-1}\bs b_n
 +\bs b_n^\top\mathbf G_n^{-1}\bs h_v
 -\bs b_n^\top\mathbf G_n^{-1}\mathbf L_v\mathbf G_n^{-1}\bs b_n\\
&=2(\mathbf G_n^{-1}\bs b_n)^\top\bs h_v
 -(\mathbf G_n^{-1}\bs b_n)^\top\mathbf L_v(\mathbf G_n^{-1}\bs b_n),
\end{aligned}
$$
where the two terms in the last line are linear in $\bs h_v$ and
$\mathbf L_v$, respectively.
Taylor's theorem for $g_v$ at $s=0$ therefore gives
\begin{equation}\label{gcrtb:eq:null-gain-taylor}
\begin{aligned}
\mathcal A_v(0)-\mathcal K_n
&=g_v(1)-g_v(0)=g_v'(0)+R_{nv}\\
&=2(\mathbf G_n^{-1}\bs b_n)^\top\bs h_v
 -(\mathbf G_n^{-1}\bs b_n)^\top\mathbf L_v(\mathbf G_n^{-1}\bs b_n)
 +R_{nv},
\end{aligned}
\end{equation}
where $R_{nv}=\int_0^1(1-s)g_v''(s)\,ds$.
When both endpoints lie in the neighbourhood, the whole path lies there
by convexity.
The chain rule and the bound on the second derivatives of $g$ give
$$
|R_{nv}|\leq\frac12\sup_{0\leq s\leq1}|g_v''(s)|
\leq C(\|\bs h_v\|+\|\mathbf L_v\|_{\rm op})^2.
$$
Thus $R_{nv}$ is of second order in the perturbations.
Since all expansion points lie in the neighbourhood with probability
tending to one, $\max_{0\leq v\leq M}|R_{nv}|=O_{\mathbb P}(n^{-1})$.

Because $\mathbf D_\infty\bs1_2=(3/2)\bs1_2$, we have
$\mathbf D_\infty^{-1}(1/4,1/4)^\top=(1/6,1/6)^\top$.
Equation~\eqref{gcrtb:eq:random-center-bounds} therefore gives
\begin{equation}\label{gcrtb:eq:null-gain-coefficient}
\mathbf G_n^{-1}\bs b_n
=\tfrac16\bs1_2+O_{\mathbb P}(n^{-1/2}).
\end{equation}
We now substitute \eqref{gcrtb:eq:null-gain-coefficient} into
\eqref{gcrtb:eq:null-gain-taylor}, then express $\bs h_v$ and
$\mathbf L_v$ through the inner products in \eqref{gcrtb:eq:null-score}:
$$
\begin{aligned}
&\sqrt n\{\mathcal A_v(0)-\mathcal K_n\}\\
\stackrel{(1)}{=}{}&\frac{\sqrt n}{3}\bs1_2^\top\bs h_v
 -\frac{\sqrt n}{36}\bs1_2^\top\mathbf L_v\bs1_2
 +o_{\mathbb P}(1)\\
\stackrel{(2)}{=}{}&\frac1{3\sqrt n}
 (\bs f_1^{(v)\top}\bs\varepsilon_\perp
  +\bs f_2^{(v)\top}\bs\varepsilon_\perp)\\
&\quad-\frac{\bs f_1^{(v)\top}\bs f_2^{(v)}}{36\sqrt n}
 \bs1_2^\top\begin{pmatrix}0&1\\1&0\end{pmatrix}\bs1_2
 +o_{\mathbb P}(1)\\
\stackrel{(3)}{=}{}&\frac13
 \left(\frac{\bs f_1^{(v)\top}\bs\varepsilon_\perp}{\sqrt n}
       +\frac{\bs f_2^{(v)\top}\bs\varepsilon_\perp}{\sqrt n}\right)
 -\frac1{18}\frac{\bs f_1^{(v)\top}\bs f_2^{(v)}}{\sqrt n}
 +o_{\mathbb P}(1),
\end{aligned}
$$
where the steps are justified as follows:
\begin{itemize}
\item step~(1) uses \eqref{gcrtb:eq:null-gain-coefficient}.
The coefficient replacement changes each first order term by
$O_{\mathbb P}(n^{-1})$, since
$\|\bs h_v\|+\|\mathbf L_v\|_{\rm op}=O_{\mathbb P}(n^{-1/2})$.
These errors and $R_{nv}$ become $o_{\mathbb P}(1)$ after multiplication
by $\sqrt n$, simultaneously for $v=0,\ldots,M$;
\item step~(2) substitutes \eqref{gcrtb:eq:null-score} into
$\bs h_v=\bs c_v(0)-\bs b_n$ and
$\mathbf L_v=\mathbf D_v-\mathbf G_n$;
\item step~(3) uses
$\bs1_2^\top\left(\begin{smallmatrix}0&1\\1&0\end{smallmatrix}\right)\bs1_2=2$.
\end{itemize}
Joint convergence in
\eqref{gcrtb:eq:sphere-joint-limit} now gives
\begin{equation}\label{gcrtb:eq:null-gain-limit}
 \sqrt n\{\mathcal A_v(0)-\mathcal K_n\}
 \ \Longrightarrow\ \tfrac13(Z_{1v}+Z_{2v})-\tfrac1{18}Z_{12v},
 \qquad v=0,\ldots,M,
\end{equation}
where
the $Z$ variables are those in \eqref{gcrtb:eq:sphere-joint-limit}.
Each limit has variance $2/9+1/324=73/324$.

\medskip
\noindent\textit{(iv) Effect of the weak orthogonal signal.}\par\nopagebreak
To transfer \eqref{gcrtb:eq:null-gain-limit} to $r=r_n$, we will show
$\max_{0\leq v\leq M}|\mathcal A_v(r_n)-\mathcal A_v(0)|
=o_{\mathbb P}(n^{-1/2})$.
As in Step~(ii), keep the covariates, randomised copies, and noise
$\bs\varepsilon$ fixed as $r$ varies.
The response at signal strength $r$ is
$\bs Y_c=\bs Y_0+r\mathbf E\bs v_-$, where
$\bs Y_0=\mathbf S+\bs\varepsilon$ is the response at $r=0$.
For $v=0,\ldots,M$, define
\begin{equation}\label{gcrtb:eq:signal-score-coefficient}
\bs s_v=\mathbf X_c^{(v)\top}\mathbf Q\mathbf E\bs v_-/n.
\end{equation}
Equation~\eqref{gcrtint:eq:centered-residual-scores} then gives
$\bs c_v(r)=\bs c_v(0)+r\bs s_v$.
The matrix $\mathbf D_v$ is unchanged as $r$ varies.
We substitute $\bs c_v(r)=\bs c_v(0)+r\bs s_v$ into
\eqref{gcrtb:eq:gain-at-signal-strength} and expand the quadratic
form to obtain
\begin{equation}\label{gcrtb:eq:gain-perturbation-identity}
 \mathcal A_v(r)-\mathcal A_v(0)
 =2r\bs s_v^\top\mathbf D_v^{-1}\bs c_v(0)
 +r^2\bs s_v^\top\mathbf D_v^{-1}\bs s_v.
\end{equation}

To control the two terms in
\eqref{gcrtb:eq:gain-perturbation-identity}, we first bound
$\bs s_0$ and the copy coefficients $\bs s_m$.
For the observed target data, \eqref{gcrtint:eq:centered-model} gives
$\mathbf X_c=\mathbf S\bs1_2^\top+\mathbf E$.
Recall the representation
$\mathbf H=\omega_n\mathbf P_z+(\mathbf I_N-\mathbf P_z)$
given in \eqref{gcrtb:eq:sphere-ridge-operators}.
It gives
$\mathbf Q=\mathbf H^2=\mathbf I_N-(1-\omega_n^2)\mathbf P_z$.
Since $\mathbf Q\mathbf S=\omega_n^2\mathbf S$,
the definition of $\bs s_v$ in \eqref{gcrtb:eq:signal-score-coefficient} gives
\begin{equation}\label{gcrtb:eq:observed-signal-contribution}
\bs s_0=\omega_n^2\bs1_2\frac{\mathbf S^\top\mathbf E\bs v_-}{n}
 +\frac{\mathbf E^\top\mathbf E}{n}\bs v_-
 -(1-\omega_n^2)\frac{\mathbf E^\top\mathbf P_z\mathbf E}{n}\bs v_-.
\end{equation}
We bound the three terms in
\eqref{gcrtb:eq:observed-signal-contribution} separately.
Equation~\eqref{gcrtb:eq:sphere-basic-rates} gives
$\mathbf S^\top\mathbf E/n=O_{\mathbb P}(n^{-1/2})$.
Conditional on $\mathbf S$, the vector
$\mathbf E^\top\mathbf S/\|\mathbf S\|$ has distribution
$\mathcal N_2(\bs0,\mathbf I_2)$, so
$\mathbf E^\top\mathbf P_z\mathbf E=O_{\mathbb P}(1)$.
Also, $\mathbf E$ has $N=n-1$ independent
$\mathcal N_2(\bs0,\mathbf I_2)$ rows by
\eqref{gcrtint:eq:centered-model}. Chebyshev's inequality applied to its
four Gram entries gives
$\mathbf E^\top\mathbf E/n=\mathbf I_2+O_{\mathbb P}(n^{-1/2})$.
Since $0<\omega_n\leq1$ and $\|\bs v_-\|=1$, we obtain
$$
\begin{aligned}
 \omega_n^2\bs1_2\frac{\mathbf S^\top\mathbf E\bs v_-}{n}
   &=O_{\mathbb P}(n^{-1/2}),\\
 \frac{\mathbf E^\top\mathbf E}{n}\bs v_-
   &=\bs v_-+O_{\mathbb P}(n^{-1/2}),\\
 (1-\omega_n^2)\frac{\mathbf E^\top\mathbf P_z\mathbf E}{n}\bs v_-
   &=O_{\mathbb P}(n^{-1}).
\end{aligned}
$$
Adding these terms with the signs in
\eqref{gcrtb:eq:observed-signal-contribution} gives
\begin{equation}\label{gcrtb:eq:observed-signal-bound}
\bs s_0=\bs v_-+O_{\mathbb P}(n^{-1/2}).
\end{equation}

For  $\bs s_m$ ($1\leq m\leq M$), define
$\bs q=(\mathbf I_N-\mathbf P_z)\mathbf E\bs v_-$ and
$\xi_{jm}=\bs f_j^{(m)\top}\bs q/n$ for $j=1,2$.
Insert \eqref{gcrtb:eq:candidate-sphere-representation} into the
definition of $\bs s_m$ in
\eqref{gcrtb:eq:signal-score-coefficient}. The identities
$\mathbf Q\mathbf S=\omega_n^2\mathbf S$ and
$\mathbf Q\bs f_j^{(m)}=\bs f_j^{(m)}$ give
\begin{equation}\label{gcrtb:eq:copy-signal-decomposition}
\bs s_m=\omega_n^2\widehat{\bs a}\,\mathbf S^\top\mathbf E\bs v_-/n
       +(\xi_{1m},\xi_{2m})^\top.
\end{equation}
By \eqref{gcrtb:eq:sphere-factorisation}, the copy residuals are
independent of the original residuals given $\mathcal C$.
Equation~\eqref{gcrtb:eq:sphere-decomposition} gives
$\mathbf E_j=(\widehat a_j-1)\mathbf S+\bs f_j^{(0)}$ for $j=1,2$.
Thus, given $\mathcal C$, the original residuals determine $\mathbf E$,
so further conditioning on $(\mathbf S,\mathbf E)$ leaves the conditional
law of $\bs f_j^{(m)}$ unchanged.
It is still uniform on the sphere of radius $\rho_j$ in
$\operatorname{col}(\mathbf I_N-\mathbf P_z)$, of dimension $d=n-2$.
Under this conditioning, $\bs q$ is fixed, while $\bs f_j^{(m)}$ has
mean zero and covariance $(\rho_j^2/d)(\mathbf I_N-\mathbf P_z)$.
Since $\bs q$ lies in this subspace, the definition of $\xi_{jm}$ gives
$$
\mathbb E(\xi_{jm}\mid\mathbf S,\mathbf E)=0,\qquad
\mathbb E(\xi_{jm}^2\mid\mathbf S,\mathbf E)
=\frac{\rho_j^2\|\bs q\|^2}{n^2d}.
$$

Given $\mathbf S$, the vectors $\mathbf E_1$ and $\mathbf E_2$ are
independent $\mathcal N_N(0,\mathbf I_N)$ vectors, as in Step~(ii).
Hence $\mathbf E\bs v_-=(\mathbf E_1-\mathbf E_2)/\sqrt2$
also has distribution $\mathcal N_N(0,\mathbf I_N)$.
The matrix $\mathbf I_N-\mathbf P_z$ is an orthogonal projection of
rank $N-1=n-2=d$.
By \eqref{gcrtb:eq:sphere-decomposition} and the definition of $\bs q$,
$\rho_j^2$ and $\|\bs q\|^2$ are the squared norms of the projections
of $\mathbf E_j$ and $\mathbf E\bs v_-$, respectively.
Thus each has the $\chi_d^2$ distribution conditional on $\mathbf S$.
Their second moments are $d(d+2)$.
The tower property and Cauchy--Schwarz give
$$
\begin{aligned}
\mathbb E\xi_{jm}^2
&=\frac{\mathbb E(\rho_j^2\|\bs q\|^2)}{n^2d}\\
&\leq\frac{\{\mathbb E \rho_j^4\,\mathbb E\|\bs q\|^4\}^{1/2}}{n^2d}
=\frac{d+2}{n^2}=\frac1n.
\end{aligned}
$$
For any $B>0$, Markov's inequality and a union bound over the $2M$
pairs $(j,m)$ give
$\mathbb P(\max_{j,m}|\xi_{jm}|>B/\sqrt n)\leq2M/B^2$.
Since $M$ is fixed,
\begin{equation}\label{gcrtb:eq:xi-max-bound}
\max_{j,m}|\xi_{jm}|=O_{\mathbb P}(n^{-1/2}).
\end{equation}

Equation~\eqref{gcrtb:eq:sphere-basic-rates} gives
$\omega_n^2\widehat{\bs a}\,\mathbf S^\top\mathbf E\bs v_-/n
=O_{\mathbb P}(n^{-1/2})$.
Substitution of this bound and \eqref{gcrtb:eq:xi-max-bound} into
\eqref{gcrtb:eq:copy-signal-decomposition} gives the bound for
$\bs s_m$. Together with \eqref{gcrtb:eq:observed-signal-bound},
this yields
\begin{equation}\label{gcrtb:eq:signal-coefficient-bounds}
 \bs s_0=\bs v_-+O_{\mathbb P}(n^{-1/2}),
 \qquad
 \max_{1\leq m\leq M}\|\bs s_m\|=O_{\mathbb P}(n^{-1/2}).
\end{equation}

We now evaluate the coefficients of $2r$ and $r^2$ in
\eqref{gcrtb:eq:gain-perturbation-identity} for $v=0$.
We first control $\mathbf D_0^{-1}$ and
$\mathbf D_0^{-1}\bs c_0(0)$.
Equation~\eqref{gcrtb:eq:direction-limits} gives
$\mathbf D_\infty=\mathbf I_2+\tfrac14\bs1_2\bs1_2^\top
\succeq\mathbf I_2$, and hence
\begin{equation}\label{gcrtb:eq:population-inverse-score-value}
\begin{aligned}
\mathbf D_\infty\bs1_2
&=\bs1_2+\tfrac14\bs1_2(\bs1_2^\top\bs1_2)
=\tfrac32\bs1_2,\\
\mathbf D_\infty^{-1}(\tfrac14\bs1_2)
&=\tfrac14\cdot\tfrac23\bs1_2=\tfrac16\bs1_2.
\end{aligned}
\end{equation}
By \eqref{gcrtb:eq:null-center-bounds},
$\|\mathbf D_0-\mathbf D_\infty\|_{\mathrm{op}}\leq1/2$
with probability tending to one.
On this event, Weyl's inequality gives
$\lambda_{\min}(\mathbf D_0)\geq1/2$, so
$\|\mathbf D_0^{-1}\|_{\mathrm{op}}\leq2$.
Thus $\|\mathbf D_0^{-1}\|_{\mathrm{op}}=O_{\mathbb P}(1)$.
The inverse identity and \eqref{gcrtb:eq:null-center-bounds} then give
\begin{equation}\label{gcrtb:eq:observed-inverse-gram-rate}
\begin{aligned}
\mathbf D_0^{-1}-\mathbf D_\infty^{-1}
&=\mathbf D_0^{-1}(\mathbf D_\infty-\mathbf D_0)
  \mathbf D_\infty^{-1}\\
&=O_{\mathbb P}(n^{-1/2}).
\end{aligned}
\end{equation}
Combining \eqref{gcrtb:eq:observed-inverse-gram-rate} with the score bound
in \eqref{gcrtb:eq:null-center-bounds}, we obtain
\begin{equation}\label{gcrtb:eq:observed-inverse-score-bound}
\begin{aligned}
\mathbf D_0^{-1}\bs c_0(0)
&=\mathbf D_\infty^{-1}(\tfrac14\bs1_2)
  +\mathbf D_0^{-1}(\bs c_0(0)-\tfrac14\bs1_2)\\
&\quad
  +(\mathbf D_0^{-1}-\mathbf D_\infty^{-1})(\tfrac14\bs1_2)\\
&=\tfrac16\bs1_2+O_{\mathbb P}(n^{-1/2}),
\end{aligned}
\end{equation}
where the last equality uses \eqref{gcrtb:eq:population-inverse-score-value}
to evaluate $\mathbf D_\infty^{-1}(\tfrac14\bs1_2)$.

For the coefficient of $2r$, we substitute
\eqref{gcrtb:eq:observed-signal-bound} and
\eqref{gcrtb:eq:observed-inverse-score-bound}:
\begin{equation}\label{gcrtb:eq:observed-linear-gain-coefficient}
\begin{aligned}
\bs s_0^\top\mathbf D_0^{-1}\bs c_0(0)
&=\bigl(\bs v_-+O_{\mathbb P}(n^{-1/2})\bigr)^\top
  \bigl(\tfrac16\bs1_2+O_{\mathbb P}(n^{-1/2})\bigr)\\
&=\tfrac16\bs v_-^\top\bs1_2+O_{\mathbb P}(n^{-1/2})\\
&=O_{\mathbb P}(n^{-1/2}),
\end{aligned}
\end{equation}
where the last equality uses $\bs v_-^\top\bs1_2=0$.

For the coefficient of $r^2$, the definition in
\eqref{gcrtb:eq:direction-limits} gives
$\mathbf D_\infty\bs v_-=
\bs v_-+\tfrac14\bs1_2(\bs1_2^\top\bs v_-)=\bs v_-$.
Thus $\mathbf D_\infty^{-1}\bs v_-=\bs v_-$.
Equations~\eqref{gcrtb:eq:observed-signal-bound} and
\eqref{gcrtb:eq:observed-inverse-gram-rate} give
\begin{equation}\label{gcrtb:eq:observed-quadratic-gain-coefficient}
\begin{aligned}
\bs s_0^\top\mathbf D_0^{-1}\bs s_0
&=\bs v_-^\top\mathbf D_\infty^{-1}\bs v_-
  +O_{\mathbb P}(n^{-1/2})\\
&=\|\bs v_-\|^2+O_{\mathbb P}(n^{-1/2})\\
&=1+O_{\mathbb P}(n^{-1/2}).
\end{aligned}
\end{equation}
Substitution of
\eqref{gcrtb:eq:observed-linear-gain-coefficient} and
\eqref{gcrtb:eq:observed-quadratic-gain-coefficient} into
\eqref{gcrtb:eq:gain-perturbation-identity} gives
\begin{equation}\label{gcrtb:eq:observed-signal-perturbation}
\begin{aligned}
\mathcal A_0(r)-\mathcal A_0(0)
&=2r\,O_{\mathbb P}(n^{-1/2})
  +r^2\{1+O_{\mathbb P}(n^{-1/2})\}\\
&=r^2+O_{\mathbb P}(r/\sqrt n+r^2/\sqrt n).
\end{aligned}
\end{equation}

For the copies, the decompositions in \eqref{gcrtb:eq:null-score}
and the joint limit in \eqref{gcrtb:eq:sphere-joint-limit} show that
$\bs c_m(0)-\bs b_n$ and $\mathbf D_m-\mathbf G_n$ are
$O_{\mathbb P}(n^{-1/2})$ uniformly over $m=1,\ldots,M$,
since $M$ is fixed.
Combining these bounds with \eqref{gcrtb:eq:random-center-bounds} gives
$$
\begin{aligned}
\max_{1\leq m\leq M}\|\bs c_m(0)-\tfrac14\bs1_2\|
&=O_{\mathbb P}(n^{-1/2}),\\
\max_{1\leq m\leq M}\|\mathbf D_m-\mathbf D_\infty\|_{\mathrm{op}}
&=O_{\mathbb P}(n^{-1/2}).
\end{aligned}
$$
The score norms are therefore bounded in probability.
Since $\mathbf D_\infty\succeq\mathbf I_2$, the eigenvalue argument
preceding \eqref{gcrtb:eq:observed-inverse-gram-rate} applies
uniformly over the copies and bounds their inverse Gram norms.
Thus
$$
\max_{1\leq m\leq M}\|\bs c_m(0)\|=O_{\mathbb P}(1),
\qquad
\max_{1\leq m\leq M}\|\mathbf D_m^{-1}\|_{\mathrm{op}}
=O_{\mathbb P}(1).
$$
Combining these bounds with \eqref{gcrtb:eq:signal-coefficient-bounds}
yields
\begin{equation}\label{gcrtb:eq:copy-gain-coefficient-bounds}
\begin{aligned}
 \max_{1\leq m\leq M}
 |\bs s_m^\top\mathbf D_m^{-1}\bs c_m(0)|
   &=O_{\mathbb P}(n^{-1/2}),\\
 \max_{1\leq m\leq M}
 |\bs s_m^\top\mathbf D_m^{-1}\bs s_m|
   &=O_{\mathbb P}(n^{-1}).
\end{aligned}
\end{equation}
For the copies, we first take absolute values in
\eqref{gcrtb:eq:gain-perturbation-identity} and apply the triangle inequality.
We then maximise over $m=1,\ldots,M$ and use
\eqref{gcrtb:eq:copy-gain-coefficient-bounds}.
This gives
\begin{equation}\label{gcrtb:eq:signal-perturbation}
\begin{aligned}
\max_{1\leq m\leq M}|\mathcal A_m(r)-\mathcal A_m(0)|
&\leq2r\max_{1\leq m\leq M}
  |\bs s_m^\top\mathbf D_m^{-1}\bs c_m(0)|\\
&\quad+r^2\max_{1\leq m\leq M}
  |\bs s_m^\top\mathbf D_m^{-1}\bs s_m|\\
&=2r\,O_{\mathbb P}(n^{-1/2})+r^2O_{\mathbb P}(n^{-1})\\
&=O_{\mathbb P}(r/\sqrt n+r^2/n).
\end{aligned}
\end{equation}
Take $r=r_n=n^{-1/3}$ in
\eqref{gcrtb:eq:observed-signal-perturbation} and
\eqref{gcrtb:eq:signal-perturbation}.
On the $n^{-1/2}$ scale, the observed and copy changes satisfy
$$
\begin{aligned}
 \sqrt n\{\mathcal A_0(r_n)-\mathcal A_0(0)\}
   &=n^{-1/6}+O_{\mathbb P}(n^{-1/3}+n^{-2/3})
      \xrightarrow{\mathbb P}0,\\
 \max_{1\leq m\leq M}\sqrt n\,
 |\mathcal A_m(r_n)-\mathcal A_m(0)|
   &=O_{\mathbb P}(n^{-1/3}+n^{-7/6})
      \xrightarrow{\mathbb P}0.
\end{aligned}
$$
Thus $\max_{0\leq v\leq M}\sqrt n\,
|\mathcal A_v(r_n)-\mathcal A_v(0)|\xrightarrow{\mathbb P}0$,
as required at the start of Step~(iv).
We use $\mathcal K_n$ as defined at $r=0$ in
\eqref{gcrtb:eq:common-reference-gain}.
Since $M$ is fixed, Slutsky's theorem and the joint limit in
\eqref{gcrtb:eq:null-gain-limit} give
\begin{equation}\label{gcrtb:eq:joint-rank-limit}
 \sqrt n\bigl(\mathcal A_0(r_n)-\mathcal K_n,\ldots,
             \mathcal A_M(r_n)-\mathcal K_n\bigr)
 \Longrightarrow (Z_0,\ldots,Z_M),
\end{equation}
where $Z_0,\ldots,Z_M$ are independent
$\mathcal{N}(0,73/324)$ variables.

\medskip
\noindent\textit{(v) Rank comparisons for $\mathcal A_v$ and $\mathcal B_v$.}\par\nopagebreak
We now derive both rejection probabilities from
\eqref{gcrtb:eq:joint-rank-limit}.
For each statistic, we will use exchangeability and the absence of ties
in its limiting vector. 

We start with $\mathrm{pVal}_{\rm unc}$.
Let $k_\alpha=\lfloor\alpha(M+1)\rfloor$ and define
$R(\bs x)=1+\sum_{m=1}^M\mathbf1\{x_m\geq x_0\}$ for
$\bs x=(x_0,\ldots,x_M)\in\mathbb R^{M+1}$.
By \eqref{eqn:pval-1},
$$
\{\mathrm{pVal}_{\rm unc}\leq\alpha\}
=\left\{R((\mathcal A_v(r_n))_{v=0}^M)\leq k_\alpha\right\}.
$$
The common shift by $-\mathcal K_n$ and the positive scaling by
$\sqrt n$ preserve the rank:
$$
R((\mathcal A_v(r_n))_{v=0}^M)
=R\bigl((\sqrt n\{\mathcal A_v(r_n)-\mathcal K_n\})_{v=0}^M\bigr).
$$

The indicator of $\{R(\bs x)\leq k_\alpha\}$ is locally constant
whenever $x_m\ne x_0$ for every $m=1,\ldots,M$.
The Gaussian limit in \eqref{gcrtb:eq:joint-rank-limit} assigns
probability zero to the excluded hyperplanes.
The limit rank $R((Z_v)_{v=0}^M)$ is uniform on
$\{1,\ldots,M+1\}$ by exchangeability and the absence of ties.
The continuous mapping theorem and boundedness of the indicator
therefore give
\begin{equation}\label{gcrtb:eq:uncentred-rank-probability}
\mathbb P\left(R\bigl(\sqrt n(\mathcal A_v(r_n)-\mathcal K_n)_{v=0}^M\bigr)
                  \leq k_\alpha\right)
\longrightarrow\mathbb P(R((Z_v)_{v=0}^M)\leq k_\alpha)
=\frac{k_\alpha}{M+1}.
\end{equation}
Together with Step~(i), \eqref{gcrtb:eq:uncentred-rank-probability}
proves \eqref{gcrtcent:eq:strict-comparison}.

\smallskip
We next consider $\mathcal B_v=(\mathcal A_v-\bar{\mathcal A})^2$
from \eqref{gcrtb:eq:scalar-centered-statistic}, evaluated at $r=r_n$.
The cancellation of $\mathcal K_n$ and the scale change are explicit in
$$
n\mathcal B_v
=\left[\sqrt n\{\mathcal A_v(r_n)-\mathcal K_n\}
-\frac1{M+1}\sum_{u=0}^M
 \sqrt n\{\mathcal A_u(r_n)-\mathcal K_n\}\right]^2.
$$
The continuous mapping theorem and \eqref{gcrtb:eq:joint-rank-limit}
therefore give
\begin{equation}\label{gcrtb:eq:scalar-centred-joint-limit}
 n(\mathcal B_0,\ldots,\mathcal B_M)
 \Longrightarrow((Z_0-\bar Z)^2,\ldots,(Z_M-\bar Z)^2), 
\end{equation}
where the $Z_v$ are those in \eqref{gcrtb:eq:joint-rank-limit} and $\bar Z=\frac1{M+1}\sum_{v=0}^M Z_v$. 
The vector $((Z_v-\bar Z)^2)_{v=0}^M$ is exchangeable.
We can show that there is no tie after centring and squaring.
For distinct $i,j\in\{0,\ldots,M\}$,
$$
(Z_i-\bar Z)^2-(Z_j-\bar Z)^2
=(Z_i-Z_j)(Z_i+Z_j-2\bar Z).
$$
This means that a tie requires $Z_i=Z_j$ or $Z_i+Z_j-2\bar Z=0$.
Since $M\geq2$, there is an index $k\notin\{i,j\}$, whose coefficient
in the second relation is $-2/(M+1)\ne0$.
Both relations therefore define proper hyperplanes, which have
probability zero under the nondegenerate joint Gaussian law of
$(Z_0,\ldots,Z_M)$.
Thus the limit vector has no ties, and the rank of $(Z_0-\bar Z)^2$
is uniform on $\{1,\ldots,M+1\}$.

Equation~\eqref{eqn:pval-1} and invariance of ranks under positive
scaling give
$$
\{\mathrm{pVal}_{\mathcal B}\leq\alpha\}
=\left\{R((n\mathcal B_v)_{v=0}^M)\leq k_\alpha\right\}.
$$
Apply the rank indicator argument from
\eqref{gcrtb:eq:uncentred-rank-probability} to
\eqref{gcrtb:eq:scalar-centred-joint-limit}. We obtain
$$
\mathbb P(\mathrm{pVal}_{\mathcal B}\leq\alpha)
\longrightarrow\frac{k_\alpha}{M+1}
=\frac{\lfloor\alpha(M+1)\rfloor}{M+1}.
$$
This completes part (b).

\end{proof}

\subsection[Proofs for MS-CRT robustness]{Proofs for MS-CRT robustness}\label{app:proof-comparison-class-properties}

We first prove Theorem~\ref{thm: MS-CRT robustness}, then the convexity and mixture results in Appendix~\ref{app:comparison-class-properties}.

\noindent\begin{minipage}{\linewidth}
\subsubsection[Proof of Theorem~\ref{thm: MS-CRT robustness}]{Proof of Theorem~\ref{thm: MS-CRT robustness}.}\label{app:proof-mscrt-robustness}
The common kernel, copy construction, and scoring notation used in this proof are specified in the technical setting for Theorem~\ref{thm: MS-CRT robustness} in Appendix~\ref{app:additional-robustness}.
\end{minipage}\par
\begin{proof}
Fix $(\bs{y},\mathbf x_{\mc S})$ for which the conditional independence
in \(H_0\) gives
\begin{equation}
\label{eq:null-conditional-law-at-fixed-Y-XS}
\mathcal L_*(\mathbf X_{\mc T}\mid
(\bs{Y}=\bs{y},\mathbf X_{\mc S}=\mathbf x_{\mc S}))
=
\mathcal L_*(\mathbf X_{\mc T}\mid
\mathbf X_{\mc S}=\mathbf x_{\mc S}). 
\end{equation}
Define
\begin{equation}\label{eq:robustness-rejection-function}
g_{\bs{y},\mathbf x_{\mc S}}(\mathbf x_{\mc T})
=
\mathbb P\!\left(
p_T(\bs{y},\mathbf x_{\mc S},\mathbf x_{\mc T},W)\leq\alpha
\right),
\end{equation}
where the probability is over $W\sim\nu$ only. Clearly,
$0\leq g_{\bs{y},\mathbf x_{\mc S}}\leq1$.

We first establish the size bound under any fixed \(R\in\widetilde{\mc P}_{\Phi}(\mathbf x_{\mc S})\). Let \(\mathbf Z\sim R\) and
\(W\sim\nu\) be independent, and set
\(U=\phi(\mathbf Z;\mathbf x_{\mc S})\).
By the membership of $R$
in \eqref{eq:Phi-comparison-class-at-fixed-XS}, conditional on \(U=u\), the
matrix $\mathbf Z$ has law \(L_{\mathbf x_{\mc S},u}\), while \(W\) remains independent of \(\mathbf Z\). 
Therefore, 
$\mathcal L(\mathbf Z,W\mid U=u)
=L_{\mathbf x_{\mc S},u}\otimes\nu$
for almost every $u$. Since $\nu=\nu_{\mathrm c}\otimes\nu_{\mathrm s}$, we have \begin{equation}\label{eq:robustness-conditional-product}
\mathcal L(\mathbf Z,W_{\mathrm c},W_{\mathrm s}\mid U=u)
=L_{\mathbf x_{\mc S},u}\otimes\nu_{\mathrm c}\otimes\nu_{\mathrm s}.
\end{equation}

Define
$$
\mathcal Z=\bigl(\mathbf Z,
Q_1(\mathbf x_{\mc S},\mathbf Z,W_{\mathrm c}),\ldots,
Q_M(\mathbf x_{\mc S},\mathbf Z,W_{\mathrm c})\bigr).
$$
Marginalising \eqref{eq:robustness-conditional-product} over
$W_{\mathrm s}$ gives the conditional law
$L_{\mathbf x_{\mc S},u}\otimes\nu_{\mathrm c}$ for
$(\mathbf Z,W_{\mathrm c})$. Thus
\eqref{eq:Q-exchangeability-at-fixed-XS} applies with $\varphi=u$.
Therefore, for every permutation $\pi$ of $\{0,\ldots,M\}$, we have
\begin{equation}\label{eq:robustness-candidate-exchangeability}
\mathcal L(\pi\mathcal Z\mid U=u)=\mathcal L(\mathcal Z\mid U=u).
\end{equation}
By \eqref{eq:robustness-conditional-product}, $W_{\mathrm s}$ has
conditional law $\nu_{\mathrm s}$ and is independent of
$(\mathbf Z,W_{\mathrm c})$ given $U=u$.
Since $\mathcal Z$ is a measurable function of $(\mathbf Z,W_{\mathrm c})$,
it is also independent of $W_{\mathrm s}$ under this conditioning.
Then 
\begin{equation}\label{eq:robustness-joint-permutation-invariance}
\begin{aligned}
\mathcal L(\pi\mathcal Z,\pi W_{\mathrm s}\mid U=u)
&\overset{\mathrm{(a)}}=\mathcal L(\pi\mathcal Z \mid U=u)\otimes \mathcal L(\pi W_{\mathrm s}\mid U=u)\\
&\overset{\mathrm{(b)}}=
\mathcal L(\mathcal Z\mid U=u)\otimes\nu_{\mathrm s}\\
&\overset{\mathrm{(c)}}=
\mathcal L(\mathcal Z,W_{\mathrm s}\mid U=u),
\end{aligned}
\end{equation}
where (a) and (c) use the conditional independence, and (b) uses \eqref{eq:robustness-candidate-exchangeability} and \eqref{eq:mscrt-scoring-randomness-conditions}. 

Define $(T_0,\ldots,T_M)=
\mathscr T(\bs y,\mathbf x_{\mc S},\mathcal Z;W_{\mathrm s})$ as in
\eqref{eq:robustness-score-array}. Then
\begin{equation}\label{eq:robustness-score-exchangeability}
\begin{aligned}
\mathcal L\bigl(\pi(T_0,\ldots,T_M)\mid U=u\bigr)
&\overset{\mathrm{(a)}}=
\mathcal L\bigl(\mathscr T(\bs y,\mathbf x_{\mc S},\pi\mathcal Z;
\pi W_{\mathrm s})\mid U=u\bigr)\\
&\overset{\mathrm{(b)}}=
\mathcal L\bigl(\mathscr T(\bs y,\mathbf x_{\mc S},\mathcal Z;
 W_{\mathrm s})\mid U=u\bigr)
= 
\mathcal L\bigl((T_0,\ldots,T_M)\mid U=u\bigr), 
\end{aligned}
\end{equation}
where (a) follows from the pointwise
equivariance \eqref{eq:mscrt-randomised-scoring-equivariance}, and 
(b) follows from \eqref{eq:robustness-joint-permutation-invariance}. 

Since $\pi$ is arbitrary, the scores are
exchangeable conditional on $U=u$ for almost every $u$.
For deterministic scoring, $W_{\mathrm s}$ is constant and
\eqref{eq:mscrt-scoring-equivariance} supplies the pointwise equivariance,
so the same argument applies.
Consequently, for the p-value defined in \eqref{eqn:pval-1}, we have 
\begin{equation}
\label{eq:conditional-rank-bound-at-fixed-XS}
\mathbb P\!\left(
p_T(\bs{y},\mathbf x_{\mc S},\mathbf Z,W)\leq\alpha
\mid U
\right)
\leq\alpha.
\end{equation}
Averaging over $U$ gives
\begin{equation}
\label{eq:comparison-size-at-fixed-XS}
\begin{aligned}
\int
g_{\bs{y},\mathbf x_{\mc S}}(\mathbf x_{\mc T})\,
R(d\mathbf x_{\mc T})
&=
\mathbb P\!\left(
p_T(\bs{y},\mathbf x_{\mc S},\mathbf Z,W)\leq\alpha
\right)\\
&=
\mathbb E\!\left[
\mathbb P\!\left(
p_T(\bs{y},\mathbf x_{\mc S},\mathbf Z,W)\leq\alpha
\mid U
\right)
\right]\\
&\leq\alpha,
\end{aligned}
\end{equation}
{where the first equality follows from
\eqref{eq:robustness-rejection-function} and the independence of
$\mathbf Z$ and $W$, the second equality uses the tower property, and the final
inequality follows from \eqref{eq:conditional-rank-bound-at-fixed-XS}.}

We now bound the rejection probability under the true conditional law.
Since $W\indp(\bs Y,\mathbf X)$ and $W\sim\nu$, for almost every data value $(\bs y,\mathbf x_{\mc S},\mathbf x_{\mc T})$
under the true law, we have
\begin{equation}\label{eq:robustness-full-data-rejection}
\begin{aligned}
&\mathbb P_*\!\left(
p_T(\bs Y,\mathbf X_{\mc S},\mathbf X_{\mc T},W)\leq\alpha
\,\middle|\,
\bs Y=\bs y,\,
\mathbf X_{\mc S}=\mathbf x_{\mc S},\,
\mathbf X_{\mc T}=\mathbf x_{\mc T}\right)\\
={}&\int\mathbf 1\!\left\{
p_T(\bs y,\mathbf x_{\mc S},\mathbf x_{\mc T},w)\leq\alpha
\right\}\nu(dw)\\
={}&g_{\bs y,\mathbf x_{\mc S}}(\mathbf x_{\mc T}).
\end{aligned}
\end{equation}

Conditioning on $(\bs Y,\mathbf X_{\mc S})=(\bs y,\mathbf x_{\mc S})$
fixes the first two arguments of $p_T$. 
The tower property applied to
\eqref{eq:robustness-full-data-rejection} then gives
\begin{equation}\label{eq:robustness-conditional-rejection-representation}
\begin{aligned}
&\mathbb P_*\!\left(
p_T(\bs y,\mathbf x_{\mc S},\mathbf X_{\mc T},W)\leq\alpha
\mid(\bs Y=\bs y,\mathbf X_{\mc S}=\mathbf x_{\mc S})\right)\\
={}&\mathbb E_*\!\left[
g_{\bs y,\mathbf x_{\mc S}}(\mathbf X_{\mc T})
\mid(\bs Y=\bs y,\mathbf X_{\mc S}=\mathbf x_{\mc S})\right]\\
={}&\int g_{\bs y,\mathbf x_{\mc S}}(\mathbf x_{\mc T})
\mathcal L_*\!\left(d\mathbf x_{\mc T}\mid
(\bs Y=\bs y,\mathbf X_{\mc S}=\mathbf x_{\mc S})\right).
\end{aligned}
\end{equation}

Under $H_0$, 
the conditional target law in
\eqref{eq:robustness-conditional-rejection-representation} equals 
$\mathcal L_*(\mathbf X_{\mc T}\mid\mathbf X_{\mc S}=\mathbf x_{\mc S})$.
Using the total variation bound, $0\leq g_{\bs y,\mathbf x_{\mc S}}\leq1$, and \eqref{eq:comparison-size-at-fixed-XS}, we have
\begin{equation}\label{eq:TV-transfer-at-fixed-XS}
\begin{aligned}
&\mathbb P_*\!\left(
p_T(\bs y,\mathbf x_{\mc S},\mathbf X_{\mc T},W)\leq\alpha
\mid(\bs Y=\bs y,\mathbf X_{\mc S}=\mathbf x_{\mc S})\right)\\
={}&\int g_{\bs y,\mathbf x_{\mc S}}(\mathbf x_{\mc T})
\mathcal L_*\!\left(d\mathbf x_{\mc T}\mid
\mathbf X_{\mc S}=\mathbf x_{\mc S}\right)\\
\leq{}&\int g_{\bs y,\mathbf x_{\mc S}}(\mathbf x_{\mc T})
R(d\mathbf x_{\mc T})
+d_{\mathrm{TV}}\!\left(
\mathcal L_*(\mathbf X_{\mc T}\mid\mathbf X_{\mc S}=\mathbf x_{\mc S}),R
\right)\\
\leq{}&\alpha+d_{\mathrm{TV}}\!\left(
\mathcal L_*(\mathbf X_{\mc T}\mid\mathbf X_{\mc S}=\mathbf x_{\mc S}),R
\right).
\end{aligned}
\end{equation}
Since \eqref{eq:TV-transfer-at-fixed-XS} holds for every
$R\in\widetilde{\mc P}_{\Phi}(\mathbf x_{\mc S})$, taking the infimum
over this class and using the definition of
$\Delta_{\Phi}(\mathbf x_{\mc S})$ gives
\eqref{eq:MSCRT-conditional-robustness}.

\end{proof}

\subsubsection[Proof of Proposition~\ref{prop:convexity-of-enlarged-class}]{Proof of \Cref{prop:convexity-of-enlarged-class}.}\label{app:proof-convexity-enlarged-class}
\begin{proof}
Fix $\mathbf x_{\mc S}$. For each $j\in[J]$, let $\eta_j$ be the law of $\phi(\mathbf Z;\mathbf x_{\mc S})$ when $\mathbf Z\sim R_j$, and define $\eta^\pi=\sum_{j=1}^J\pi_j\eta_j$.
For $A\in\mathscr X$ and $B\in\mathscr V(\mathbf x_{\mc S})$, the membership of each $R_j$ in \eqref{eq:Phi-comparison-class-at-fixed-XS} gives

\begin{equation}\label{eq:convex-mixture-kernel-identity}
\begin{aligned}
R^\pi\!\left(A\cap\{\mathbf z:\phi(\mathbf z;\mathbf x_{\mc S})\in B\}\right)
&=\sum_{j=1}^J\pi_j\int_B L_{\mathbf x_{\mc S},\varphi}(A)\,\eta_j(d\varphi)\\
&=\int_B L_{\mathbf x_{\mc S},\varphi}(A)\,\eta^\pi(d\varphi).
\end{aligned}
\end{equation}

Let $\mathbf Z\sim R^\pi$ and define $U=\phi(\mathbf Z;\mathbf x_{\mc S})$.
Then, the law of $U$ is $\eta^\pi$, supported on $\mc V(\mathbf x_{\mc S})$.

Equation~\eqref{eq:convex-mixture-kernel-identity} identifies
$L_{\mathbf x_{\mc S},U}$ as a version of the conditional law of
$\mathbf Z$ given $U$, so \eqref{eq:Phi-comparison-class-at-fixed-XS}
gives $R^\pi\in\widetilde{\mc P}_{\Phi}(\mathbf x_{\mc S})$.
The convex hull inclusion follows because $\mc P_{\Theta}(\mathbf x_{\mc S})\subseteq\widetilde{\mc P}_{\Phi}(\mathbf x_{\mc S})$.
\end{proof}

\subsubsection[Proof of Proposition~\ref{prop:mixtures-in-enlarged-class}]{Proof of \Cref{prop:mixtures-in-enlarged-class}.}\label{app:proof-mixtures-enlarged-class}

\begin{proof}
Let \(\mathbf Z\sim P^\Pi_{\mathbf x_{\mc S}}\), and set
\(U=\phi(\mathbf Z;\mathbf x_{\mc S})\). For each \(\theta\), define the law
of the sufficient statistic under \(P_{\theta,\mathbf x_{\mc S}}\) by
\[
\eta_{\theta,\mathbf x_{\mc S}}(B)
=
P_{\theta,\mathbf x_{\mc S}}
\left(\{\mathbf z:\phi(\mathbf z;\mathbf x_{\mc S})\in B\}\right),
\qquad B\in\mathscr V(\mathbf x_{\mc S}).
\]

The measurability assumption makes $\theta\mapsto\eta_{\theta,\mathbf x_{\mc S}}(B)$ measurable.
Define 
\begin{equation}\label{eq:mixture-statistic-law}
\eta^\Pi_{\mathbf x_{\mc S}}(B)
=
\int_\Theta
\eta_{\theta,\mathbf x_{\mc S}}(B)\,
\Pi_{\mathbf x_{\mc S}}(d\theta), \qquad B\in\mathscr V(\mathbf x_{\mc S}).
\end{equation}
By the definition of $P^\Pi_{\mathbf x_{\mc S}}$, we have 
\[
\begin{aligned}
P^\Pi_{\mathbf x_{\mc S}}\left(\{U\in B\}\right) & = 
P^\Pi_{\mathbf x_{\mc S}}\left(
\{\mathbf z:\phi(\mathbf z;\mathbf x_{\mc S})\in B\} \right)  \\
&=
\int_\Theta
P_{\theta,\mathbf x_{\mc S}} \left(
\{\mathbf z:\phi(\mathbf z;\mathbf x_{\mc S})\in B\} \right) \,
\Pi_{\mathbf x_{\mc S}}(d\theta)\\
&=
\int_\Theta
\eta_{\theta,\mathbf x_{\mc S}}(B)\,
\Pi_{\mathbf x_{\mc S}}(d\theta)\\
&=
\eta^\Pi_{\mathbf x_{\mc S}}(B).
\end{aligned}
\]
Thus \(\eta^\Pi_{\mathbf x_{\mc S}}\) is the law of \(U\), and
\(U\in\mc V(\mathbf x_{\mc S})\) almost surely.

It remains to identify the conditional law of \(\mathbf Z\) given \(U\). Let
\(A\) be a measurable set in the sample space of $\mathbf{X}_{\mc T}$ and let
$B\in\mathscr V(\mathbf x_{\mc S})$. Since
\(L_{\mathbf x_{\mc S},\varphi}\) is the common conditional law of $\mathbf{X}_{\mc T}$ given \(\Phi=\varphi\) under each
\(P_{\theta,\mathbf x_{\mc S}}\), we have 
\[
P_{\theta,\mathbf x_{\mc S}}
\!\left(A\cap
\{\mathbf z:\phi(\mathbf z;\mathbf x_{\mc S})\in B\}\right)
=
\int_B
L_{\mathbf x_{\mc S},\varphi}(A)\,
\eta_{\theta,\mathbf x_{\mc S}}(d\varphi).
\]
Integrating both sides with respect to $\Pi_{\mathbf x_{\mc S}}(d\theta)$ gives
\begin{equation}\label{eq:mixture-joint-law-before-integration}
\begin{aligned}
P^\Pi_{\mathbf x_{\mc S}}(\mathbf Z\in A,U\in B)
&=
\int_\Theta
\int_B
L_{\mathbf x_{\mc S},\varphi}(A)\,
\eta_{\theta,\mathbf x_{\mc S}}(d\varphi)\,
\Pi_{\mathbf x_{\mc S}}(d\theta).
\end{aligned}
\end{equation}
The definition of \(\eta^\Pi_{\mathbf x_{\mc S}}\) in \eqref{eq:mixture-statistic-law} implies the following identity for every bounded nonnegative measurable function \(f\) on \(\mc V(\mathbf x_{\mc S})\): 
\begin{equation}\label{eq:mixture-integration-identity}
\int_{\mc V(\mathbf x_{\mc S})}
f(\varphi)\,\eta^\Pi_{\mathbf x_{\mc S}}(d\varphi)
=
\int_\Theta
\int_{\mc V(\mathbf x_{\mc S})}
f(\varphi)\,\eta_{\theta,\mathbf x_{\mc S}}(d\varphi)\,
\Pi_{\mathbf x_{\mc S}}(d\theta).
\end{equation}
This identity can be verified using the canonical method: it holds first for indicators by definition of
\(\eta^\Pi_{\mathbf x_{\mc S}}\), then for simple functions by linearity, and
then for bounded nonnegative measurable functions by monotone convergence.

In particular, for the bounded nonnegative measurable function
\(
f(\varphi)
=
\mathbf 1\{\varphi\in B\}L_{\mathbf x_{\mc S},\varphi}(A)
\), equations~\eqref{eq:mixture-joint-law-before-integration} and~\eqref{eq:mixture-integration-identity} give 
\begin{equation}\label{eq:mixture-joint-law-after-integration}
P^\Pi_{\mathbf x_{\mc S}}(\mathbf Z\in A,U\in B)
=
\int_B
L_{\mathbf x_{\mc S},\varphi}(A)\,
\eta^\Pi_{\mathbf x_{\mc S}}(d\varphi).
\end{equation}
Since $\eta^\Pi_{\mathbf x_{\mc S}}$ is the law of $U$,
\eqref{eq:mixture-joint-law-after-integration} shows that
$\mathbf Z\mid U\sim L_{\mathbf x_{\mc S},U}$.
By \eqref{eq:Phi-comparison-class-at-fixed-XS}, we conclude that 
\(P^\Pi_{\mathbf x_{\mc S}}\in
\widetilde{\mc P}_{\Phi}(\mathbf x_{\mc S})\).
\end{proof}

\subsection[Proofs for robustness to graph misspecification]{Proofs for robustness to graph misspecification}
\label{app: pf robust}

The Gaussian error bound uses the
conditional laws in Appendix~\ref{app:gaussian-conditional-laws}.
For \Cref{prop:learned-graph-size-rewrite}, we use the joint Gaussian
model of Section~\ref{sec:ggm-robustness-rewrite} and the blocks of the precision matrix
defined in Appendix~\ref{app:gaussian-conditional-laws}.
The proof uses the reference construction in
Appendix~\ref{app:proof-gaussian-reference}.

\subsubsection[Proof of (\ref{eq:gaussian-misspecification-size-supp})]{Proof of \eqref{eq:gaussian-misspecification-size-supp}.}\label{app:proof-gaussian-discrepancy}

\begin{proof}

At each fixed $\bs x_{\mc S}$, we apply the standard formula for the Kullback--Leibler divergence between Gaussian laws.

Write $D_{\mathrm{KL}}(P\|R)=\int\log(dP/dR)\,dP$ for Kullback--Leibler divergence. The conditional divergence for one row is therefore
$$
\begin{aligned}
&D_{\mathrm{KL}}\!\left(
\mathcal N_t(\bs\mu(\bs x_{\mc S}),\mathbf V)
\,\middle\|\,
\mathcal N_t(\bs\mu_0(\bs x_{\mc S}),\mathbf V_0)
\right)\\
={}&\frac12\left\{\operatorname{tr}(\mathbf V_0^{-1}\mathbf V)-t
+\log\frac{\det\mathbf V_0}{\det\mathbf V}
+\bs b(\bs x_{\mc S})^\top\mathbf V_0^{-1}
\bs b(\bs x_{\mc S})\right\}.
\end{aligned}
$$

Let $P_{\mathbf x_{\mc S}}$ and $R_{\mathbf x_{\mc S}}$ be the true and reference conditional laws of $\mathbf X_{\mc T}$ given $\mathbf X_{\mc S}=\mathbf x_{\mc S}$ in Appendix~\ref{app:gaussian-conditional-laws}.
Both are product laws over the $n$ rows, so their Kullback--Leibler divergence is additive across rows.
The background rows are identically distributed, and the expected divergence for each row is $B/2$ by \eqref{eq:gaussian-covariate-discrepancy-supp}.

Independence of the rows and \eqref{eq:gaussian-covariate-discrepancy-supp} therefore give $\mathbb E_*D_{\mathrm{KL}}(P_{\mathbf X_{\mc S}}\|R_{\mathbf X_{\mc S}})=nB/2$. Pinsker's inequality and Jensen's inequality imply
$$
\mathbb E_*d_{\mathrm{TV}}(P_{\mathbf X_{\mc S}},R_{\mathbf X_{\mc S}})
\leq\left\{\tfrac12\mathbb E_*D_{\mathrm{KL}}(P_{\mathbf X_{\mc S}}\|R_{\mathbf X_{\mc S}})\right\}^{1/2}
=\frac{\sqrt{nB}}{2}.
$$
Since $R_{\mathbf x_{\mc S}}$ belongs to the comparison class for almost every $\mathbf x_{\mc S}$, \eqref{eq:TV-transfer-at-fixed-XS} applies with $R=R_{\mathbf x_{\mc S}}$. Averaging over the true law of $(\bs Y,\mathbf X_{\mc S})$ bounds the rejection probability by $\alpha+\mathbb E_*d_{\mathrm{TV}}(P_{\mathbf X_{\mc S}},R_{\mathbf X_{\mc S}})$, proving \eqref{eq:gaussian-misspecification-size-supp}.\qedhere

\end{proof}

\subsubsection[Gaussian references compatible with a graph]{Gaussian references compatible with a graph}\label{app:proof-gaussian-reference}

Fix a graph $G$. We construct Gaussian reference laws for the proof of
Proposition~\ref{prop:learned-graph-size-rewrite}.
Let $\mathbf C\succ0$ be a symmetric $t\times t$ matrix and $\mathbf D$ a $t\times s$ matrix, where $t=|\mc T|$ and $s=|\mc S|$, with zeros wherever $G$ forbids the corresponding edges incident to $\mc T$. Consider the following conditional distribution of the target $\bs{X}_{\mathcal T}$ given the background $\bs x_{\mc S}$: 
\begin{equation}\label{eq:learned-graph-reference}
R_{\mathbf C,\mathbf D}(\cdot\mid\bs x_{\mc S})
=\mathcal{N}_t\!\left(\bs\mu_{\mc T}
-\mathbf C^{-1}\mathbf D(\bs x_{\mc S}-\bs\mu_{\mc S}),\mathbf C^{-1}\right).
\end{equation}
We show that, for every such pair $(\mathbf C,\mathbf D)$, the product of the
conditional laws in \eqref{eq:learned-graph-reference} over the $n$ observations
belongs to $\widetilde{\mathcal P}_\Phi(\mathbf x_{\mc S})$ for almost every
$\mathbf x_{\mc S}$ under the true law, where
$\phi(\mathbf z;\mathbf x_{\mc S})=\psi_G([\mathbf z,\mathbf x_{\mc S}])$.
Hence, this reference can be used in
\eqref{eq:gaussian-misspecification-size-supp}.

\begin{proof}

We first realise \eqref{eq:learned-graph-reference} as the conditional law of a joint Gaussian reference in $\mathcal M_G$. 
Choose $a>\lambda_{\max}(\mathbf D^\top\mathbf C^{-1}\mathbf D)$. 
Let the joint Gaussian reference have mean $\bs\mu$ and 
the precision matrix 
$
\begin{pmatrix}
\mathbf C&\mathbf D\\
\mathbf D^\top&a\mathbf I_s
\end{pmatrix}
$, which 
is positive definite by the Schur complement and obeys every zero restriction of $G$, including those within $\mc S$.
By the standard Gaussian conditioning formula, the conditional
distribution of $\bs X_{\mc T}$ given $\bs X_{\mc S}=\bs x_{\mc S}$
is Gaussian with mean
$\bs\mu_{\mc T}-\mathbf C^{-1}\mathbf D
(\bs x_{\mc S}-\bs\mu_{\mc S})$
and covariance $\mathbf C^{-1}$.
This is precisely $R_{\mathbf C,\mathbf D}(\cdot\mid\bs x_{\mc S})$
in \eqref{eq:learned-graph-reference}.

For $n$ independent observations from this joint Gaussian reference
distribution, 
Proposition~\ref{prop:ggm-gcrt-valid} establishes the joint exchangeability of the observed target matrix and its randomised copies. Because the joint reference belongs to $\mathcal M_G$, the product of the
conditional laws in \eqref{eq:learned-graph-reference} over the $n$ independent
observations belongs to $\mathcal P_\Theta(\mathbf x_{\mc S})$ for almost every
$\mathbf x_{\mc S}$ under the reference law. By
\eqref{eq:Phi-comparison-class-at-fixed-XS}, it also belongs to
$\widetilde{\mathcal P}_\Phi(\mathbf x_{\mc S})$.

The completed reference and true Gaussian laws of
$\mathbf X_{\mc S}$ have strictly positive densities, so every exceptional
background set under the reference law is also null under the true law.
This establishes the claimed membership and allows us to apply
\eqref{eq:gaussian-misspecification-size-supp}, whose comparison keeps the true
law of $(\bs Y,\mathbf X_{\mc S})$ and changes only the target conditional law.
No coverage of the true edges within $\mc S$ by $G$ is required.

\end{proof}

\subsubsection[Proof of Proposition~\ref{prop:learned-graph-size-rewrite}]{Proof of \Cref{prop:learned-graph-size-rewrite}.}\label{app:proof-learned-graph-size}

\begin{proof}
Fix $\alpha\in(0,1)$. Define
$\mathbf B=\mathbf\Omega_{\mc T\mc S}$, where
$\mathbf\Omega=\mathbf\Sigma^{-1}$ is the true precision matrix.
Throughout this proof, we condition on the auxiliary data $\mathcal U$. The selected graph
$\widehat G=([p],\widehat{\mc E})$, the sampler settings, and the statistic
$T$ are then fixed, as is the matrix $\mathbf\Omega_{\mc T\mc S}^{\mathrm{miss}}$ with entries
$\Omega_{ij}^{\mathrm{miss}}=\Omega_{ij}\mathbf1\{(i,j)\notin\widehat{\mc E}\}$,
$i\in\mc T$, $j\in\mc S$.
Since $\mathcal U$ is independent of $(\bs Y,\mathbf X,W)$,
conditioning on $\mathcal U$ does not change their joint law.

Suppose that the event $\mathcal F$ in
\Cref{prop:learned-graph-size-rewrite} occurs. We use the reference law in
\eqref{eq:learned-graph-reference} with
$\mathbf C=\mathbf\Omega_{\mc T\mc T}$ and $\mathbf D=\mathbf B-\mathbf\Omega_{\mc T\mc S}^{\mathrm{miss}}$.
The matrix $\mathbf C$ is positive definite, since it is a principal
submatrix of $\mathbf\Omega\succ0$. On $\mathcal F$, all true edges within
$\mc T$ are included in $\widehat G$, so $\mathbf C$ has zero off-diagonal
entries wherever $\widehat G$ has no edge. For $i\in\mc T$ and $j\in\mc S$,
the definition of $\mathbf\Omega_{\mc T\mc S}^{\mathrm{miss}}$ gives
$D_{ij}=\Omega_{ij}-\Omega_{ij}^{\mathrm{miss}}
=\Omega_{ij}\mathbf1\{(i,j)\in\widehat{\mc E}\}$.
Thus $\mathbf D$ also has zeros wherever the corresponding edge is absent
from $\widehat G$. Appendix~\ref{app:proof-gaussian-reference}, applied with
$G=\widehat G$, therefore shows that the product of these reference laws
over the $n$ observations belongs to the comparison class required by
\eqref{eq:gaussian-misspecification-size-supp}, for almost every background
matrix under the true law.

We next evaluate the scalar discrepancy $B$ in
\eqref{eq:gaussian-covariate-discrepancy-supp} for this reference. For one inference observation, define
$\bs Z_{\mc S}=\bs X_{\mc S}-\bs\mu_{\mc S}$.
The Gaussian conditioning formula and
\eqref{eq:learned-graph-reference} give the common conditional covariance
$\mathbf\Omega_{\mc T\mc T}^{-1}$ and the respective true and reference conditional means
$\bs\mu_{\mc T}-\mathbf\Omega_{\mc T\mc T}^{-1}\mathbf B\bs Z_{\mc S}$ and
$\bs\mu_{\mc T}-\mathbf\Omega_{\mc T\mc T}^{-1}(\mathbf B-\mathbf\Omega_{\mc T\mc S}^{\mathrm{miss}})\bs Z_{\mc S}$.
Subtracting the reference mean from the true mean gives
$\bs b(\bs X_{\mc S})=\bs\mu(\bs X_{\mc S})-\bs\mu_0(\bs X_{\mc S})=-\mathbf\Omega_{\mc T\mc T}^{-1}\mathbf\Omega_{\mc T\mc S}^{\mathrm{miss}}\bs Z_{\mc S}$.
The common covariance makes the covariance contribution to
\eqref{eq:gaussian-covariate-discrepancy-supp} vanish. Moreover,
$\mathbf\Omega_{\mc T\mc S}^{\mathrm{miss}}$ is fixed given $\mathcal U$, and independence gives
$\mathbb E_*[\bs Z_{\mc S}\bs Z_{\mc S}^\top\mid\mathcal U]
=\mathbf\Sigma_{\mc S\mc S}$.
Substituting the mean difference and then evaluating the expectation of the quadratic form gives
$$
\begin{aligned}
B
&=\mathbb E_*\!\left[
\bs b(\bs X_{\mc S})^\top\mathbf\Omega_{\mc T\mc T}\bs b(\bs X_{\mc S})
\mid\mathcal U\right]\\
&=\mathbb E_*\!\left[
\bs Z_{\mc S}^\top(\mathbf\Omega_{\mc T\mc S}^{\mathrm{miss}})^\top\mathbf\Omega_{\mc T\mc T}^{-1}\mathbf\Omega_{\mc T\mc S}^{\mathrm{miss}}\bs Z_{\mc S}
\mid\mathcal U\right]\\
&=\mathbb E_*\!\left[
\operatorname{tr}\!\left(\mathbf\Omega_{\mc T\mc T}^{-1}\mathbf\Omega_{\mc T\mc S}^{\mathrm{miss}}
\bs Z_{\mc S}\bs Z_{\mc S}^\top(\mathbf\Omega_{\mc T\mc S}^{\mathrm{miss}})^\top\right)
\mid\mathcal U\right]\\
&=\operatorname{tr}\!\left(\mathbf\Omega_{\mc T\mc T}^{-1}\mathbf\Omega_{\mc T\mc S}^{\mathrm{miss}}
\mathbb E_*\!\left[\bs Z_{\mc S}\bs Z_{\mc S}^\top\mid\mathcal U\right]
(\mathbf\Omega_{\mc T\mc S}^{\mathrm{miss}})^\top\right)\\
&=\operatorname{tr}(\mathbf\Omega_{\mc T\mc T}^{-1}\mathbf\Omega_{\mc T\mc S}^{\mathrm{miss}}
\mathbf\Sigma_{\mc S\mc S}(\mathbf\Omega_{\mc T\mc S}^{\mathrm{miss}})^\top)\leq f,
\end{aligned}
$$
where the inequality is the assumption on $\mathcal F$ in
\Cref{prop:learned-graph-size-rewrite}.

Let
$\rho(\mathcal U)=\mathbb P_*(\mathrm{pVal}_T\leq\alpha\mid\mathcal U)$
and $b=\min\{1,\alpha+\sqrt{nf}/2\}$.
On $\mathcal F$, the product of the chosen reference laws belongs to the
comparison class by Appendix~\ref{app:proof-gaussian-reference}.
Applying \eqref{eq:gaussian-misspecification-size-supp} conditional on
$\mathcal U$ and using $B\leq f$ gives $\rho(\mathcal U)\leq b$ on $\mathcal F$.
On $\mathcal F^c$, we use $\rho(\mathcal U)\leq1$.
To obtain the unconditional rejection probability, the tower property
and these two bounds give
$$
\begin{aligned}
\mathbb P_*(\mathrm{pVal}_T\leq\alpha)
&=\mathbb E_*\rho(\mathcal U)\\
&\leq b\,\mathbb P_*(\mathcal F)+\mathbb P_*(\mathcal F^c)\\
&=b+(1-b)\mathbb P_*(\mathcal F^c)\\
&\leq(1-\delta)b+\delta,
\end{aligned}
$$
where the last inequality uses
$\mathbb P_*(\mathcal F^c)\leq\delta$ and $1-b\geq0$.
Since $b\leq1$, $b\leq\alpha+\sqrt{nf}/2$, and $1-\delta\geq0$,
$$
(1-\delta)b+\delta
\leq\min\left\{1,\alpha+(1-\alpha)\delta
+\frac{1-\delta}{2}\sqrt{nf}\right\},
$$
which proves \eqref{eq:learned-graph-size-rewrite}.
If every true edge incident to $\mc T$ is included in $\widehat G$ on
$\mathcal F$, then $\mathbf\Omega_{\mc T\mc S}^{\mathrm{miss}}=\mathbf0$ there, so we may take $f=0$ and obtain
$\mathbb P_*(\mathrm{pVal}_T\leq\alpha)\leq\alpha+(1-\alpha)\delta$.
\end{proof}

\subsubsection[Proof of Corollary~\ref{cor:learned-graph-threshold}]{Proof of Corollary~\ref{cor:learned-graph-threshold}.}\label{app:proof-learned-graph-threshold}
\begin{proof}
On the event in \eqref{eq:learned-graph-estimation-event}, for every
$i\ne j$, $\Omega_{ij}=0$ implies $|\widehat\Omega_{ij}|\leq r_u$,
whereas $|\Omega_{ij}|>2r_u$ implies
$|\widehat\Omega_{ij}|\geq|\Omega_{ij}|-r_u>r_u$.
Thus the strict threshold includes every true edge with
$|\Omega_{ij}|>2r_u$, and an omitted true edge $(i,j)$ must satisfy
$0<|\Omega_{ij}|\leq2r_u$. All true edges within $\mc T$ are therefore covered.

Each edge between $\mc T$ and $\mc S$ contributes one entry to
the rectangular matrix $\mathbf\Omega_{\mc T\mc S}^{\mathrm{miss}}$. On the estimation event,
$$
\|\mathbf\Omega_{\mc T\mc S}^{\mathrm{miss}}\|_F^2
=\sum_{i\in\mc T}\sum_{j\in\mc S}
  \Omega_{ij}^2\mathbf1\{(i,j)\notin\widehat{\mc E}\}
\leq k_u(2r_u)^2.
$$
The spectral assumptions then give
$$
\begin{aligned}
\operatorname{tr}(\mathbf\Omega_{\mc T\mc T}^{-1}\mathbf\Omega_{\mc T\mc S}^{\mathrm{miss}}
                 \mathbf\Sigma_{\mc S\mc S}(\mathbf\Omega_{\mc T\mc S}^{\mathrm{miss}})^\top)
&=\|\mathbf\Omega_{\mc T\mc T}^{-1/2}\mathbf\Omega_{\mc T\mc S}^{\mathrm{miss}}\mathbf\Sigma_{\mc S\mc S}^{1/2}\|_F^2\\
&\leq\|\mathbf\Omega_{\mc T\mc T}^{-1/2}\|_{\mathrm{op}}^2
       \|\mathbf\Omega_{\mc T\mc S}^{\mathrm{miss}}\|_F^2
       \|\mathbf\Sigma_{\mc S\mc S}^{1/2}\|_{\mathrm{op}}^2
\leq\frac{4\Lambda_Sk_ur_u^2}{\kappa_T}.
\end{aligned}
$$
The event in \eqref{eq:learned-graph-estimation-event}
is determined by $\mathcal U$ and has probability at least $1-\delta$.
On this event, all true edges within $\mc T$ are included in
$\widehat G$, and the trace in
Proposition~\ref{prop:learned-graph-size-rewrite} is at most
$f=4\Lambda_S k_u r_u^2/\kappa_T$.
That proposition therefore gives
\eqref{eq:learned-graph-threshold-size}, since $r_u\geq0$ gives
$$
\frac{1-\delta}{2}\sqrt{nf}
=(1-\delta)r_u\sqrt{\frac{n\Lambda_S k_u}{\kappa_T}}.
$$

If $|\Omega_{ij}|>2r_u$ for every
true edge $(i,j)$ incident to $\mc T$, all true edges incident to
$\mc T$ are included in $\widehat G$. The bound then reduces to $\alpha+(1-\alpha)\delta$.
\end{proof}

\subsubsection[Proof of Proposition~\ref{prop:clime-estimation-guarantee}]{Proof of Proposition~\ref{prop:clime-estimation-guarantee}.}\label{app:proof-learned-graph-clime}
\begin{proof}
We prove the estimation guarantee \eqref{eq:learned-graph-estimation-event}
with the radius $r_u$ specified in
\eqref{eq:learned-graph-clime-radius}.
Although \citet[Theorem~4(a)]{cai_constrained_2011} already provide
a probability bound for CLIME, using it directly would yield a
radius with constants for their general exponential moment conditions.
We therefore exploit the Gaussian assumption to obtain simpler
constants, deriving the covariance error bound $a_u$ in
\eqref{eq:learned-graph-clime-radius} for the unbiased sample covariance
from the standard chi-squared inequality in
Lemma~\ref{chi square concentration}.
Combining this covariance bound with the deterministic CLIME
inequality of \citet[Theorem~6, equation~(13)]{cai_constrained_2011}
then gives $r_u=4K^2a_u$.

For a Gaussian variable with variance $\tau^2>0$, its unbiased
sample variance $\widehat\tau^2$ satisfies
$(n_u-1)\widehat\tau^2/\tau^2\sim\chi^2_{n_u-1}$.
The two tails of Lemma~\ref{chi square concentration}, with
$D=n_u-1$, $B=0$, and $x=\ell_u$, give
\begin{equation}\label{eq:clime-sample-variance-concentration-proposed}
\mathbb P_*\left(
|\widehat\tau^2-\tau^2|
>2\tau^2\left\{
\sqrt{\frac{\ell_u}{n_u-1}}+\frac{\ell_u}{n_u-1}
\right\}\right)
\leq2e^{-\ell_u}.
\end{equation}
Apply \eqref{eq:clime-sample-variance-concentration-proposed} to
$X_i$ and $X_i+X_j,X_i-X_j$ for $i<j$.
A union bound over these $p^2$ variables shows that all their sample
variance errors are at most $a_u/v$ times the corresponding population
variances, with probability at least
$1-2p^2e^{-\ell_u}=1-\delta/2\geq1-\delta$. 
For each $i<j$, let $\tau_\pm^2=\operatorname{Var}_*(X_i\pm X_j)$,
and let $\widehat\tau_\pm^2$ be the corresponding unbiased sample variances.
Expanding the population variances and the centred sample sums of squares gives
\begin{equation}\label{eq:clime-polarisation-identities}
\Sigma_{ij}=\frac{\tau_+^2-\tau_-^2}{4},
\qquad
\widehat\Sigma_{ij}=\frac{\widehat\tau_+^2-\widehat\tau_-^2}{4}.
\end{equation}
On the simultaneous event from the union bound,
$|\widehat\tau_\pm^2-\tau_\pm^2|\leq(a_u/v)\tau_\pm^2$ for both signs.
Consequently,
$$
\begin{aligned}
|\widehat\Sigma_{ij}-\Sigma_{ij}|
&=\frac14\left|
(\widehat\tau_+^2-\tau_+^2)-(\widehat\tau_-^2-\tau_-^2)
\right|\\
&\leq\frac14\left(
|\widehat\tau_+^2-\tau_+^2|+|\widehat\tau_-^2-\tau_-^2|
\right)\\
&\leq\frac{a_u}{4v}(\tau_+^2+\tau_-^2)\\
&=\frac{a_u}{2v}(\Sigma_{ii}+\Sigma_{jj})\leq a_u,
\end{aligned}
$$
where the first equality uses \eqref{eq:clime-polarisation-identities},
the first inequality is the triangle inequality, 
the second inequality uses
the simultaneous sample variance bounds,
the next equality follows from
$\tau_+^2+\tau_-^2=2(\Sigma_{ii}+\Sigma_{jj})$,
and the final inequality uses $\Sigma_{ii},\Sigma_{jj}\leq v$.
For the diagonal entries, the same event gives
$|\widehat\Sigma_{ii}-\Sigma_{ii}|\leq(a_u/v)\Sigma_{ii}\leq a_u$.
Thus $\|\widehat{\mathbf\Sigma}-\mathbf\Sigma\|_{\max}\leq a_u$
with probability at least $1-\delta$.

Whenever $\|\widehat{\mathbf\Sigma}-\mathbf\Sigma\|_{\max}\leq a_u$,
our choice of $\lambda_u$ satisfies
$$
\lambda_u=Ka_u\geq\|\mathbf\Omega\|_1
\|\widehat{\mathbf\Sigma}-\mathbf\Sigma\|_{\max}.
$$
All column problems are then feasible, so the fallback is not used.
For the symmetrised CLIME estimator defined in
Appendix~\ref{app:clime-graph-selection}, we apply
\citet[Theorem~6, equation~(13)]{cai_constrained_2011}
with $\rho=0$ in their notation.
Their matrix class can be chosen to allow all $p$ entries per row,
and their proof of (13) is a deterministic matrix argument that also
applies to the unbiased sample covariance used here.
Consequently,
$$
\|\widehat{\mathbf\Omega}-\mathbf\Omega\|_{\max}
\leq4\|\mathbf\Omega\|_1\lambda_u
\leq4K^2a_u=r_u.
$$
This proves \eqref{eq:learned-graph-estimation-event}.
\end{proof}

\section{Computational settings and simulation studies}
\label{app:simulation-details}
\AppendixTableFigureLayout
This appendix gives computational settings and simulation results for Section~\ref{sec: Simulation}. It presents the four response models, experiments with a single target covariate and with supergraphs, comparisons of joint and coordinatewise tests, and simulations with non-Gaussian covariates.

\subsection{Statistics and computational settings}\label{app:experimental-pvalue-conventions}

Table~\ref{tab:mscrt-statistic-catalogue} lists several statistics used in our numerical studies, grouped by construction. LM denotes linear regression, GLM a generalised linear model, and RF a random forest. We write $\mathbf X_{\mc T}^{(v)}$, $v=0,\ldots,M$, for a candidate target matrix.

For distillation and distilled residuals, a background fit regresses $\bs Y$ on $\mathbf X_{\mc S}$. Let $\widehat{\bs Y}_0$ be its fitted values and $\bs R_Y=\bs Y-\widehat{\bs Y}_0$ its response residuals. Both constructions fit the background model once and reuse it across candidates, avoiding repeated fits involving all background covariates \citep{liu_fast_2022}.

\begin{table}[!b]
\centering
\small
\setstretch{1.05}
\setlength{\tabcolsep}{4pt}
\renewcommand{\arraystretch}{1.2}
\captionsetup{font=normalsize,labelfont=bf}
\caption{Names and constructions of test statistics. Here $\bs Y$ is the response vector, $\mathbf X_{\mc S}$ the background covariate matrix, and $\mathbf X_{\mc T}^{(v)}$ an observed or randomised target matrix. For each specified background model, let $\widehat{\bs Y}_0$ be the fitted values from regressing $\bs Y$ on $\mathbf X_{\mc S}$. In the distillation and distilled residual constructions, the scoring rule is taken from the full fit construction and applied to the specified final fit.}
\label{tab:mscrt-statistic-catalogue}

\begin{subtable}[t]{\linewidth}
\caption{Full fit}
\raggedright
A full fit uses response $\bs Y$ and predictors $[\mathbf X_{\mc S},\mathbf X_{\mc T}^{(v)}]$. The score summarises overall model fit or combines importance measures for the target covariates.
\par\smallskip
\begin{tabular}{@{}>{\raggedright\arraybackslash}p{0.23\linewidth}>{\raggedright\arraybackslash}p{0.23\linewidth}>{\raggedright\arraybackslash}p{0.46\linewidth}@{}}
\hline
\textbf{Name} & \textbf{Fit} & \textbf{Score} \\
\hline
\textbf{LM-SST} & Linear regression & Sum of squared t-statistics for the target coefficients. \\
\textbf{LM-SSR} & Linear regression & Negative residual sum of squares. \\
\textbf{GLM-Dev} & GLM & Negative residual deviance. \\
\textbf{RF} & Random forest & Sum of the positive parts of the variable importance scores for the target covariates. \\
\hline
\end{tabular}
\end{subtable}

\par\medskip
\begin{subtable}[t]{\linewidth}
\caption{Distillation}
\raggedright
Distillation uses the fitted values from the background model to summarise the background covariates. The final fit uses response $\bs Y$ and predictors $[\widehat{\bs Y}_0,\mathbf X_{\mc T}^{(v)}]$.
\par\smallskip
\begin{tabular}{@{}>{\raggedright\arraybackslash}p{0.23\linewidth}>{\raggedright\arraybackslash}p{0.25\linewidth}>{\raggedright\arraybackslash}p{0.20\linewidth}>{\raggedright\arraybackslash}p{0.24\linewidth}@{}}
\hline
\textbf{Name} & \textbf{Background fit} & \textbf{Final fit} & \textbf{Scoring rule} \\
\hline
\textbf{GLM-L1-D} & $\ell_1$-penalised GLM & GLM & GLM-Dev \\
\textbf{RF-D} & Random forest & Random forest & RF \\
\hline
\end{tabular}
\end{subtable}

\par\medskip
\begin{subtable}[t]{\linewidth}
\caption{Distilled residuals, with $\bs R_Y=\bs Y-\widehat{\bs Y}_0$.}
\raggedright
Distilled residuals remove the fitted background contribution from the response. The final fit uses response $\bs R_Y$ and predictors $\mathbf X_{\mc T}^{(v)}$.
\par\smallskip
\begin{tabular}{@{}>{\raggedright\arraybackslash}p{0.23\linewidth}>{\raggedright\arraybackslash}p{0.25\linewidth}>{\raggedright\arraybackslash}p{0.20\linewidth}>{\raggedright\arraybackslash}p{0.24\linewidth}@{}}
\hline
\textbf{Name} & \textbf{Background fit} & \textbf{Final fit} & \textbf{Scoring rule} \\
\hline
\textbf{LM-L1-R-SST} & $\ell_1$-penalised linear regression & Linear regression & LM-SST \\
\textbf{LM-L1-R-SSR} & $\ell_1$-penalised linear regression & Linear regression & LM-SSR \\
\textbf{GLM-L1-R-SST} & $\ell_1$-penalised GLM & Linear regression & LM-SST \\
\textbf{RF-RR} & Random forest regression & Random forest regression & RF \\
\hline
\end{tabular}
\end{subtable}
\end{table}

For RF, RF-D and RF-RR, classification uses the overall mean decrease in accuracy, and regression uses the increase in mean squared prediction error. The importance values are unscaled. Each forest uses 100 trees and \texttt{importance=TRUE}; all other settings are the defaults in \texttt{randomForest} version 4.7.1.1. 
RF and RF-D use classification for categorical responses and regression for continuous responses. Both stages of RF-RR use regression, including for binary responses. For RF-RR, the background residual for each observation uses predictions from trees whose training samples exclude that observation.

Lasso background fits use five-fold cross-validation with \texttt{lambda.1se}. Fitting and tuning involving target covariates follow the scoring conditions in Section~\ref{sec: MSCRT} and Appendix~\ref{app:scoring-conditions}.

In Section~\ref{sec:simulation-g-crt}, we compute the $G$-CRT p-values by the discrete randomised rule in \eqref{eq:crt-pvalue-discrete}. For the coordinatewise dCRT baselines, we compute analytic p-values as in \citet{liu_fast_2022} and apply Bonferroni correction across $\mc T$.

\subsection{Simulations for four response models}\label{app:four-response-models}

We evaluate the power of the $G$-CRT for GGMs using selected statistics from Table~\ref{tab:mscrt-statistic-catalogue} to test multivariate conditional independence. 
We consider a range of regression models for $\mc{L}(Y\mid \bs{X})$, including linear regression, nonlinear regression, logistic regression, and nonlinear binary regression. 
For each regression model, we consider a low-dimensional setting ($p=20$ and $n=50$) and a high-dimensional setting ($p=120$ and $n=80$), and we compare our tests with several existing testing methods.

In each of these experiments, we start with a banded precision matrix $\mathbf{\Omega}_0=(\omega^{(0)}_{i,j})_{i,j=1}^p$ defined by
$$ \omega^{(0)}_{i,j}=\left\{
\begin{array}{rcl}
1       &      & \text{if ~~} i=j, \\
s     &      & \text{if ~~} 1\leq |i-j|\leq K,\\
0       &      & \text{if ~~} |i-j|> K,
\end{array} \right. $$
where $K=6$ is the bandwidth and all nonzero off-diagonal entries of $\mathbf{\Omega}_0$ equal $s=0.2$.

We standardise $\mathbf{\Omega}_0^{-1}$ to a correlation matrix. In each replication, we apply a random permutation simultaneously to its rows and columns. Let $\mathbf{\Omega}$ be the inverse of the resulting covariance matrix, and define $G$ by the nonzero off-diagonal entries of $\mathbf{\Omega}$.
The covariate vectors are i.i.d. samples drawn from $\mathcal{N}_p( \bs{0}, \mathbf{\Omega}^{-1})$.
Depending on the simulation setting, we generate the response variable $Y$ under a linear or nonlinear model. A constant factor $\theta$ controls the magnitude of the dependence of $Y$ on $\bs{X}$.

We consider testing the null hypothesis $H_0: Y\indp \bs{X}_{\mc{T}}\mid  \bs{X}_{-\mc{T}}$ for $\mc{T}=\{1,2,\ldots, 8\}$. 
Because the nodes of $G$ have been permuted, choosing $\mc{T}$ to consist of the first 8 indices does not tie $\mc{T}$ to the structure of the banded graph. Our $G$-CRT procedure uses the graph $G$ and knows neither $\mathbf{\Omega}$ nor the mean parameters.
For the four response models below, we fix $M=100$ and $L=3$ in Algorithm~\ref{alg: exchangeable}.
Since the banded graph has maximum degree \(12\),  \(n>|N_i|+1\) for every updated node in both dimension settings. 
For testing procedures that compute p-values for the individual hypotheses $H_{0,i}:Y\indp X_{i}\mid \bs{X}_{-i}$, we apply Bonferroni correction across $\mc T$. In these four experiments, such procedures should be understood as coordinatewise baselines implemented through multiplicity adjustment. 
The experiments use $400$ replications for linear, nonlinear and logistic regression, and $800$ replications for nonlinear binary regression. The curves show rejection rates at level $\alpha=0.05$; error bars show one Monte Carlo standard error. The points at $\theta=0$ show the estimated size of each test. 

\subsubsection{Gaussian linear regression}\label{sec: simulation CRT linear regression}
In the first experiment, the response variable $Y$ given the covariates $\bs{X}\in \mathbb{R}^{p}$ is generated under a Gaussian linear model such that $Y=\bs{X}\tp \bs{\beta} + \varepsilon$ with some independent noise $\varepsilon\sim \mathcal{N}(0,1)$. 
The regression coefficient $\bs{\beta}$ is determined randomly at the beginning of each replication, whose first 20 coefficients are independently drawn from a uniform distribution on the interval $(1/\sqrt{20},2/\sqrt{20})$ and all remaining coefficients are set to zero. 
The first 8 coordinates of $\bs{\beta}$ corresponding to $\bs{X}_{\mc{T}}$ (recall that $\mc{T}=\{1,2,\ldots, 8\}$) are then scaled by a factor $\theta$, which controls the signal strength for the conditional independence testing problem. 

For a Gaussian linear model, the testing problem is equivalent to testing
\begin{equation}\label{eq: linear model hypothesis}
H_0: \bs{\beta}_{\mc{T}}=\bs{0} \quad \text{versus} \quad H_1: \bs{\beta}_{\mc{T}} \neq \bs{0}. 
\end{equation}
If $p<n$ and $|\mc{T}|>1$, a uniformly most powerful unbiased test does not exist, but the $F$-test is uniformly most powerful invariant (UMPI) under a certain group of transformations; see \citet[Section~6.3.2]{shaoMathematicalStatistics2003}. 
We therefore use the $F$-test as a baseline in the low-dimensional experiments.

In a high-dimensional setting, \citet{verzelen_goodness--fit_2010}  consider the same hypothesis as in \eqref{eq: linear model hypothesis} with Gaussian random designs, but their procedure is feasible only if the size of $\mc{T}^{c}$ (i.e., $p-|\mc{T}|$) is smaller than $n$, which is not satisfied in our high-dimensional simulations. 
We use the coordinatewise de-sparsified Lasso \citep{van_de_geer_asymptotically_2014} as a high-dimensional baseline.
Its asymptotic validity depends on sparsity and design conditions; see Appendix~\ref{app:regression-ci-comparison}.
In contrast, $G$-CRT controls the Type I error in finite samples without a sparsity assumption on $\bs{\beta}$ (Proposition~\ref{prop: CRT valid}).

\begin{figure}[!b]
    \centering
    \begin{subfigure}[b]{0.48\textwidth}
\includegraphics[width=1\textwidth]{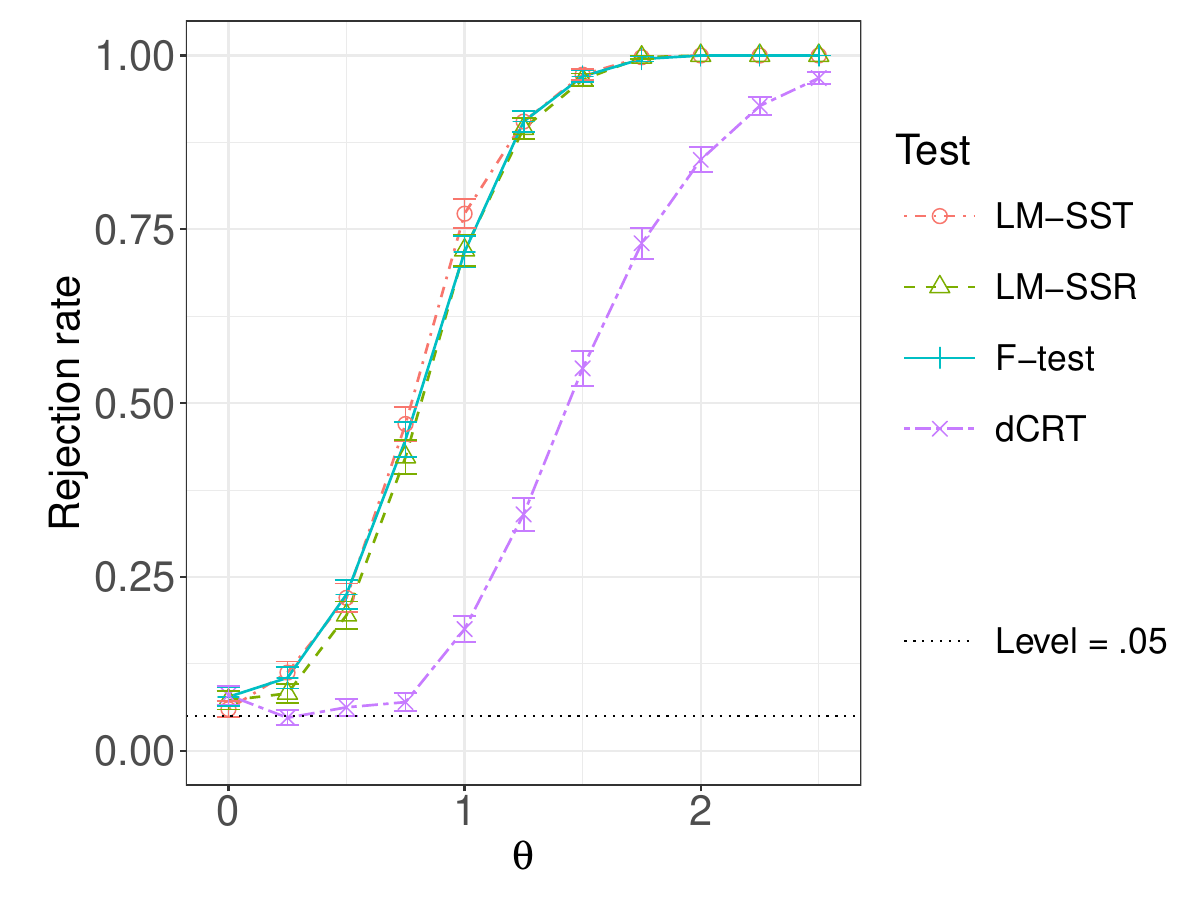}
\caption{Low-dimensional}
\label{fig: app llw}\label{fig:llw}
\end{subfigure}
    \begin{subfigure}[b]{0.48\textwidth}
\includegraphics[width=1\textwidth]{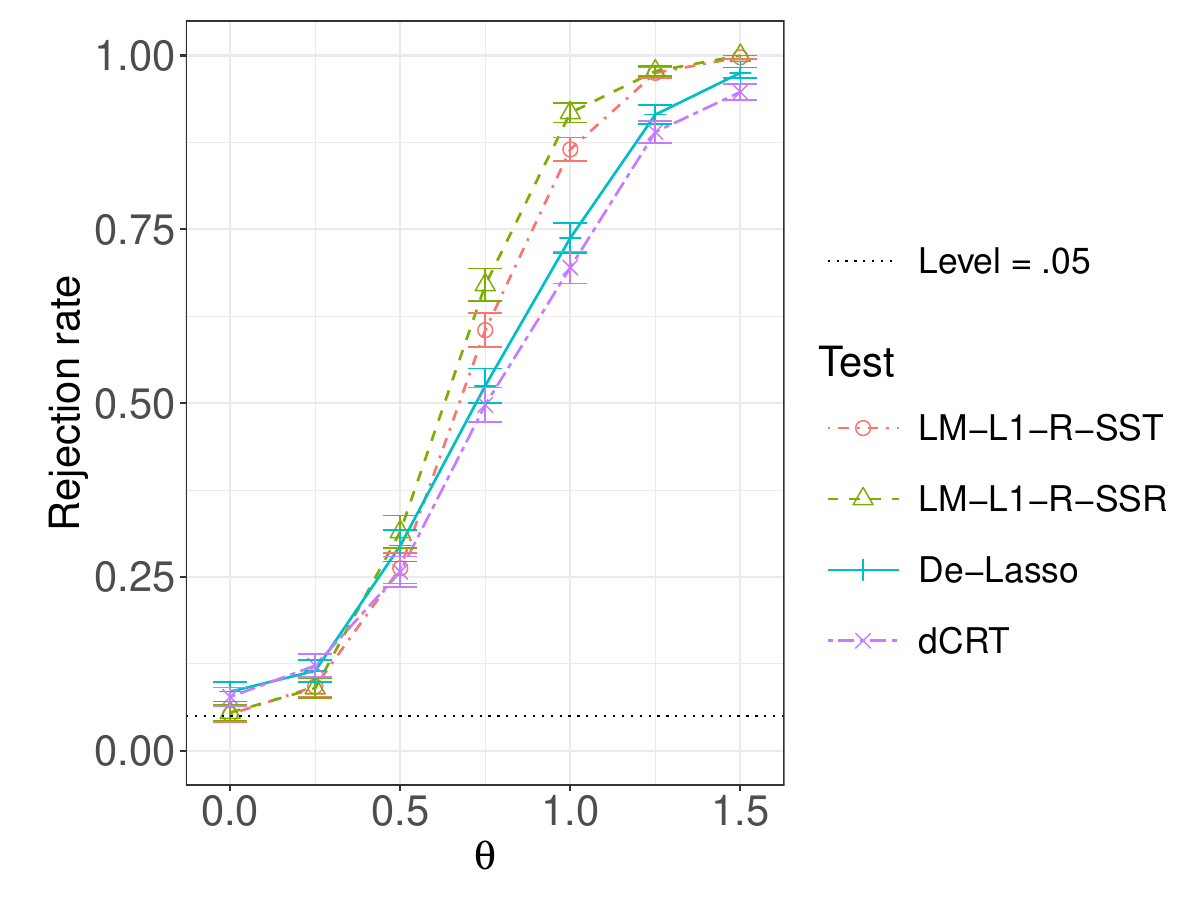}
\caption{High-dimensional}
\label{fig: app hlw}\label{fig:hlw}
    \end{subfigure}
    \caption{Rejection rates at nominal level $0.05$ for testing conditional independence in linear regression. Each point uses $400$ replications; error bars show one Monte Carlo standard error. The points at $\theta=0$ estimate size.}\label{fig: app linear regression}\label{fig: linear regression}
\end{figure}

Figure~\ref{fig: app llw} presents the power curves for several testing methods in the low-dimensional experiment with $p=20$ and $n=50$, where power is plotted against varying values of $\theta$. 
We compare the $G$-CRT with the LM-SST and LM-SSR statistics against the classical $F$-test for \eqref{eq: linear model hypothesis} and a coordinatewise dCRT baseline, implemented via Bonferroni-adjusted dCRT using Gaussian Lasso models.
The dCRT is implemented using the first author's source code with a properly chosen model and other default arguments. 
At $\theta=0$ in Figure~\ref{fig: app llw}, the estimated sizes are $0.0600$, $0.0725$, $0.0775$ and $0.0800$ for LM-SST, LM-SSR, the $F$-test and dCRT, respectively. Their Monte Carlo standard errors are at most $0.0136$. 
The dCRT also suffers from a loss of power relative to the other methods. 
In addition, the $G$-CRTs with LM-SSR and LM-SST perform comparably and achieve high power on par with the $F$-test. 
The empirical power of both $G$-CRTs is close to 1 at the larger values of $\theta$ shown in Figure~\ref{fig: app llw}.

Figure~\ref{fig: app hlw} presents the power curves in the high-dimensional experiment with $p=120$ and $n=80$. We consider the test statistics LM-L1-R-SST and LM-L1-R-SSR for the $G$-CRT. In addition to the dCRT, we also consider the Bonferroni-adjusted coordinatewise de-sparsified Lasso \citep{van_de_geer_asymptotically_2014}, labelled as De-Lasso, which can be implemented using the function \textbf{lasso.proj()} in the \texttt{R} package \textbf{hdi}. 
Figure~\ref{fig: app hlw} shows estimated sizes of $0.0525$ and $0.0550$ for the $G$-CRT with LM-L1-R-SST and LM-L1-R-SSR, respectively, compared with $0.0850$ for De-Lasso and $0.0775$ for dCRT. The rejection rates of all methods increase as $\theta$ grows. The four curves are close at weak signal strengths, while the $G$-CRT has substantially higher rejection rates than the coordinatewise baselines at moderate signal strengths. The rejection rates of both $G$-CRT tests are close to $1$ at $\theta=1.25$.

\subsubsection{Nonlinear regression}\label{sec: simulation CRT nonlinear regression}

In this subsection, we consider a nonlinear dependence between the response and the covariates. We expand $\bs{X}$ using nonlinear basis functions, including squares, reciprocals, trigonometric functions, and interactions, and generate $Y$ as a linear combination of these functions plus standard Gaussian noise. For $\bs{x}=(x_1, \ldots, x_p)$, define the expanded basis $\bs{B}(\bs{x})$ by 
\begin{equation}\label{eq: simulation nonlinear basis}
\bs{B}(\bs{x}) = \left[
\underbrace{\frac{x_{1}^2}{2}, \frac{1}{1+x_{2}^2}, \cos(\pi x_{3}), \sin(\pi x_{4}), x_{5}x_{6}, \sin(x_{7}x_{8}), \frac{\sin(\pi x_{9})}{4+x_{10}^2}}_{\text{I: 7 terms of nonlinear functions}}, \underbrace{ \frac{x_{11}^2}{2}, \ldots, \frac{\sin(\pi x_{19})}{4+x_{20}^2} }_{\text{II: 7 terms in the same form as I}}
\right]. 
\end{equation}
Let $\bs{\beta}$ be a coefficient vector of dimension 14 so that its first 7 entries (corresponding to the functions involving $x_1, \ldots, x_{10}$) are all equal to $\theta$ and its next 7 entries (corresponding to the functions involving $x_{11}, \ldots, x_{20}$)  are all equal to 1. 
Then the response is defined as 
\begin{equation}\label{eq: simulation nonlinear regression}
Y=\bs{B}(\bs{X}) \bs{\beta} + \varepsilon
\end{equation}
with some independent noise $\varepsilon\sim \mathcal{N}(0,1)$. 
Here $\theta$ scales the six terms involving the target covariates and the background term $\sin(\pi X_9)/(4+X_{10}^2)$, while the seven terms involving $X_{11},\ldots,X_{20}$ have coefficients fixed at one. At $\theta=0$, the response depends only on background covariates and independent noise, so the MCI null holds.
For this nonlinear regression model, the tests based on Gaussian linear models in Section~\ref{sec: simulation CRT linear regression} may have low power because the linear dependence between the response and the covariates can be very weak or absent. 
To capture the nonlinear dependence, we consider the $G$-CRT with the statistic RF and its two-step variants, RF-D and RF-RR. 
For comparison, we consider the same baseline tests as in Section~\ref{sec: simulation CRT linear regression} but use random forest regression for the response fits in the dCRT; the covariate fits remain Gaussian Lasso regression. 

Figures~\ref{fig: app llm} and \ref{fig: app hlm} represent the results in low-dimensional and high-dimensional settings, respectively.
In both settings, our $G$-CRT with either of the two-step statistics (RF-D and RF-RR) exhibits superior performance.
The advantages of these statistics against the RF statistic are significant in the high-dimensional setting, which demonstrates the potential of distillation and distilled residuals in boosting power.

In Figure~\ref{fig: app llm}, the $F$-test has power no greater than $0.11$ over the displayed alternatives.
Furthermore, the Bonferroni-adjusted coordinatewise baselines considered here, including the dCRT using random forests and the de-sparsified Lasso, fail to reach the same power as the $G$-CRTs.
These results illustrate the ability of the $G$-CRT with random forest statistics to detect the nonlinear group signals in \eqref{eq: simulation nonlinear regression}.

\begin{figure}[!tbp]
    \centering
    \begin{subfigure}[b]{0.48\textwidth}
\includegraphics[width=1\textwidth]{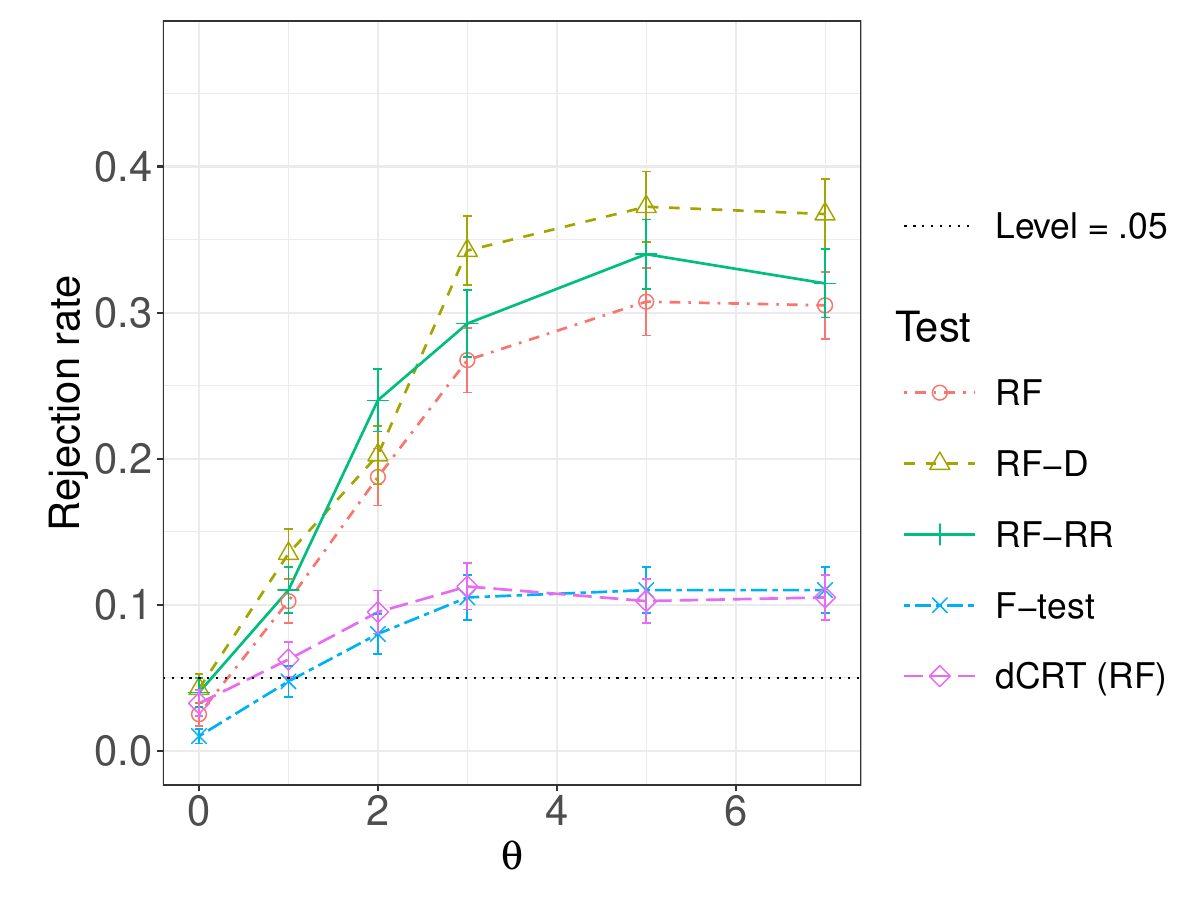}
\caption{Low-dimensional}
\label{fig: app llm}\label{fig:llm}
\end{subfigure}
    \begin{subfigure}[b]{0.48\textwidth}
\includegraphics[width=1\textwidth]{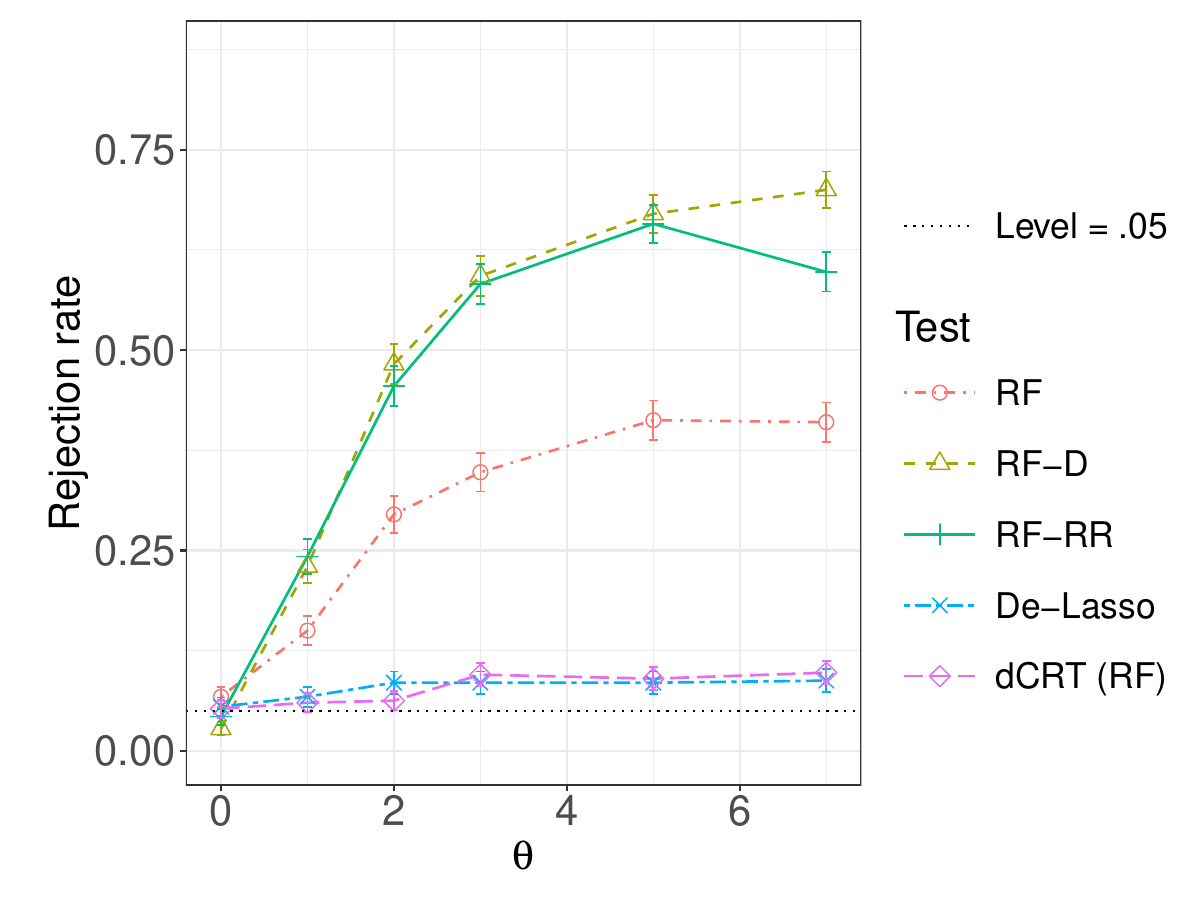}
\caption{High-dimensional}
\label{fig: app hlm}\label{fig:hlm}
    \end{subfigure}
    \caption{Rejection rates at nominal level $0.05$ for testing conditional independence in nonlinear regression. Each point uses $400$ replications; error bars show one Monte Carlo standard error. The points at $\theta=0$ estimate size.}\label{fig: app nonlinear regression}\label{fig: nonlinear regression}
\end{figure}

\subsubsection{Logistic regression}\label{sec: simulation CRT logistic regression}
We generate the random regression coefficients $\bs{\beta}$ as in Section~\ref{sec: simulation CRT linear regression}, but sample the response from the Bernoulli distribution with probability $\sigma( \bs{X}\tp \bs{\beta})$, where $\sigma(t)=e^{t}/(1+e^{t})$ is the sigmoid function (i.e. the inverse of the logit function). 

In the low-dimensional setting with $p=20$ and $n=50$, we evaluate the power of our $G$-CRT with the statistics GLM-Dev and RF. 
For comparison, we use  a coordinatewise dCRT baseline, implemented via Bonferroni-adjusted dCRT using logistic Lasso models. 
We also consider the classical Chi-squared test based on the difference in residual deviance between the nested models (see \citet[Section~4.4.3]{mccullagh1989generalized} and \citet{HastiePregibon1992} for more details).

In Figure~\ref{fig: app lglw}, the Chi-squared test has an estimated size of $0.5375$ ($215/400$). Its curve is shown for reference and does not provide a valid power comparison. 
The other three tests are near nominal size; among them, the \(G\)-CRT with the GLM-Dev statistic has the highest power, while the \(G\)-CRT with the RF statistic is less powerful in this logistic setting.

In the high-dimensional setting with $p=120$ and $n=80$, we evaluate the $G$-CRT using two additional statistics besides RF: GLM-L1-D and GLM-L1-R-SST. 
For comparison, we also consider a bias-correction method for high-dimensional logistic Lasso regression proposed by \citet{cai2023statistical}, which can be implemented using the \textbf{LF()} function in the \texttt{R} package \textbf{SIHR}. 
The Bonferroni-adjusted test based on the asymptotic normality of this bias-corrected estimator will be referred to as the CGM test. 

Figure~\ref{fig: app hglw} shows the power curves for the high-dimensional logistic regression experiment. The estimated sizes of the $G$-CRT tests range from $0.0400$ to $0.0475$, compared with $0.0025$ for the CGM test. The power of all methods increases with $\theta$.
It seems that the CGM test is too conservative and powerless for $\theta\in (0,1)$.
The dCRT has higher power than the $G$-CRT with the RF statistic in this experiment.
In Figure~\ref{fig: app hglw}, the $G$-CRT tests using GLM-L1-D and GLM-L1-R-SST have close power and outperform the other methods.
The $G$-CRT with GLM-L1-R-SST achieves this power even though the statistic's second step uses neither the original response nor the logistic regression likelihood.
This observation highlights the potential to boost power when using a flexible test statistic in a CRT method.

\begin{figure}[!tbp]
    \centering
    \begin{subfigure}[b]{0.48\textwidth}
\includegraphics[width=1\textwidth]{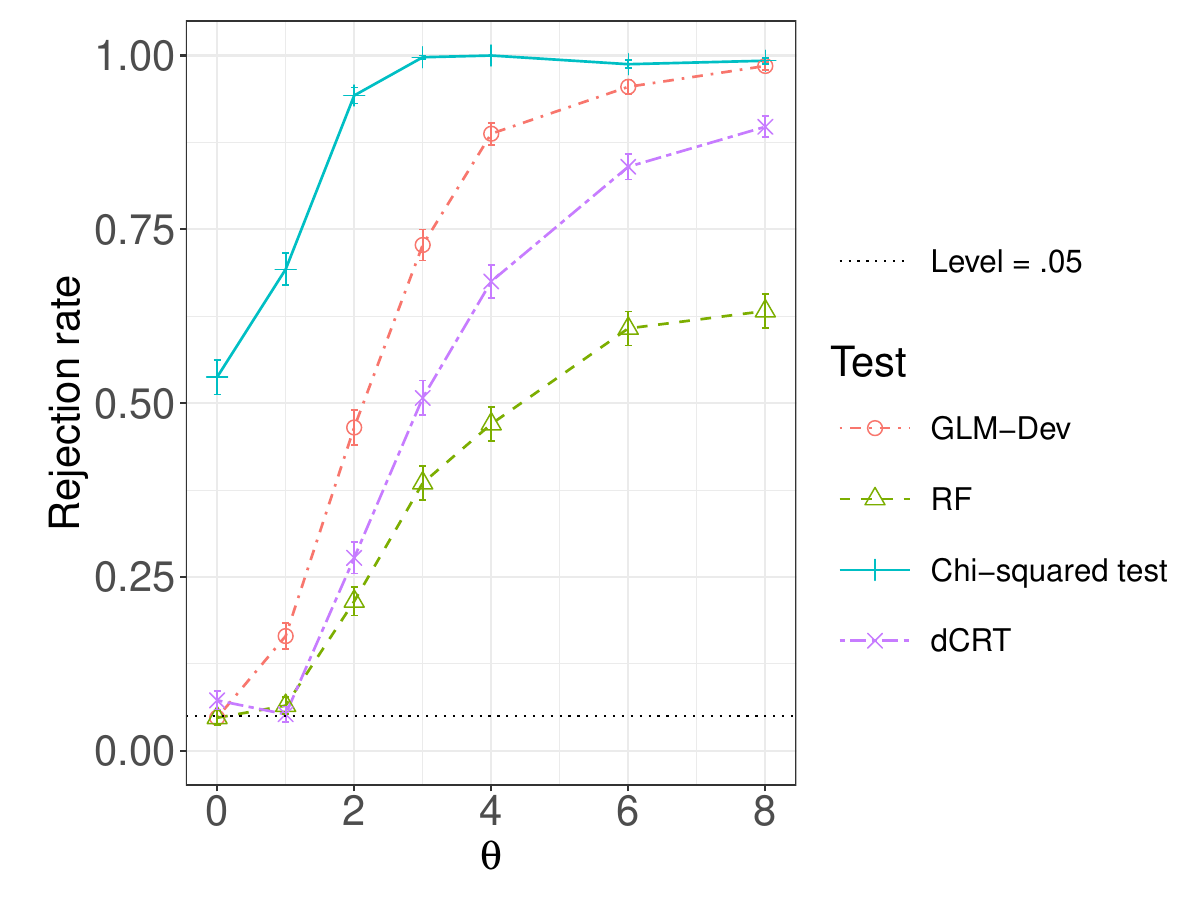}
\caption{Low-dimensional}
\label{fig: app lglw}\label{fig:lglw}
\end{subfigure}
    \begin{subfigure}[b]{0.48\textwidth}
\includegraphics[width=1\textwidth]{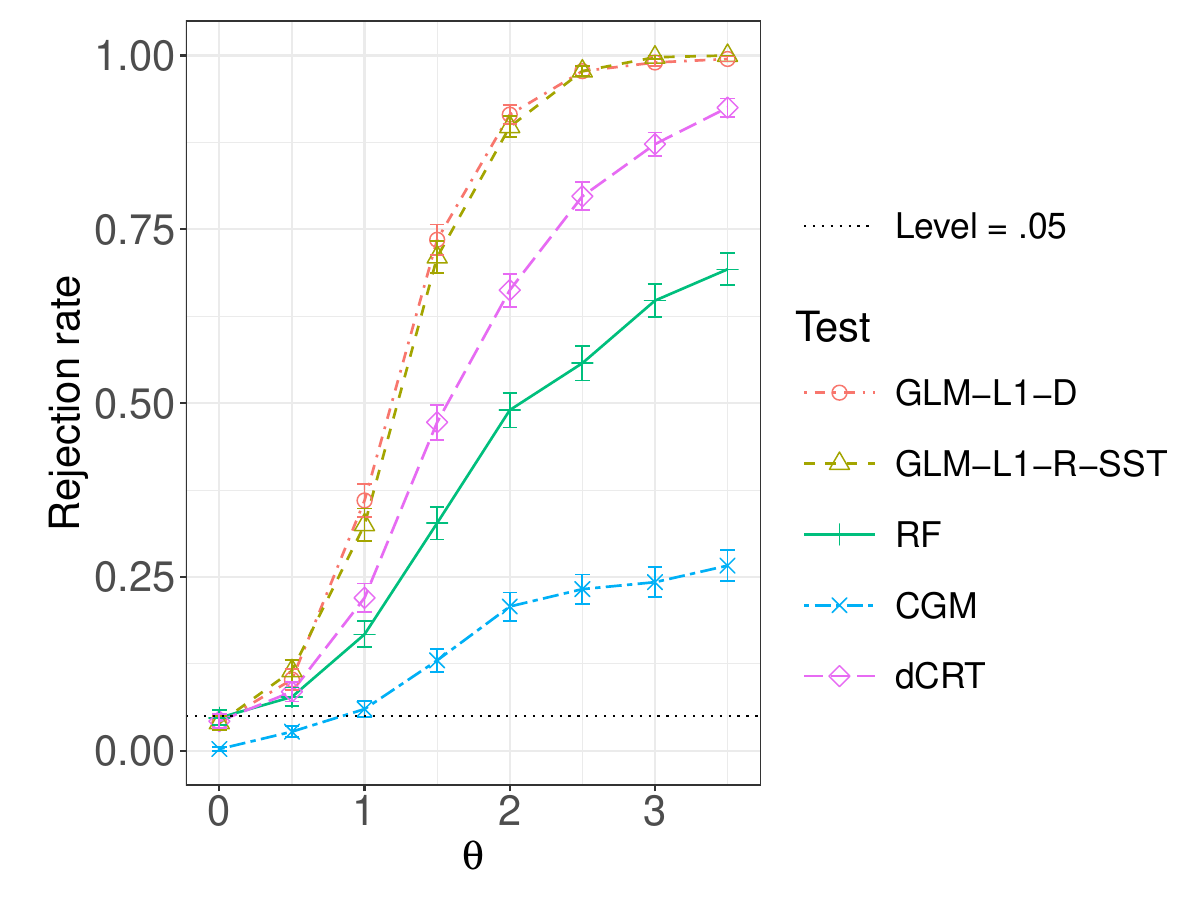}
\caption{High-dimensional}
\label{fig: app hglw}\label{fig:hglw}
    \end{subfigure}
    \caption{Rejection rates at nominal level $0.05$ for testing conditional independence in logistic regression. Each point uses $400$ replications, except CGM at $\theta=3.5$ in panel (b), where in two replications it does not yield an output. Error bars show one Monte Carlo standard error. The points at $\theta=0$ estimate size; the oversized Chi-squared test in panel (a) is shown only for reference.}\label{fig: app logistic regression}\label{fig: logistic regression}
\end{figure}

\subsubsection{Nonlinear binary regression}

In this subsection, we consider a binary response variable whose success probability depends on nonlinear transformations of the covariates. Specifically, we transform the first 20 coordinates of any covariate vector $\bs{X}$ componentwise through a step function $$b(x)= \left({1}_{\left\{x< l\right\} }-2\cdot  {1}_{\left\{ l \leq x \leq u\right\} } + {1}_{\left\{x > u\right\} } \right) / 4,$$ where $l$ and $u$ are the $1/3$ and $2/3$ quantiles of the standard normal distribution, and obtain a vector $\bs{Z}=b(\bs{X}_{[20]})$. 
We then sample $Y$ from the Bernoulli distribution with probability $\sigma( \bs{Z} \tp \bs{\beta})$, where the first 10 coordinates of $\bs{\beta}$ are equal to $\theta$ and the rest are equal to 1. 
Here $\theta$ scales the terms $b(X_j)$ for $j=1,\ldots,10$, including those for the background covariates $X_9$ and $X_{10}$, while the terms for $j=11,\ldots,20$ have coefficients fixed at one. At $\theta=0$, the response depends only on background covariates, so the MCI null holds.
In each experiment, we will evaluate our $G$-CRT procedure with the statistic RF and its variants, RF-D and RF-RR. 
Note that the statistic RF-RR is based on regression models that treat the response variable as continuous, but it is compatible with the CRT framework. 
For comparison, we consider the same baseline tests as in Section~\ref{sec: simulation CRT logistic regression} but use random forest regression for the response fits in the dCRT (the covariate fits remain Gaussian Lasso regression). 
To reduce the variability of the power estimates in these nonlinear binary regression experiments, we increase the number of replications for each experiment from 400 to 800.

\begin{figure}[!tbp]
    \centering
    \begin{subfigure}[b]{0.48\textwidth}
\includegraphics[width=1\textwidth]{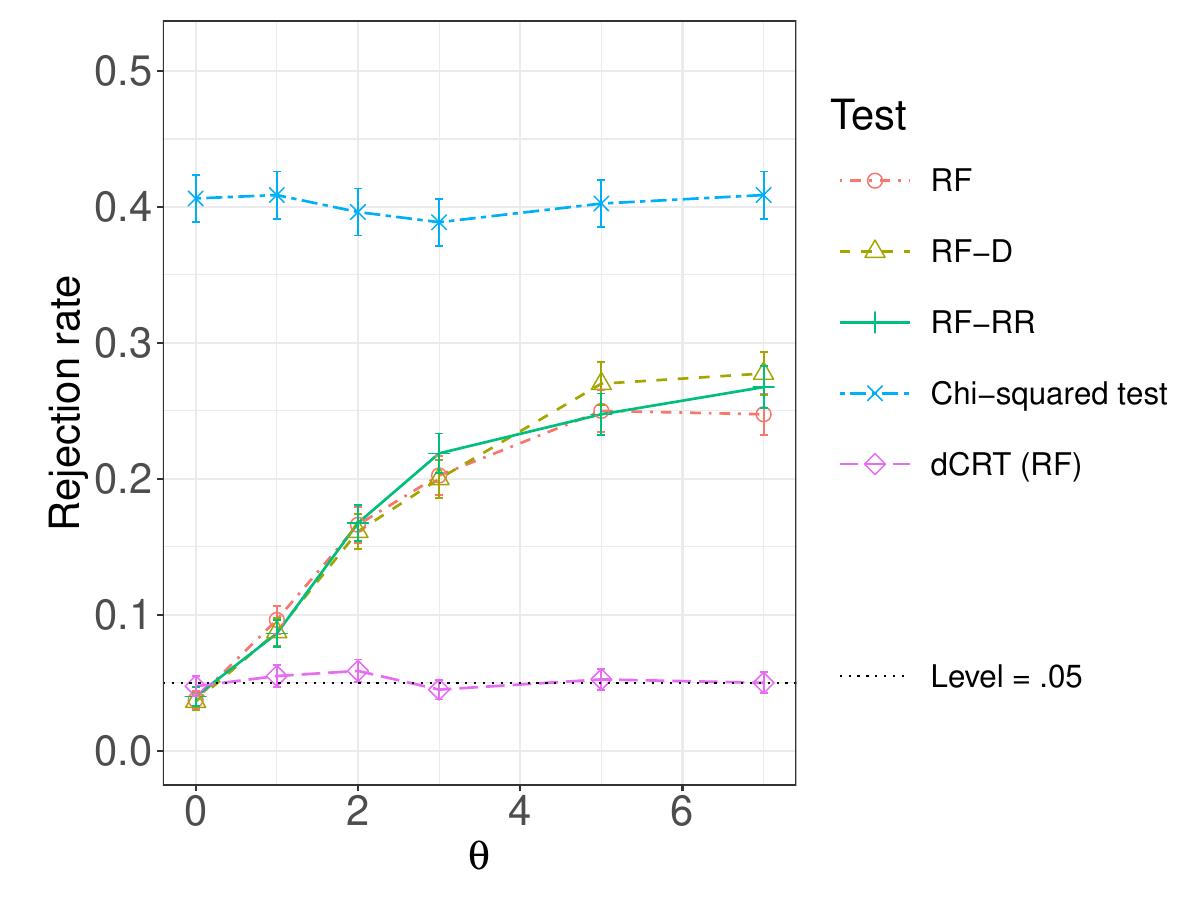}
\caption{Low-dimensional}
\label{fig: app lglm}\label{fig:lglm}
\end{subfigure}
    \begin{subfigure}[b]{0.48\textwidth}
\includegraphics[width=1\textwidth]{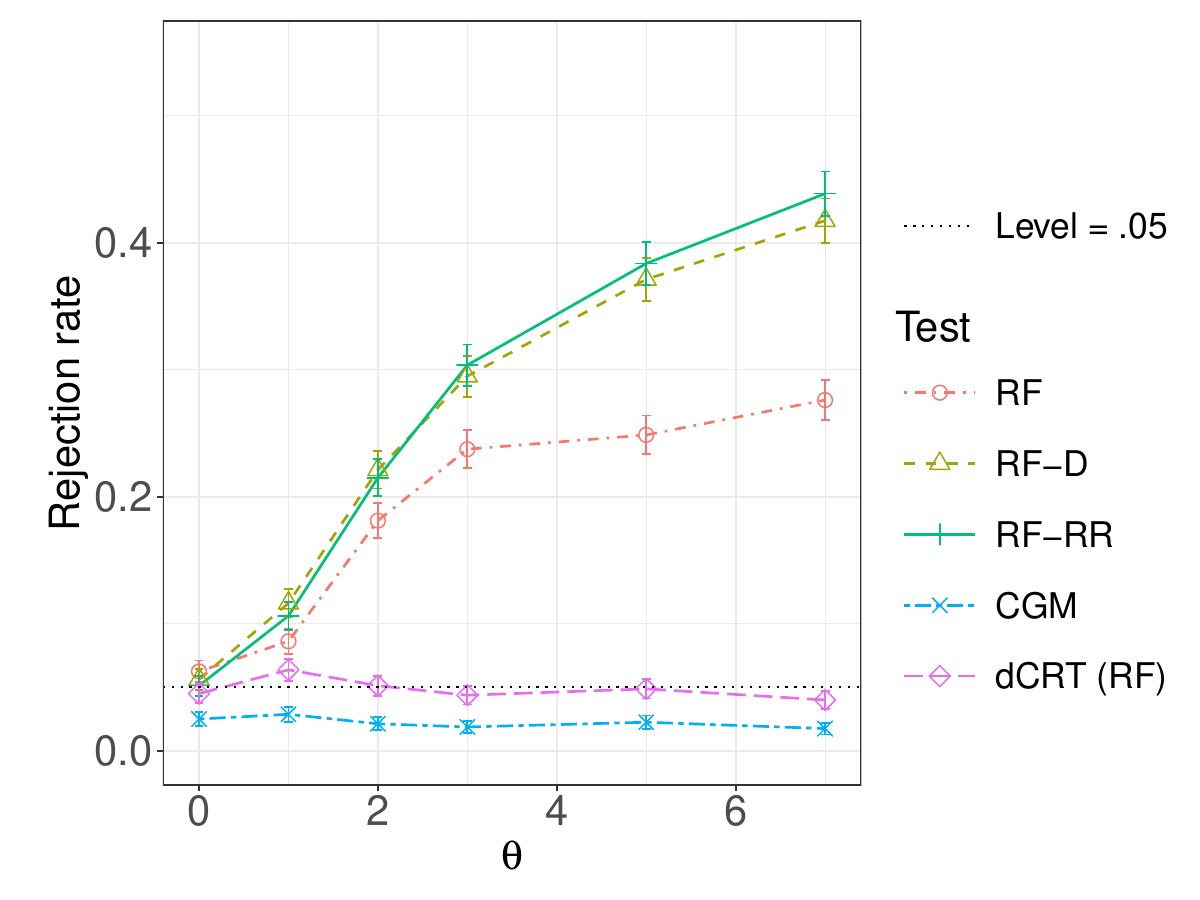}
\caption{High-dimensional}
\label{fig: app hglm}\label{fig:hglm}
    \end{subfigure}
    \caption{Rejection rates at nominal level $0.05$ for testing MCI in nonlinear binary regressions. The first three tests are the $G$-CRT with different statistics, and the other two are existing methods. Each point uses $800$ replications; error bars show one Monte Carlo standard error. The points at $\theta=0$ estimate size; the oversized Chi-squared test in panel (a) is shown only for reference.}\label{fig: app nonlinear binary}\label{fig: nonlinear binary}
\end{figure}

Figure~\ref{fig: app lglm} presents the results for low-dimensional nonlinear binary regression. The Chi-squared test has an estimated size of $0.40625$ ($325/800$). Its curve is shown for reference and does not provide a valid power comparison. 
In addition, the coordinatewise dCRT baseline is powerless even when $\theta$ is as large as 5. 
The $G$-CRTs using RF, RF-D and RF-RR have similar power curves, reaching approximately $25\%$--$28\%$ at $\theta=7$.

Figure~\ref{fig: app hglm} shows the results for high-dimensional nonlinear binary regressions.
Similar to the results in Figure~\ref{fig: app hglw}, the CGM test is too conservative and lacks power across the considered range of $\theta$.
In addition, the coordinatewise dCRT baseline falls short of achieving high power.
For the $G$-CRT with any of the three statistics, power shows an increasing trend as $\theta$ increases. However, the statistics RF-D and RF-RR exhibit significantly better performance than the RF statistic.
This discrepancy in performance corroborates the finding from Figure~\ref{fig: app hlm} that distillation and distilled residuals can boost the power of the $G$-CRT, especially in high-dimensional problems.

\subsection{Univariate CRT}\label{app:crt-univariate}

\begin{figure}[!tbp]
    \centering

    \begin{minipage}{0.48\textwidth}
        \centering
        \includegraphics[width=\linewidth]{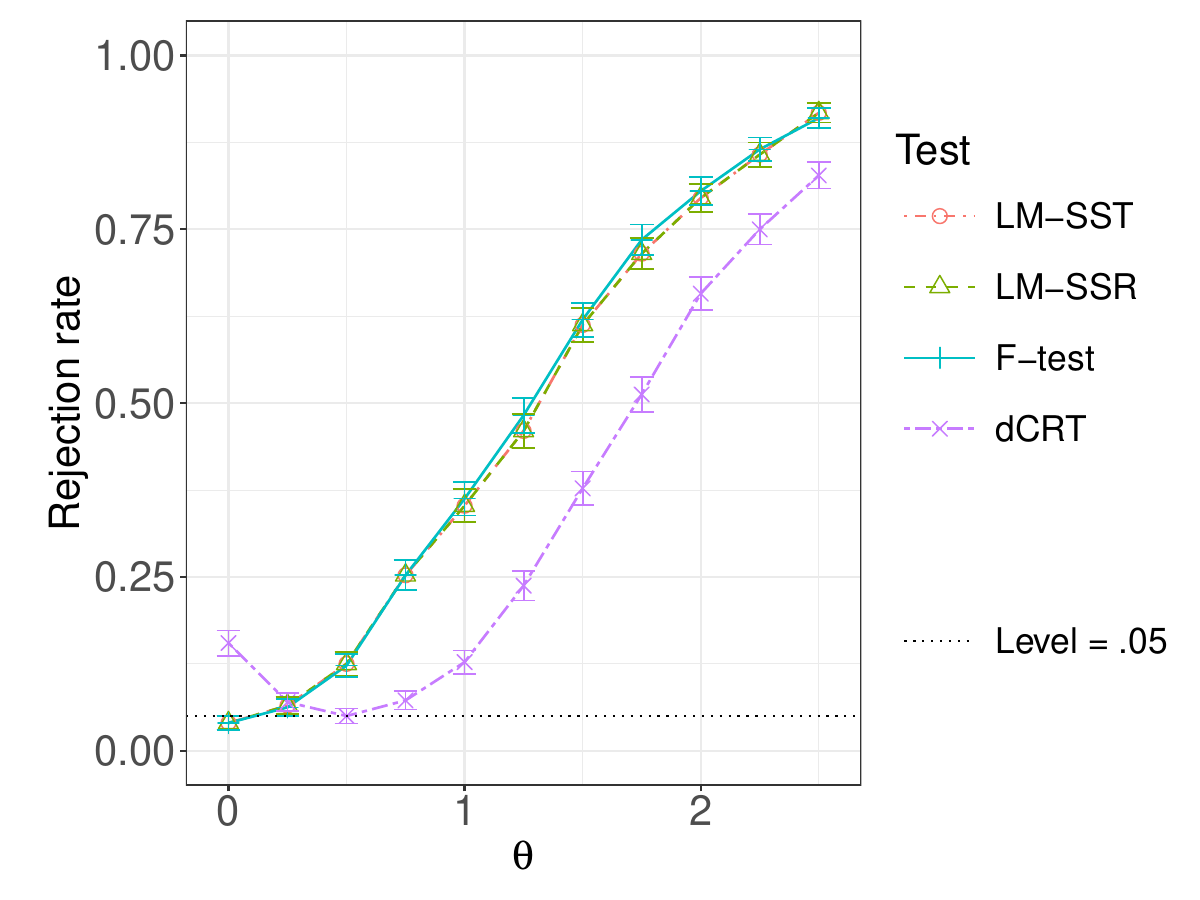}
        \textbf{(a)} Low-dimensional
    \end{minipage}
    \hfill
    \begin{minipage}{0.48\textwidth}
        \centering
        \includegraphics[width=\linewidth]{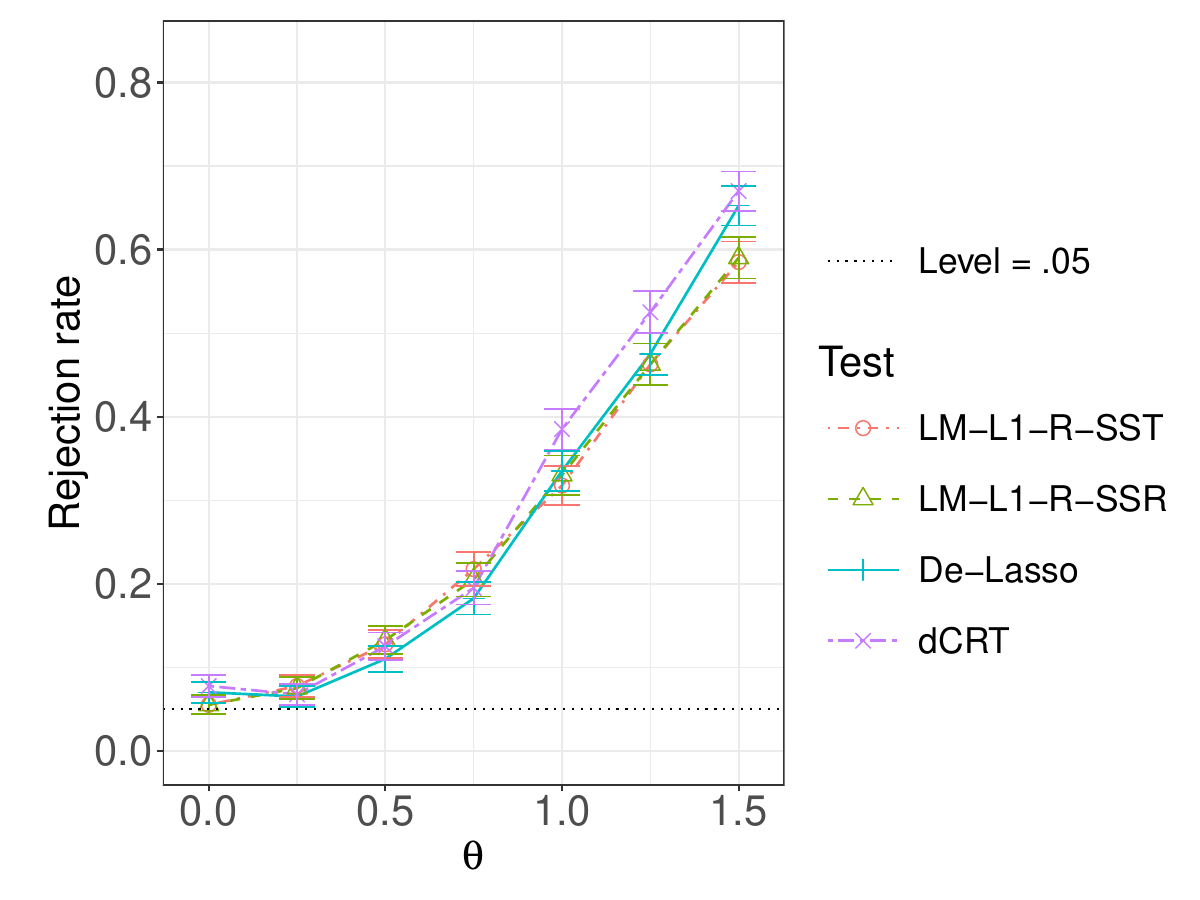}
        \textbf{(b)} High-dimensional
    \end{minipage}

    \caption{Rejection rates at nominal level $0.05$ for testing univariate conditional independence in linear regression with target set $\mc T=\{1\}$. (a) $n=50$ and $p=20$; (b) $n=80$ and $p=120$. Each point uses $400$ replications; error bars show one Monte Carlo standard error. The points at $\theta=0$ estimate size. The estimated size of dCRT in panel (a) is $0.1550$, so its curve is shown only for reference.}
    \label{fig: linear regression-univariate}
\end{figure}

To examine whether the conditional sampling used by the $G$-CRT leads to a substantial loss of power for a single target covariate, we revisit the linear regression experiment from Appendix~\ref{sec: simulation CRT linear regression} with $\mc T=\{1\}$. Figure~\ref{fig: linear regression-univariate} presents the results. In the low-dimensional setting, the $G$-CRT with either the LM-SST or LM-SSR statistic has power comparable to that of the classical $F$-test. In the high-dimensional setting, the $G$-CRT with either the LM-L1-R-SST or LM-L1-R-SSR statistic has power broadly comparable to that of the de-sparsified Lasso and dCRT, although its observed power is lower at the larger signal strengths. For example, at $\theta=1.5$, the $G$-CRT with LM-L1-R-SSR has power $0.5900$, compared with $0.6525$ for the de-sparsified Lasso and $0.6700$ for dCRT.

\subsection{Supergraph in $G$-CRT}\label{app: CRT super}

The $G$-CRT can use any supergraph of the conditional independence graph. 
However, the sufficient statistic for a supergraph includes additional sample cross-products, which can further constrain the randomised covariate copies and may reduce power.
We therefore compare the power of the $G$-CRT under the conditional independence graph and two supergraphs.

We use the high-dimensional linear regression setting from Appendix~\ref{sec: simulation CRT linear regression}, with $n=80$ and $p=120$. The conditional independence graph $G_0$ is a banded graph with bandwidth $K=6$. Its two nested supergraphs, $G_1$ and $G_2$, have bandwidths $K=12$ and $K=18$, respectively. We use \textbf{LM-L1-R-SSR} as the test statistic for all three graphs. The maximum neighbourhood size is $36$ under $G_2$, so $n=80$ satisfies $n>|N_i|+1$ for every target node in all three graphs, as required for the update by residual rotation.

\begin{figure}[!tbp]
    \centering
    \includegraphics[width=0.6\linewidth]{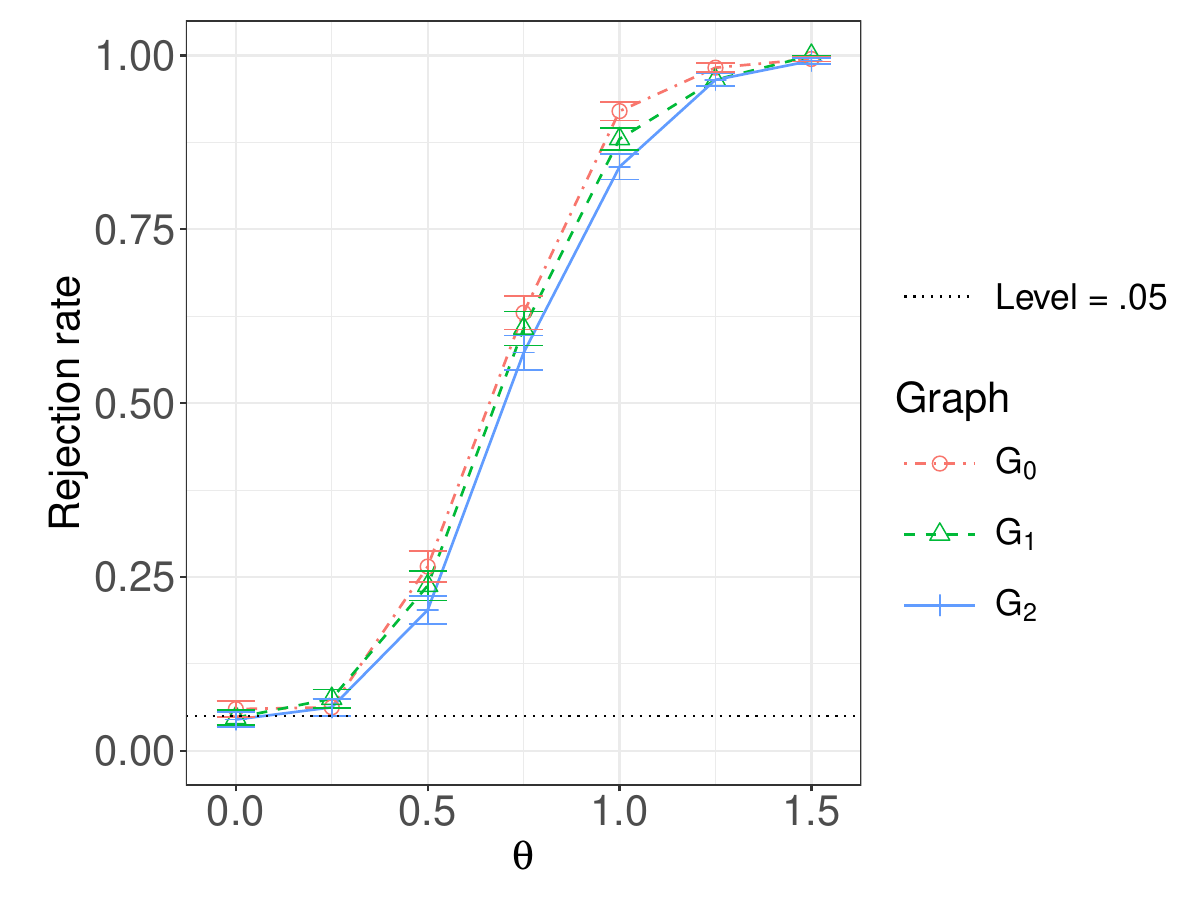}
    \caption{Rejection rates at nominal level $0.05$ for the $G$-CRT in linear regression with $n=80$ and $p=120$. The conditional independence graph $G_0$ has bandwidth $K=6$; its supergraphs $G_1$ and $G_2$ have bandwidths $K=12$ and $K=18$, respectively. Each point uses $400$ replications; error bars show one Monte Carlo standard error. The points at $\theta=0$ estimate size.}
    \label{fig: varying graphsize}
\end{figure}

Figure~\ref{fig: varying graphsize} shows lower observed power under $G_1$ and $G_2$ than under $G_0$ at intermediate signal strengths.
At $\theta=0.5$, the powers under $G_0$, $G_1$ and $G_2$ are $0.2650$, $0.2375$ and $0.2025$, respectively; at $\theta=1$, they are $0.9200$, $0.8800$ and $0.8400$.
At $\theta=1.5$, all three powers exceed $0.99$, despite the maximum neighbourhood size of $36$ under $G_2$ compared with $12$ under $G_0$.
Thus, in this experiment, the $G$-CRT with \textbf{LM-L1-R-SSR} has high estimated power under the two supergraphs when the signal is strong.

\newpage
\subsection{Joint and coordinatewise tests with graph information}\label{sec:matched-ridge}\label{app:matched-ridge}

We assess the power of $G$-CRT when the covariate mean and precision matrix are unknown by comparing it with oracle joint CRTs using the same statistics. We also compare $G$-CRT with a coordinatewise test using the same true conditional independence graph, across covariance designs and target signal directions. For each joint CRT, we consider the uncentred ridge statistic, the maximum absolute t-statistic among the target coefficients (MaxCoefT) in \eqref{eq:matched-maxcoeft}, and the combined test that applies Bonferroni correction to their two p-values. The coordinatewise test applies Bonferroni correction across targets.

Let $\bs Y\in\mathbb R^n$ and $\mathbf X\in\mathbb R^{n\times p}$ be the response vector and covariate matrix, respectively, with target set $\mathcal T$, background set $\mathcal S=\mathcal T^c$ and $t=|\mathcal T|$.

\subsubsection*{Joint CRT tests}

We consider $G$-CRT and the oracle joint CRT with each of two statistics: the uncentred ridge statistic and MaxCoefT. For each method, the combined test applies Bonferroni correction to the two p-values. The oracle joint CRT draws target vectors from their true conditional Gaussian law; $G$-CRT generates its randomised copies using the supplied graph and Algorithm~\ref{alg: exchangeable}. The uncentred ridge statistic is defined in \eqref{gcrtint:eq:statistic}, with the same residual operator from ridge regression on the background covariates applied to the observed target matrix and every randomised copy.

The MaxCoefT statistic is designed to detect a large effect at an individual target. It uses the two regression fits of LM-L1-R-SST in Table~\ref{tab:mscrt-statistic-catalogue}, but replaces the sum of squared t-statistics for the target coefficients with the maximum absolute t-statistic. Let $\widehat{\bs Y}_0$ be the fitted values from the Gaussian Lasso regression of $\bs Y$ on the background covariates $\mathbf X_{\mc S}$, and let $\bs R_Y=\bs Y-\widehat{\bs Y}_0$.

For each candidate target matrix $\mathbf x\in\mathbb R^{n\times t}$, regress $\bs R_Y$ on an intercept and all columns of $\mathbf x$ by ordinary least squares. For $j\in\mathcal T$, let $\widehat b_j(\mathbf x)$ be the coefficient of target $j$ and $\widehat{\operatorname{se}}_j(\mathbf x)$ its usual regression standard error, using $n-t-1$ residual degrees of freedom. Thus each coefficient adjusts for the other targets, and its standard error provides a common scale for the maximum. Define
\begin{equation}\label{eq:matched-maxcoeft}
 T^{\max}(\mathbf x)=\max_{j\in\mathcal T}
 \frac{|\widehat b_j(\mathbf x)|}{\widehat{\operatorname{se}}_j(\mathbf x)}.
\end{equation}
The background fit and $\bs R_Y$ are shared by all candidates, while the target regression and its standard errors are recomputed for each candidate.

Let $p_R$ and $p_{\max}$ be the $G$-CRT p-values computed by \eqref{eqn:pval-1} with the uncentred ridge and MaxCoefT statistics, respectively. The combined test uses the Bonferroni p-value $\min\{1,2\min(p_R,p_{\max})\}$.

\subsubsection*{Graph-based coordinatewise test}

For comparison, we use the same graph to test each target for conditional independence from the response given the remaining covariates. We combine the p-values for the individual targets by Bonferroni correction.

Let $\mathbf H\in\mathbb R^{n\times n}$ be the residual operator for the background ridge fit described in Section~\ref{sec:gcrt-ridge-power}, and define $\bs Y_R=\mathbf H\bs Y$. For each target $j$, first regress $\bs Y_R$ on an intercept and the other residualised targets. Next, regress this response residual and $\mathbf X_j$ separately on an intercept and the neighbours of $j$ in the supplied graph. The correlation between the two resulting residuals provides the statistic for that target.

For a matrix $\mathbf A$, let $\mathbf P_{\mathbf A}$ be the orthogonal projection onto its column space. For each $j\in\mathcal T$, define $\mathbf A_j=[\bs1_n,\mathbf H\mathbf X_{\mathcal T\setminus\{j\}}]$ and $\mathbf D_j=[\bs1_n,\mathbf X_{N_j}]$, where $N_j$ is the supplied neighbourhood of target $j$, and compute
\begin{equation}\label{eq:matched-coordinate-residuals}
\begin{aligned}
 \bs v_j&=(\mathbf I_n-\mathbf P_{\mathbf A_j})\bs Y_R,&
 \bs r_j&=(\mathbf I_n-\mathbf P_{\mathbf D_j})\mathbf X_j,\\
 \bs u_j&=(\mathbf I_n-\mathbf P_{\mathbf D_j})\bs v_j,&
 d_j&=n-\operatorname{rank}(\mathbf D_j).
\end{aligned}
\end{equation}

Let $F_{a,b}$ be the distribution function of $\operatorname{Beta}(a,b)$. For $d_j>1$, $\|\bs r_j\|>0$ and $\|\bs u_j\|>0$, the p-value for target $j$ is
\begin{equation}\label{eq:matched-coordinate-pvalue}
 w_j=\frac{(\bs r_j^\top\bs u_j)^2}{\|\bs r_j\|^2\|\bs u_j\|^2},
 \qquad p_j^{G}=1-F_{1/2,(d_j-1)/2}(w_j).
\end{equation}
This p-value is one if $d_j\leq1$, $\|\bs r_j\|=0$ or $\|\bs u_j\|=0$. The graph-based coordinatewise test uses the Bonferroni p-value $\min\{1,t\min_{j\in\mathcal T}p_j^G\}$.

Under the Gaussian covariate model with the correct supplied graph, consider the null hypothesis $Y\perp X_j\mid\boldsymbol X_{-j}$. Conditional on $(\bs Y,\mathbf X_{-j})$ and $\|\bs r_j\|$, the direction $\bs r_j/\|\bs r_j\|$ is uniform on the unit sphere in the $d_j$-dimensional orthogonal complement of the column space of $\mathbf D_j$. The vector $\bs u_j$ is fixed under this conditioning and lies in the same subspace. Hence $w_j$ has the $\operatorname{Beta}(1/2,(d_j-1)/2)$ distribution, which gives the upper-tail p-value in \eqref{eq:matched-coordinate-pvalue}.

\begin{table}[!b]
\centering
\small\setstretch{1.05}
\setlength{\tabcolsep}{5pt}
\renewcommand{\arraystretch}{1.15}
\caption{Rejection rates (\%) at nominal level $0.05$, based on 500 replications per setting. Coord. G is the graph-based coordinatewise test with Bonferroni correction. Ridge denotes the uncentred ridge statistic, and MaxCoefT is defined in \eqref{eq:matched-maxcoeft}. Combined applies Bonferroni correction to the ridge and MaxCoefT p-values. Each marginal Monte Carlo standard error is at most 2.3 percentage points.}
\label{tab:matched-direct-comparison}
\label{tab:matched-statistic-geometry}
\begin{tabular*}{\linewidth}{@{\extracolsep{\fill}}lr|r|rrr|rrr@{}}
\hline
& & & \multicolumn{3}{c|}{$G$-CRT} & \multicolumn{3}{c}{Oracle joint CRT} \\
Direction & $\theta$ & \shortstack{Coord.\\G} & Ridge & MaxCoefT & Combined & Ridge & MaxCoefT & Combined \\
\hline
\multicolumn{9}{l}{\textbf{Design A}} \\[2pt]
Null & $0$ & $3.8$ & $5.8$ & $4.6$ & $5.0$ & $5.4$ & $4.2$ & $5.4$ \\
Dense & $0.5$ & $47.2$ & $16.8$ & $35.2$ & $29.2$ & $18.2$ & $33.8$ & $26.2$ \\
Dense & $0.75$ & $82.8$ & $39.6$ & $69.6$ & $67.8$ & $38.8$ & $68.8$ & $66.2$ \\
Dense & $1$ & $97.2$ & $68.6$ & $88.2$ & $89.2$ & $68.2$ & $88.8$ & $88.4$ \\
Sparse & $0.5$ & $44.0$ & $34.0$ & $65.2$ & $62.2$ & $34.0$ & $62.4$ & $59.4$ \\
Sparse & $0.75$ & $79.8$ & $65.0$ & $93.2$ & $93.6$ & $64.8$ & $93.4$ & $93.4$ \\
Sparse & $1$ & $97.8$ & $89.2$ & $99.6$ & $100.0$ & $89.0$ & $99.8$ & $99.8$ \\
\noalign{\vskip 0.8em}
\multicolumn{9}{l}{\textbf{Design B}} \\[2pt]
Null & $0$ & $4.8$ & $4.8$ & $4.8$ & $5.2$ & $4.4$ & $4.8$ & $5.4$ \\
Dense & $0.5$ & $9.6$ & $48.4$ & $9.2$ & $39.0$ & $52.8$ & $9.2$ & $43.0$ \\
Dense & $0.75$ & $10.0$ & $75.8$ & $12.4$ & $70.0$ & $77.6$ & $12.4$ & $72.0$ \\
Dense & $1$ & $23.4$ & $93.6$ & $27.2$ & $91.8$ & $94.2$ & $25.6$ & $92.6$ \\
Sparse & $0.5$ & $36.6$ & $36.6$ & $43.0$ & $45.4$ & $36.4$ & $42.0$ & $45.0$ \\
Sparse & $0.75$ & $75.2$ & $75.8$ & $81.2$ & $87.0$ & $77.2$ & $81.4$ & $85.4$ \\
Sparse & $1$ & $95.0$ & $95.2$ & $98.2$ & $99.2$ & $95.8$ & $97.8$ & $98.8$ \\
\hline
\end{tabular*}

\end{table}

\subsubsection*{Design and implementation details}

We fix $n=80$, $p=120$ and $t=8$, with target set $\mathcal T=\{1,\ldots,8\}$ and background set $\mathcal S=\mathcal T^c$. Before marginal rescaling, the precision matrix has diagonal entries one and off-diagonal entries $a$ for $0<|i-j|\leq b$ and zero otherwise, with $(b,a)=(6,0.20)$ in design A and $(1,-0.45)$ in design B. We rescale each covariate to have variance one and let $\mathbf\Sigma$ be the resulting covariance matrix. We generate independent covariate vectors $\boldsymbol X_i\sim \mathcal{N}_p(\boldsymbol0,\mathbf\Sigma)$, $i=1,\ldots,n$.

Let $\mathbf V=\mathbf\Sigma_{\mathcal T\mid\mathcal S}$ be the conditional covariance of the target covariates given the background covariates, and define $\mathbf C=\mathbf\Sigma_{\mathcal T\mathcal S}\mathbf\Sigma_{\mathcal S\mathcal S}^{-1}$. Let $\boldsymbol d$ be $\boldsymbol1_8/\sqrt8$ for the dense direction or the first coordinate unit vector in $\mathbb R^8$ for the sparse direction. Define $\boldsymbol v=\boldsymbol d/\sqrt{\boldsymbol d^\top\mathbf V\boldsymbol d}$. Let $\theta\geq0$ measure signal strength. The response is
\begin{equation}\label{eq:matched-ridge-response}
Y_i=\boldsymbol X_{i,\mathcal S}^{\top}\boldsymbol\gamma
 +\theta(\boldsymbol X_{i,\mathcal T}-\mathbf C\boldsymbol X_{i,\mathcal S})^\top\boldsymbol v
 +\varepsilon_i,
\end{equation}
where $\boldsymbol\gamma=(1,-0.7,0.5,0.3,0,\ldots,0)^\top$, and the $\varepsilon_i$ are independent standard Gaussian variables, independent of the covariates. 
The normalisation of $\boldsymbol v$ makes the target component in \eqref{eq:matched-ridge-response} have conditional variance $\theta^2$ given the background covariates. For each design, we use 500 independent replications for the null $\theta=0$ and for both dense and sparse alternatives at $\theta=0.5,0.75,1$.

The uncentred ridge statistic and the coordinatewise test use a background ridge fit with penalty $0.2n$. The target penalty for the uncentred ridge statistic is one. The MaxCoefT statistic uses a Lasso background fit with the tuning rule in Appendix~\ref{app:experimental-pvalue-conventions}. Both $G$-CRT and the oracle joint CRT use $M=999$ randomised copies, and $G$-CRT uses $L=3$ sweeps. Within each sampler and replication, the two statistics use the same copies. Each test rejects when its p-value is at most $0.05$.

\subsubsection*{Comparison results}

In the settings considered, $G$-CRT achieves power close to the corresponding oracle joint CRTs without knowing the covariate mean or precision matrix. Across all settings in Table~\ref{tab:matched-direct-comparison}, their rejection rates differ by at most $4.4$ percentage points. Under the null, the rejection rates of the four graph-based tests range from $3.8\%$ to $5.8\%$.

In these experiments, $G$-CRT can achieve high power when the coordinatewise test has low power, but this advantage depends on the statistic. For the dense alternatives, $G$-CRT with the uncentred ridge statistic is more powerful than the coordinatewise test in design B, whereas the ordering is reversed in design A. Replacing the uncentred ridge statistic with MaxCoefT increases the power of $G$-CRT in design A but reduces it in design B. These patterns hold at each of the three signal strengths considered. For example, at $\theta=0.75$, $G$-CRT with the uncentred ridge statistic has a rejection rate of $75.8\%$ in design B, compared with $10.0\%$ for the coordinatewise test. In design A, the corresponding rates are $39.6\%$ and $82.8\%$, respectively. For the sparse alternatives, $G$-CRT with MaxCoefT is more powerful than both the coordinatewise test and $G$-CRT with the uncentred ridge statistic in both designs.

For the dense alternatives, design B illustrates how the targets can jointly provide substantial signal beyond the background covariates while each target contributes little additional signal once the other targets are also taken into account. The coordinatewise test looks for a large absolute residual correlation, and MaxCoefT selects the largest absolute t-statistic for a target coefficient after fitting all targets jointly. Both therefore seek a large individual score. The uncentred ridge statistic in \eqref{gcrtint:eq:statistic} uses the overall improvement in the penalised fit, allowing contributions from several targets to enter its score.

To quantify the individual conditional effects, let $\mathbf\Omega=\mathbf\Sigma^{-1}$, and let $\beta_j$ be the coefficient of target $j$ in \eqref{eq:matched-ridge-response}. Since $1/\Omega_{jj}=\operatorname{Var}(X_j\mid\boldsymbol X_{-j})$, the part of the response signal that remains after conditioning on all covariates other than target $j$ has standard deviation $|\beta_j|/\sqrt{\Omega_{jj}}$. We measure this signal relative to the noise and account for sample size through $\sqrt n|\beta_j|/(\sigma\sqrt{\Omega_{jj}})$, where $\sigma=1$ is the noise standard deviation. For example, at $\theta=0.75$, these effects range from $3.514$ to $3.722$ in design A and from $0.934$ to $1.198$ in design B, while the conditional signal variance of the target group is $0.5625$ in both designs. Thus, at the same group signal strength, design B provides much less signal for detecting an individual conditional effect. This is consistent with the low power of the coordinatewise test and $G$-CRT with MaxCoefT in B, alongside the higher power of $G$-CRT with the uncentred ridge statistic, which uses the overall improvement in fit.

To understand why $G$-CRT with the uncentred ridge statistic is more powerful in design B than in design A under the dense alternatives, we compare the observed and randomised gains. By \eqref{gcrtb:eq:rank}, a small p-value requires few randomised gains to be at least as large as the observed gain. In the leading approximations \eqref{gcrtb:eq:observed-score-approximation} and \eqref{gcrtb:eq:randomised-score-approximation}, randomisation changes the target contribution to the residual score vector, while the background contribution left by ridge residualisation is common to the observation and copies. The corresponding leading gains use the same quadratic form. The term involving only the background therefore cancels from their difference, while its interaction with the changing target contribution remains. This interaction can reinforce or offset the gain difference produced by the target contribution.

For the dense alternative at $\theta=0.75$, we numerically approximate the expectations in \eqref{gcrtb:eq:population} using the background covariate samples from the 500 replications in each design. The target contribution gives a larger positive difference between the observed and randomised gains in design B than in design A. Its interaction with the remaining background contribution further increases this difference in B and reduces it in A. Together, the two terms give leading approximations of $0.0408$ in A and $0.0945$ in B for the observed gain minus a randomised gain. Thus, the leading calculation places the observed gain further above the randomised gains in B, providing a possible explanation for the higher rejection rate of $G$-CRT with the uncentred ridge statistic: $75.8\%$ in B compared with $39.6\%$ in A, despite the same group signal strength.

The combined test draws on both statistics without selecting one in advance. Across the alternatives considered, its rejection rate exceeds that of the less powerful component, while it is lower than that of the more powerful component in some settings and higher in others. The Bonferroni adjustment imposes a stricter threshold on each p-value, while allowing either component to trigger rejection. For example, in design B at $\theta=0.75$, the combined test has a rejection rate of $70.0\%$ under the dense alternative, compared with $75.8\%$ for $G$-CRT with the uncentred ridge statistic. Under the sparse alternative, its rejection rate is $87.0\%$, compared with $81.2\%$ for $G$-CRT with MaxCoefT. Relative to the more powerful component in each setting, these are a decrease and an increase of $5.8$ percentage points, with paired Monte Carlo standard errors of $1.1$ and $1.4$ percentage points, respectively.

\subsection{Performance under non-Gaussian covariate models}\label{app:nongaussian-covariates}

We examine how departures from Gaussianity affect the size of $G$-CRT and MVN-CRT. 

\subsubsection{G-CRT with violation of the Gaussian assumption}\label{app: CRT vio}

The $G$-CRT considered here assumes multivariate Gaussian covariates, an assumption that may be violated in applications.
To assess its robustness, we examine estimated size and power in the linear regression experiment from Appendix~\ref{sec: simulation CRT linear regression} with heavy-tailed, non-Gaussian covariates.
We generate the covariates as follows:
\begin{enumerate}
    \item We generate an initial design matrix $\mathbf{X}^{(0)}\in\mathbb R^{n\times p}$ from an AR(6) model. Let $t_3$ be the standard $t$ distribution with $3$ degrees of freedom. The entries of the first column $\mathbf{X}_{:,1}^{(0)}$ and of the noise vectors $\bs{\varepsilon}_j\in\mathbb R^n$, $j=2,\ldots,p$, are mutually independent draws from $t_3$.
Define $k=6$, $\boldsymbol{\gamma}=(\gamma_1,\ldots,\gamma_k)=(1,1,-1,-1,1,1)/5$ and $\sigma_{\varepsilon}=0.7$.
\par
For each $j=2,\ldots,p$, let $k_j=\min(j-1,k)$ and define
$$
\mathbf{X}_{:,j}^{(0)} \;=\;
\sum_{\ell=1}^{k_j} \gamma_{\ell}\,\mathbf{X}_{:,\,j-\ell}^{(0)} \;+\; \sigma_{\varepsilon}\bs{\varepsilon}_j.
$$

\item Define $G_0$ on the node set $[p]$, with an edge between $i$ and $j$ if and only if $1\leq|i-j|\leq6$. We apply the same random permutation to the columns of $\mathbf{X}^{(0)}$ and the nodes of $G_0$ to obtain the covariate matrix $\mathbf{X}$ and the graph $G$ used by the $G$-CRT.

\end{enumerate}

In the AR(6) model, the conditional distribution of each covariate given all the others depends only on covariates within six positions of it, which are its neighbours in $G_0$.
Thus, $\mathbf{X}$ satisfies the Markov property with respect to $G$.
The use of $t_3$ distributions for the initial covariate and noise terms produces heavy-tailed, non-Gaussian covariates, so the Gaussian assumption is violated.

\begin{figure}[!tbp]
    \centering

    \begin{minipage}{0.48\textwidth}
        \centering
        \includegraphics[width=\linewidth]{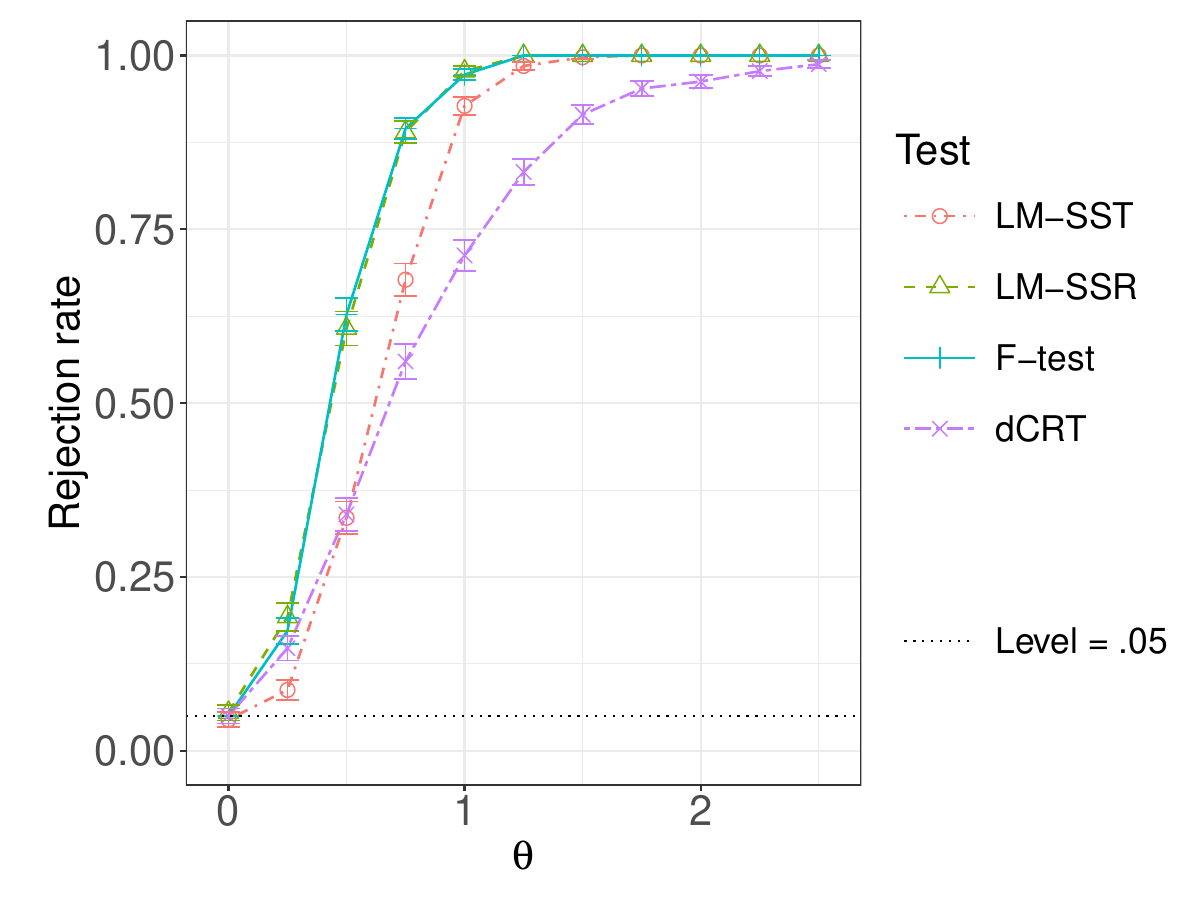}
        \textbf{(a)} Low-dimensional
    \end{minipage}
    \hfill
    \begin{minipage}{0.48\textwidth}
        \centering
        \includegraphics[width=\linewidth]{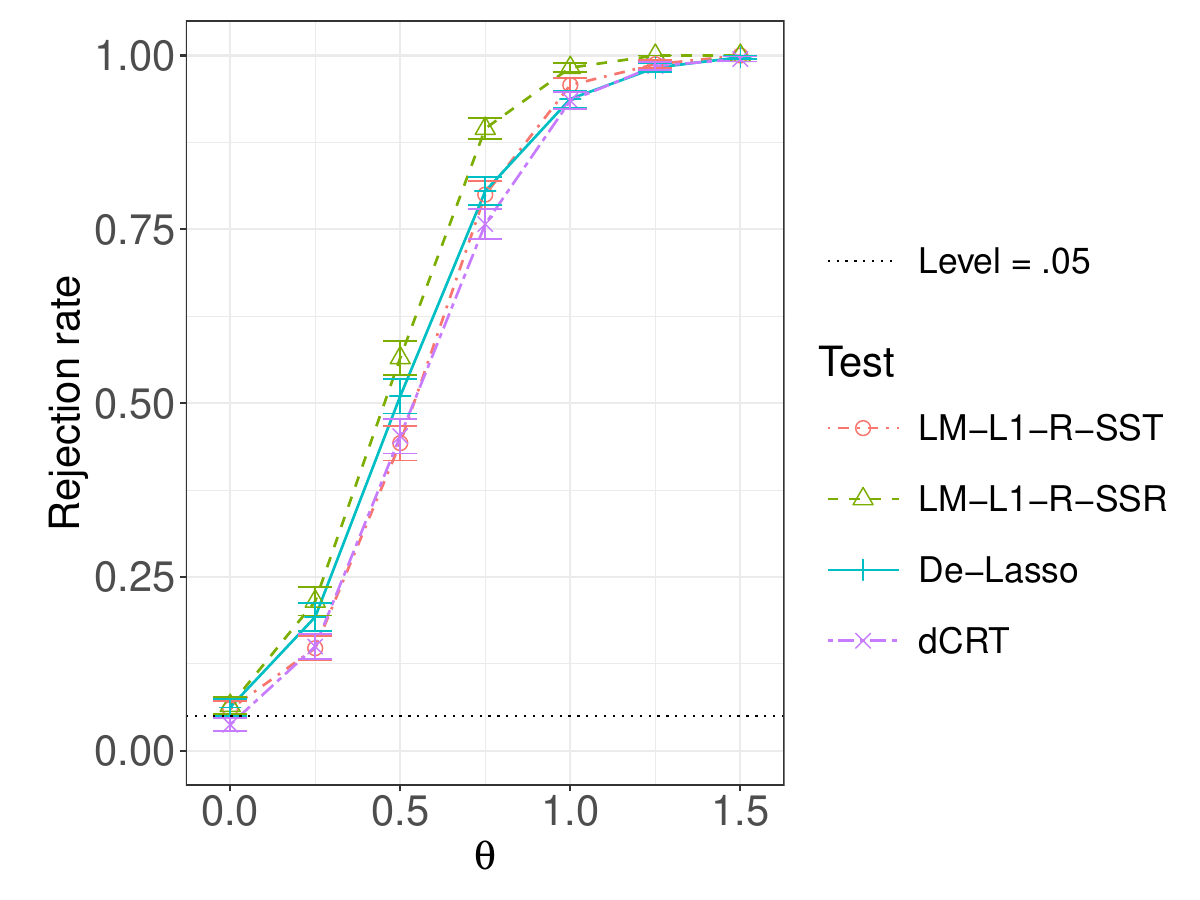}
        \textbf{(b)} High-dimensional
    \end{minipage}

    \caption{Rejection rates at nominal level $0.05$ for testing conditional independence in linear regression with non-Gaussian covariates. (a) $n=50$ and $p=20$; (b) $n=80$ and $p=120$. Each point uses $400$ replications; error bars show one Monte Carlo standard error. The points at $\theta=0$ estimate size.}
    \label{fig: linear regression-t3}
\end{figure}

Figure~\ref{fig: linear regression-t3} shows estimated sizes between $0.0450$ and $0.0650$ for the $G$-CRT at nominal level $0.05$.
In the low-dimensional setting, the $G$-CRT with the LM-SSR statistic has power comparable to that of the classical $F$-test.
In the high-dimensional setting, the $G$-CRT with the LM-L1-R-SSR statistic has higher observed power than the Bonferroni-adjusted de-sparsified Lasso and dCRT at all positive signal strengths shown; at $\theta=0.75$, the respective powers are $0.8950$, $0.8050$ and $0.7575$.
These results provide numerical evidence of robustness to the violation of the Gaussian assumption considered in this experiment.

\subsubsection{MVN-CRT under non-Gaussian covariate models}
\label{app:mvn-nongaussian}

We examine the size of MVN-CRT when its Gaussian covariate model is misspecified. Let $\bs Z_i\in\mathbb R^{12}$ and $\bs U_i\in\mathbb R^8$ be the background and target covariate vectors for observation $i$, respectively, and let $Y_i$ be the response. The observations are independent and identically distributed for $i=1,\ldots,n$, with $n\in\{50,100,200,400\}$.

We generate $\bs Z_i\sim\mathcal N_{12}(0,\mathbf I_{12})$ and $\bs G_i\sim\mathcal N_8(0,\mathbf\Sigma)$ independently, where the covariance matrix $\mathbf\Sigma$ has entries $\Sigma_{jk}=0.5^{|j-k|}$ for $j,k=1,\ldots,8$. We compare Gaussian covariate errors $\bs{\varepsilon}_i=\bs G_i$ with standardised multivariate $t_\nu$ errors $\bs{\varepsilon}_i=\sqrt{(\nu-2)/V_i}\,\bs G_i$, for $\nu=5$ and $3$, where $V_i\sim\chi^2_\nu$ is independent of $(\bs Z_i,\bs G_i)$. Each observation has its own $V_i$, shared by its eight target coordinates. All three error distributions have mean zero and covariance $\mathbf\Sigma$.

Let $\varepsilon_{iY}\sim\mathcal N(0,1)$ be independent of the variables used to generate the covariates, and define

$$
\begin{aligned}
U_{ij}&=\frac{Z_{ij}+Z_{i,j+1}}{\sqrt2}+\varepsilon_{ij},
\qquad j=1,\ldots,8,\\
Y_i&=\frac1{\sqrt{12}}\sum_{k=1}^{12}Z_{ik}
     +\frac{Z_{i1}^2-1}{\sqrt2}+\varepsilon_{iY}.
\end{aligned}
$$

The response depends only on the background covariates and independent noise, so $Y_i\indp\bs U_i\mid\bs Z_i$. The conditional mean and covariance of the target covariates are the same in all three models, and all population parameters remain fixed as $n$ varies.

In all three models, we apply the Gaussian MVN-CRT with the intercept adjustment in Remark~\ref{rem: general mean for MVN}. We use the ridge gain statistic in \eqref{eq:simulation-ridge-gain}, with penalty one. The background projections for both the response and the targets use an intercept and the twelve original background covariates. Each setting uses $1000$ independent replications and $M=199$ randomised copies per replication. We compute p-values by \eqref{eqn:pval-1} and evaluate rejection rates at nominal levels $\alpha=0.005,0.010,\ldots,0.100$. Figure~\ref{fig:mvn-nongaussian-rejection-rates} displays these rates with pointwise $95\%$ Wilson confidence intervals based on the observed rejection counts; Table~\ref{tab:mvn-nongaussian-size} reports the rates and intervals at $\alpha=0.05$.

\begin{figure}[!b]
    \centering
    \includegraphics[width=\linewidth]{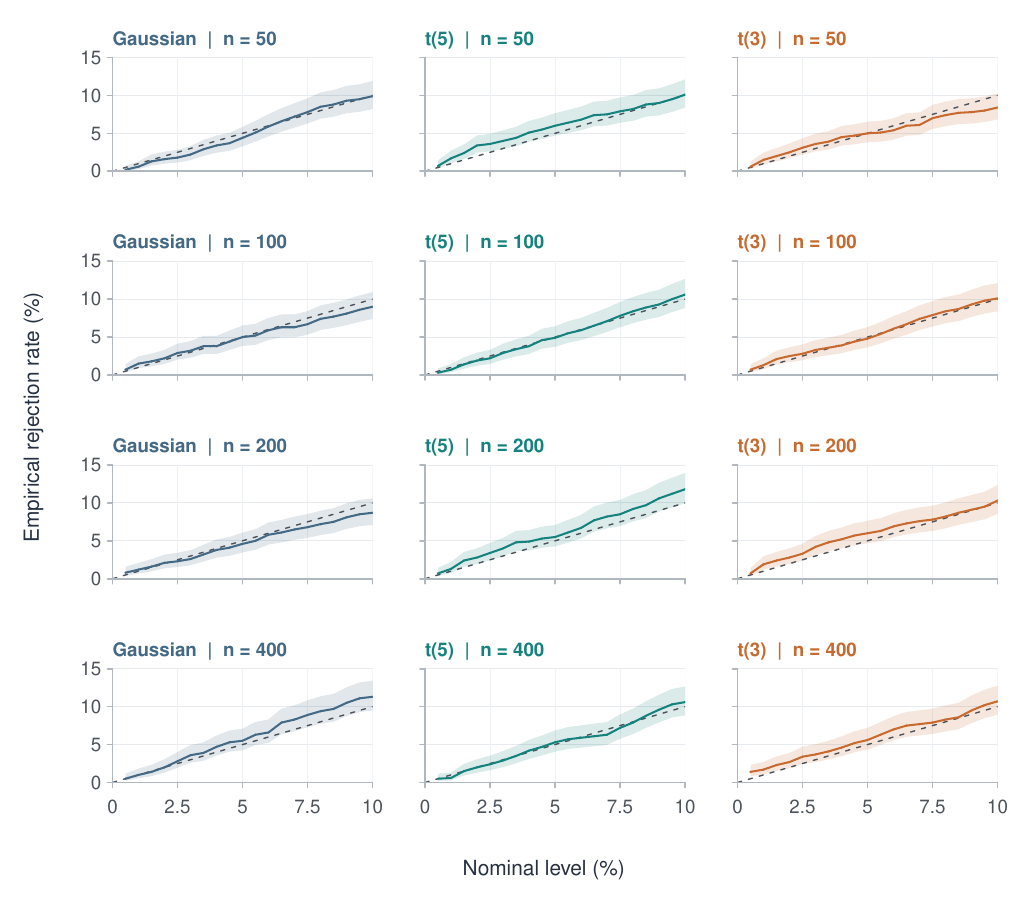}
    \caption{Empirical rejection rates of MVN-CRT under Gaussian and standardised multivariate $t_5$ and $t_3$ covariate errors. Rows give sample sizes $n=50,100,200,400$; columns give the three error distributions in that order. Nominal levels range from $0.005$ to $0.100$ in increments of $0.005$. Both axes are percentages. Each setting uses $1000$ null replications with $M=199$ randomised copies and the ridge gain statistic. The dashed line marks equality of the rejection rate and nominal level. Shading shows pointwise $95\%$ Wilson confidence intervals based on the observed rejection counts.}
    \label{fig:mvn-nongaussian-rejection-rates}
\end{figure}

\begin{table}[!t]
    \centering
\small\setstretch{1.05}
\setlength{\tabcolsep}{5pt}
\renewcommand{\arraystretch}{1.15}
\caption{Estimated size (\%) of MVN-CRT at nominal level $0.05$ under Gaussian and standardised multivariate $t_5$ and $t_3$ covariate errors. The background and target dimensions are $12$ and $8$, respectively. Each estimate uses $1000$ replications with $M=199$ randomised copies and the ridge gain statistic. Pointwise $95\%$ Wilson confidence intervals are in brackets.}
    \label{tab:mvn-nongaussian-size}
\begin{tabular*}{\linewidth}{@{\extracolsep{\fill}}r@{\hspace{14pt}}rrr@{}}
\hline
$n$ & Gaussian & $t_5$ & $t_3$ \\
\hline
$50$ & 4.4 [3.3, 5.9] & 6.0 [4.7, 7.6] & 5.0 [3.8, 6.5] \\
$100$ & 5.0 [3.8, 6.5] & 4.9 [3.7, 6.4] & 4.8 [3.6, 6.3] \\
$200$ & 4.6 [3.5, 6.1] & 5.5 [4.2, 7.1] & 6.0 [4.7, 7.6] \\
$400$ & 5.5 [4.2, 7.1] & 5.3 [4.1, 6.9] & 5.6 [4.3, 7.2] \\
\hline
\end{tabular*}
\end{table}

At nominal level $0.05$, Table~\ref{tab:mvn-nongaussian-size} gives estimated rejection rates of $4.8\%$--$6.0\%$ under the two multivariate $t$ distributions, compared with $4.4\%$--$5.5\%$ under Gaussian errors. The largest estimated excess over $0.05$ among the non-Gaussian settings is one percentage point.

Figure~\ref{fig:mvn-nongaussian-rejection-rates} also shows departures at smaller nominal levels. For $t_3$ errors, the estimated rejection rate is $1.9\%$ at nominal level $0.01$ when $n=200$, and $1.4\%$ at nominal level $0.005$ when $n=400$. The latter exceeds the nominal level by $0.9$ percentage points and is $2.8$ times that level. Thus, a small absolute difference can correspond to appreciable relative inflation at a small nominal level.

\section[Numerical comparisons for the G-CRT power analysis]{Numerical comparisons for the $G$-CRT power analysis}\label{app:power-numerical-comparisons}\label{app:additional-simulations}
We use the residual score vectors defined in Section~\ref{sec:gcrt-ridge-power} to compare the power of centred and uncentred $G$-CRT tests under different target signal directions and signs of the background contribution. Appendix~\ref{app:direction-mechanism} examines how reversing the target direction at the same signal strength changes the power of $G$-CRT with the uncentred ridge statistic. Appendices~\ref{app:centering-simulation} and~\ref{app:connected-centering} compare centred and uncentred tests in the model with two targets and with dependent targets, respectively.

\subsection{A finite-sample direction contrast}\label{app:direction-mechanism}
We examine how reversing the target signal direction at the same signal strength affects the finite-sample power of $G$-CRT with the uncentred ridge statistic. In the model with two targets in \eqref{gcrtb:eq:direction-example}, we write the response as $Y=Z+b(E_1+E_2)+\varepsilon$. We take $b=1/4$ and $b=-1/4$ for the positive and negative fixed alternatives, respectively, and include the null case $b=0$.

For each $n\in\{40,120,360\}$, we use 500 independent replications. Within each replication, the three responses share the same Gaussian variables and randomised covariate copies. We use Algorithm~\ref{alg: exchangeable} with the conditional independence graph and $M=999$ copies, and the uncentred ridge statistic in \eqref{gcrtint:eq:statistic} with background penalty $n$ and target penalty 1. We take $L=1$: each target has only $Z$ as a graph neighbour, so one sweep suffices to redraw both residual directions. We compute the p-value by \eqref{eq:crt-pvalue-conservative} at nominal level $0.05$.

The test rejects when the observed score $T_0$ is sufficiently large
relative to the copy scores $T_1,\ldots,T_M$.
Equation~\eqref{gcrtb:eq:rank} gives the same ranks and p-value for
$T_v$ and the gains $\mathcal A_v$, so we compare the gains directly.
Here ridge residualisation leaves a background contribution in the
residual score vectors. In the observed residual score vector, the
target signal reinforces this contribution under the positive alternative
and cancels it asymptotically under the negative alternative.
Equation~\eqref{gcrtb:eq:direction-gain-limits} gives, for each $m=1,\ldots,M$,
$$
\mathcal A_0-\mathcal A_m\xrightarrow{\mathbb P}
\begin{cases}
  1/4, & b=1/4,\\
  -1/12, & b=-1/4,\\
  0, & b=0.
\end{cases}
$$
Since $M$ is fixed, $\mathcal A_0>\max_{1\leq m\leq M}\mathcal A_m$
under the positive alternative and
$\mathcal A_0<\min_{1\leq m\leq M}\mathcal A_m$ under the negative
alternative with probability tending to one.
The rank equivalence in \eqref{gcrtb:eq:rank} gives the same comparisons
for $T_0,\ldots,T_M$.
The corresponding rejection probabilities tend to one and zero, respectively.

We summarise the copy scores by their mean, $\bar T=M^{-1}\sum_{m=1}^M T_m$. Table~\ref{tab:direction-mechanism} reports rejection rates together with the means of $T_0$, $\bar T$ and $T_0-\bar T$ across replications, relating the finite-sample power contrast to the predicted score separation.

\begin{table}[!t]
\centering
\small\setstretch{1.05}
\setlength{\tabcolsep}{5pt}
\renewcommand{\arraystretch}{1.15}
\caption{Scores and rejection rates in the finite-sample direction contrast. We use 500 independent replications for each sample size. Within each replication, we generate the three responses from the same Gaussian draws and use the same randomised covariate copies to compute the scores and p-values for all three responses. Scores include the response normalisation in \eqref{gcrtb:eq:rank}. SE is the Monte Carlo standard error of the mean of $T_0-\bar T$ across replications. Rejection rates are percentages at nominal level $0.05$, with pointwise 95\% exact binomial confidence intervals in brackets.}
\label{tab:direction-mechanism}
\begin{tabular*}{\linewidth}{@{\extracolsep{\fill}}rlrrr@{\hspace{16pt}}r@{}}
\hline
& & \multicolumn{3}{c}{Scores} & \shortstack{Rejection rate (\%)\\{[95\% CI]}} \\
\cline{3-5}
$n$ & Response & Mean $T_0$ & Mean $\bar T$ & \shortstack{Mean $T_0-\bar T$\\(SE)} & \\
\hline
$40$ & Null & $0.109$ & $0.108$ & $0.0012\;(0.0033)$ & $5.6\;[3.8, 8.0]$ \\
$40$ & Positive & $0.273$ & $0.104$ & $0.1696\;(0.0048)$ & $60.6\;[56.2, 64.9]$ \\
$40$ & Negative & $0.040$ & $0.102$ & $-0.0618\;(0.0019)$ & $0.4\;[0.0, 1.4]$ \\
\noalign{\vskip 0.5em}
$120$ & Null & $0.081$ & $0.081$ & $0.0009\;(0.0017)$ & $5.6\;[3.8, 8.0]$ \\
$120$ & Positive & $0.254$ & $0.075$ & $0.1786\;(0.0026)$ & $99.0\;[97.7, 99.7]$ \\
$120$ & Negative & $0.012$ & $0.074$ & $-0.0621\;(0.0008)$ & $0.0\;[0.0, 0.7]$ \\
\noalign{\vskip 0.5em}
$360$ & Null & $0.072$ & $0.071$ & $0.0008\;(0.0010)$ & $5.4\;[3.6, 7.8]$ \\
$360$ & Positive & $0.248$ & $0.065$ & $0.1825\;(0.0015)$ & $100.0\;[99.3, 100.0]$ \\
$360$ & Negative & $0.005$ & $0.065$ & $-0.0607\;(0.0004)$ & $0.0\;[0.0, 0.7]$ \\
\hline
\end{tabular*}

\end{table}

The results show a pronounced power difference between the two target directions at every sample size. Power increases towards one under the positive alternative and remains near zero under the negative alternative. For example, at $n=120$, the rejection rates are $99.0\%$ and $0.0\%$, respectively. The mean score gaps have the predicted signs throughout; at $n=360$, they are $0.1825$ and $-0.0607$. Under the null, the rejection rates range from $5.4\%$ to $5.6\%$.

\subsection{Vector and scalar centring}\label{app:centering-simulation}
\label{sec:centering-simulation}

We examine whether $G$-CRT using vector centring or scalar centring and squaring recovers power against the negative fixed alternative in Appendix~\ref{app:direction-mechanism}. For the positive fixed alternative, we compare their power with that of $G$-CRT using the uncentred ridge statistic. We also compare the two forms of centring under the weak orthogonal alternative in Appendix~\ref{sec:gcrt-ridge-comparison}, taking $r_n=n^{-1/3}$. The null case $b=0$ is included.

\newpage
We compare four $G$-CRT tests. The uncentred test uses the ridge statistic in \eqref{gcrtint:eq:statistic}. For scalar centring and squaring, we centre the uncentred gains across candidates and square the deviations, as defined in Appendix~\ref{sec:gcrt-ridge-comparison}. The other two tests centre the residual score vectors before forming quadratic scores, following \eqref{gcrtcent:eq:residual-scores}--\eqref{eq:symmetric-centered-statistic}. One uses the inverse pooled Gram weight, $\mathbf W=\mathbf D_{\rm pool}^{-1}$, and the other uses unit weight, $\mathbf W=\mathbf I_2$. All four statistics use a penalty of $n$ ($\ell=1$) for ridge regression on the background covariates. We take $\lambda_n=1$ for the uncentred gains and the pooled Gram matrix.

For each $n\in\{40,120,360\}$, we use 100 independent replications.
The randomised covariate copies are generated by
Algorithm~\ref{alg: exchangeable} with the conditional independence graph, $M=19$ copies and
$L=1$ sweep. Within each replication, the four responses share the same
Gaussian draws, and all four tests use the same covariate copies.
We compute the conservative p-value in \eqref{eq:crt-pvalue-conservative}
from unrounded scores and reject at level $0.05$.

Figure~\ref{fig:centering-comparison} and Table~\ref{tab:centering-rates}
report the rejection rates. Table~\ref{tab:centering-paired} gives paired
differences relative to $G$-CRT with the uncentred ridge statistic and their
Monte Carlo standard errors.

\begin{figure}[!b]
\centering
\includegraphics[width=0.96\textwidth]{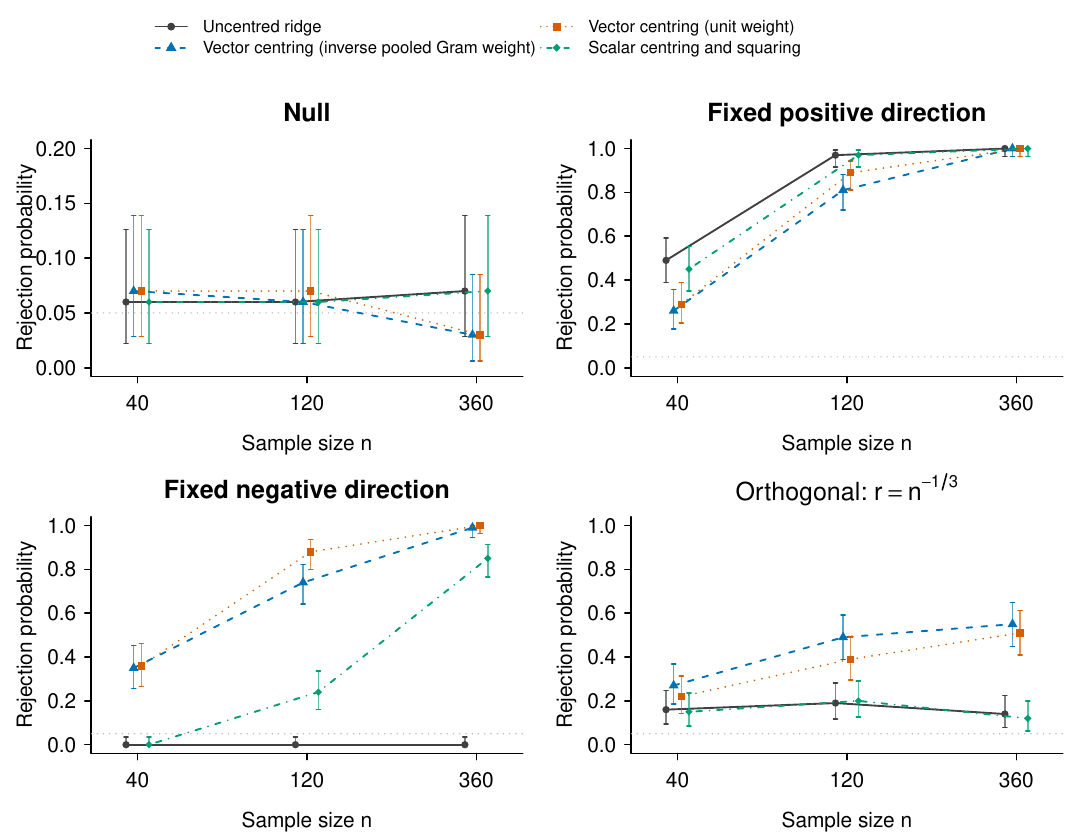}
\caption{Paired comparison of the four statistics at level 0.05.
Each point is based on 100 replications. Error bars are pointwise 95\% exact
binomial intervals; the null panel has a different vertical scale.
Horizontal offsets distinguish methods at the same sample size.
The dotted horizontal lines mark the nominal level of 0.05.}
\label{fig:centering-comparison}
\end{figure}

\newpage
\begin{table}[!t]
\centering
\small\setstretch{1.05}
\setlength{\tabcolsep}{5pt}
\renewcommand{\arraystretch}{1.15}
\caption{Rejection rates (\%) for four $G$-CRT tests at nominal level $0.05$, based on 100 replications per setting. For vector centring, the inverse pooled Gram weight is $\mathbf W=\mathbf D_{\rm pool}^{-1}$ and unit weight is $\mathbf W=\mathbf I_2$.}
\label{tab:centering-rates}
\begin{tabular*}{\linewidth}{@{\extracolsep{\fill}}r@{\hspace{16pt}}r@{\hspace{16pt}}rr@{\hspace{16pt}}r@{}}
\hline
& \shortstack{Uncentred\\ridge} & \multicolumn{2}{c}{Vector centring} & \shortstack{Scalar centring\\and squaring} \\
\cline{3-4}
$n$ &  & \shortstack{Inverse pooled\\Gram weight} & Unit weight & \\
\hline
\multicolumn{5}{l}{\textbf{(a) Null}} \\[2pt]
$40$ & $6.0$ & $7.0$ & $7.0$ & $6.0$ \\
$120$ & $6.0$ & $6.0$ & $7.0$ & $6.0$ \\
$360$ & $7.0$ & $3.0$ & $3.0$ & $7.0$ \\
\noalign{\vskip 0.8em}
\multicolumn{5}{l}{\textbf{(b) Positive fixed alternative}} \\[2pt]
$40$ & $49.0$ & $26.0$ & $29.0$ & $45.0$ \\
$120$ & $97.0$ & $81.0$ & $89.0$ & $97.0$ \\
$360$ & $100.0$ & $100.0$ & $100.0$ & $100.0$ \\
\noalign{\vskip 0.8em}
\multicolumn{5}{l}{\textbf{(c) Negative fixed alternative}} \\[2pt]
$40$ & $0.0$ & $35.0$ & $36.0$ & $0.0$ \\
$120$ & $0.0$ & $74.0$ & $88.0$ & $24.0$ \\
$360$ & $0.0$ & $99.0$ & $100.0$ & $85.0$ \\
\noalign{\vskip 0.8em}
\multicolumn{5}{l}{\textbf{(d) Weak orthogonal alternative}} \\[2pt]
$40$ & $16.0$ & $27.0$ & $22.0$ & $15.0$ \\
$120$ & $19.0$ & $49.0$ & $39.0$ & $20.0$ \\
$360$ & $14.0$ & $55.0$ & $51.0$ & $12.0$ \\
\hline
\end{tabular*}

\end{table}

Under the negative fixed alternative, both forms of centring gain power as $n$ increases, while the uncentred test never rejects. With vector centring and the inverse pooled Gram weight, the rejection rates are 35.0\%, 74.0\% and 99.0\% at $n=40,120,360$, respectively. Scalar centring and squaring reaches a rejection rate of $85.0\%$ at $n=360$. Under the positive fixed alternative, the test with vector centring and the inverse pooled Gram weight has lower rejection rates than the uncentred test at $n=40$ and $120$: 26.0\% versus 49.0\% and 81.0\% versus 97.0\%, respectively. All four tests reject in every replication at $n=360$ under the positive fixed alternative.

Under the weak orthogonal alternative in Proposition~\ref{prop:gcrt-centering-comparisons}(b), both tests with vector centring have higher rejection rates than the other two tests at every sample size. Their rejection rates increase with $n$, even though the signal strength $r_n=n^{-1/3}$ decreases. For example, at $n=360$, the rejection rates are $55.0\%$ with the inverse pooled Gram weight and $51.0\%$ with unit weight, compared with $12.0\%$ for scalar centring and squaring and $14.0\%$ for the uncentred test.

\begin{table}[!t]
\centering
\small\setstretch{1.05}
\setlength{\tabcolsep}{5pt}
\renewcommand{\arraystretch}{1.15}
\caption{Paired differences in rejection rates (percentage points), calculated as the rejection rate of each indicated $G$-CRT test minus that of $G$-CRT with the uncentred ridge statistic. Results use the same 100 replications per setting as Table~\ref{tab:centering-rates}. Monte Carlo standard errors are in parentheses.}
\label{tab:centering-paired}
\begin{tabular*}{\linewidth}{@{\extracolsep{\fill}}r@{\hspace{16pt}}rr@{\hspace{16pt}}r@{}}
\hline
& \multicolumn{2}{c}{Vector centring} & \shortstack{Scalar centring\\and squaring} \\
\cline{2-3}
$n$ & \shortstack{Inverse pooled\\Gram weight} & Unit weight & \\
\hline
\multicolumn{4}{l}{\textbf{(a) Null}} \\[2pt]
$40$ & $1.0\;(3.3)$ & $1.0\;(3.3)$ & $0.0\;(0.0)$ \\
$120$ & $0.0\;(2.5)$ & $1.0\;(2.7)$ & $0.0\;(1.4)$ \\
$360$ & $-4.0\;(2.8)$ & $-4.0\;(2.4)$ & $0.0\;(2.0)$ \\
\noalign{\vskip 0.8em}
\multicolumn{4}{l}{\textbf{(b) Positive fixed alternative}} \\[2pt]
$40$ & $-23.0\;(4.5)$ & $-20.0\;(4.3)$ & $-4.0\;(2.0)$ \\
$120$ & $-16.0\;(3.7)$ & $-8.0\;(2.7)$ & $0.0\;(0.0)$ \\
$360$ & $0.0\;(0.0)$ & $0.0\;(0.0)$ & $0.0\;(0.0)$ \\
\noalign{\vskip 0.8em}
\multicolumn{4}{l}{\textbf{(c) Negative fixed alternative}} \\[2pt]
$40$ & $35.0\;(4.8)$ & $36.0\;(4.8)$ & $0.0\;(0.0)$ \\
$120$ & $74.0\;(4.4)$ & $88.0\;(3.3)$ & $24.0\;(4.3)$ \\
$360$ & $99.0\;(1.0)$ & $100.0\;(0.0)$ & $85.0\;(3.6)$ \\
\noalign{\vskip 0.8em}
\multicolumn{4}{l}{\textbf{(d) Weak orthogonal alternative}} \\[2pt]
$40$ & $11.0\;(4.5)$ & $6.0\;(4.0)$ & $-1.0\;(1.0)$ \\
$120$ & $30.0\;(4.8)$ & $20.0\;(4.3)$ & $1.0\;(1.0)$ \\
$360$ & $41.0\;(5.5)$ & $37.0\;(5.3)$ & $-2.0\;(2.0)$ \\
\hline
\end{tabular*}

\end{table}

Among the two tests with vector centring, the test with unit weight gives rejection rates
at least as high as the test with the inverse pooled Gram weight under the fixed alternatives, whereas
the latter gives higher rates at each of the three sample sizes under
the weak orthogonal alternative. In the null panel of
Figure~\ref{fig:centering-comparison}, all 95\% confidence intervals for the
rejection probabilities contain the nominal level $0.05$.

\newpage

\subsection{Centring with dependent targets}\label{app:connected-centering}
We examine whether vector centring reduces the sensitivity of $G$-CRT power to the sign of the background contribution, with the conditional target signal held fixed. The experiment extends the comparisons in Appendix~\ref{app:centering-simulation} to conditionally dependent targets and more background variables than observations. We use $n=80$ observations, $t=8$ target variables and $s=112$ background variables. Let $\mathbf\Theta$ be the $8\times8$ tridiagonal matrix with diagonal entries one and adjacent entries $-0.2$, and define $\mathbf V=\mathbf\Theta^{-1}$. Define $\lambda_*=1-0.4\cos(\pi/9)$ and the unit vector $\bs v\in\mathbb R^8$ by $v_j=\sqrt{2/9}\sin(j\pi/9)$, $j=1,\ldots,8$, so $\mathbf\Theta\bs v=\lambda_*\bs v$. Independently for each $i=1,\ldots,n$, generate mutually independent $\bs S_i\sim \mathcal{N}_{112}(0,\mathbf I_{112})$, $\bs E_i\sim \mathcal{N}_8(0,\mathbf V)$ and $\varepsilon_i\sim \mathcal{N}(0,1)$. For each $b\in\{-1,1\}$ and $r\in\{0,0.15,0.30,0.60\}$, use
\begin{equation}\label{eq:connected-centering-model}
 \bs X_{i,\mathcal S}=\bs S_i,\qquad
 \bs X_{i,\mathcal T}=0.5\bs v S_{i1}+\bs E_i,\qquad
 Y_i=(b-0.5r\sqrt{\lambda_*})S_{i1}+r\sqrt{\lambda_*}\bs X_{i,\mathcal T}^\top\bs v+\varepsilon_i.
\end{equation}
In this model, $\bs\beta_{\mathcal T}=r\sqrt{\lambda_*}\bs v$. The identity $\mathbf V\bs v=\bs v/\lambda_*$ and unit noise variance give the conditional signal-to-noise ratio $\bs\beta_{\mathcal T}^\top\mathbf V\bs\beta_{\mathcal T}=r^2$. Since $\mathbb E(Y_i\mid\bs X_{i,\mathcal S})=bS_{i1}$, the two values of $b$ reverse this conditional mean while leaving the conditional target signal unchanged. The null corresponds to $r=0$.

In the conditional independence graph of \eqref{eq:connected-centering-model}, the targets form a chain in the order $1,\ldots,8$, and each target is connected to just one background variable, the first. Each local regression in Algorithm~\ref{alg: exchangeable} therefore uses at most three covariates, so residual rotation remains possible even though $s>n$. We use this graph to generate $M=99$ copies with $L=1$: a forward sweep in target order $1,\ldots,8$ forms the hub, and a reverse sweep generates each branch.

With conditionally dependent targets, the residual score vectors of the copies can retain part of the target signal after these scans. The matrix $\mathbf F$ in \eqref{gcrtint:eq:scan} describes this retention: in the leading approximations \eqref{gcrtb:eq:observed-score-approximation} and \eqref{gcrtb:eq:randomised-score-approximation}, it maps the target signal contribution for the observation to that for each copy. For this graph and scan order, $\|\mathbf F\|_{\mathbf\Theta}\approx0.0598$, so each copy retains at most about $6\%$ of the observed target signal contribution in the norm induced by $\mathbf\Theta$. Thus, even with $L=1$, the difference between the leading target signal contributions retains most of the observed signal.

We compare three $G$-CRT tests using the residual score vectors and regularised residual Gram matrices in \eqref{gcrtcent:eq:residual-scores}. One test uses the uncentred gain $\bs c_v^\top\mathbf D_v^{-1}\bs c_v$, which yields the same candidate ordering as the uncentred ridge statistic by \eqref{gcrtb:eq:rank}. The other two use the symmetrically centred ridge residual scores in \eqref{eq:symmetric-centered-statistic}, with $\mathbf W=\mathbf D_{\rm pool}^{-1}$ for the inverse pooled Gram weight and $\mathbf W=\mathbf I_8$ for unit weight. The background penalty is $0.2n$, the target penalty is $\lambda_n=1$, and the intercept is unpenalised. For each test, we calculate the conservative p-value from unrounded scores by \eqref{eq:crt-pvalue-conservative} and reject when $p\leq0.05$.

The uncentred ridge statistic uses $\mathbf D_v^{-1}$ for candidate $v$, whereas each centred statistic uses the same weight for all candidates. To compare vector centring at a fixed weight, we also evaluate the uncentred statistic $\bs c_v^\top\mathbf W\bs c_v$, $v=0,\ldots,M$, with $\mathbf W=\mathbf D_{\rm pool}^{-1}$ and $\mathbf W=\mathbf I_8$. For each weight, the centred and uncentred $G$-CRT tests use the same randomised covariate copies and conservative p-value rule.

\begin{table}[!t]
\centering
\small\setstretch{1.05}
\setlength{\tabcolsep}{5pt}
\renewcommand{\arraystretch}{1.15}
\caption{Rejection rates (\%) for five $G$-CRT tests with dependent targets at nominal level $0.05$, based on 2,500 replications per setting. The inverse pooled Gram and unit weights are $\mathbf W=\mathbf D_{\rm pool}^{-1}$ and $\mathbf W=\mathbf I_8$, respectively. For each weight, the two columns compare the uncentred statistic with the statistic using vector centring.}
\label{tab:connected-centering-rates}
\begin{tabular*}{\linewidth}{@{\extracolsep{\fill}}r@{\hspace{16pt}}r@{\hspace{16pt}}rr@{\hspace{16pt}}rr@{}}
\hline
& \shortstack{Uncentred\\ridge} & \multicolumn{2}{c}{\shortstack{Inverse pooled\\Gram weight}} & \multicolumn{2}{c}{Unit weight} \\
\cline{3-4}\cline{5-6}
$r$ &  & Uncentred & Centred & Uncentred & Centred \\
\hline
\multicolumn{6}{l}{\textbf{(a) Background sign $b=-1$}} \\[2pt]
$0.00$ & $5.0$ & $4.5$ & $5.2$ & $4.2$ & $4.6$ \\
$0.15$ & $4.5$ & $4.6$ & $5.8$ & $4.1$ & $6.9$ \\
$0.30$ & $7.0$ & $7.0$ & $10.0$ & $9.4$ & $13.8$ \\
$0.60$ & $28.3$ & $22.6$ & $29.0$ & $34.6$ & $41.2$ \\
\noalign{\vskip 0.8em}
\multicolumn{6}{l}{\textbf{(b) Background sign $b=1$}} \\[2pt]
$0.00$ & $4.4$ & $5.0$ & $4.4$ & $4.7$ & $4.4$ \\
$0.15$ & $9.4$ & $8.4$ & $6.5$ & $9.6$ & $6.8$ \\
$0.30$ & $20.2$ & $15.6$ & $10.1$ & $20.2$ & $14.0$ \\
$0.60$ & $59.9$ & $42.7$ & $29.4$ & $55.3$ & $42.1$ \\
\hline
\end{tabular*}

\end{table}

We estimate the rejection rates of the five tests from 2,500 independent replications. In each replication, we generate the Gaussian variables in \eqref{eq:connected-centering-model} once and form the responses for all eight settings. We use one set of randomised covariate copies for all eight settings and all five tests. Table~\ref{tab:connected-centering-rates} reports the rejection rates. Table~\ref{tab:connected-centering-paired} reports centred minus uncentred rejection rates at each common weight and rejection rates at $b=1$ minus those at $b=-1$ for each test. The Monte Carlo standard errors of the differences are calculated from paired rejection indicators.

\newpage
\begin{figure}[!t]
\centering
\includegraphics[width=0.96\textwidth]{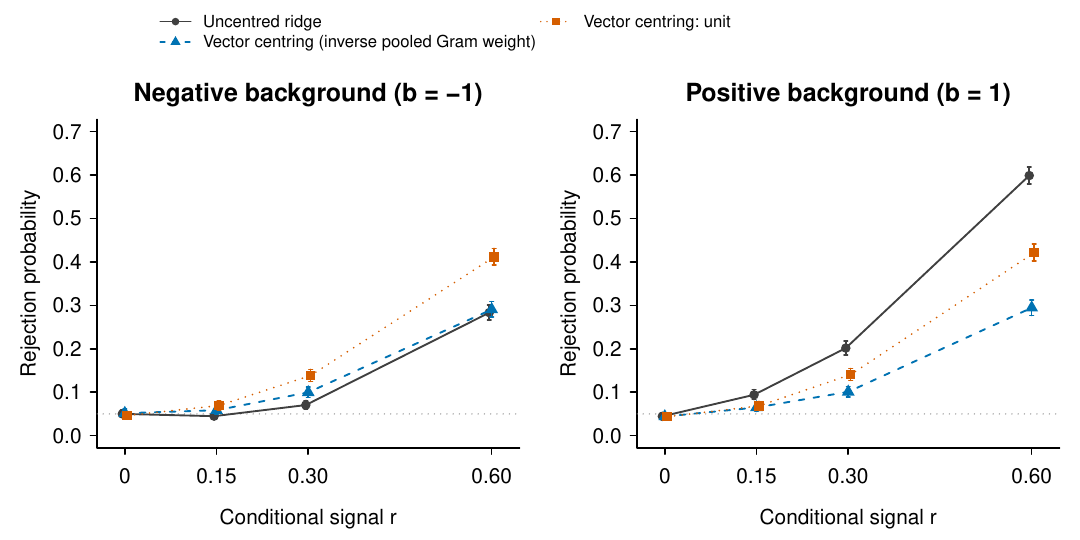}
\caption{Rejection rates for centred and uncentred $G$-CRT tests with eight dependent
targets and 112 background variables, at $n=80$ and level 0.05.
Each response setting uses 2,500 independent replications. Error bars are pointwise 95\% exact binomial intervals;
horizontal offsets separate the methods.}
\label{fig:connected-centering}
\end{figure}

Figure~\ref{fig:connected-centering} shows that the power of each of the two $G$-CRT tests with centred statistics changes little when the sign of the background contribution is reversed, whereas the test with the uncentred ridge statistic has higher power for $b=1$ than for $b=-1$. For example, at $r=0.60$, the rejection rates under $b=-1$ and $b=1$ are $29.0\%$ and $29.4\%$ for the centred test with the inverse pooled Gram weight, and $41.2\%$ and $42.1\%$ for the centred test with unit weight. The corresponding rates for the test with the uncentred ridge statistic are $28.3\%$ and $59.9\%$. At this signal strength, the centred tests have similar or higher power than the test with the uncentred ridge statistic for $b=-1$, and lower power for $b=1$.

Under the null ($r=0$), the three tests in Figure~\ref{fig:connected-centering} have rejection rates close to the nominal level $0.05$ for both signs of the background contribution; all six 95\% confidence intervals contain $0.05$.

\begin{table}[!tbp]
\centering
\small\setstretch{1.05}
\setlength{\tabcolsep}{5pt}
\renewcommand{\arraystretch}{1.15}
\caption{Paired differences in rejection rates (percentage points) for the $G$-CRT tests with dependent targets. Results use the same 2,500 replications per setting as Table~\ref{tab:connected-centering-rates}. Monte Carlo standard errors are in parentheses.}
\label{tab:connected-centering-paired}
\textbf{(a)} Centred minus uncentred at each common weight\par\smallskip
\begin{tabular*}{\linewidth}{@{\extracolsep{\fill}}rrr@{}}
\hline
$r$ & \shortstack{Inverse pooled\\Gram weight} & Unit weight \\
\hline
\multicolumn{3}{l}{\textbf{Background sign $b=-1$}} \\[2pt]
$0.00$ & $0.7\;(0.3)$ & $0.4\;(0.4)$ \\
$0.15$ & $1.2\;(0.3)$ & $2.8\;(0.4)$ \\
$0.30$ & $3.0\;(0.4)$ & $4.4\;(0.5)$ \\
$0.60$ & $6.4\;(0.6)$ & $6.6\;(0.7)$ \\
\noalign{\vskip 0.6em}
\multicolumn{3}{l}{\textbf{Background sign $b=1$}} \\[2pt]
$0.00$ & $-0.6\;(0.3)$ & $-0.3\;(0.4)$ \\
$0.15$ & $-1.9\;(0.4)$ & $-2.8\;(0.4)$ \\
$0.30$ & $-5.5\;(0.5)$ & $-6.2\;(0.6)$ \\
$0.60$ & $-13.2\;(0.7)$ & $-13.2\;(0.7)$ \\
\hline
\end{tabular*}\par\bigskip
\textbf{(b)} Background sign: $b=1$ minus $b=-1$\par\smallskip
\begin{tabular*}{\linewidth}{@{\extracolsep{\fill}}r@{\hspace{12pt}}r@{\hspace{12pt}}rr@{\hspace{12pt}}rr@{}}
\hline
& \shortstack{Uncentred\\ridge} & \multicolumn{2}{c}{\shortstack{Inverse pooled\\Gram weight}} & \multicolumn{2}{c}{Unit weight} \\
\cline{3-4}\cline{5-6}
$r$ &  & Uncentred & Centred & Uncentred & Centred \\
\hline
$0.00$ & $-0.6\;(0.5)$ & $0.5\;(0.5)$ & $-0.8\;(0.5)$ & $0.5\;(0.5)$ & $-0.3\;(0.4)$ \\
$0.15$ & $4.9\;(0.6)$ & $3.8\;(0.5)$ & $0.6\;(0.5)$ & $5.5\;(0.6)$ & $-0.1\;(0.5)$ \\
$0.30$ & $13.1\;(0.8)$ & $8.6\;(0.7)$ & $0.1\;(0.6)$ & $10.8\;(0.8)$ & $0.2\;(0.7)$ \\
$0.60$ & $31.6\;(1.0)$ & $20.0\;(0.9)$ & $0.4\;(0.8)$ & $20.8\;(0.9)$ & $1.0\;(0.9)$ \\
\hline
\end{tabular*}

\end{table}

For each common weight, vector centring increases the observed rejection rate for $b=-1$ and decreases it for $b=1$ at all three nonzero signal strengths (Table~\ref{tab:connected-centering-paired}(a)). At each nonzero signal strength, the observed rejection rates under $b=1$ and $b=-1$ are much closer for the centred tests than for their uncentred counterparts (Table~\ref{tab:connected-centering-paired}(b)). In the leading approximations \eqref{gcrtb:eq:observed-score-approximation} and \eqref{gcrtb:eq:randomised-score-approximation}, reversing $b$ reverses the common background contribution to all candidate score vectors while leaving their target signal contributions unchanged. Vector centring cancels this common term, as in \eqref{gcrtcent:eq:centered-expansion}, which is consistent with the smaller differences between the two background signs.

The centred test with unit weight has higher observed rejection rates than the centred test with the inverse pooled Gram weight at every nonzero signal strength under both signs of $b$ (Table~\ref{tab:connected-centering-rates}). The ordering is reversed for the weak orthogonal alternative in Section~\ref{app:centering-simulation} (Table~\ref{tab:centering-rates}).

\FloatBarrier

\section{Details of exploratory data illustrations}\label{app:application-details}
In this section, we provide detailed information about the applications of our testing methods to real-world data in Section~\ref{sec: Application} in the main paper.

In the applications, each forward sweep visits the updated coordinates in increasing column order in the saved covariate matrix; the reverse sweep uses the opposite order.
The statistics and fitting rules are given in Appendix~\ref{app:experimental-pvalue-conventions}. The $G$-CRT p-values are computed from unrounded scores using \eqref{eq:crt-pvalue-discrete}.

\subsection{Dependence of fund returns}\label{app: stock detail}\label{app:stock-detail}

\subsubsection{Information on the data}\label{app:stock-dataset}
We consider an exchange-traded fund (ETF) that tracks a stock index.
The SPDR Dow Jones Industrial Average ETF Trust (DIA) tracks the performance of the Dow Jones Industrial Average (DJIA, also known as Dow 30).
The DJIA is a stock market index that measures the stock performance of 30 well-established companies listed on U.S. exchanges. Its constituent weights are proportional to share prices.

We study prices from 1 September 2020 to 30 June 2022 for 103 large-cap U.S. stocks, including all stocks in either the Standard \& Poor's 100 Index or the DJIA.
These stocks define the market used to study dependence of the DIA return. 
We obtain the original Yahoo Finance daily adjusted closing prices for the 103 stocks and use a complete DIA series downloaded from the same source on 23 October 2026. The daily return is $(P_t-P_{t-1})/P_{t-1}$, where $P_t$ and $P_{t-1}$ are the adjusted closing prices on consecutive trading days. We average the returns over the same 92 non-overlapping blocks of five trading days, with the final block ending on 30 June 2022. The DIA response is expressed as a percentage return.

The Standard and Poor's 100 Index (SP100) includes 101 stocks (there are two classes of stock corresponding to one of the component companies). The constituents of the SP100 can be identified using the holdings of the index-linked product iShares S\&P 100 ETF (\url{https://www.ishares.com/us/products/239723/ishares-sp-100-etf}). We used the version of the holdings on 29 September 2023 for our analysis. We define sectors of these stocks almost the same as shown on the webpage but put the three stocks \verb|MA|, \verb|PYPL|, and \verb|V| of payment technology companies in the sector of Information Technology rather than Financials, and put \verb|TGT| in Consumer Discretionary rather than Consumer Staples. In our analysis, we have also included the two stocks \verb|TRV| and \verb|WBA| that are included by DJIA but not by the SP100. In total, we consider 103 stocks and they are divided into 11 sectors:
\begin{enumerate}
\item Communication Services (10): CHTR, CMCSA, DIS, GOOG, GOOGL, META, NFLX, T, TMUS, VZ
\item Consumer Discretionary (12): AMZN, BKNG, F, GM, HD, LOW, MCD, NKE, SBUX, TGT, TRV, TSLA
\item Consumer Staples (10): CL, COST, KHC, KO, MDLZ, MO, PEP, PG, PM, WMT
\item Energy (3): COP, CVX, XOM
\item Financials (15): AIG, AXP, BAC, BK, BLK, BRK-B, C, COF, GS, JPM, MET, MS, SCHW, USB, WFC
\item Health Care (15): ABBV, ABT, AMGN, BMY, CVS, DHR, GILD, JNJ, LLY, MDT, MRK, PFE, TMO, UNH, WBA
\item Industrials (13): BA, CAT, DE, EMR, FDX, GD, GE, HON, LMT, MMM, RTX, UNP, UPS
\item Information Technology (17): AAPL, ACN, ADBE, AMD, AVGO, CRM, CSCO, IBM, INTC, MA, MSFT, NVDA, ORCL, PYPL, QCOM, TXN, V
\item Materials (2): DOW, LIN
\item Real Estate (2): AMT, SPG
\item Utilities (4): DUK, EXC, NEE, SO
\end{enumerate}

For each stock, let $Z_i$ be its $i$th 5-day average return.
To improve marginal normality, we first subtract the sample median and apply the piecewise function
\begin{equation}\label{eq:stock-piece-transform}
\psi(z) =
\begin{cases}
    \operatorname{sign}(z)|z|^\lambda, & |z|<c,\\
    \operatorname{sign}(z)\{a+0.01\log|z|\}, & |z|\geq c,
\end{cases}
\end{equation}
where $a=c^\lambda-0.01\log c$ and $z$ is a median-centred return.
We select $\lambda$ from 20 equally spaced values between 0.1 and 2, and
$c$ from 20 equally spaced values between the 80\% empirical quantile and
the maximum of the absolute median-centred returns, to maximise the
Shapiro--Wilk p-value of the transformed column.

The transformed returns for MMM and MO have Shapiro--Wilk p-values below
0.05. For these two stocks, we instead start with the uncentred returns
and define the positive values
$u_i=Z_i-\min_k Z_k+0.01(\max_k Z_k-\min_k Z_k)$.
The observed range is positive for both stocks. We apply the Box--Cox transformation
\begin{equation}\label{eq:stock-boxcox-transform}
\varphi(u)=
\begin{cases}
 (u^\lambda-1)/\lambda,&\lambda\ne0,\\
 \log u,&\lambda=0.
\end{cases}
\end{equation}
The parameter maximises the profile log likelihood returned by
\texttt{MASS::boxcox()} for an intercept-only model, on its grid from $-2$
to $2$ in steps of $0.1$. The selected values are $\lambda=1.4$ for MMM
and $\lambda=1.5$ for MO.
Finally, all 103 transformed columns are centred and scaled to unit sample
variance.

\subsubsection{Preliminary analysis}\label{app:stock-preliminary}

We begin with a graphical model for the stock returns $\mathbf{X}$. 
As a working model, we approximate the rows of $\mathbf{X}$ as i.i.d. samples from a normal distribution and seek a graph \(G\) such that this distribution is well represented by a GGM with respect to \(G\).
As discussed in Section~\ref{sec: G-CRT assumption}, all we need to implement the $G$-CRT is one particular supergraph of the conditional independence graph.
To facilitate the selection of such a supergraph, we incorporate the prior knowledge of the stock sectors (based on their industries).
We estimate $G$ using the graphical Lasso (GLasso) with regularisation parameter $\lambda=0.15$ and with no penalty on edges between stocks in the same sector.
The resulting graph $\widehat{G}$ has maximum degree $30$ and median degree 19, with around $18.4\%$ of node pairs connected.
Since \(n=92\) and the maximal degree is \(30\), we have \(n>|N_i|+1\) for every node in \(\widehat G\). 
The graph $\widehat{G}$ we estimated using GLasso is shown in Figure \ref{fig: sp100_G}, where nodes are coloured according to their sectors.
\begin{figure}[!tbp]
\centering

    \includegraphics[scale=0.8]{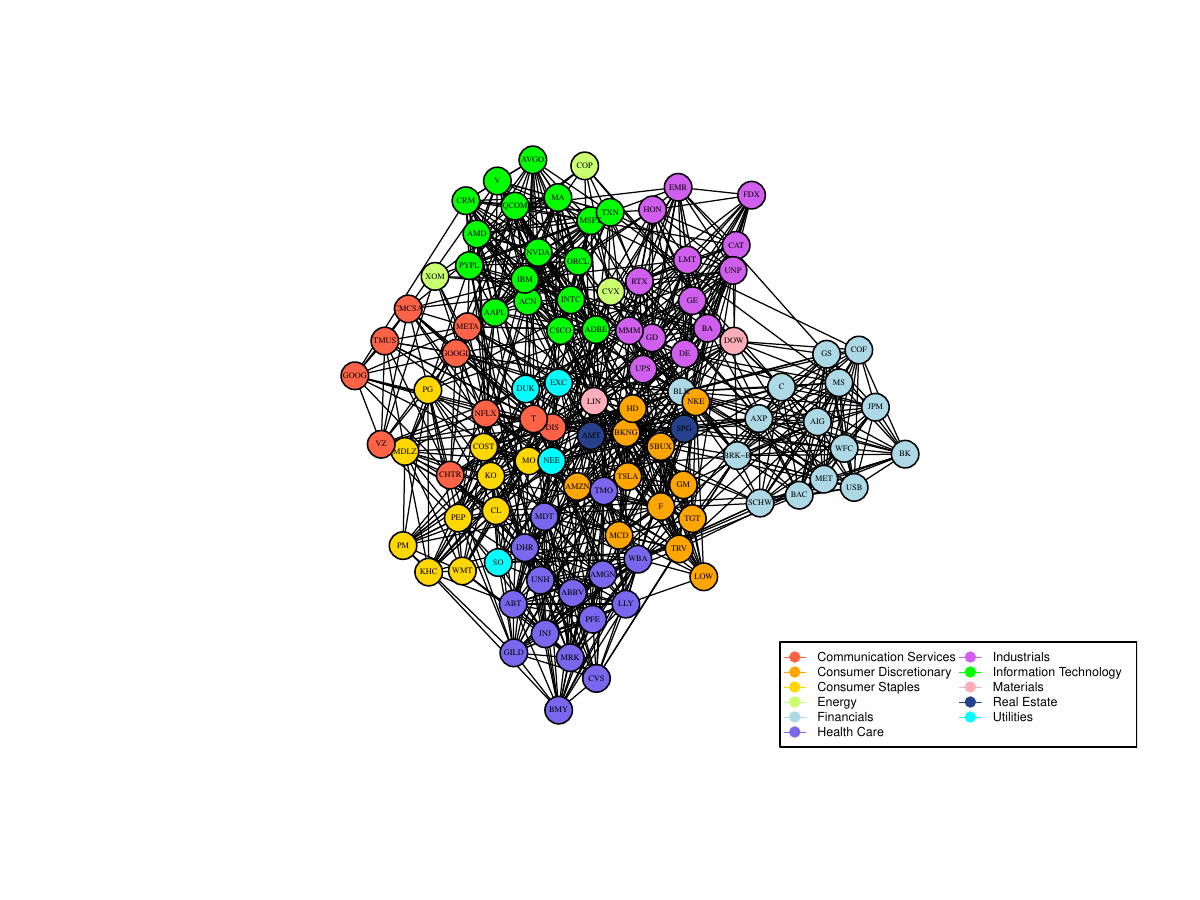}
    \caption{Estimated Graph with Colour for Sector}\label{fig: sp100_G}
\end{figure}

\enlargethispage{2\baselineskip}
We assess the GGM with $\widehat{G}$ using the Monte Carlo goodness-of-fit test of \citet{lin_goodness--fit_2025} with statistic F$_{\Sigma}$, together with the \nameVV procedure of Verzelen and Villers and the Bonferroni-adjusted (\nameDP) procedure of Drton and Perlman, as described in Section~5 of that reference.
The p-values are listed in Table~\ref{tab: GoF stock} and are all above 0.3. We report them as descriptive diagnostics for the GGM with $\widehat{G}$ in this exploratory analysis.
The transformations and graph are selected from the observations used for testing. Proposition~\ref{prop:learned-graph-size-rewrite} assumes independent auxiliary data for graph selection and therefore does not directly apply to this exploratory analysis.

\begin{table}[!b]
    \centering
    \caption{Results of GoF tests for modelling the stock returns by the GGM with the estimated graph. }
    \label{tab: GoF stock}
    \begin{tabular}{l|ccc}
\hline
 & \multicolumn{3}{c}{Method} \\ 
  & \nameVV & \nameDP & \multicolumn{1}{c}{F$_{\Sigma}$} \\ 
\hline
\hline p-value  & 1.000 & 0.352 & 0.310 \\
\hline 
\end{tabular}

\end{table}

The relationship between $Y$ and $\bs{X}$ can be modelled either by a linear model or by a nonparametric regression model. In addition, the dimension of data ($p=103$) is larger than the number of observations ($n=92$).
Therefore, we should consider the testing methods that have been numerically studied in Sections~\ref{sec: simulation CRT linear regression} and \ref{sec: simulation CRT nonlinear regression}.
For the $G$-CRT, we choose LM-L1-R-SSR and RF-RR based on their performance in our high-dimensional linear and nonlinear regression simulations, respectively. We generate the copies using Algorithm~\ref{alg: exchangeable}, with $\mc I=\mc T$ and $L=3$. We use $M=1000$ copies for the observed data and $M=400$ copies in each replication under the fitted null model or with added noise. We also consider Bonferroni-adjusted coordinatewise baselines, including the de-sparsified Lasso and the dCRT with either Gaussian Lasso models (L1) or random forests (RF).

All random forests in the stock analysis use 1,000 trees, including both stages of RF-RR and the response fits for dCRT (RF). We chose this setting after diagnostics for $\mc T_1$ showed that p-values for the $G$-CRT with RF-RR varied substantially with forest randomness at 100 trees.

We examine rejection rates under a fixed fitted null model. For each target set, we fit a Lasso regression of $Y$ on the background covariates, choosing the penalty that minimises the ten-fold cross-validation error. The transformation, graph, target set for testing, fitted coefficients and residual standard deviation are then held fixed. Each replication starts from the observed covariate matrix and generates a new matrix by applying Algorithm~\ref{alg: exchangeable} with all $p$ covariates treated as targets, $\mc I=[p]$, $M=1$ and $L=20$. The response is generated from the linear predictor fitted using the background covariates plus independent Gaussian errors with the fitted residual standard deviation. Table~\ref{tab: fund error check} reports rejection rates over 400 replications, averaging over this covariate generation, the response errors and the testing randomness. The distribution of the generated covariates is conditional on the observed matrix. Using $L=20$ does not establish that the generated rows are i.i.d.\ from a GGM\@.

At nominal level $0.05$, the rejection rates of the four CRTs range from $0.8\%$ to $7.3\%$, compared with $26.0\%$ and $15.0\%$ for the de-sparsified Lasso with $\mc T_1$ and $\mc T_2$, respectively (Table~\ref{tab: fund error check}).

\begin{table}[!tbp]
    \centering
\small\setstretch{1.05}
\setlength{\tabcolsep}{5pt}
\renewcommand{\arraystretch}{1.15}
\caption{Rejection rates (\%) under the fixed fitted null model for the stock return analysis. Results are based on 400 replications at nominal level $0.05$. Monte Carlo standard errors are in parentheses. The $G$-CRT columns use the indicated statistics; the dCRT columns use Lasso (L1) or random forest (RF) fits.}
    \label{tab: fund error check}
   \begin{tabular*}{\linewidth}{@{\extracolsep{\fill}}l@{\hspace{12pt}}rr@{\hspace{12pt}}r@{\hspace{12pt}}rr@{}}
\hline
& \multicolumn{2}{c}{$G$-CRT} & \shortstack{De-sparsified\\Lasso} & \multicolumn{2}{c}{dCRT} \\
\cline{2-3}\cline{5-6}
$\mc T$ & LM-L1-R-SSR & RF-RR &  & L1 & RF \\
\hline
$\mc T_1$ & $5.3\;(1.1)$ & $4.8\;(1.1)$ & $26.0\;(2.2)$ & $6.0\;(1.2)$ & $2.3\;(0.7)$ \\
$\mc T_2$ & $4.5\;(1.0)$ & $5.0\;(1.1)$ & $15.0\;(1.8)$ & $7.3\;(1.3)$ & $0.8\;(0.4)$ \\
\hline
\end{tabular*}

\end{table}

\subsubsection{Analysis on the fund dependence}\label{app:stock-fund-dependence}

Because of the high rejection rates of the de-sparsified Lasso in Table~\ref{tab: fund error check}, we focus on the four CRTs. Their p-values for the observed data are reported in Table~\ref{tab: fund p-value}.

\begin{table} [t]
\small\setstretch{1.05}
\setlength{\tabcolsep}{5pt}
\renewcommand{\arraystretch}{1.15}
\caption{p-values of MCI tests with different $\mc T$ on the DIA returns. The $G$-CRT columns use the indicated statistics; the dCRT columns use Lasso (L1) or random forest (RF) fits.}\label{tab: fund p-value}
      \centering
\begin{tabular*}{\linewidth}{@{\extracolsep{\fill}}l@{\hspace{14pt}}rr@{\hspace{18pt}}rr@{}}
\hline
& \multicolumn{2}{c}{$G$-CRT} & \multicolumn{2}{c}{dCRT} \\
\cline{2-3}\cline{4-5}
$\mc T$ & LM-L1-R-SSR & RF-RR & L1 & RF \\
\hline
$\mc{T}_1$  & $0.010$ & $0.237$ & $0.020$ & $0.254$ \\
$\mc{T}_2$  & $0.001$ & $0.004$ & $<0.001$ & $0.084$ \\
\hline
\end{tabular*}

\end{table}

For $\mc T_1$, the $G$-CRT with LM-L1-R-SSR and dCRT (L1) reject at the $0.05$ level, whereas the $G$-CRT with RF-RR and dCRT (RF) do not (Table~\ref{tab: fund p-value}).

For $\mc{T}_2$, both $G$-CRTs and dCRT (L1) reject at the $0.05$ level, whereas dCRT (RF), with p-value $0.084$, does not.

 We assess sensitivity to added noise using $\tilde{\bs Y}=\bs Y+\sigma\bs\varepsilon$, where $\bs\varepsilon\sim\mathcal N_n(\bs0,\mathbf I_n)$. The response is expressed as a percentage return, so $\sigma\in\{0.05,0.10,0.15,0.20,0.25\}$ is measured in percentage points. All methods use the same perturbed response in each replication and refit their background models. The covariate matrix, graph and target set remain fixed. We use 400 replications at each noise level, with the same seed for corresponding replications across levels.

Figure~\ref{fig: stock noise} in Section~\ref{sec: fund} presents the empirical detection rates across the five noise levels. The $G$-CRT with LM-L1-R-SSR has the highest empirical detection rate among the four CRT procedures at all five noise levels, tied with dCRT (L1) at the lowest noise level. For example, at $\sigma=0.10$, its detection rate is $94.3\%$, compared with $72.8\%$ for dCRT (L1). Both tests based on linear models have higher detection rates than the two tests based on random forests throughout the examined range. Detection rates generally decline as noise increases; dCRT (RF) rises from $31.5\%$ at $\sigma=0.05$ to $35.0\%$ at $\sigma=0.10$ before declining.

The observed p-values and detection rates compare the specified procedures on these stock groups. Their differences reflect the statistics and background fits as well as joint or coordinatewise aggregation. Appendix~\ref{app:matched-ridge} compares $G$-CRT with a coordinatewise test using the same graph and with oracle joint CRTs using the same statistics.

\subsection{Breast cancer relapse}\label{app: BCR details}\label{app:bcr-details}

We compare the $G$-CRT with the likelihood ratio test (LRT) and coordinatewise dCRT baselines in a breast cancer relapse analysis.

\subsubsection{Information on the data}
The breast cancer relapse (BCR) data were collected by \citet{wang_gene-expression_2005} to improve risk assessment in patients with lymph-node-negative breast cancer by identifying genome-wide measures of gene expression associated with the development of distant metastases.

We obtain the breast cancer relapse (BCR) data from the NCBI GEO database with accession GSE2034 \citep{barrett2012ncbi} at
\url{https://www.ncbi.nlm.nih.gov/geo/query/acc.cgi?acc=gse2034}.
The original data contain the expression of 22,000 transcripts from the total RNA of frozen tumour samples from 286 patients with lymph-node-negative breast cancer \citep{wang_gene-expression_2005}.
The data include patients' clinical information, including oestrogen receptor (ER) status, time to relapse or last follow-up (months), relapse status, and brain relapse status.
Patients were divided into two groups based on ER levels: those with more than 10 fmol per mg protein (ER-positive) and those with less.
We analyse the 209 ER-positive patients using the 60 probes identified for this subgroup in \citet[Table 3]{wang_gene-expression_2005}.

The response $Y$ has three categories: no relapse within 48 months, relapse within 48 months without recorded brain relapse, and relapse within 48 months with recorded brain relapse. The three response categories contain 152, 52 and 5 patients, respectively. Brain relapse is recorded without a separate event date. No patient was censored without relapse at or before 48 months.
The covariates are the gene expression levels, after applying the rank-based Inverse Normal Transformation (INT) \citep{fisher1953statistical}.
For an observed vector $\bs z=(z_1,\ldots,z_n)^\top$, the INT is defined componentwise by
$$
z_{i}^{(\text{INT})} = \Phi^{-1}\left( \frac{r_{i}-0.5}{n} \right),
$$
where $r_i$ is the rank of $z_i$ among $z_1,\ldots,z_n$ and $\Phi^{-1}$ is the inverse CDF of the standard normal distribution. See \citet{beasley2009rank, mccaw2020operating} for more details about INT.
We apply INT to each probe using all 286 patients, then restrict the analysis to the 209 ER-positive patients.

\subsubsection{Preliminary analysis}

To define the background and target probe sets, we fit a symmetric multinomial logistic regression of $Y$ on the transformed expression levels of the 60 probes, with a lasso penalty \citep{friedman2010}. We use \texttt{cv.glmnet()} in the \textbf{R} package \texttt{glmnet}, with \texttt{family} set to \texttt{multinomial} and five-fold cross-validation. A probe is selected if at least one of its estimated coefficients for the response categories is nonzero.

We compare two choices of the penalty. The first targets predictive performance: $\hat{\lambda}_{\min}$ minimises the mean validation deviance and selects 21 probes. To favour a more parsimonious model, the ``one-standard-error'' rule chooses the largest penalty $\lambda$ whose validation deviance is within one standard error of the minimum \citep{hastie2009elements}; it selects three probes.

The three probes selected by the ``one-standard-error'' rule are \texttt{218478\_s\_at}, \texttt{209835\_x\_at}, and \texttt{221916\_at}; they form the set $\mc{S}$. Furthermore, $\mc{T}$ consists of the 18 probes selected by $\hat{\lambda}_{\min}$ but not by the ``one-standard-error'' rule: \texttt{201091\_s\_at}, \texttt{221344\_at}, \texttt{210028\_s\_at}, \texttt{218782\_s\_at}, \texttt{219724\_s\_at}, \texttt{204768\_s\_at}, \texttt{218914\_at}, \texttt{201068\_s\_at}, \texttt{214919\_s\_at}, \texttt{217471\_at}, \texttt{205034\_at}, \texttt{218533\_s\_at}, \texttt{215633\_x\_at}, \texttt{221928\_at}, \texttt{214806\_at}, \texttt{204540\_at}, \texttt{209524\_at}, and \texttt{211040\_x\_at}. Choosing $\mc S$ and $\mc T$ using the response makes the real-data analysis exploratory; a formal post-selection guarantee would require external selection, sample splitting, or an argument based on selective inference.

We test whether the 18 selected target probes in $\mathcal T$ provide additional information about the response given the three selected background probes in $\mathcal S$.
The inference problem is to test
\begin{equation}\label{eq: cancer CI}
H_{0}: Y\indp \bs{X}_{\mc{T}}\mid \bs{X}_{\mc{S}}.
\end{equation}

The likelihood ratio test compares the unpenalised multinomial logistic models using $\bs{X}_{\mc{S}}$ alone and using both $\bs{X}_{\mc{S}}$ and $\bs{X}_{\mc{T}}$. Its p-value uses the asymptotic chi-squared reference distribution with $18(K-1)$ degrees of freedom, where $K$ is the number of observed response categories. Some fits exhibit separation, so we compute the log-likelihood suprema even when finite maximum likelihood estimates do not exist. For both the observed data and all 400 simulated datasets, we verify the log-likelihood values using lower and upper bounds with a gap of at most $10^{-5}$.
This gives a p-value of approximately $9.58\times10^{-6}$. In the simulation under conditional independence reported in Table~\ref{tab:appendix BCR type 1}, however, the test rejects in $40.0\%$ of replications at the nominal $0.05$ level. We use the $G$-CRT to avoid relying on the chi-squared approximation.
When implementing the $G$-CRT,
we use a complete graph for $G$ because the dimension of covariates $\bs{X}_{\mc{T}\cup \mc{S}}$ is $21$, much smaller than the sample size $n$.
Consequently, the modelling assumption is that $\bs{X}_{\mc{T}\cup \mc{S}}$ follows a multivariate Gaussian distribution.
Given that $Y$ is a multiclass response, we can conduct the $G$-CRT with the RF statistic. The $G$-CRT with RF uses classification. We generate the copies using Algorithm~\ref{alg: exchangeable}, with $\mc I=\mc T$ and $L=3$. We use $M=1000$ copies for the observed data and $M=400$ copies in each replication under the fitted null model.
For comparison, we also conduct Bonferroni-adjusted coordinatewise dCRT baselines. The dCRT baselines code the response categories as $1,2,3$ and use Gaussian Lasso regression or random forest regression.

\subsubsection{\texorpdfstring{Simulation under the fitted null model}{Simulation under the fitted null model}}
This simulation examines the specified procedures under a fitted null model with $\mc S$ and $\mc T$ held fixed. We fit a multivariate Gaussian law using the sample mean and covariance of the transformed expression levels of the 21 selected probes. We also fit a lasso-penalised multinomial logistic model of $Y$ on $\bs{X}_{\mc S}$, choosing the penalty that minimises the five-fold cross-validation deviance. In each replication, we sample $\tilde{\bs{X}}_{\mc T\cup\mc S}$ from the fitted Gaussian law and generate $\tilde Y$ from the fitted multinomial model using $\tilde{\bs{X}}_{\mc S}$. We compute p-values for the LRT, the $G$-CRT with the RF statistic, and the two dCRT variants.

We run 400 replications with $\mc S$ and $\mc T$ fixed. When a simulated response omits a category, we use the observed number of categories to determine the LRT degrees of freedom and omit unused category levels from the RF fit. All four methods yield p-values in all 400 replications. This simulation does not repeat probe selection or the inverse normal transformation and does not establish validity of the p-value obtained after selection in the original analysis.

\begin{table}[!b]
    \centering
\small\setstretch{1.05}
\setlength{\tabcolsep}{5pt}
\renewcommand{\arraystretch}{1.15}
\caption{Rejection rates (\%) under the fitted null model with the selected probe sets held fixed. Results are based on 400 replications at $\alpha=0.05$. Monte Carlo standard errors are in parentheses.}
    \label{tab:appendix BCR type 1}
    \begin{tabular*}{\linewidth}{@{\extracolsep{\fill}}lrrrr@{}}
\hline
& LRT & $G$-CRT (RF) & dCRT (L1) & dCRT (RF) \\
\hline
\shortstack[l]{Rejection rate (\%)\\(SE)} & $40.0\;(2.5)$ & $7.3\;(1.3)$ & $3.3\;(0.9)$ & $3.3\;(0.9)$ \\
\hline
\end{tabular*}

\end{table}

\begin{figure}[!t]
    \centering
    \includegraphics[scale=0.5]{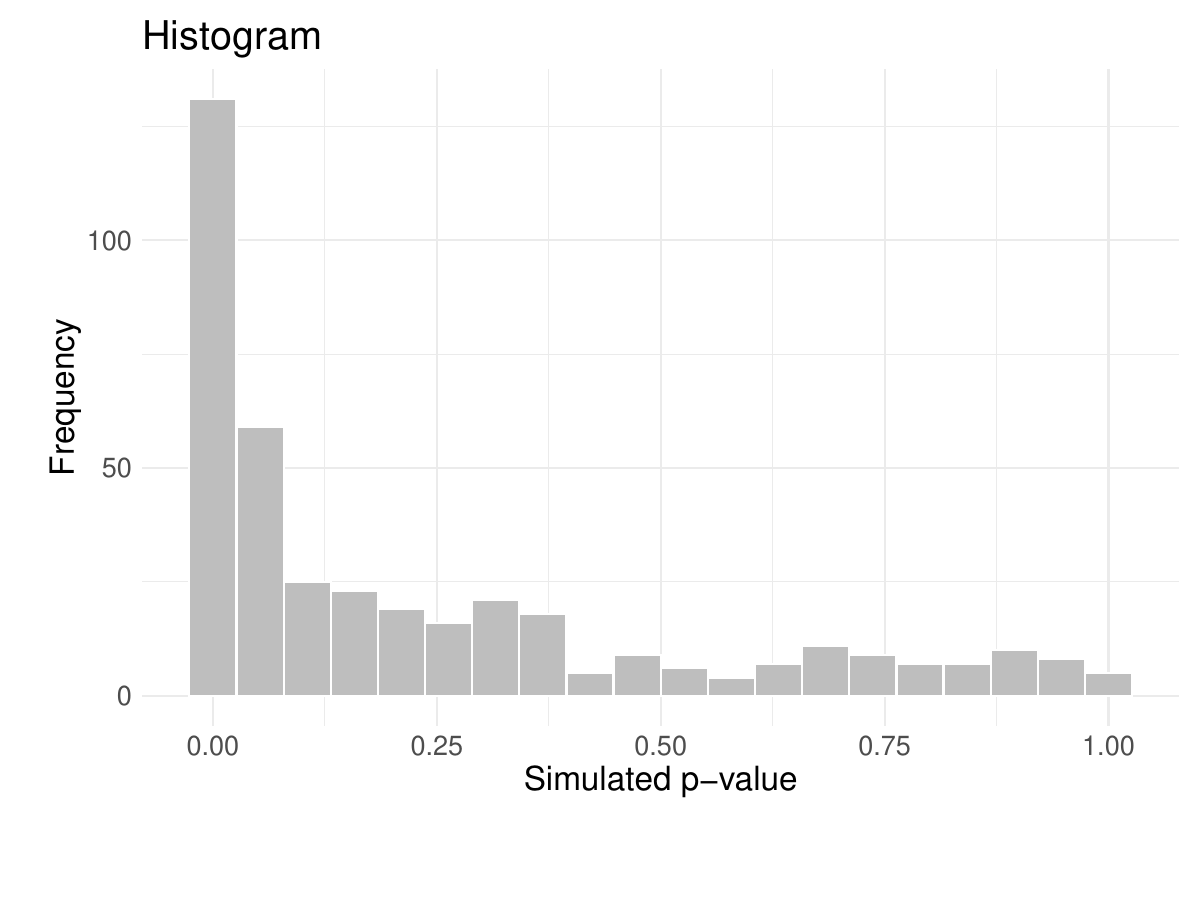}
    \caption{Histogram of p-values obtained by the likelihood ratio test in 400 replications under the fitted null model, with the selected probe sets held fixed. The p-values use the chi-squared reference distribution after numerical verification of the log-likelihood values.}
    \label{fig: BCR pvalue lrt}
\end{figure}

At the nominal $0.05$ level, the LRT rejects in $40.0\%$ of replications, compared with $7.3\%$ for $G$-CRT and $3.3\%$ for each dCRT baseline (Table~\ref{tab:appendix BCR type 1}). Figure~\ref{fig: BCR pvalue lrt} shows an excess of small LRT p-values, indicating that the chi-squared approximation is inaccurate in this setting. The rejection rates of the three CRTs are much closer to the nominal level.

\newpage
\subsubsection{Results of tests}

On the observed data, the p-values are $0.018$, $0.697$, and $0.407$ for $G$-CRT with the RF statistic, dCRT (L1), and dCRT (RF), respectively. 

At the $0.05$ level, the $G$-CRT rejects the joint null, whereas neither coordinatewise dCRT does. Its rejection rate under the fitted null is much closer to the nominal level than that of the likelihood ratio test. The differences in the observed p-values also reflect the treatment of the response categories and the scoring methods. Appendix~\ref{app:matched-ridge} compares $G$-CRT with a coordinatewise test using the same graph and with oracle joint CRTs using the same statistics.

\clearpage
\begingroup
\hypersetup{pdfborder={0 0 0}}
\putbib[reference]
\endgroup
\end{bibunit}
\end{document}